\documentclass[aps,notitlepage,twocolumn,nofootinbib,pra,10pt]{revtex4-1}
\usepackage{amsmath,amsfonts,amssymb,amsthm,bbm,graphicx,enumerate,times}
\usepackage{mathtools}
\usepackage[usenames,dvipsnames]{color}

\usepackage{hyperref}
\usepackage{comment}
\usepackage{tikz}
\usepackage{graphicx}
\usepackage{float}
\usetikzlibrary{shapes,positioning}
\usepackage{etoolbox}
\usepackage{comment}
\usepackage{todonotes}
\usepackage{accents}

\definecolor{darkgreen}{rgb}{.15,.6,.15}
\definecolor{darkcyan}{rgb}{.15,.5,.5}
\definecolor{darkred}{rgb}{.6,.15,.15}

\usepackage[normalem]{ulem}

\usepackage{dsfont}
\newcommand{\id}{\mathds{1}}

\newcommand{\tr}{\mathrm{tr}}

 \renewcommand{\i}{\,\ensuremath\mathrm{i}}

\newtheorem{theorem}{Theorem}
\newtheorem{lemma}[theorem]{Lemma}

\newcommand{\ket}[1]{\left.\left|{#1}\right.\right\rangle}
\newcommand{\sket}[1]{\left.|{#1}\right.\rangle}

\newcommand{\bra}[1]{\left.\left\langle{#1}\right.\right|}

\newcommand{\braket}[2]{\left\langle #1 \middle| #2 \right\rangle}

\newcommand{\ketbra}[2]{\ket{#1} \!\! \bra{#2}}

\newcommand{\sandwich}[3]
  {\left\langle  #1 \right| #2 \left| #3 \right\rangle}

\newcommand\vacbra{{\bra{\emptyset}}}
\newcommand\vacket{{\ket{\emptyset}}}

\DeclareMathAlphabet{\mathpzcc}{OT1}{pzc}{m}{it}
\DeclareMathAlphabet{\mathpzc}{T1}{pzc}{m}{it}{\huge}

\newcommand{\m}{\operatorname{\gamma}}
\newcommand{\mt}{\operatorname{\tilde{\gamma}}}
\newcommand{\fe}{\operatorname{f}^{\phantom{\dagger}}}

\newcommand{\fd}{\operatorname{f}^\dagger}
\newcommand{\ftd}{\tilde{\operatorname{f}}^\dagger}

\newcommand\Par[1][]{%
  \ifstrempty{#1}{%
    \mathcal P_\text{tot}
  }{
    \mathcal P_{#1}
  }
}

\def\and{\quad\text{and}\quad}

\newcommand\BC[1]{ {\{0,1\}^{\times #1}}}

\begin{document}
\title{Majorana dimers and holographic quantum error-correcting codes}
\author{A.\ Jahn, M.\ Gluza, F.\ Pastawski, J.\ Eisert}
\address{Dahlem Center for Complex Quantum Systems, Freie Universit{\"a}t Berlin, 14195 Berlin, Germany}

\begin{abstract}
Holographic quantum error-correcting codes have been proposed as toy models that describe key aspects of the AdS/CFT correspondence. In this work, we introduce a versatile framework of Majorana dimers capturing the intersection of stabilizer and Gaussian Majorana states. This picture allows for an efficient contraction with a simple diagrammatic interpretation and is amenable to analytical study of  holographic quantum error-correcting codes. Equipped with this framework, we revisit the recently proposed hyperbolic pentagon code (HyPeC). Relating its logical code basis to Majorana dimers, we efficiently compute boundary state properties even for the non-Gaussian case of generic logical input.
The dimers characterizing these boundary states coincide with discrete bulk geodesics, leading to a geometric picture from which properties of entanglement, quantum error correction, and bulk/boundary operator mapping immediately follow. 
We also elaborate upon the emergence of the Ryu-Takayanagi formula from our model, which realizes many of the properties of the recent bit thread proposal.
Our work thus elucidates the connection between bulk geometry, entanglement, and quantum error correction in AdS/CFT, and lays the foundation for new models of holography.
\end{abstract}

\maketitle
\date{\today}
\tableofcontents
\section{Introduction}
The holographic principle -- the idea that certain theories of gravity are dual to lower-dimensional quantum field theory -- has had wide-ranging applications within theoretical physics.
In particular, the AdS/CFT correspondence has changed our understanding of theories of both (quantum) gravity and quantum field theory, by giving a specific relationship between gravity on $d{+}1$-dimensional negatively curved \textsl{Anti-de Sitter} spacetime (AdS) and $d$-dimensional conformal field theory (CFT) \cite{Maldacena98,Witten:1998qj}.
A number of simple models capturing key aspects of holography have been constructed \cite{Almheiri15,Qi2013,Pastawski2015,Hayden2016,Mintun2015,PhysRevD.86.065007}, largely relying on \textsl{tensor network} descriptions of bulk AdS geometry and boundary states. 
Tensor networks have long been understood as describing a state in terms of its entanglement structure \cite{verstraete2006matrix}, thus serving as an ideal tool to study holography in terms 
of notions of quantum information theory
\cite{fannes1992finitely,VerstraeteBig,schollwock2011density,orus2014practical,AreaReview}.
The basis of this work is the tensor network construction of the \emph{hyperbolic pentagon code}  (HyPeC), a class of holographic models often named \emph{HaPPY codes} after the authors' initials \cite{Pastawski2015}. These codes explicitly realize \textsl{holographic quantum error correction} \cite{Almheiri15} by providing an error-correctable mapping from bulk to boundary degrees of freedom, reproducing many of the features of AdS/CFT. 
However, the boundary states of the HyPeC differ from other tensor network models specifically designed to produce physical CFTs, such as the MERA \cite{PhysRevLett.101.110501}. 
For computational basis bulk inputs, where the tensor network becomes Gaussian and efficiently contractible, earlier studies revealed a pair-wise correlation structure in terms of boundary Majorana modes \cite{Jahn:2017tls}. 
As we show in this work, HyPeC states are in fact a special case of a \textsl{Majorana dimer model}, and can be described by entangled fermionic pairs.
Majorana dimers have previously been used to describe superconducting phases on lattices \cite{PhysRevB.94.115127,PhysRevB.94.115115}, as instances of
tensor networks that have a fermionic component 
\cite{PhysRevA.80.042333,PhysRevA.81.052338,CorbozPEPSFermions,PhysRevB.95.245127,PhysRevB.95.075108}.
We show that the contraction of dimer-based tensor networks is equivalent to combining entangled Majorana pairs, replacing the computational difficulties of contraction by simple rules on dimer diagrams. 
This graphical language directly visualizes parities, physical correlations, and the entanglement structure of quantum states spanning the entire fermionic Hilbert space.
By deriving the holographic properties of the HyPeC merely from emergent entangled pairs, we connect to recent proposals of AdS/CFT models based on bit-threads \cite{freedman2017bit,Cui:2018dyq}.
Thus, our work is also an important step towards integrating discrete tensor network models of AdS/CFT into a unified setting.

\section{A simple model of holography}
Consider the \emph{boundary} and \emph{bulk} Hilbert spaces denoted by $\mathcal H_\partial$ and  $\mathcal H_\text{bulk}$ respectively.
A holographic quantum error-correcting code is formed by an \emph{encoding isometry} $E$ from the logical states in $H_\text{bulk}$ to boundary states in $\mathcal H_\text{code}\subset \mathcal H_\partial$. 
Indeed, $EE^\dagger$ is the projector onto the code $\mathcal H_\text{code}$ of the boundary Hilbert space $\mathcal H_\partial$.
Any bulk operator $\mathcal O$ acting on the states in $\mathcal H_\text{bulk}$ can be represented by at least one operator $\mathcal O_\partial$ acting on $\ket {\psi_\text{code}} \in \mathcal H_\text{code}$ with the property $E^\dagger {\mathcal O}_\partial E = \mathcal O$ while preserving the code subspace ($[{\mathcal O}_\partial, EE^\dagger]=0$). 
The specific form of such a mapping from bulk to boundary is the \emph{holographic dictionary} obtained in continuum AdS/CFT  by equating bulk and boundary partition functions \cite{Witten:1998qj}, which is equivalent to considering boundary CFT operators $\mathcal O_\partial$ as limits of fields on the gravitational AdS background \cite{Harlow:2011ke}.
As we visualize in Fig.\ \ref{FIG_ADS_EE_WEDGE} (left), oftentimes $\mathcal O_\partial$ acts non-trivially only on a subregion of the total boundary. 
Given a subregion $A$ on the boundary one can perform the so-called \emph{AdS/Rindler reconstruction} \cite{Susskind:1998dq,Banks:1998dd,PhysRevD.59.104021,Harlow:2011ke,PhysRevLett.117.021601,PhysRevD.74.066009} to associate to any boundary operator $\mathcal O_A$ a corresponding bulk operator $\mathcal O$ acting within the \emph{wedge} $\mathcal{W}[A]$ which is a subset of the bulk.

Due to the computational difficulties in studying continuum AdS/CFT, discrete toy models often provide an easier approach to understanding its properties. These models usually consider a space-like slice of the full AdS spacetime, discretized by a tiling whose open boundary edges correspond to the AdS boundary. Subsets of these open edges are then identified with subregions of the boundary CFT (see Fig.\ \ref{FIG_ADS_EE_WEDGE}, right).

What properties should the discretized boundary states in $H_\text{code}$ fulfill? As a bulk operator can be represented equivalently on different parts of the boundary, e.g.\ two regions $A$ and $B$, we are led to the condition
\begin{equation}
\label{EQ_BDY_OPERATORS}
\mathcal O_A \ket {\psi_\text{code}} = \mathcal O_B \ket {\psi_\text{code}} \ ,
\end{equation}
where $\mathcal O_A$ and $\mathcal O_B$ are boundary representations on $A$ and $B$ of an operator $\mathcal O$ inserted somewhere in the bulk.
For this condition to hold for any $\mathcal O$ and any suitable $A$ and $B$, the states in $H_\text{code}$ must necessarily possess multi-partite and nonlocal entanglement to allow for operators that act equivalently on distant parts of the boundary.

\begin{figure}[tb]
\centering
\includegraphics[height=0.18\textheight]{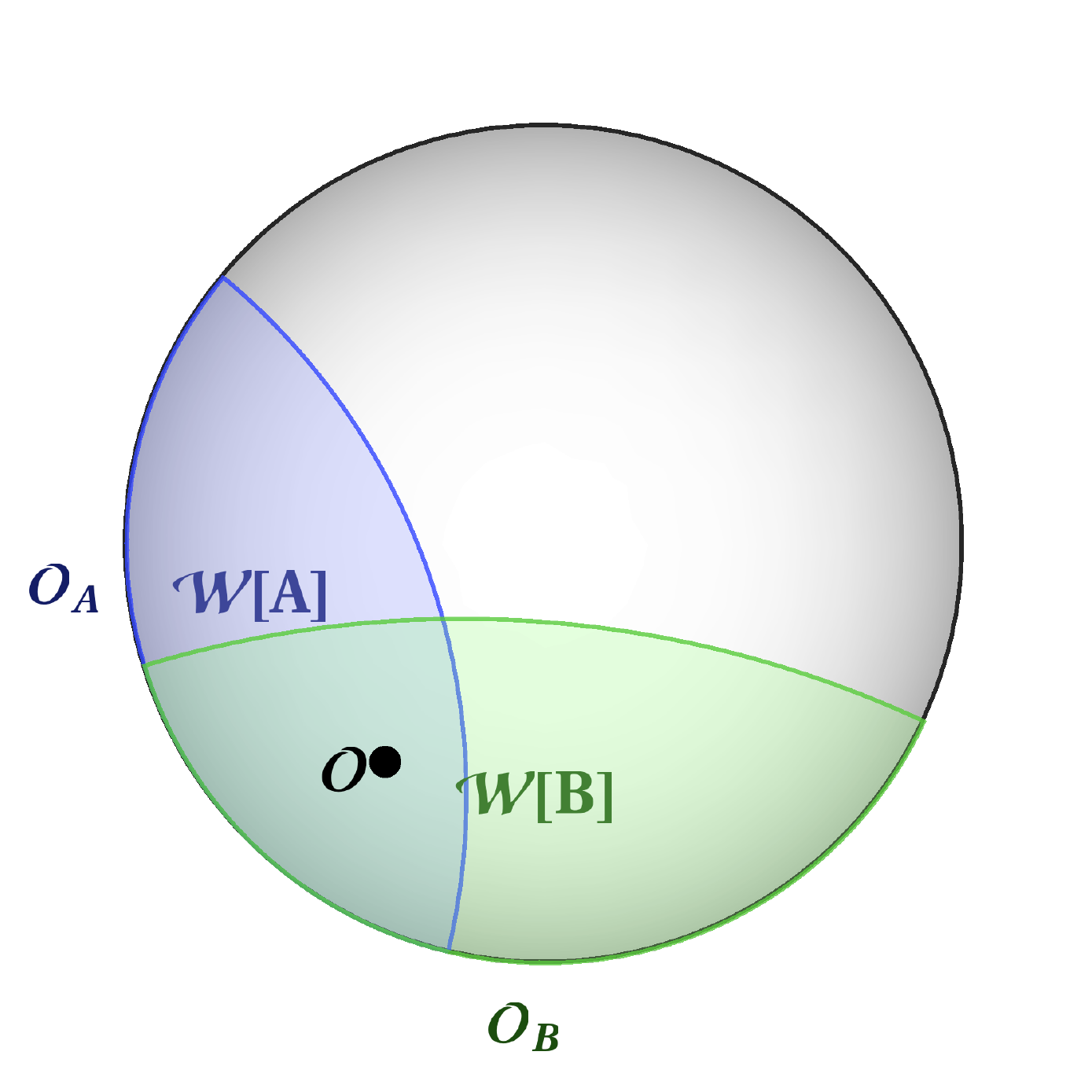}
\includegraphics[height=0.18\textheight]{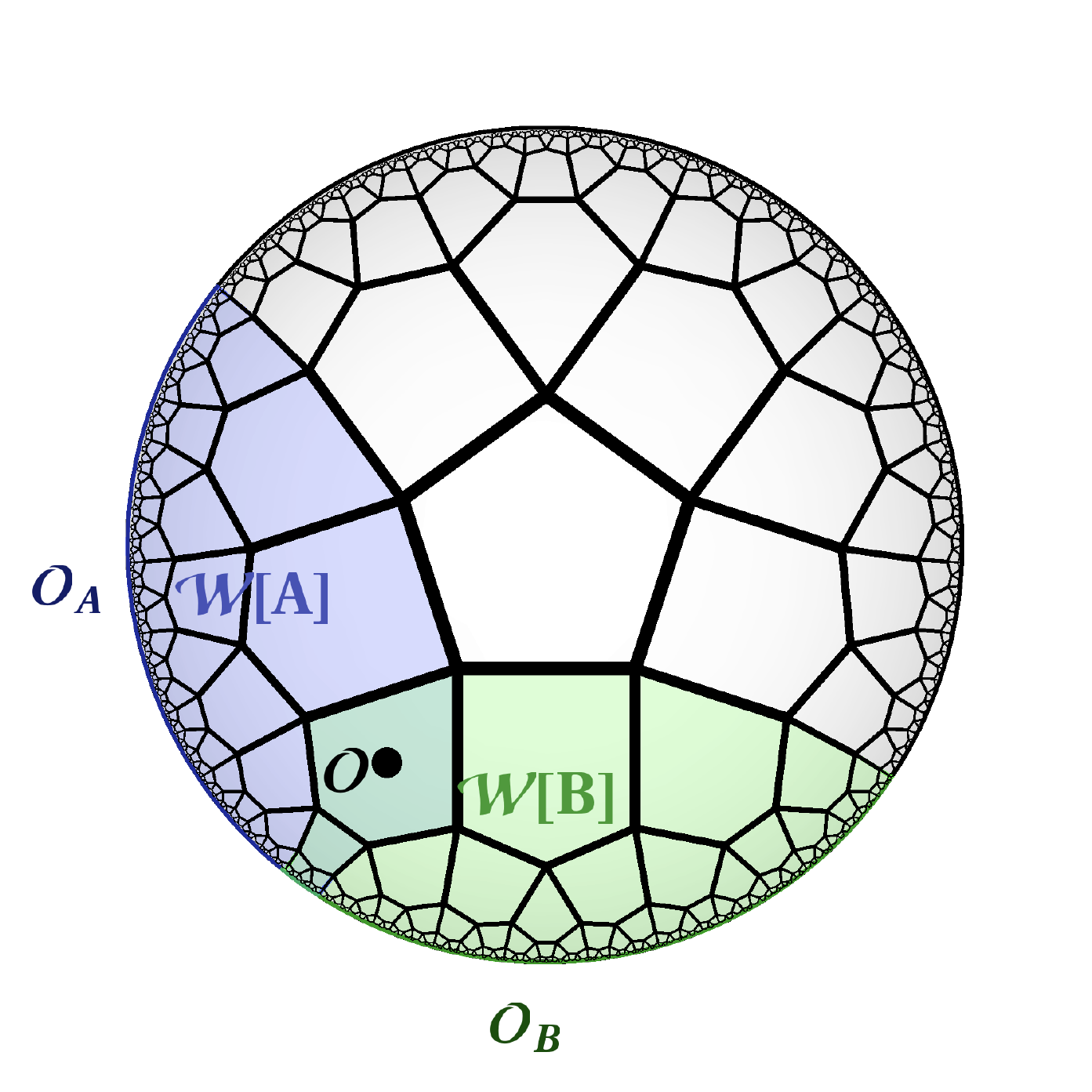}
\caption{Continuous (left) and discretized (right) reconstruction of an AdS bulk operator $\mathcal O$ along two (causal) wedges $\mathcal{W}[A]$ and $\mathcal{W}[B]$ \cite{Almheiri15}, leading to two boundary operators $\mathcal O_A$ and $\mathcal O_B$ with support on boundary regions $A$ and $B$. The AdS time slice is projected onto the Poincar\'e disk, with the AdS boundary corresponding to the black outer circle. The discretization is a $\{5,4\}$ tiling.}
\label{FIG_ADS_EE_WEDGE}
\end{figure}

In this work we show that the \emph{holographic pentagon code} implements these properties through an underlying fermionic structure.
To motivate the use of fermions in the context of holographic quantum error correction, consider a simple toy model of entangled fermionic modes.
Throughout, we denote fermionic canonically anti-commuting operators by $\fe_j$ satisfying $\fd_j \fe_k+\fe_k\fd_j = \delta_{j,k}$ and distinguish the vacuum state vector $\vacket$ satisfying $\fe_j\vacket =0$ for any $j$.
The counterpart of a Bell pair for fermions is the so-called BCS state which has the form 
\begin{align}
  \label{EQ_BCS}
  \ket{\psi_\text{BCS}} = (1+\fd_j \fd_k)\vacket \text{ .}
\end{align}
By a simple calculation, we find that 
\begin{align}
   \fe_j \ket{\psi_\text{BCS}} = \fd_k \ket{\psi_\text{BCS}} = \fd_k \vacket \text{ ,}
   \label{EQ_BCS_ACT}
\end{align}
which implies that if $j,k$ are boundary indices, we found a mapping between boundary operators that resembles \eqref{EQ_BDY_OPERATORS}.
For holographic quantum error-correction, however, this mapping is insufficient: After acting with the operator, the result \eqref{EQ_BCS_ACT} is an unentangled Fock state vector $\fd_k\vacket$, which is no longer in the desired code-space of entangled states. Furthermore, $\ket{\psi_\text{BCS}}$ does not exhibit any multi-partite entanglement necessary for holography \cite{walter2016multi}.
Fortunately, both problems can be resolved by fermionic mode fractionalization by means of \emph{Majorana dimers}. Consider the action of Majorana operators, defined as
\begin{align}
\m_{2k-1} &= \fd_k + \fe_k \text{ ,} & 
\m_{2k} &= \i\, (\fd_k - \fe_k) \text{ ,}
\end{align}
and fulfilling $\{\m_j, \m_k \} = 2\delta_{j,k}$, on the BCS state vector \eqref{EQ_BCS} as
\begin{align}
\m_{2j-1} \ket{\psi_\text{BCS}} &= -\i\,\m_{2k} \ket{\psi_\text{BCS}} &= (\fd_j + \fd_k) \vacket \text{ ,} \\
\m_{2k-1} \ket{\psi_\text{BCS}} &= \i\,\m_{2j} \ket{\psi_\text{BCS}} &= (\fd_j - \fd_k) \vacket \text{.}
\end{align}
This shows that a mapping between Majorana operators, unlike one relying on standard fermionic operators as in \eqref{EQ_BCS_ACT}, can be performed without destroying entanglement.
To achieve multi-partite entanglement, BCS-type states are insufficient. However, a suitable model is provided by the hyperbolic pentagon code (HyPeC). Let us briefly review its construction:
The HyPeC is an isometry between bulk and boundary degrees of freedom. An AdS time slice is discretized by a finite tiling of $M$ pentagons, the Poincar\'e disk projection of which is shown in Fig.~\ref{FIG_ADS_EE_WEDGE}. Each pentagon is associated with one logical qubit, i.e.\ one bulk degree of freedom, encoded in five spins (the pentagon edges) via the $[[5,1,3]]$ quantum error-correcting code. This code can be expressed by a six-leg tensor, with one ``bulk'' leg corresponding to the logical qubit and the remaining five to the physical spins.
The tiling is connected by tracing out spins on the edges of two adjacent pentagon tiles, i.e.\ by contracting the corresponding tensor indices. This contraction can be understood as a projection of the spins on the two connected edges onto a Bell pair. In this paper, we will usually consider this setup with each bulk input fixed to a certain state. Before contraction, the bulk is then effectively composed of a product state of $M$ local quantum states on five spins each. Contraction locally entangles the spins with each other, thus leading to a larger entangled state on the remaining $N$ spins at the boundary of the pentagon tiling.
If we consider instead an arbitrary bulk input on each pentagon\footnote{For the purposes of this paper, bulk inputs between different pentagons are assumed to be unentangled.}, contraction combines the local 5-spin Hilbert spaces into a larger $N$-spin Hilbert space that defines our code space $\mathcal H_\text{code}$.

By merit of the $[[5,1,3]]$ code, the five spins on the edges of each pentagon are \emph{absolutely maximally entangled}.
A pure state of $n$ qubits is absolutely maximally entangled if all of its reductions to $\lfloor{n/2}\rfloor$ subsystems are maximally mixed \cite{AMS1,AMS2,AMS3} and hence the states are maximally entangled over all such cuts. The isometric properties of the code follow from this construction.

A useful approach to understanding these states is to represent this spin picture of the HyPeC in terms of Majorana fermions \cite{Jahn:2017tls}. This is achieved by a Jordan-Wigner transformation between $L$ spins and $2L$ Majorana modes:
\begin{equation}
\begin{aligned}
\label{EQ_JORDANWIGNER}
\m_{2k-1} &= Z_1 Z_2 \dots Z_{k-1} X_k \text{ ,}\\
\m_{2k}  &=Z_1 Z_2 \dots Z_{k-1} Y_k \text{ ,}
\end{aligned}
\end{equation}
where we have used the $k$-site Pauli operators defined as
\begin{equation}
\begin{aligned}
X_k := {\id_2}^{\otimes(k-1)} \otimes \sigma_x \otimes {\id_2}^{\otimes(L-k)} \text{ ,}\\
Y_k := {\id_2}^{\otimes(k-1)} \otimes \sigma_y \otimes {\id_2}^{\otimes(L-k)} \text{ ,}\\
Z_k := {\id_2}^{\otimes(k-1)} \otimes \sigma_z \otimes {\id_2}^{\otimes(L-k)} \text{ ,}
\end{aligned}
\end{equation}
in terms of the Pauli matrices $\sigma_x,\sigma_y,\sigma_z$.
It will be useful to define the total parity operator
\begin{align}
  \Par = Z_1 Z_2 \dots Z_L = (-\i)^L\m_1\m_2\ldots\m_{2L}\ .
\end{align}
In the HyPeC, we take $L=5$ spins for each pentagon.
The logical eigenvectors $\ket{\bar{0}}$ and $\ket{\bar{1}}$ of the $[[ 5, 1, 3]]$ code have $\Par$ eigenvalues ${+}1$ and ${-}1$, respectively, corresponding to even and odd fermionic parity. For fixed bulk input (and thus parity), the stabilizers are quadratic in Majorana operators. Thus, $\ket{\bar{0}}$ and $\ket{\bar{1}}$ are ground state vectors of a Hamiltonian describing free Majorana modes, given by
\begin{align}
  H = \i \sum_{j=1}^{L=5} (P_\text{tot})^j \gamma_{j}\gamma_{j+5}
  \label{eq:H_penta}
\end{align}
where $P_\text{tot}=\pm1$ is the eigenvalue of $\Par$ and indices follow periodic boundary conditions. If we replace $P_\text{tot} \to \Par$, we recover the original $[[5,1,3]]$ stabilizer Hamiltonian with its two-fold degenerate ground state.
%
Before considering contractions of these fermionic code states, we now develop a comprehensive framework for Majorana dimers that allows us to study the fermionic HyPeC in detail.

\section{Majorana dimers}
\label{SEC_MAJ_DIM}

\subsection{Definition}

Majorana dimers are effectively a reordering of the vacuum state in terms of Majorana modes. The $L$-fermion vacuum state vector is defined by being annihilated by all of the fermionic annihilation operators $\fe_k$ for $k \in \{ 1,2,\dots , L \}$ as
\begin{align}
  \fe_k \vacket = \frac{1}{2} \left( \m_{2k-1}+\i\,\m_{2k} \right)\vacket =0 \ .
\end{align}
Thus, the vacuum state effectively relates $L$ pairs of Majorana modes $(2k{-}1,2k)$ in an operator equation. By permuting Majorana indices, we can generalize this state to any pairing of modes.
Such a \textsl{Majorana dimer state} is determined via $L$ conditions on distinct pairs $(j,k)$ (choosing $j{<}k$ as convention) of Majorana operators
\begin{equation}
\label{EQ_DIMER_COND}
\left( \m_j + \i\, p_{j,k} \m_k \right) \ket\psi = 0 \text{ .}
\end{equation}
The \textsl{dimer parities} $p_{j,k} \in \lbrace -1, 1 \rbrace$ give each pair an ``orientation'' with respect to the index ordering. We refer to $p_{j,k}=1$ as ``even'' and $p_{j,k}=-1$ as ``odd''. 
To recapitulate, a Majorana dimer state is defined to be a (normalized) state vector of $L$ fermionic modes which is annihilated by $L$ independent conditions of the form \eqref{EQ_DIMER_COND}.
Note that we have fixed a vacuum state which under the Jordan-Wigner transformation corresponds to a product state in spins, but non-trivial Majorana dimer states can be highly entangled, as we shall see.

Equivalently, we may characterize Majorana dimer state vectors $\ket\psi$ as ground states of specific quadratic Hamiltonians: Multiplying \eqref{EQ_DIMER_COND} with its Hermitian conjugate from the left yields
\begin{align}
\bra\psi \left( 2 + 2\,\i\, p_{j,k} \m_j \m_k \right) \ket\psi = 0 \text{ ,}
\end{align} 
which implies that the Hermitian operator $\i \m_j \m_k$ has expectation value $-p_{j,k}$. We can now construct the Hamiltonian
\begin{equation}
H = \frac{\i\,}{2} \sum_{(j,k)\in \Omega} p_{j,k} \m_{j} \m_{k} \text{ ,}
\end{equation}
where we sum over all $L$ Majorana dimers $\Omega=\{(j,k)\}$. $H$ is a \emph{parent Hamiltonian} of $\ket\psi$, meaning that $\ket\psi$ is the unique ground state vector of $H$ with energy $-L$, being in the ${-}1$ eigenspace of all the summands.

These two equivalent characterizations are most intuitively visualized through a diagrammatic notation.
Consider $L$ fermionic modes, ordered as a chain visualized by an $L$-gon, with the Majorana modes shown as dots on the edge (mode).
Arrows between the Majorana modes represent the pairing.
For example, for $L=5$, the state visualized by
\begin{equation}
\label{DG_PENTA_EX}
\begin{aligned}
\includegraphics[height=0.12\textheight]{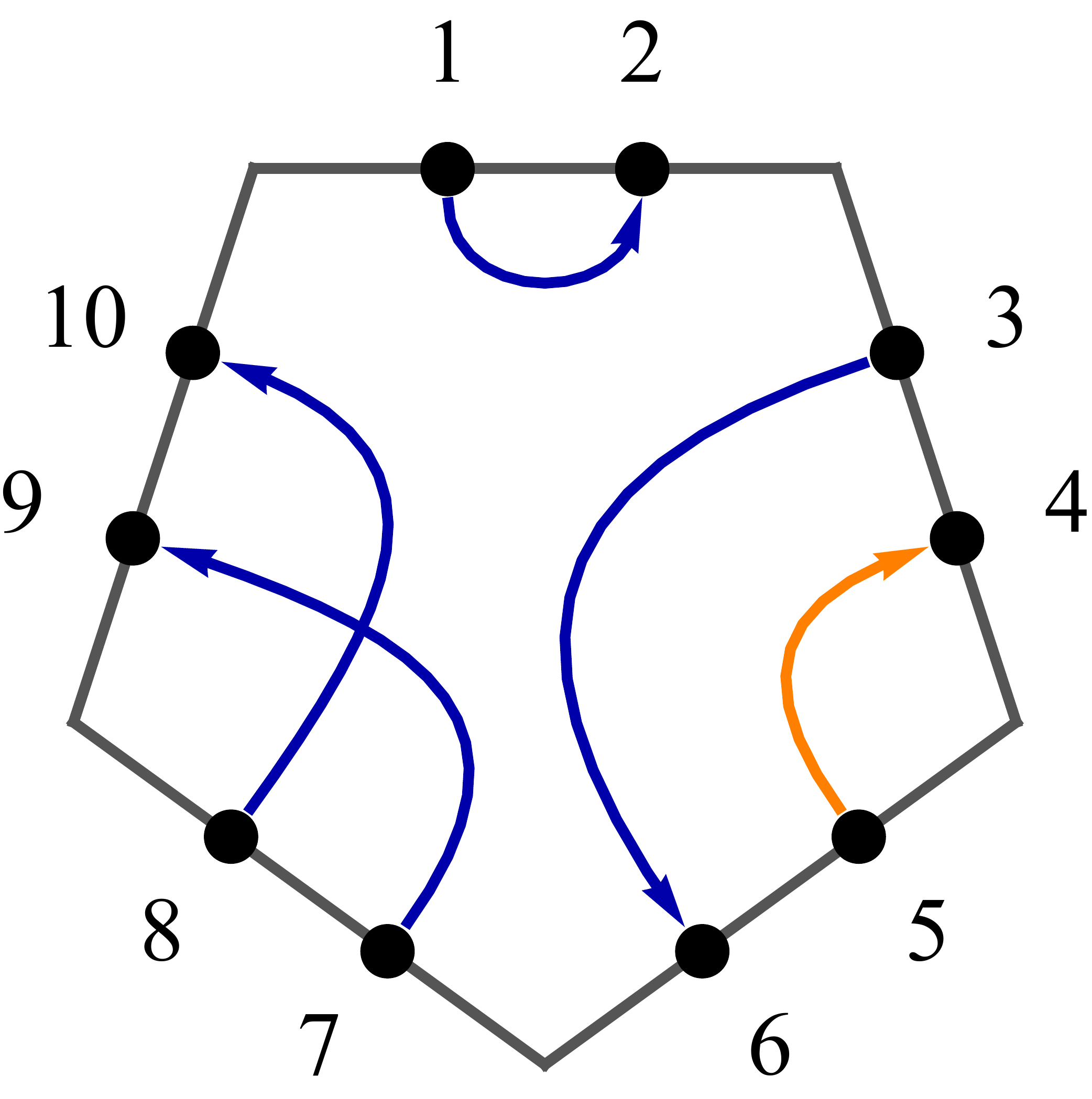}
\end{aligned}
\end{equation}
is the ground state of the Hamiltonian
\begin{equation}
H = \frac{\i\,}{2} \left( \m_1 \m_2 + \m_3 \m_6 - \m_4 \m_5 + \m_7 \m_9 + \m_8 \m_{10} \right) \text{ .}
\end{equation} 
An arrow $j \to k $ along the index orientation ($j<k$, blue) corresponds to a dimer parity $p_{j,k} = +1$, while an arrow against it ($j>k$, orange) corresponds to $p_{j,k} = -1$. Note that these diagrams only specify the state up to a scalar $c\in \mathbb{C}$, as $c$ affects neither the ground state property nor the dimer conditions \eqref{EQ_DIMER_COND}.
A particularly symmetric case is the aforementioned vacuum $\vacket$ represented by a diagram
\begin{equation}
\label{EQ_PENTA_VAC}
\begin{aligned}
\includegraphics[height=0.12\textheight]{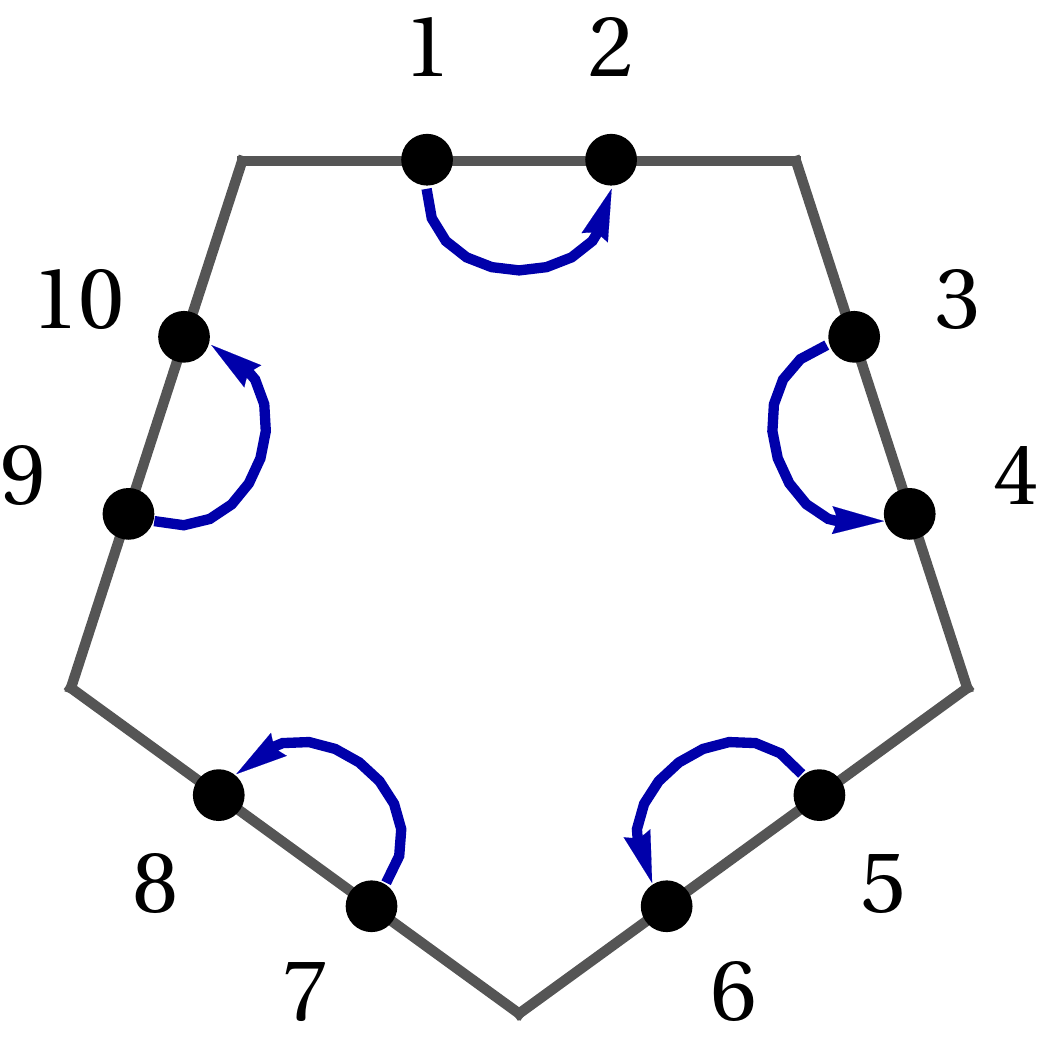}
\end{aligned}
\end{equation}
for $L=5$. Unsurprisingly, $\vacket$ is also the ground state vector of the Hamiltonian $H_0=\sum_k n_k$ with the local number operators $n_k = \fd_k \fe_k = (1 + \i \m_{2k-1} \m_{2k}) / 2$.
We can construct any Majorana dimer state from the vacuum by applying swap operators 
$S_{j,k}\coloneqq  \Par (\m_j-\m_k)/\sqrt{2}$ onto $\vacket$, $\Par$ being the total parity operator defined in the last section.
For example, the state expressed by diagram \eqref{DG_PENTA_EX} is given by $S_{8,9} S_{4,6} \vacket$.
It should be noted that while these swap operators violate the fermionic super-selection rule in an actual fermionic systems, we are merely interested in Majorana dimers as an effective representation of spins (such as the HyPeC).

As Majorana dimer states are Gaussian, all expectation values 
are determined by the entries of the \textsl{covariance matrix} with entries
\begin{equation}
	\Gamma_{j,k}^\psi = \frac{\i\,}{2}\bra\psi  [\m_j, \m_k] \ket\psi \ . 
  \label{eq:cov_def}
\end{equation}	
	We can read off $\Gamma_{j,k}^\psi$ directly from the corresponding diagram: As $\ket\psi$ is constructed from $\vacket$ by acting with a product $\mathcal{S}$ of swap operators mapping each index $k$ to an index $S(k)$, $\Gamma_{j,k}^\psi$ is simply $\Gamma_{j,k}^\emptyset$ with interchanged rows and columns
\begin{align}
\Gamma_{j,k}^\psi &= \frac{\i\,}{2} \vacbra \mathcal{S}^\dagger [\m_i, \m_j] \mathcal{S} \vacket \\
&= \frac{\i\,}{2} \vacbra [ \m_{S(i)}, \m_{S(j)} ] \vacket = \Gamma_{S(j),S(k)}^\emptyset \text{ .}
\end{align}
The only non-zero entries of the vacuum covariance matrix are $\Gamma_{2k,2k-1}^\emptyset = - \Gamma_{2k-1,2k}^\emptyset = 1$. We can thus infer $\Gamma_{j,k}^\psi$ from its diagram using the rules 
\begin{equation}
\Gamma_{j,k}^\psi = \begin{cases} 
-1 & \text{ for an arrow }j \to k \\
1 & \text{ for an arrow }k \to j \\
0 & \text{ if no arrow connects $j$ and $k$}
\end{cases} \text{ .}
\end{equation}
For example, the covariance matrix for diagram \eqref{DG_PENTA_EX} is 
\begin{equation}
\label{EQ_COVMM_EX}
\Gamma_{j,k}^\psi \;=\;
\begin{gathered}
\vspace{-5pt}
\includegraphics[width=0.25\textwidth]{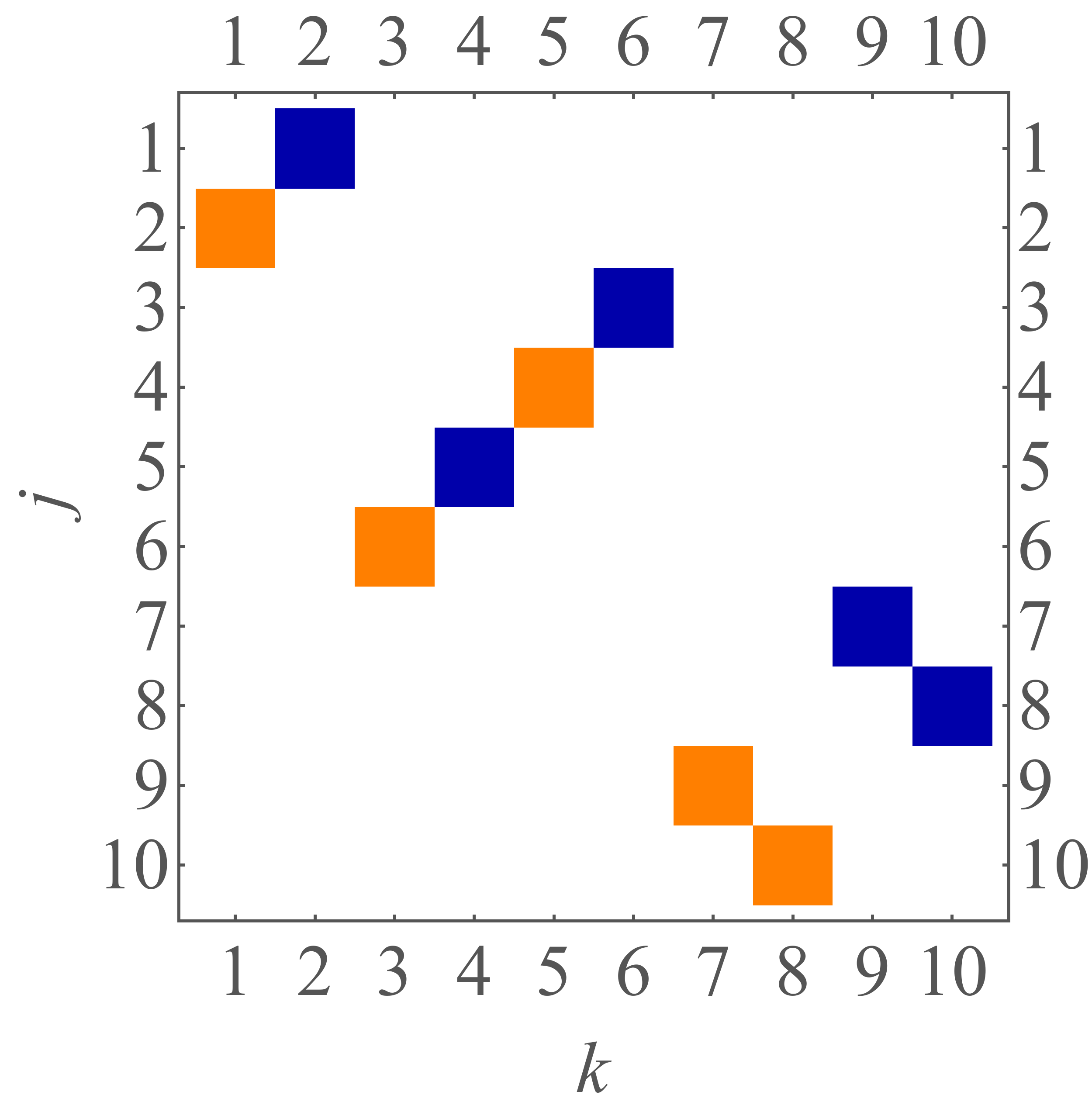}
\end{gathered}
\end{equation}
with color-coded entries (orange$={+}1$, blue$={-}1$). Note that we have chosen the colors to match with the dimer parities when reading the entries above the main diagonal ($j<k$). We assume that the state vector 
$\ket\psi$ is normalized.
Equivalently, we can think of the swap operators as acting on the Hamiltonian, yielding $H_\psi = \mathcal{S}\, H_0\, \mathcal{S}^\dagger$. Clearly, the spectrum of $H_\psi$ is simply a permutation of the spectrum of $H_0$, consistent with the covariance matrix picture.

By Eq.~\eqref{eq:H_penta}, the $[[5,1,3]]$ code states are ground states of Hamiltonians quadratic in Majorana operators, and can thus be represented as Majorana dimers.
As diagrams, they are given by
\begin{align}
\label{HAPPY_ZERO}
\ket{\bar{0}}_5\; = \quad
\begin{gathered}
\includegraphics[height=0.12\textheight]{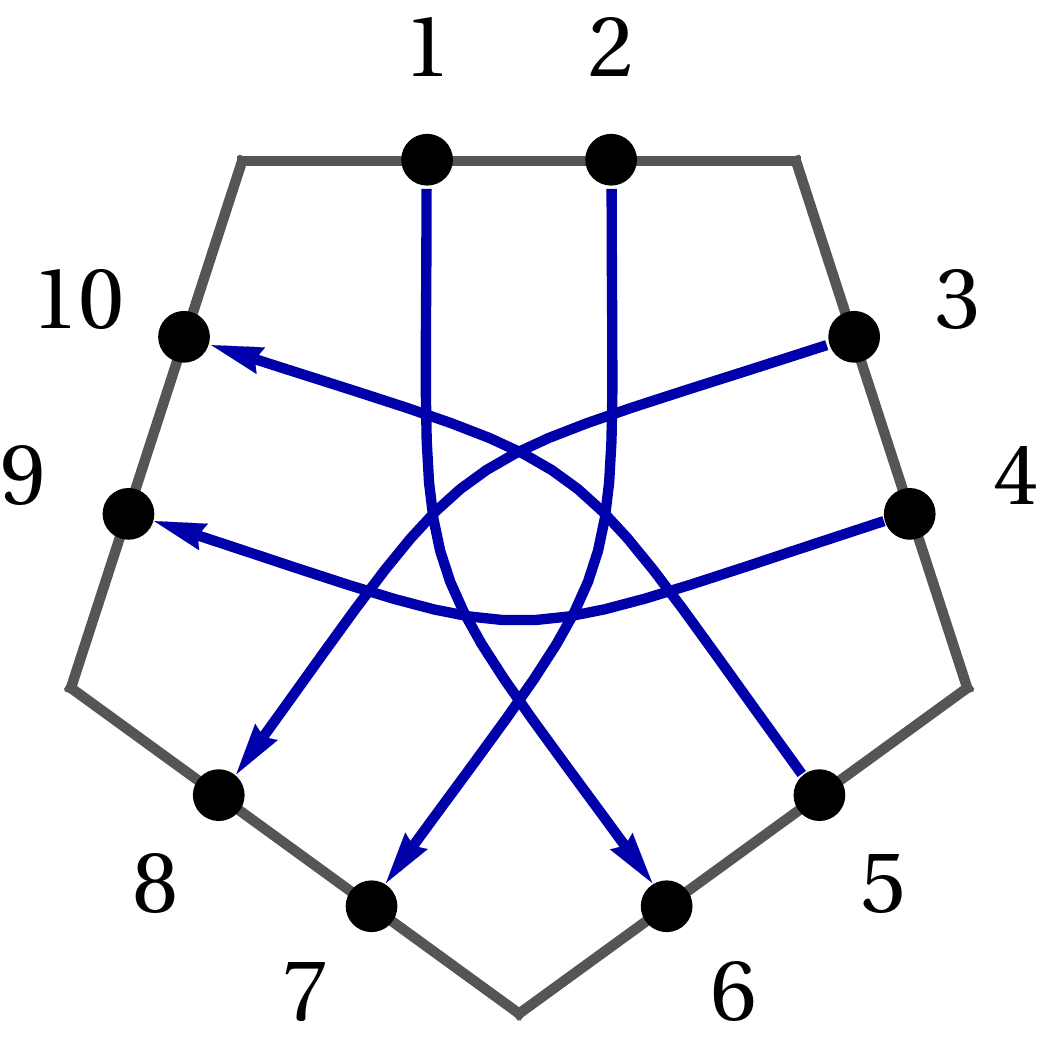}
\end{gathered} \\
\label{HAPPY_ONE}
\ket{\bar{1}}_5\; = \quad
\begin{gathered}
\includegraphics[height=0.12\textheight]{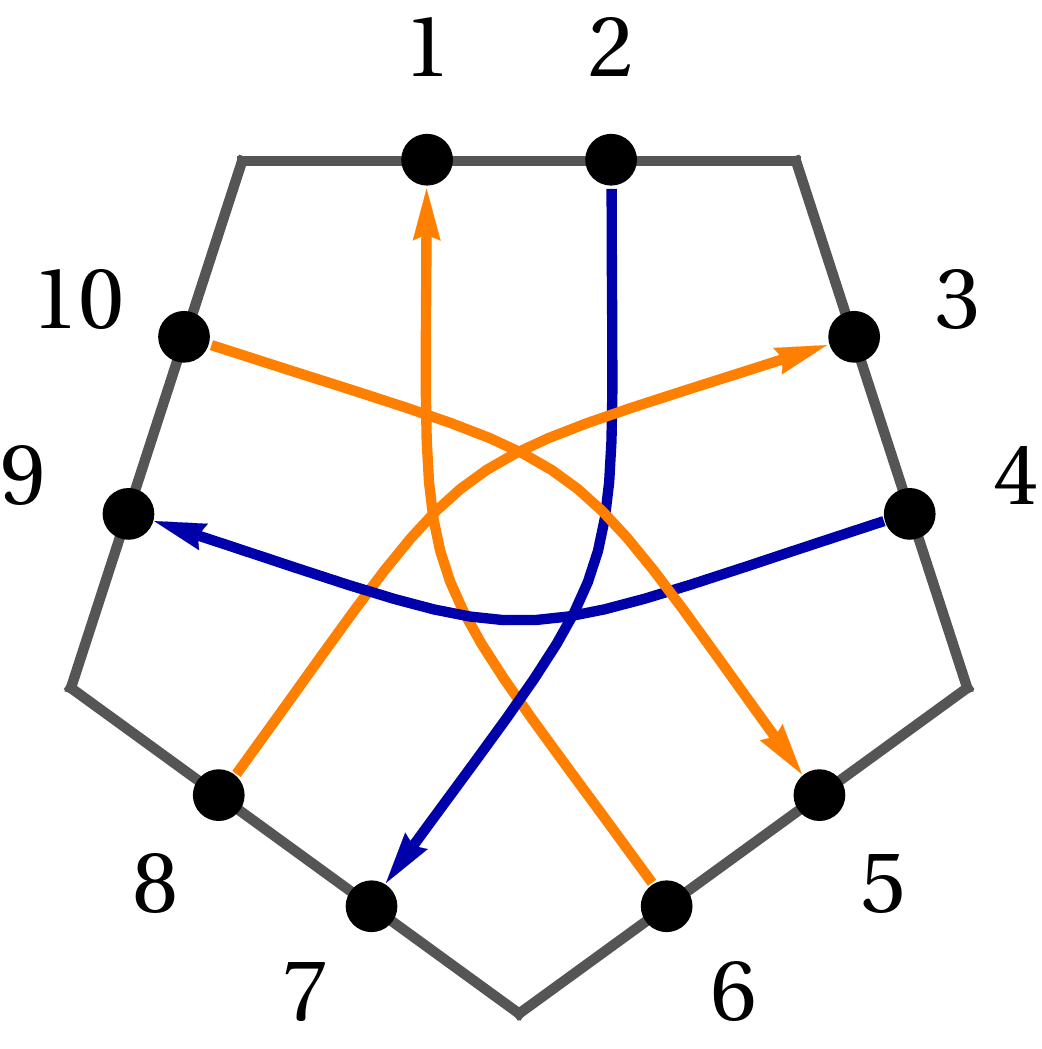}
\end{gathered} 
\end{align}
As we will see in the next section, the code distance $d=3$ between these two states in terms of Pauli operations can be shown graphically.

\subsection{Pauli operations and total parity}

As the Majorana operators are obtained from spin operators through a Jordan-Wigner transformation, local Pauli operations in the spin picture generally act non-locally on the Majorana dimers. Specifically, the reverse transformation of \eqref{EQ_JORDANWIGNER} is given by
\begin{equation}
\label{EQ_JORDANWIGNER_REV}
\begin{aligned}
X_k &= (-\i)^{k-1} \prod_{j=1}^{2k-1} \m_j \text{ ,} \\
Y_k &= (-\i)^{k-1} \left( \prod_{j=1}^{2k-2} \m_j \right) \m_{2k} \text{ ,} \\
Z_k &= -\i\, \m_{2k-1} \m_{2k} \text{ .}
\end{aligned}
\end{equation}
A Majorana operator $\m_k$ acting on a Majorana dimer state flips the parity of the dimer ending on site $k$. We show this by noting that if a state vector $\ket\psi$ is annihilated by the operator $\m_a + \i\, p \m_b$ (with dimer parity $p \in \lbrace -1,+1 \rbrace$ and $a \neq b$), then both $\m_a \ket\psi$ and $\m_b \ket\psi$ are annihilated by  $\m_a - \i\, p \m_b$:
\begin{align}
(\m_a - \i\, p \m_b) \m_a \ket\psi &= \m_a (\m_a + \i\, p \m_b) \ket\psi = 0 \text{ ,} \\
(\m_a - \i\, p \m_b) \m_b \ket\psi &= -\m_b (\m_a + \i\, p \m_b) \ket\psi = 0 \text{ .}
\end{align}
All other dimer conditions remain unaffected.
As a graphical notation, we highlight the affected edges of the state in red. 
Some examples of these operations on a Majorana dimer state vector $\ket\psi$ are shown here,
\begin{align}
X_2 \ket\psi\; = \;
\begin{gathered}
\includegraphics[height=0.09\textheight]{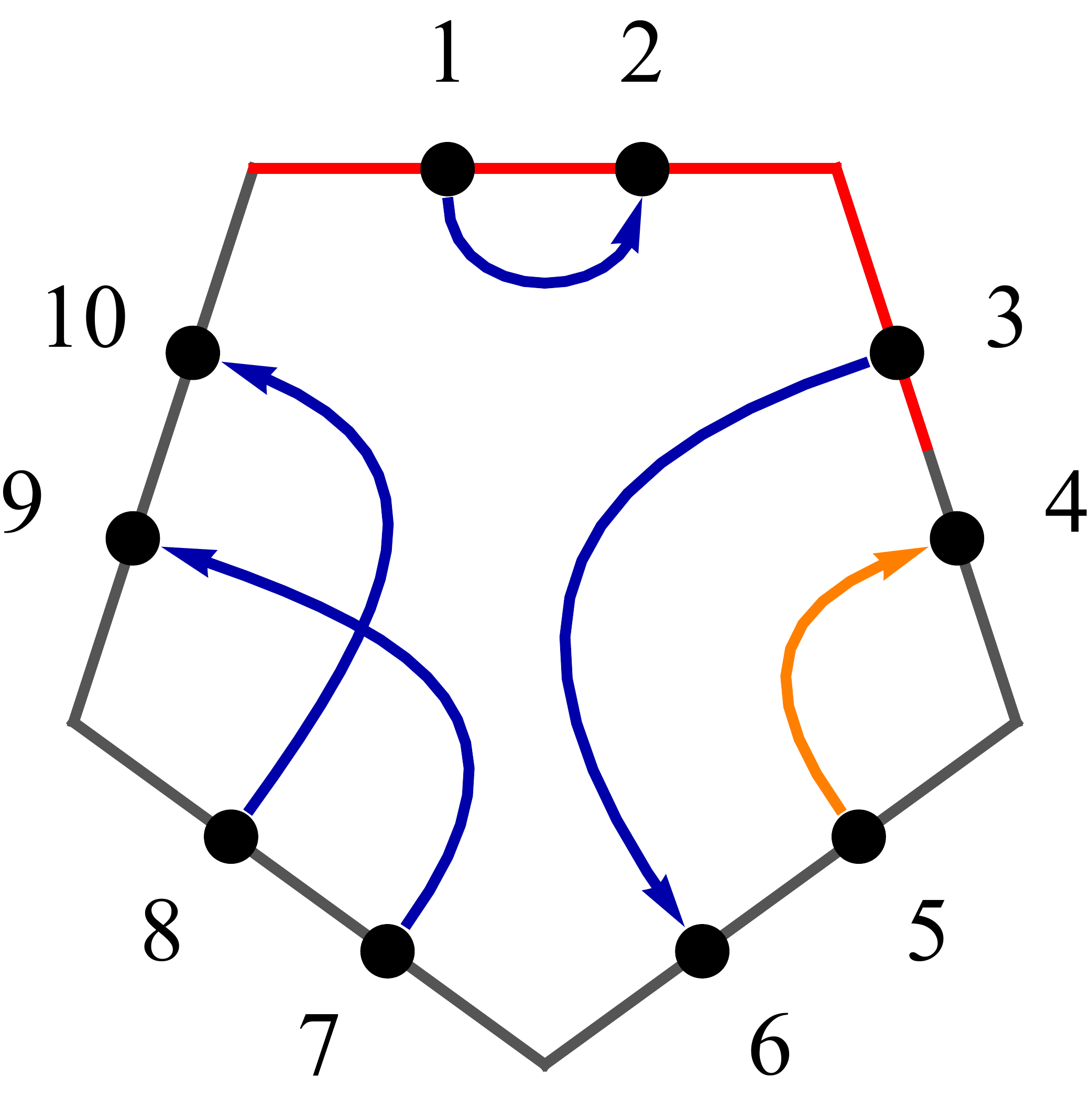}
\end{gathered}
\; = \;
\begin{gathered}
\includegraphics[height=0.09\textheight]{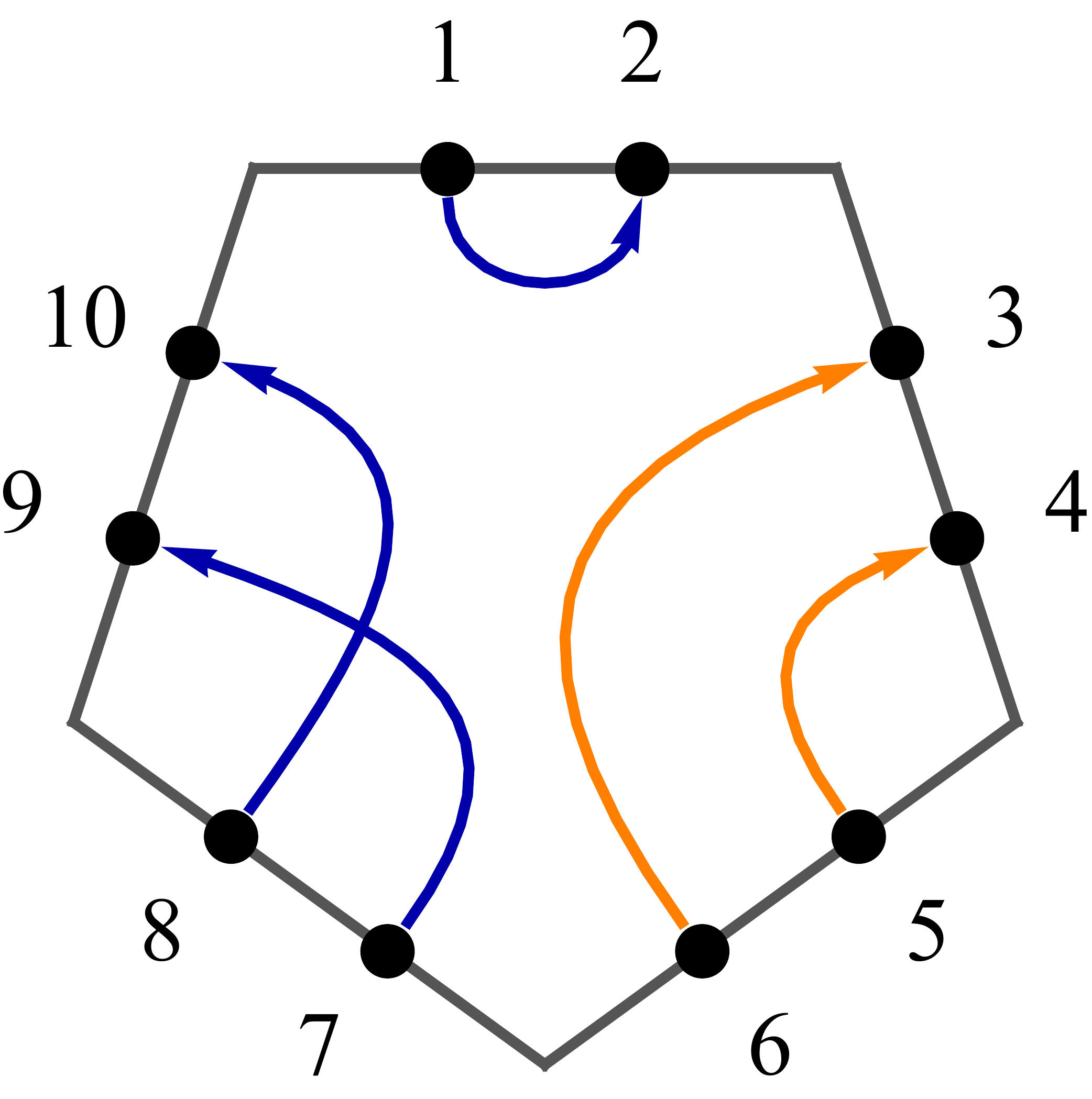}
\end{gathered} \\
Y_3 \ket\psi\; = \;
\begin{gathered}
\includegraphics[height=0.09\textheight]{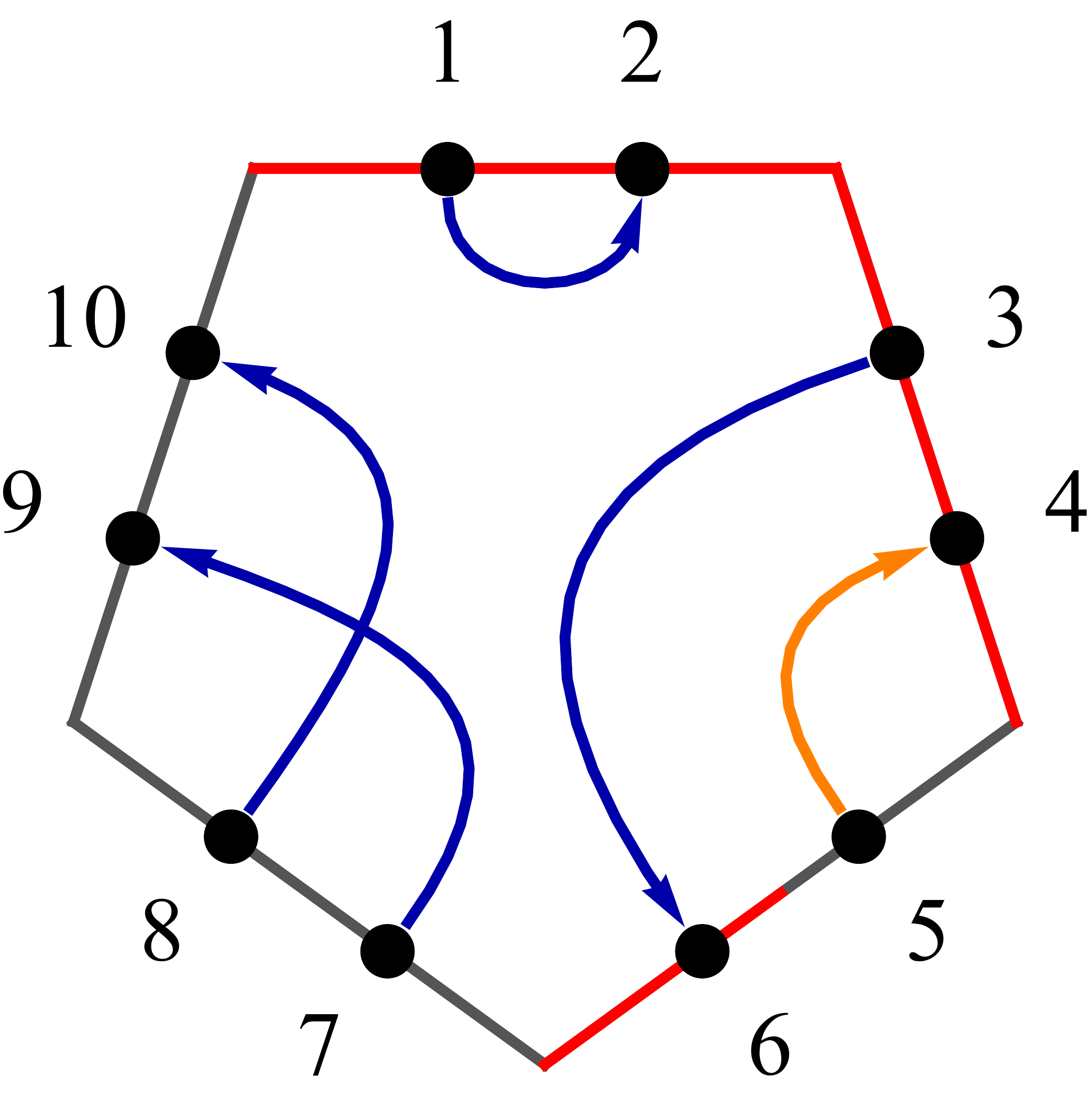}
\end{gathered}
\; = \;
\begin{gathered}
\includegraphics[height=0.09\textheight]{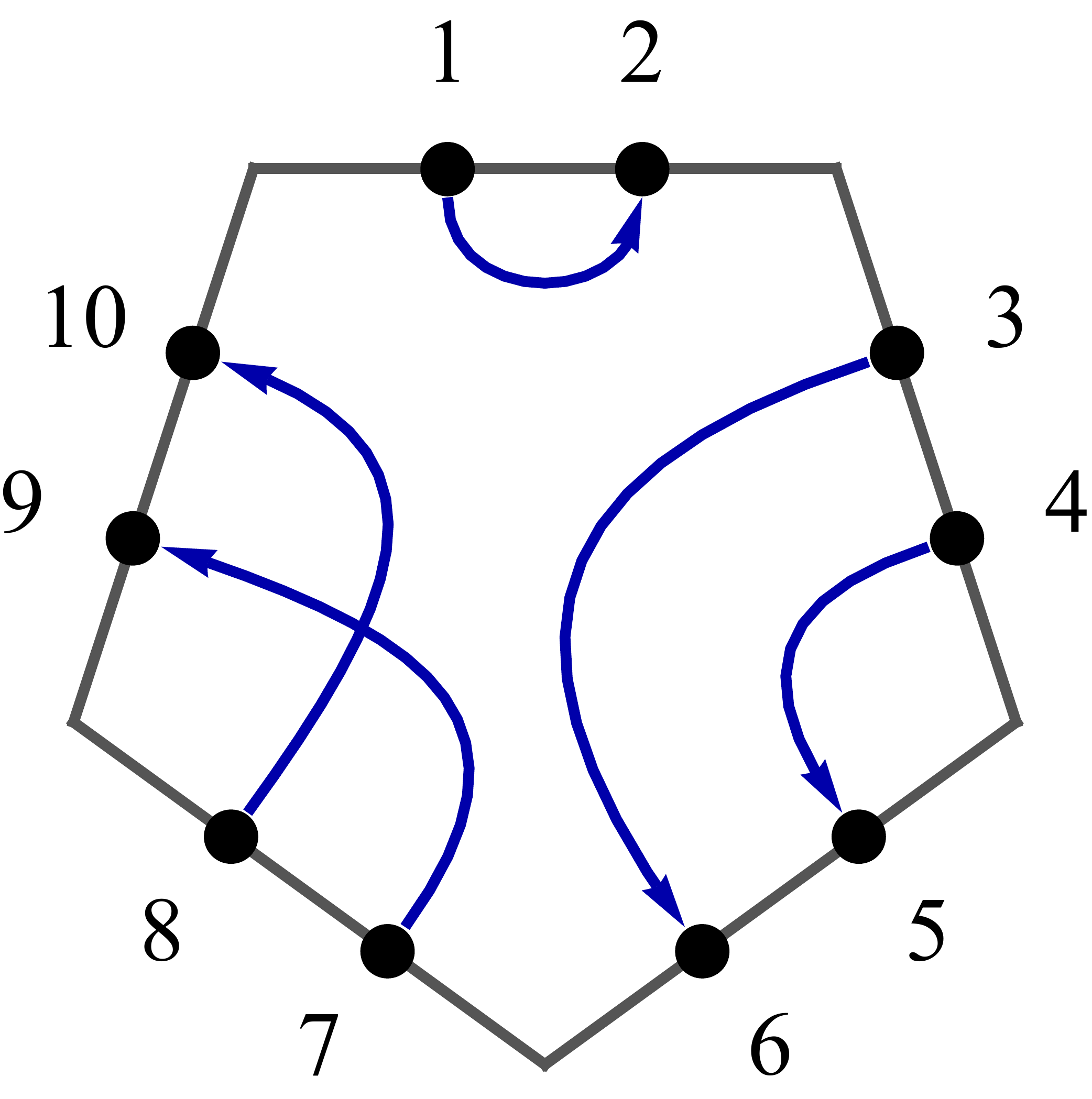}
\end{gathered} \\
Z_4 \ket\psi\; = \;
\begin{gathered}
\includegraphics[height=0.09\textheight]{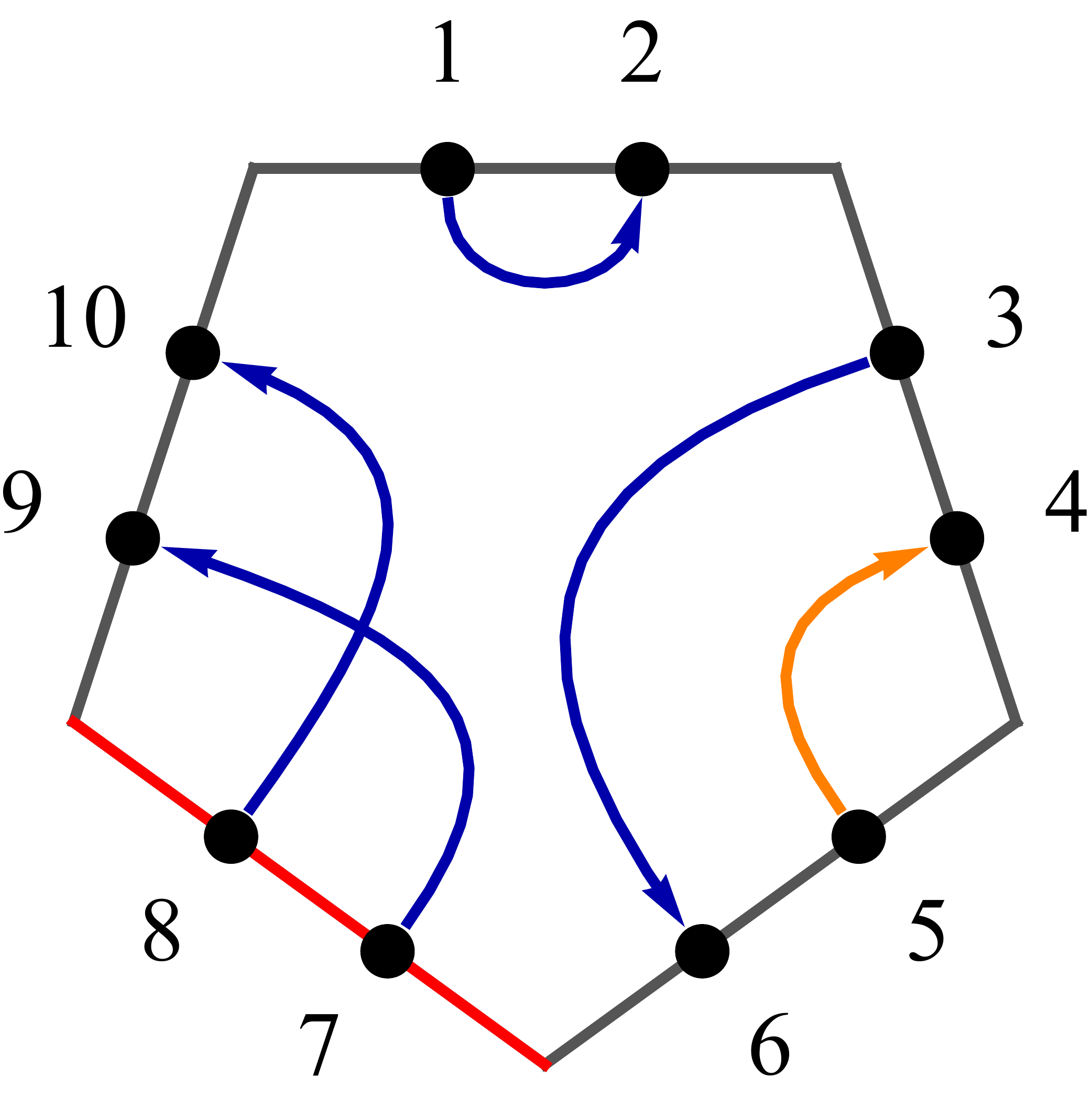}
\end{gathered}
\; = \;
\begin{gathered}
\includegraphics[height=0.09\textheight]{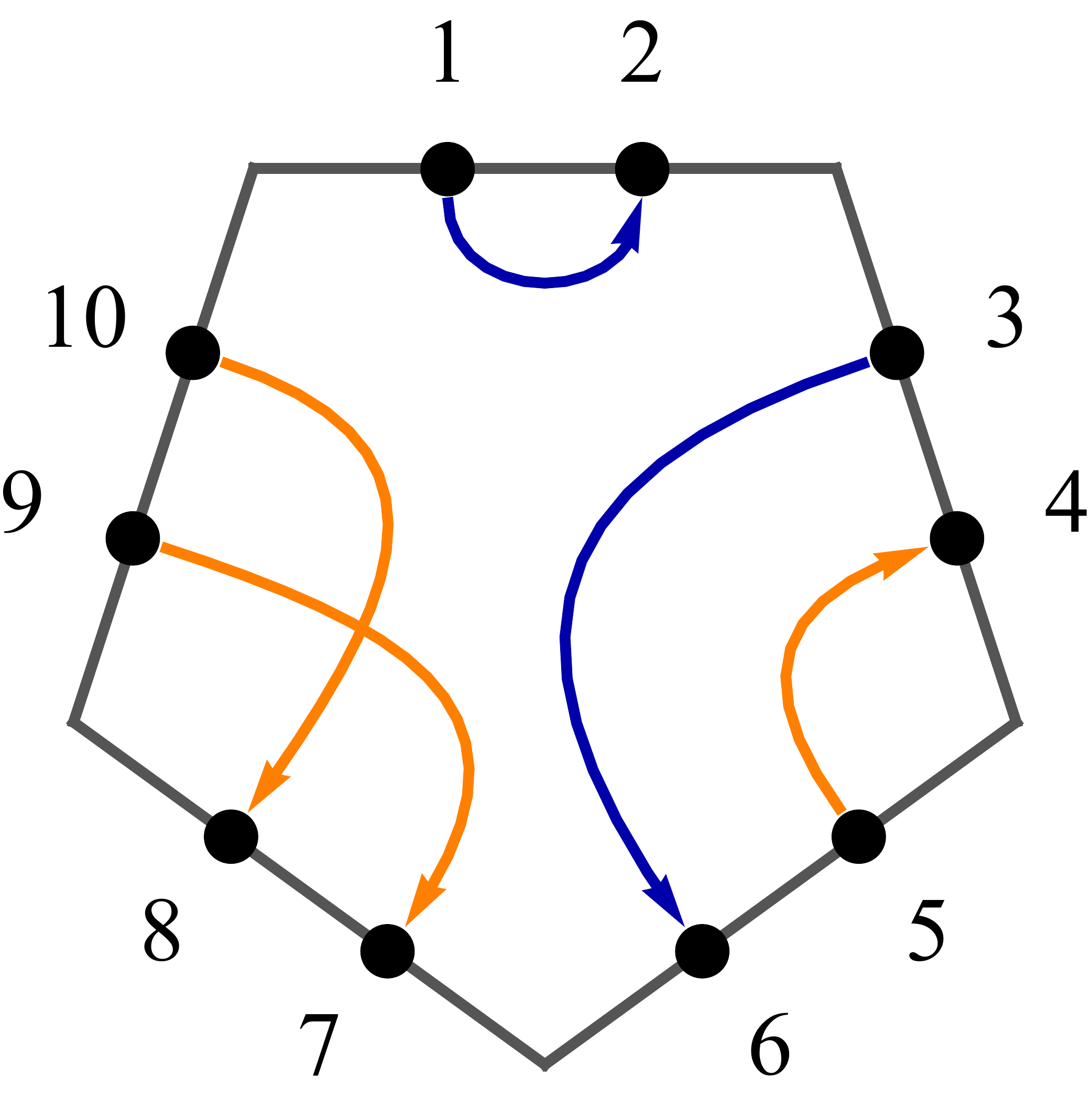}
\end{gathered} 
\end{align}
When both ends of a dimer are acted upon with a Majorana operator, the local parity stays the same.
Note that $Z_k$ operations only affect the $k$th edge, while $X_k$ and $Y_k$ combine a local Majorana operation with a $Z$ string on the first $k{-}1$ edges. Using this graphical calculus, we can now see that it requires three Pauli operations to map \eqref{HAPPY_ZERO} into \eqref{HAPPY_ONE} or each into itself. These correspond to bit flip (e.g.\ $\i \m_1 \m_3 \m_5 = X_1 Y_2 X_3$) and phase flip errors (e.g.\ $\i \m_1 \m_6 = Y_1 Z_2 Y_3$), respectively.

Now consider the total parity operator $\Par=\prod_i Z_i$, which affects all Majorana sites at once. Clearly, acting with $\Par$ leaves the state invariant (up to the parity eigenvalue), which implies that all Majorana dimer states have definite total parity $p_\text{tot}$. In fact, this parity is given by
\begin{equation}
\label{EQ_TOTAL_PARITY}
p_\text{tot} = (-1)^{N_c}\prod_{(j,k)} p_{j,k} \text{ ,}
\end{equation}
depending on the dimer parities $p_{j,k}$ of all dimers $(j,k)$ as well the number $N_c$ of crossings between dimers. This statement can be proven inductively: We start with the vacuum $\vacket$ with $p_\text{tot}=+1$. It corresponds to a diagram with $p_k=+1$ for all dimers $k$ and no crossings. We can now construct any state vector $\ket\psi$ from $\vacket$ by applying swap operators $S_{j,k}=\Par (\m_j-\m_k)/\sqrt{2}$. Since $S_{i,j}$ anticommutes with $\Par$, each swap inverts $p_\text{tot}$. To see that \eqref{EQ_TOTAL_PARITY} reflects this, consider how a swap $S_{j,k}$ affects the pairs ending in $j$ and $k$ for each possible initial configuration
as
\begin{align}
\footnotesize
\begin{gathered}
\includegraphics[height=0.08\textheight]{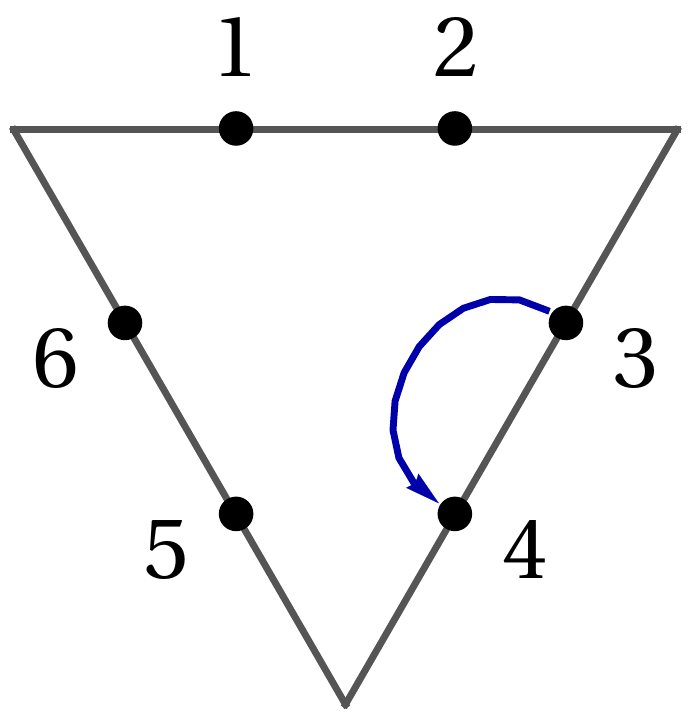}\\
S_{3,4} \bigg\updownarrow \phantom{S_{3,4}} \\
\includegraphics[height=0.08\textheight]{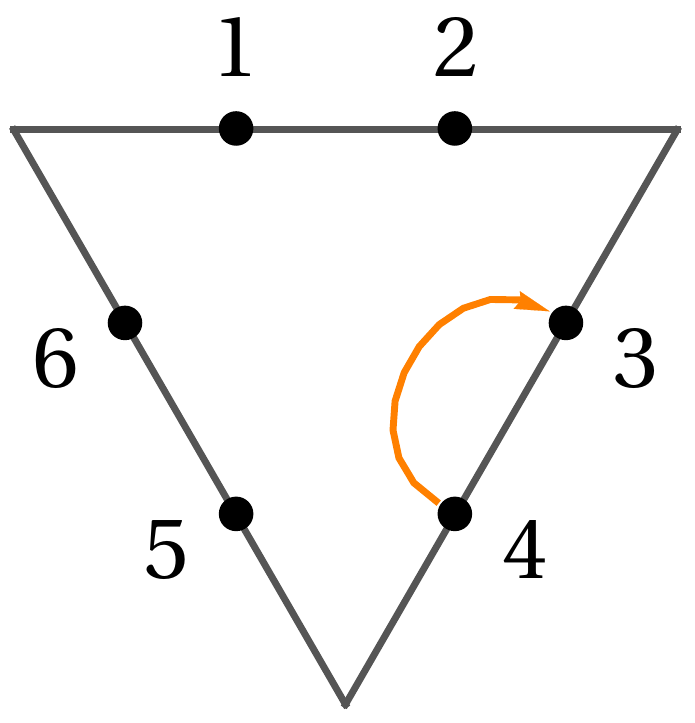}
\end{gathered}\,
\begin{gathered}
\includegraphics[height=0.08\textheight]{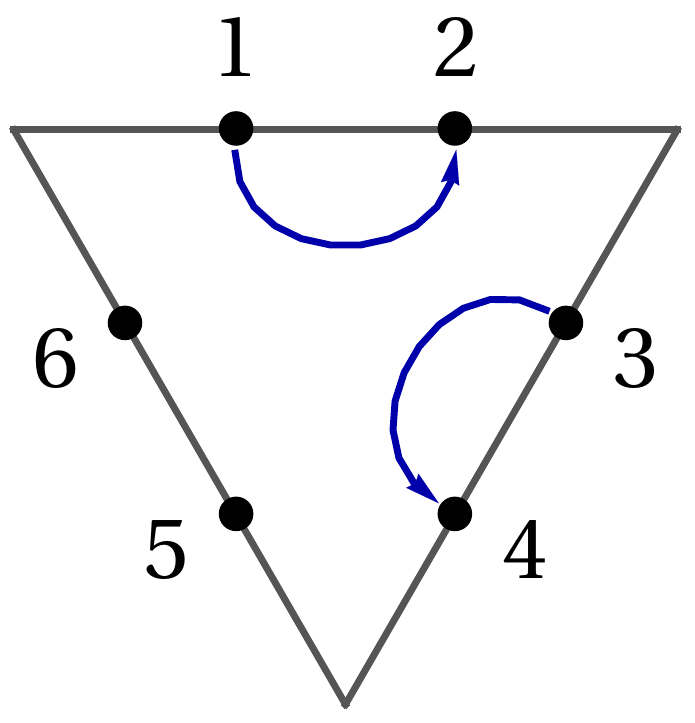}\\
S_{1,3} \bigg\updownarrow S_{2,4} \\
\includegraphics[height=0.08\textheight]{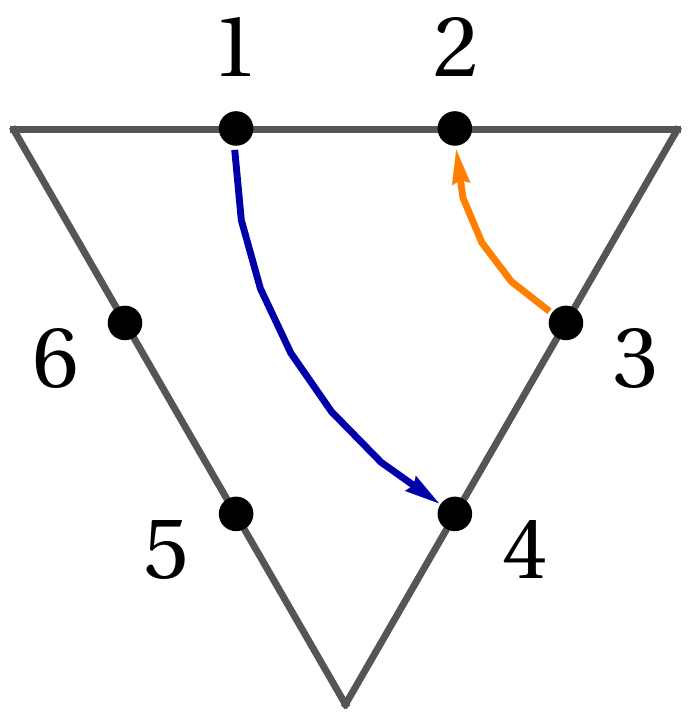}
\end{gathered}\,
\begin{gathered}
\includegraphics[height=0.08\textheight]{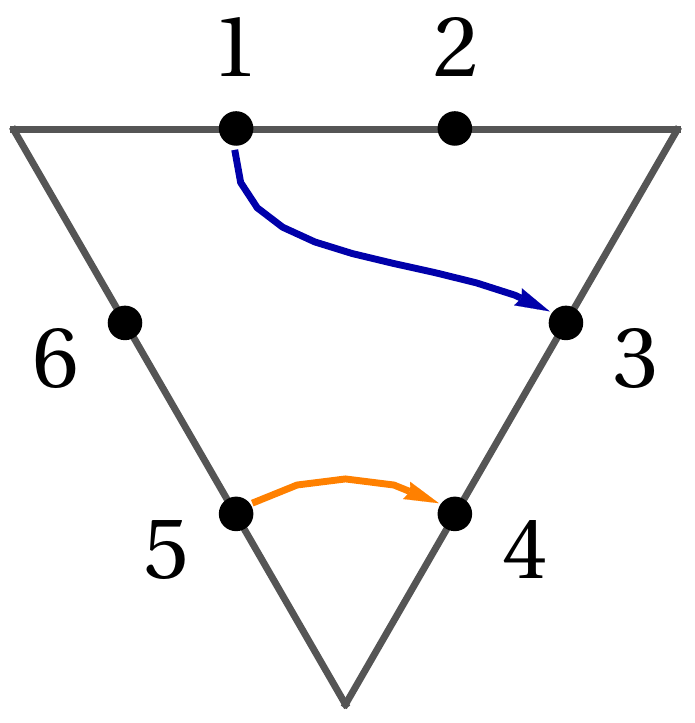}\\
S_{3,4} \bigg\updownarrow S_{1,5} \\
\includegraphics[height=0.08\textheight]{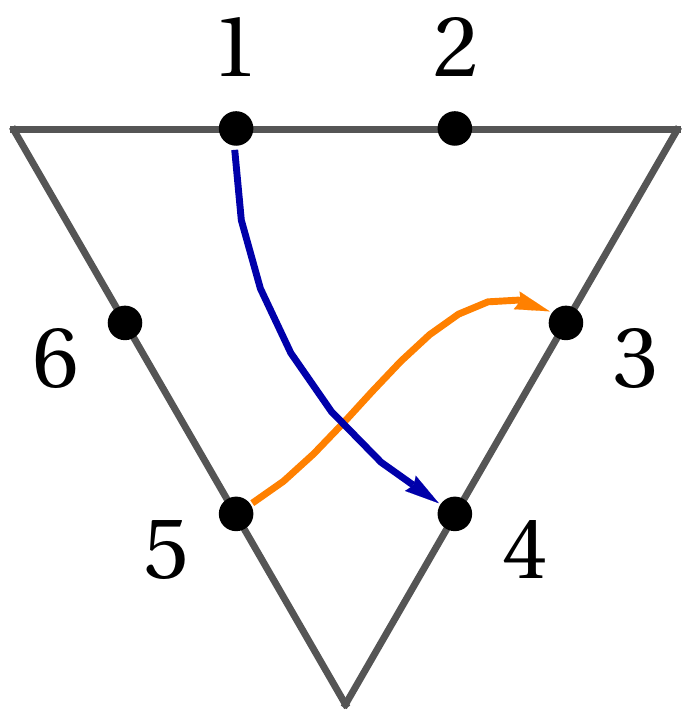}
\end{gathered}\,
\begin{gathered}
\includegraphics[height=0.08\textheight]{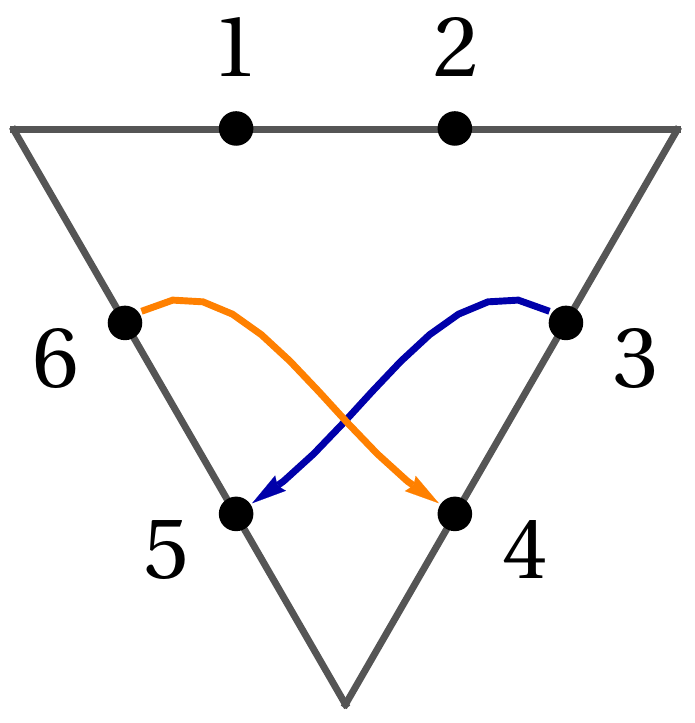}\\
S_{5,6} \bigg\updownarrow  \phantom{S_{3,4}} \\
\includegraphics[height=0.08\textheight]{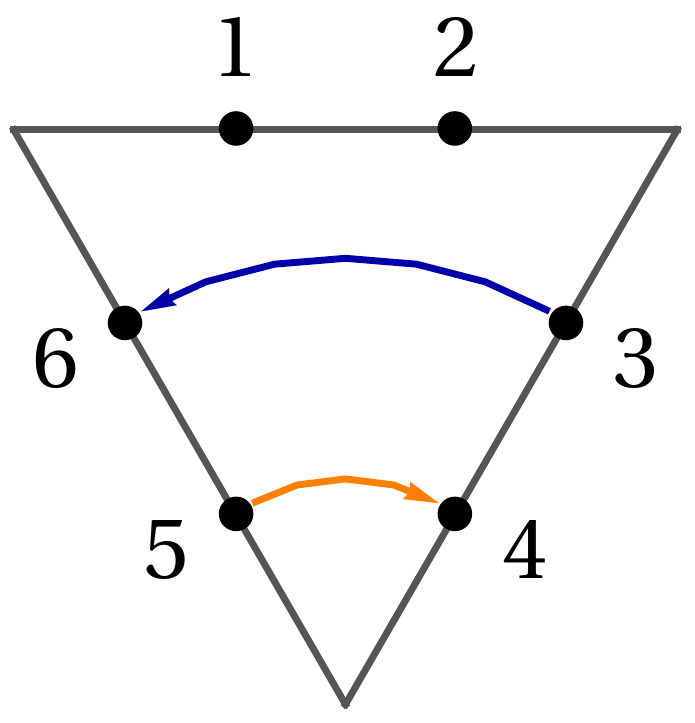}
\end{gathered}
\end{align}
Up to mirroring, relabeling and relative shifting of indices, all possible swaps belong to one the four categories shown above. The first two swaps flip one local parity but create no additional crossings; the last two either add or remove one crossing while flipping an even number of parities. Thus \eqref{EQ_TOTAL_PARITY} is always satisfied. Note that we are free to move around the dimer curve between the fixed endpoints, which means we can make two (or more) paths overlap. However, this will always change the number of crossings by an even number. For example, the logical $\bar{0}$ state of the $[[5,1,3]]$ code corresponds to both of the following diagrams (each with ten crossings):
\begin{align}
\ket{\bar{0}}_5\; = \;
\begin{gathered}
\includegraphics[height=0.09\textheight]{pentagon_net_happy_p.pdf}
\end{gathered} \;= \;
\begin{gathered}
\includegraphics[height=0.09\textheight]{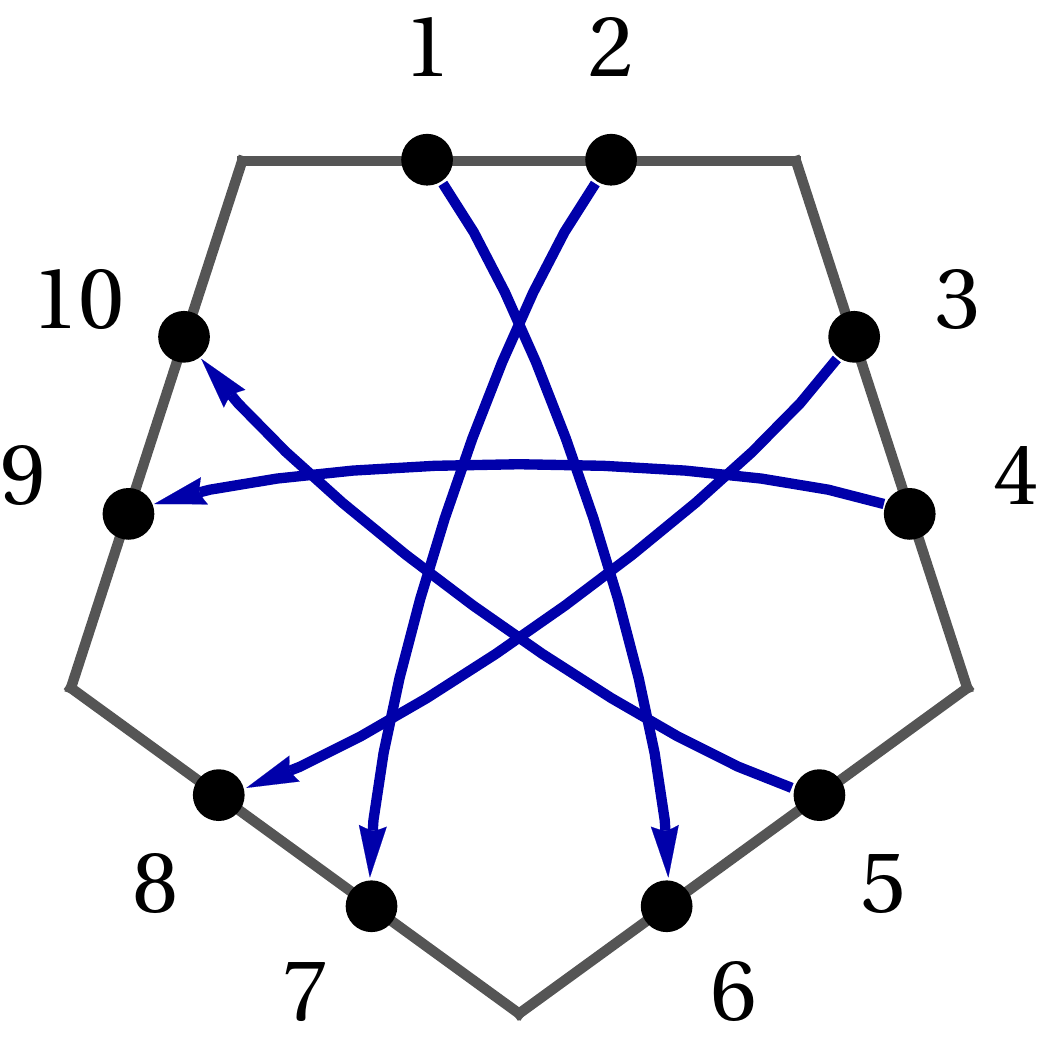}
\end{gathered} 
\end{align}
As expected, \eqref{EQ_TOTAL_PARITY} tells us that this state has positive parity.
For a fixed dimer configuration but variable dimer parities $p_k$, only the second factor of \eqref{EQ_TOTAL_PARITY} is relevant. Thus we find that acting with an $X_k$ or $Y_k$ operator, which changes an odd number of dimer parities, also flips the total parity. A $Z_k$ error, which always flips two dimer parities, leaves the total parity invariant.

\subsection{Contracting dimers}
\label{SEC_T_CONTR}
We will now show how the notion of tensor network contraction applies to Majorana dimer states.
To begin with, consider a state of $N$ spins 
\begin{align}
  \ket {\psi_T} = \sum_{j\in\{0,1\}^{\times N}} T_{j_1,\ldots,j_N} \ket{j_1,\ldots,j_N}\ .
\end{align}
Here the amplitudes $T_{j_1,\ldots,j_N} = \braket{j_1,\ldots,j_N}{\psi}$ can be viewed as a tensor $T$ which fully specifies the state vector $\ket {\psi_T}$.
A tensor network is a means of specifying a tensor  describing a state of a large number of spins through multiple contractions of tensors of a smaller rank.
Specifically, the contraction of two tensors $S$ and $T$ of rank $N_S$ and $N_T$
between the last index of $S$ and the first index of $T$ is defined to be a new tensor $U$ of rank $N_U=N_S+N_T-2$, with entries
\begin{align}
\label{EQ_TENSOR_CONTRACTION}
U_{k_1,k_2,\dots,k_{N_U}} &= S_{k_1,k_2,\dots,k_{N_S-1},0}\, T_{0, k_{N_S},k_{N_S+1},\dots,k_{N_U}} \nonumber \\
&+ S_{k_1,k_2,\dots,k_{N_S-1},1}\, T_{1, k_{N_S},k_{N_S+1},\dots,k_{N_U}} \text{ .}
\end{align}
We see that by contracting the respective tensors, this operation allows us to merge two state vectors $\ket{\psi_S}$ and $\ket{\psi_T}$ into a larger one $\ket{\psi_U}$ given by
\begin{equation}
\ket{\psi_U} = \sum_{k\in \BC {N^\prime}} U_{k_1,k_2,\dots,k_{N^\prime}} \ket{k_1,\ldots, k_{N'}}\text{ .}
\end{equation}
A tensor network state can thus describe a large state by the relatively few entries of its contracted tensors. This process can be generalized to fermions by identifying the spin basis with a fermionic one as
\begin{equation}
\label{EQ_SPIN_FERMION_BASES}
\ket{j_1,\ldots,j_N} \leftrightarrow (\fd_1)^{j_1} (\fd_2)^{j_2} \dots (\fd_N)^{j_N} \vacket \text{ .}
\end{equation}
In this picture, tensors are associated with pure fermionic states. As these expressions only use creation operators, they obey a \emph{Grassmann algebra}. The tensor contraction \eqref{EQ_TENSOR_CONTRACTION} can then be expressed by a Grassmann integration over fermionic degrees of freedom \cite{Bravyi2008}.
Specifically, a contraction of two fermionic state vectors $\ket\phi$ and $\ket\psi$ into a state vector $\ket\omega$ over the same indices as in \eqref{EQ_TENSOR_CONTRACTION} has the form
\begin{align}
\ket\omega &= \int \text{d}\fd_{M+1} \text{d}\fd_M\, (1 + \fd_M \fd_{M+1}) \ket\phi \ket\psi \\
&= \int \text{d}\fd_{M+1} \text{d}\fd_M\, e^{\,\fd_M \fd_{M+1}} \ket\phi \ket\psi \text{ ,}
\end{align} 
where we have used the Grassmann integration $\int \text{d}\fd_k \,{\fd_k}^n = \delta_{n,1}$ (for more information, see Refs.\ \cite{Berezin, cahill1999density,Bravyi2008, bravyi2004lagrangian}). Note that $\int \text{d}\fd_k$ acts like an annihilator $\fe_k$ on a fermionic state, with a subsequent projection onto the fermionic subspace excluding the $k$th mode. This requires a relabeling of the remaining degrees of freedom and a truncation of the Jordan-Wigner string in the corresponding spin representation.

We can now apply this machinery to Majorana dimer states. In our graphical language, tensor contraction is equivalent to connecting two polygon edges and integrating out the four Majorana modes on them. What happens to the dimers of the original states? It is easy to see that dimers $(j,k)$ of a state vector $\ket\phi$ not connected to the contracted edges remain dimers, i.e.,\ if $( \m_j + \i\, p_{j,k} \m_k ) \ket\phi$ vanishes, we also find
\begin{align}
&\left( \m_j + \i\, p_{j,k} \m_k \right)  \int \text{d}\fd_{M+1} \text{d}\fd_M\, e^{\,\fd_M \fd_{M+1}} \ket\phi \ket\psi \nonumber \\
&= \int \text{d}\fd_{M+1} \text{d}\fd_M\, e^{\,\fd_M \fd_{M+1}} \left( \m_j + \i\, p_{j,k} \m_k \right)  \ket\phi \ket\psi = 0 \text{ ,}
\end{align}
as $\m_j$ and $\m_k$ commute with the integration. We now claim that the dimers connected to the contracted edge become new dimers of the contracted state $\omega$. This leads to the following statement.

\begin{theorem}[Contractions of Majorana dimer states]
\label{THM_MS_CLOSED}
The contraction of two Majorana dimer state vectors $\ket\phi$ and $\ket\psi$ yields either a new Majorana dimer state vector $\ket\omega$ or zero.
\end{theorem}
An example for the contraction of two pentagon state vectors $\ket\phi$ and $\ket\psi$ is given by 
\begin{equation}
\label{EQ_CONTR_EX1}
\begin{gathered}
\includegraphics[height=0.12\textheight]{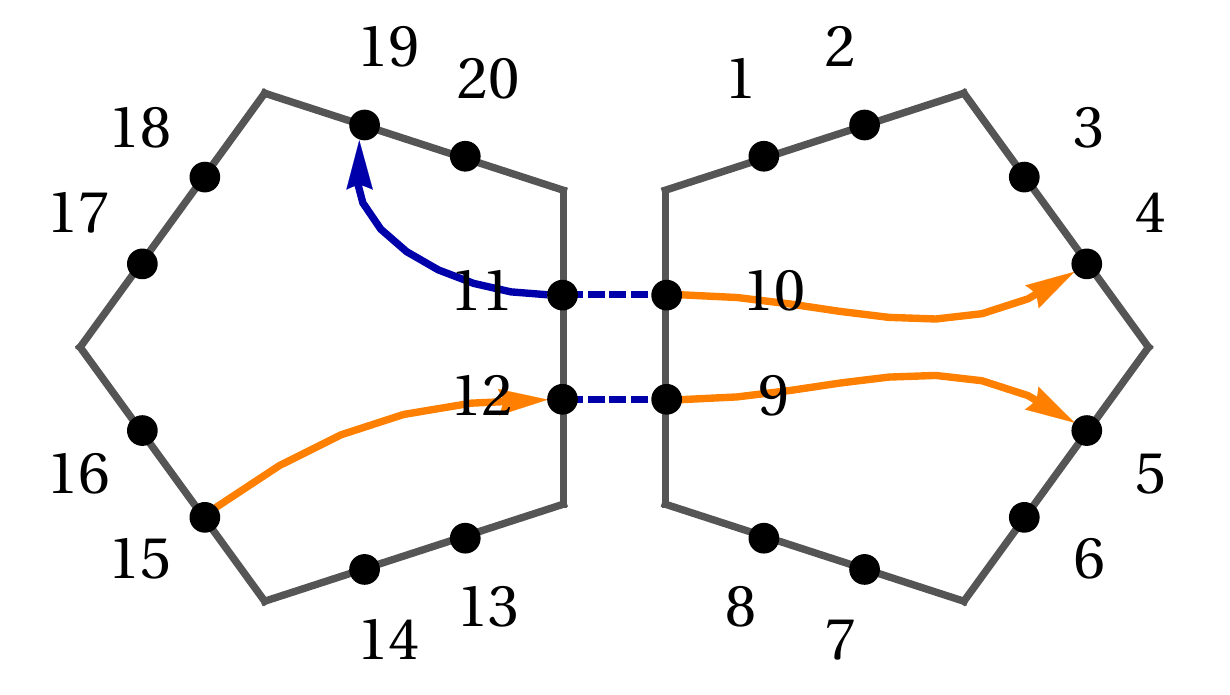}\\
\scalebox{1.2}{$\parallel$}\\
\includegraphics[height=0.12\textheight]{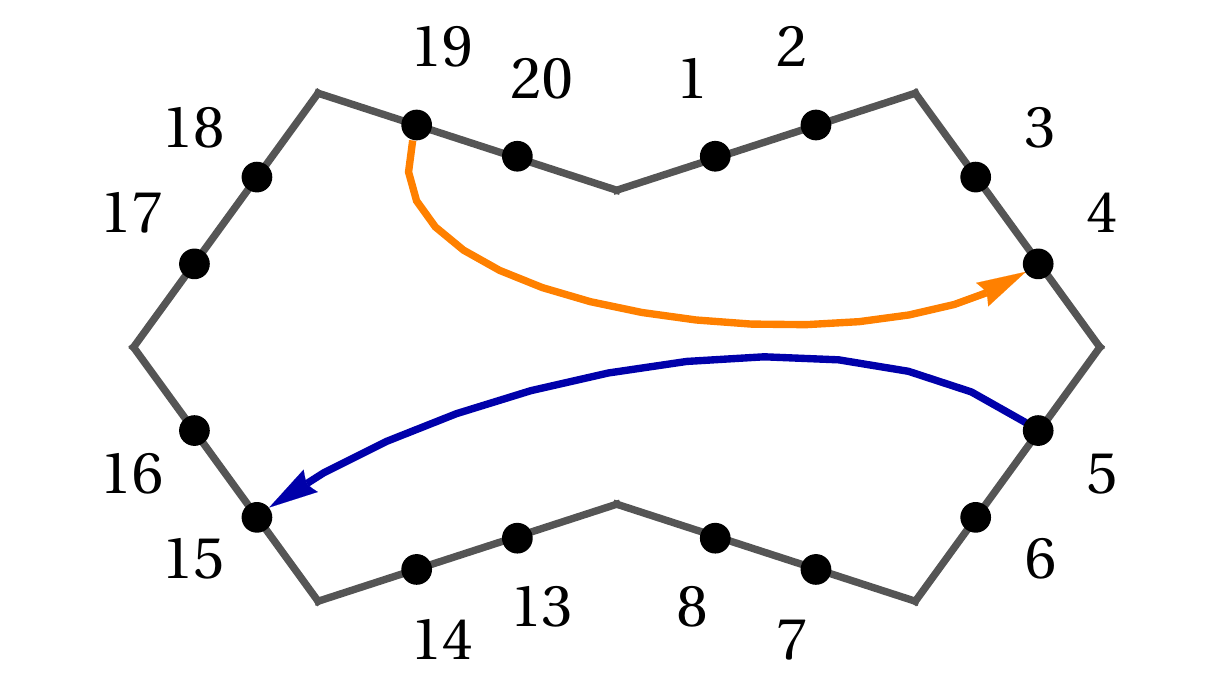}
\end{gathered}
\end{equation}
We have visualized the contraction by a pair of dashed lines. In this example, dimers not connected to the contracted edges are omitted. The upper diagram corresponds to the conditions
\begin{align}
\label{EQ_CONTR_EX_COND1}
( \m_4 - \i \m_{10} ) \ket\phi &= 0 \text{ ,}
& 
( \m_{11} + \i \m_{19} ) \ket\psi &= 0 \text{ ,} \\
\label{EQ_CONTR_EX_COND2}
( \m_5 - \i \m_9 ) \ket\phi &= 0  \text{ ,}
& 
( \m_{12} - \i \m_{15} ) \ket\psi &= 0  \text{ .}
\end{align}
We now prove that \eqref{EQ_CONTR_EX_COND1} implies $( \m_4 - \i \m_{19} ) \ket\omega = 0$ for the contracted state vector $\ket\omega$, i.e.,\ that the two original dimers fuse into a larger one:
\begin{align}
&( \m_4 - \i \m_{19} ) \ket\omega= \int \text{d}\fd_6 \text{d}\fd_5\, e^{\,\fd_5 \fd_6} ( \m_4 - \i \m_{19} ) \ket\psi \ket\phi \nonumber\\
&= \int \text{d}\fd_6 \text{d}\fd_5\, e^{\,\fd_5 \fd_6} ( \i \m_{10} + \m_{11} ) \ket\psi \ket\phi \nonumber\\
&= \int \text{d}\fd_6 \text{d}\fd_5\, e^{\,\fd_5 \fd_6} ( - \fd_5 + \fd_6 + \fe_5 + \fe_6 ) \ket\psi \ket\phi \nonumber\\
&= \int \text{d}\fd_6 \text{d}\fd_5\, ( - \fd_5 + \fd_6 - \fd_6 \fd_5 \fe_5  + \fd_5 \fd_6 \fe_6 ) \ket\psi \ket\phi \nonumber\\
&= 0 \text{ ,}
\end{align}
where we have used the identities $\int \text{d}\fd_k \fe_k = 0$ and $\lbrace \fe_k, \fd_l \rbrace  = \delta_{k,l}$. A similar proof using \eqref{EQ_CONTR_EX_COND2} leads to $( \m_5 + \i \m_{15} ) \ket\omega = 0$.

The full proof for all possible dimer contractions is given in Appendix \ref{APP_CONTR_RULES}. The resulting rules are:
\begin{itemize}
\item Contracting neighbouring edges $k$ and $k{+}1$ removes the Majorana modes $\lbrace 2k-1, 2k, 2k+1, 2k+2 \rbrace$. The dimers ending on $2k-1$ and $2k+2$ as well as $2k$ and $2k+1$ are fused into larger dimers. 
\item The dimer parity $p_{j,k}$ of a fused dimer is the product of parities of the original dimers. In addition, every crossing of the path of a contracted dimer with itself reverses $p_{j,k}$.
\item Every contraction that creates closed loops leads to a vanishing contracted state if at least one loop has an odd dimer parity.
\end{itemize}
The last case refers to diagrams such as the following:
\begin{align}
\label{EQ_CONTR_EX2}
\begin{gathered}
\includegraphics[height=0.12\textheight]{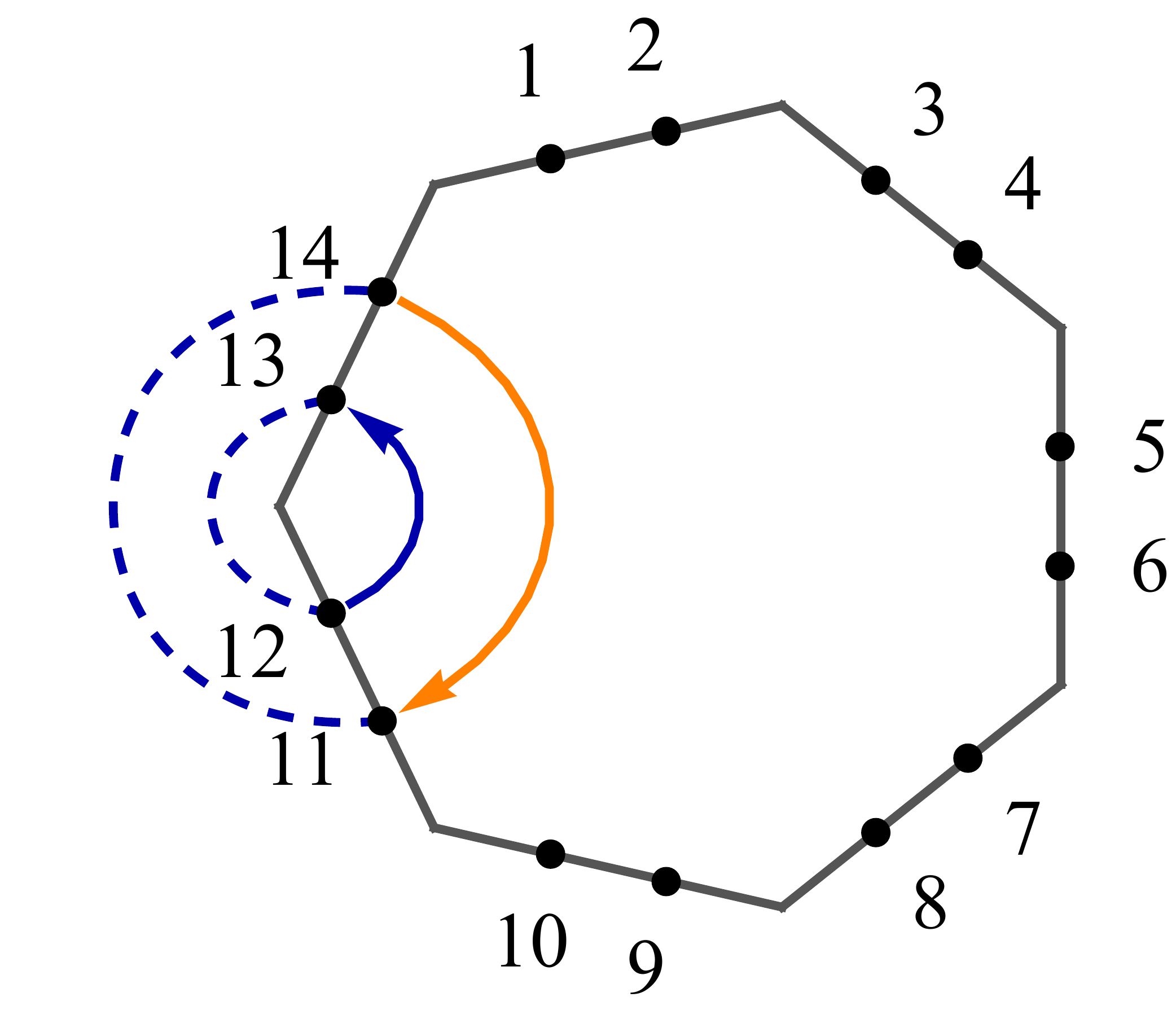}
\end{gathered} \quad = \; 0 \text{  .}
\end{align}
Loops with even total dimer parity only change the contracted state by a non-zero constant.

\subsection{Ordering and cyclic permutations}
\label{APP_CYCL_PERM}

As fermionic operators anti-commute, defining a chain of fermionic sites requires a definite ordering of the site indices. This is also true for the Majorana modes that make up Majorana dimer states.
For convention, we assume clockwise-oriented indices starting from an initial index which we call the \textsl{pivot} and mark by a little circle in the following diagrams. Shifting the pivot corresponds to a cyclic permutation of all indices:
\begin{equation}
\label{EQ_CYCL_PERM}
\begin{gathered}
\includegraphics[height=0.1\textheight]{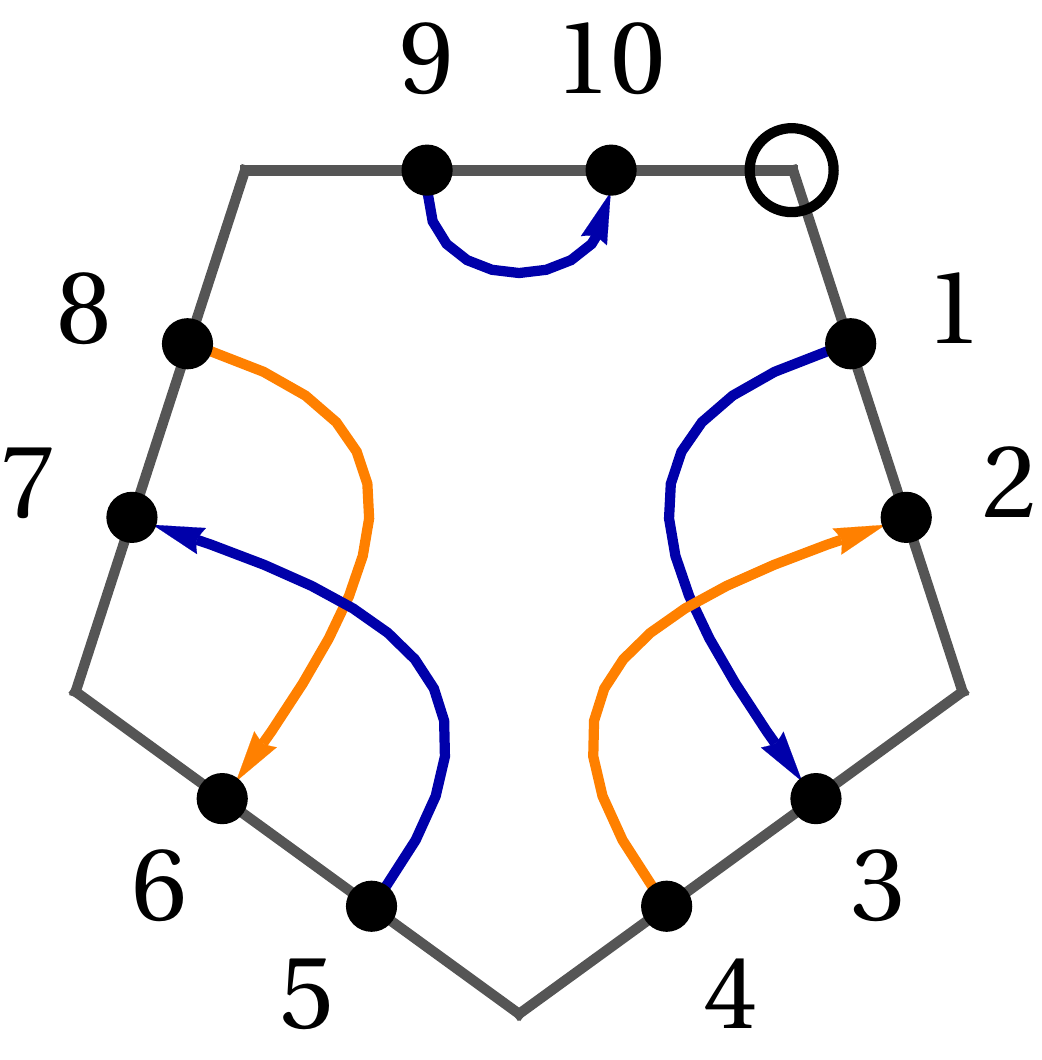}
\end{gathered}
\quad\rightarrow\quad
\begin{gathered}
\includegraphics[height=0.1\textheight]{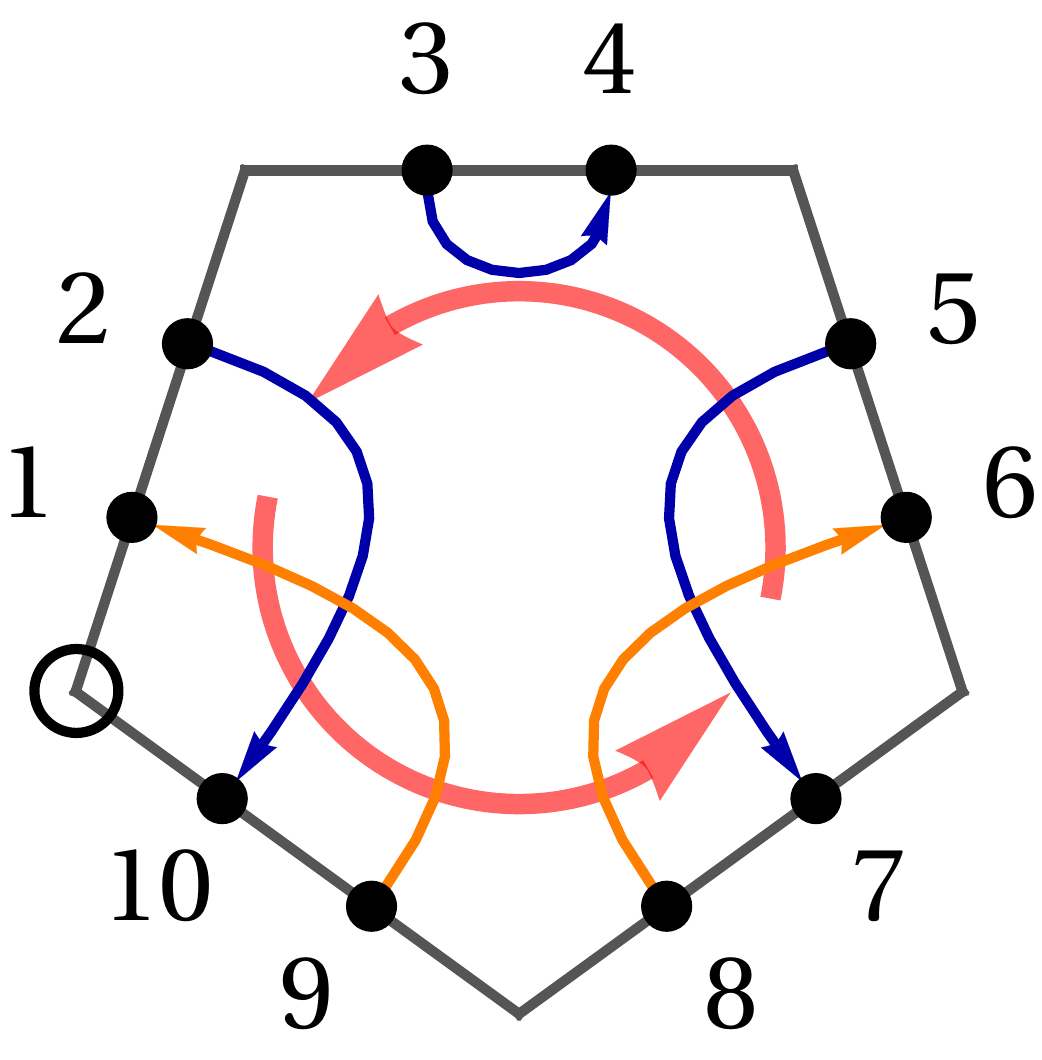}
\end{gathered}
\end{equation}
How does this transformation affect the dimer parities? First, let us interpret \eqref{EQ_CYCL_PERM} as a cyclic permutation of Majorana indices through an operator $\mathcal{M}_4$ acting on Majorana modes. For example, in the initial labeling we find $(\m_5 + \i \m_7)\ket\psi=0$. After a cyclic permutation $\ket\psi \mapsto \ket{\psi^\prime}= \mathcal{M}_4 \ket\psi$ with $\mathcal{M}_4 \m_k \mathcal{M}_4^{-1} = \m_{(k+4)\,\mathrm{mod}\, 10}$, it follows that 
\begin{align}
\label{EQ_CYCL_PERM_M_EX}
(\m_9 + \i \m_1)\ket{\psi^\prime}\propto (\m_1 - \i \m_9)\ket{\psi^\prime} = 0\ .
\end{align}
Hence, while a cyclic permutation of Majorana operators does not change the direction of the arrows, it flips the local parities of dimers ending on the edges between the initial and final pivot (with indices $i,i+1,\dots,f$), effectively acting like a product of all Majorana operators $\m_{2i-1},\m_{2i},\dots,\m_{2f}$ on these edges.

However, in our use of Majorana dimers as a description of the HyPeC, the underlying physical system is one of spins, where (Pauli) operators on different sites commute. Hence, we are interested in cyclic permutations not of Majorana modes but of the underlying spin degrees of freedom. This modifies the previous result, as each spin permutation effects a different Jordan-Wigner transformation. Consider an initial spin state vector
\begin{equation}
\ket\psi = \sum_{k\in \BC {N}} T_{k_1,k_2,\dots,k_{N}} \ket{k_1,k_2,\ldots, k_{N}} \ .
\end{equation}
A permutation $i \mapsto \sigma(i)$ of indices $i$ gives rise to the spin state vector $\sket{\tilde{\psi}} = \mathcal{S}_\sigma \ket\psi$, where $\mathcal{S}_\sigma$ is the spin-picture unitary permutation operator. Explicitly,
\begin{align}
\sket{\tilde{\psi}} &= \sum_{k\in \BC {N}} T_{k_1,k_2,\dots,k_N} \ket{\sigma(k_1),\sigma(k_2),\ldots,\sigma(k_N)} \nonumber \\
&= \sum_{j\in \BC {N}} \tilde{T}_{j_1,j_2,\dots,j_{N}} \ket{j_1,j_2,\ldots, j_{N}}\text{ }
\end{align}
with
\begin{align}
  \tilde{T}_{j_1,j_2,\dots,j_{N}} = T_{\sigma^{-1}(j_1), \sigma^{-1}(j_2), \ldots \sigma^{-1}(j_N)}\ .
\end{align}
In terms of fermionic operators, the initial and final state vectors $\ket\psi$ and $\sket{\tilde{\psi}}$ are identified as
\begin{align}
\ket\psi &= \sum_{k\in \BC {N}} T_{k_1,k_2,\dots,k_{N}}\, (\fd_1)^{k_1} (\fd_2)^{k_2} \dots (\fd_N)^{k_N} \vacket \text{ ,} \\
\sket{\tilde{\psi}} &= \sum_{k\in \BC {N}} \tilde{T}_{k_1,k_2,\dots,k_{N}}\, (\ftd_1)^{k_1} (\ftd_2)^{k_2} \dots (\ftd_N)^{k_N} \vacket \text{ .}
\end{align}
The operators $\fd_k$ and $\ftd_k$ are fermionic operators defined by the respective Jordan-Wigner transformations. Consider a one-step cyclic spin permutation $i \mapsto (i+1) \mod N$ through a permutation operator $\mathcal{S}_{+1}$. The Majorana operators transform as $\m_k \mapsto \mathcal{S}_{+1} \m_k \mathcal{S}_{+1}^\dagger$, and explicitly as
\begin{align}
\m_1 &= X_1 \mapsto X_2 = Z_1 \mt_3 \text{ ,} \nonumber\\
\m_2 &= Y_1 \mapsto Y_2 = Z_1 \mt_4 \text{ ,} \nonumber\\
\m_3 &= Z_1 X_2 \mapsto Z_2 X_3 = Z_1 \mt_5 \text{ ,} \nonumber\\
&\dots \\
\m_{2N-2} &= Z_1 \dots Z_{N-2} Y_{N-1}  \mapsto Z_2 \dots Z_{N-1} Y_N = Z_1 \mt_N \text{ ,} \nonumber\\
\m_{2N-1} &= Z_1 \dots Z_{N-1} X_N  \mapsto X_1 Z_2 \dots Z_N = - Z_1 \mt_1 \Par \text{ ,} \nonumber\\
\m_{2N} &= Z_1 \dots Z_{N-1} Y_N  \mapsto Y_1 Z_2 \dots Z_N = - Z_1 \mt_2 \Par \text{ .} \nonumber
\end{align}
Hence, the transformed Majorana operators are \textsl{not} the Majorana operators defined by the new Jordan-Wigner transformation! Instead, under the cyclic permutation $\mathcal{S}_{+1}$ all $\m_k$ for $k<2N-1$ transform as $\m_k \mapsto Z_1 \mt_{k+2}$, while $\m_{2N-1}$ and $\m_{2N}$ transform with an additional total parity $\Par = Z_1 Z_2 \dots Z_N$. This changes the dimer conditions: If the state vector $\ket\psi$ is annihilated by the operator $\m_j + \i\, p_{j,k} \m_k$ (with $j<k$), then this operator changes under cyclic spin permutation to
\begin{align}
(\m_j + \i\, p_{j,k} \m_k) \ket\psi &\mapsto \mathcal{S}_{+1} (\m_j + \i\, p_{j,k} \m_k) \mathcal{S}_{+1}^\dagger \mathcal{S}_{+1} \ket\psi \nonumber \\
&= ( \mathcal{S}_{+1} \m_j \mathcal{S}_{+1}^\dagger + \i\, p_{j,k}\, \mathcal{S}_{+1} \m_k \mathcal{S}_{+1}^\dagger ) \sket{\tilde{\psi}}  \nonumber \\
\end{align}
Let us distinguish this result by the parity of $\sket{\tilde{\psi}}$. For even parity $\Par\sket{\tilde{\psi}} = \sket{\tilde{\psi}}$, we find
\begin{align}
\label{EQ_PARITY_FLIP_PLUS}
&( \mathcal{S}_{+1} \m_j \mathcal{S}_{+1}^\dagger + \i\, p_{j,k}\, \mathcal{S}_{+1} \m_k \mathcal{S}_{+1}^\dagger ) \sket{\tilde{\psi}}
\nonumber\\
= &\begin{cases}
Z_1 (\mt_{j+2} + \i\, p_{j,k} \mt_{k+2}) \sket{\tilde{\psi}} \ , & j,k{<}2N{-}1 ,\\
Z_1 (\mt_{j+2} - \i\, p_{j,k} \mt_{k+2-2N}) \sket{\tilde{\psi}} \ , & j{<}2N{-}1,k {\ge} 2N{-}1 ,\\
- Z_1 (\mt_{1} + \i\, p_{j,k} \mt_{2}) \sket{\tilde{\psi}} \ , & j{=}2N{-}1,k{=}2N .\\
\end{cases}
\end{align}
For odd parity $\Par\sket{\tilde{\psi}} = -\sket{\tilde{\psi}}$, the result is given by
\begin{align}
\label{EQ_PARITY_FLIP_MINUS}
&( \mathcal{S}_{+1} \m_j \mathcal{S}_{+1}^\dagger + \i\, p_{j,k}\, \mathcal{S}_{+1} \m_k \mathcal{S}_{+1}^\dagger ) \sket{\tilde{\psi}}
\nonumber\\
= &\begin{cases}
Z_1 (\mt_{j+2} + \i\, p_{j,k} \mt_{k+2}) \sket{\tilde{\psi}} \ , & j,k{<}2N{-}1 ,\\
Z_1 (\mt_{j+2} + \i\, p_{j,k} \mt_{k+2-2N}) \sket{\tilde{\psi}} \ , & j{<}2N{-}1,k {\ge} 2N{-}1 ,\\
Z_1 (\mt_{1} + \i\, p_{j,k} \mt_{2}) \sket{\tilde{\psi}} \ , & j{=}2N{-}1,k{=}2N .\\
\end{cases}
\end{align}
For odd-parity states, we thus find from \eqref{EQ_PARITY_FLIP_MINUS} that the dimer parities flip as one dimer endpoint moves past the pivot, just as in \eqref{EQ_CYCL_PERM_M_EX}. However, for an even-parity state \eqref{EQ_PARITY_FLIP_PLUS} tells us that the dimer parities remain invariant under spin permutations. This is the only difference between cyclic permutations of Majorana modes and of the underlying spin sites.
Omitting arrows and showing only dimer parities, spin cyclic permutation correspond to diagrams such as 
\begin{align}
\begin{gathered}
\includegraphics[height=0.1\textheight]{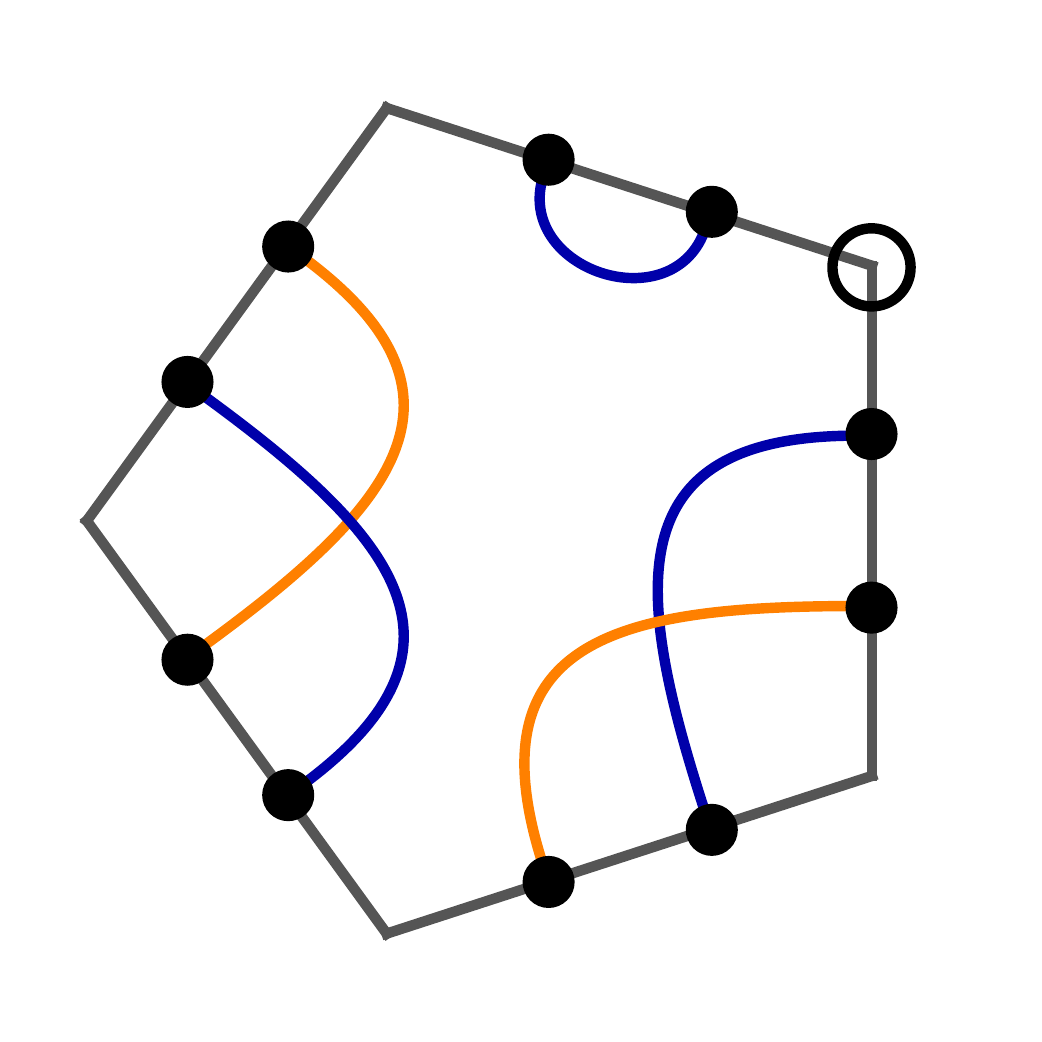}
\end{gathered}
&\mapsto \,
\begin{gathered}
\includegraphics[height=0.1\textheight]{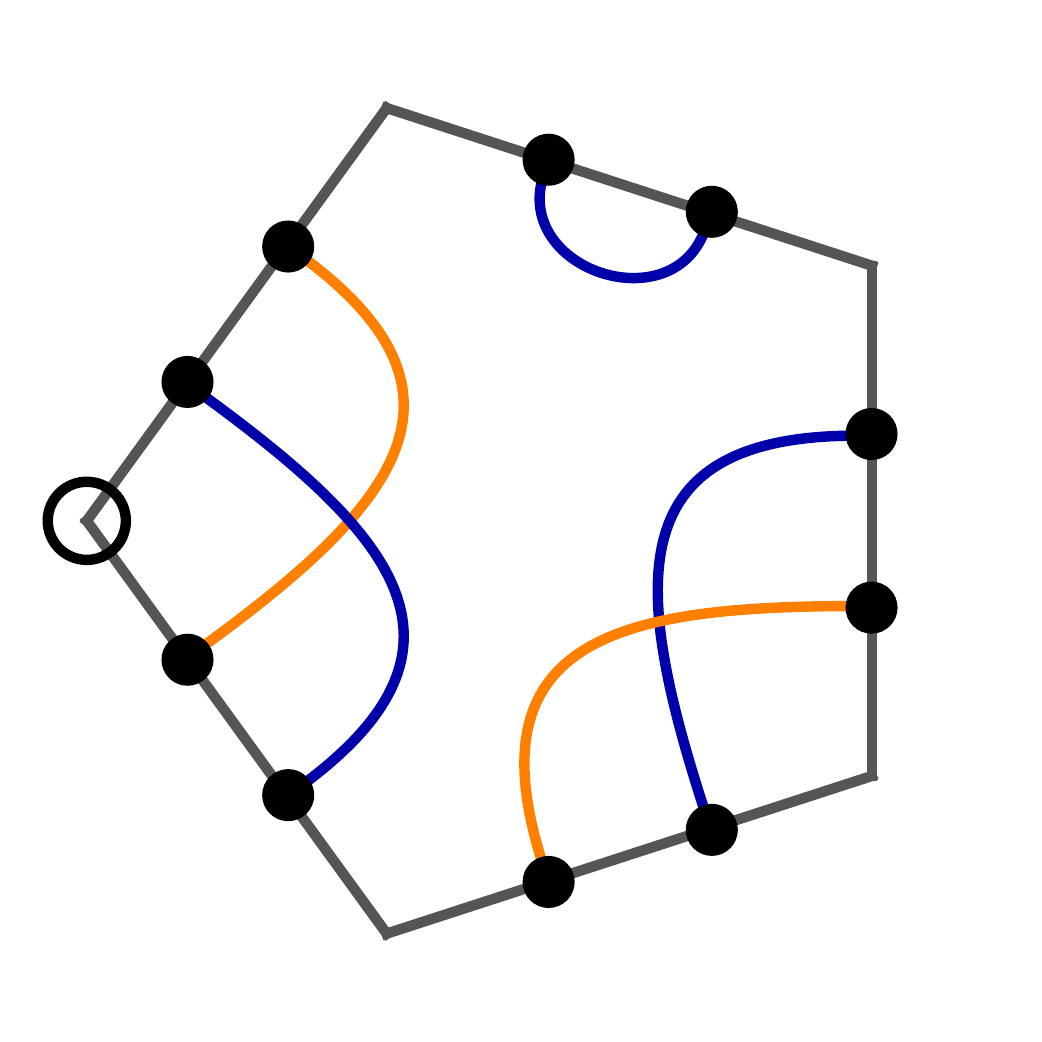}
\end{gathered} \\
\begin{gathered}
\includegraphics[height=0.1\textheight]{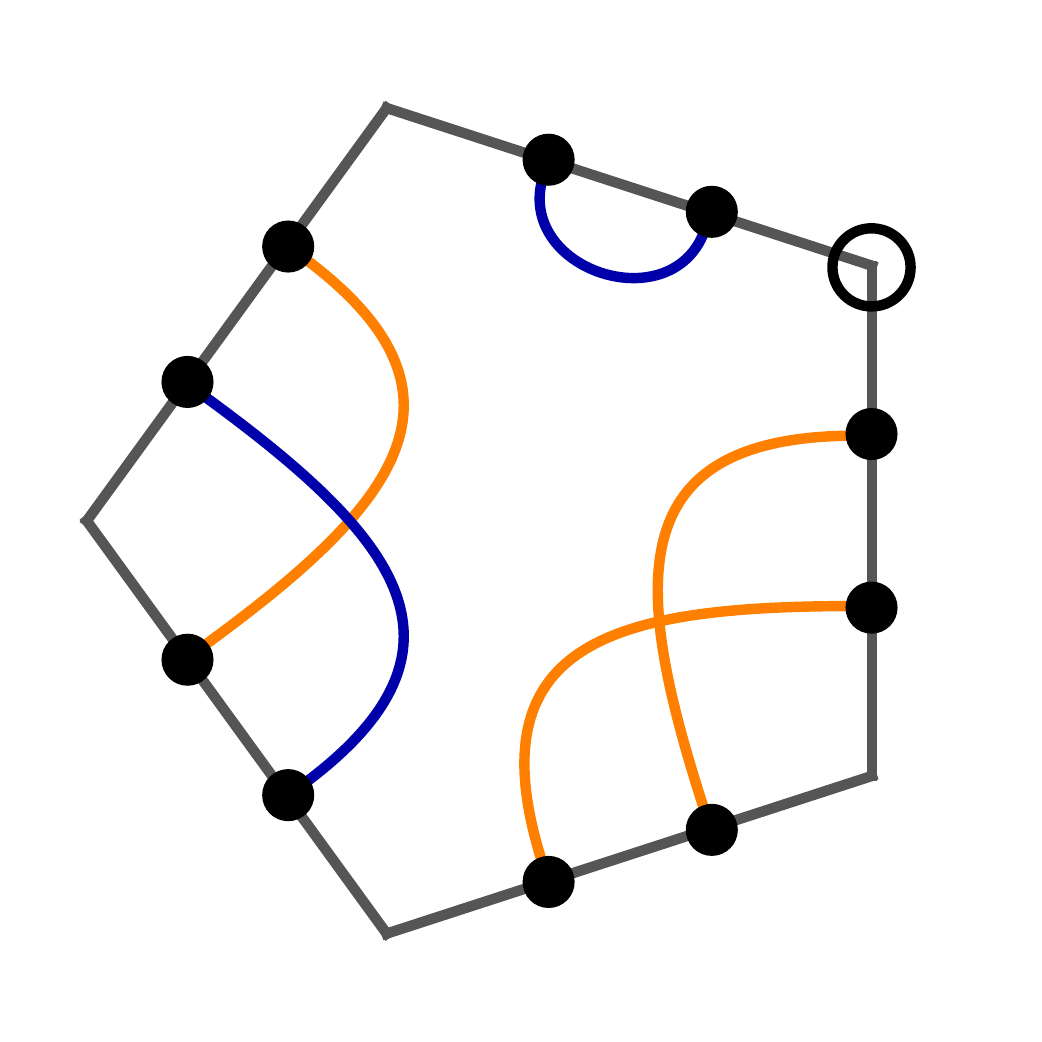}
\end{gathered}
&\mapsto \,
\begin{gathered}
\includegraphics[height=0.1\textheight]{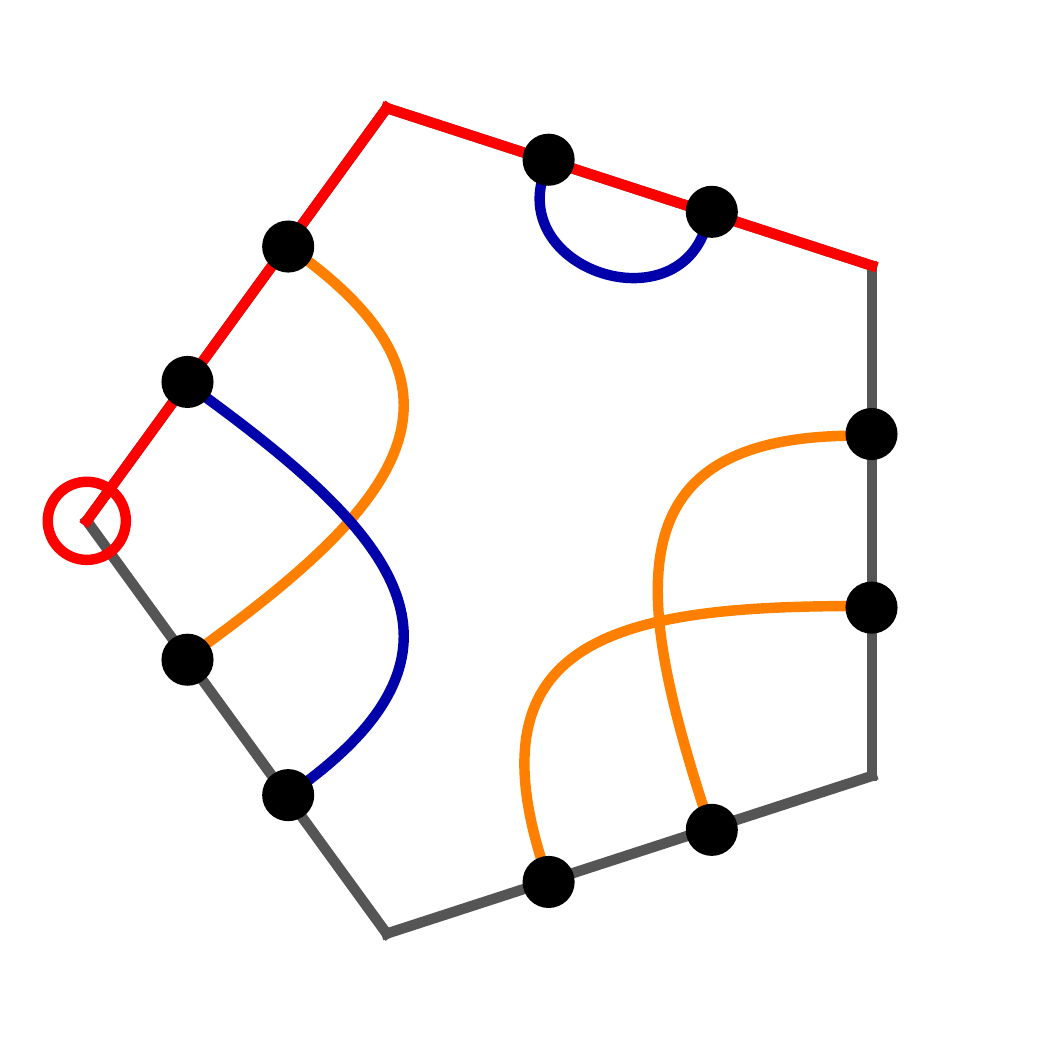}
\end{gathered}
=\,
\begin{gathered}
\includegraphics[height=0.1\textheight]{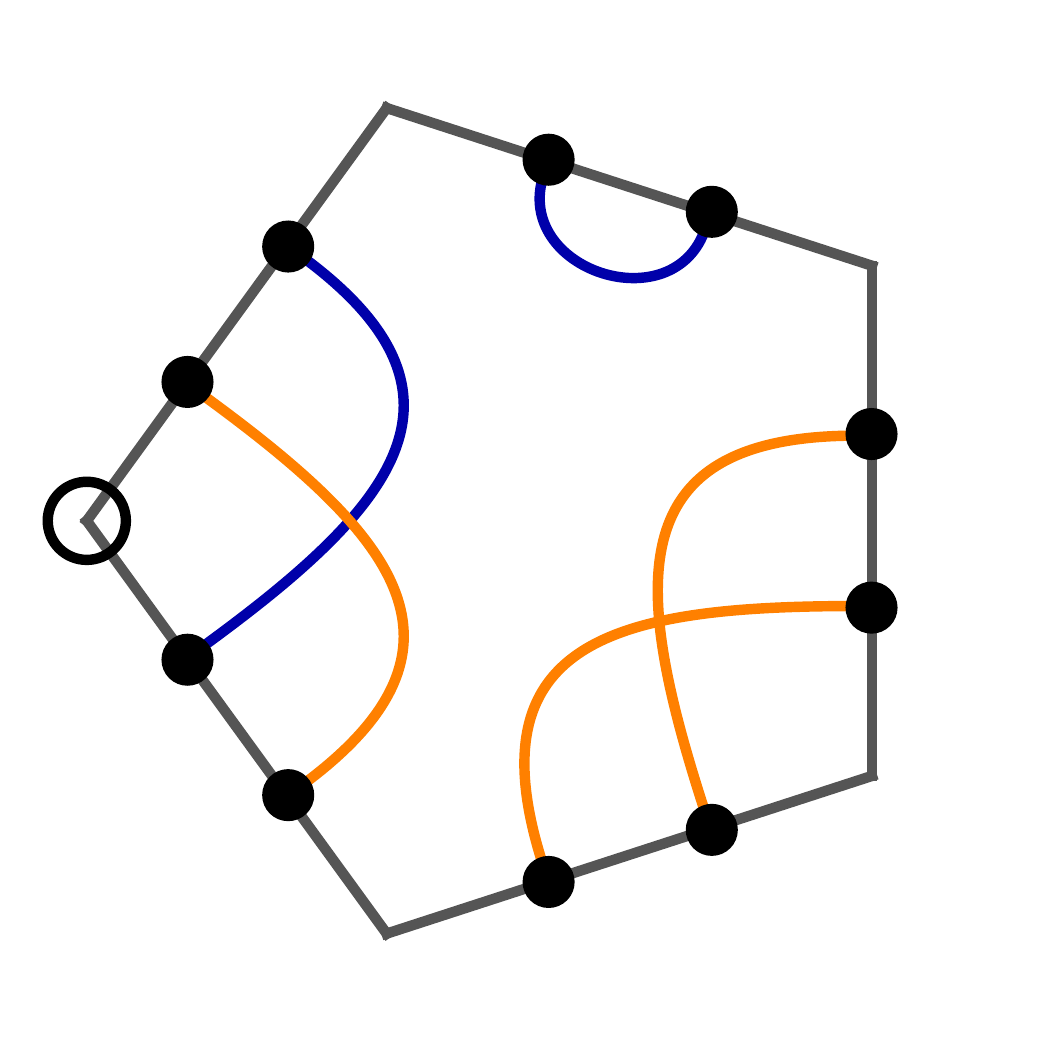}
\end{gathered}
\end{align}
These two diagrams illustrate a more general observation.
Namely, the upper diagram shows a parity-even Majorana dimer state, where changing the index labeling (i.e.,\ shifting the pivot) does not change the dimer parities. The lower one, however, is parity-odd: The red-shaded edges, following the path along which the pivot was moved, represent Majorana operators on the edges that flip the connected dimer parities. Note that there are two possible paths (clockwise and anti-clockwise) between the initial and final position of the pivot, and that both lead to the same final state up to a total sign.

As a special case, consider the behaviour of the logical code state vectors $\ket{\bar{0}}_5$ and $\ket{\bar{1}}_5$ of the $[[5,1,3]]$ code under cyclic permutations (here, for a clockwise shift of two edges):
\begin{align}
\label{EQ_PENTAGON_STATES}
\ket{\bar{0}}_5\; = \;
\begin{gathered}
\includegraphics[height=0.1\textheight]{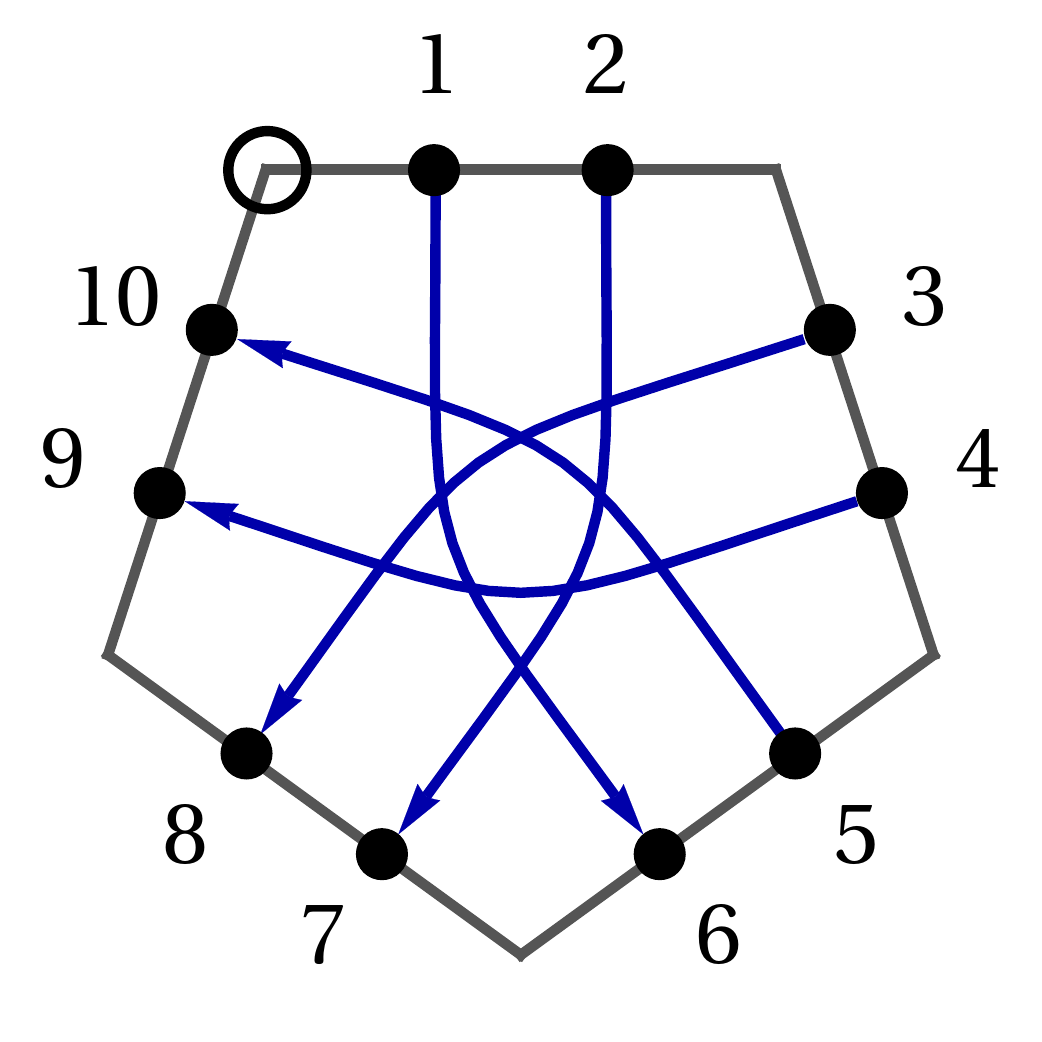}
\end{gathered}
\;\mapsto \;
&\begin{gathered}
\includegraphics[height=0.1\textheight]{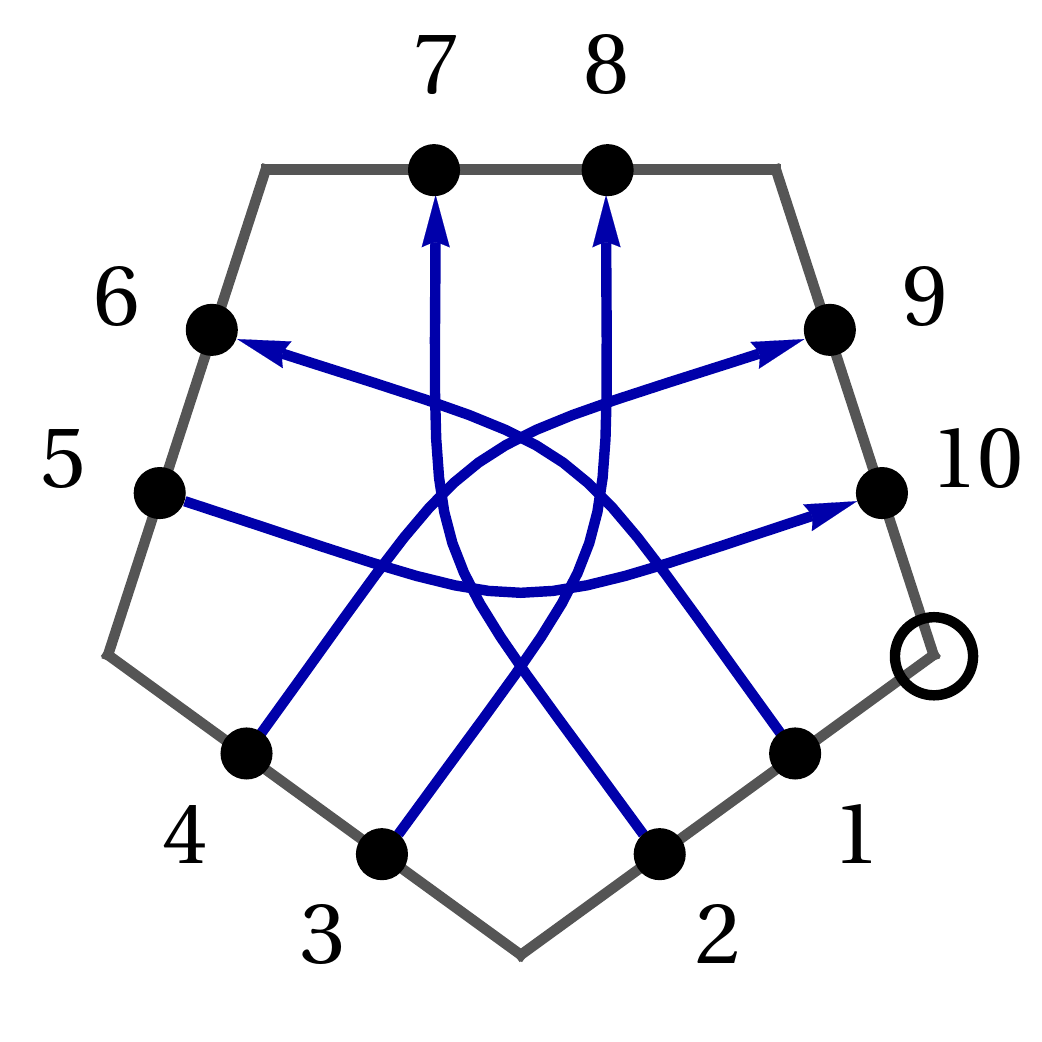}
\end{gathered} \\
\label{EQ_PENTAGON_STATES2}
\ket{\bar{1}}_5\; = \;
\begin{gathered}
\includegraphics[height=0.1\textheight]{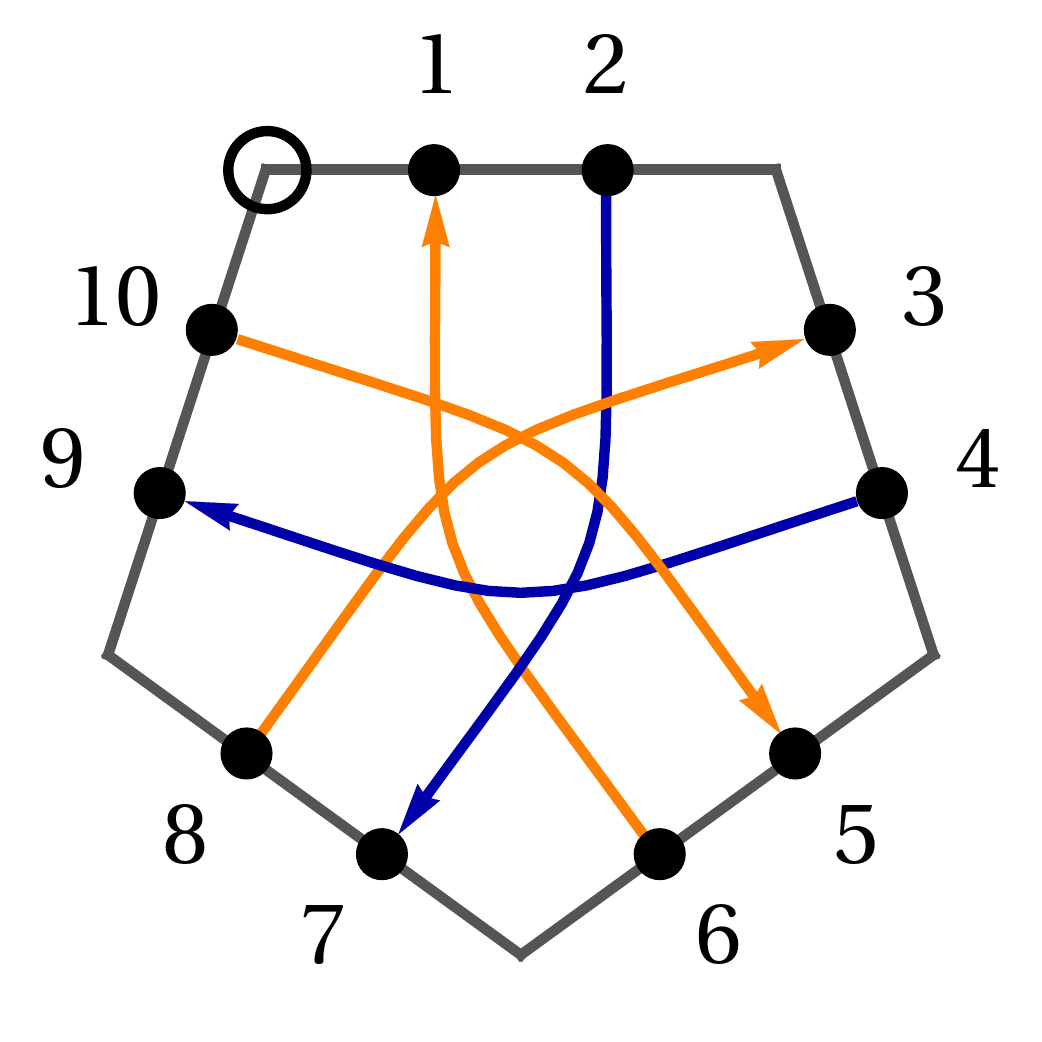}
\end{gathered} 
\;\mapsto \;
&\begin{gathered}
\includegraphics[height=0.1\textheight]{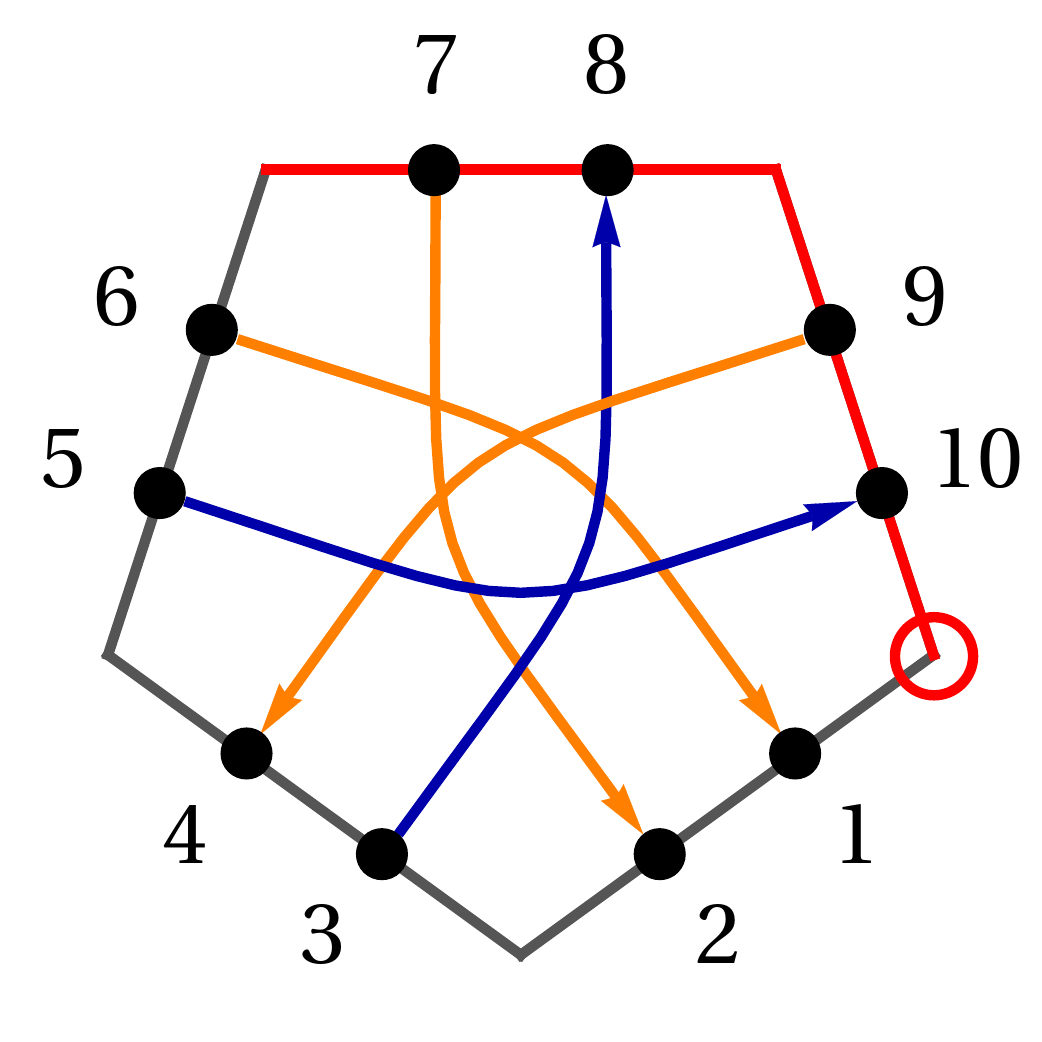}
\end{gathered} \nonumber\\
\;=\; 
&\begin{gathered}
\includegraphics[height=0.1\textheight]{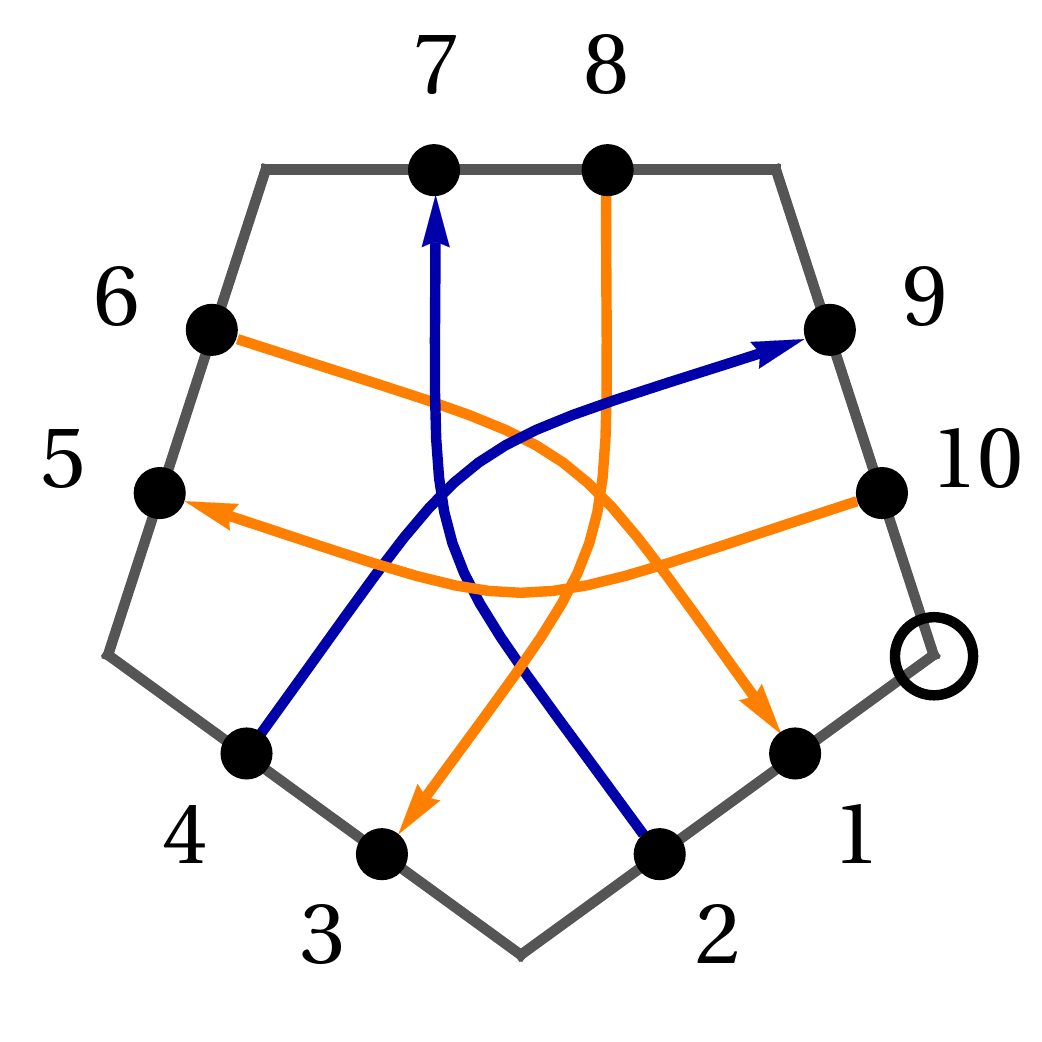}
\end{gathered}
\end{align}
We see that a cyclic permutation of these states is equivalent to a rotation of its dimer parities, which is simply a rotation of the corresponding diagram. This is because the tensors $T$ corresponding to these states are invariant under cyclic permutations of indices, i.e., $T_{i,j,k,l,m}=T_{m,i,j,k,l}$. The explicit construction of these states can be easily performed using \emph{matchgate tensors} \cite{Jahn:2017tls}.

\subsection{Computing entanglement}

The entanglement entropy $S_A{=}-\tr_A (\rho_A \log \rho_A)$ of a subsystem $A$ and its corresponding reduced density matrix $\rho_A{=}\tr_{A^\text{C}}(\rho_\text{tot})$ can be evaluated diagrammatically. Given the $2M \times 2M$ Majorana covariance matrix $\Gamma^A$ of the subspace belonging to $A$ (i.e.,\ the rows and columns of the full $2L \times 2L$ covariance matrix $\Gamma$ whose Majorana modes are contained in $A$), we can perform 
a special orthogonal transformation $\Gamma^A = Q \tilde{\Gamma}^A Q^T $ to the form
\begin{equation}
\tilde{\Gamma}^A = \bigoplus_{i=1}^M \left(\begin{matrix}
  0 & \lambda_i \\
  -\lambda_i & 0
\end{matrix}\right) \text{ ,}
\end{equation}
where $\pm i \lambda_k$ are the eigenvalues of $\Gamma^A$,
some of which may be zero. From this form the entanglement entropy follows as
\begin{equation}
\label{EQ_EE_SYMPL}
S_A = \sum_{i=1}^M \left( - \frac{1+\lambda_i}{2} \log \frac{1+\lambda_i}{2} - \frac{1-\lambda_i}{2} \log \frac{1-\lambda_i}{2} \right) \text{ .}
\end{equation}
As we have found in Section \ref{SEC_MAJ_DIM}, the covariance matrix entries $\Gamma_{j,k}$ of Majorana dimer states can only be $\pm 1$ or zero. Consider the $j$th row (or column) of the sub-matrix $\Gamma^A$: The dimer connected to Majorana mode $j$ ends on another mode $k$ (with $1 \leq k \leq 2L$.) If $j,k \in A$, the $j$ and $k$th row will jointly contribute to a $\lambda_i$ of $\pm 1$, i.e.,\ zero entanglement entropy. However, if $j \in A, k \not\in A$, the $j$th row of $\Gamma^A$ will be zero. As the number of such ``dimer leaks'' must be even, each pair contributes to a vanishing $\lambda_i$. Thus each dimer connecting  $A$ with its complement $A^{\text{C}}$ contributes an entanglement entropy of $\frac{1}{2}\log 2$, i.e.,\ half of an EPR pair. Graphically, the entanglement entropy reduces to counting such dimers
\begin{equation}
\label{EQ_DIMER_EE}
S_A = \text{(\# dimers between $A$ and $A^{\text{C}}$)} \times \frac{\log 2}{2} \ .
\end{equation}
Consider the following example state.
\begin{equation}
\label{EQ_DIMER_EE_EXAMPLE}
\begin{gathered}
\includegraphics[height=0.15\textheight]{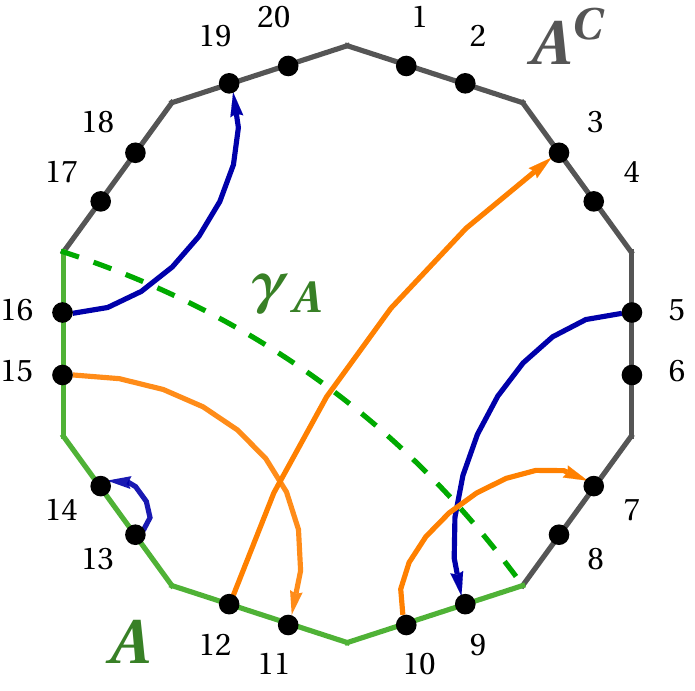}
\end{gathered}
\end{equation}
The subsystem $A$ comprises four edges with the Majorana modes $9$ to $16$. As four dimers connect $A$ with $A^{\text{C}}$, the entanglement entropy is given by $S_A = 2 \log 2$. Effectively, $S_A$ counts the number of dimers across the cut $\gamma_A$ separating $A$ from $A^{\text{C}}$ (shown as a dashed line). For contracted states, $S_A \leq |\gamma_A| \log 2$, where $|\gamma_A|$ is the length of the shortest cut through the contracted network. Thus, we recover the tensor network interpretation of the Ryu-Takayanagi surface $\gamma_A$, which appears in continuum AdS/CFT in the holographic entanglement entropy formula \cite{PhysRevLett.96.181602}
\begin{equation}
\label{EQ_RT_FORMULA}
S_A = \frac{|\gamma_A|}{4\, G_N} \ ,
\end{equation}
which expresses $S_A$ in terms of the area of the minimal surface $\gamma_A$, denoted $|\gamma_A|$, and Newton's constant $G_N$. In our two-dimensional bulk space, $\gamma_A$ is simply a geodesic and $|\gamma_A|$ its length. We will see later how the discrete analog of \eqref{EQ_RT_FORMULA}, where 
$1/(4G_N) \to \log 2$, is saturated in the HyPeC.

The definition of entanglement entropy $S_A$ can be ambiguous under a mapping from spins to fermions, as fermionic operators on different sites do not commute. As long as we consider connected subsystems $A$, such a mapping \eqref{EQ_SPIN_FERMION_BASES} yields the same $S_A$, as both are related only by cyclical permutation of fermionic sites, which only leads to a sign flip along the permuted rows and columns of the covariance matrix $\Gamma$.
For a region $A$ composed of disconnected parts, $S_A$ is generally not preserved by a mapping from spins to fermions. If, as in the HyPeC model, we want to compute the spin entanglement entropy in the effective fermionic picture, we first need to permute the spin degrees of freedom so that $A$ becomes connected, and then apply the mapping to fermions. However, such a spin transposition usually breaks the Majorana dimer structure, as it leads to fermionic states that are not ground states of Hamiltonians quadratic in Majorana operators.
Thus \eqref{EQ_DIMER_EE} describes fermionic entanglement that remains valid in the spin picture only for connected regions $A$.

In the fermionic picture, we can easily generalize \eqref{EQ_DIMER_EE} to disjoint subsystems, such as the \textsl{mutual information}
\begin{align}
\label{EQ_DIMER_MI}
I(A:B) &= S_A + S_B - S_{AB} \nonumber\\
&= \text{(\# dimers between $A$ and $B$)} \times \log 2 \text{ .}
\end{align}
Compared to \eqref{EQ_DIMER_EE}, each dimer in \eqref{EQ_DIMER_MI} is counted twice. 
In terms of the geometry of the dimer graph itself, \eqref{EQ_DIMER_MI} corresponds to a system with an \emph{exact area law}. \cite{Bao:2015boa}
One of the properties of this form of the mutual information is an always vanishing \textsl{tripartite information} \cite{Casini:2008wt}
\begin{align}
I_3(A:B:C) &= I(A:B) + I(A:C) - I(A:BC) \nonumber\\
 &= 0 \text{ .}
\end{align}
This implies that Majorana dimer models are compatible with holographic theories, where $I_3 \leq 0$ \cite{Hayden:2011ag}. Furthermore, as we show in Appendix \ref{APP_EE_RULES}, the spectrum of R\'enyi entropies 
\begin{align}
S_A^{(n)}=\frac{1}{1-n} \log \tr(\rho_A^n )
\end{align}
is flat, a property of the underlying stabilizer state structure \cite{PhysRevLett.103.261601}. 
We show in Appendix \ref{APP_EE_RULES} that this property also follows from the Majorana dimer picture for arbitrary local superpositions of bulk input in the HyPeC under certain constraints on the (compact) boundary region considered.

To clarify the connection between Majorana dimers and EPR pairs, we can explicitly construct Bell states from pairs of dimers. Consider the following two even-parity dimers connecting edges $j$ and $k$ (with $j<k$) without crossing:
\begin{equation}
\label{EQ_EPR1}
\begin{gathered}
\includegraphics[height=0.08\textheight]{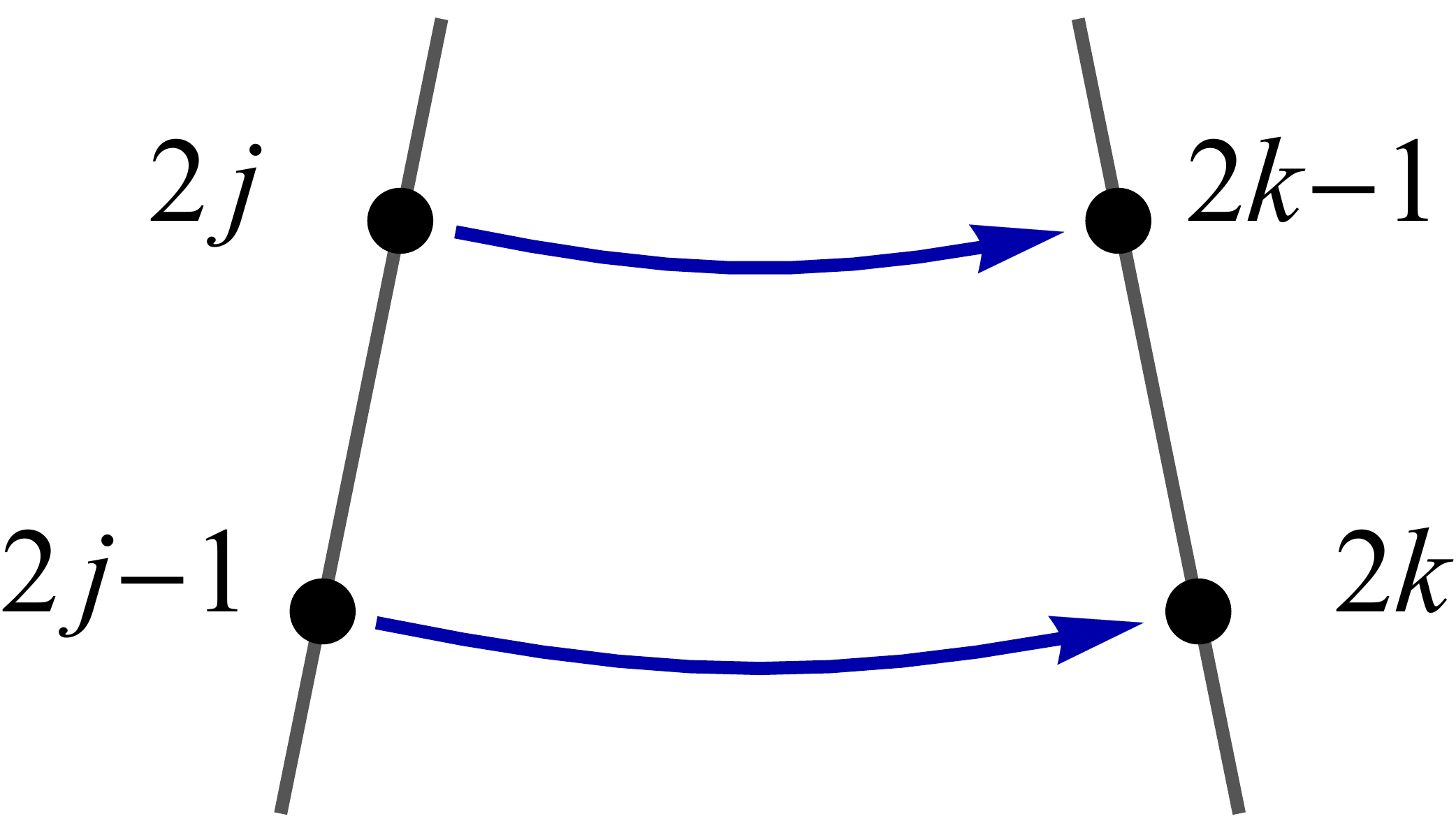}
\end{gathered}
\end{equation}
This corresponds to two conditions on the total state vector $\ket\psi$,
\begin{equation}
\label{EQ_EPR_COND}
\begin{aligned}
(\m_{2j-1} + \i\, \m_{2k}) \ket\psi &= \phantom{\i\,} (\fd_j + \fe_j - \fd_k + \fe_k) \ket\psi = 0 \text{ ,}\\
(\m_{2j} + \i\, \m_{2k-1}) \ket\psi &= \i\,(\fd_j - \fe_j + \fd_k + \fe_k) \ket\psi = 0 \text{ .}
\end{aligned}
\end{equation}
As no entanglement between edges $j$ and $k$ and the rest of the system exists, $\ket\psi$ should be factorizable with regards to these degrees of freedom:
\begin{equation}
\ket\psi \propto (a + b \fd_j + c \fd_k + d \fd_j \fd_k) (\dots) \vacket \text{ ,}
\end{equation}
where $(\dots)$ includes terms containing creation operators $\fd_i$ with $i\neq j,k$. Up to a complex phase, the parameters $a,b,c,d \in \mathbb{C}$ can be fixed using \eqref{EQ_EPR_COND}, which leads to $b=c=0$ and $a=d=1/\sqrt{2}$ (assuming normalization). This corresponds to a Bell state vector $\ket{\Phi^+} = (\ket{0} \ket{0} + \ket{1} \ket{1})/\sqrt{2}$ on sites $j$ and $k$. This analysis can be repeated for all possible dimer configurations, yielding Table \ref{TAB_BELL_STATES}. Conveniently, this allows us to form superpositions of dimers, for example
\begin{equation}
\label{EQ_EPR2}
\begin{gathered}
\includegraphics[height=0.05\textheight]{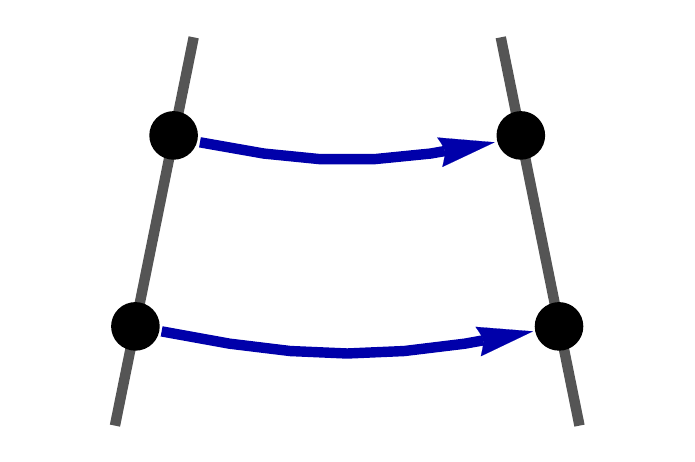}
\end{gathered}
\;=\;
\frac{1}{\sqrt{2}}
\begin{gathered}
\includegraphics[height=0.05\textheight]{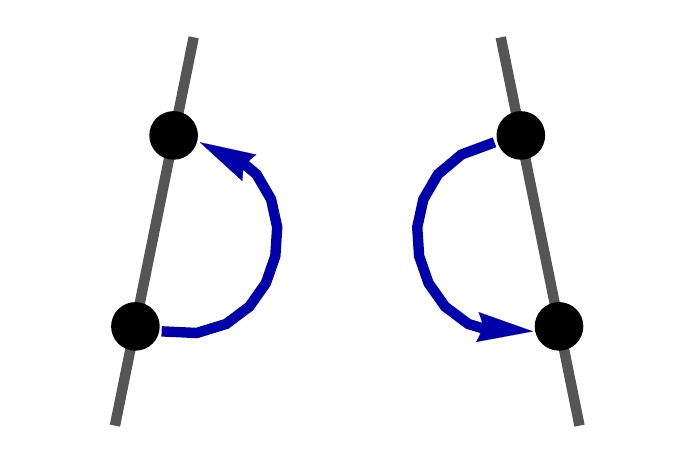}
\end{gathered}
\;+\;
\frac{1}{\sqrt{2}}
\begin{gathered}
\includegraphics[height=0.05\textheight]{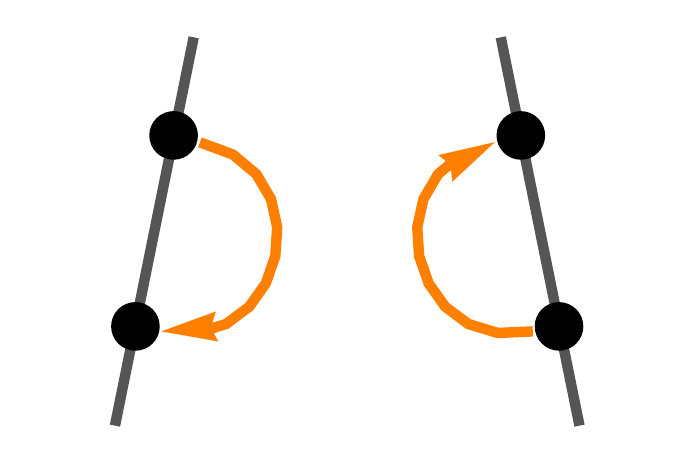}
\end{gathered} 
\end{equation}
Each diagram in this expression corresponds to a normalized Majorana dimer state. Note that this diagram confirms our intuition that a contraction, which is the sum of projections onto $\ket{0}$ and $\ket{1}$, is equivalent to connecting pairs of Majoranas via dimers. 
In a mild abuse of notation, we may thus write
\begin{equation}
\label{EQ_EPR3}
\begin{gathered}
\includegraphics[height=0.05\textheight]{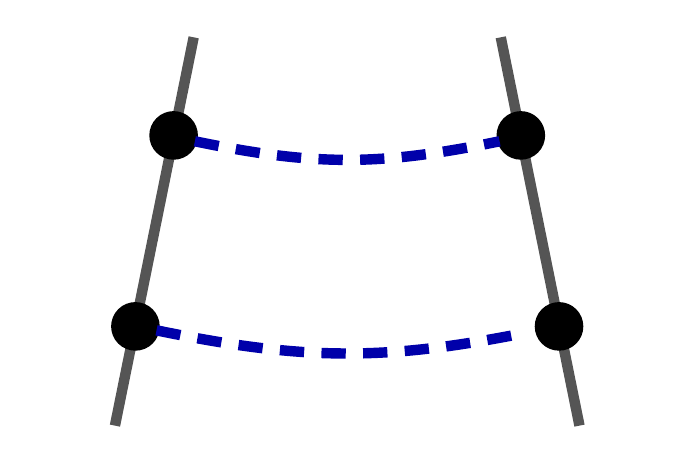}
\end{gathered}
\;=\;
\sqrt{2}
\begin{gathered}
\includegraphics[height=0.05\textheight]{epr_pair1.pdf}
\end{gathered}
\end{equation}
to express a contraction (dashed lines).
This also allows us to fix relative factors that appear through contraction, such as in the following projection of \eqref{EQ_EPR2} onto a $\ket{0}$ state vector:
\begin{align}
\label{EQ_EPR2EX}
\begin{gathered}
\includegraphics[height=0.05\textheight]{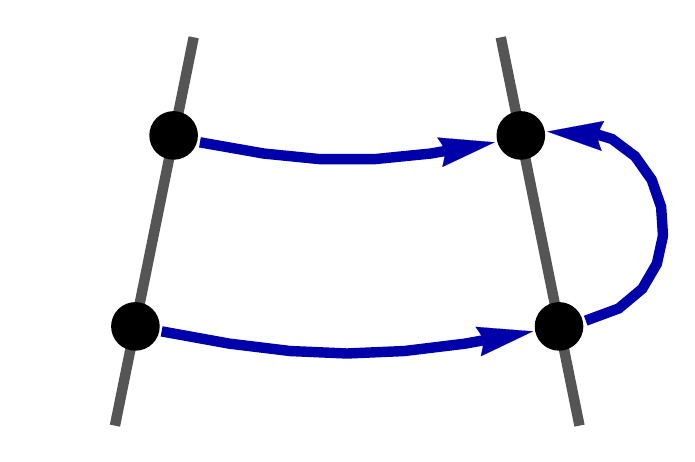}
\end{gathered}
\;&=\;
\frac{1}{\sqrt{2}}
\begin{gathered}
\includegraphics[height=0.05\textheight]{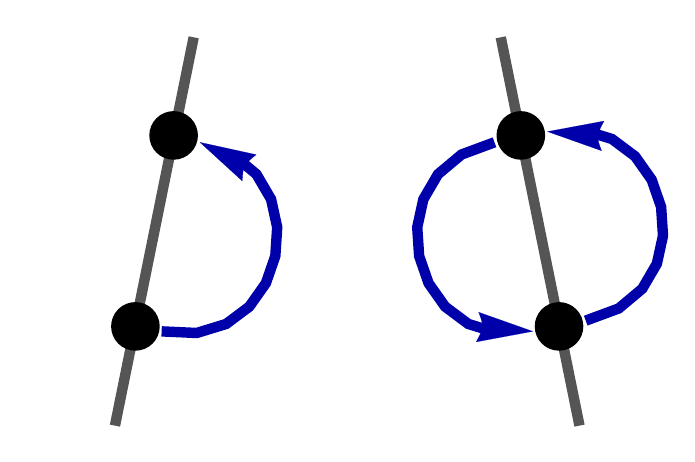}
\end{gathered}
\;+\;
\frac{1}{\sqrt{2}}
\begin{gathered}
\includegraphics[height=0.05\textheight]{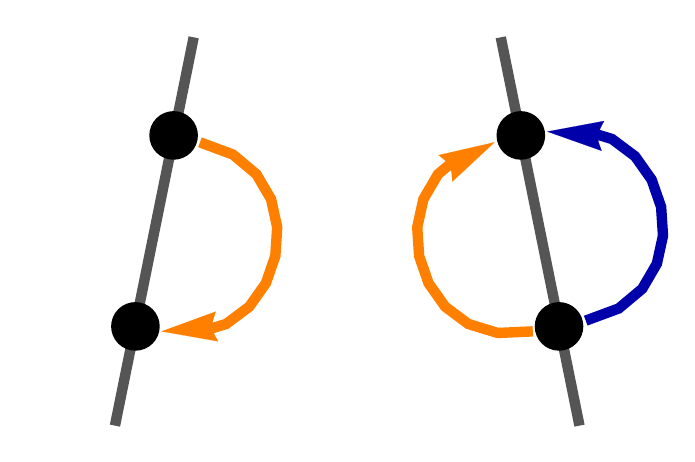}
\end{gathered} \nonumber\\
&=\;\frac{1}{\sqrt{2}}
\begin{gathered}
\includegraphics[height=0.05\textheight]{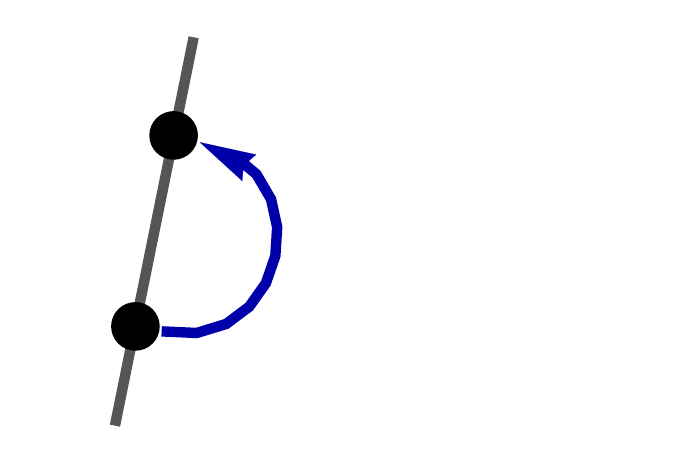}
\end{gathered}
\end{align}
The second term vanishes from the condition $\braket{0}{1}=0$, in agreement with the rule that loops of total odd parity vanish (compare Eq.\ \eqref{EQ_CONTR_EX2}). Note that the arrow orientation of the dimer for $\ket{0}$ is reversed, as it is used in its adjoint form $\bra{0}$ (more on Hermitian conjugates in the next section).
Projections like \eqref{EQ_EPR2EX} can be evaluated for each of the entries in Table \ref{TAB_BELL_STATES}, always leading to a resulting factor of $1/\sqrt{2}$. This result is heavily used in Appendix \ref{APP_EE_RULES}, where we study the entanglement properties of superpositions of HyPeC code states, where norms of Majorana dimer states become relevant.

\begin{table}
\begin{tabular}{c | c || c | c}
Majorana dimer & Bell state & Majorana dimer & Bell state \\
\hline
$\begin{gathered}
\includegraphics[height=0.04\textheight]{epr_pair1.pdf}
\end{gathered}$
 & $\frac{\ket{0} \ket{0} + \ket{1} \ket{1}}{\sqrt{2}}$
&$\begin{gathered}
\includegraphics[height=0.04\textheight]{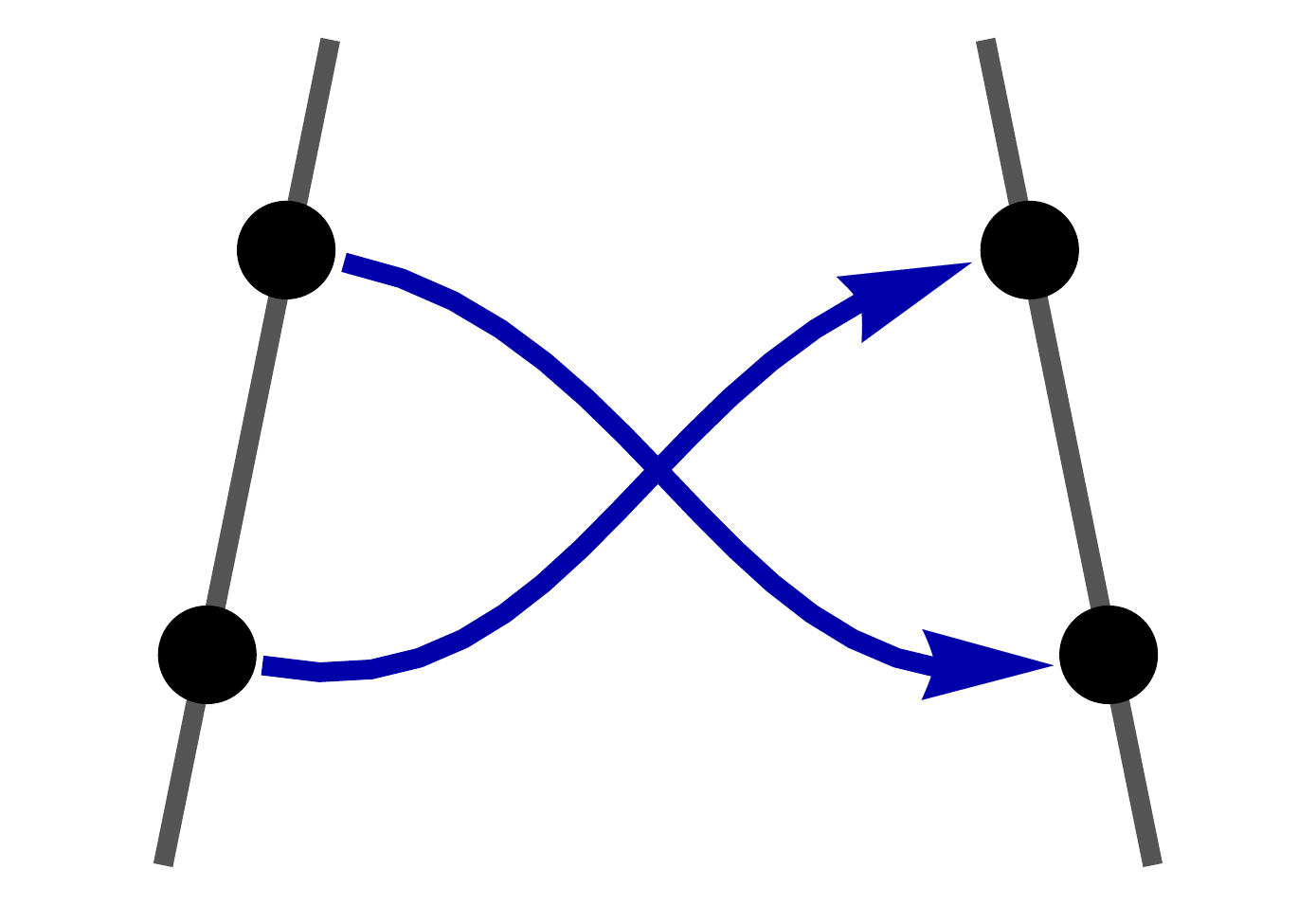}
\end{gathered}$
 & $\frac{\ket{0} \ket{1} + \i \ket{1} \ket{0}}{\sqrt{2}}$\\
$\begin{gathered}
\includegraphics[height=0.04\textheight]{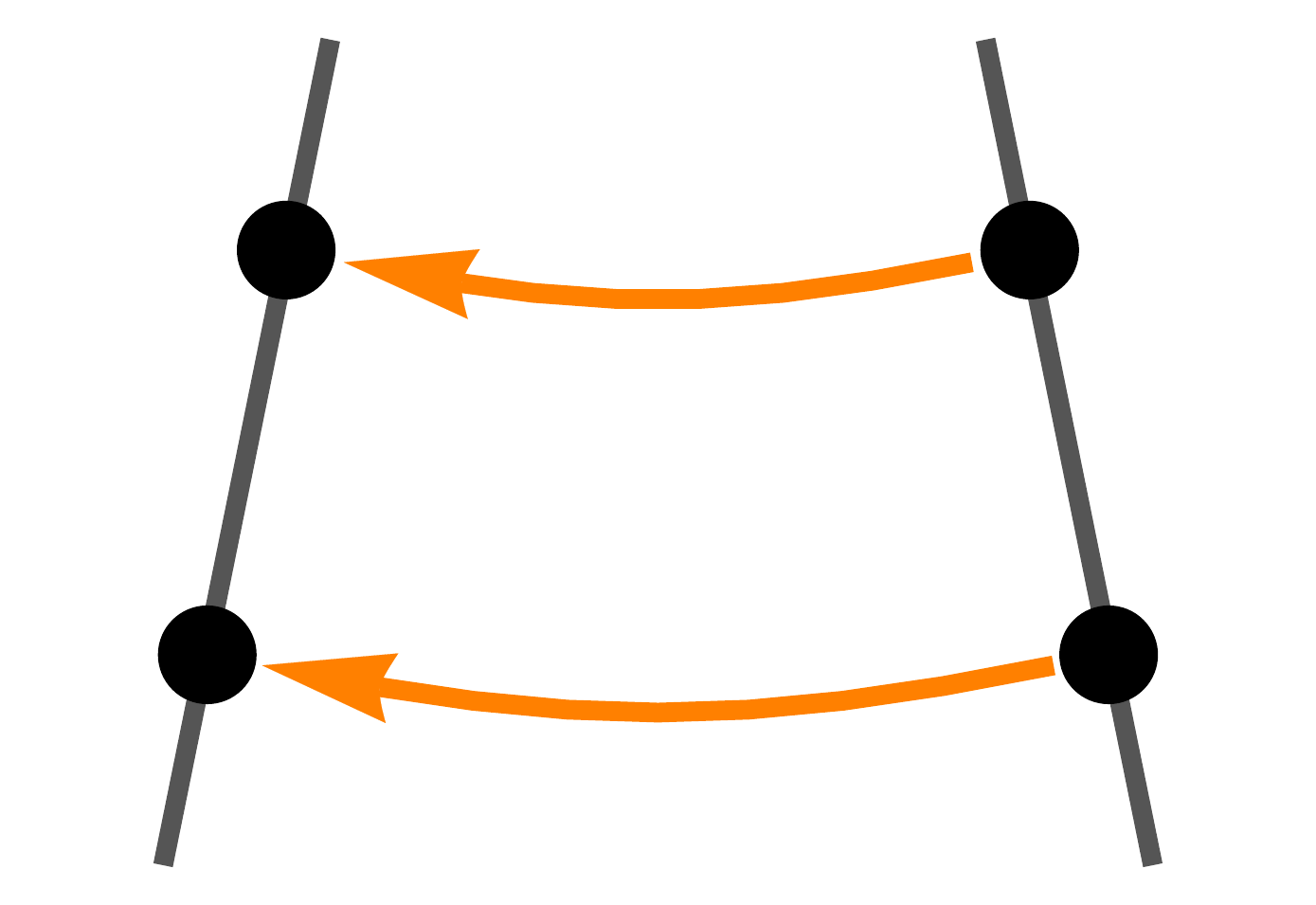}
\end{gathered}$
 & $\frac{\ket{0} \ket{0} - \ket{1} \ket{1}}{\sqrt{2}}$
&$\begin{gathered}
\includegraphics[height=0.04\textheight]{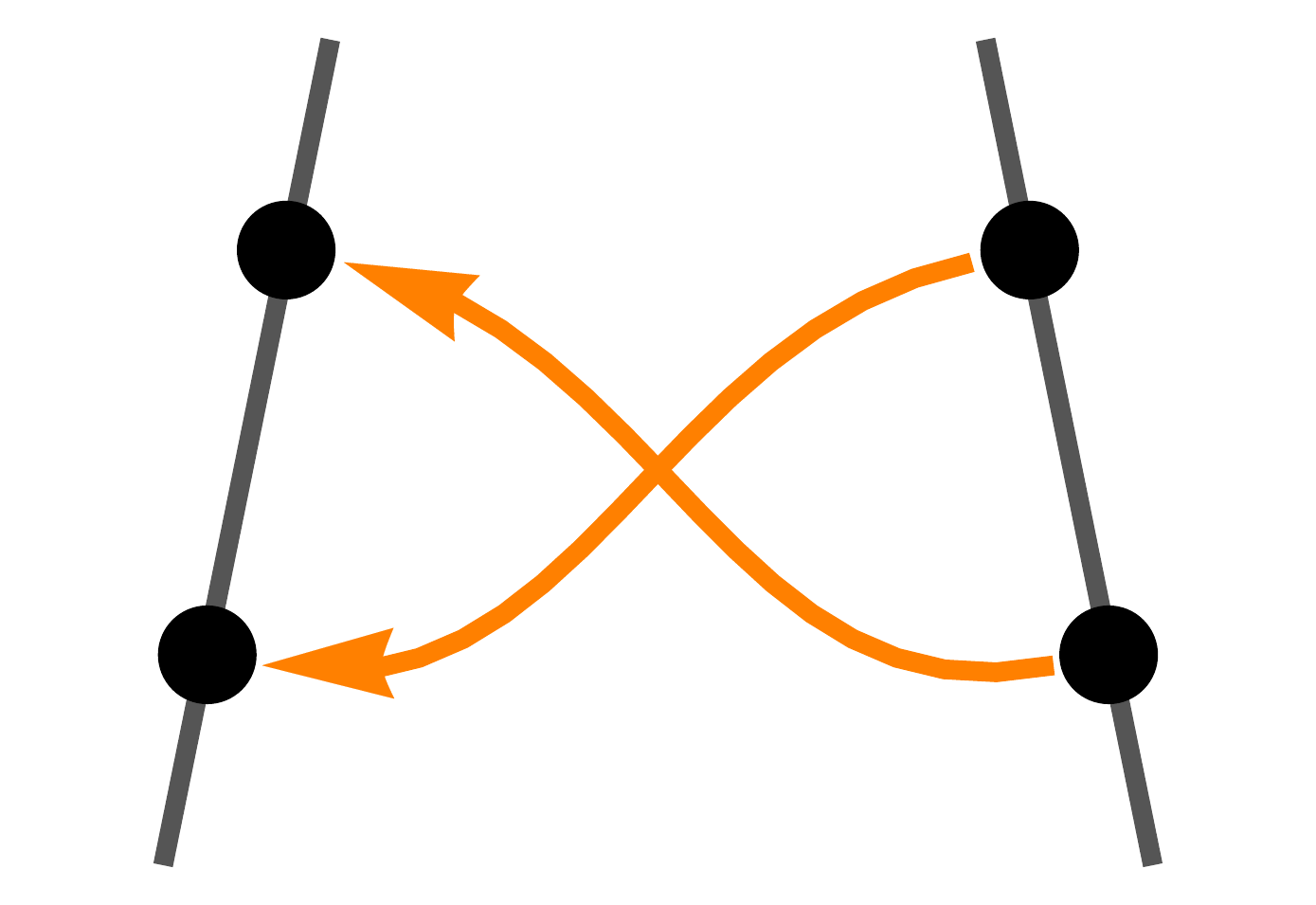}
\end{gathered}$
 & $\frac{\ket{0} \ket{1} - \i \ket{1} \ket{0}}{\sqrt{2}}$ \\
$\begin{gathered}
\includegraphics[height=0.04\textheight]{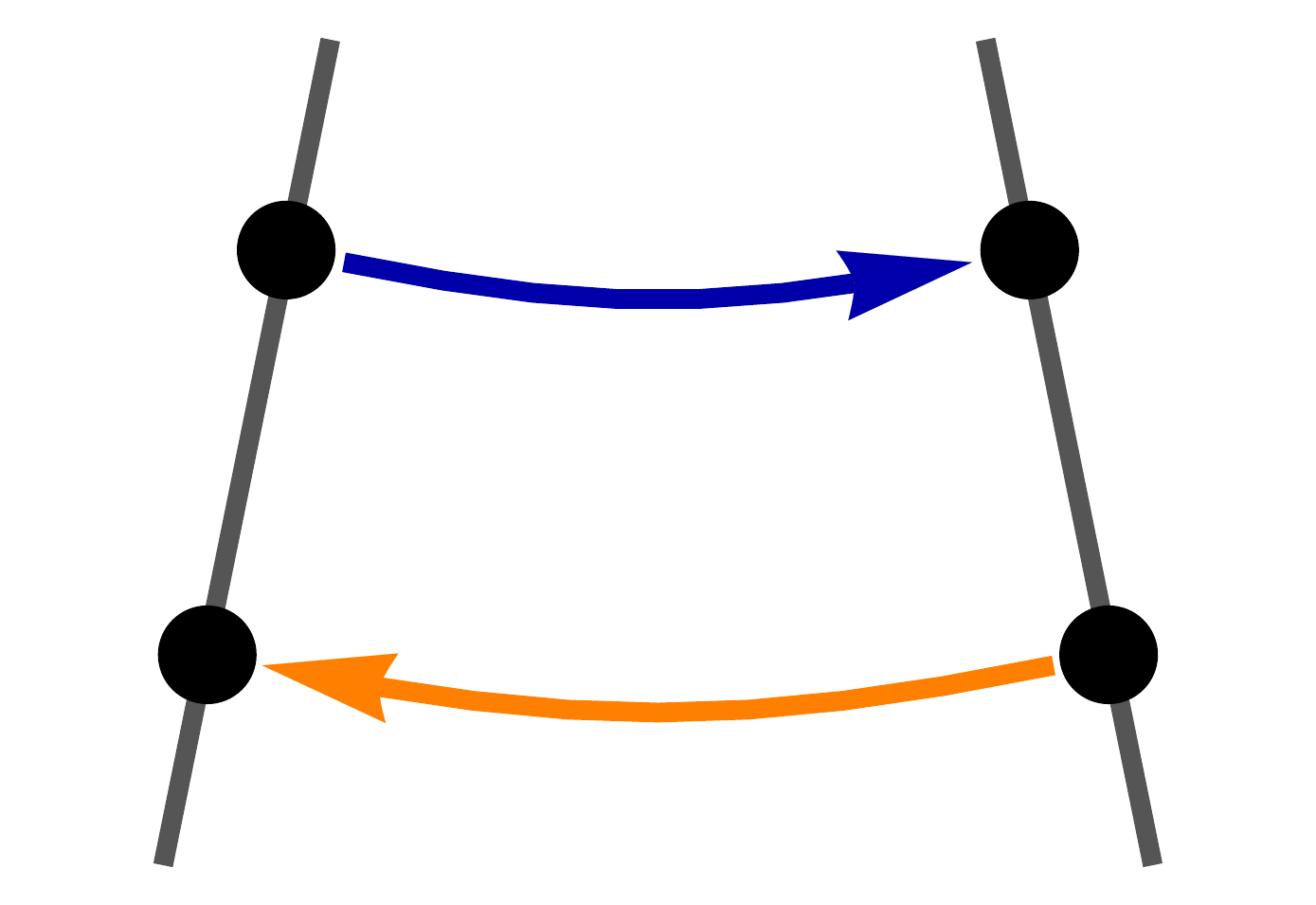}
\end{gathered}$
 & $\frac{\ket{0} \ket{1} + \ket{1} \ket{0}}{\sqrt{2}}$
&$\begin{gathered}
\includegraphics[height=0.04\textheight]{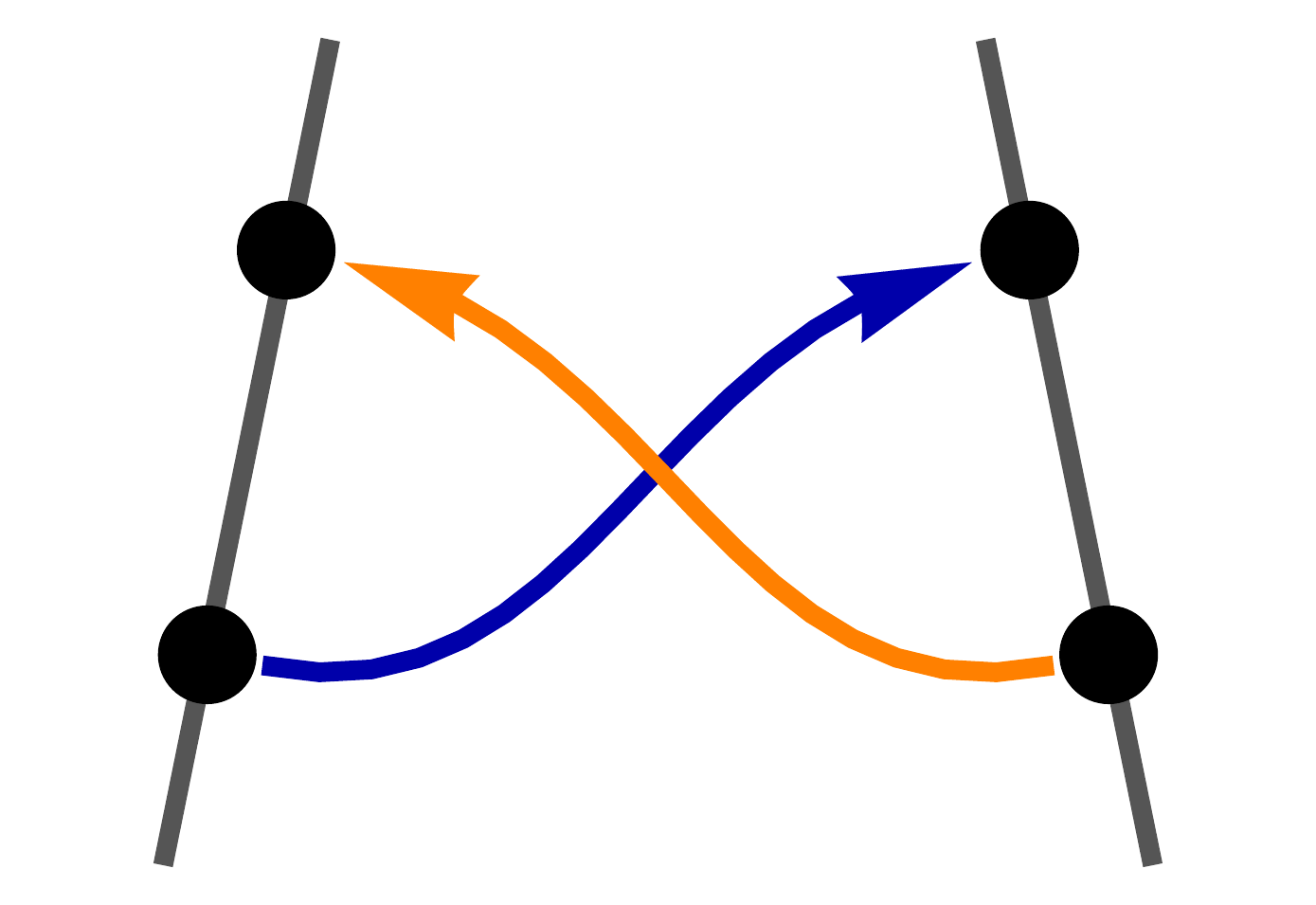}
\end{gathered}$
 & $\frac{\ket{0} \ket{0} + \i \ket{1} \ket{1}}{\sqrt{2}}$ \\
$\begin{gathered}
\includegraphics[height=0.04\textheight]{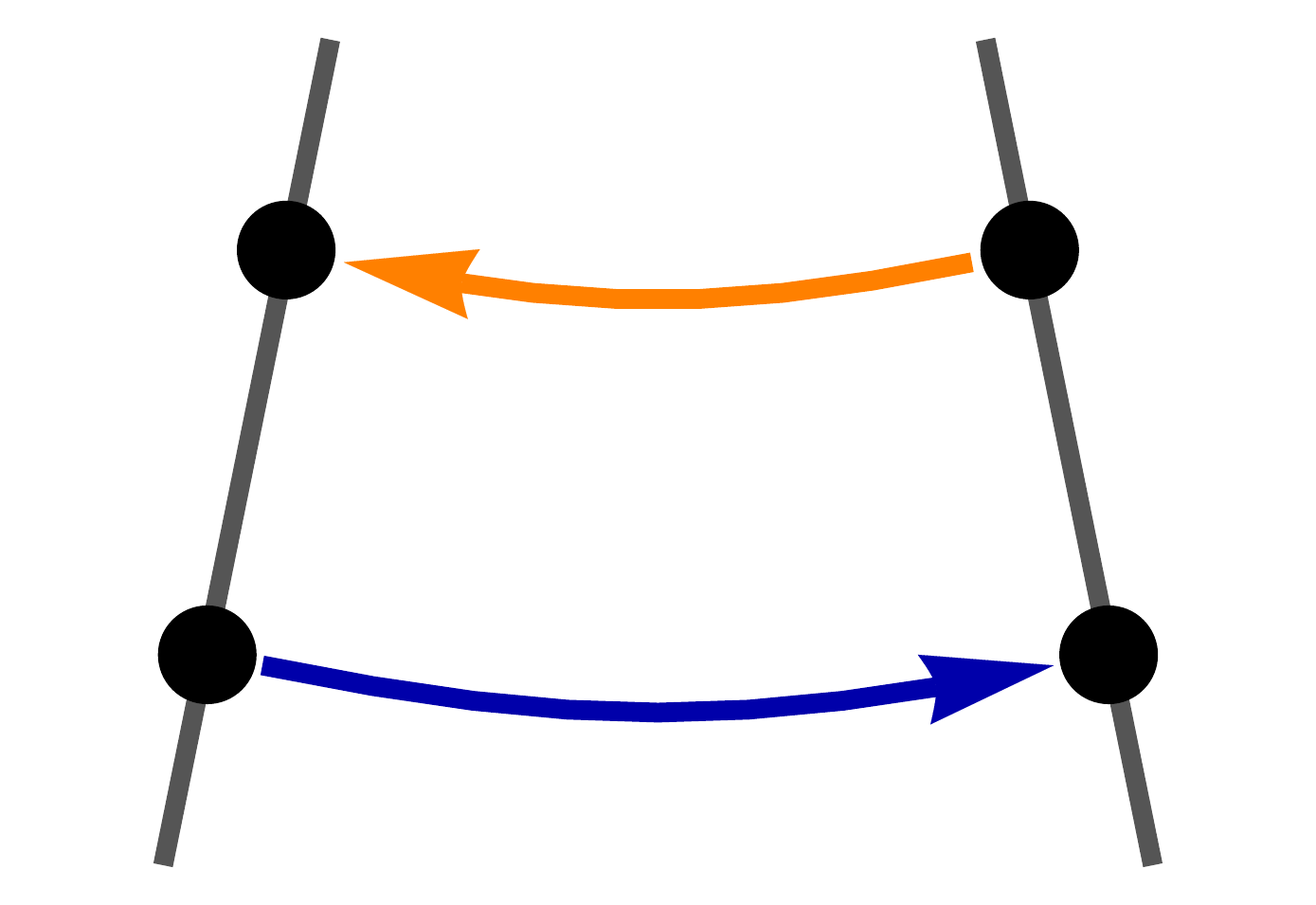}
\end{gathered}$
 & $\frac{\ket{0} \ket{1} - \ket{1} \ket{0}}{\sqrt{2}}$
&$\begin{gathered}
\includegraphics[height=0.04\textheight]{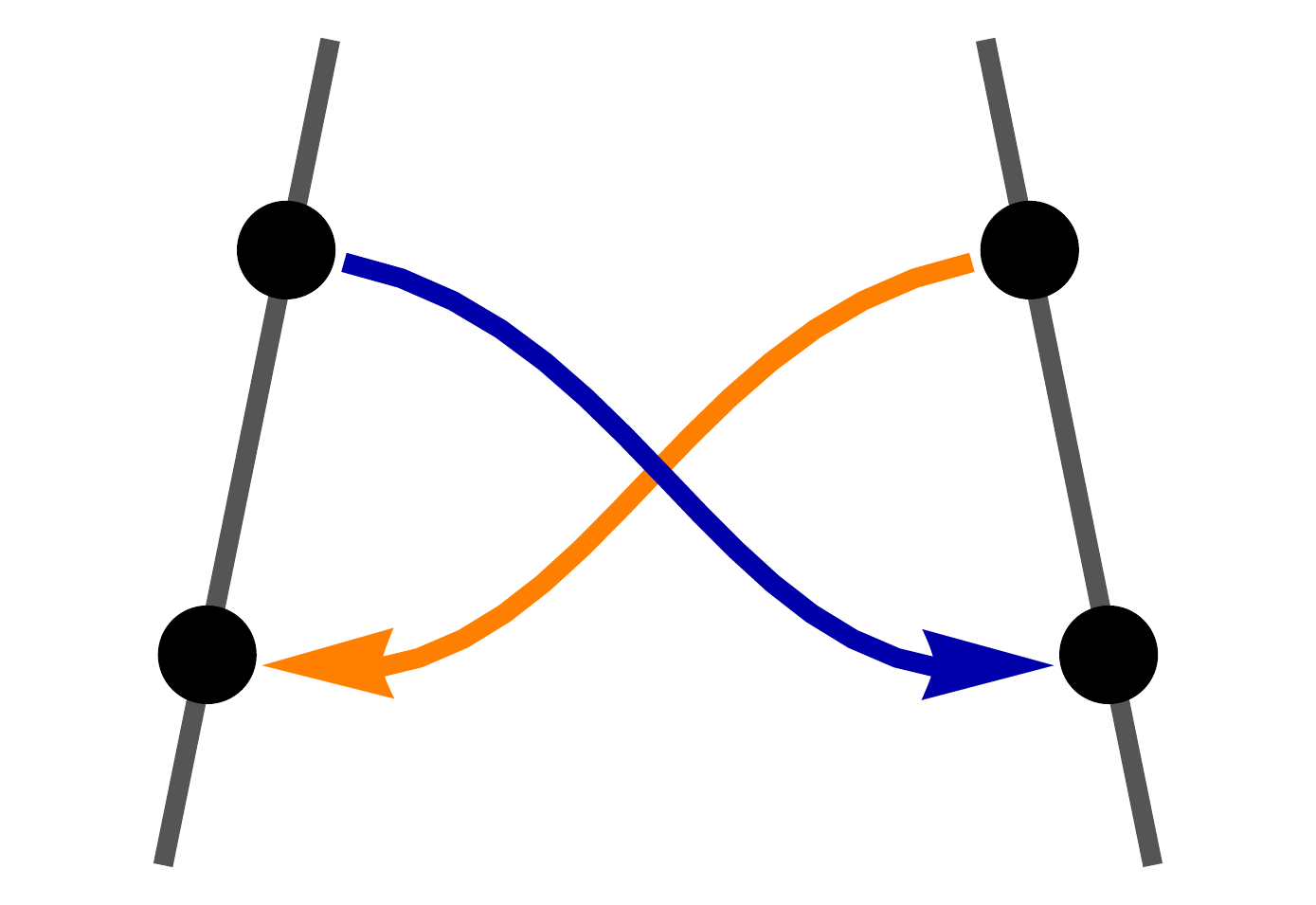}
\end{gathered}$
 & $\frac{\ket{0} \ket{0} - \i \ket{1} \ket{1}}{\sqrt{2}}$ \\
\end{tabular}
\caption{Bell states expressed as Majorana dimers.}
\label{TAB_BELL_STATES}
\end{table}

\subsection{Orthogonality and completeness}

Our diagrammatic notation can also express inner products. Consider the bra $\bra\psi$ corresponding to a ket $\ket\psi$. Clearly, if $(\m_j + \i\, \m_k) \ket\psi = 0$ then $\bra\psi (\m_j - \i\, \m_k) = 0$ holds for the adjoint. Thus, we can visualize adjoints by inverting all arrows and corresponding parities, for example (omitting labels):
\begin{align}
\label{EQ_BRA_AND_KET}
\ket\psi \;&=\;
\begin{gathered}
\includegraphics[height=0.09\textheight]{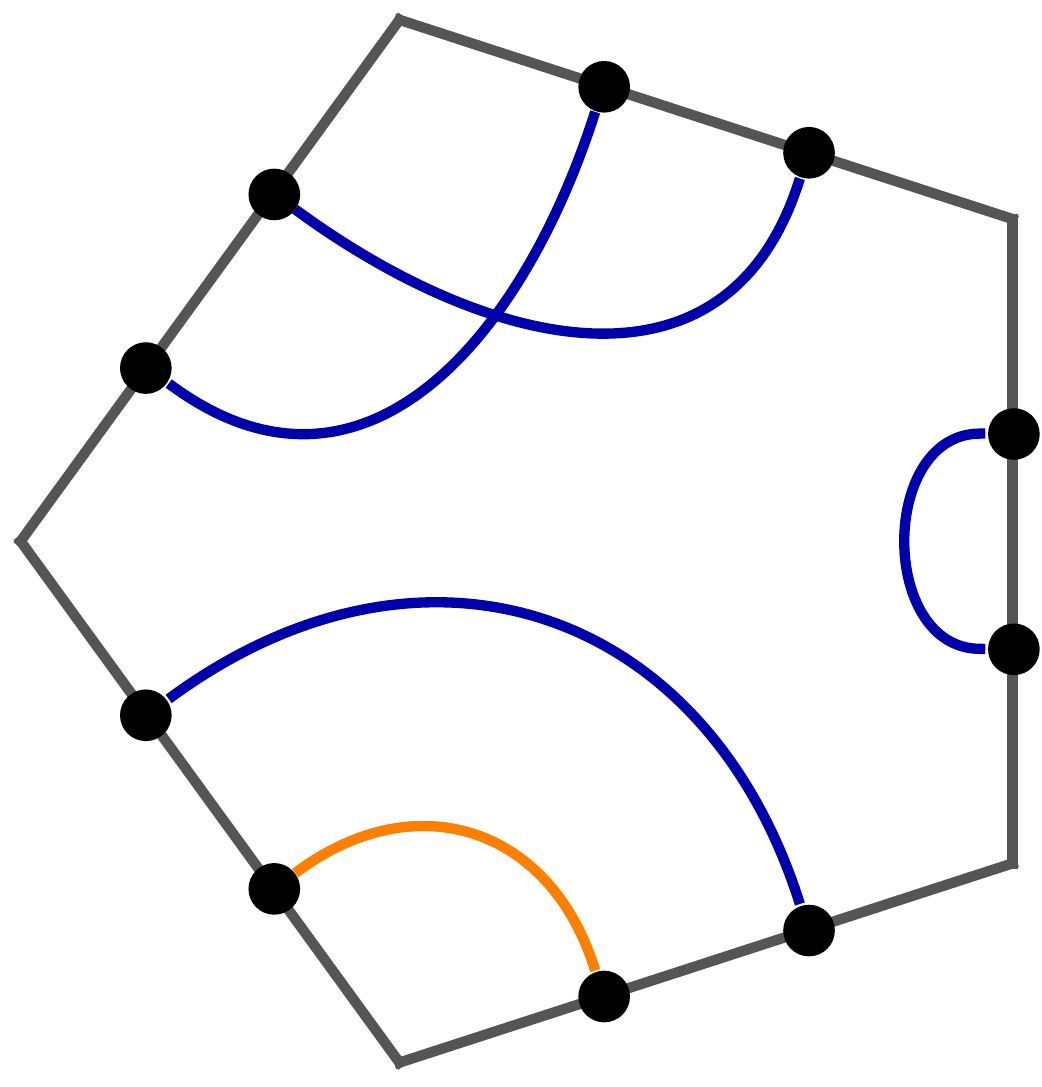}
\end{gathered}&
\bra\psi\;&=\;
\begin{gathered}
\includegraphics[height=0.09\textheight]{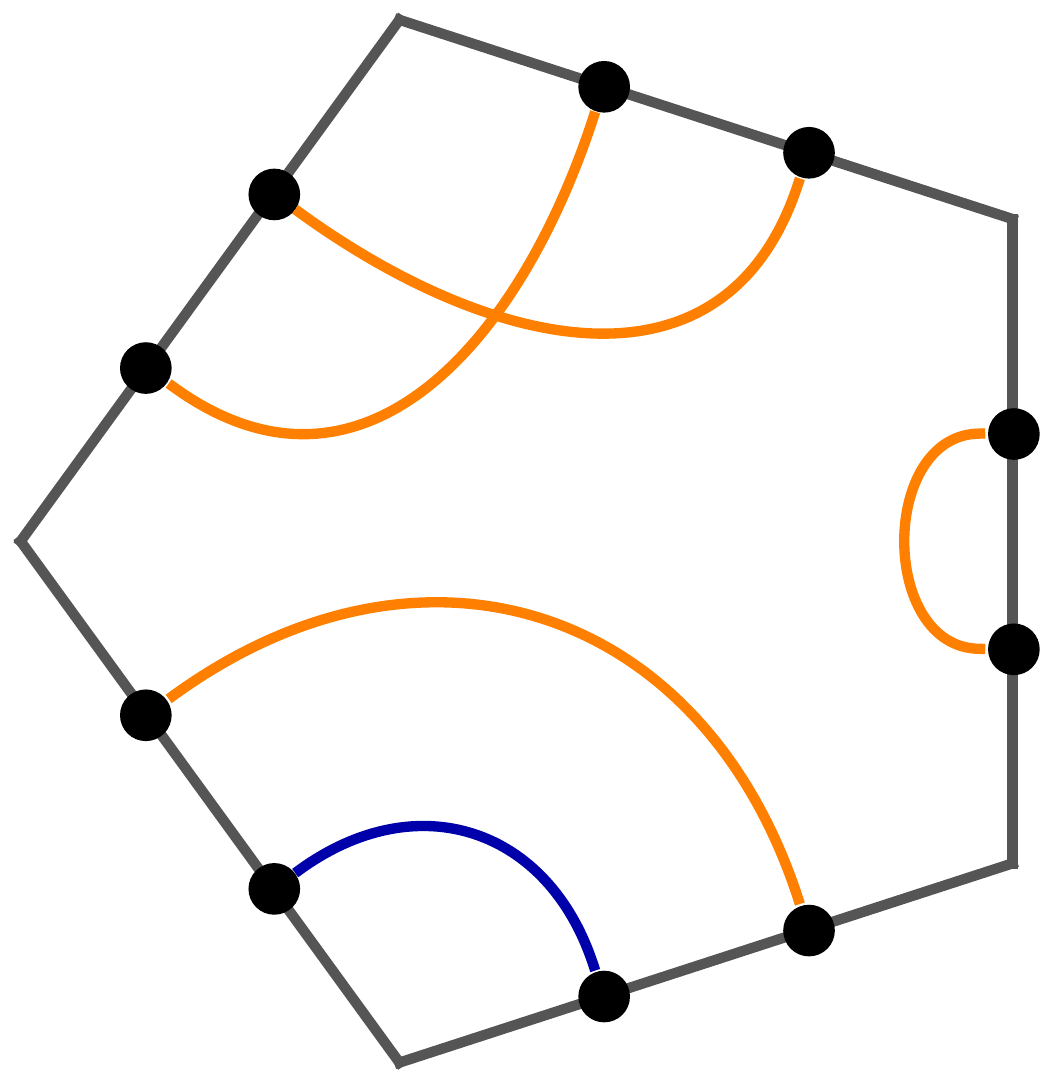}
\end{gathered}
\end{align}
The inner product $\braket\psi\psi$ is a contraction between $\ket\psi$ and $\bra\psi$ over all indices, expressed 
as
\begin{align}
\label{EQ_INNER_PROD}
\braket\psi\psi \;=\;
\begin{gathered}
\includegraphics[height=0.09\textheight]{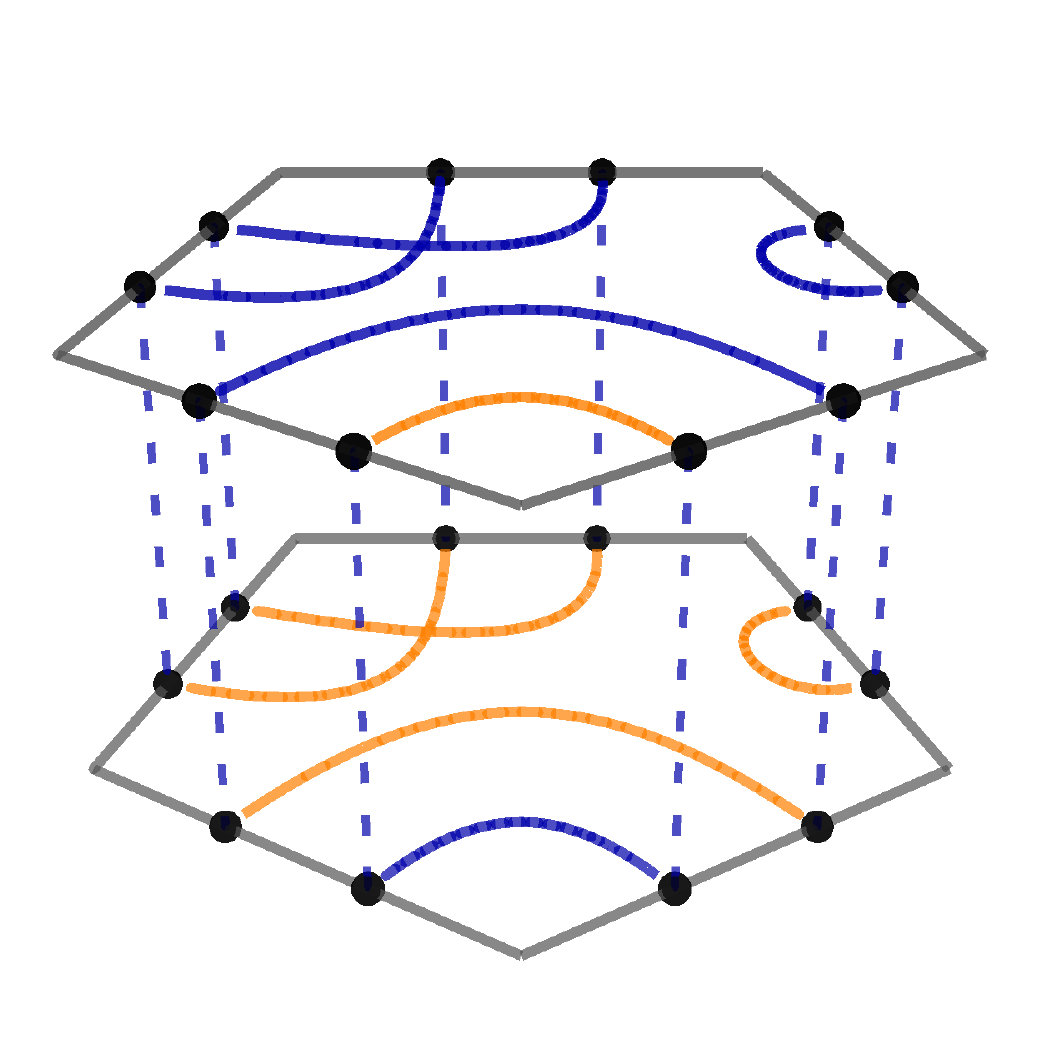}
\end{gathered}
\;=\;
\begin{gathered}
\includegraphics[height=0.1\textheight]{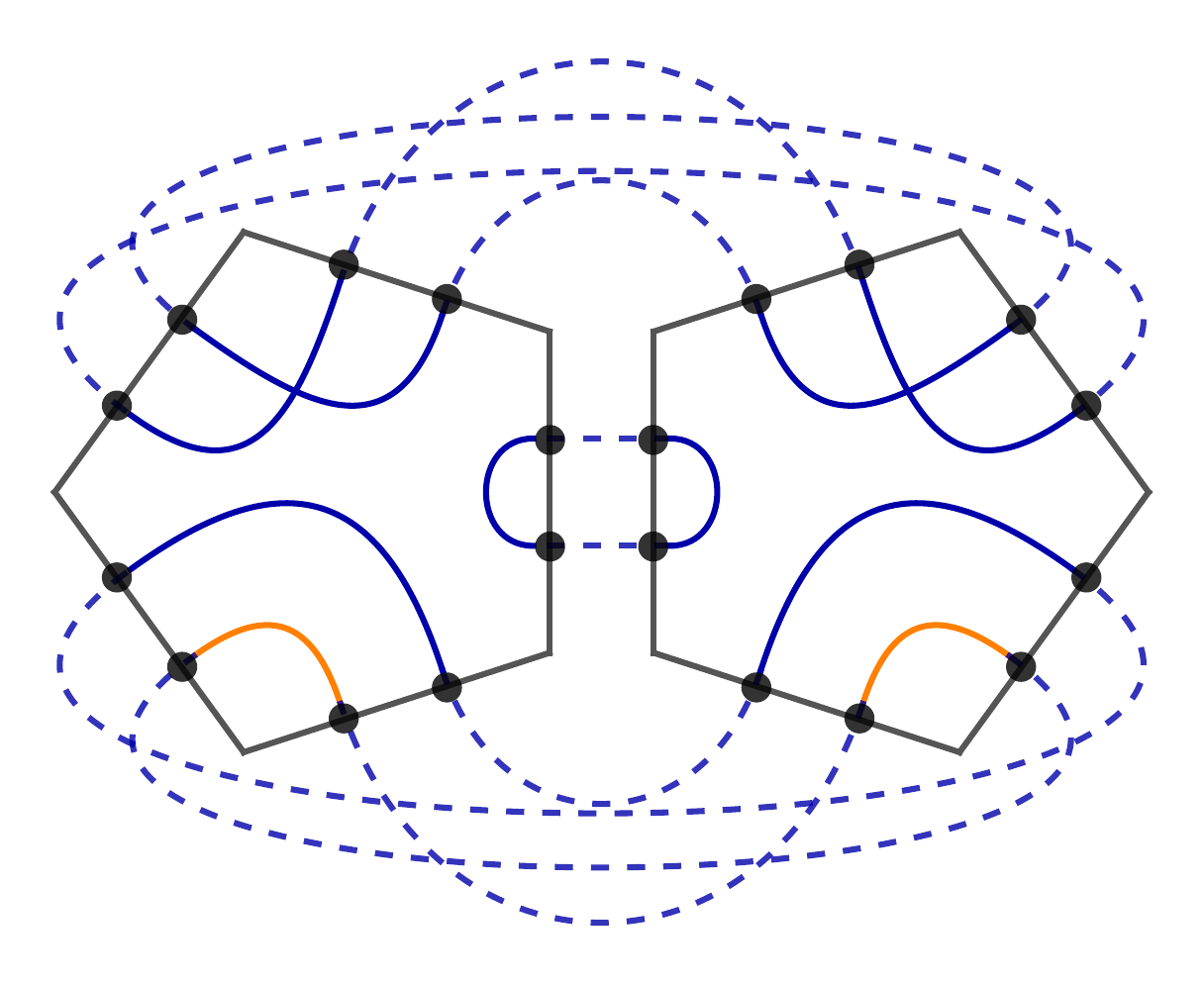}
\end{gathered}
\end{align}
The right-hand side, showing a ``flipped'' $\bra\psi$, represents a Choi-Jamiolkowski isomorphism expressing $\bra\psi$ in the same Hilbert space as $\ket\psi$. This involves an inversion of the orientation which flips all dimer parities. As all diagrams are defined for normalized state vectors that satisfy $\braket\psi\psi=1$.
Furthermore, we can easily evaluate whether two diagrams correspond to orthogonal states as a contraction $\braket\psi\phi$ of $\ket\psi$ and $\ket\phi$ vanishes for any odd-parity loop (see Appendix \ref{APP_CONTR_RULES}). In particular, $\ket\psi$ and $\ket\phi$ are orthogonal if they share the same correlation structure (i.e.,\ pairing of Majorana modes) but differ in at least one dimer parity.
This allows us to construct the complete Hilbert space with Majorana dimer states on $N$ edges by fixing a correlation structure and then flipping through all $2^N$ possible dimer parities, resulting in $2^N$ mutually orthogonal state vectors.
Since the Hilbert space is also $2^N$-dimensional, we can express any state in it by a superposition of Majorana dimer states under the given correlation structure.
This is equivalent to obtaining a orthogonal stabilizer state basis by considering all $2^N$ sign combinations of $N$ stabilizer generators.

\section{The HyPeC with Majorana dimers}
\subsection{Overview}

As we saw in the previous section, the computational basis logical code states of the $[[5,1,3]]$ quantum error-correcting code can be expressed as Majorana dimers. Furthermore, we showed that identifying Majorana dimer states as tensors and contracting them yields new Majorana dimer states, and that these contractions are easy to evaluate diagrammatically.
Because the HyPeC is built from tensors each representing the $[[5,1,3]]$ code, we find the following key result:
\begin{theorem}[Representing the HyPeC with Majorana dimers]
The hyperbolic pentagon code (HyPeC) with logical bulk input fixed to local basis states $\bar{0}$ or $\bar{1}$ yields a Majorana dimer state on the boundary.
Each input corresponds to a (unique) pattern of dimer parities on the boundary state.
\end{theorem}
While fermionic modes require an explicit ordering, we show in Appendix \ref{APP_CONTR_ORDER} that different contraction orderings lead to equivalent boundary states.
We will now show how the geometry of the dimers in the HyPeC determines its properties, using the tools developed in the previous section.

\subsection{Dimers and entanglement structure}

First, we will consider the physical properties of the HyPeC for logical inputs fixed locally to $\bar{0}$ or $\bar{1}$. 
The code is constructed from a hyperbolic $\{5,4\}$ tiling\footnote{This \emph{Schl\"afli symbol} denotes a polygon tiling with four pentagons at each vertex.}, with each tile now set to \eqref{HAPPY_ZERO} or \eqref{HAPPY_ONE} (the full HyPeC also allows for superpositions between the two). 
As the model consists of asymptotically infinite tiles, we have to define a UV cutoff at which the tiling is truncated. We do this by starting with a central tile and iteratively adding tiles on all free edges. The number $n$ of iterations thus gives the graph distance between each boundary tile and the centre, determining the cutoff. Such a series of iterations for an all-$\bar{0}$ bulk input is visualized in Fig.\ \ref{FIG_HAPPY_CONTR_SERIES}.

\begin{center}
\begin{figure*}[htb]
\begin{equation*}
\begin{gathered}
n=0:\\
\vspace{7pt}
\includegraphics[height=0.14\textheight]{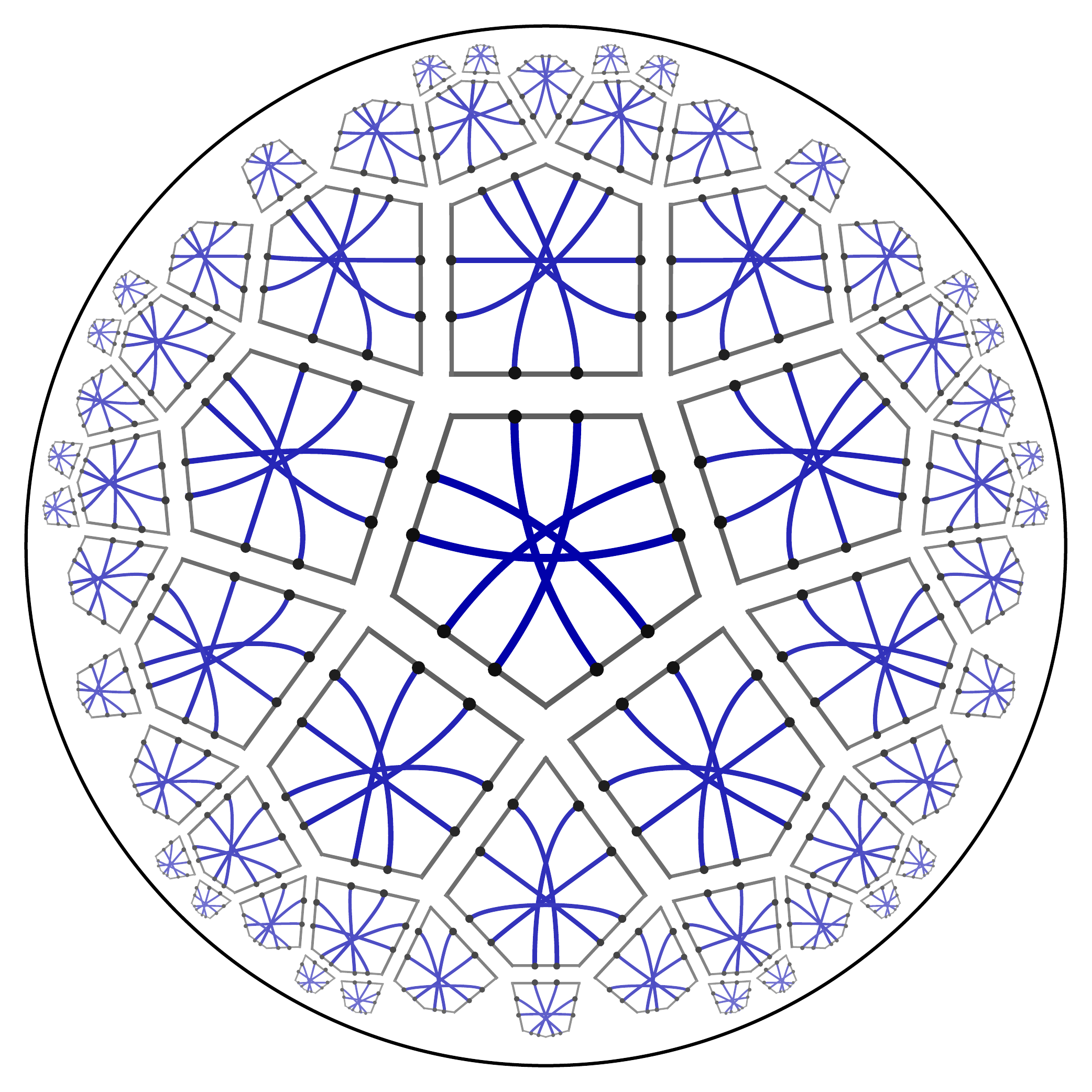}
\end{gathered}
\scalebox{1.3}{$\quad\rightarrow\quad$}
\begin{gathered}
n=1:\\
\vspace{7pt}
\includegraphics[height=0.14\textheight]{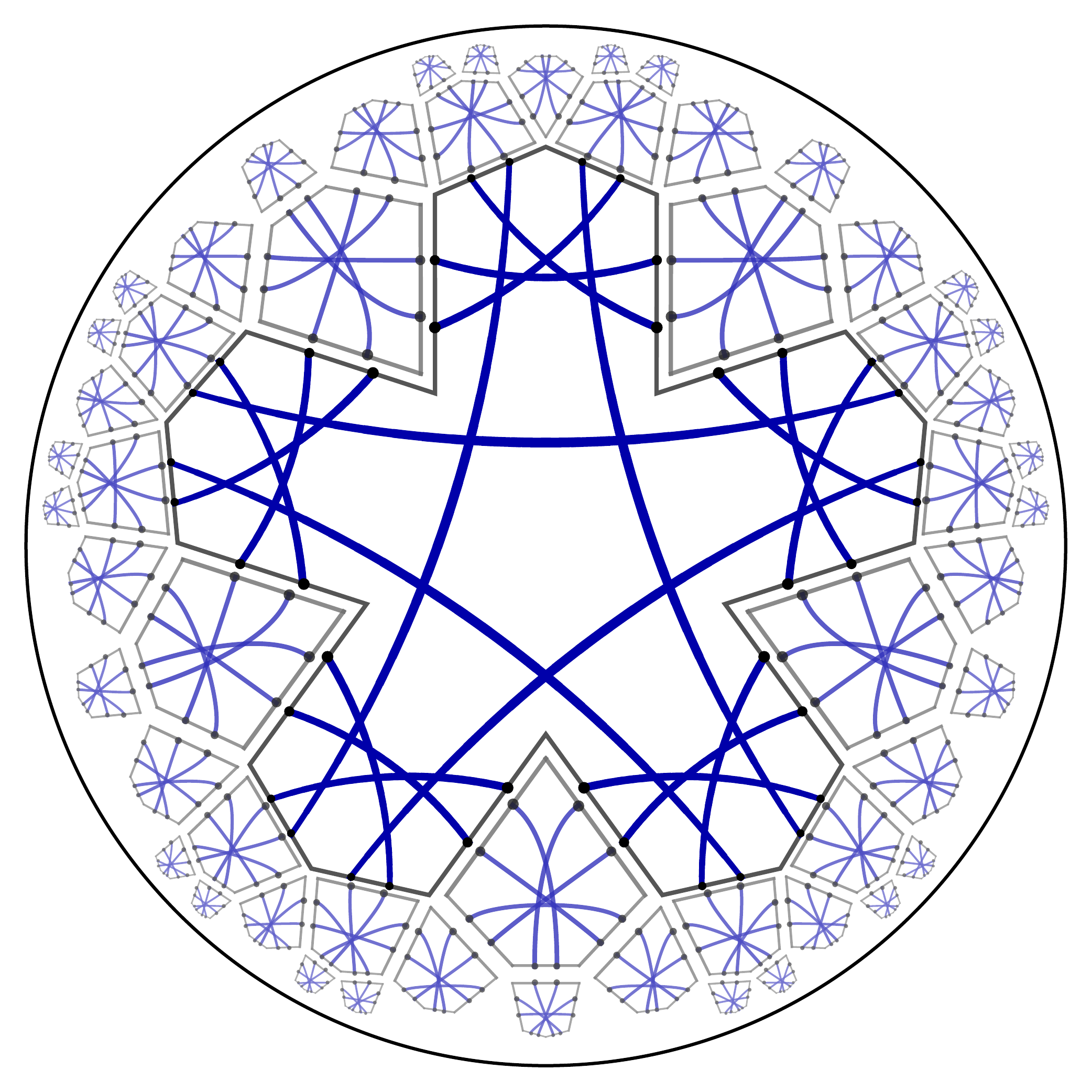}
\end{gathered}
\scalebox{1.3}{$\quad\rightarrow\quad$}
\begin{gathered}
n=2:\\
\vspace{7pt}
\includegraphics[height=0.14\textheight]{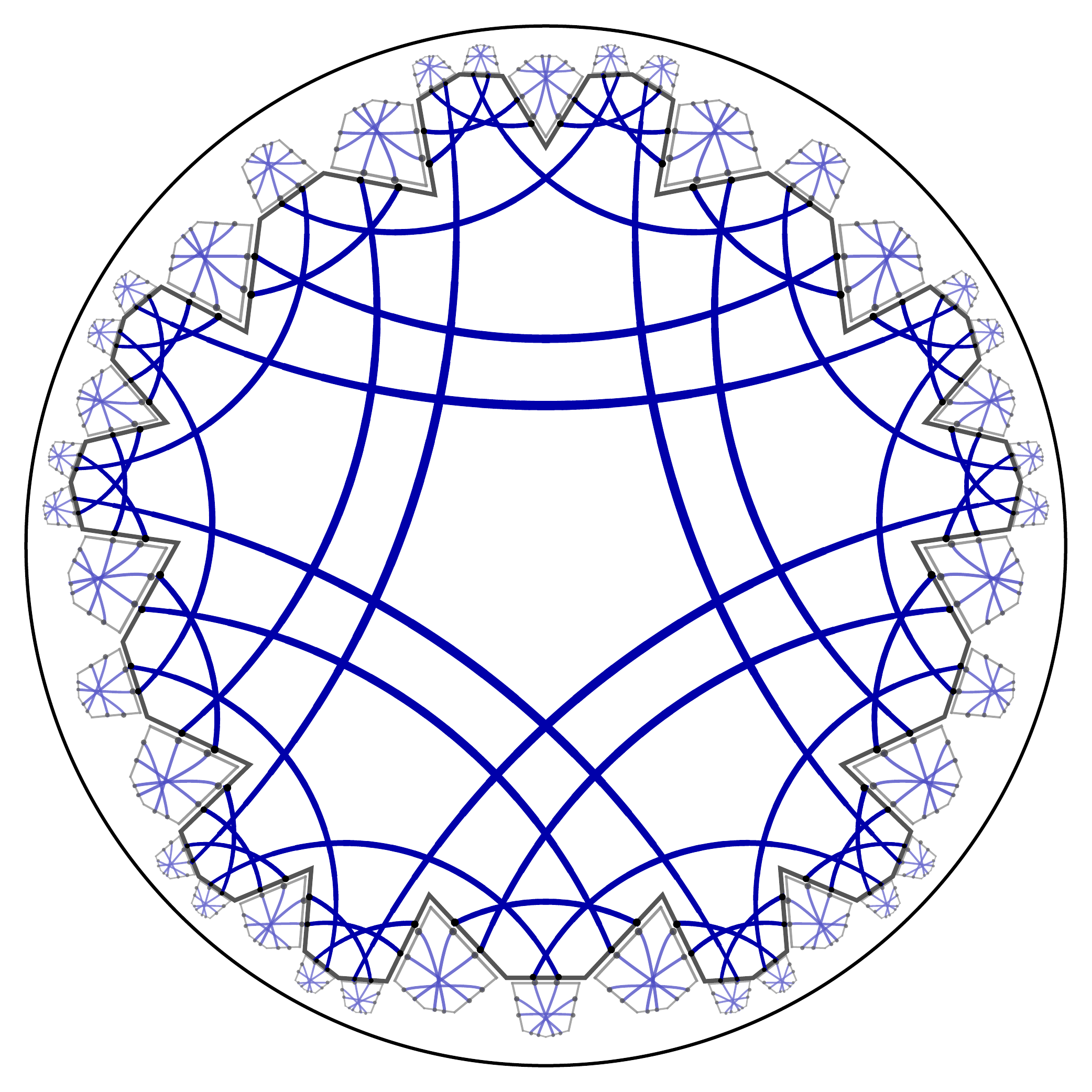}
\end{gathered}
\scalebox{1.3}{$\quad\rightarrow\quad$}
\begin{gathered}
n=3:\\
\vspace{7pt}
\includegraphics[height=0.14\textheight]{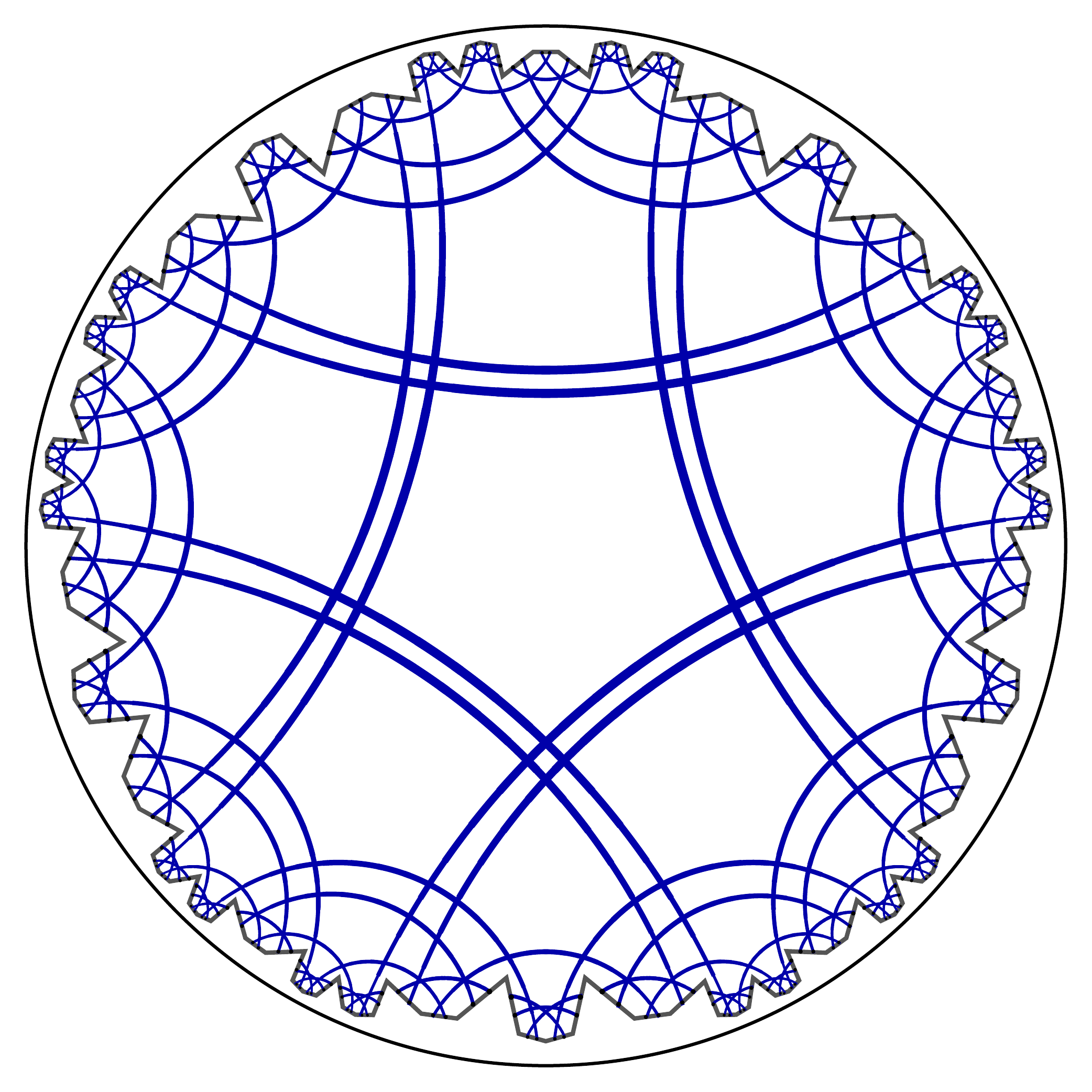}
\end{gathered}
\end{equation*}
\caption{Iterative contraction of the HyPeC for fixed bulk inputs of $\ket{\bar{0}}_5$ state vectors, with a cutoff after 3 iterations. Each step involves contracting a further layer of tiles, starting from the centre at $n=0$. The asymptotic boundary of the Poincar\'e disk is drawn as a black circle.
}
\label{FIG_HAPPY_CONTR_SERIES}
\end{figure*}
\end{center}

\vspace{-30pt}
The contracted dimers are drawn as geodesics in the Poincar\'e disk. This is not an arbitrary choice, as the dimers follow discrete geodesics (i.e.,\ shortest paths) in the $\{5,4\}$ tiling. Fig.\ \ref{FIG_HAPPY_DIMER_ASYMPT} shows the $n \to \infty$ limit both for an example of two dimers and the whole contraction.
Because of the particular property of the $\lbrace 5,4 \rbrace$ tiling that the pentagon edges connect to continuous geodesics, the asymptotic endpoints of a contracted dimer are also endpoints of such a geodesic.\footnote{In the dual $\lbrace 4,5 \rbrace$ tiling, each four-sided tile contains an intersection of two such geodesics meeting at right angles.} 
Tracing this geodesic back into the bulk, we see that it passes along all tiles that contained the uncontracted dimer pieces. Furthermore, as Fig.\ \ref{FIG_HAPPY_DIMER_ASYMPT} also shows, there are always two dimers with the same pair of asymptotic boundary points, resulting in a bulk geodesic that is dual to a boundary Bell state.
While Fig.\ \ref{FIG_HAPPY_DIMER_ASYMPT} only shows a uniform $\bar{0}$ bulk input, 
the dimer parities generally differ with the input. The dimer pairs then correspond to different types of Bell states, as in Table \ref{TAB_BELL_STATES}.
Note that the Majorana modes composing the effective fermions of these Bell states are located on neighbouring boundary edges at any finite cutoff.
This elucidates the code's error correction properties: Any product of pairs of Majorana operators $\i\,\m_j \m_k$ acting on dimer endpoints $(j,k)$ can only change the state up to a total sign, and is thus a representation of a logical parity operation in the bulk. While single Majorana operators are nonlocal in terms of spins, pairs of Majorana operators on neighboring sites can be expressed by a local pair of Pauli operators (compare \eqref{EQ_JORDANWIGNER_REV}). For each pair of HyPeC dimers, there then exists a boundary operator $\mathcal{O}$ of weight $|\mathcal{O}|=4$, i.e., consisting of four Pauli operators acting on the boundary, which represents a logical bulk operation. Thus, even for an infinitely large number of HyPeC tiles, the code distance $d$ never exceeds $d=3$, as such an $\mathcal{O}$ represents an error on the code space.	
\begin{figure*}
\begin{equation*}
\begin{gathered}
\includegraphics[height=0.14\textheight]{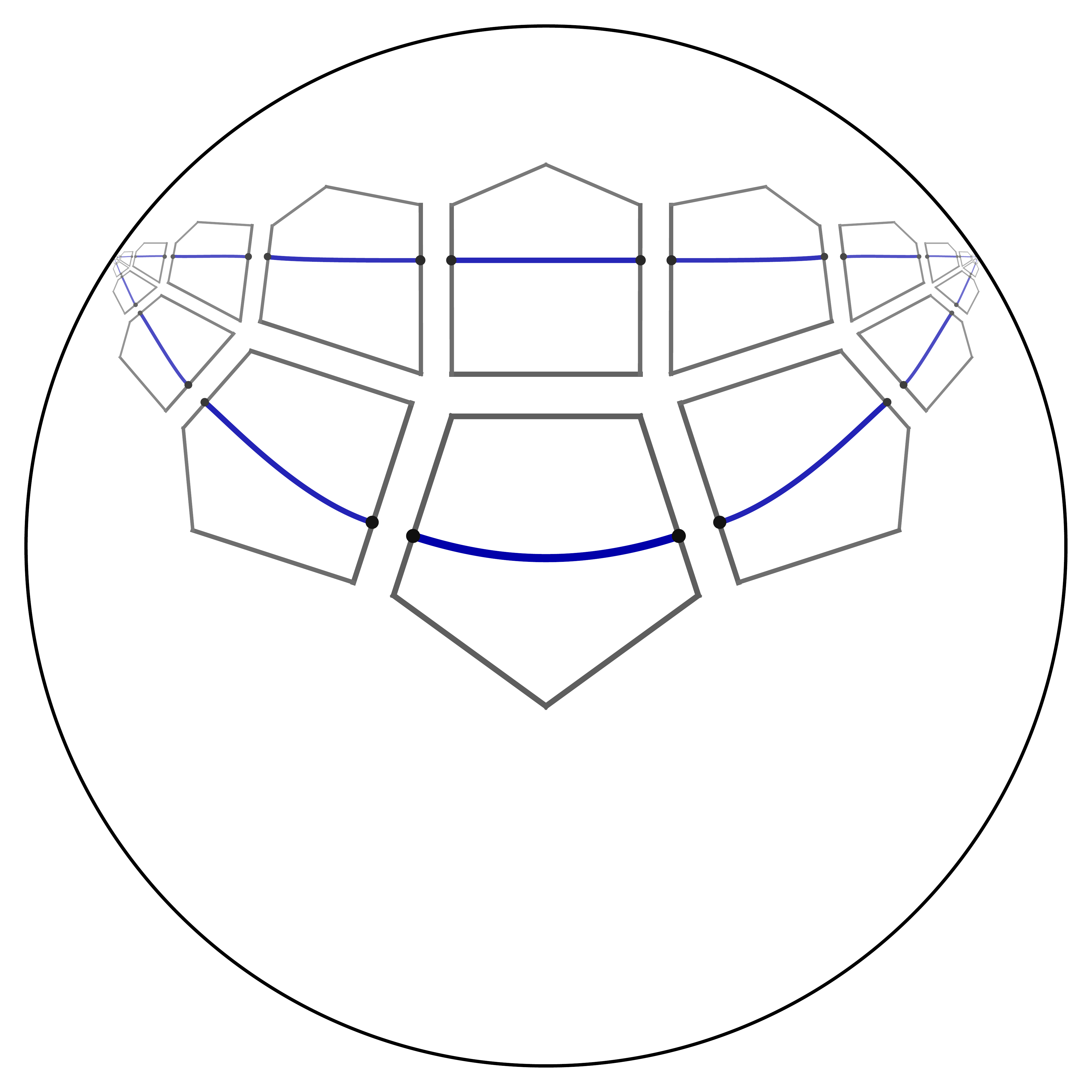}
\end{gathered}
\scalebox{1.3}{$\quad\rightarrow\quad$}
\begin{gathered}
\includegraphics[height=0.14\textheight]{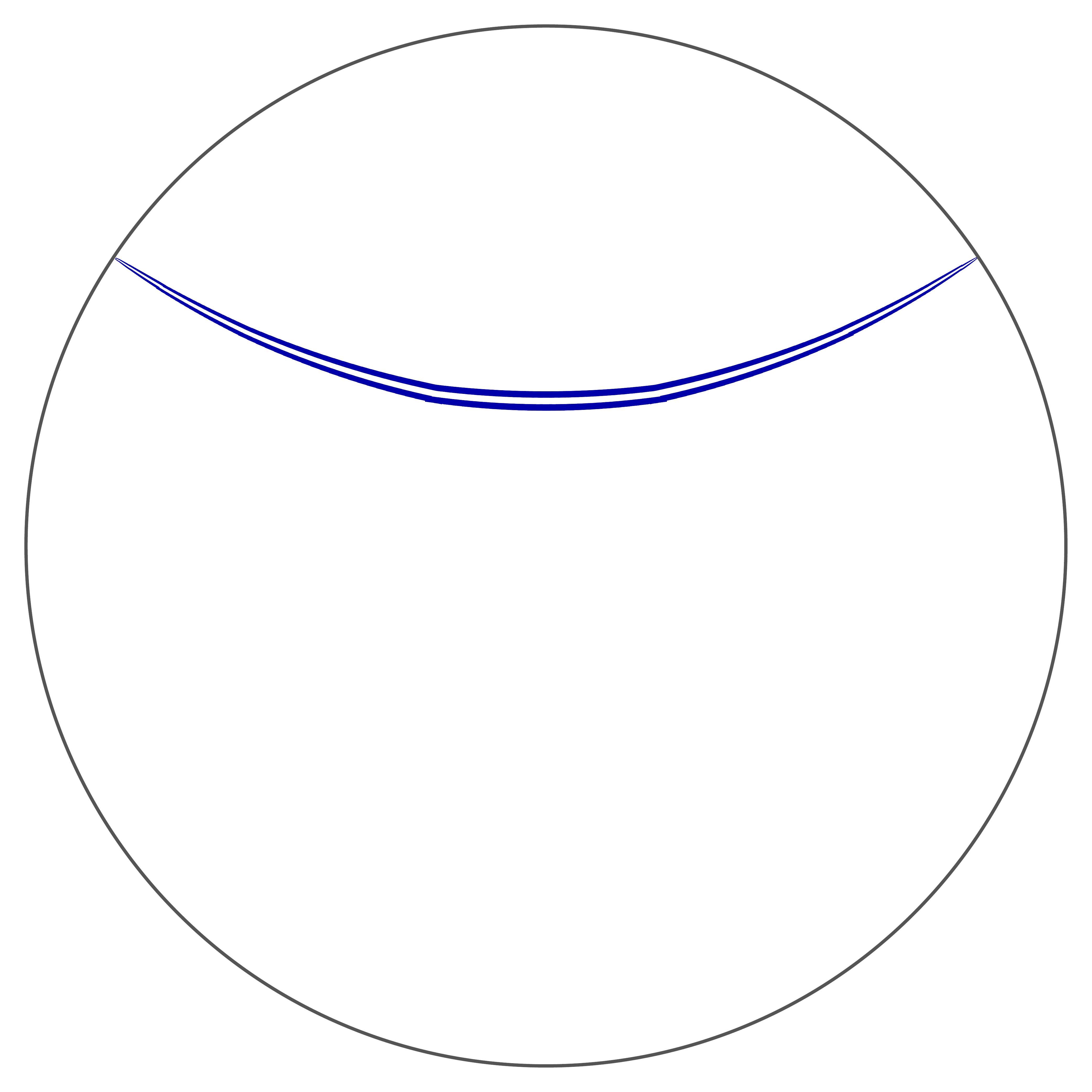}
\end{gathered}
\hspace{1.2cm}
\begin{gathered}
\includegraphics[height=0.14\textheight]{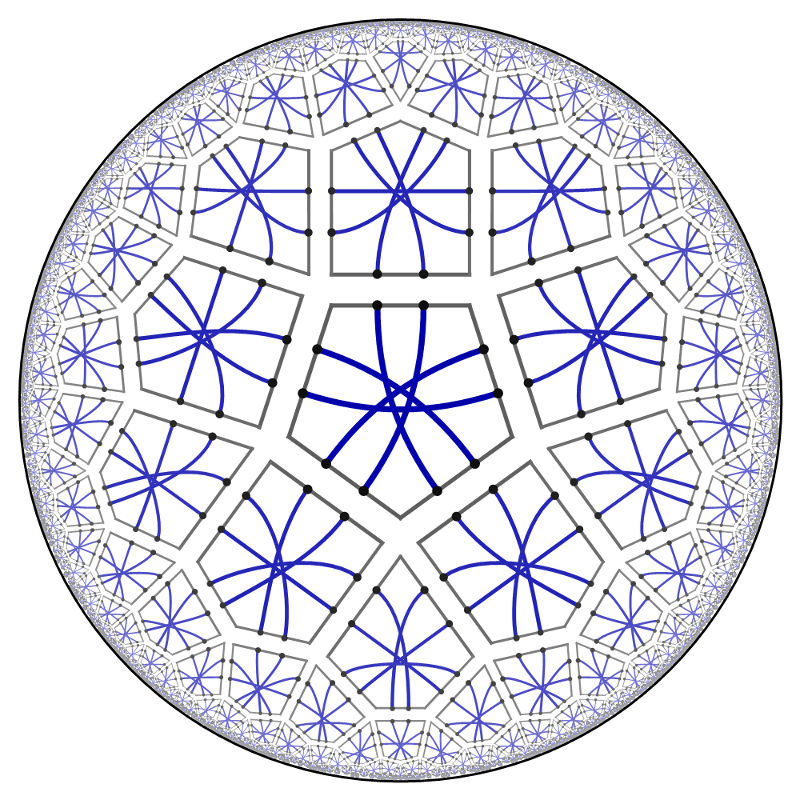}
\end{gathered}
\scalebox{1.3}{$\quad\rightarrow\quad$}
\begin{gathered}
\includegraphics[height=0.14\textheight]{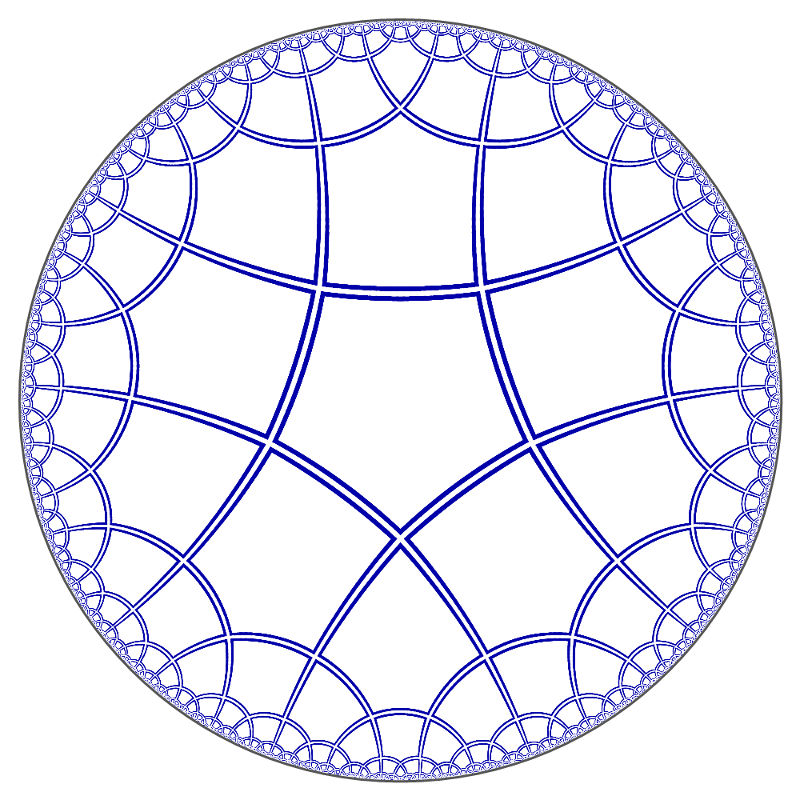}
\end{gathered}
\end{equation*}
\caption{\textsc{Left:} A Majorana dimer pair in an infinitely large contraction of HyPeC tiles. The endpoints of both dimers meet at the asymptotic boundary, thus the dimer pair can be drawn as a double-geodesic. 
\textsc{Right:} Full contraction for a $\bar{0}$ input on all tiles, with all dimers pairing up.}
\label{FIG_HAPPY_DIMER_ASYMPT}
\end{figure*}

Given this picture of pair entanglement on the boundary through bulk geodesics, the dependence of the entanglement entropy $S_A$ on the boundary subsystem size $|A|=\ell$ can be explicitly computed. Clearly, the position of this subsystem affects the value of $S_A$, as the distribution of entangled pairs in Fig.~\eqref{FIG_HAPPY_DIMER_ASYMPT} does not respect translation invariance on the boundary. Thus, we consider the average expectation value $\mathbb{E}_\ell(S)$ of the entanglement entropy instead. The results in Fig.~\ref{FIG_HAPPY_EE} show an approximate logarithmic growth $S_A \propto \log\ell$, as expected of a critical theory. Fitting against the expected logarithmic scaling expected for a CFT \cite{CalabreseReview},
\begin{equation}
\label{EQ_CALABRESE_CARDY}
S_A = \frac{c}{3} \log \left( \frac{L}{\pi \epsilon} \sin\frac{\pi \ell}{L} \right) \simeq \frac{c}{3} \log \frac{\ell}{\epsilon} + O\left( (\ell/L)^2 \right) \text{ ,}
\end{equation}
we find a central charge $c \approx 4.2$ (dashed line in Fig.\ \ref{FIG_HAPPY_EE}). For a finite system of boundary size $L$, $S_A$ reaches its maximum at $\ell = L/2$, in agreement with the full form of \eqref{EQ_CALABRESE_CARDY}. Each iteration preserves the entanglement entropy scaling of the previous one up to $\ell \approx L/4$. We already observed this behaviour in a previous analysis using \emph{matchgate tensors} \cite{Jahn:2017tls}.

\begin{figure}[tb]
\centering
\includegraphics[height=0.18\textheight]{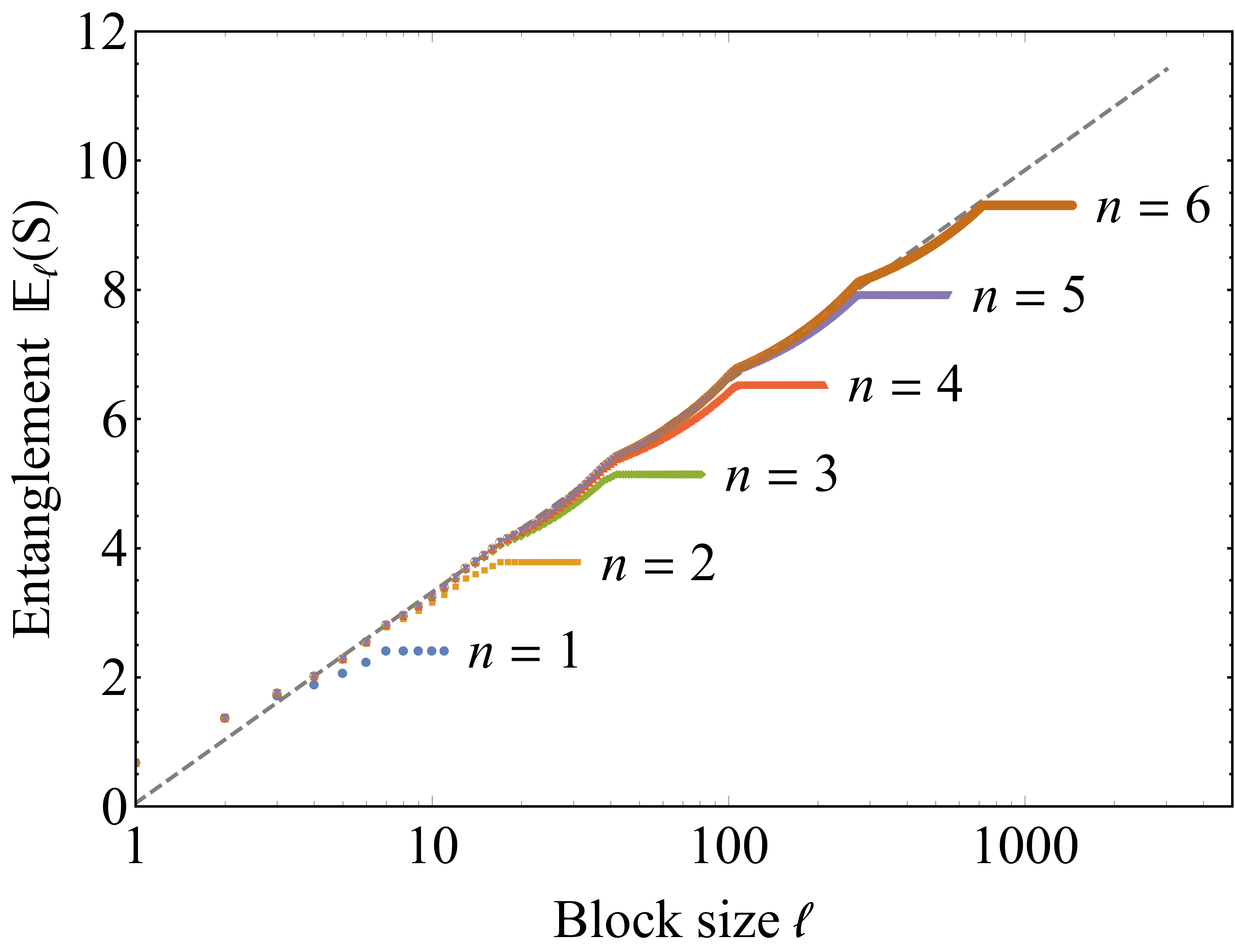}
\caption{Entanglement entropy scaling with the size $\ell$ of the subsystem block for successive iterations $n$ of the contraction. Dashed line is a logarithmic fit of the $n\to\infty$ limit.}
\label{FIG_HAPPY_EE}
\end{figure}

The logarithmic entanglement entropy scaling saturates the bound we observed in \eqref{EQ_DIMER_EE_EXAMPLE}:
The maximum entanglement between two boundary regions $A$ and $A^{\text{C}}$ is proportional to the maximum number of dimers that can connect them, or equivalently, the number of edges $|\gamma_A|$ of a minimal cut $\gamma_A$ through the bulk separating $A$ from $A^{\text{C}}$. Due to the hyperbolic geometry of the $\{ 5, 4 \}$ tiling, $|\gamma_A| \propto \log |A|$. As $\gamma_A$ is a geodesic in this discrete geometry, no other geodesic -- and thus no dimer -- can pass through it twice (up to cases such as in Fig.~\ref{FIG_GREEDY2}), turning the upper bound into an equality:
\begin{equation}
\label{EQ_EE_BOUND_SAT}
S_A =  |\gamma_A|\,\log 2 \propto \log\ell \text{ .}
\end{equation}
Clearly, this result is the same for each computational basis state input, as changing this input only changes dimer parities, leaving $S_A$ invariant.

We can modify the IR structure of our boundary states by modifying the central tiles. There are two possible approaches: One is the replacement of the dimer structure in these tiles, and the other is the complete removal of the tiles. In the first case, we simply reconnect the dimers with each other, so that they no longer follow geodesics. This breaks the conditions for the saturation of the bound \eqref{EQ_EE_BOUND_SAT}, so that we reduce the entanglement of the boundary states. 
The more tiles in the centre are changed (e.g.\ to the vacuum \eqref{EQ_PENTA_VAC}), the further long-range entanglement is suppressed, so that we approach a \textsl{gapped} boundary state with constant (i.e., area-law) entanglement.

The second case was already considered in the original HyPeC model: When removing entire tiles, auxiliary bulk degrees of freedom (open edges, or open legs in tensor network language) are opened up. The usual interpretation of this setup is that of a black hole, or when extending the open edges to a non-contracted auxiliary system, that of a \textsl{wormhole}. In both cases, the boundary state of this setup should exhibit an additional thermal entropy, which the Ryu-Takayanagi formula interprets as a deformation of geodesics around the horizon. 
In the language of Majorana dimers, this additional entropy is explained by dimers ending on the open edges: Following \eqref{EQ_EE_SYMPL}, any dimer in a region $A$ that connects to a site outside of $A$ contributes $\log(2)/2$ to the entanglement entropy $S_A$. When $A$ becomes large, this entropy contribution scales with the ``horizon area'' of the black hole, i.e.,\ the number of fermionic modes on the open edges. As we increase the radius of the black hole, $S_A$ will begin to grow linearly with the size of $A$, as expected of a thermal CFT.

As we show in Appendix \ref{APP_EE_RULES}, the entanglement entropy $S_A$ of compact subsystems $A$ of the HyPeC for logical bulk input can be generalized to arbitrary \emph{local} input, i.e.\ superpositions of $\bar{0}$ and $\bar{1}$ that factorize between the tiles. 
Without additional entanglement between bulk sites, $S_A$ is independent of the specific bulk input, as long as the boundary regions $A$ and its complement $A^{\text{C}}$ are reducible to the same discrete bulk geodesic $\gamma_A$ via the \emph{greedy algorithm} \cite{Pastawski2015}, which can be completely re-derived using Majorana dimers. This algorithm iteratively removes tiles with three or four open edges (see Fig.\ \ref{FIG_GREEDY}), deforming $A$ into a region $A^\prime$ further in the bulk, while keeping $S_A=S_{A^\prime}$ invariant.
Fig.\ \ref{FIG_GREEDY2} illustrates how some boundary regions $B$ are not reducible to the same geodesic $\gamma_B$ as their complement regions $B^\text{C}$. In these cases, no cut along the pentagon edges can completely separate dimers with endpoints in $B$ from those with endpoints in $B^\text{C}$, leaving dimers in a \emph{residual bulk region}.
While \eqref{EQ_DIMER_EE} still holds if local bulk inputs are fixed to basis states $\bar{0}$ or $\bar{1}$, the entanglement entropy $S_B$ for local superpositions can generally be larger, as there is additional input-dependent entanglement between the residual dimers. For example, in the setup of Fig.\ \ref{FIG_GREEDY2} (bottom), $S_B$ can be up to $\log 2$ larger than the fixed-input result (see Appendix \ref{APP_EE_RULES} for details).

\begin{figure}[htb]
\centering
\begin{equation*}
\begin{gathered}
\includegraphics[height=0.1\textheight]{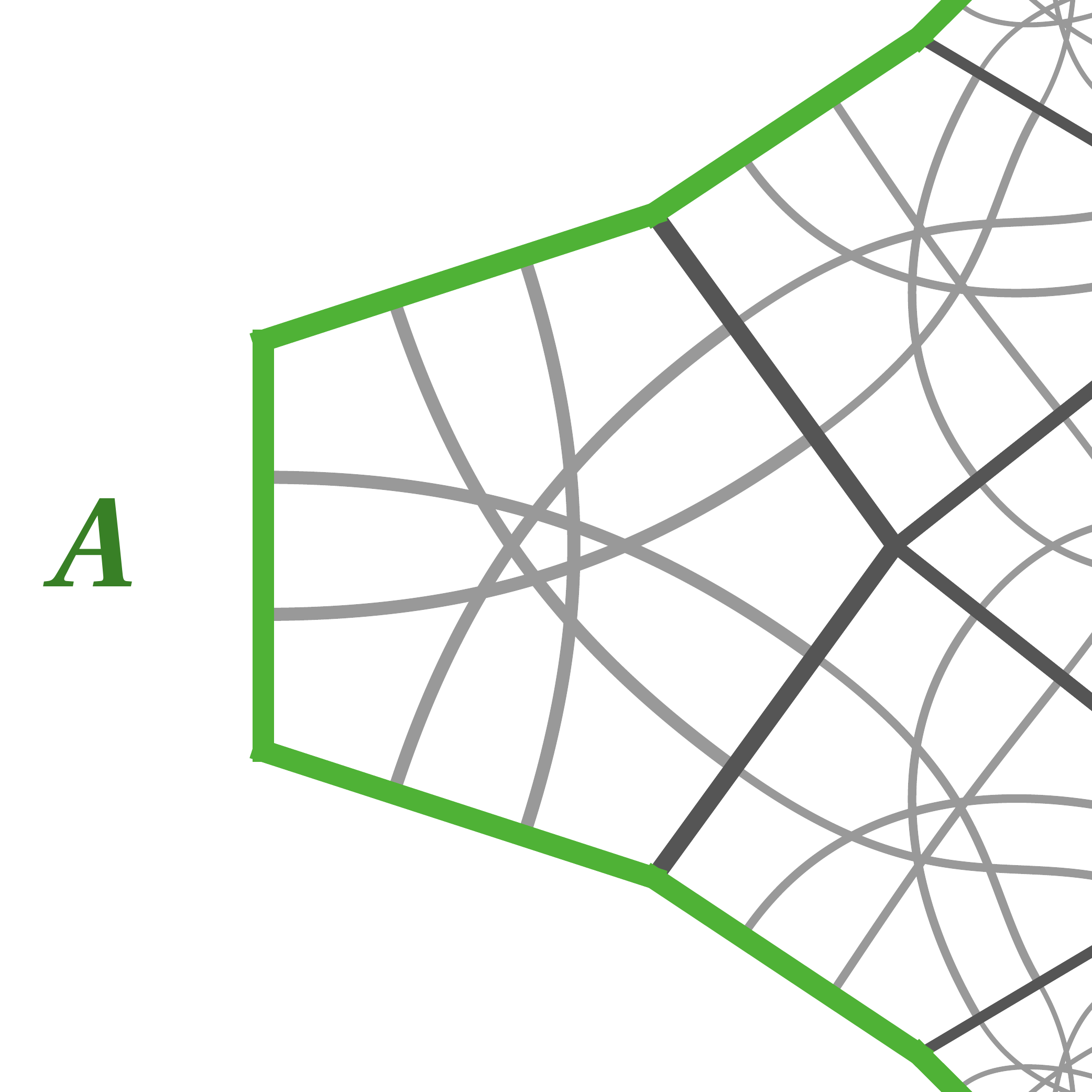}
\end{gathered}
\scalebox{1.3}{$\;\quad\quad\rightarrow\quad\;$}
\begin{gathered}
\includegraphics[height=0.1\textheight]{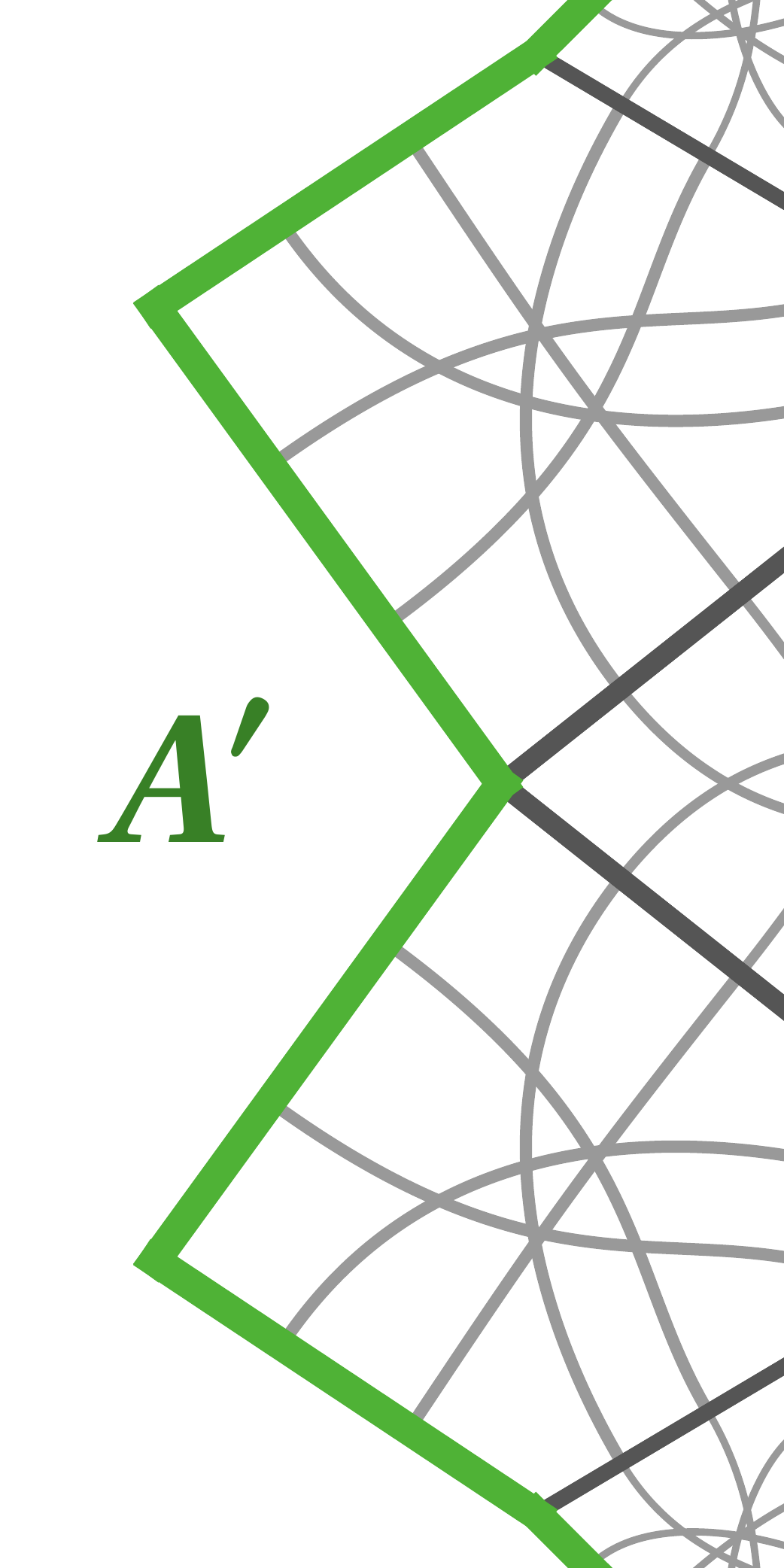}
\end{gathered}
\end{equation*}
\begin{equation*}
\begin{gathered}
\includegraphics[height=0.1\textheight]{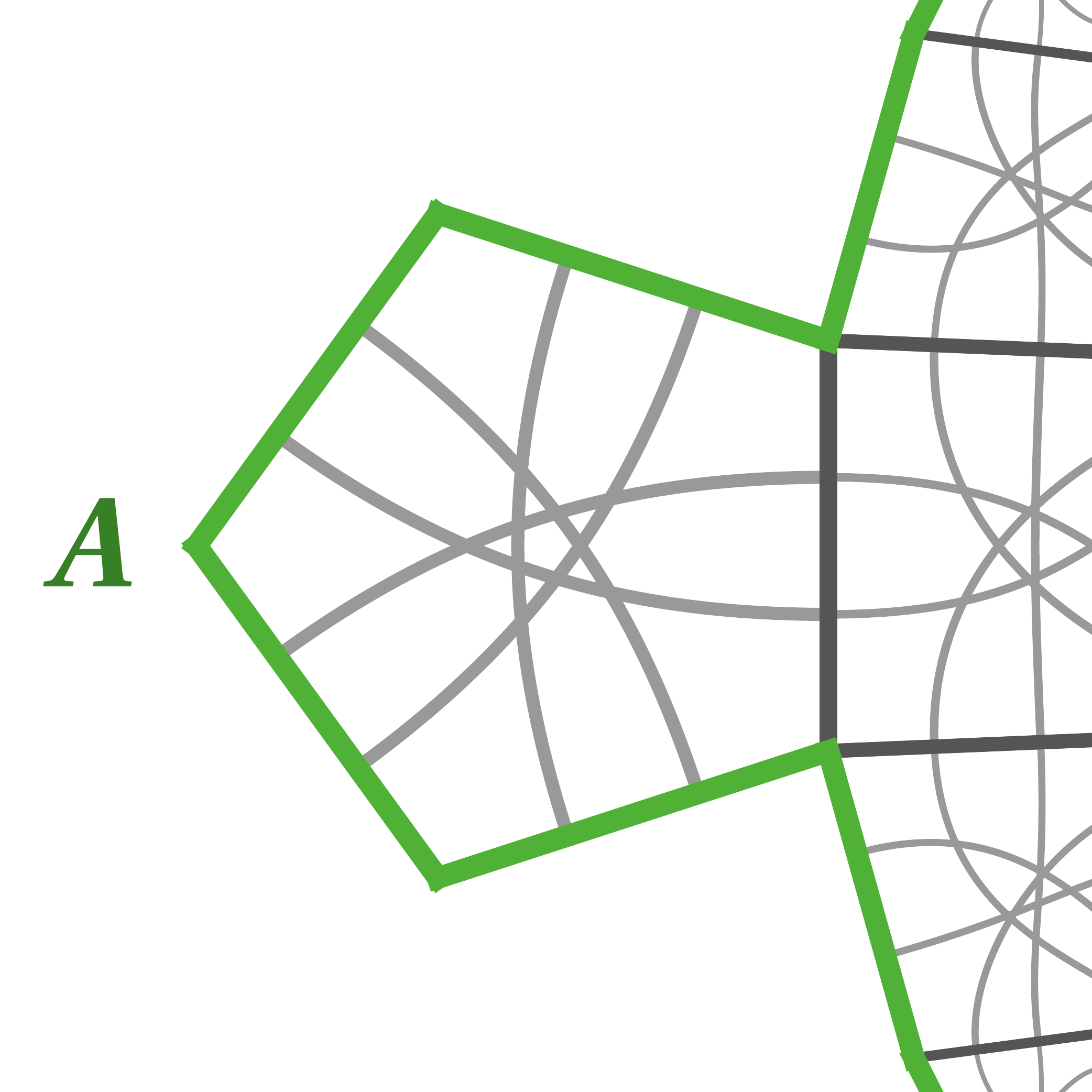}
\end{gathered}
\scalebox{1.3}{$\;\quad\quad\rightarrow\quad\;$}
\begin{gathered}
\includegraphics[height=0.1\textheight]{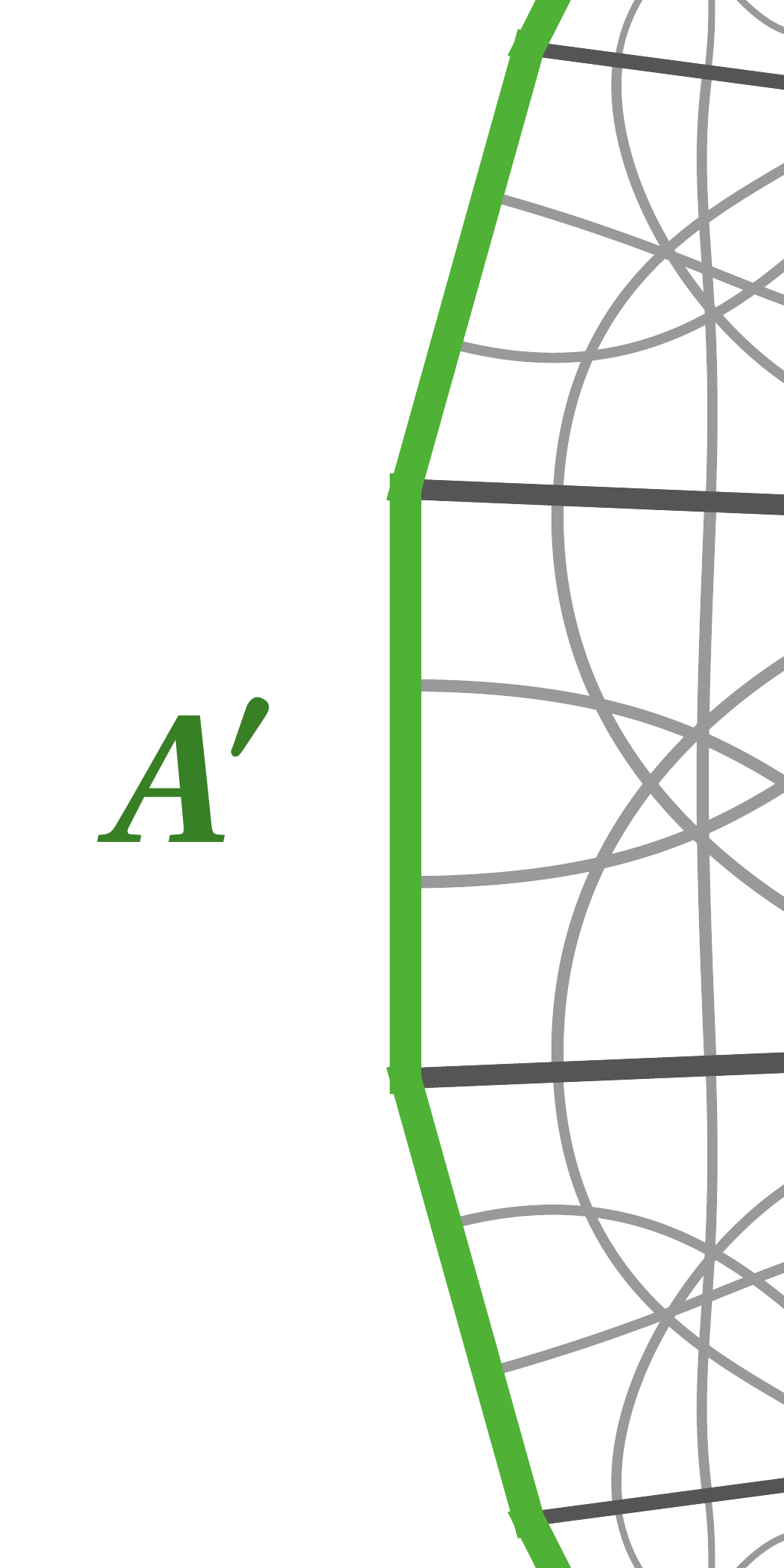}
\end{gathered}
\end{equation*}
\caption{The greedy algorithm: The boundary region $A$ is pushed into the bulk to a new region $A^\prime$ by removing pentagon tiles with three (top) and four (bottom) open edges. Each pentagon can be in an arbitrary local superposition of $\bar{0}$ and $\bar{1}$, shown as grey-shaded dimers.}
\label{FIG_GREEDY}
\end{figure}

\begin{figure}[htb]
\centering
\begin{equation*}
\begin{gathered}
\includegraphics[height=0.16\textheight]{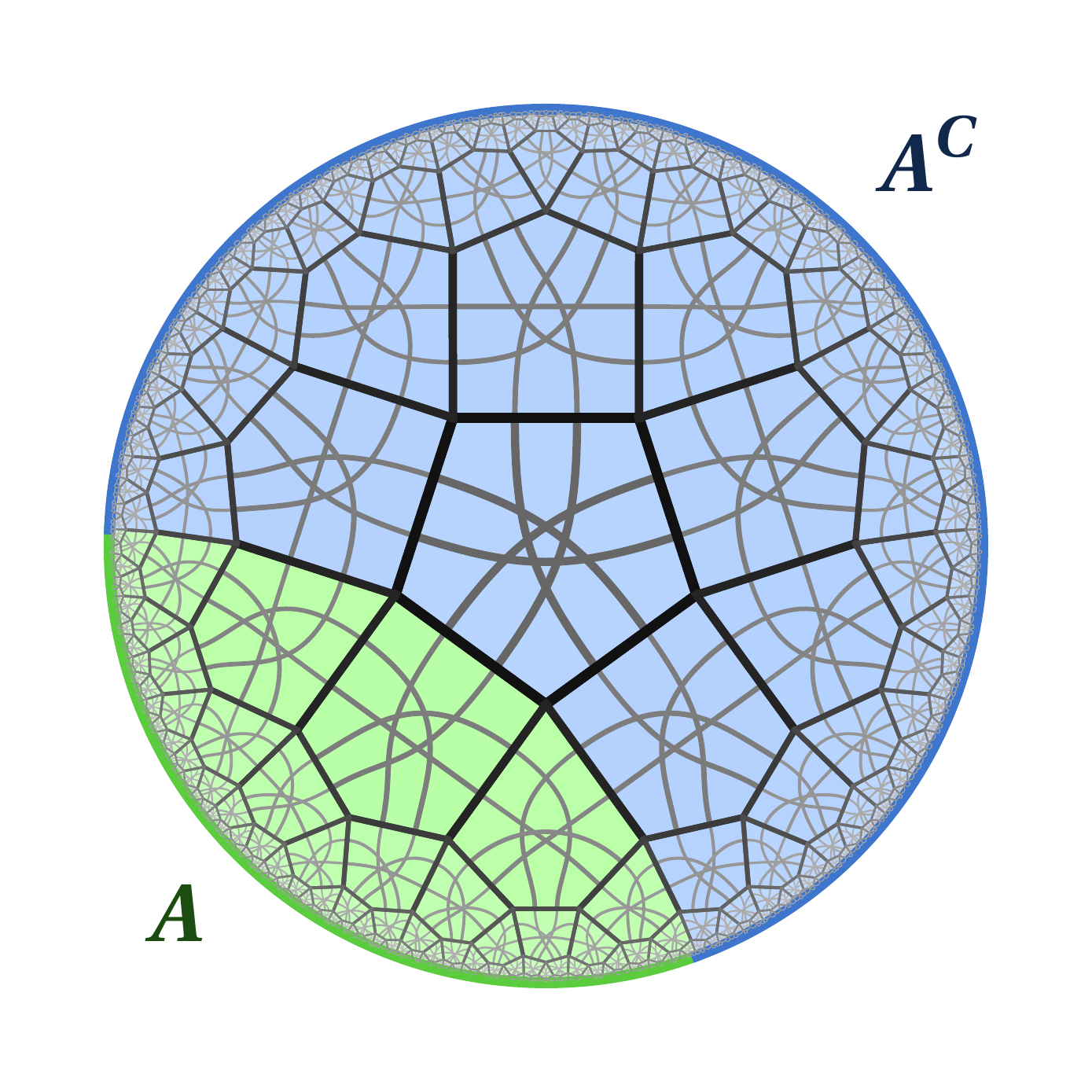}
\end{gathered}
\scalebox{1.3}{$\;\rightarrow\;$}
\begin{gathered}
\includegraphics[height=0.16\textheight]{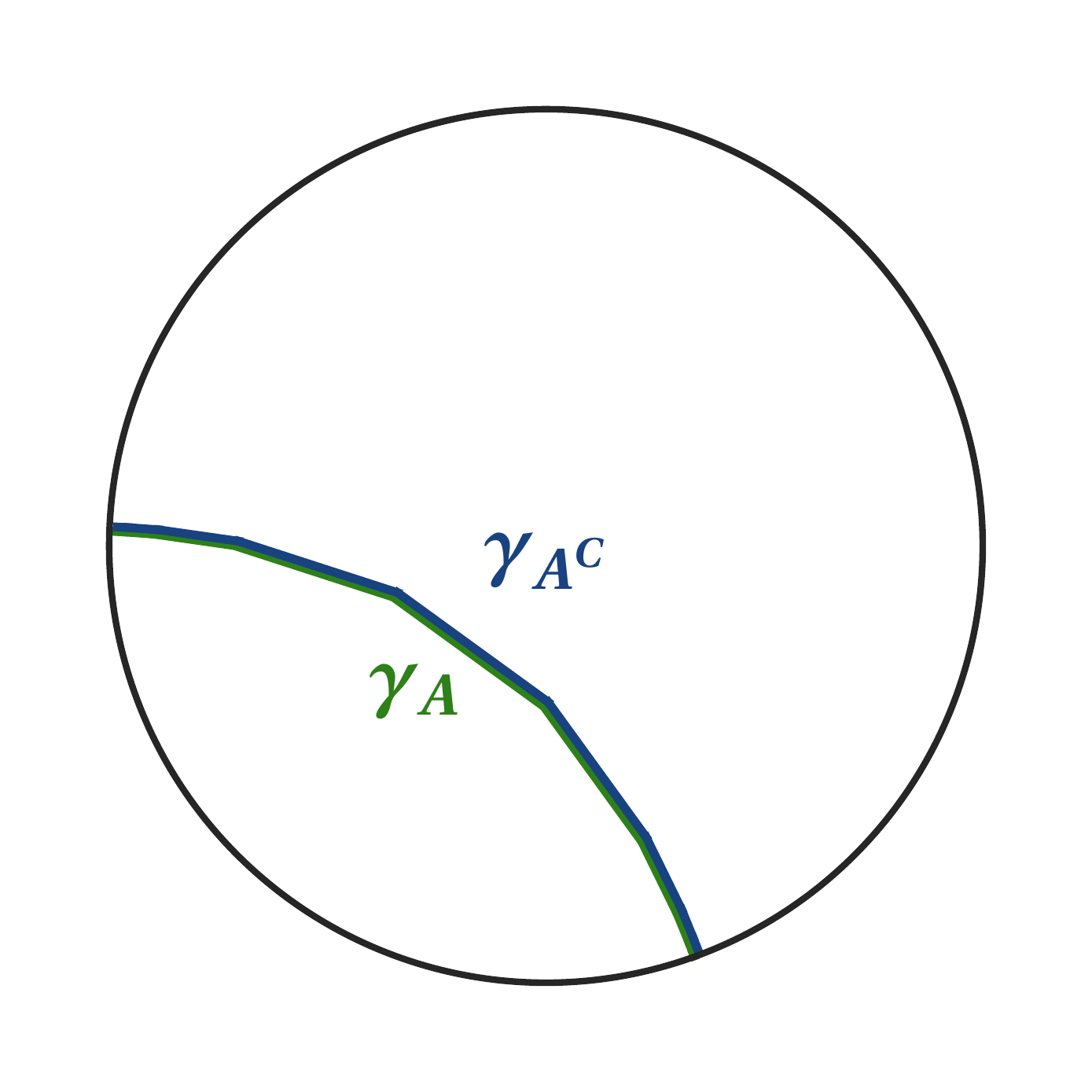}
\end{gathered}
\end{equation*}
\begin{equation*}
\begin{gathered}
\includegraphics[height=0.16\textheight]{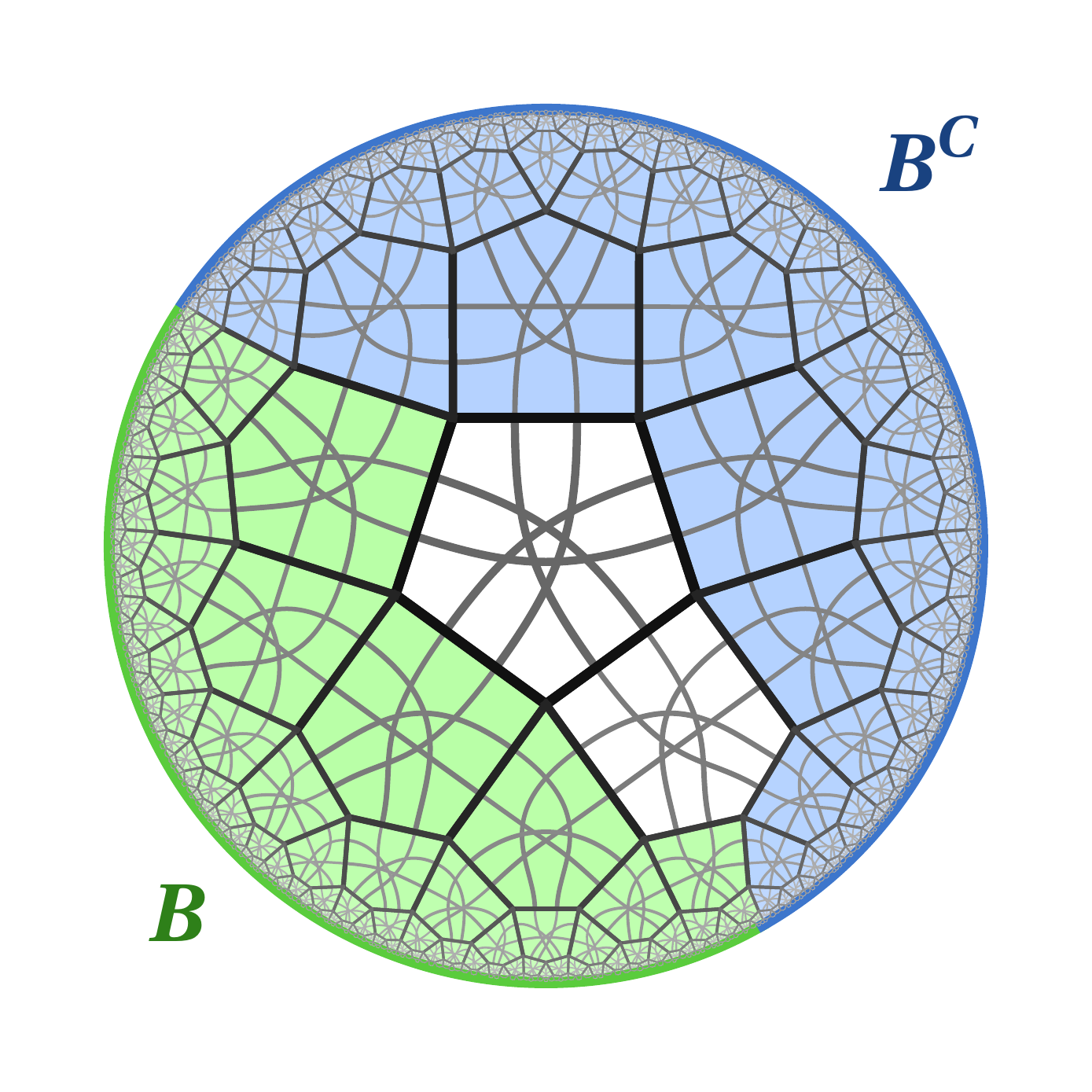}
\end{gathered}
\scalebox{1.3}{$\;\rightarrow\;$}
\begin{gathered}
\includegraphics[height=0.16\textheight]{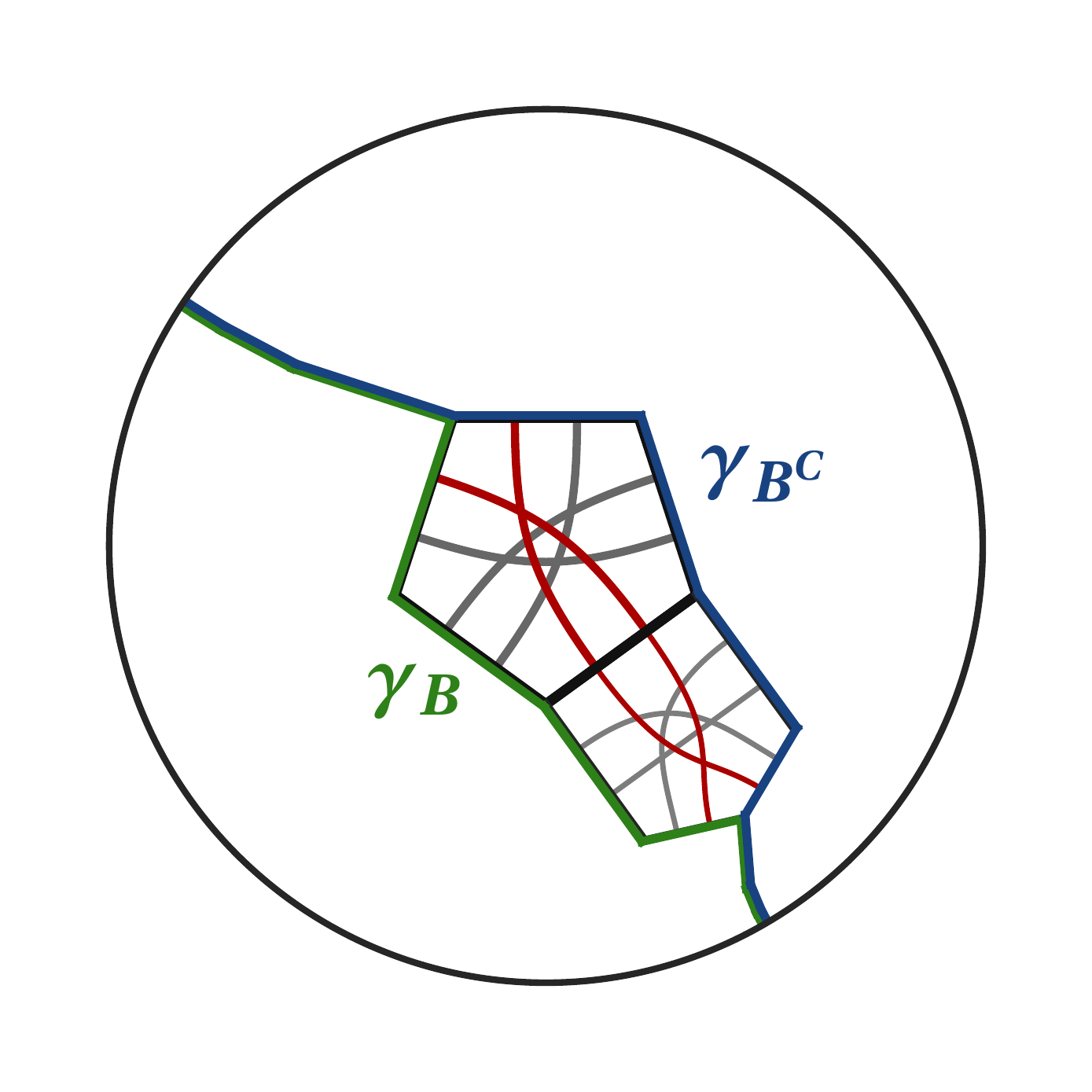}
\end{gathered}
\end{equation*}
\caption{Reducing boundary regions in the HyPeC with the greedy algorithm, for two boundary regions $A$ and $B$. $A$ reduces to the same geodesic $\gamma_A=\gamma_{A^{\text{C}}}$ as its complement $A^{\text{C}}$, while $B$ does not. On the left-hand side, the corresponding ``greedy wedge'' is shaded in the same color as the boundary regions. The residual dimers are shaded in red.}
\label{FIG_GREEDY2}
\end{figure}

\subsection{Scaling and RG flow}

As we saw in Fig.\ \ref{FIG_HAPPY_DIMER_ASYMPT}, contracting the HyPeC produces effective boundary EPR pairs connected along geodesics through the bulk. This allows for a naive interpretation in terms of IR/UV scaling: Longer geodesics that probe deeper into the bulk are then associated with the IR scale, while short-range geodesics close to the boundary correspond to UV modes. The iteration of contractions in Fig.\ \ref{FIG_HAPPY_CONTR_SERIES} is then interpreted as a renormalization group (RG) flow, with each new iteration adding additional degrees of freedom. As each tile connects to either one or two tiles of the previous iteration, there are two possible local scaling steps, both of which are shown in Fig.\ \ref{FIG_HAPPY_RENORM}. Thus either one or three new dimers are added in each local step.

\begin{figure}[tb]
\centering
\includegraphics[height=0.14\textheight]{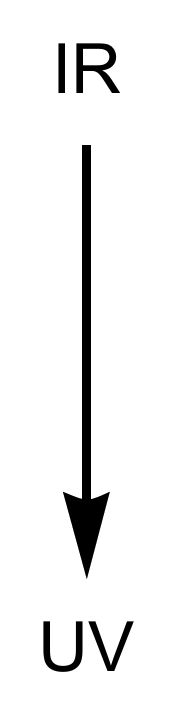}
\includegraphics[height=0.14\textheight]{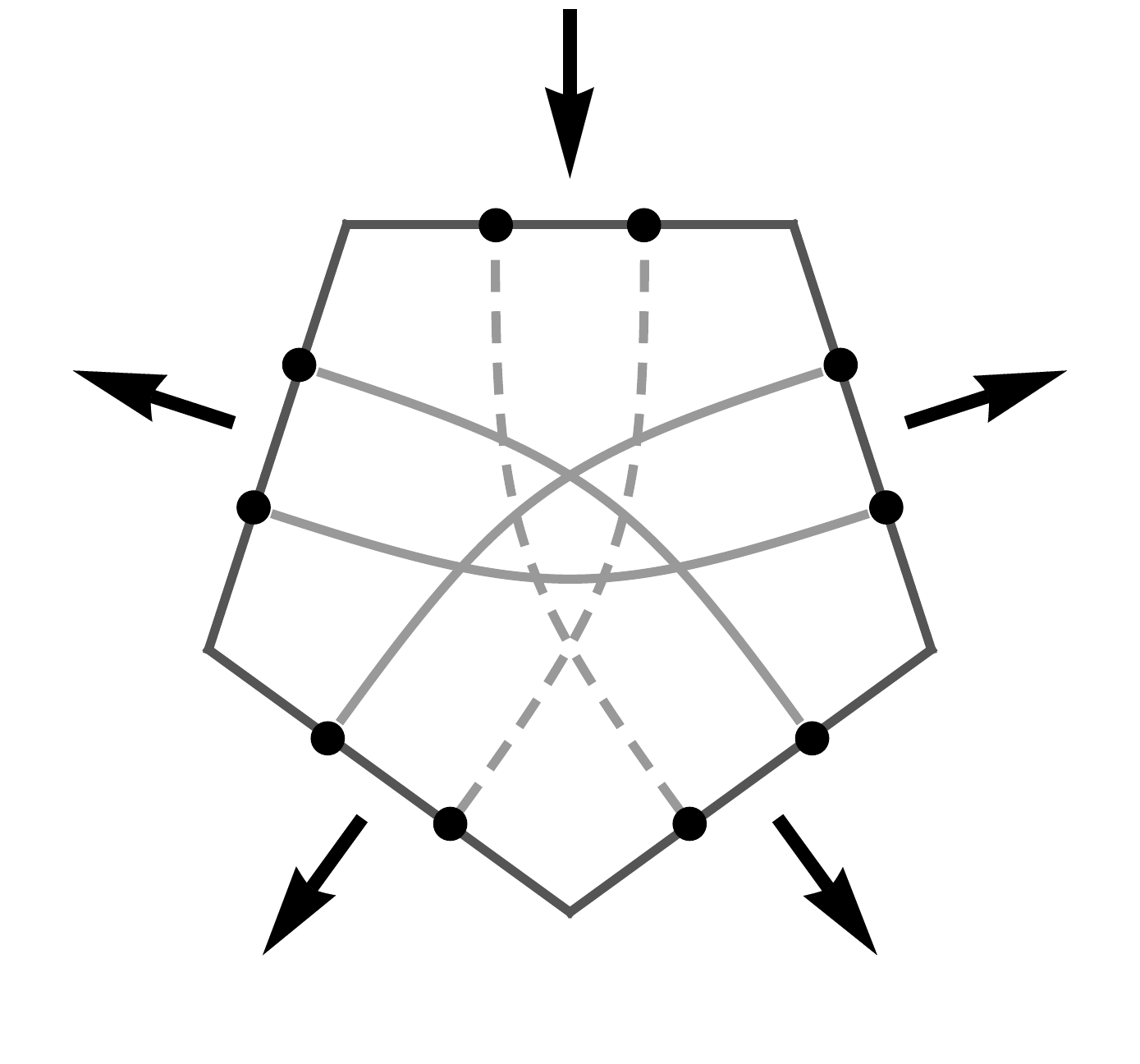}
\includegraphics[height=0.14\textheight]{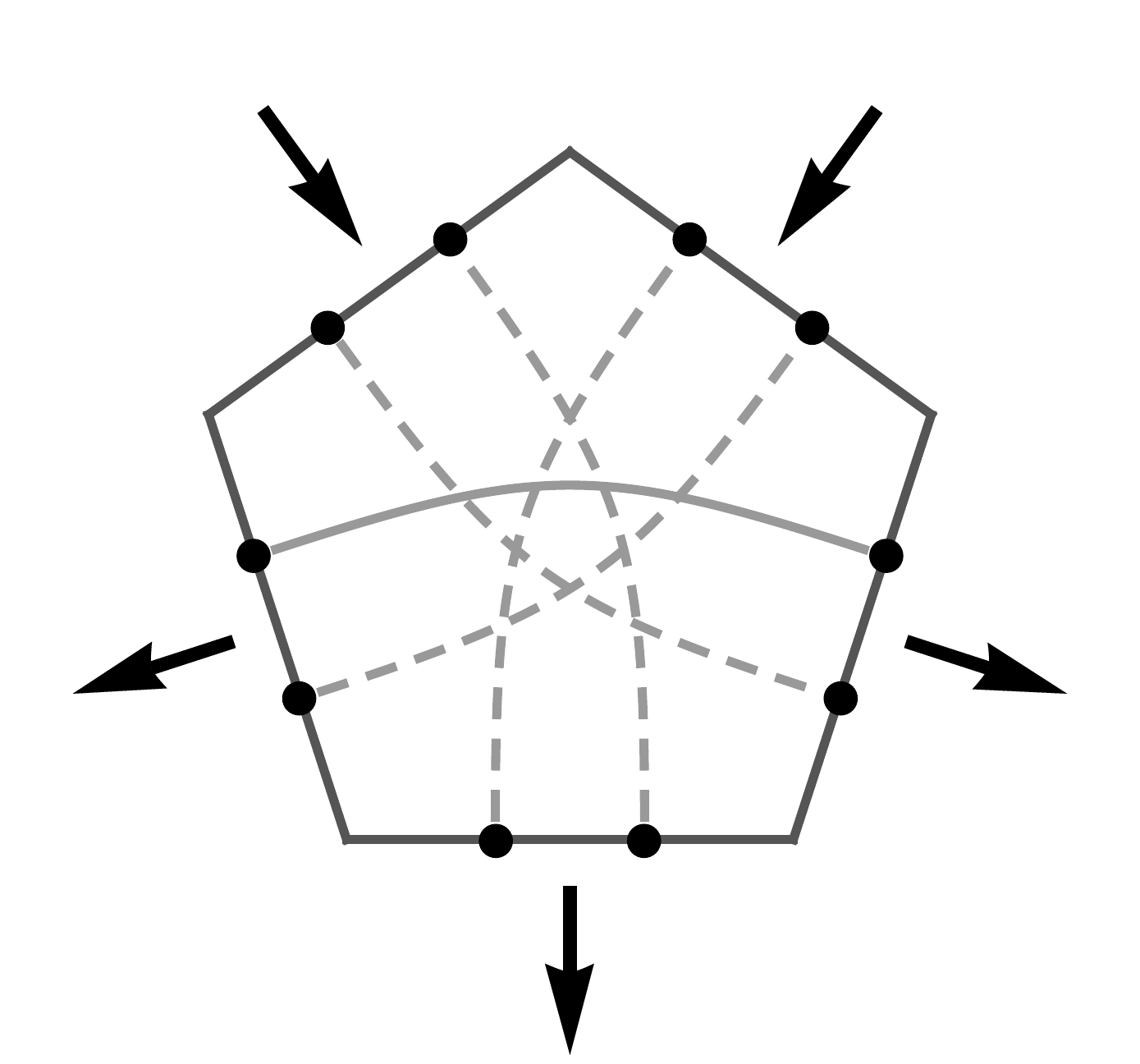}
\caption{Each tile in the HyPeC can act as a mapping of $1\to 4$ edges (left) or $2\to 3$ edge (right). Arrows distinguish between ``input'' (IR) and ``output'' (UV) edges. Dimers extended from previous tiles are drawn dashed, while new ones are drawn as solid curves. The dimer parities depend on the actual logical input on the tile.}
\label{FIG_HAPPY_RENORM}
\end{figure}

This IR$\to$UV renormalization step has a well-defined UV$\to$IR inverse constructed from the Hermititian conjugate of a specific tile. Consider the $2 \to 3$ step from Fig.\ \ref{FIG_HAPPY_RENORM}: The IR$\to$UV step consists of contracting the edges 1 and 2 of $\ket\psi = \alpha \ket{\bar{0}}_5 + \beta \ket{\bar{1}}_5$ (with $|\alpha|^2 + |\beta|^2=1$) onto the boundary state. To inverse this operation, we trace out the edges $3,4,5$ of $\ketbra\psi\psi$, which we can express using Majorana dimers (full calculation \eqref{EQ_HAPPY_SP_3C} in Appendix \ref{APP_EE_RULES}):
\begin{align}
\label{EQ_PENTA_TRACE}
\tr_{(3,4,5)}\, \ketbra{\psi}{\psi}&= \frac{|\alpha|^2 + |\beta|^2}{2}\;
\begin{gathered}
\includegraphics[height=0.06\textheight]{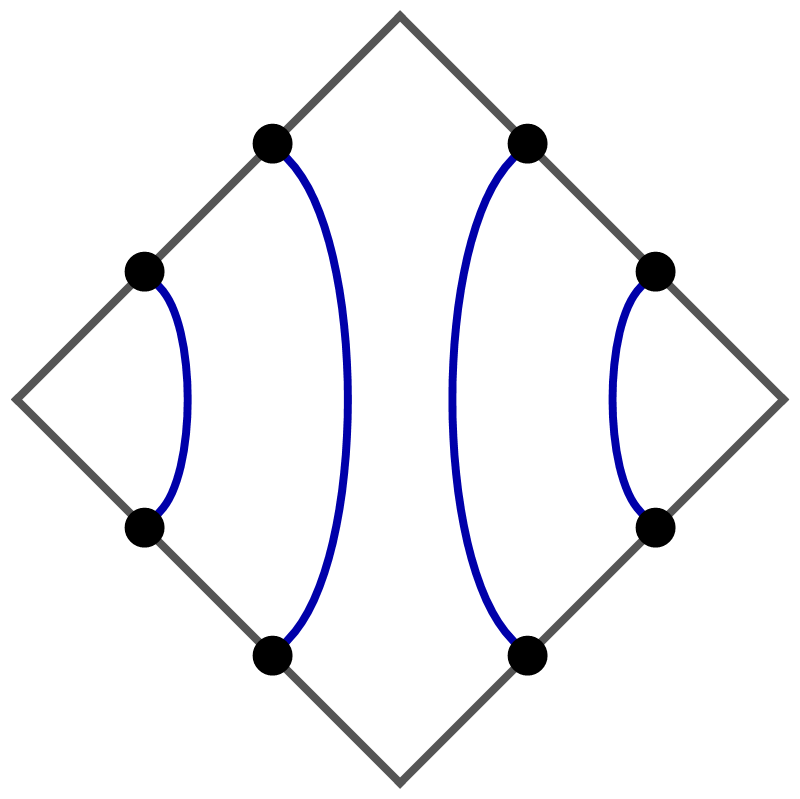}
\end{gathered} 
\;= \frac{\id_4}{4} \text{ .}
\end{align}
Note the representation of the identity as a set of dimers directly connecting two pairs of edges: Any Majorana dimer state is left invariant by contraction with such a state, and hence any other state, which we can always express in a dimer basis, as well.
We can similarly find that $\tr_{(2,3,4,5)}\, \ketbra{\psi}{\psi} = \id_2/2$. The inverse renormalization step is thus simply the reversal of Fig.\ \ref{FIG_HAPPY_RENORM}: Some dimers are closed into loops, thus tracing out the associated degrees of freedom. 
In fact, this result is closely related to the \textsl{perfect tensor} property of the HyPeC, whereby any pentagon tile can act as an isometry of $k \to 5-k$ edges as long as $k \leq 5-k$. We also use it in Appendix \ref{APP_EE_RULES} to prove the greedy algorithm with Majorana dimers.
\begin{center}
\begin{figure*}
\begin{equation*}
\begin{gathered}
n=1:\\
\vspace{7pt}
\includegraphics[height=0.14\textheight]{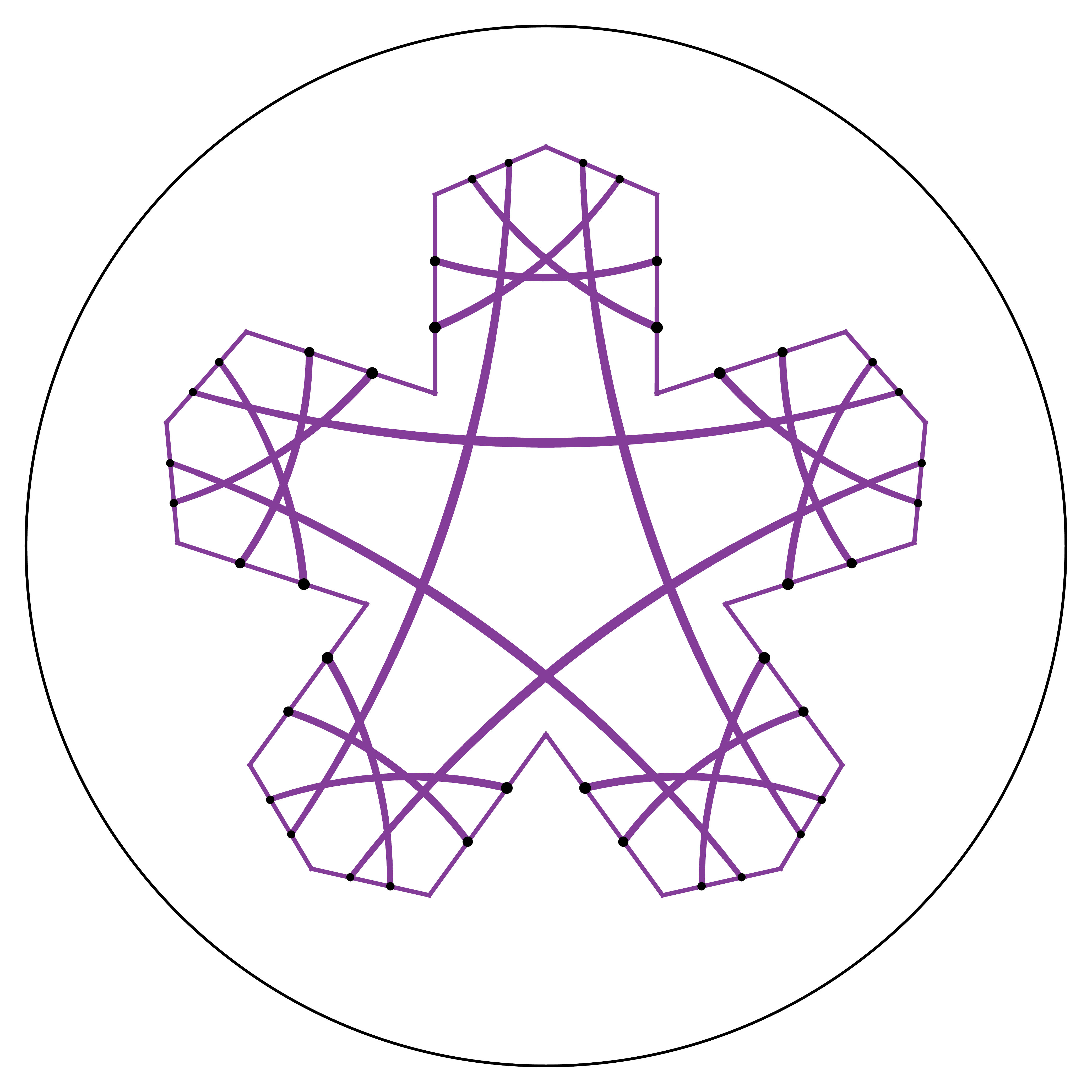}
\end{gathered}
\scalebox{1.3}{$\quad\rightarrow\quad$}
\begin{gathered}
n=2:\\
\vspace{7pt}
\includegraphics[height=0.14\textheight]{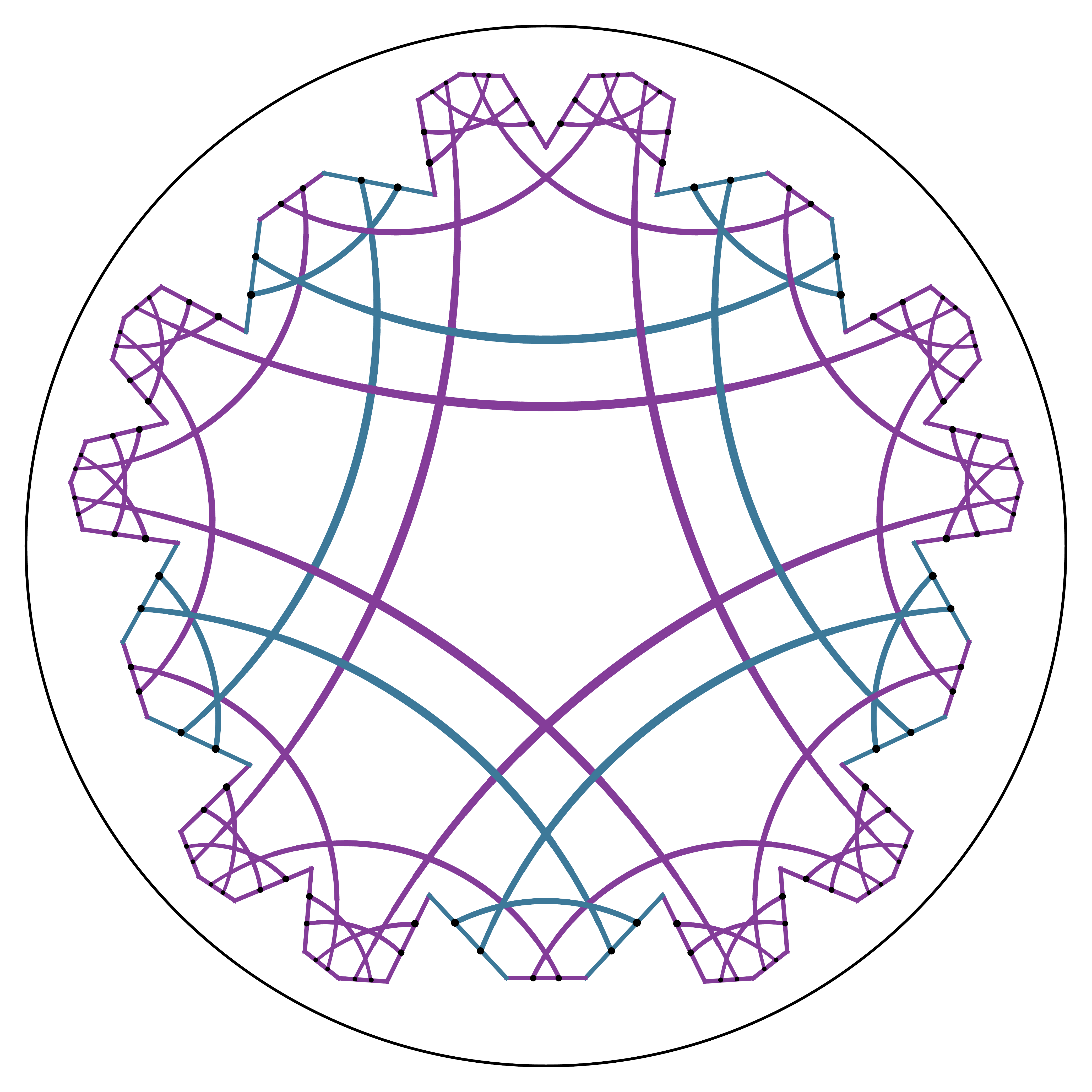}
\end{gathered}
\scalebox{1.3}{$\quad\rightarrow\quad$}
\begin{gathered}
n=3:\\
\vspace{7pt}
\includegraphics[height=0.14\textheight]{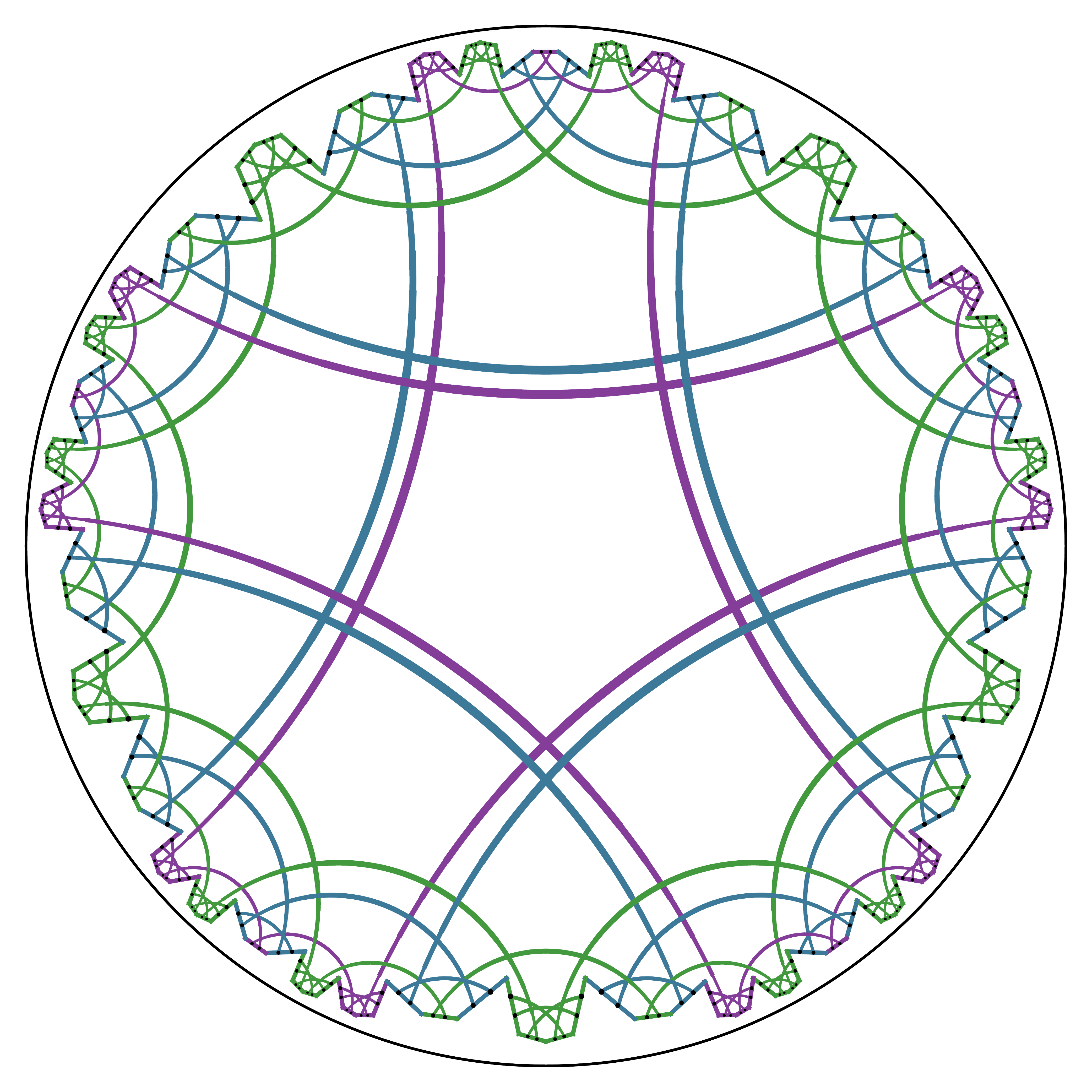}
\end{gathered}
\quad
\scalebox{1.3}{$\quad\rightarrow\quad\dots$}
\end{equation*}
\caption{Iterative contraction of the HyPeC, with Majorana dimers belonging to decoupled fermionic subsystems drawn in different colors. The number of such subsystems increases with iteration number $n$.}
\label{FIG_HAPPY_CONTR_SERIES2}
\end{figure*}
\end{center}

\vspace{-20pt}

While we saw that Majorana dimers form effective EPR pairs in the asymptotic limit of infinitely many contractions, we can also observe a separation of the boundary into separate fermionic subsystems at finite cutoff. The physical fermion corresponding to each uncontracted edge can be coupled to at most two other fermions/edges via the dimers it contains. These two fermions are again coupled to further fermions, so we end up with a -- necessarily closed -- chain of fermions, each only coupled to its two non-local ``neighbors''. However, as we contract more and more tiles, we find that our boundary fermions are separated into an increasing number of independent chains. This is shown in Fig.\ \ref{FIG_HAPPY_CONTR_SERIES2} for the first few iterations, where the dimers are colored according to the decoupled fermionic chain they belong to.
The appearance of additional decoupled subsystems at larger iterations is another sign of an RG flow: Increasing the number of iterations encodes more and more subsystems of varying sizes on the boundary. For the full HyPeC beyond basis-state input, correlations between these subsystems can be nonzero. As we will show next, however, such correlations can only be captured by $n$-point correlators with $n>2$.

\subsection{Correlation functions for general bulk input}

\begin{figure}[tb]
\centering
\begin{align*}
\begin{gathered}
\includegraphics[height=0.16\textheight]{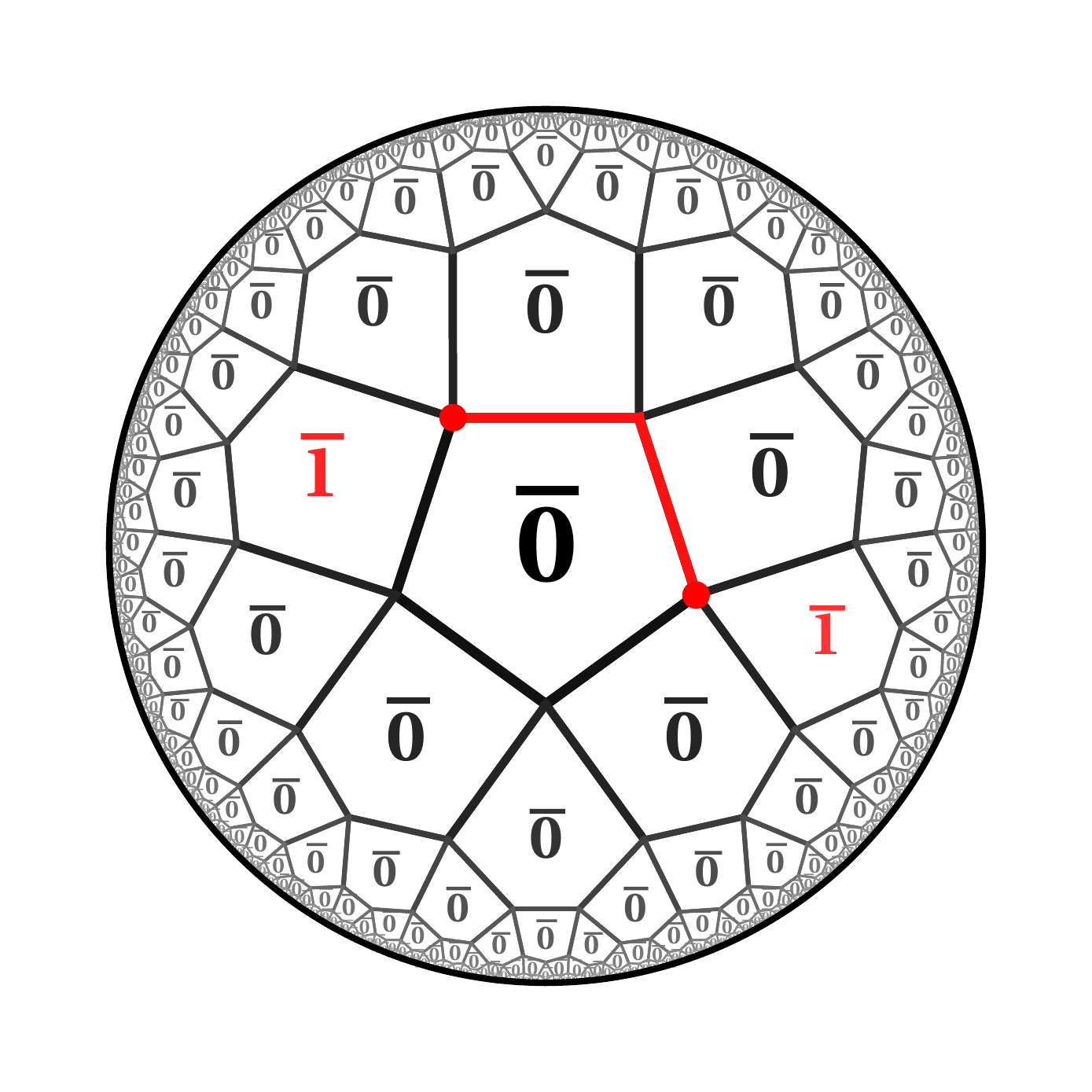}\vspace{-0.25cm}\\
\vspace{-0.25cm}
\scalebox{1.3}{$\updownarrow$}\\
\includegraphics[height=0.16\textheight]{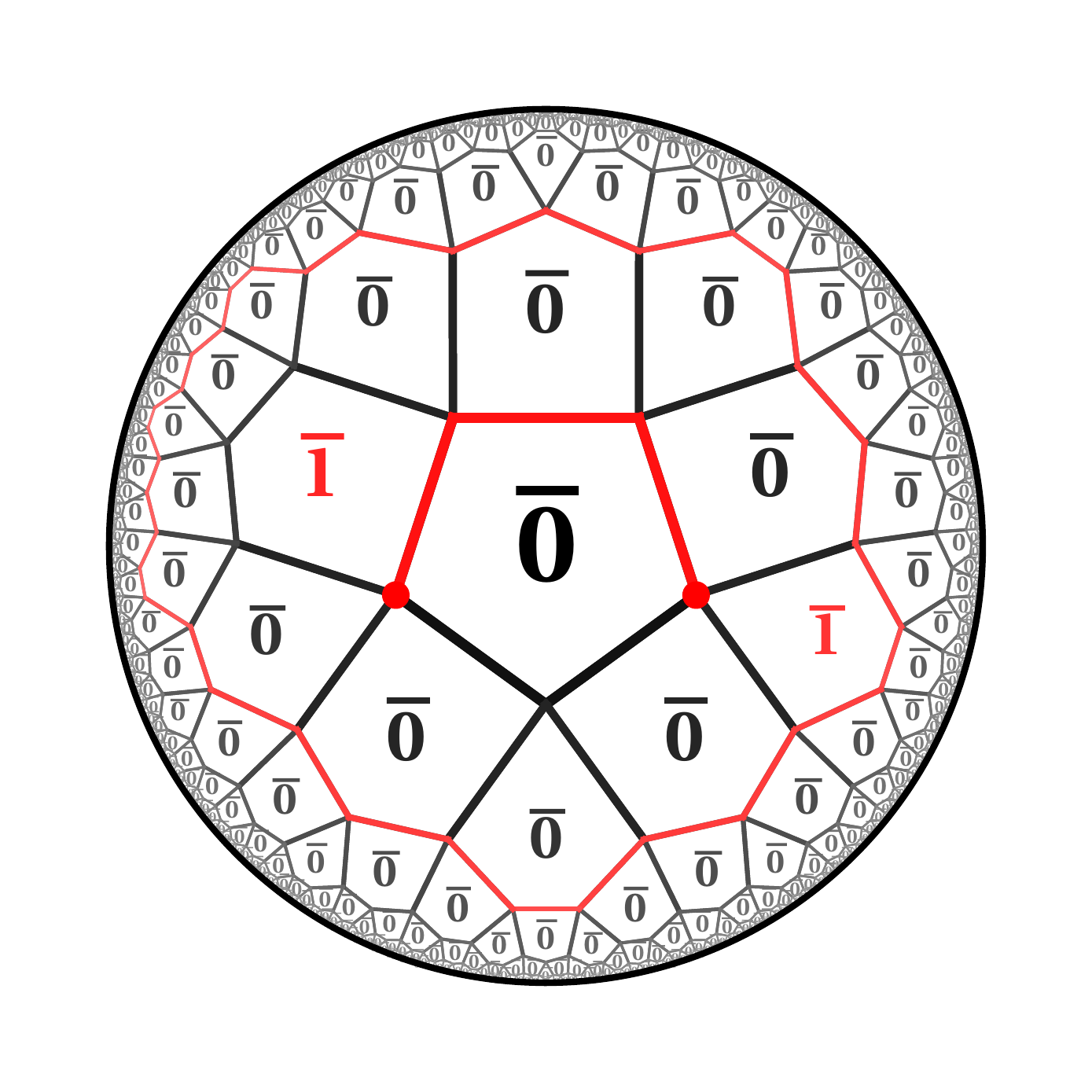}
\end{gathered}
\begin{gathered}
\includegraphics[height=0.16\textheight]{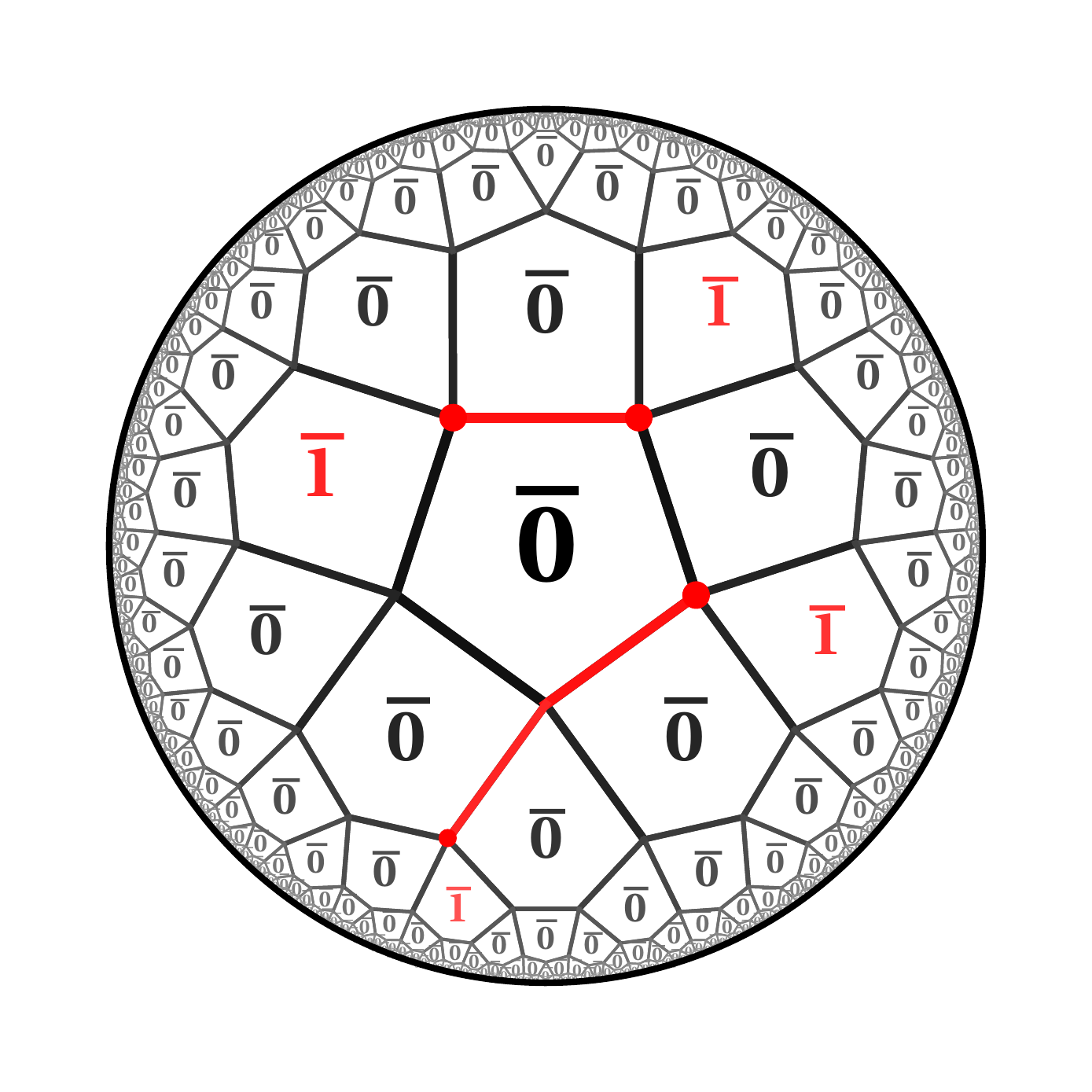}\vspace{-0.25cm}\\
\vspace{-0.25cm}
\scalebox{1.3}{$\updownarrow$}\\
\includegraphics[height=0.16\textheight]{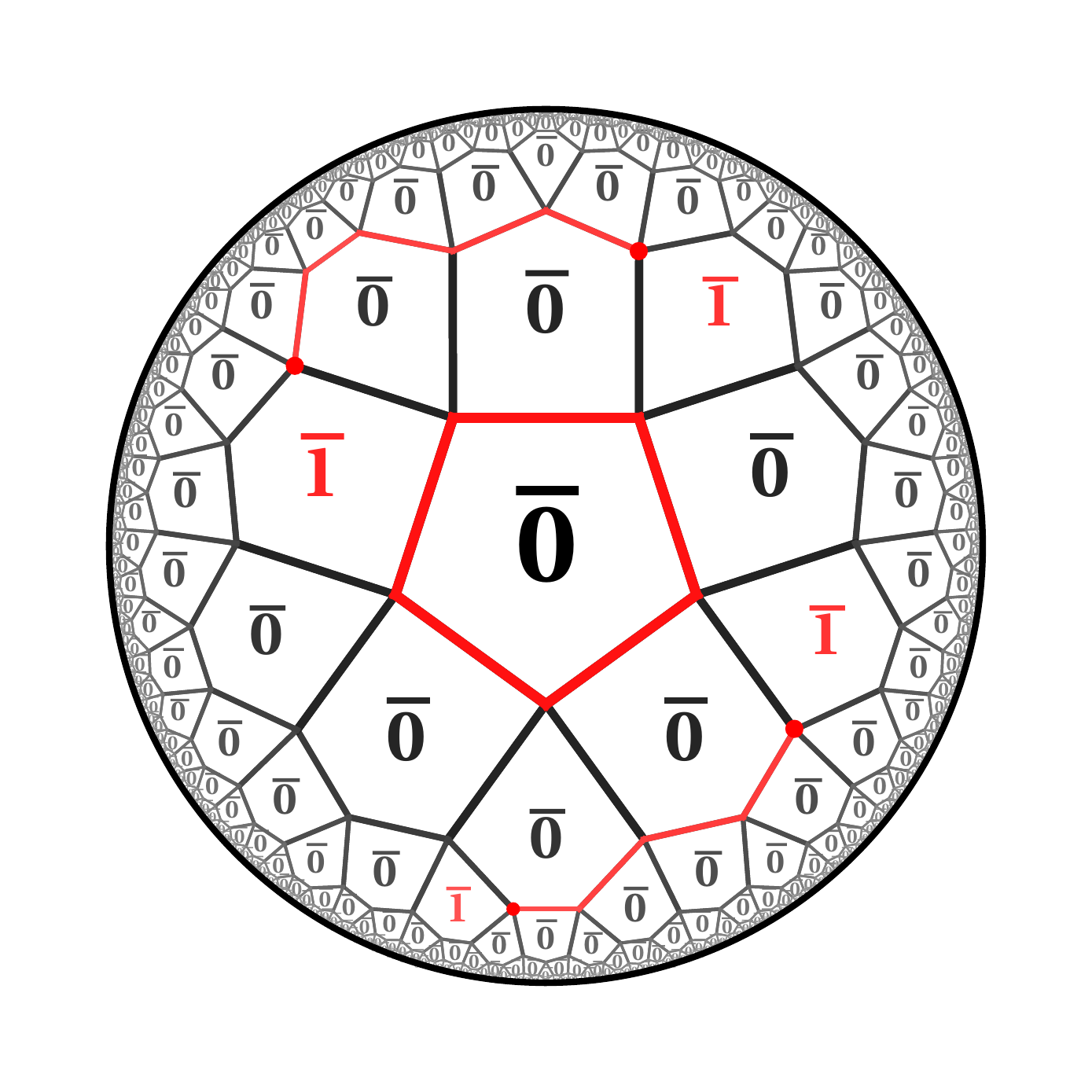}
\end{gathered}
\end{align*}
\caption{\textsc{Top:} Inserting two (left) and four (right) $\bar{1}$ tiles in the HyPeC. Beyond the local dimer parity flips in each tile, pairs of $\bar{1}$ tiles are connected by $Z$ strings (red lines), which flip dimer parities non-locally. The endpoints of these strings (red dots) set the local orientation of the $\bar{1}$ tiles connected to it. \textsc{Bottom:} Equivalent $Z$ string configuration of the upper two states.
}
\label{FIG_HAPPY_INPUT}
\end{figure}

By counting the dimers by the boundary distance over which they reach, the average correlation falloff of the Majorana covariance matrix $\Gamma$ defined in Eq.~\eqref{eq:cov_def}
can be determined. For fixed input, this leads to a polynomial $\Gamma_{j,k} \propto |j-k|^{-1}$ falloff of two-point correlations \cite{Jahn:2017tls}, again resembling a CFT scaling. Naively, this holds only for the case of a fixed logical input $\bar{0}$ or $\bar{1}$ on each tile, as superpositions of Majorana dimer states are generally non-Gaussian and have a complicated two-point correlation structure.

However, we will now show that two-point correlations for the HyPeC with general bulk input, where local superpositions of $\bar{0}$ and $\bar{1}$ inputs are allowed, are surprisingly similar to the fixed-input case.
First, consider the dimer parity structure caused by local $\bar{1}$ inputs. As we showed in Fig.\ \ref{FIG_HAPPY_CONTR_SERIES}, using even-parity $\bar{0}$ input over the entire bulk leads to a simple contracted state, where all resulting dimer parities are even. When contracting over odd-parity $\bar{1}$ inputs, index permutations necessary during the contraction process can lead to additional dimer parity flips caused by $Z$ operators on some of the pentagon edges. After going through the contraction process, which is laid out in Appendix \ref{APP_CONTR_ORDER}, we find that these dimer parity flips can be grouped into strings of $Z$ operators between the tiles with $\bar{1}$ input. Possible configurations are shown in Fig.\ \ref{FIG_HAPPY_INPUT} for two and four $\bar{1}$ insertions. While neither the pairing of $\bar{1}$ tiles with $Z$ strings nor the paths of these strings are unique, we can freely deform them without changing the boundary state (bottom diagrams in Fig.\ \ref{FIG_HAPPY_INPUT}). Furthermore, we can freely add closed $Z$ loops around a set of tiles with an even number of $\bar{1}$ tiles in it, as this is equivalent to evaluating the total (even) parity of the contained tiles. 
Intriguingly, we can relate this to physical rotations of tiles: While the dimer parities of $\bar{0}$ tiles are invariant under cyclic permutations (in the spin picture), we showed in \eqref{EQ_PENTAGON_STATES2} that a rotation of a tile with $\bar{1}$ input is equivalent to tracing the shifted endpoint of the local ordering with a $Z$ string. A full ``rotation'' (leading to a closed $Z$ loop around a $\bar{1}$ tile) changes the state by a total minus sign. In other words, as shown in Fig.\ \ref{FIG_HAPPY_INPUT}, $Z$ string loops around tiles with an even number of $\bar{1}$ insertions leave the state invariant.
Thus, it is tempting to interpret the $\bar{0}$ tiles as local fermionic vacua and the $\bar{1}$ tiles as logical fermions, emergent from the underlying spin degrees of freedom.

The set of boundary states for all possible basis-state bulk inputs ($\bar{0}$ or $\bar{1}$ on each tile) gives us a basis set for the states of the full HyPeC. In general, boundary $n$-point functions for an arbitrary input can have a correlation structure completely different from the dimer structure we saw for logical basis-state input, and overlaps between different basis states can change the entanglement structure. Fortunately, as we show in Appendix \ref{APP_CONTR_ORDER}, the HyPeC boundary states for different basis inputs are all distinct by an operator of \emph{Majorana weight} $\mathfrak{w}{>}2$, i.e., at least three Majorana operators $\m_k$ are required to map one basis state to another.
This leads us to the following conclusion:
\begin{theorem}
\label{THM_TWO_POINT_CORR}
For a contraction of $N$ pentagon tiles of the HyPeC, two-point correlation functions of the boundary states are convex combinations of the covariance matrices for any logical basis input.
\end{theorem}
\begin{proof}
We denote by $\ket{b}:= \ket{b_1, b_2,\dots,b_N}$ the state vector for a fixed basis-state input $b_k$ on the $k$th pentagon. A general HyPeC boundary state vector is given by the superposition
\begin{equation}
\ket\omega = \sum_{b\in \{\bar{0},\bar{1}\}^{\times N}} c_b \ket{b} \text{ ,}
\end{equation}
with $c_b \in \mathbb{C}$. 
A fermionic two-point correlation function with entries
\begin{align}
G^{(2)}_{j,k} &= \frac{\i\,}{2}\sandwich{\omega}{[\m_j, \m_k]}{\omega} \\
&=\sum_{b,b^\prime \in \{\bar{0},\bar{1}\}^{\times N}} \frac{\i\, c_b^\star c_{b^\prime}}{2} \sandwich{b}{[ \m_j,\m_k ]}{b^\prime}
\end{align}
is generally a sum of $2^{2N}$ terms. However, we assumed that two boundary states for different basis-state inputs $b$ and $b^\prime$ are separated by a $\mathfrak{w}>2$ operator, i.e.\ fulfill the conditions
\begin{align}
\braket{b}{b^\prime} &= 0 \text{ ,} \\
\bra{b} \m_j \ket{b^\prime} &= 0 \text{ ,} \\
\bra{b} \m_j \m_k \ket{b^\prime} &= 0 \text{ .}
\end{align}
In other words, the expectation values of operators with Majorana weight $\mathfrak{w}\leq 2$ are diagonal in the logical basis. This implies 
\begin{align}
G^{(2)}_{j,k} &= \sum_{b \in \{\bar{0},\bar{1}\}^{\times N}} \frac{\i\,c_b^* c_{b}}{2} \sandwich{b}{[ \m_j,\m_k ]}{b} \nonumber\\
&= \sum_{b \in \{\bar{0},\bar{1}\}^{\times N}} |c_b|^2\, \Gamma^b_{j,k} \text{ ,}
\end{align}
where $\Gamma^b_{j,k}= \i \sandwich{b}{[\m_j,\m_k]}{b}/2$ are the covariance matrices for the Gaussian boundary state for a basis-state input $b$. 
\end{proof}
This enormously simplifies the computation of fermionic two-point correlation functions. For example, consider the contraction of only two pentagon states: There are four possible fixed logical bulk inputs, with a $\bar{0}$ or $\bar{1}$ input on either pentagon. When contracted, these lead to the ``boundary'' state vectors $\ket{\bar{0},\bar{0}}_8$, $\ket{\bar{0},\bar{1}}_8$, $\ket{\bar{1},\bar{0}}_8$, and $\ket{\bar{1},\bar{1}}_8$ on eight edges. Now consider a general logical input, i.e.,\ a state vector $\alpha_1 \ket{\bar{0}}_5 + \beta_1 \ket{\bar{1}}_5$ on the first tile and $\alpha_2 \ket{\bar{0}}_5 + \beta_2 \ket{\bar{1}}_5$ on the second (with $|\alpha_k|^2 + |\beta_k|^2=1$). As tensor contraction is a linear operation, the contracted state vector is given by
\begin{align}
\ket\omega = \alpha_1 \alpha_2 &\ket{\bar{0},\bar{0}}_8 + \beta_1 \beta_2 \ket{\bar{1},\bar{1}}_8  \\
+ \alpha_1 \beta_2 &\ket{\bar{0},\bar{1}}_8 + \beta_1 \alpha_2 \ket{\bar{1},\bar{0}}_8 \text{ .}
\end{align}
In dimers, the explicit basis-state contractions are
\begin{align}
\ket{\bar{0},\bar{0}}_8 \;&=
\begin{gathered}
\includegraphics[height=0.1\textheight]{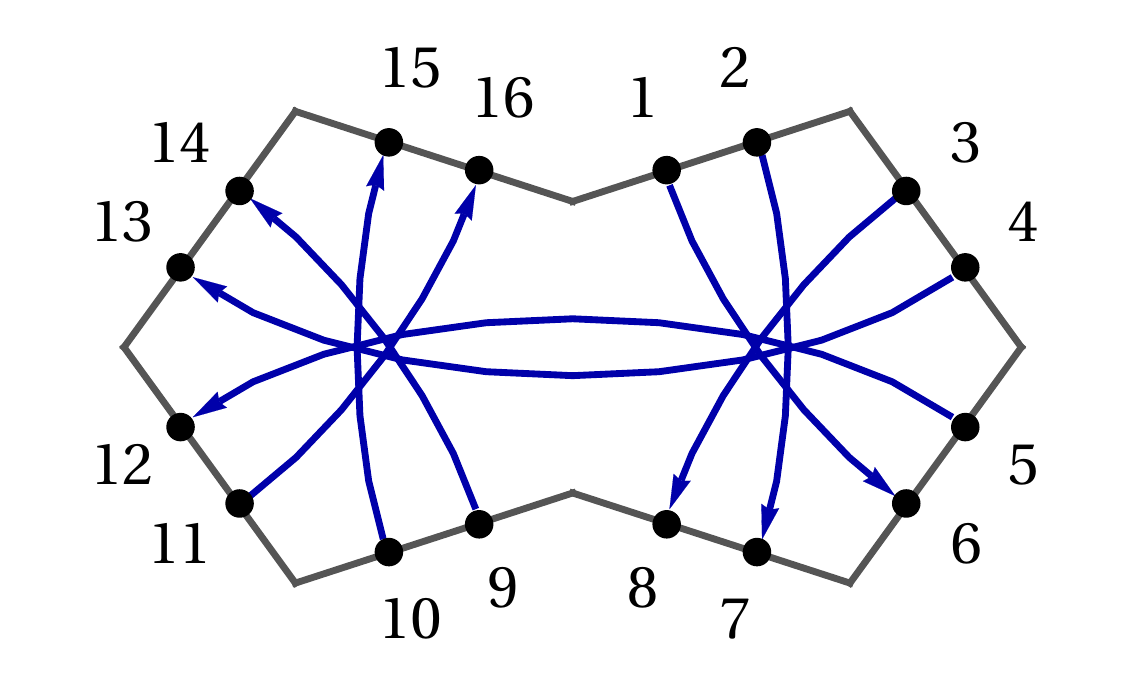}
\end{gathered} \ ,\\
\ket{\bar{0},\bar{1}}_8 \;&=
\begin{gathered}
\includegraphics[height=0.1\textheight]{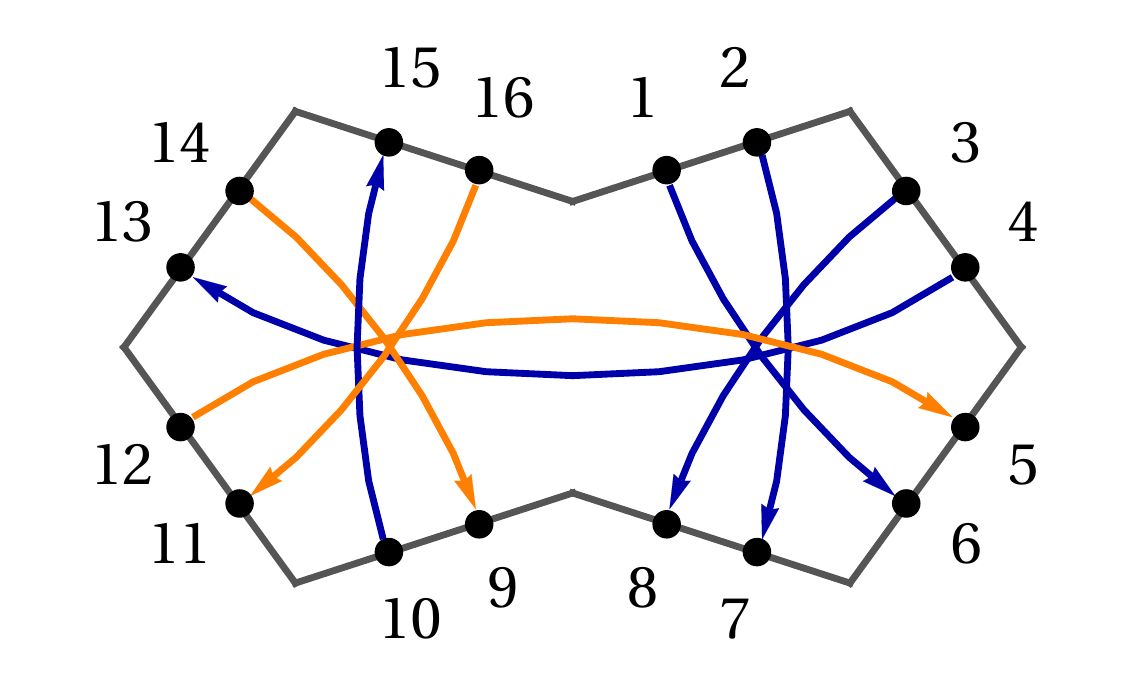}
\end{gathered}\ ,\\
\ket{\bar{1},\bar{0}}_8 \;&=
\begin{gathered}
\includegraphics[height=0.1\textheight]{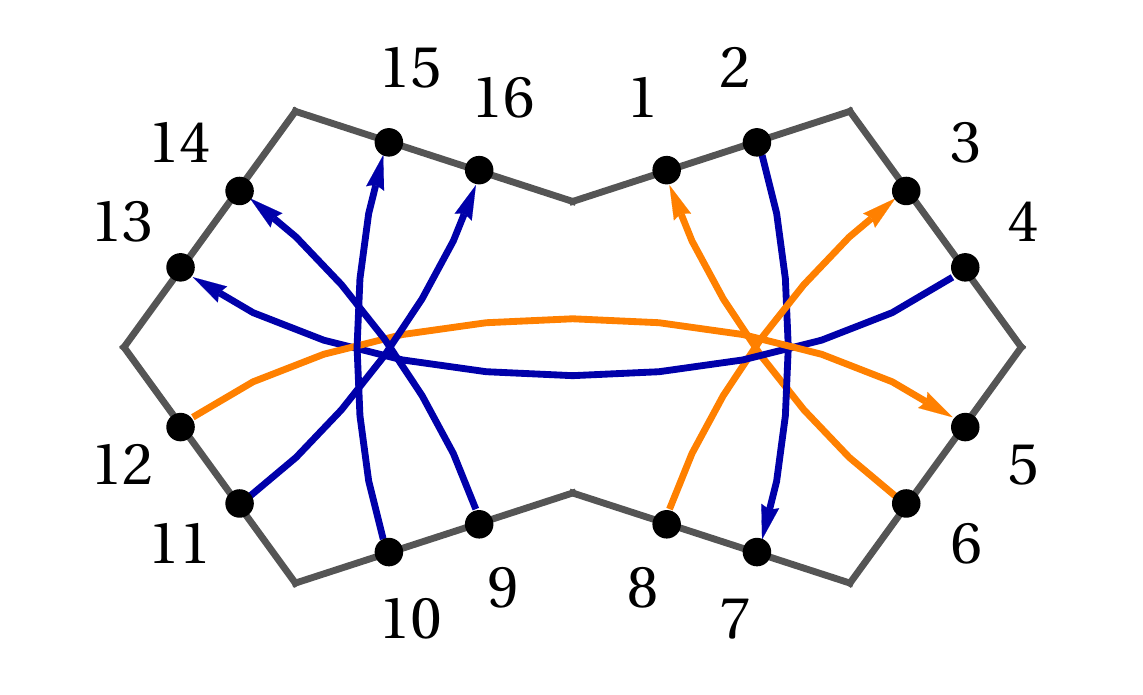}
\end{gathered}\ ,\\
\ket{\bar{1},\bar{1}}_8 \;&=
\begin{gathered}
\includegraphics[height=0.1\textheight]{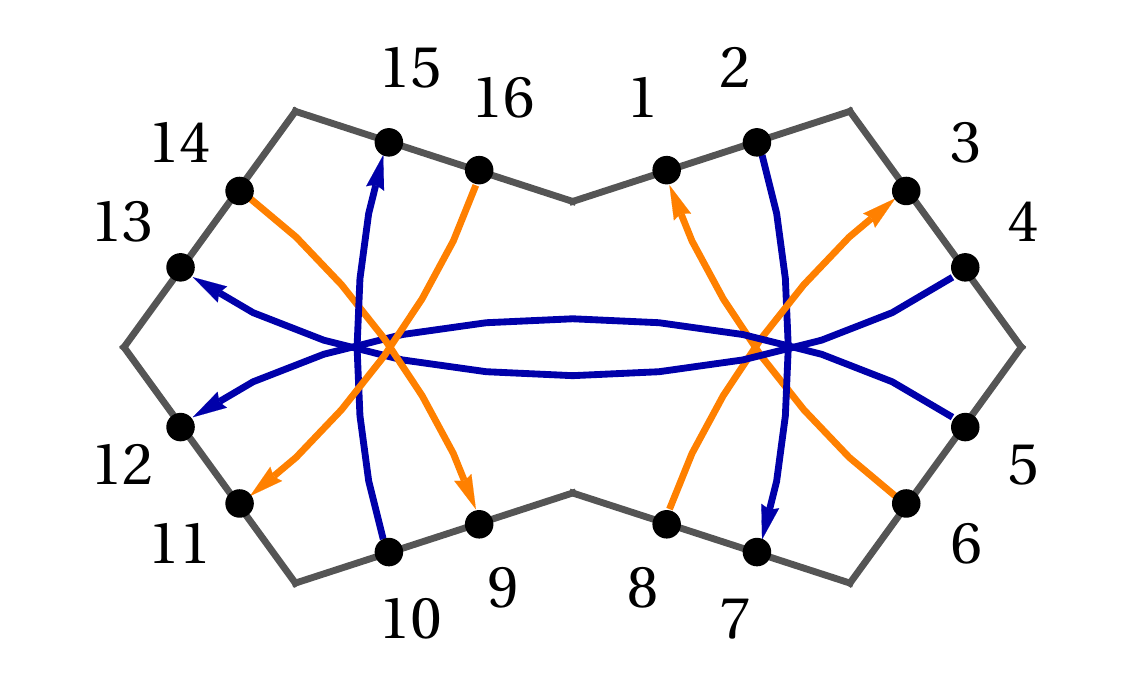}
\end{gathered} \ .
\end{align}
Note that in this labeling, the first pentagon is on the right. As we can see, each of these state vectors is distinguished from the others by at least three dimer parity flips, i.e.\ it requires more than two Majorana operators to map between them. As a result, $G^{(2)}$ only contains four diagonal terms,
\begin{align}
G^{(2)}_{j,k} = &|\alpha_1|^2 |\alpha_2|^2\, \Gamma_{j,k}^{\bar{0},\bar{0}} + |\beta_1|^2 |\beta_2|^2\, \Gamma_{j,k}^{\bar{1},\bar{1}} \nonumber\\
&+ |\alpha_1|^2 |\beta_2|^2\, \Gamma_{j,k}^{\bar{0},\bar{1}}+ |\beta_1|^2 |\alpha_2|^2\, \Gamma_{j,k}^{\bar{1},\bar{0}} \text{ ,}
\end{align}
with the covariance matrix $\Gamma_{j,k}^{b_1,b_2}$ containing two-point correlations for the basis state vector $\ket{b_1,b_2}_8$. 

This example, as well as the general case of Theorem \ref{THM_TWO_POINT_CORR}, implies that two-point functions $G^{(2)}$ preserve the correlation structure of the fixed-input covariance matrices $\Gamma$, whose entries only differ by signs (i.e.,\ dimer parities). Hence, if $\Gamma_{j,k}=0$ for a fixed logical input (no dimer connecting Majorana modes $j$ and $k$), then two-point correlations between the two modes vanish for \emph{any} bulk input. 
This is even true for the case of superpositions with entangled bulk input, where $G^{(2)}$ is still a convex sum. 
Higher-order correlation functions separate into a Gaussian part that follows Wick's theorem and have the form of products of dimer terms, and a non-Gaussian part which contains contributions from overlaps between boundary states for different logical input. 
To illustrate, consider a single pentagon with arbitrary logical input, described by $\ket\psi = \alpha \ket{\bar{0}}_5 + \beta \ket{\bar{1}}_5$ (with $|\alpha|^2 + |\beta|^2=1$). The $n$-point correlators $G^{(n)}$ until $n=4$ are given by
\begin{align}
G^{(1)}_j = &\sandwich{\psi}{\m_j}{\psi} = 0 \text{ ,} \\
G^{(2)}_{j,k} = &\i\,\sandwich{\psi}{\m_{[j} \m_{k]}}{\psi} = |\alpha|^2\, \Gamma^{\bar{0}}_{j,k} + |\beta|^2\, \Gamma^{\bar{1}}_{j,k} \text{ ,}\\
G^{(3)}_{j,k,l} = &-\i\,\sandwich{\psi}{ \m_{[j} \m_k \m_{l]}}{\psi} \nonumber\\
= &-\i\,\alpha^\star \beta \sandwich{\bar{0}}{\m_{[j} \m_k \m_{l]}}{\bar{1}} + \text{ h.c. ,} \\
G^{(4)}_{j,k,l,m} = &\sandwich{\psi}{\m_{[j} \m_k \m_l \m_{m]}}{\psi} \nonumber\\
= &|\alpha|^2\, (\Gamma^{\bar{0}}_{i,j}\Gamma^{\bar{0}}_{k,m} - \Gamma^{\bar{0}}_{i,k}\Gamma^{\bar{0}}_{j,m} + \Gamma^{\bar{0}}_{i,m}\Gamma^{\bar{0}}_{j,k}) \nonumber\\
&+ |\beta|^2\, (\Gamma^{\bar{1}}_{i,j}\Gamma^{\bar{1}}_{k,m} - \Gamma^{\bar{1}}_{i,k}\Gamma^{\bar{1}}_{j,m} + \Gamma^{\bar{1}}_{i,m}\Gamma^{\bar{1}}_{j,k}) \text{ .}
\end{align}
We have used square brackets around indices to denote anti-symmetrization. Gaussian contributions can occur only for even $n$, as only pairs of Majorana operators can map a dimer state onto itself. As $\ket{\bar{0}}_5$ and $\ket{\bar{1}}_5$ are mapped to each other by a $\mathfrak{w}=3$ operator, the non-Gaussian part appears at $n=3$:
The correlator $G^{(3)}_{j,k,l}$ can have non-zero entries for $j \in \{ 1,6 \}$, $k\in\{3,8\}$, $l\in\{5,10\}$ and its permutations, corresponding to the dimers differing between both input states (compare \eqref{HAPPY_ZERO} and \eqref{HAPPY_ONE}). As the exact entries of $G^{(3)}$ depend on the complex phases with which we define $\ket{\bar{0}}_5$ and $\ket{\bar{1}}_5$, they are not determined by the Majorana dimer structure.

Our example generalizes to large HyPeC contractions: The Gaussian part of $n$-point correlations $G^{(n)}$ is described by a convex combination of Gaussian covariance matrices, while all boundary states for fixed logical input that differ by $n$ dimer parities contribute to its non-Gaussian part. We can think of the latter as an ``interaction'' between code words that depends on how much the logical bulk input is in a superposition of the basis state vectors $\bar{0}$ and $\bar{1}$. For a completely classical version of the code, no non-Gaussianity appears.

\section{Generalized codes with Majorana dimers}
\subsection{Other stabilizer codes}
We have extensively focused on the $[[5,1,3]]$ stabilizer code as the building block of the HyPeC. However, we can construct Majorana dimer models for states on other $n$-gons, i.e.,\ more general $[[n,1,d]]$ stabilizer codes. We now show that these have properties very similar to the $n=5$ case.

We set a number of requirements to such generalizations:
\begin{itemize}
\item[I.] Stabilizer code: We require $n-1$ stabilizers (commuting products of Pauli operators) that lead to a twofold degenerate ground state, stabilizing one logical qubit.
\item[II.] Majorana dimer representation: All stabilizers should be products of two Majorana operators, up to a total parity operator $\Par$.
\item[III.] Rotational symmetry: All stabilizers $S_k$ should be cyclic permutations of $S_1$.
\end{itemize}
We may also wish to construct $n$-qubit codes that correspond to \emph{perfect tensors}. For fixed input $\bar{b}$ with $b \in \{ 0,1 \}$, this requires an isometric reduced density matrix 
\begin{equation}
\rho_A =\tr_{A^\text{C}}\ket{\bar{b}}_n \bra{\bar{b}}_n \propto \id_{2^{|A|}}
\end{equation}
for any subset $A$ of sites with size $|A|\leq n/2$. To hold for arbitrary input (i.e.,\ superpositions of $\bar{0}$ and $\bar{1}$), it is also necessary that $\bar{0}$ and $\bar{1}$ are partially orthogonal on $A^\text{C}$, i.e.,
\begin{equation}
\tr_{A^\text{C}}\ket{\bar{0}}_n \bra{\bar{1}}_n = \tr_{A^\text{C}}\ket{\bar{1}}_n \bra{\bar{0}}_n = 0 \ ,
\end{equation}
again assuming $|A| \leq n/2$.
Unfortunately, perfect tensors for qubits require states that are maximally entangled for any subdivision of sites, a condition which cannot be satisfied for $n=4$ or any $n>6$ \cite{Goyeneche:2015fda,2017PhRvL.118t0502H}. As $n<3$ does not correspond to a physical tile and we already covered the $n=5$ case, this leaves only $n=3$ and $n=6$ to be studied with Majorana dimers. However, as we will see below, none of the corresponding Majorana dimer codes can be perfect for arbitrary bulk input.

Let us start with the $n=3$ case. We can easily find a stabilizer code that conditions I - III. The stabilizers $S$ are
\begin{align}
Y_1 Y_2 &= \i\m_1\m_4 \text{ ,}& Y_1 Y_3 &= \i\Par \m_2\m_5 \text{ ,}& Y_2 Y_3 &= \i\m_3\m_6 \text{ .}
\end{align}
The twofold degenerate ground state of the stabilizer Hamiltonian $H=-\sum_k S_k$ is spanned by one parity-even and one parity-odd Majorana dimer state with pairing between modes on opposite sites.
Furthermore, for a fixed logical input $\bar{0}$ or $\bar{1}$ (but not its superpositions), the boundary state is decribed by a perfect tensor. This implies that adding such triangular tiles into the pentagon code preserves its entanglement structure only for logical basis-state input.
Note that this code is equivalent to a repetition code under $Y_k \to Z_k$. We will explore the connection to GHZ states in the next section.
Contrary to the pentagon code, embedding the states of this ``triangle code'' into a regular $\{3,k\}$ bulk tiling does not lead to interesting bulk/boundary relations, as the dimers close into loops. 

Similarly, we can consider a ``square code'' for $n=4$, where we find yet another stabilizer code with similar properties. The following stabilizers lead to a familiar Majorana form: 
\begin{equation}
\begin{aligned}
X_1 X_2 Z_3 Z_4 &= -\i\Par\m_1\m_4 \text{ ,}& Z_1 X_2 X_3 Z_4 &= -\i\Par\m_3\m_6 \text{ ,} \\
Z_1 Z_2 X_3 X_4 &= -\i\Par\m_5\m_8 \text{ ,}& X_1 Z_2 Z_3 X_4 &= -\i\m_2\m_7 \text{ .}
\end{aligned}
\end{equation}
Note that by applying a total parity operator, we can map this to an equivalent code with stabilizers $S=\langle Y_1 Y_2,\, Y_2 Y_3,\, Y_3 Y_4,\, Y_1 Y_4\rangle$ (which again highlights the GHZ-type entanglement).
As in the triangle code, Majorana dimers at a distance of three Majorana sites are paired up. Again, this implies trivial bulk loops of dimers for a regular $\{ 4,k\}$ tiling. Furthermore, this code does not lead to a perfect tensor for any logical input, as this is impossible to achive with 4-leg tensors.
One may be tempted to construct a stabilizer code with $S=\langle X_1 Z_2 X_3,\, X_2 Z_3 X_4,\, X_1 X_3 Z_4,\, Z_1 X_2 X_4 \rangle$, where each element of $S$ can also be written as a product of two Majorana operators on opposite edges (see Table \ref{TAB_STAB_CODES}). However, this choice of $S$ only stabilizes a single state instead of a full qubit, as the ground state is non-degenerate. Interestingly, this ground state does fulfill the perfect tensor property for connected subsets of the boundary legs.

Beyond the familiar $n=5$ case (with permutations of $X_1 Z_2 Z_3 X_4$ as stabilizers), we can construct another code by exchanging $X_k \leftrightarrow Y_k$, which is equivalent to exchanging $\m_{2k-1} \leftrightarrow \m_{2k}$ and leads to the stabilizers
\begin{equation}
\begin{aligned}
Y_1 Y_3 Z_4 Z_5 &= \i\Par\m_2\m_5 \text{ ,}& Y_1 Z_2 Z_3 Y_4 &= \i\m_1\m_8 \text{ ,} \\
Z_1 Y_2 Y_4 Z_5 &= \i\Par\m_4\m_7 \text{ ,}& Y_2 Z_3 Z_4 Z_5 &= \i\m_3\m_{10} \text{ ,}\\
Z_1 Z_2 Y_3 Y_5 &= \i\Par\m_6\m_9 \text{ .}
\end{aligned}
\end{equation}
Naturally, this code inherits the properties of the original $[[5,1,3]]$ code, including the perfect tensor property for any logical input.

Attempting to generalize to $n=6$, we find the following choice for the elements of $S$:
\begin{equation}
\begin{aligned}
X_1 X_4 Z_5 Z_6 &= -\i\Par\m_1\m_8 \text{ ,}& X_1 Z_2 Z_3 X_4 &= -\i\m_2\m_7 \text{ ,} \\
Z_1 X_2 X_5 Z_6 &= -\i\Par\m_3\m_{10} \text{ ,}& X_2 Z_3 Z_4 X_5 &= -\i\m_4\m_9 \text{ ,}\\
Z_1 Z_2 X_3 X_6 &= -\i\Par\m_5\m_{12} \text{ ,}& X_3 Z_4 Z_5 X_6 &= -\i\m_6\m_{11} \text{ .}
\end{aligned}
\end{equation}
The $n=6$ case resembles the $n=3$ result, as partial traces $\tr_{A^\text{C}}\ketbra{\bar{0}}{\bar{1}}$ do not usually vanish. Contrary to the $n=3$ case, it is also possible to form subsystems $A$ of size $|A|\leq n/2$ that are disjoint.
In both cases, the reduced density matrix $\rho_A$ is not an isometry. In other words, this code is only perfect for basis-state inputs and connected subsystems $A$.

We find similar properties for $n>6$ codes: While it is impossible to construct a perfect tensor for all (possibly disjoint) boundary regions $A$, we can always construct a Majorana dimer code with basis states $\bar{0}$ and $\bar{1}$ that are each perfect for connected subsystems $A$ by connecting Majorana dimer modes on opposite edges.
For $n=4k{+}1$, $k\in \mathbb{N}$, this construction even leads to codes where $\tr_{A^\text{C}}\ket{\bar{0}} \bra{\bar{1}} =0$ for a connected subsystem $A$ with $|A| \leq n/2$. 
Such a \emph{block perfect} code leads to an isometric $\rho_A$ for superpositions of bulk input for any connected $A$. 
The $n=9$ case, whose stabilizers are permutations of $X_1 Z_2 Z_3 Z_4 Z_5 X_6$, is visualized in Table \ref{TAB_STAB_CODES}.
Note that block perfect holographic codes can also be constructed from CSS codes \citep{PhysRevA.98.052301}.

\begin{table}[htb]
\begin{tabular}{l | c | c | c | c}
$n$ & Stabilizer & $\ket{\bar{0}}_n$ & $\ket{\bar{1}}_n$ & P\\
\hline
3 &
$Y_1 Y_2$ &
$\begin{gathered}
\includegraphics[height=0.09\textheight]{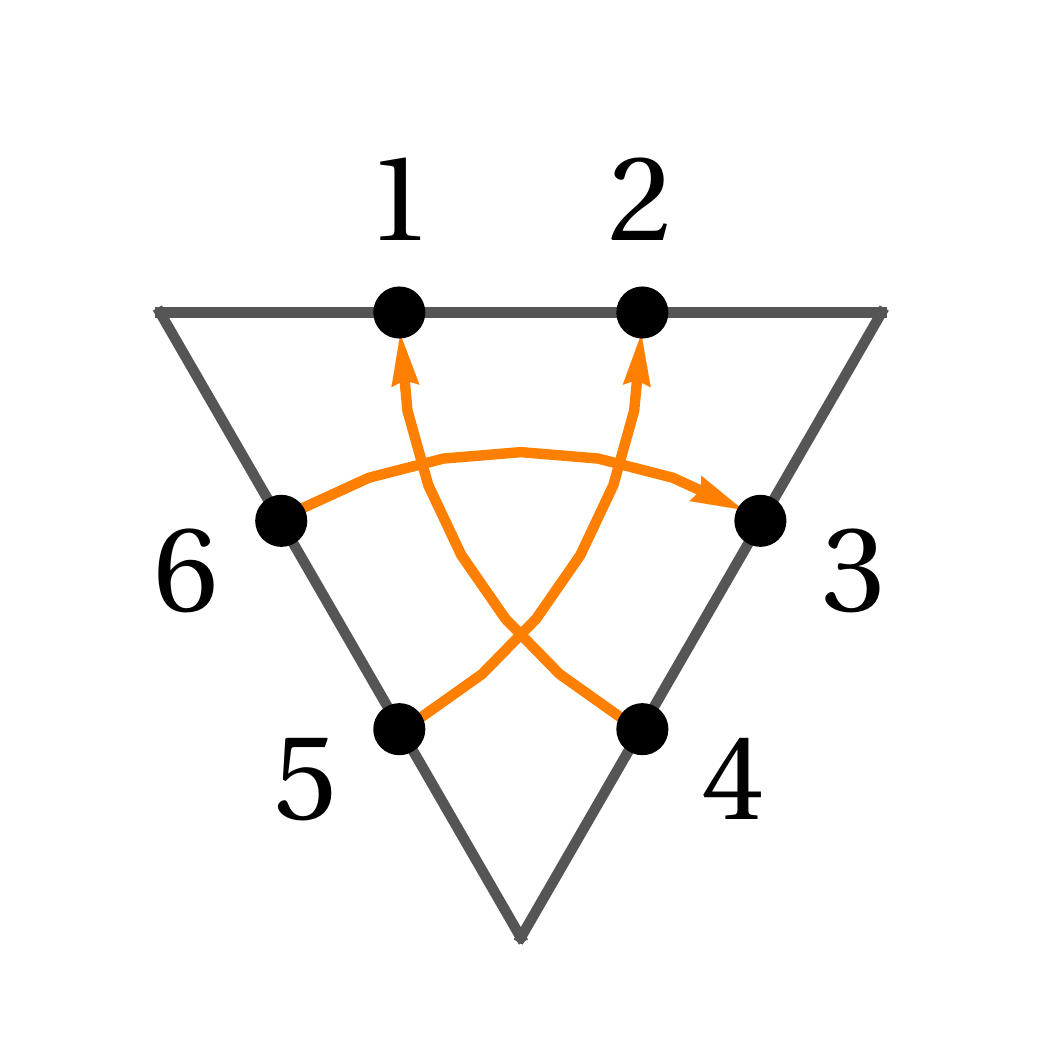}
\end{gathered}$
&$\begin{gathered}
\includegraphics[height=0.09\textheight]{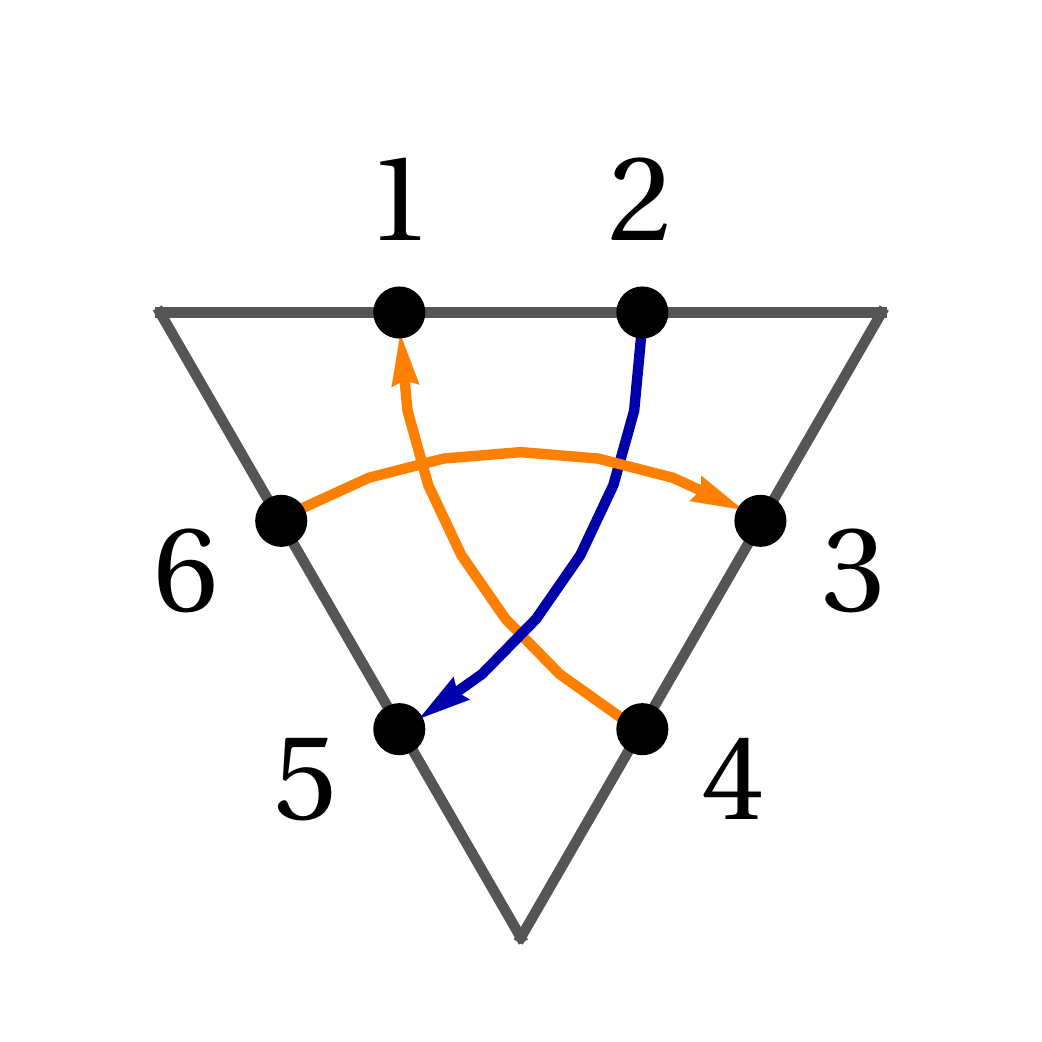}
\end{gathered}$ & 
$\times$ \\
4 &
$X_1 Z_2 Z_3 X_4$ &
$\begin{gathered}
\includegraphics[height=0.09\textheight]{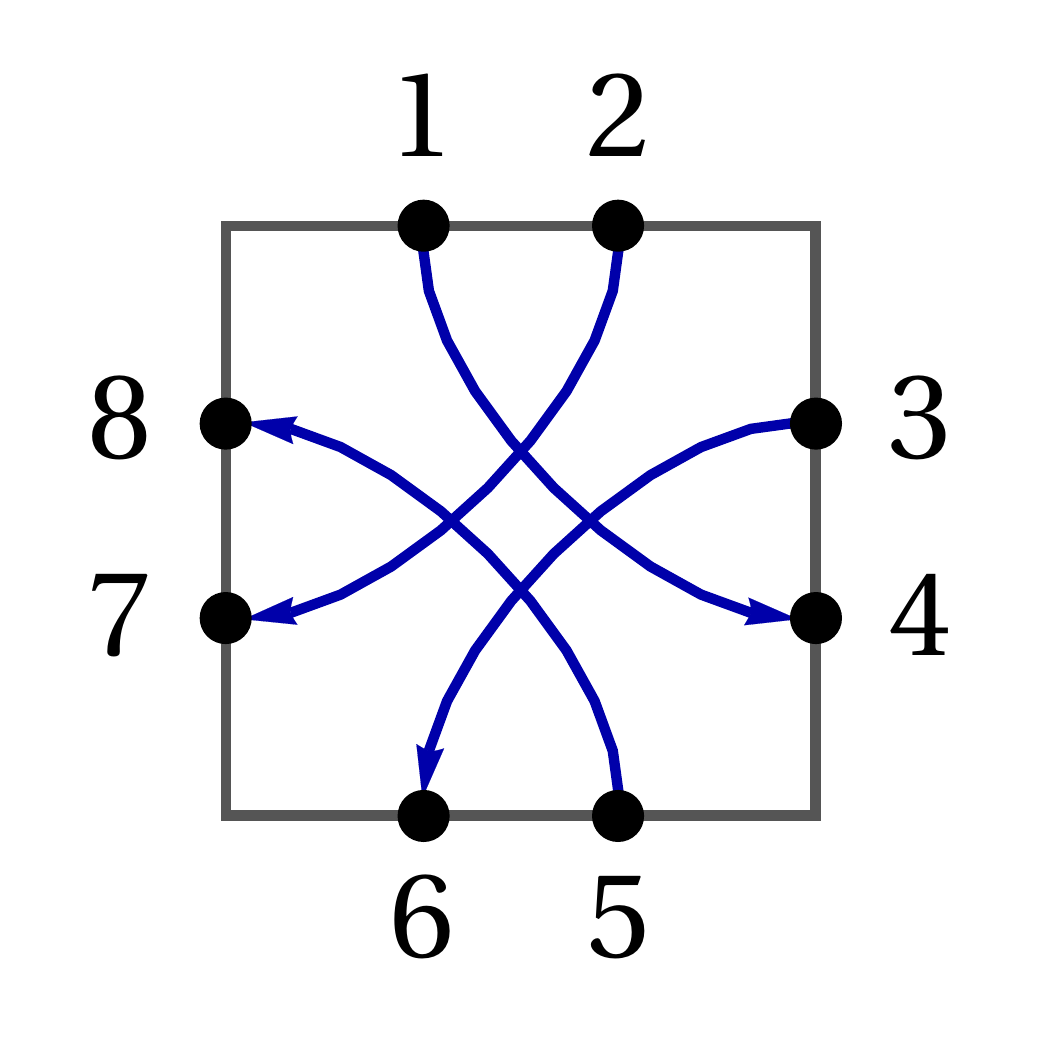}
\end{gathered}$
&$\begin{gathered}
\includegraphics[height=0.09\textheight]{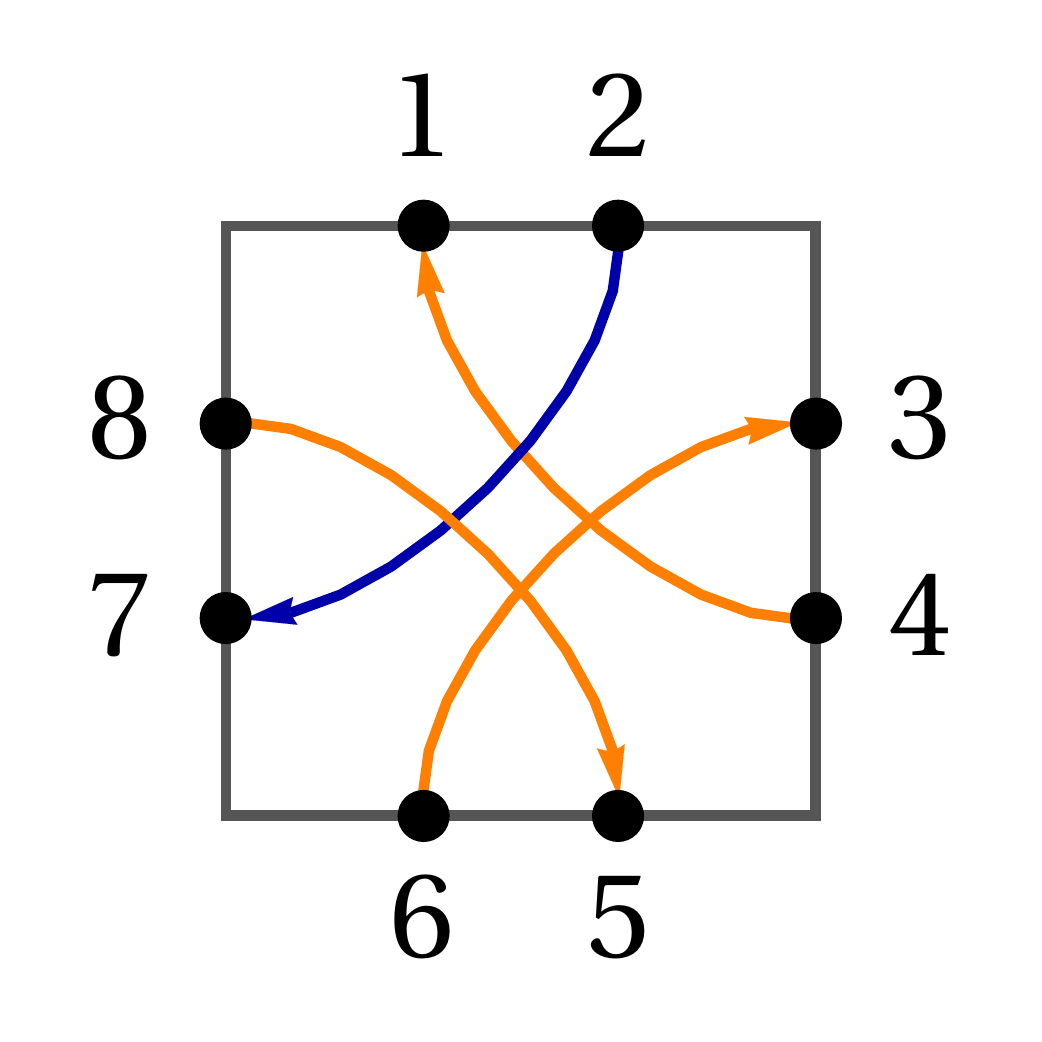}
\end{gathered}$ &
$\times$  \\
4 &
$X_1 Z_2 X_3$ &
$\begin{gathered}
\includegraphics[height=0.09\textheight]{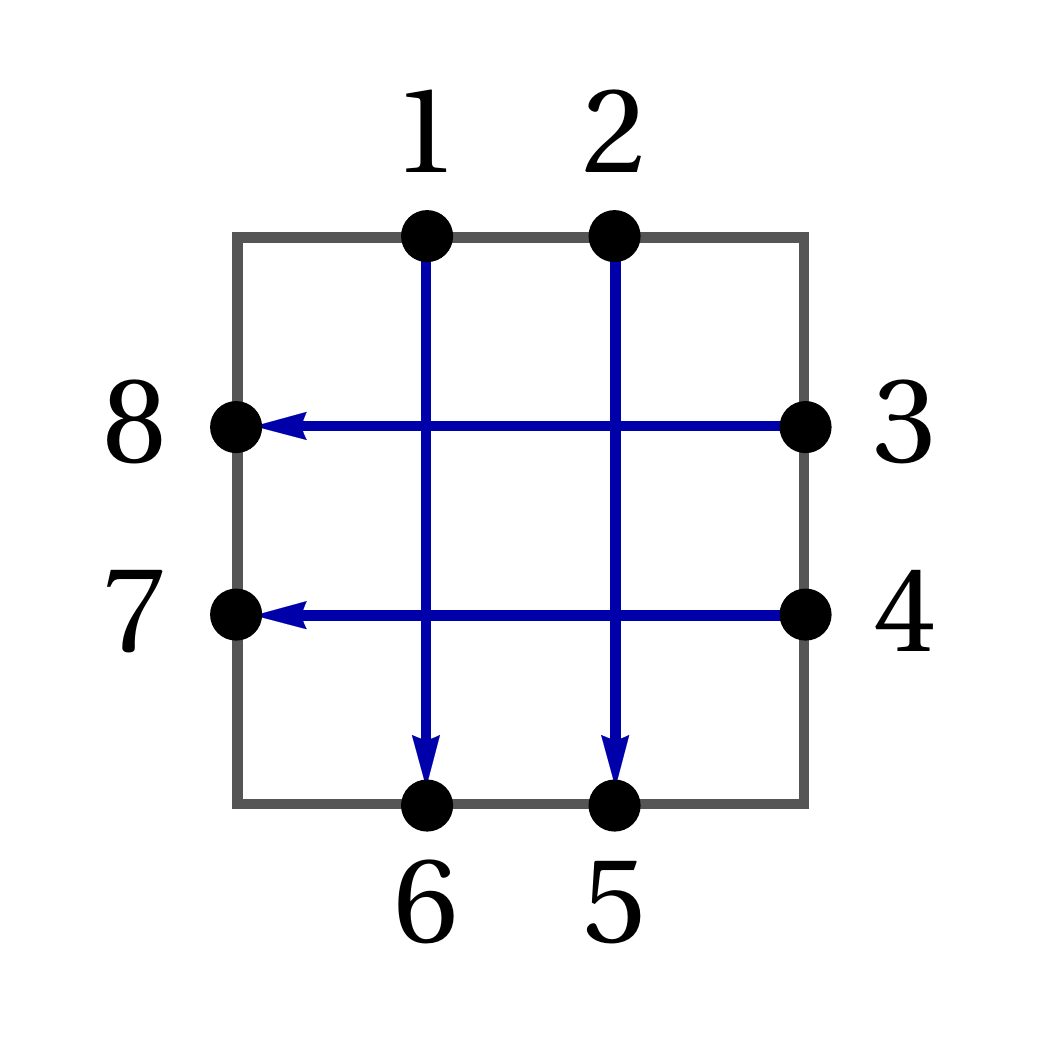}
\end{gathered}$
& - &
\checkmark \\
5 &
$X_1 Z_2 Z_3 X_4$ &
$\begin{gathered}
\includegraphics[height=0.09\textheight]{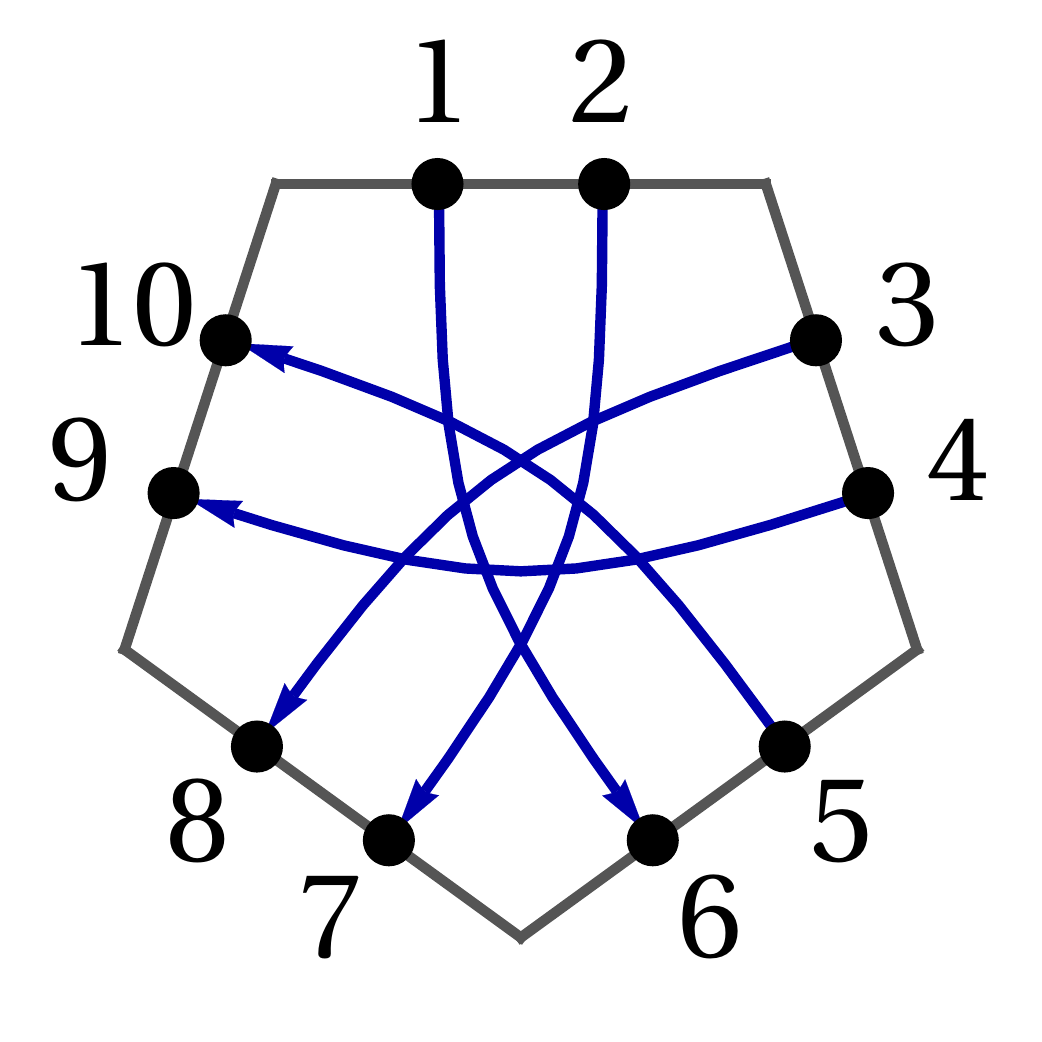}
\end{gathered}$
&$\begin{gathered}
\includegraphics[height=0.09\textheight]{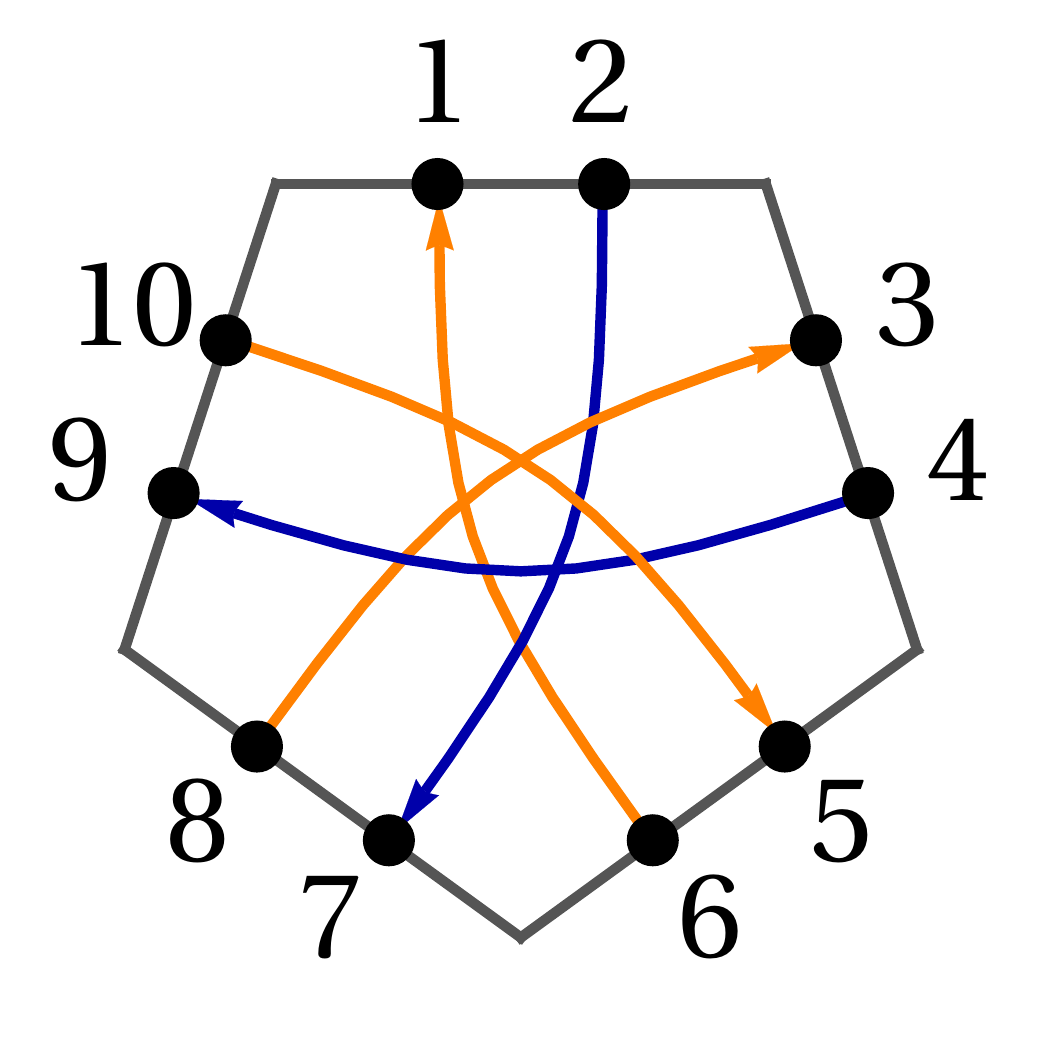}
\end{gathered}$ & 
\checkmark\checkmark\\
5 &
$Y_1 Z_2 Z_3 Y_4$ &
$\begin{gathered}
\includegraphics[height=0.09\textheight]{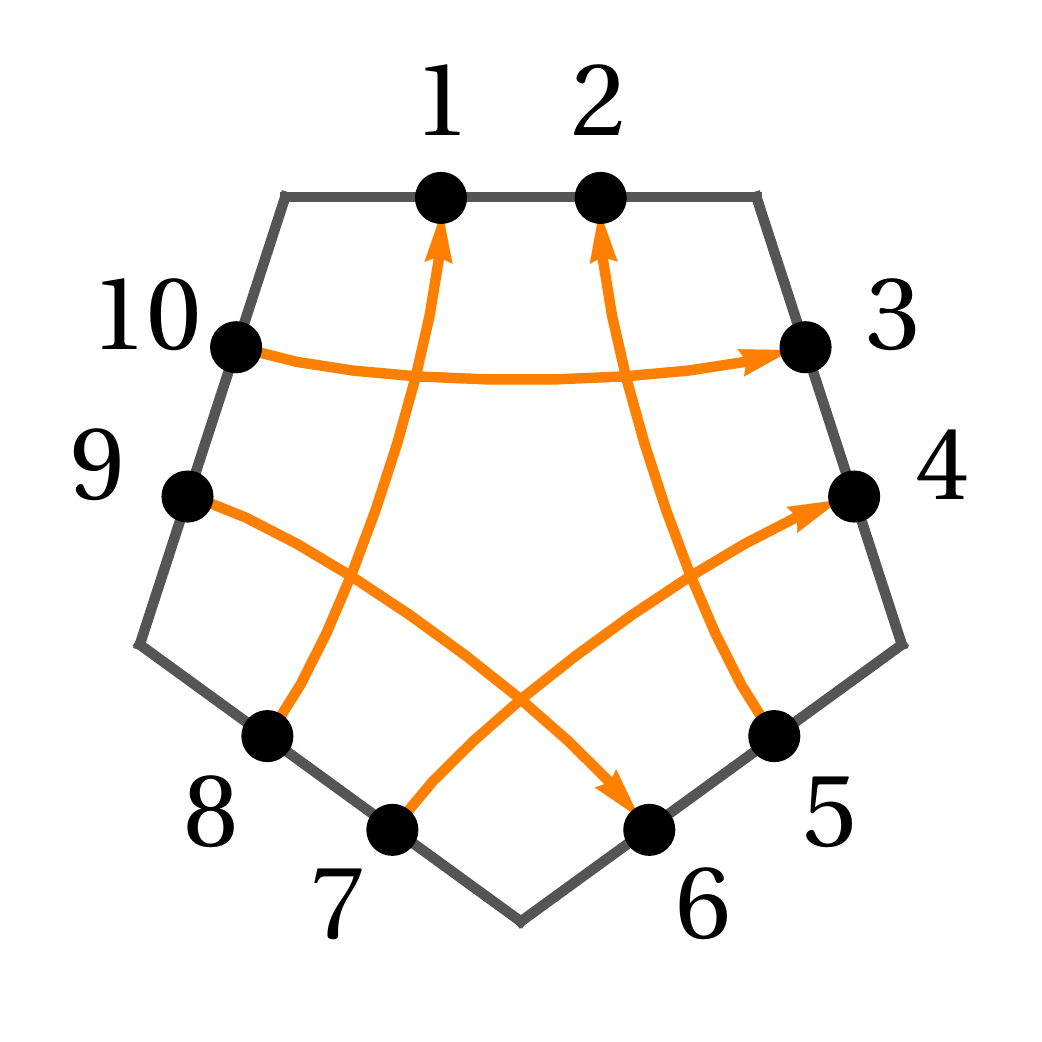}
\end{gathered}$
&$\begin{gathered}
\includegraphics[height=0.09\textheight]{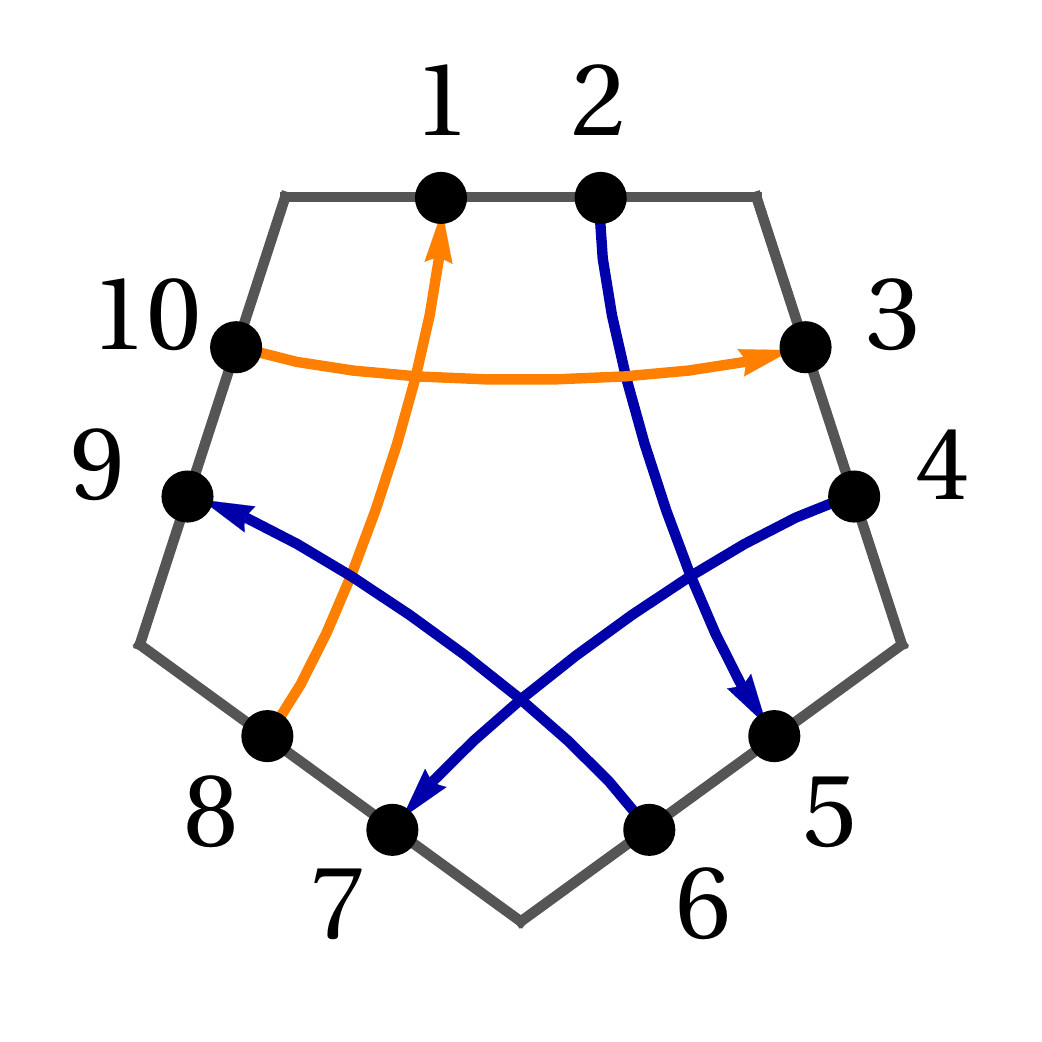}
\end{gathered}$ & 
\checkmark\checkmark\\
6 &
$X_1 Z_2 Z_3 X_4$ &
$\begin{gathered}
\includegraphics[height=0.095\textheight]{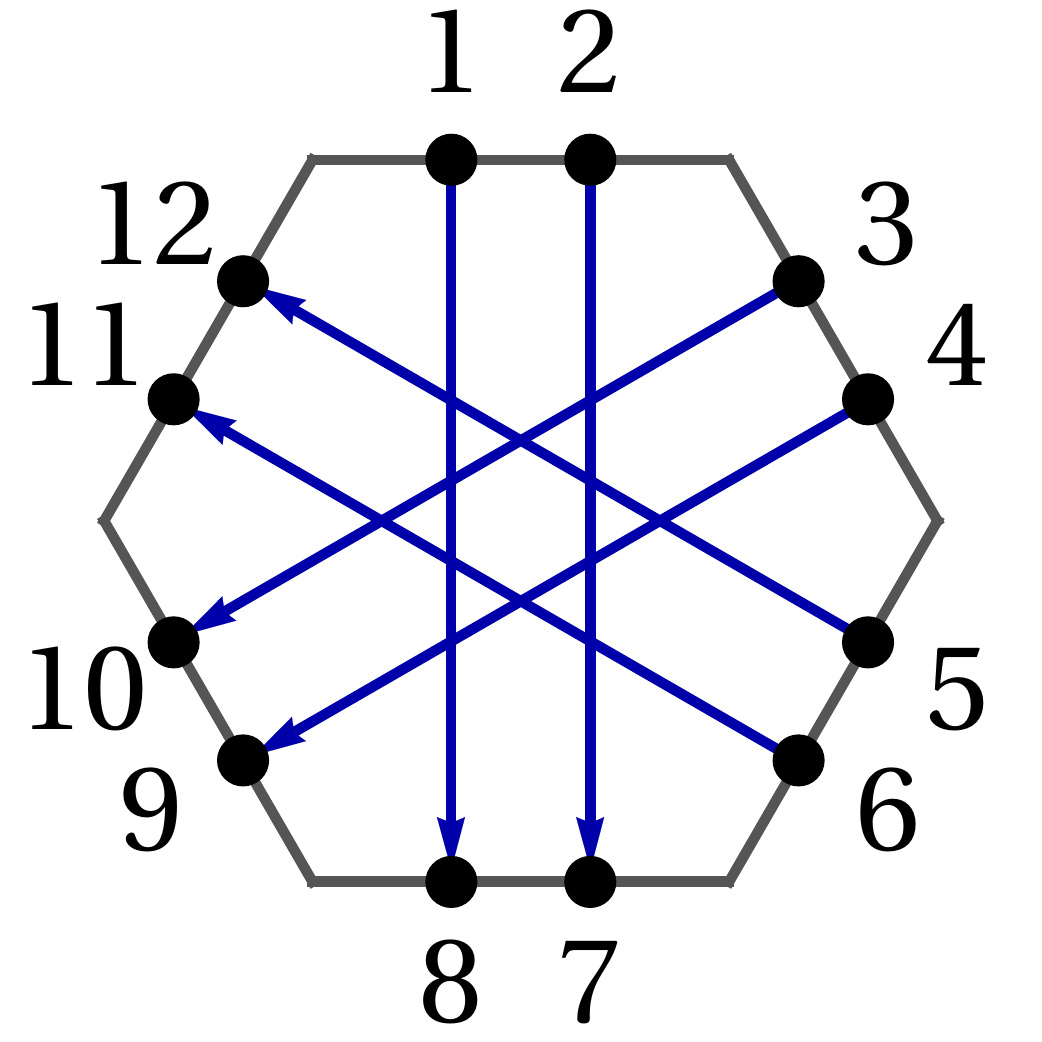}
\end{gathered}$
&$\begin{gathered}
\includegraphics[height=0.095\textheight]{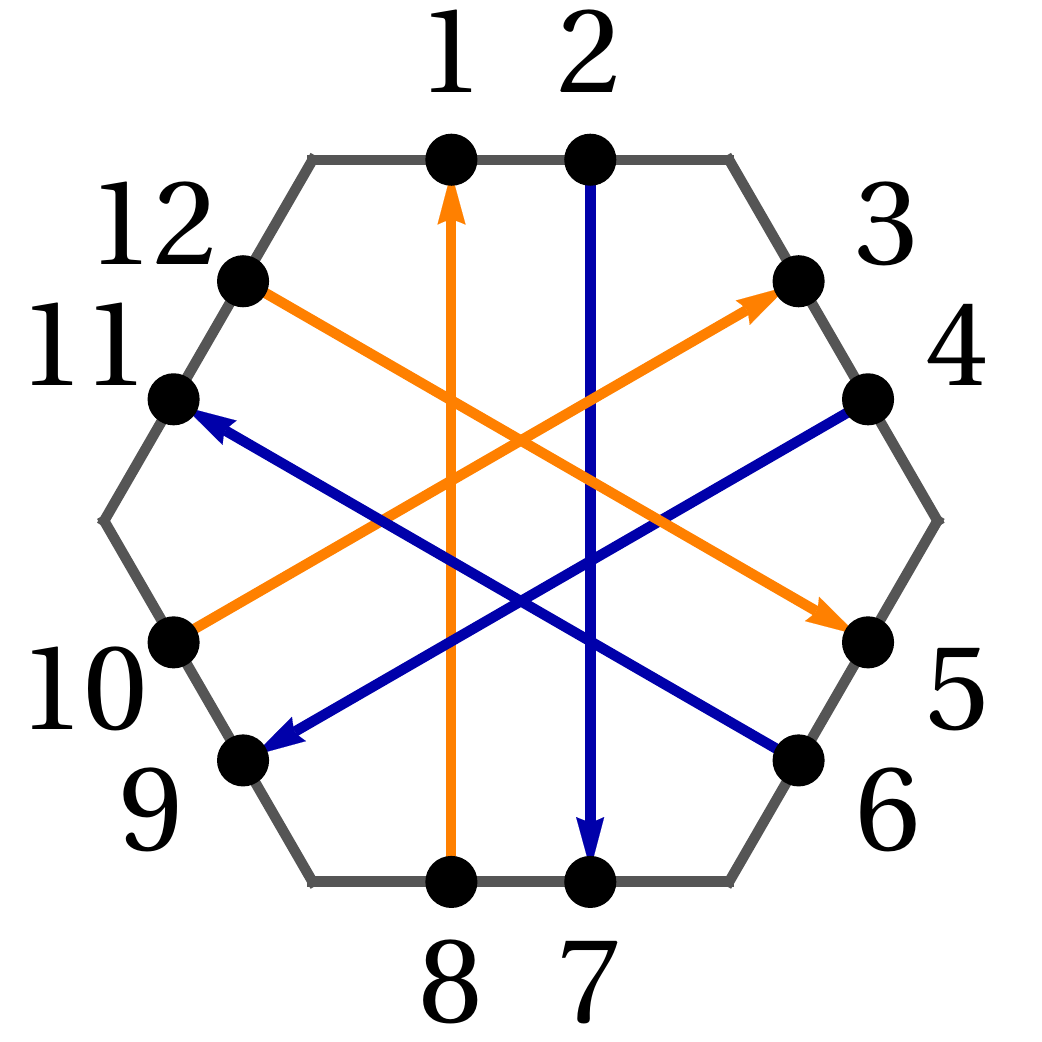}
\end{gathered}$ &
$\times$
\\
 & & & & \\
$\vdots$ & & & & \\
 & & & & \\
9 &
$X_1 Z_2 Z_3 Z_4 Z_5 X_6$ &
$\begin{gathered}
\includegraphics[height=0.1\textheight]{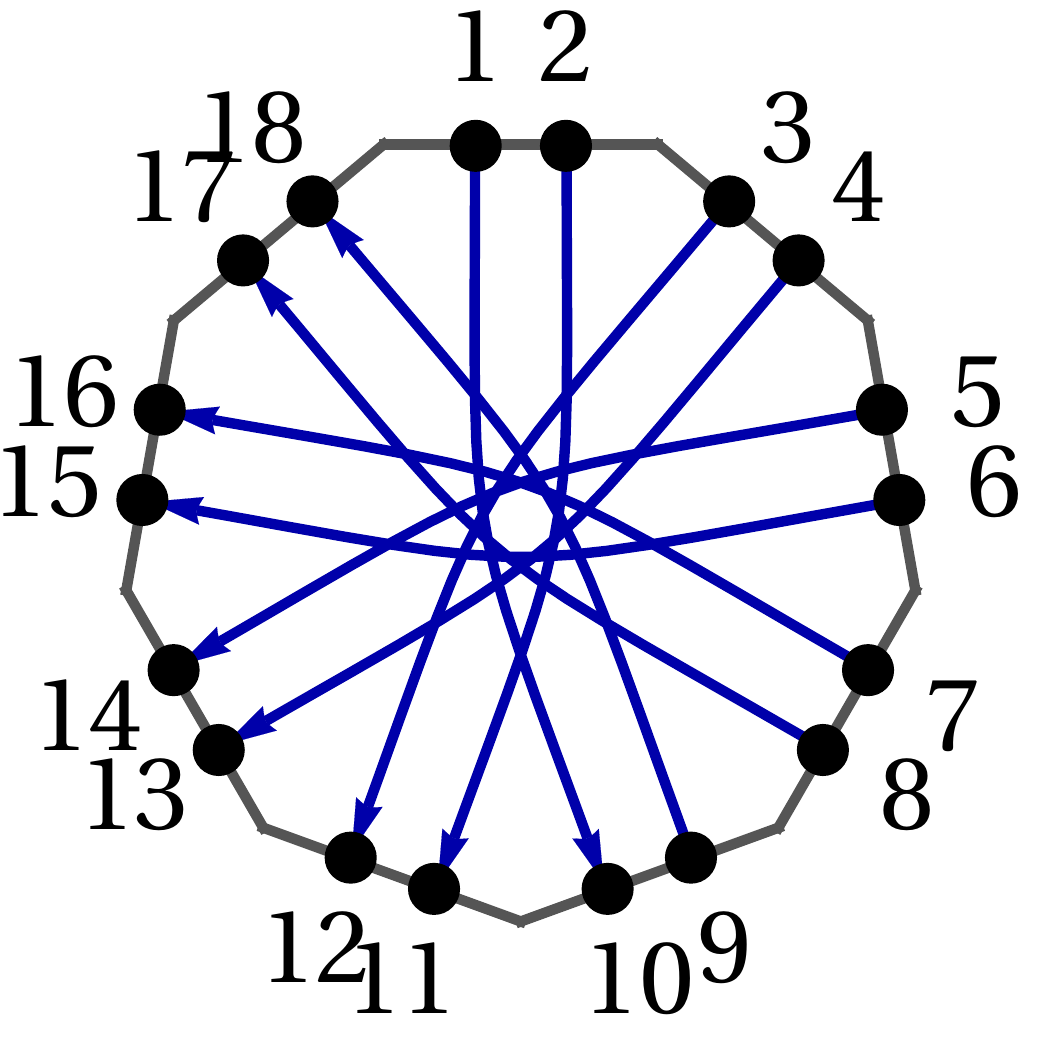}
\end{gathered}$
&$\begin{gathered}
\includegraphics[height=0.1\textheight]{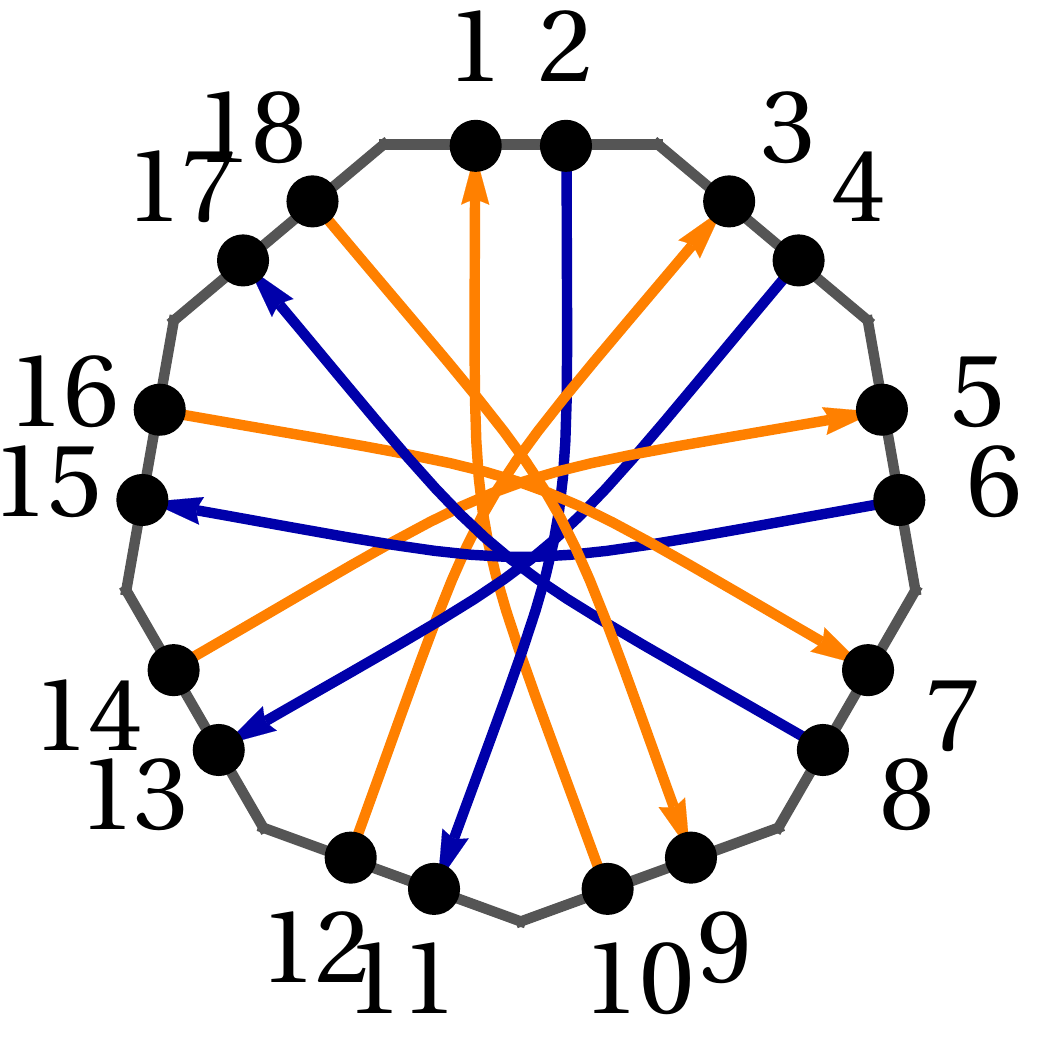}
\end{gathered}$ &
\checkmark
\end{tabular}
\caption{Possible generalizations of the $[[5,1,3]]$ pentagon code (fourth row) to an $n$-gon code. All stabilizers are cyclic permutations of the one given in the second column. The last column indicates whether boundary states lead to block perfect tensors (\checkmark) or fully perfect tensors (\checkmark\checkmark). 
}
\label{TAB_STAB_CODES}
\end{table}

While we would expect a bulk tiling of each of the two $n=5$ codes to lead to similar boundary properties, it would be interesting to investigate codes built from combinations of perfect and block perfect tensors.

\subsection{GHZ states}

The $n=3$ code considered previously possesses a peculiar property: The logical eigenstates are GHZ states in the $Y$ basis, i.e.,\ $\ket{\bar{0}}_3 = \ket{Y+}_3^\text{GHZ}$ and $\ket{\bar{1}}_3 = \ket{Y-}_3^\text{GHZ}$, using the definition
\begin{equation}
\label{EQ_GHZ_STATES}
\begin{aligned}
\ket{Y+}_n^\text{GHZ} &= \frac{1}{\sqrt{2}} \left( \ket{y_+}^{\otimes n} - \ket{y_-}^{\otimes n} \right) \text{ ,} \\
\ket{Y-}_n^\text{GHZ} &= \frac{1}{\sqrt{2}} \left( \ket{y_+}^{\otimes n} + \ket{y_-}^{\otimes n} \right)  \text{ ,} 
\end{aligned}
\end{equation}
where $\ket{y_\pm}$ are the eigenstates of $\sigma^y$ with $\sigma^y \ket{y_\pm} = \pm \ket{y_\pm}$. This is because $\ket\pm_n^\text{GHZ}$ is in the ${+}1$ eigenspace of the stabilizer $S_1 = Y_1 Y_2$ and its permutations, and thus in the ground state space of corresponding stabilizer Hamiltonian. The total parity $\Par \ket\pm_n^\text{GHZ} = \pm \ket\pm_n^\text{GHZ}$ follows from the relation $\sigma_z \ket{y_\pm} = - \ket{y_\mp}$.

We can easily generalize these $Y$-basis GHZ states to higher $n$. Using \eqref{EQ_JORDANWIGNER_REV}, we find $Y_k Y_{k+1} = \i \m_{2k-1}\m_{2k+2}$ for $k<n$ and $Y_1 Y_n = \i \Par \m_{2}\m_{2n-1}$. This fixes the Majorana dimers for any $n$ to a $(2k{+}2 \mod 2n)\mapsto(2k{-}1)$ pairing ($k \in \{1,\dots,n\}$), with the last dimer parity flipped in the $Y-$ state. For example, the GHZ state vectors on a pentagon are:
\begin{align}
\ket{Y+}_5^\text{GHZ} &\;=\;
\begin{gathered}
\includegraphics[height=0.11\textheight]{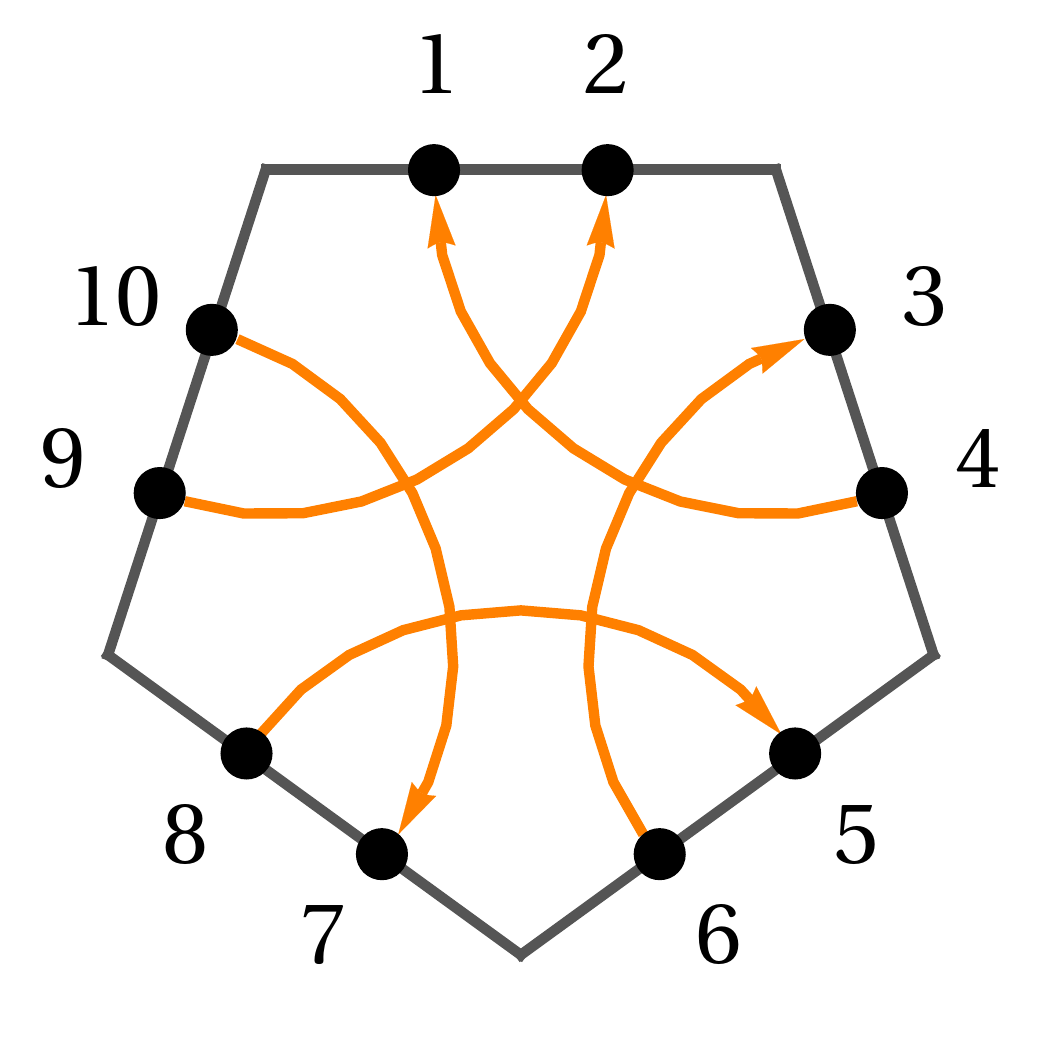}
\end{gathered} \text{ ,}\\
\ket{Y-}_5^\text{GHZ} &\;=\;
\begin{gathered}
\includegraphics[height=0.11\textheight]{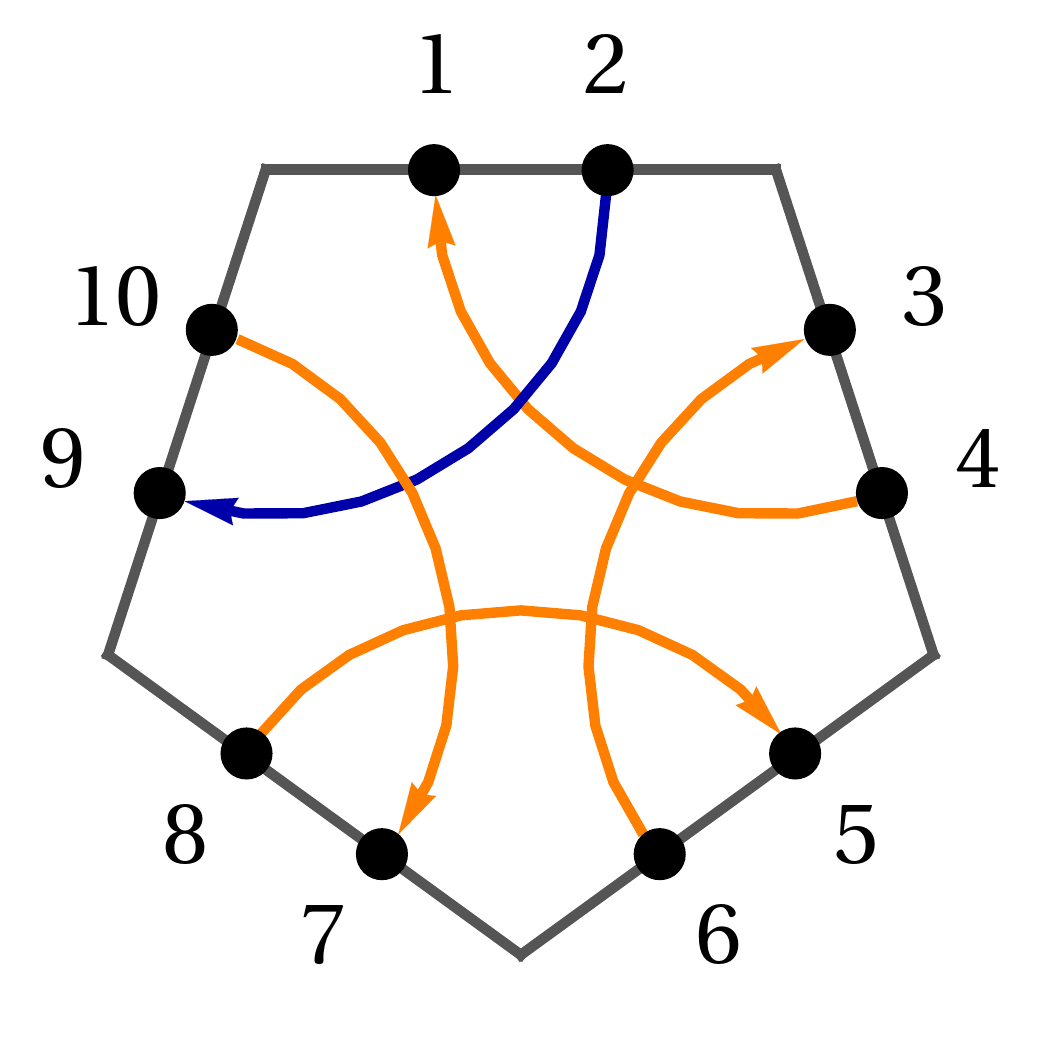}
\end{gathered} \text{ .}
\end{align}
Similarly, we can construct $n$-qubit GHZ states in the $X$ basis: As $X_k X_{k+1} = -\i\m_{2k}\m_{2k+1}$ and $X_1 X_n = \i\Par\m_1\m_{2n}$, we find a $2k \to 2k{+1} \mod 2n$ pairing, with the last dimer flipped in the $X+$ state. For the $n=5$ case, the corresponding diagrams are:
\begin{align}
\ket{X+}_5^\text{GHZ} &\;=\;
\begin{gathered}
\includegraphics[height=0.11\textheight]{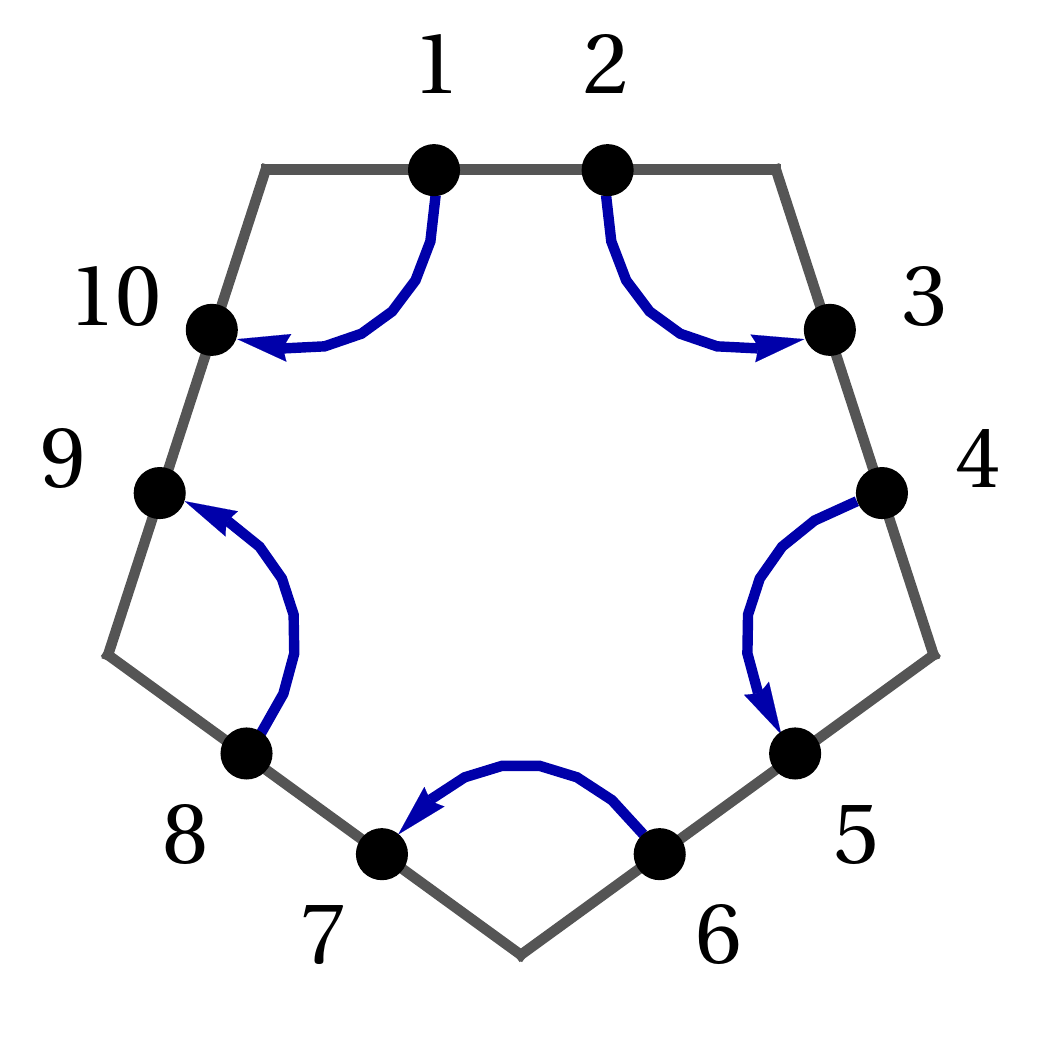}
\end{gathered} \text{ ,}\\
\ket{X-}_5^\text{GHZ} &\;=\;
\begin{gathered}
\includegraphics[height=0.11\textheight]{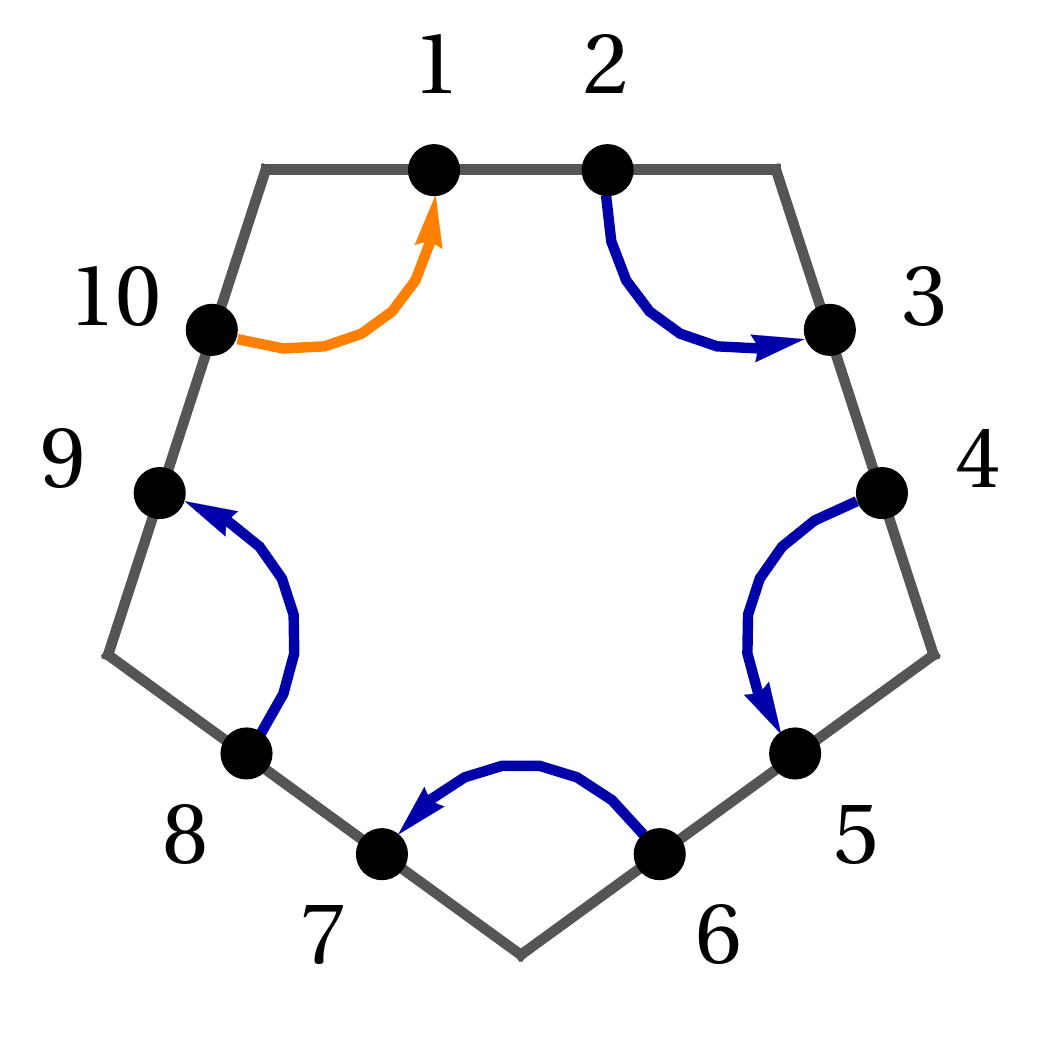}
\end{gathered} \text{ .}
\end{align}
As a general rule, the positive-parity GHZ states are rotationally invariant in dimer parities, while the negative-parity GHZ states are rotationally invariant in dimer \emph{orientation} (i.e.,\ direction of arrows). As shown in Sec.\ \ref{APP_CYCL_PERM}, this means that the underlying spin degrees of freedom are invariant under a cyclic permutation of indices of the tensors specifying the GHZ states.

In the Majorana dimer language, we can also see that the $[[5,1,3]]$ logical code states are extensions of GHZ states: All have a completely symmetric entanglement structure, but whereas the $X\pm$ and $Y\pm$ GHZ states connect Majorana modes at a distance of $d=1$ and $d=3$ modes, respectively, the $\bar{0}$ and $\bar{1}$ logical eigenstates pair modes $5$ sites apart. While an even $d$ cannot lead to rotational symmetry, we can systematically construct all of these states by considering all odd $d$. For example, the $n=9$ case in Table \ref{TAB_STAB_CODES} corresponds to a $d=9$ pairing.

\subsection{Majorana dimers and Majorana codes}

So far, we have only discussed quantum error correction in a system of spins which we effectively described by fermionic degrees of freedom. Another approach is to build quantum error correction in fundamentally fermionic systems and then describe the actions of Majorana operators in such codes \cite{Bravyi:2010de}.
While superficially similar to our treatment of the HyPeC, there are fundamental differences: The advantage of actual Majorana codes is the use of fermion super-selection to reduce the occurrence of logical errors by encoding them in operators that are odd in Majorana operators, and thus cannot occur if the system is in a purely bosonic environment. However, our Majorana dimer model encodes the $\bar{0}$ and $\bar{1}$ states in \emph{different} fermionic parity sectors, superpositions of which would thus be forbidden in a system composed of actual fermions.
It follows that our Majorana dimer description of the $[[5,1,3]]$ stabilizer code is different from the Bravyi-Terhal-Leemhuis prescription to turn stabilizer into Majorana codes, which uses four Majorana modes to encode one spin degrees of freedom. 
However, Majorana dimers can still be a useful tool for studying Majorana codes. Consider a simple \emph{Kitaev chain} \cite{Kitaev2001} of $2N$ Majorana modes in the ground state of the stabilizer Hamiltonian
\begin{equation}
H = -\i \sum_{k=1}^{N-1} \m_{2k}\m_{2k+1} \text{ .}
\end{equation}
The ground state is two-fold degenerate, but can be spanned by two Majorana dimer state vectors $\ket{\pm}_N$. Explicitly for $N=6$,
\begin{align}
\ket{+}_6 &= 
\begin{gathered}
\includegraphics[height=0.08\textheight]{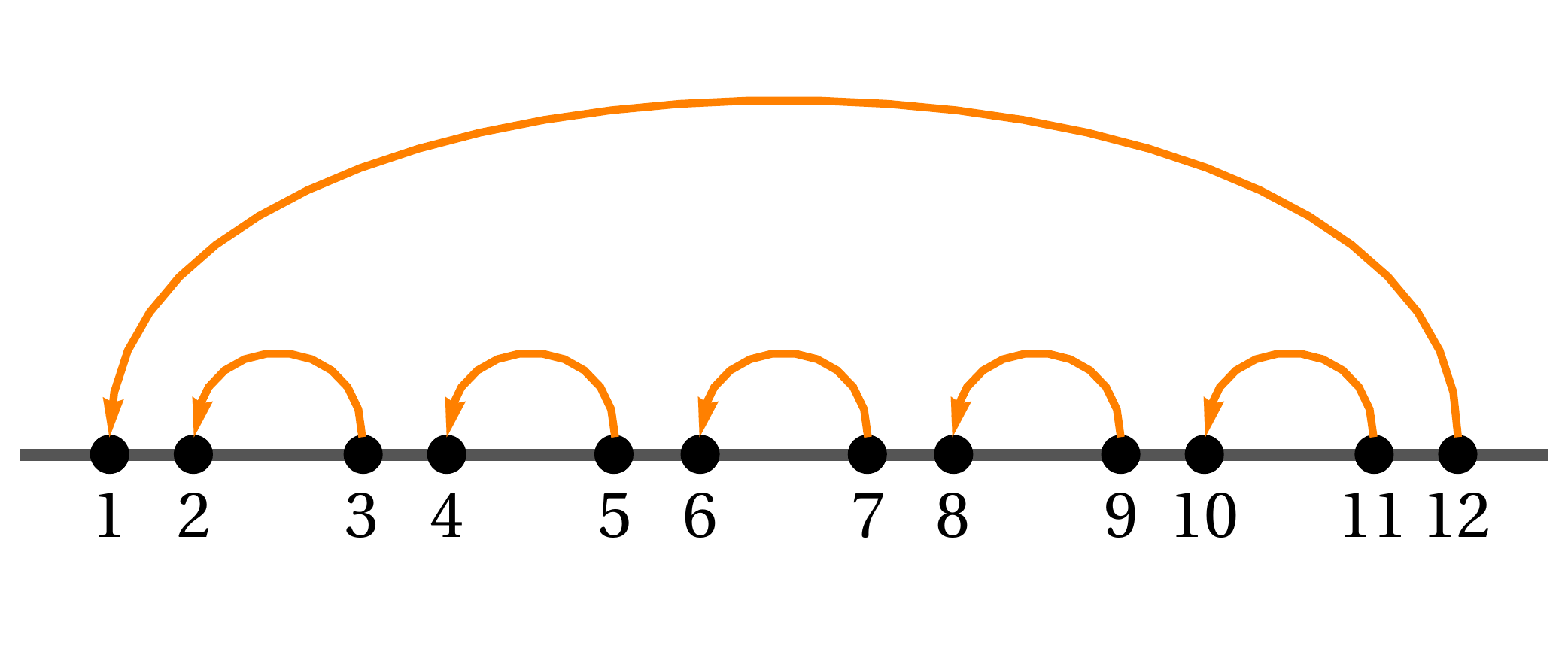}
\end{gathered} \text{ ,} \\
\ket{-}_6 &= 
\begin{gathered}
\includegraphics[height=0.08\textheight]{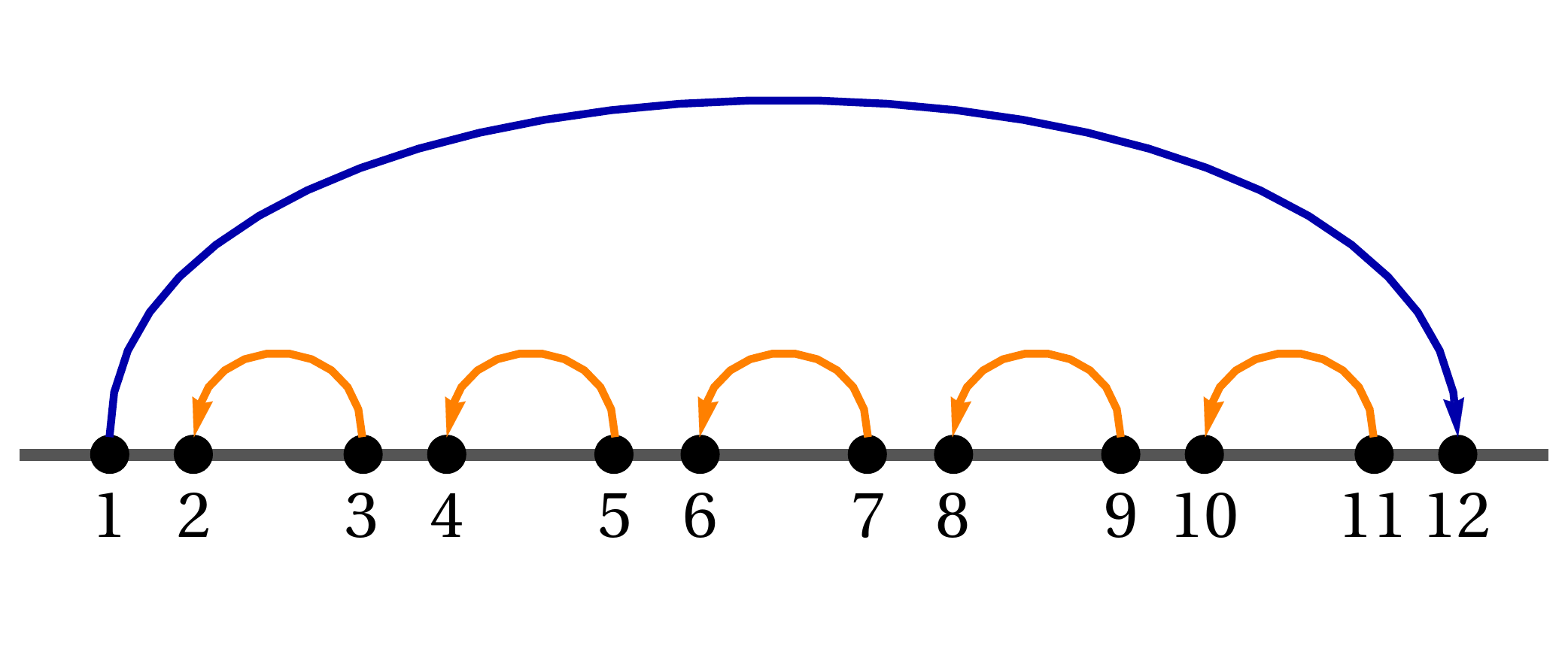}
\end{gathered} \text{ .}
\end{align}
From \eqref{EQ_TOTAL_PARITY} we immediately see that $\ket{\pm}_N$ has fermionic parity $\pm 1$. While the logical code states can be easily mapped into each other by applying the operator $\m_1$ or $\m_{12}$ (which flips the parity of the $1 \to 12$ dimer), any physical error has to respect fermion parity and locality and is therefore restricted to the form $\m_k \m_{k+1}$, i.e.,\ even nearest-neighbour terms. Thus, a phase error requires a string $\m_1\m_2\dots \m_{12}$ of Majorana operators with $\mathfrak{w}=2N$, endowing the ground state of the Kitaev chain with \emph{topological protection}. Clearly, this approach can be generalized to any Majorana dimer state of $2N$ Majorana modes: By fixing $N{-}k$ dimers, we leave a $k$-dimensional logical qubit subspace on the remaining $k$ possible dimers. If these remaining $k$ modes are far apart, then they will be robust against errors that are both even and local in Majorana operators.

Furthermore, we can express Majorana stabilizer codes with dimers even if the stabilizer generators are not quadratic in Majorana operators. Consider $N=4$ edges with eight Majorana modes under the stabilizers $S=\langle -\m_1\m_3\m_5\m_7,\, -\m_2\m_4\m_6\m_8\rangle$. The ${+}1$ eigenspace of each $S_k$ is spanned by two Majorana dimer states on the corresponding modes, one where both dimer parities are even and one where they are odd. We can thus define the logical 2-qubit state vectors $\ket{\bar{b}_1,\bar{b}_2}$ as follows:
\begin{equation}
\begin{aligned}
\ket{\bar{0},\bar{0}} &= 
\begin{gathered}
\includegraphics[height=0.09\textheight]{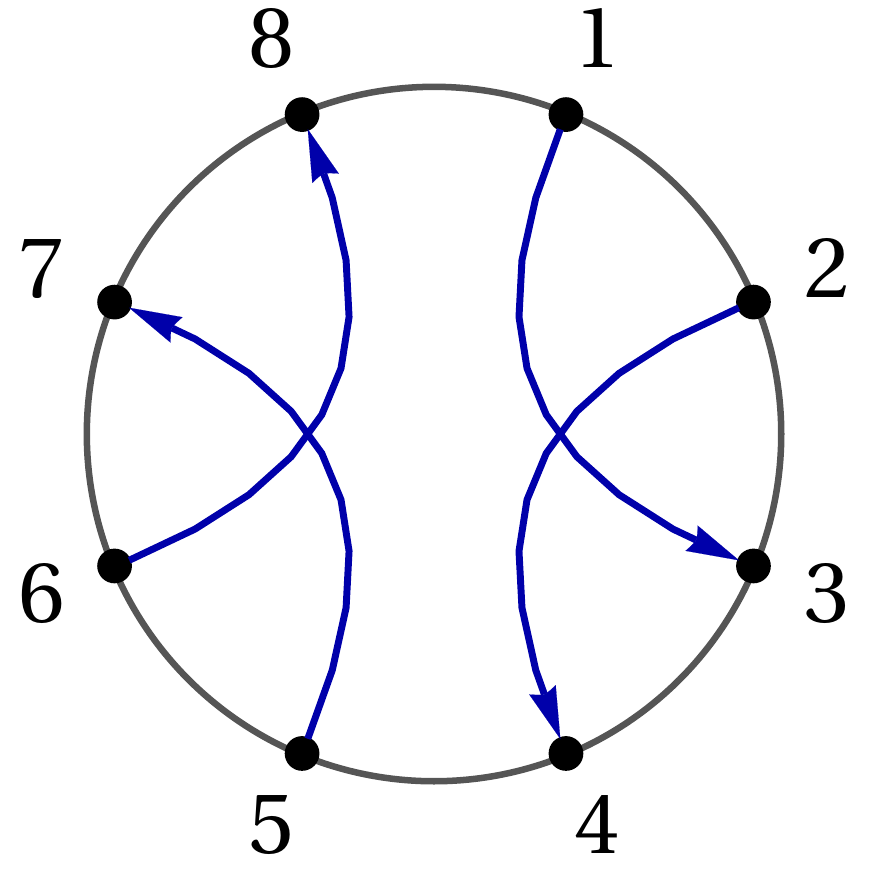}
\end{gathered} \text{ ,} &
\ket{\bar{0},\bar{1}} &= 
\begin{gathered}
\includegraphics[height=0.09\textheight]{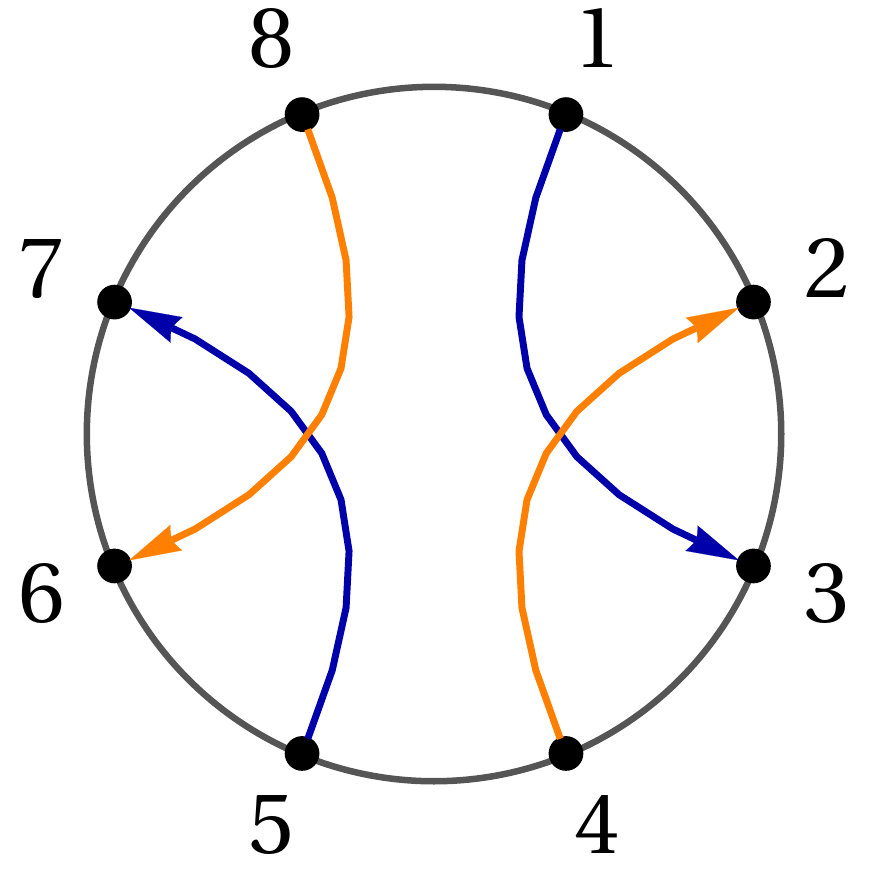}
\end{gathered} \text{ ,} \\
\ket{\bar{1},\bar{0}} &= 
\begin{gathered}
\includegraphics[height=0.09\textheight]{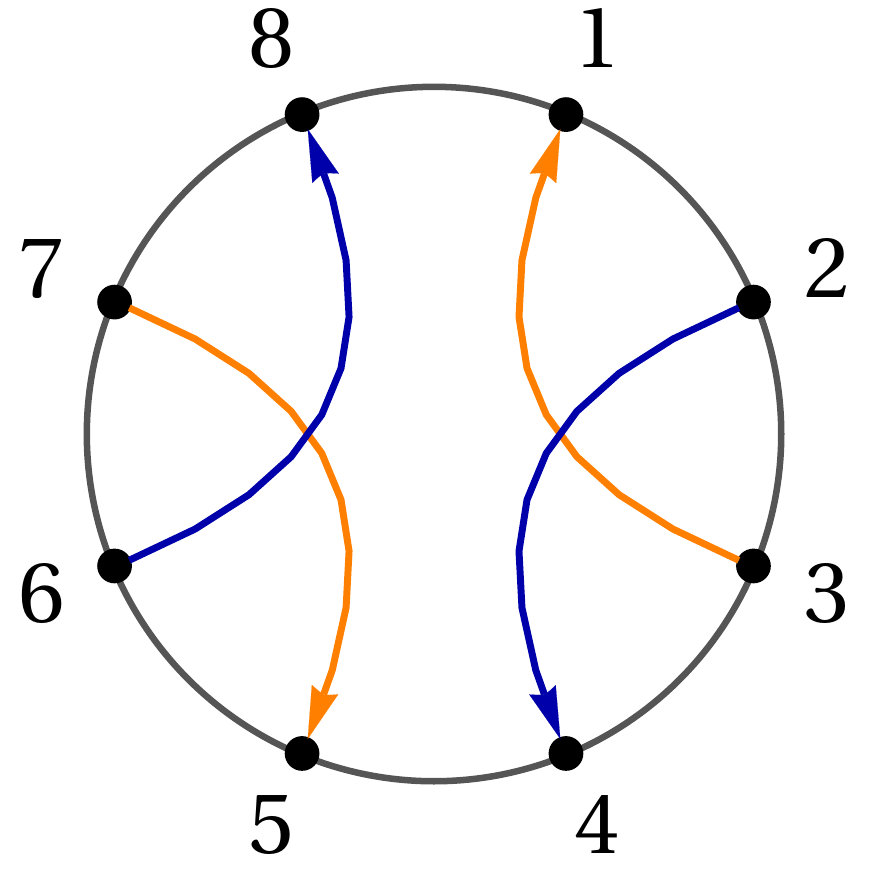}
\end{gathered} \text{ ,} &
\ket{\bar{1},\bar{1}} &= 
\begin{gathered}
\includegraphics[height=0.09\textheight]{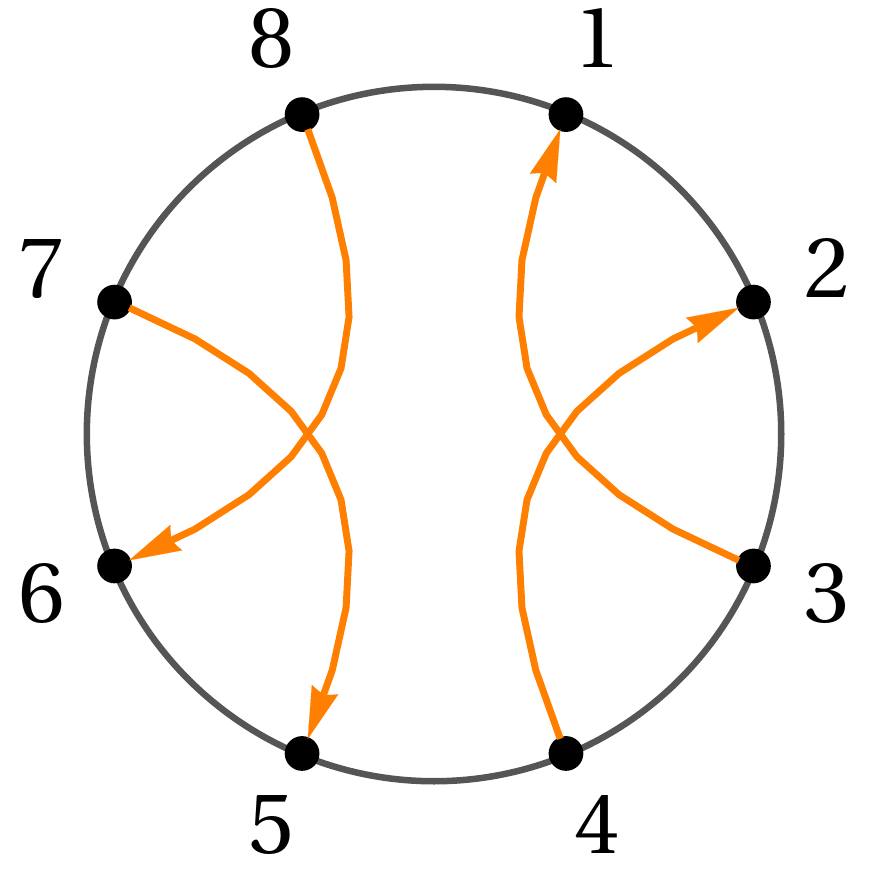}
\end{gathered} \text{ .}
\end{aligned}
\end{equation}
If we again assume physical errors to correspond to even nearest-neighbour Majorana operators, we find a code distance $d=2$ with regards to these errors. Explicitly, two such operators are required for both bit-flip and phase-flip errors, e.g.\ $\m_1\m_2\m_2\m_3 = \m_1\m_3$ for a phase flip and $\m_3\m_4\m_5\m_6$ for a bit flip. We identify these types of errors by the Majorana dimers whose parity differs between the logical code states: In the former case, we identify operators acting on both ends of one of these dimers, while the latter case corresponds to operators on one endpoint of each of them. We can thus systematically evaluate the error correction properties of any Majorana stabilizer code by expressing its logical basis in Majorana dimers.

\section{Majorana dimers and bit threads}

Our model bears close resemblance to the \emph{bit thread} proposal \cite{freedman2017bit}, a model for holographic states that re-derives the Ryu-Takayanagi (RT) formula by postulating that such states are composed by a flow of EPR pairs between boundary regions. In this proposal, the entanglement entropy $S_A$ of a boundary region $A$ is then equivalent to the \emph{maximal flow} of EPR pairs between $A$ and $A^\text{C}$ through the bulk, which is equivalent to the area of the minimal surface $\gamma_A$ in the standard RT prescription.

Clearly, this picture is satisfied by the Majorana dimer description of the HyPeC for compact regions $A$ (for which \eqref{EQ_DIMER_EE} holds). 
While each dimer only carries half the entanglement of an EPR pair, the phenomenological behaviour is identical: $S_A$ is determined by the number of dimers between $A$ and $A^\text{C}$, which is restricted by the minimal cut through the bulk tiling from the endpoints of $A$. The HyPeC leads to a special dimer configuration in which this bound is saturated for any compact region $A$ (up to degenerate cases shown in Fig.\ \ref{FIG_GREEDY2}). It thus defines a \emph{global} bit thread configuration, i.e.,\ one independent of the choice of $A$.
A special property of this configuration is that the dimers/bit threads follow discrete geodesics through the bulk, so that the bulk metric is emergent from the entanglement structure. 
Note that in the asymptotic limit of infinitely many tiles, shown in Fig.\ \ref{FIG_HAPPY_THREADS}, each geodesic can be identified with a pair of dimers, thus forming an effective EPR pair. Curiously, the resulting entanglement entropy resembles a classical fracton models on a $\{ 4,5 \}$ tiling, where the Shannon entropy scales with the number of dual geodesics, each doubling the ground state degeneracy \cite{Yan:2018nco}.

However, global bit thread configurations are generally not sufficient to reproduce the RT formula for disjoint subsystems; in such cases, the bit thread flow must differ according to the choice of the subsystems to reproduce the correct holographic multi-partite entanglement \cite{Cui:2018dyq}. This cannot be fulfilled with fermionic dimers, as the global entanglement structure is fixed for a specific state. However, the HyPeC is only a Majorana dimer model effectively, with its underlying spin degrees of freedom converted to fermionic modes through a Jordan-Wigner transformation. Thus in general, when considering disjoint subsystems, the entanglement entropy cannot be determined by a dimer counting. It is an interesting future question whether the entanglement between disjoint subsystems (or equivalently, for transpositions of boundary regions) leads to a multi-partite entanglement of the HyPeC also resembling the bit thread picture.

\begin{figure}[tb]
\centering
\includegraphics[height=0.19\textheight]{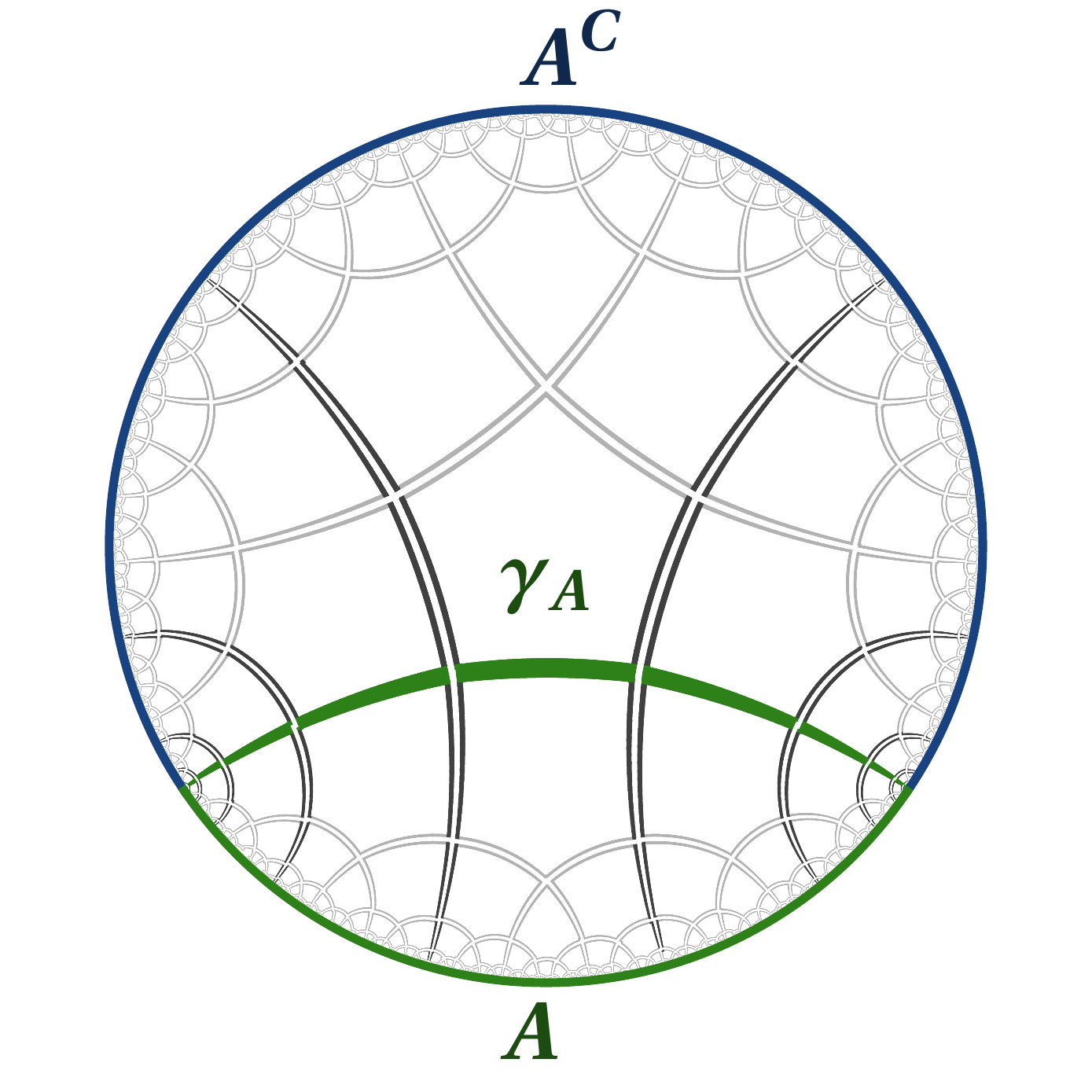}
\caption{Effective bit thread picture in the asymptotically large HyPeC: All dimers are paired along bulk geodesics and any boundary region $A$ has a maximal flow of bit threads (dimer pairs) through the Ryu-Takayanagi surface $\gamma_A$. The threads that pass $\gamma_A$, each contributing $\log 2$ to the entanglement entropy $S_A$, are highlighted. Note that both the number of such threads and $S_A$ are UV-divergent.
}
\label{FIG_HAPPY_THREADS}
\end{figure}

\section{Discussion}
In this work, we have studied the intersection of stabilizer states and fermionic Gaussian states, both efficiently describable classes of quantum states with a wide range of applications in quantum information theory and both condensed matter and high energy physics.
For this purpose, we have introduced a novel graphical formalism for describing Majorana dimer states, free fermionic states characterized by entangled Majorana modes. 
These can describe stabilizer states such as those of the $[[5,1,3]]$ quantum error-correcting code. 
We applied this formalism to the recently constructed hyperbolic pentagon code (HyPeC), a discrete toy model of the AdS/CFT correspondence \cite{Pastawski2015}.
For logical bulk input fixed to code basis states, the HyPeC was found to correspond to Majorana dimers along discrete bulk geodesics. With the bulk geometry thus encoded in boundary state entanglement, we reproduced the logarithmic scaling of the entanglement entropy $S_A$ with its subsystem size for connected subsystems $A$, reproducing the Ryu-Takayanagi formula through a calculation that sharply resembles the recent \emph{bit thread} proposal. We also extended our results to bulk inputs containing local superpositions on each pentagon tile. For this case, where boundary states are generally non-Gaussian, Majorana dimers quantify the dependence of the entanglement entropy on residual bulk regions.
We also provided a method for computing non-Gaussian $n$-point correlations function of the HyPeC for arbitrary bulk input, finding that the Majorana dimer structure -- i.e.,\ boundary correlations only between pairs of Majorana modes -- is preserved for $n=2$, a feature related to the quantum error-correcting properties of the code.
Furthermore, we showed that Majorana dimers can describe a range of entangled states, including GHZ states and models such as the Kitaev chain, while also allowing for complicated non-Gaussian states by expansion in a Majorana dimer basis.
Finally, tensor networks based on Majorana dimers provide a particularly simple model of an RG flow, where an IR$\to$UV transformation is interpreted as an addition of new dimer degrees of freedom upon contraction.


As this work has focused on the specific Majorana dimer structure of the $[[5,1,3]]$ code and the HyPeC that is built upon it, we have only glimpsed at the general relationship between Majorana fermions and stabilizers. While our graphical formalism for Majorana dimers can be used to describe a wide range of entangled quantum states, including generalized stabilizer codes, only a subset of these could be covered here. 
As this formalism allows for the construction of quantum states from their entanglement symmetries, a more systematic study of Majorana dimer states and their symmetries would be useful in the future.
With our approach allowing for a direct analytical contraction of dimer-based tensor networks through simple graphical rules, and a possible description of non-Gaussian states through dimer superpositions, there appears to be a vast number of potential applications. 
Within the Gaussian setting, an interesting question is the deformation of Gaussian stabilizer states. As each Majorana dimer state can be expressed by a \emph{matchgate tensor} \cite{Jahn:2017tls}, one may consider smooth deformation of the HyPeC (and other stabilizer models) while retaining an efficiently contractible tensor network. 
Under such deformations, it is conceivable that a picture with some effective degrees of freedom localized to geodesics is retained.
For example, there exists a possible connection to \emph{ribbon operators} \cite{PhysRevB.94.205123} which appear in the study of topological phases of matter away from fixed point models. This would also involve exploring the similarities between Majorana dimers and anyon models.
One may also wish to address the actual recovery rates of logical qubits in holographic codes, which have been studied both in the original HyPeC proposal \cite{Pastawski2015} and in extensions such as the Calderbank-Shor-Steane (CSS) holographic heptagon code \cite{PhysRevA.98.052301}. 
Their remarkable property of a resilience of logical qubits further in the bulk may be studied more directly with Majorana dimers, where an explicit mapping between bulk and boundary degrees of freedom is provided.
While the toy models studied here are inherently discrete, the many properties of the HyPeC resembling a \emph{conformal field theory} (CFT) motivate further studies on its continuum limits, analogous to continuous MERA \cite{PhysRevLett.115.171602,PhysRevLett.110.100402}. While rigorous studies of the continuum quantum fields corresponding to lattice models are ongoing \cite{Osborne:2019bsq}, the \emph{quasiregular} symmetries expected on the boundary of regular hyperbolic bulk tilings \cite{Boyle:2018uiv} may require a different notion of a CFT for regular discretizations than the familiar continuum formulation. We hope that the present work
stimulates further endeavours in this direction.

\emph{Acknowledgements.} 
We would like to thank Pawel Caputa, Xiao-Liang Qi, Tadashi Takayanagi, Nicholas Tarantino, Han Yan, Michael Walter and Carolin Wille for valuable comments and discussions.
This work has been supported by the ERC (TAQ), the Templeton Foundation, 
the Studienstiftung des Deutschen Volkes, and the DFG (EI 519/14-1, EI 519/15-1, CRC 183 B1, FOR 2724).
This work has also received funding from the European Union's Horizon 2020
research and innovation programme under grant agreement No 817482 (PASQuanS).

\appendix
\begin{widetext}
\newpage

\section{Dimer contraction rules}
\label{APP_CONTR_RULES}

We will now prove the contraction rules from the main text by considering all possible Majorana dimer configurations that can be contracted, showing that they result in either a new Majorana dimer state or vanish. 
As the Majorana dimer diagrams are defined as an effective representation of spins, we define contractions in the Majorana picture to be consistent with the result obtained by reversing the Jordan-Wigner transformation, contracting the corresponding spin degrees of freedoms, and applying a new Jordan-Wigner transformation on the remaining ones.
This is equivalent to always contracting the first two fermionic degrees of freedom under the given ordering, as this requires no re-ordering when projecting onto the $\ket{0,0}$ and $\ket{1,1}$ subspaces in the spin picture during contraction.

Note that any contraction is equivalent to a self-contraction. For example, when contracting two state vectors $\ket\phi$ and $\ket\psi$ over some fermionic degree of freedom, we can simply view this as a self-contraction of the tensor product $\ket\phi \ket\psi$. By using cyclic permutations, we can relate the contraction of any edge index pair $(j,k)$ with $j<k$ to the canonical case $(i,j)=(1,2)$. Equivalently, as we will consider below, we can apply the contraction rules to the last two edges under a given ordering (which avoids re-numbering all the edges).
As we will see throughout this section, the contraction rules rely purely on the dimer parities of dimers connected to contracted edges, so changing the index ordering for states with even total parity does not affect the result, as dimer parities are left invariant. To see that the same logic holds for parity-odd states, consider the following re-ordered versions of the contraction \eqref{EQ_CONTR_EX1} from the main text, where we assume the product state vector $\ket\phi \ket\psi$ to have odd total parity, and again omit dimers that are not connected to the contracted edge:
\begin{align}
\begin{gathered}
\begin{gathered}
\includegraphics[height=0.12\textheight]{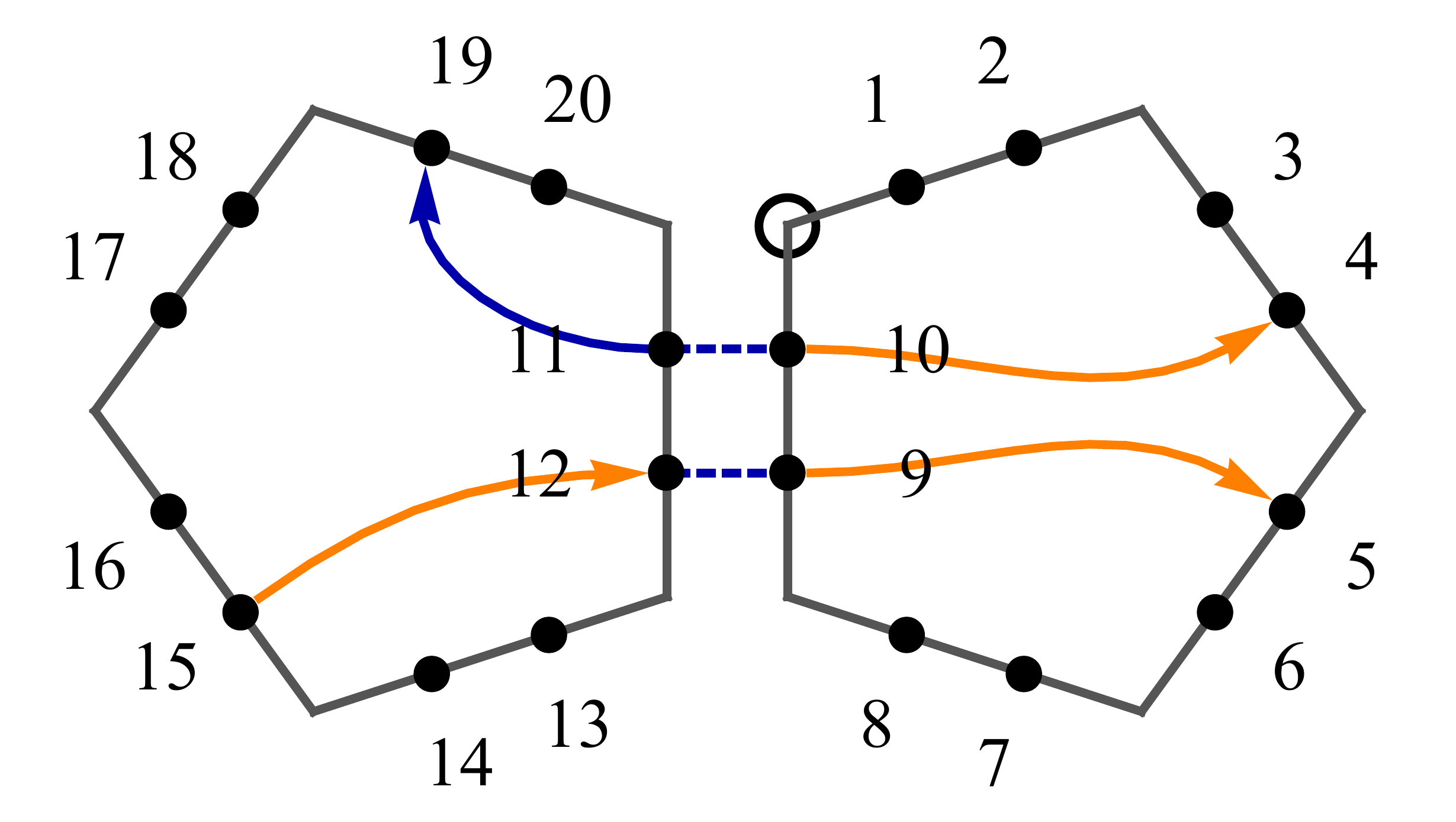}
\end{gathered}
\quad \scalebox{1.2}{$\rightarrow$} \\
\hspace{-23pt}\scalebox{1.2}{$\parallel$} \\
\begin{gathered}
\includegraphics[height=0.12\textheight]{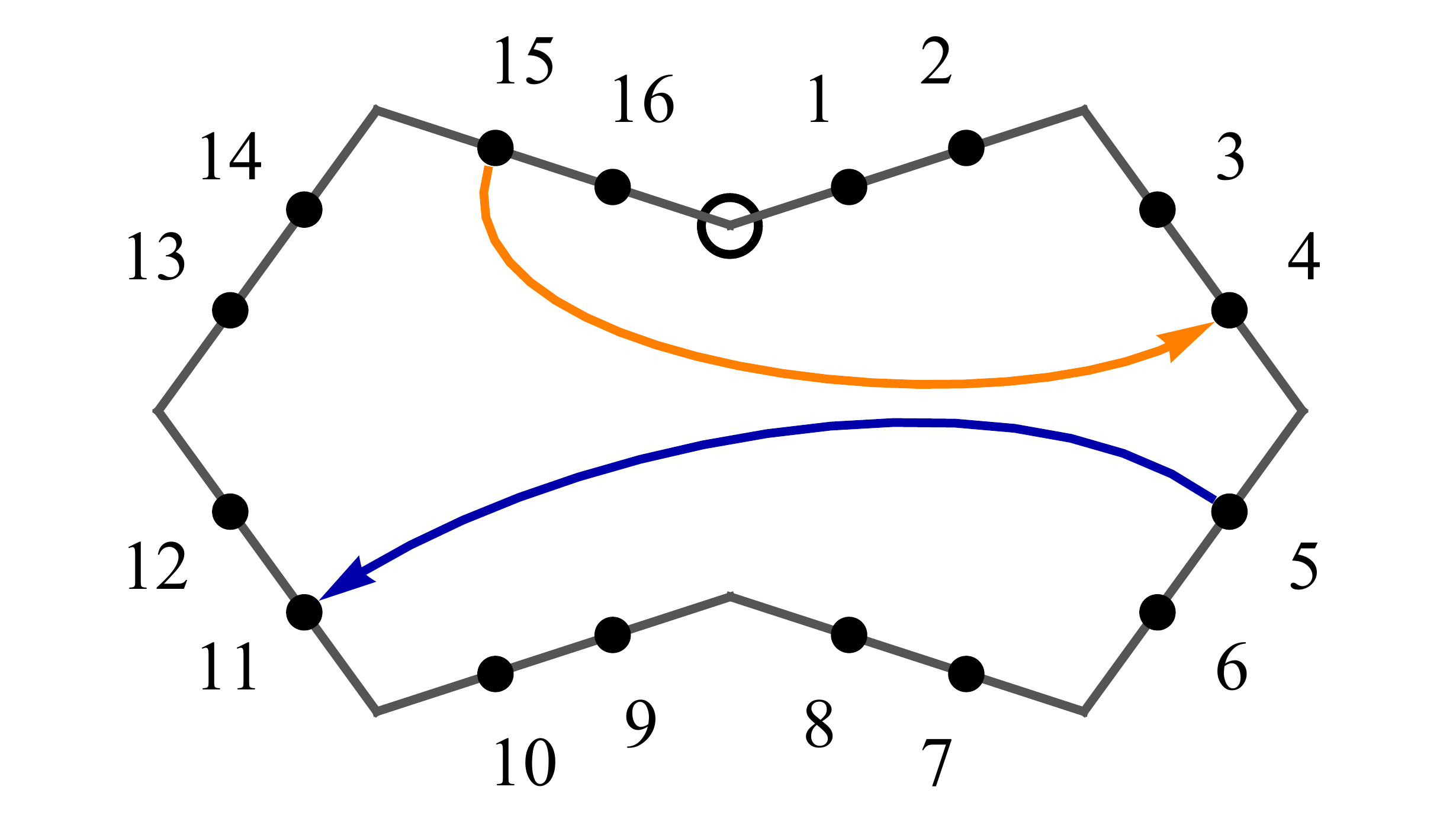}
\end{gathered}
\quad \scalebox{1.2}{$\leftarrow$}
\end{gathered}
\quad
\begin{gathered}
\begin{gathered}
\includegraphics[height=0.12\textheight]{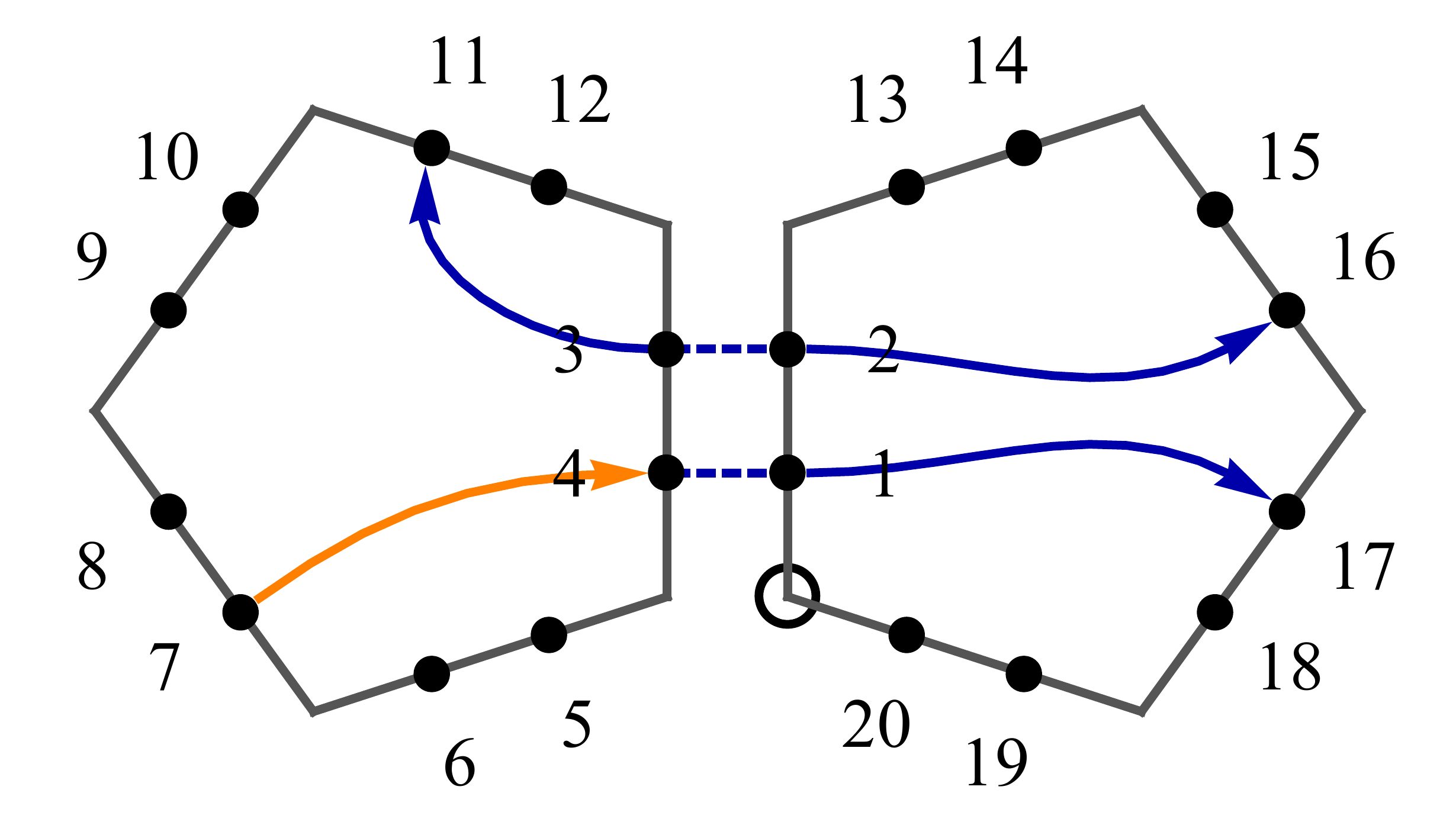}
\end{gathered}
\quad \scalebox{1.2}{$\rightarrow$} \\
\hspace{-23pt}\scalebox{1.2}{$\parallel$} \\
\begin{gathered}
\includegraphics[height=0.12\textheight]{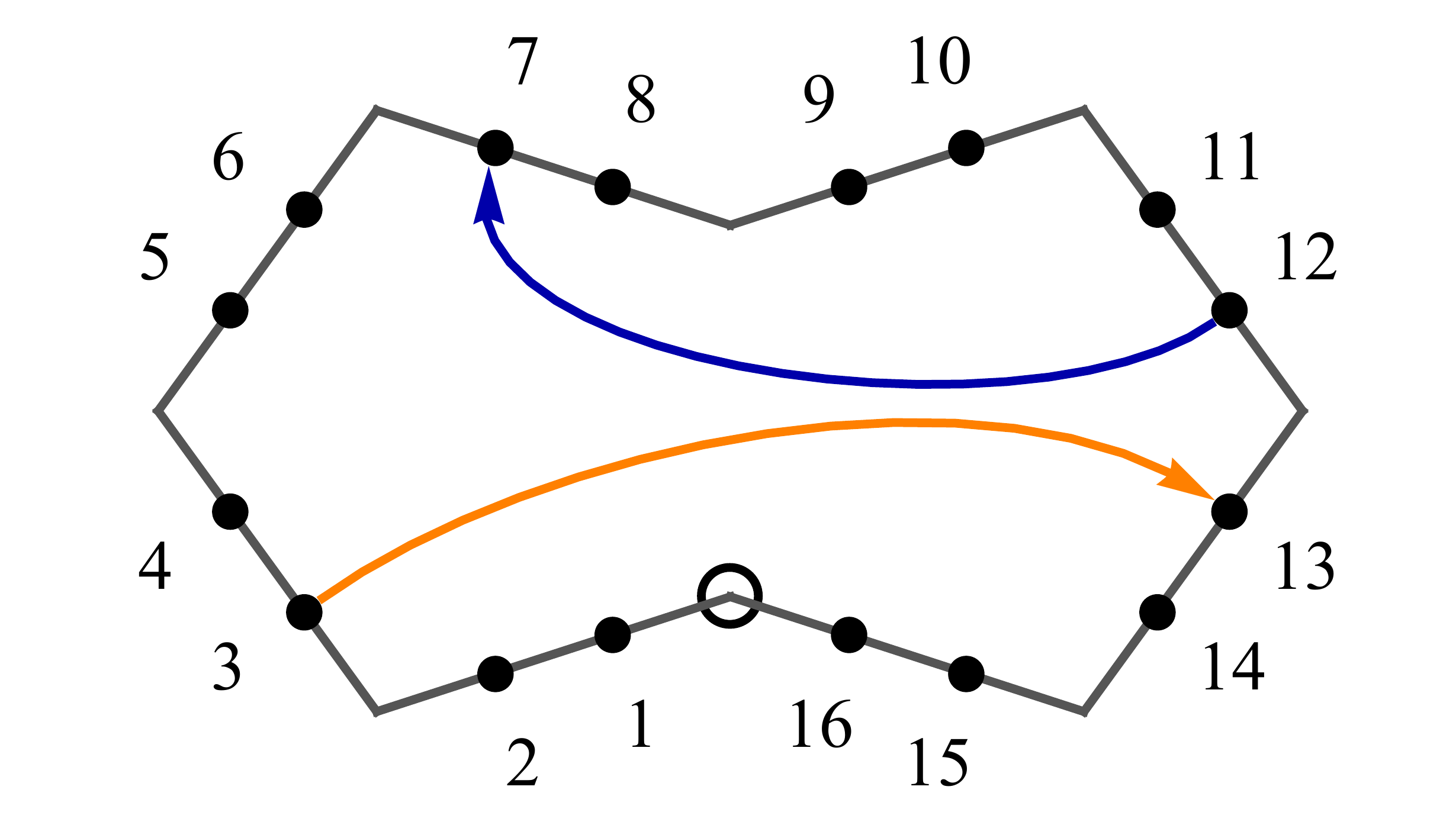}
\end{gathered}
\quad \scalebox{1.2}{$=$}
\end{gathered}
\quad
\begin{gathered}
\includegraphics[height=0.12\textheight]{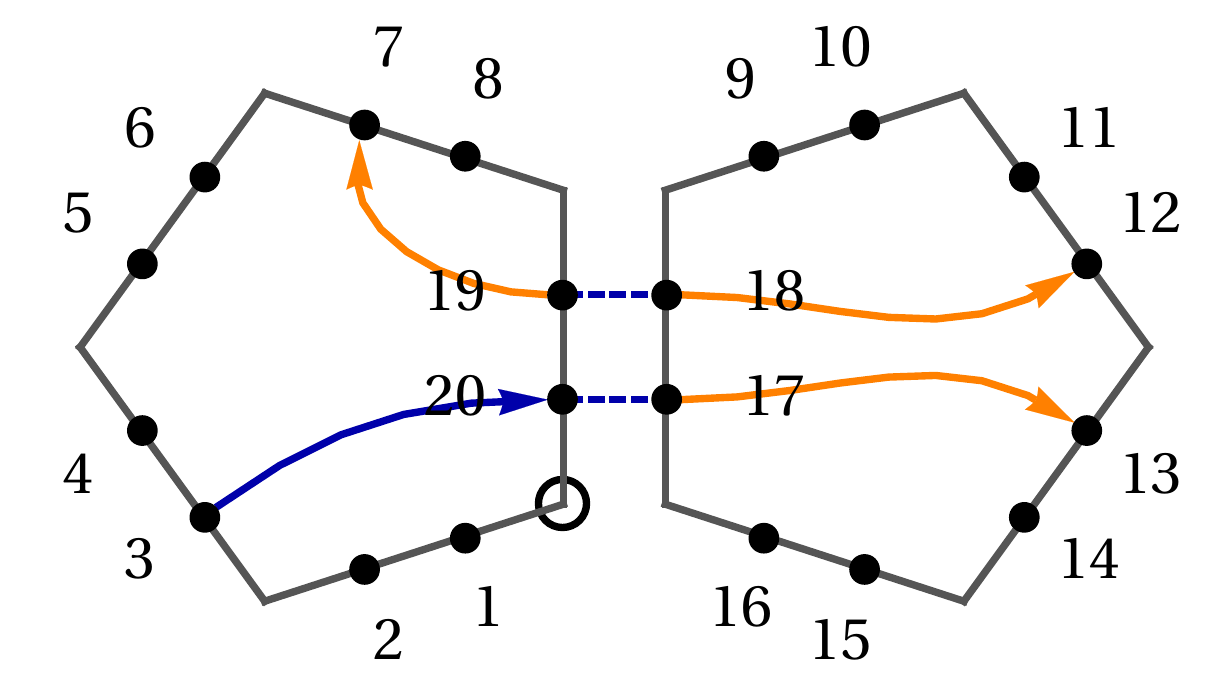}\\
\scalebox{1.2}{$\parallel$}\\
\includegraphics[height=0.12\textheight]{pentagon_net_contr_ex2e.pdf}
\end{gathered}
\end{align}
The two diagrams on the left correspond to the contraction \eqref{EQ_CONTR_EX1} proved in the main text. Cyclic permutations relate this to a self-contraction of the first two fermions (Majorana modes $1$ to $4$, centre) and alternatively of the two last fermions (Majorana modes $17$ to $20$, right). The pivot of the permutation is again represented by a small circle. As we see, applying the dimer contraction rules from the main text leads to equivalent results under cyclic permutations. Note, however, that forming product state vectors $\ket\phi \ket\psi$ requires an ordering of the modes in $\ket\phi$ before the ones in $\ket\psi$, which can still lead to additional parity shifts when contracting in parity-odd states. We resolve these ambiguities in Appendix \ref{APP_CONTR_ORDER}.

Let us now prove the general case of the previous diagram, written as a self-contraction of an arbitrary Majorana dimers state (of which a product state is only a special case). We start with contractions of the form
\begin{equation}
\label{EQ_SELF_CONTR_EX1}
\begin{gathered}
\includegraphics[height=0.15\textheight]{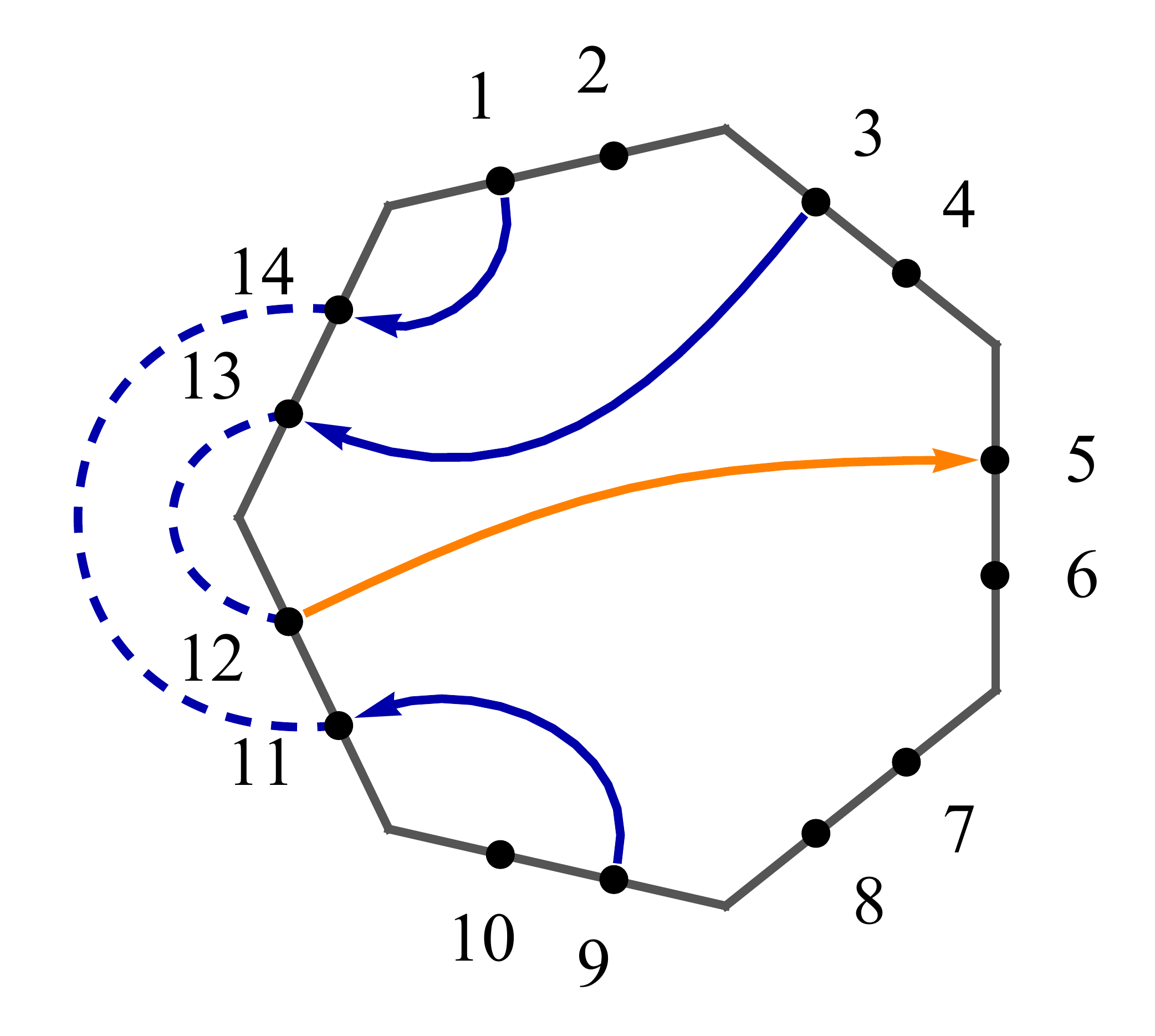}
\end{gathered}
\quad \scalebox{1.5}{$=$} 
\begin{gathered}
\includegraphics[height=0.15\textheight]{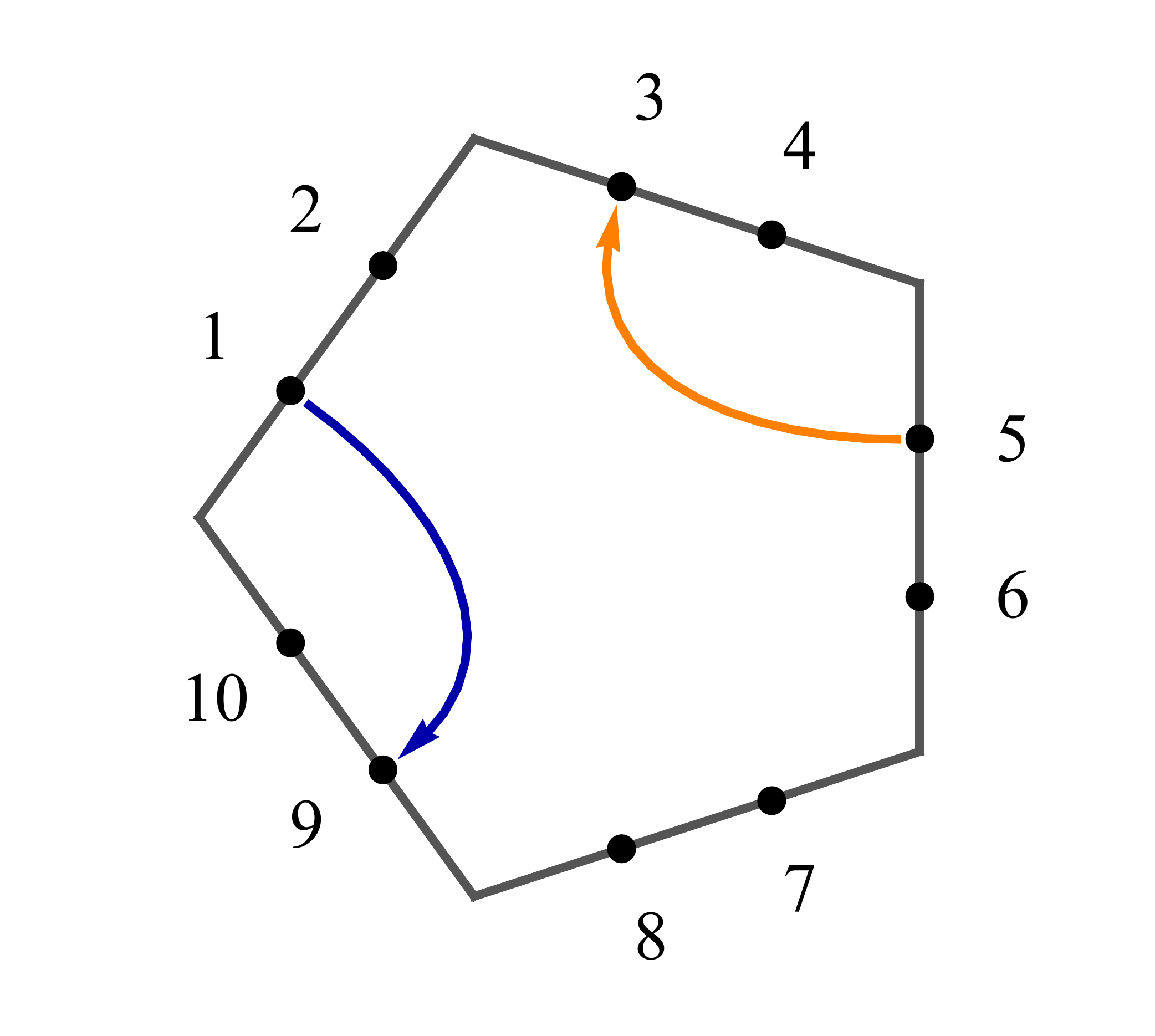}
\end{gathered}
\end{equation}
We start with an $N$-fermion state vector $\ket\chi$ ($N=7$ in \eqref{EQ_SELF_CONTR_EX1}) that obeys the Majorana dimer conditions
\begin{align}
\label{EQ_COND1A}
(\m_a + \i\, p_{a,2N} \m_{2N})\ket\chi &= 0 \text{ ,} & (\m_b + \i\, p_{b,2N-3} \m_{2N-3})\ket\chi &= 0 \text{ ,}  \\
\label{EQ_COND1B}
(\m_c + \i\, p_{c,2N-1} \m_{2N-1})\ket\chi &= 0 \text{ ,} & (\m_d + \i\, p_{d,2N-2} \m_{2N-2})\ket\chi &= 0  \text{ ,}
\end{align}
where we assume for now that $a<b$ and $c<d$, so that the dimer lines do not cross ($a=1$, $b=9$, $c=3$, $d=5$ in \eqref{EQ_SELF_CONTR_EX1}). We claim that after contraction, the contracted state vector $\ket\omega$ is again a Majorana dimer state with conditions
\begin{align}
(\m_a + \i\, p_{a,2N} p_{b,2N-3} \m_b)\ket\omega &= 0 \text{ ,} 
\label{EQ_CONTR_COND1A}\\
(\m_c + \i\, p_{c,2N-1} p_{d,2N-2} \m_d)\ket\omega &= 0 \text{ ,}
\label{EQ_CONTR_COND1B}
\end{align}
which means that the parities along a contracted path are multiplied. We write these conditions as $O_k \ket\omega =0$ with $k \in \lbrace 1,2 \rbrace$ denoting the two cases. Using the tools developed in section \ref{SEC_T_CONTR}, we will now prove them simultaneously:
\begin{align}
O_k \ket\omega &= O_k \int \text{d}\fd_N \text{d}\fd_{N-1}\, e^{\,\fd_{N-1} \fd_N} \ket\chi \nonumber \\
&= \int \text{d}\fd_N \text{d}\fd_{N-1}\, e^{\,\fd_{N-1} \fd_N}
\begin{cases}
\left(\m_a + \i\, p_{a,2N} p_{b,2N-3} \m_b\right) \ket\chi & \text{for $k=1$} \\
\left(\m_c + \i\, p_{c,2N-1} p_{d,2N-2} \m_d\right) \ket\chi & \text{for $k=2$}
\end{cases} \text{ .}
\end{align}
Using \eqref{EQ_COND1A} and \eqref{EQ_COND1B}, we can rewrite this purely in terms of operators acting locally on the contracted edges:
\begin{align}
O_k \ket\omega &= \int \text{d}\fd_N \text{d}\fd_{N-1}\, e^{\,\fd_{N-1} \fd_N}
\begin{cases}
p_{a,2N} (\m_{2N-3} - \i\, \m_{2N}) \ket\chi & \text{for $k=1$} \\
p_{c,2N-1} (\m_{2N-3} - \i\, \m_{2N-1}) \ket\chi & \text{for $k=2$}
\end{cases} \nonumber \\
&\propto \int \text{d}\fd_N \text{d}\fd_{N-1}\, e^{\,\fd_{N-1} \fd_N}
\begin{cases}
\left( \fe_{N-1} + \fd_{N-1} - \fe_N  + \fd_N \right) \ket\chi \\
\left( \fe_{N-1} - \fd_{N-1} + \fe_N  + \fd_N \right) \ket\chi
\end{cases} \nonumber\\
&= \int \text{d}\fd_N \text{d}\fd_{N-1}\,
\begin{cases}
\left(  \fd_{N-1} + \fd_N  -  \fd_N \fd_{N-1} \fe_{N-1} - \fd_{N-1} \fd_N \fe_N \right) \ket\chi \\
\left( - \fd_{N-1} + \fd_N -  \fd_N \fd_{N-1} \fe_{N-1} + \fd_{N-1} \fd_N \fe_N \right) \ket\chi
\end{cases} \nonumber\\
&= \int \text{d}\fd_N \text{d}\fd_{N-1}\,
\begin{cases}
\left( \fe_{N-1} \fd_{N-1}\fd_N + \fd_{N-1} \fe_N \fd_N \right) \ket\chi \\
\left( \fe_{N-1}\fd_{N-1}\fd_N  - \fd_{N-1} \fe_N \fd_N \right) \ket\chi
\end{cases}
= 0 \text{ .}
\end{align}
We previously assumed $a<b$ and $c<d$. What happens if e.g.\ $c>d$? As the condition \eqref{EQ_CONTR_COND1B} for the contracted state vector $\ket\omega$ still hold, we just multiply both sides by $-\i\, p_{c,2N-1} p_{d,2N-2}$, yielding
\begin{equation}
(\m_d - \i\, p_{c,2N-1} p_{d,2N-2} \m_c)\ket\omega = 0 \text{ .}
\end{equation}
In other words, contracting out two Majorana dimers that cross each other flips the parity of the resulting dimer. For our example, the corresponding diagram has the form
\begin{equation}
\begin{gathered}
\includegraphics[height=0.15\textheight]{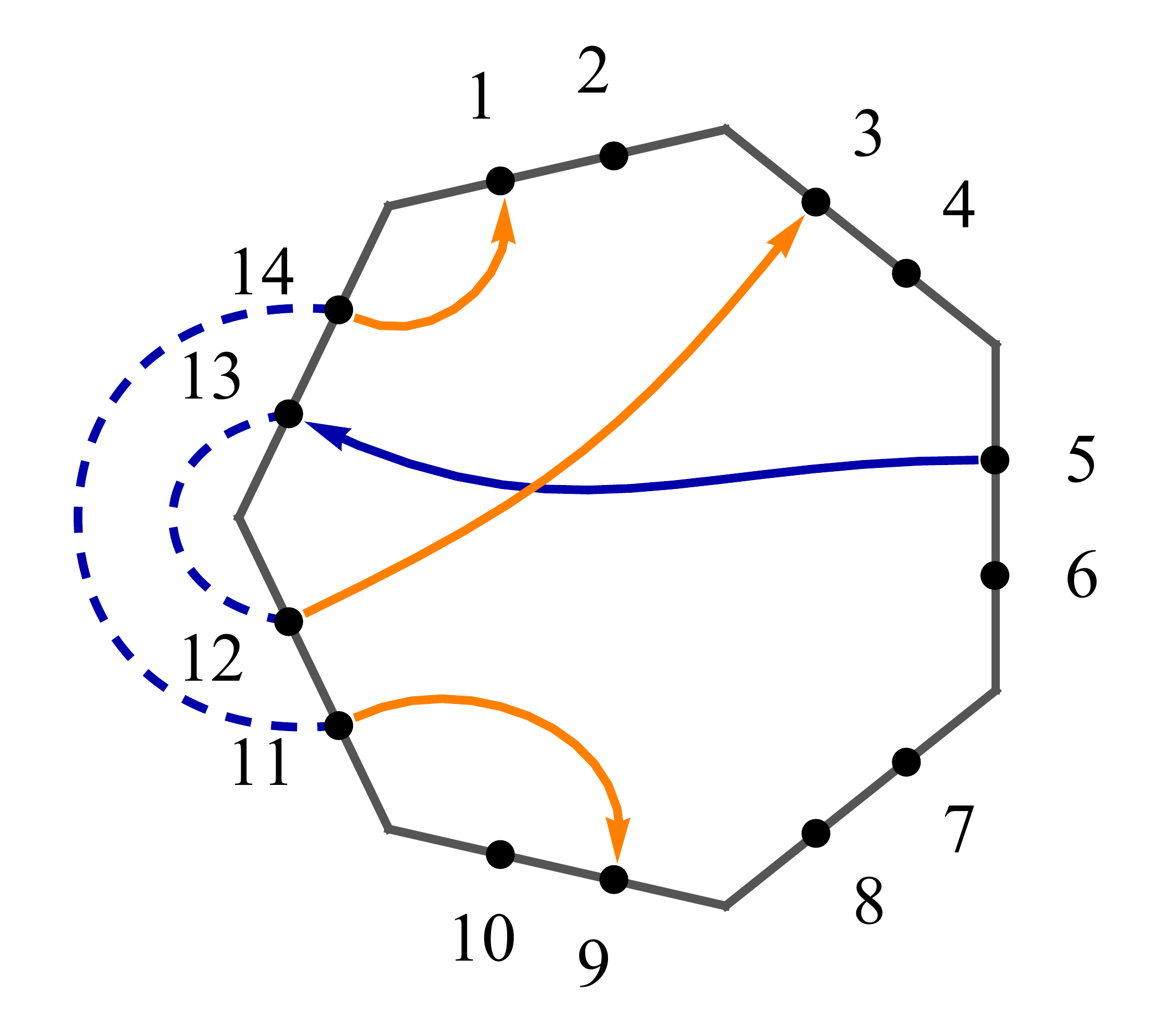}
\end{gathered}
\quad \scalebox{1.5}{$=$} 
\begin{gathered}
\includegraphics[height=0.15\textheight]{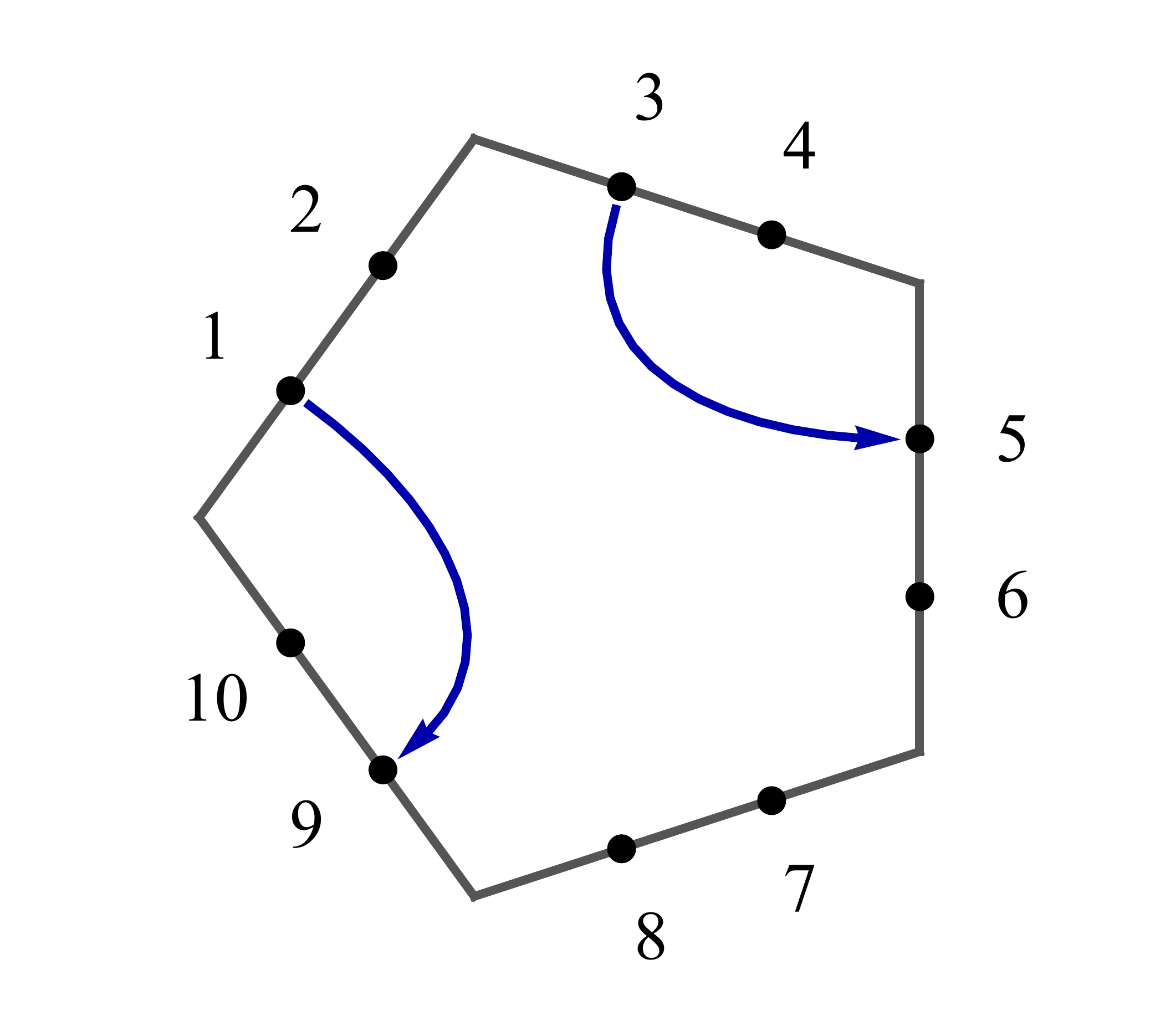}
\end{gathered}
\end{equation}

Self-contractions also allow for special cases involving dimers on the contracted edge itself, which we will now prove, as well. First, consider the case where one of the contracted edges contains a local dimer, such as the contraction
\begin{equation}
\label{EQ_SELF_CONTR_EX2}
\begin{gathered}
\includegraphics[height=0.15\textheight]{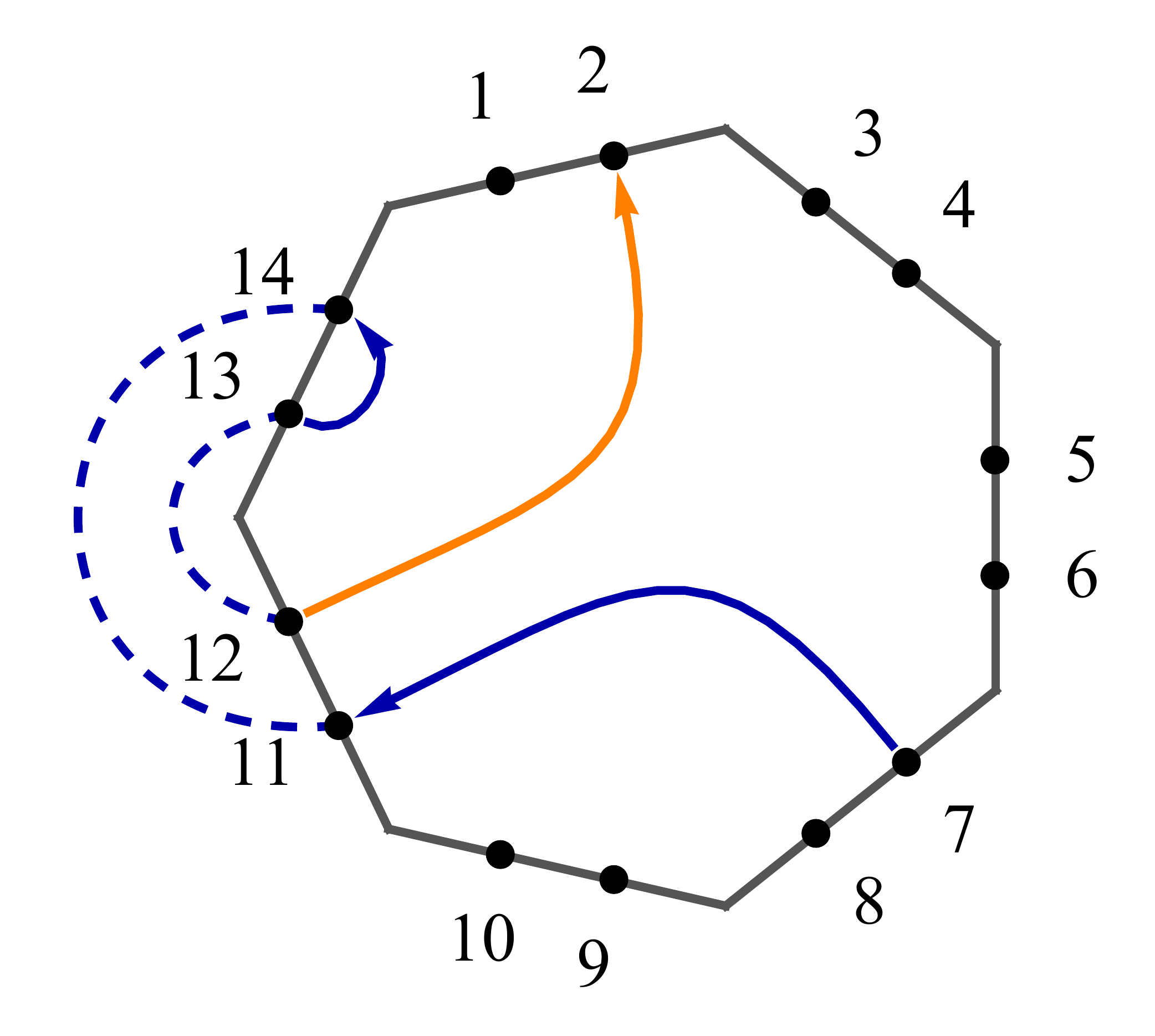}
\end{gathered}
\quad \scalebox{1.5}{$=$} 
\begin{gathered}
\includegraphics[height=0.15\textheight]{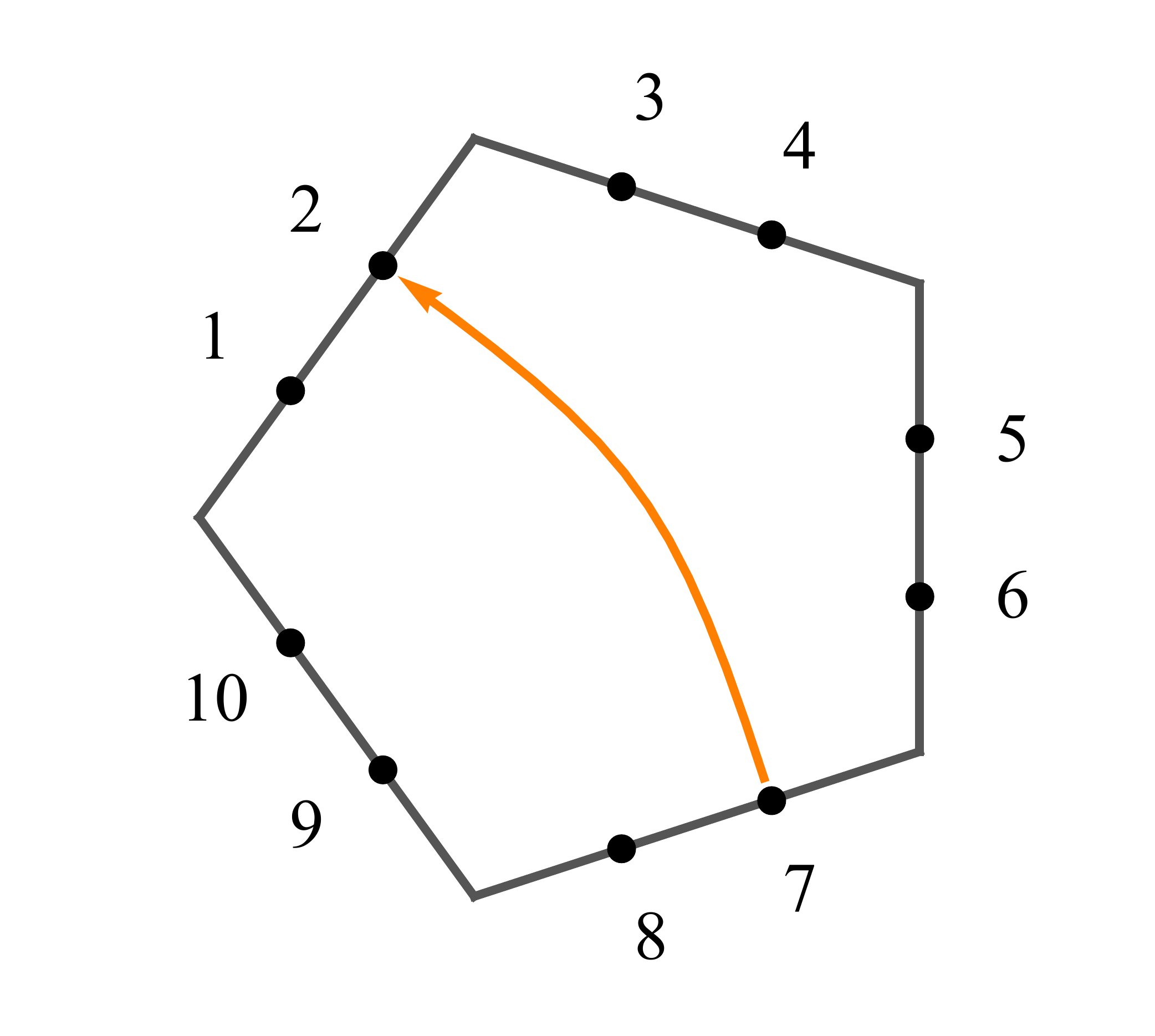}
\end{gathered}
\end{equation}
The contracted path contains contributions from three parities. Without loss of generality, we assume that the local dimer is located on the $N$th edge, so that we start with the conditions
\begin{align}
(\m_{2N-1} + \i\, p_{2N-1,2N} \m_{2N})\ket\chi = &\left( (1-p_{2N-1,2N}) f_N^\dagger + (1+p_{2N-1,2N}) f_N \right) \ket\chi = 0 \text{ ,}\\
(\m_a + \i\, p_{a,2N-2} \m_{2N-2})\ket\chi &= 0 \text{ ,}\hspace{1.5cm}
(\m_b + \i\, p_{b,2N-3} \m_{2N-3})\ket\chi = 0 \text{ .}
\end{align}
In our example \eqref{EQ_SELF_CONTR_EX2}, $a=2$ and $b=7$. The first line simply becomes $\fd_N \ket\chi = 0$ for $p_{2N-1,2N}=-1$ and $\fe_N \ket\chi = 0$ for $p_{2N-1,2N}=+1$. The latter case implies that $\int\text{d}\fd_N \ket\chi= 0$ as well, as Grassmann integrations and annihilators act equivalently. We now prove that these assumptions for the uncontracted $\ket\chi$ imply that
\begin{equation}
\label{EQ_CONTR_COND2}
(\m_a + \i\,  p_{a,2N-2} p_{2N-1,2N} p_{b,2N-3} \, \m_b) \ket\omega = 0
\end{equation}
for the contracted $\ket\omega$, similar to \eqref{EQ_CONTR_COND1A} and \eqref{EQ_CONTR_COND1B}. The proof is similar to the previous setup:
\begin{align}
(\m_a + \i\,  p_{a,2N-2} p_{2N-1,2N} p_{b,2N-3} \, \m_b) \ket\omega 
&= \int \text{d}\fd_N \text{d}\fd_{N-1}\, e^{\,\fd_{N-1} \fd_N} (\m_a + \i\,  p_{a,2N-2} p_{2N-1,2N} p_{b,2N-3} \, \m_b) \ket\chi \nonumber\\
&= p_{a,2N-2} \int \text{d}\fd_N \text{d}\fd_{N-1}\, e^{\,\fd_{N-1} \fd_N} (p_{2N-1,2N}\, \m_{2N-3} - \i\, \m_{2N-2}) \ket\chi \nonumber\\
&= 2 p_{a,2N-2} \int \text{d}\fd_N \text{d}\fd_{N-1}\, e^{\,\fd_{N-1} \fd_N}
\begin{cases}
-\fe_{N-1} \ket\chi & \text{ for $p_{2N-1,2N}=-1$} \\
\fd_{N-1} \ket\chi & \text{ for $p_{2N-1,2N}=+1$} \\
\end{cases} \nonumber\\
&= 2 p_{a,2N-2}
\begin{cases}
-\int \text{d}\fd_N \text{d}\fd_{N-1}\, \fd_{N-1} \fd_N \fe_{N-1} \ket\chi \\
\int \text{d}\fd_N \text{d}\fd_{N-1}\, \fd_{N-1} \ket\chi \\
\end{cases} \nonumber\\
&= 2 p_{a,2N-2}
\begin{cases}
\int \text{d}\fd_N \text{d}\fd_{N-1}\, \fd_{N-1} \fe_{N-1} \fd_N \ket\chi \\
\int \text{d}\fd_{N-1} \fd_{N-1} \int \text{d}\fd_N\, \ket\chi \\
\end{cases} \nonumber\\
&= 0 \text{ .}
\end{align}
Again, crossing the two initial paths so that $a>b$ introduces an additional minus sign to the contracted parity.
The next case to consider contains a Majorana pair across the two contracted edges, such as in the diagram
\begin{equation}
\label{EQ_SELF_CONTR_EX3}
\begin{gathered}
\includegraphics[height=0.15\textheight]{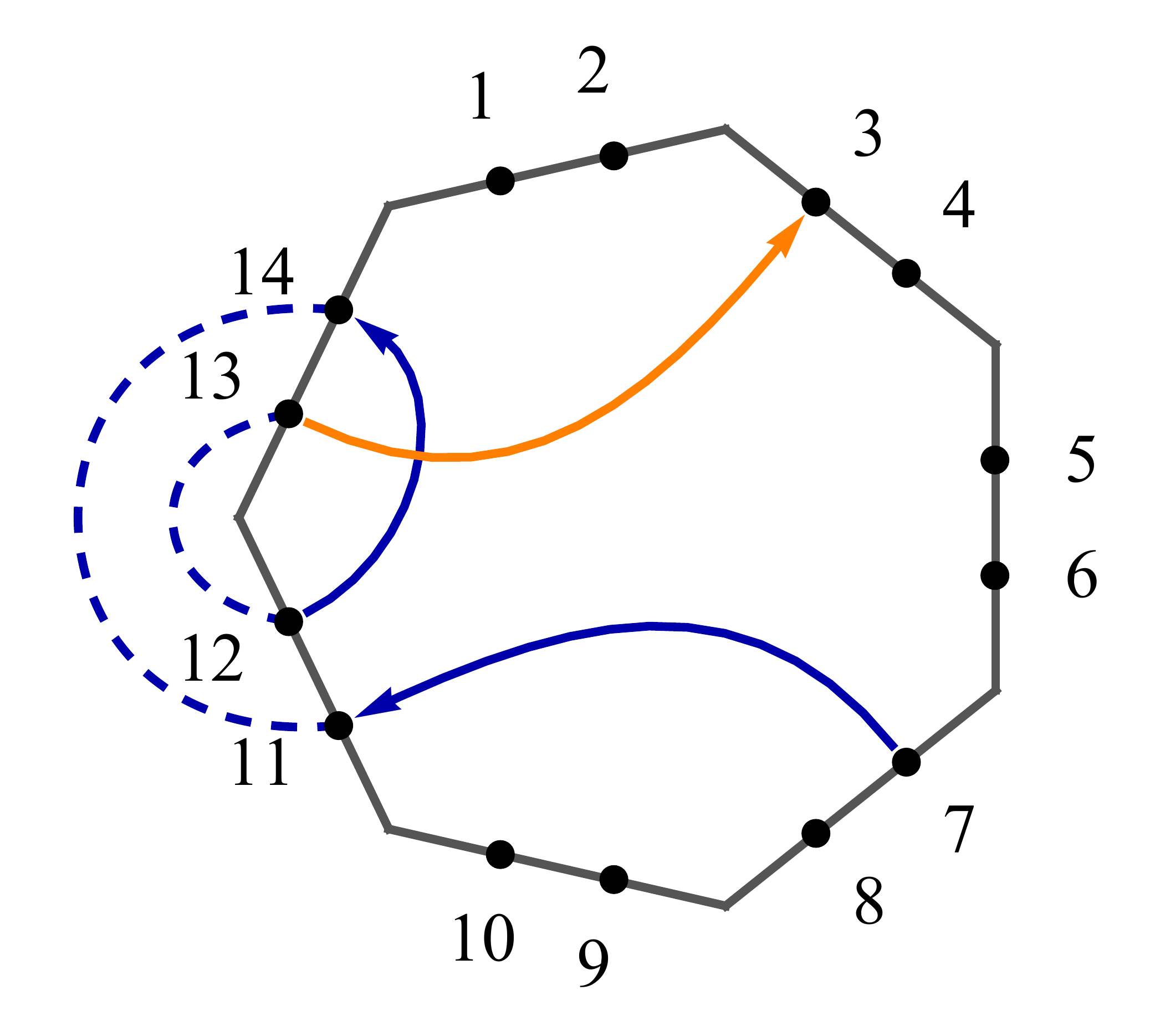}
\end{gathered}
\quad \scalebox{1.5}{$=$} 
\begin{gathered}
\includegraphics[height=0.15\textheight]{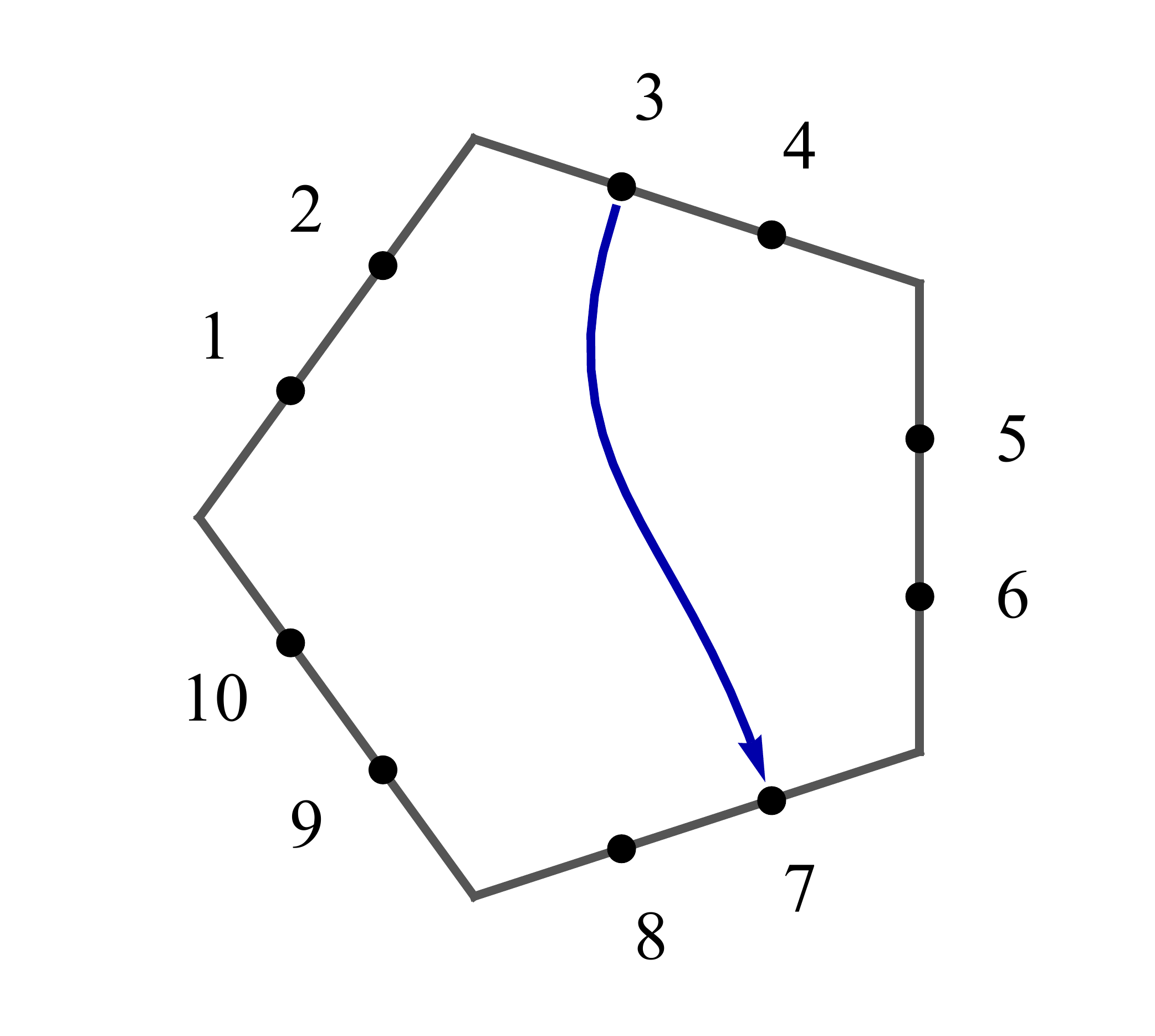}
\end{gathered}
\end{equation}
Note that this kind of contraction always contains a crossing. Again without loss of generality, we assume that the dimer on the contracted edges connects Majorana modes $2N-3$ and $2N$. The full conditions for the uncontracted state are
\begin{align}
(\m_{2N-2} +  \i\, p_{2N-2,2N}\, \m_{2N}) \ket\chi &= \left( \i\, (\fd_{N-1} - \fe_{N-1}) -  p_{2N-2,2N} (\fd_N - \fe_N) \right) \ket\chi = 0 \text{ ,}\\
(\m_a + \i\, p_{a,2N-1} \m_{2N-1})\ket\chi &= 0 \text{ ,}\hspace{2cm}
(\m_b + \i\, p_{b,2N-3} \m_{2N-3})\ket\chi = 0 \text{ ,}
\end{align}
with $a=3$ and $b=7$ in \eqref{EQ_SELF_CONTR_EX3}. The first condition can be rewritten into the form
\begin{equation}
( p_{2N-2,2N}\, \fd_{N-1} +  \i\, \fd_N ) \ket\chi = ( p_{2N-2,2N}\, \fe_{N-1} +  \i\, \fe_N ) \ket\chi \text{ .}
\end{equation}
We now prove the contracted state fulfills 
\begin{equation}
(\m_a - \i\,  p_{a,2N-1} p_{2N-2,2N} p_{b,2N-3} \, \m_b) \ket\omega = 0 \text{ .}
\end{equation} 
Note that additional minus sign in comparison to \eqref{EQ_CONTR_COND2} due to the crossing. The proof is given by
\begin{align}
(\m_a - \i\,  p_{a,2N-1} p_{2N-2,2N} p_{b,2N-3} \, \m_b) \ket\omega 
&= \int \text{d}\fd_N \text{d}\fd_{N-1}\, e^{\,\fd_{N-1} \fd_N} (\m_a - \i\,  p_{a,2N-1} p_{2N-2,2N} p_{b,2N-3} \, \m_b) \ket\chi \nonumber\\
&= -p_{a,2N-1} \int \text{d}\fd_N \text{d}\fd_{N-1}\, e^{\,\fd_{N-1} \fd_N} (p_{2N-2,2N}\, \m_{2N-3} + \i\, \m_{N2-1}) \ket\chi \nonumber\\
&= -2 p_{a,2N-1} \int \text{d}\fd_N \text{d}\fd_{N-1}\, e^{\,\fd_{N-1} \fd_N} (p_{2N-2,2N} \fd_{N-1}+ \i\; \fd_N) \ket\chi \nonumber\\
&= -2 p_{a,2N-1} \int \text{d}\fd_N \text{d}\fd_{N-1}\, (p_{2N-2,2N} \fd_{N-1} + \i\; \fd_N) \ket\chi \nonumber\\
&= -2 p_{a,2N-1} \int \text{d}\fd_N \text{d}\fd_{N-1}\, \left( p_{2N-2,2N} \fe_{N-1} + \i\; \fe_N \right) \ket\chi \nonumber\\
&= 0 \text{ .}
\end{align}

Finally, consider contractions that involve paths that get completely removed by contraction. Up to parities, there are two possible dimer configurations for such contractions:
\begin{align}
\label{EQ_SELF_CONTR_EX3A}
\begin{gathered}
\includegraphics[height=0.15\textheight]{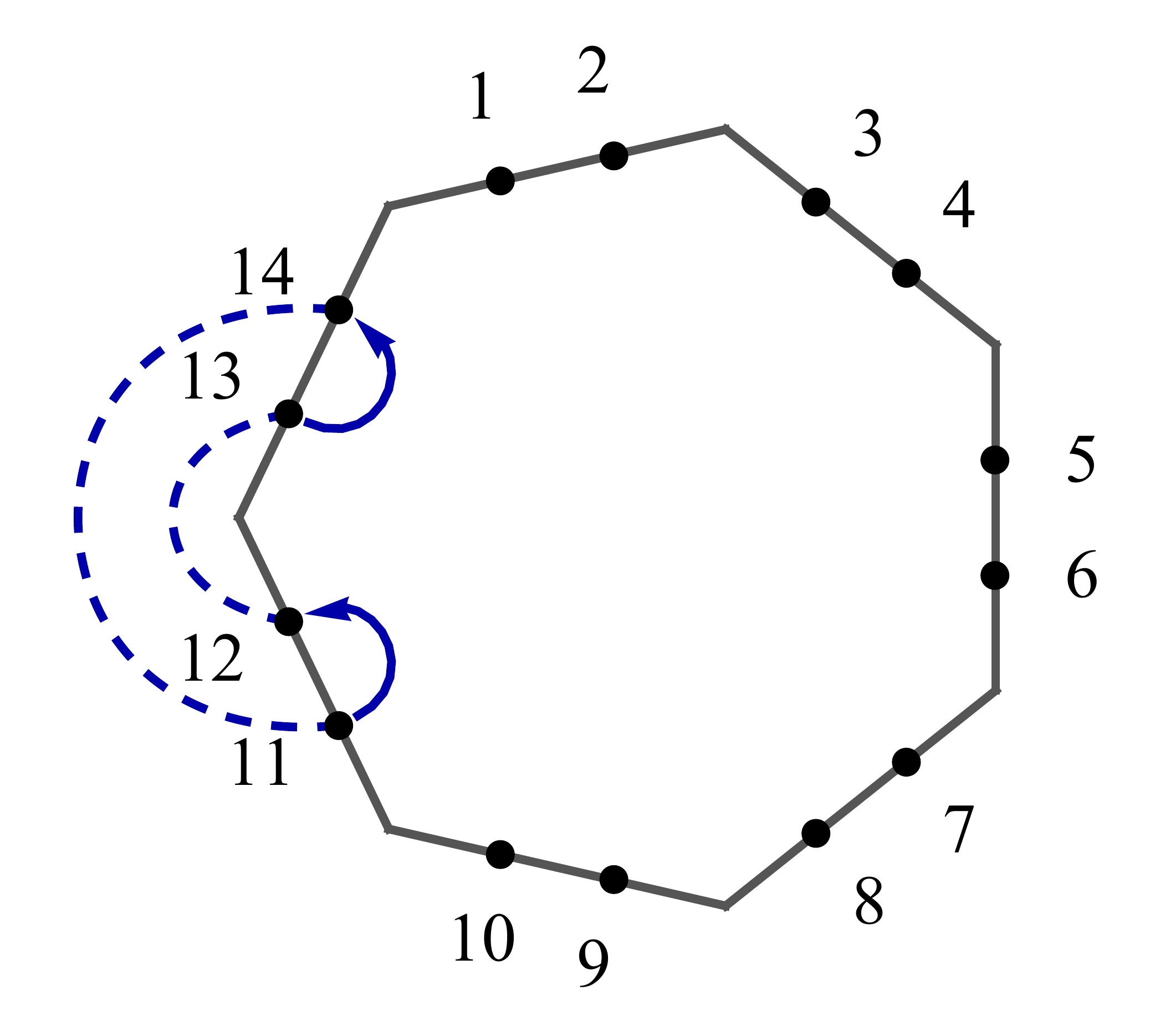}
\end{gathered}
\quad \scalebox{1.5}{$=$} 
\begin{gathered}
\includegraphics[height=0.15\textheight]{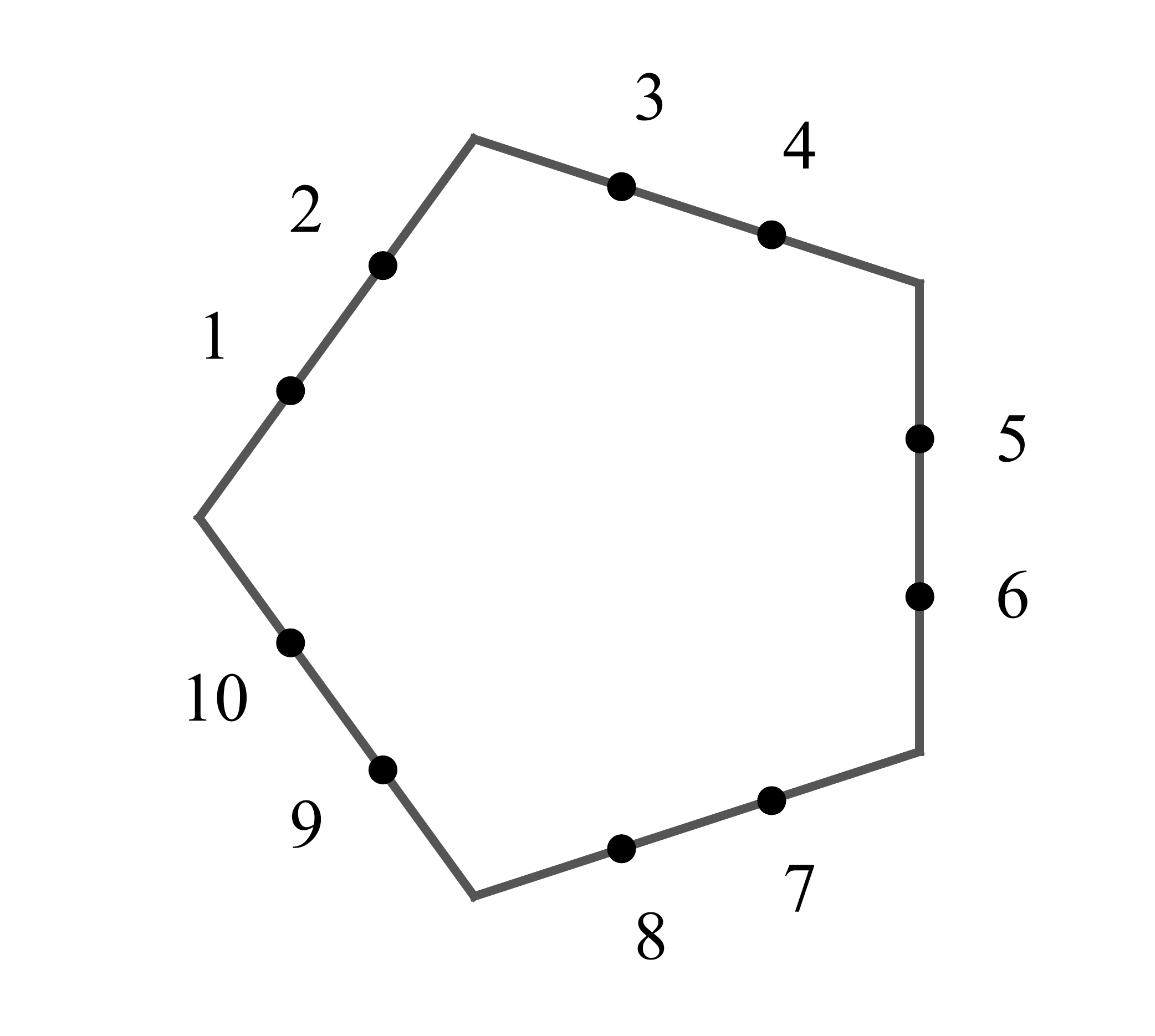}
\end{gathered}\\
\label{EQ_SELF_CONTR_EX3B}
\begin{gathered}
\includegraphics[height=0.15\textheight]{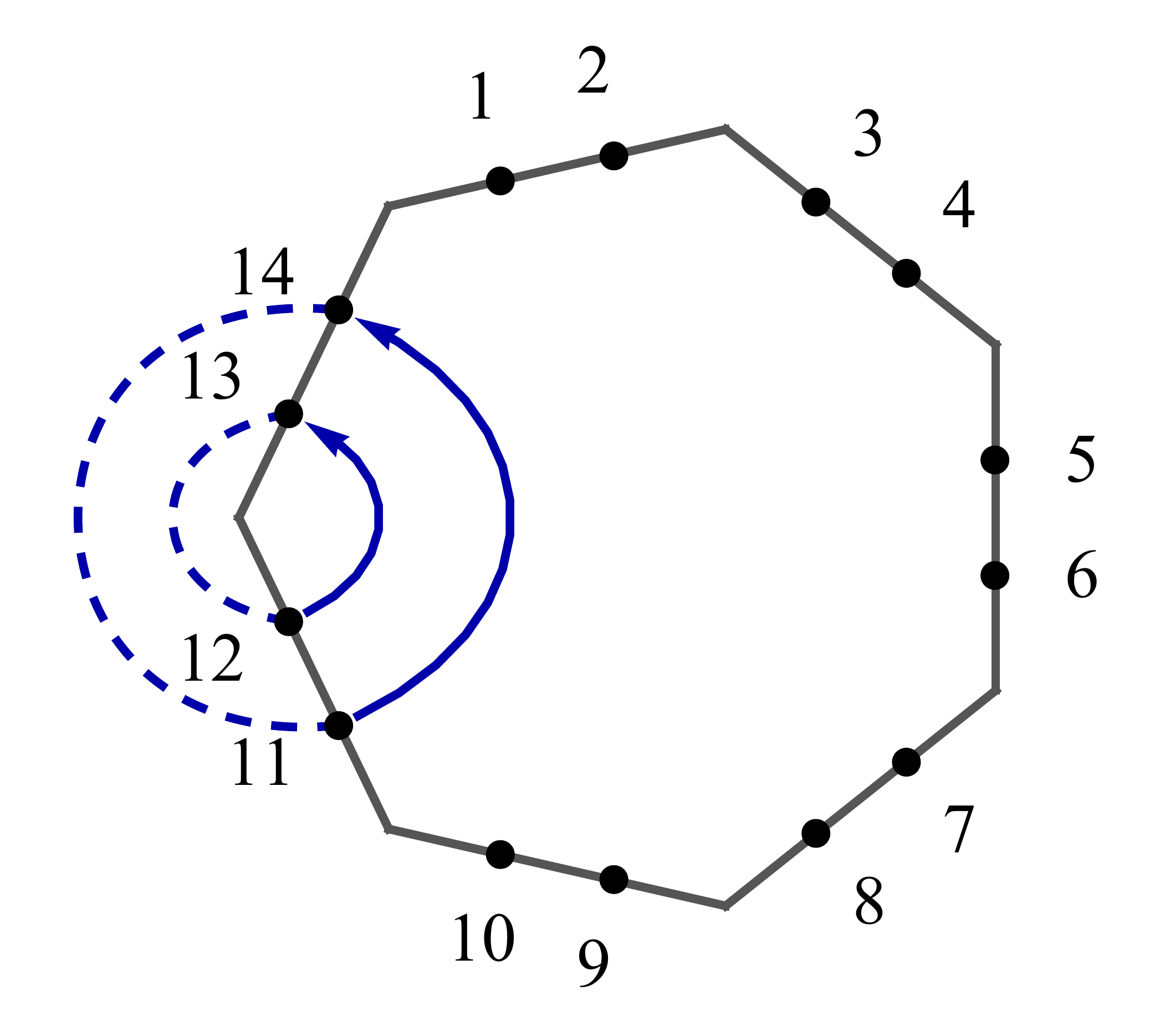}
\end{gathered}
\quad \scalebox{1.5}{$=$} 
\begin{gathered}
\includegraphics[height=0.15\textheight]{heptagon_self_contr_5c.pdf}
\end{gathered}
\end{align}
Clearly, such contractions can only affect the state on the remaining edges by an overall constant $C$. Unless $C=0$, this constant can be absorbed into an appropriate normalization. But when does $C=0$ occur? Let us consider the first diagram \eqref{EQ_SELF_CONTR_EX3A}, which can be generalized to the conditions 
\begin{align}
(\m_{2N-3} + \i\, p_{2N-3,2N-2}\, \m_{2N-2}) \ket\chi &= 0 \text{ , }\\
(\m_{2N-1} + \i\, p_{2N-1,2N}\, \m_{2N}) \ket\chi &= 0 \text{ .}
\end{align}
We claim that the contracted state vector $\ket\omega$ vanishes if $(p_{2N-3,2N-2}, p_{2N-1,2N}) \in \lbrace (1,-1), (-1,1) \rbrace$. These two cases correspond to either $f_{N-1} \ket\chi = 0$ and $f_N^\dagger \ket\chi = 0$ or $f_{N-1}^\dagger \ket\chi = 0$ and $f_N \ket\chi = 0$. It is easy to see that the contraction
\begin{align}
\label{EQ_SELF_CONTR_RESULT}
\ket\omega = \int \text{d}\fd_N \text{d}\fd_{N-1}\, e^{\,\fd_{N-1} \fd_N} \ket\chi = \int \text{d}\fd_N \text{d}\fd_{N-1}\, \ket\chi + \int \text{d}\fd_N \text{d}\fd_{N-1}\, \fd_{N-1} \fd_N \ket\chi
\end{align}
is annihilated in either case (recall that integrals $\int \text{d}f_{k}^\dagger$ act like annihilation operators $f_k$). The second diagram \eqref{EQ_SELF_CONTR_EX3A}, corresponding to the conditions 
\begin{align}
(\m_{2N-3} + \i\, p_{2N-3,2N}\, \m_{2N}) \ket\chi &= 0 \text{ , }\\
(\m_{2N-2} + \i\, p_{2N-2,2N-1}\, \m_{2N-1}) \ket\chi &= 0 \text{ ,}
\end{align}
is more involved. We seek to prove that $\ket\omega$ vanishes if $(p_{2N-3,2N}, p_{2N-2,2N-1}) \in \lbrace (1,-1), (-1,1), (-1,-1) \rbrace$, i.e.,\ for at least one odd parity. In terms of creation and annihilation operators, these three cases can be rewritten as
\begin{align}
\fe_{N-1} \ket\chi = \begin{cases} 
\fe_N \ket\chi & \text{ for }(p_{2N-3,2N}, p_{2N-2,2N-1})=(1,-1) \\ 
\fe_N \ket\chi & \text{ for }(p_{2N-3,2N}, p_{2N-2,2N-1})=(-1,1) \\
-\fd_N \ket\chi & \text{ for }(p_{2N-3,2N}, p_{2N-2,2N-1})=(-1,-1) \\
\end{cases} \text{ ,} \\
\fd_{N-1} \ket\chi = \begin{cases} 
\fd_N \ket\chi & \text{ for }(p_{2N-3,2N}, p_{2N-2,2N-1})=(1,-1) \\ 
-\fd_N \ket\chi & \text{ for }(p_{2N-3,2N}, p_{2N-2,2N-1})=(-1,1) \\
\fe_N \ket\chi & \text{ for }(p_{2N-3,2N}, p_{2N-2,2N-1})=(-1,-1) \\
\end{cases} \text{ .}
\end{align}
For the first two cases, the contraction \eqref{EQ_SELF_CONTR_RESULT} turns into
\begin{align}
\ket\omega = \int \text{d}\fd_N \text{d}\fd_{N}\, \ket\chi \pm \int \text{d}\fd_N \text{d}\fd_{N-1}\, \fd_{N-1} \fd_{N-1} \ket\chi = 0 \text{ .}
\end{align}
For the third case, we get
\begin{align}
\ket\omega = \int \text{d}\fd_N \text{d}\fd_{N-1}\, (1 - \fd_{N} \fe_{N}) \ket\chi 
= \int \text{d}\fd_N \text{d}\fd_{N-1} \fe_N\,  \fd_{N} \ket\chi = 0 \text{ .}
\end{align}
To summarize, we see that a self-contracted loop leads to vanishing contracted state if the total parity of the loop is odd, as postulated in the main text. Even-parity loops contribute an overall constant $C \neq 0$ to the contracted state.

\newpage
\section{Graphical computation of entanglement entropies}
\label{APP_EE_RULES}

In this section, we derive formula \eqref{EQ_DIMER_EE} for the entanglement entropy $S_A$ of a Majorana dimer state using diagrammatic tools and extend it to the computation of R\'enyi entropies $S_A^{(n)}$. Furthermore, we generalize these proofs to full the HyPeC with arbitrary bulk input, recovering previously known conditions on  the boundary regions $A$ \cite{Pastawski2015}.
Following \eqref{EQ_BRA_AND_KET}, we can visualize a density matrix $\rho=\ket\psi\bra\psi$ of some Majorana dimer state vector $\ket\psi$ as
\begin{equation}
\rho \;=\quad
\begin{gathered}
\includegraphics[height=0.1\textheight]{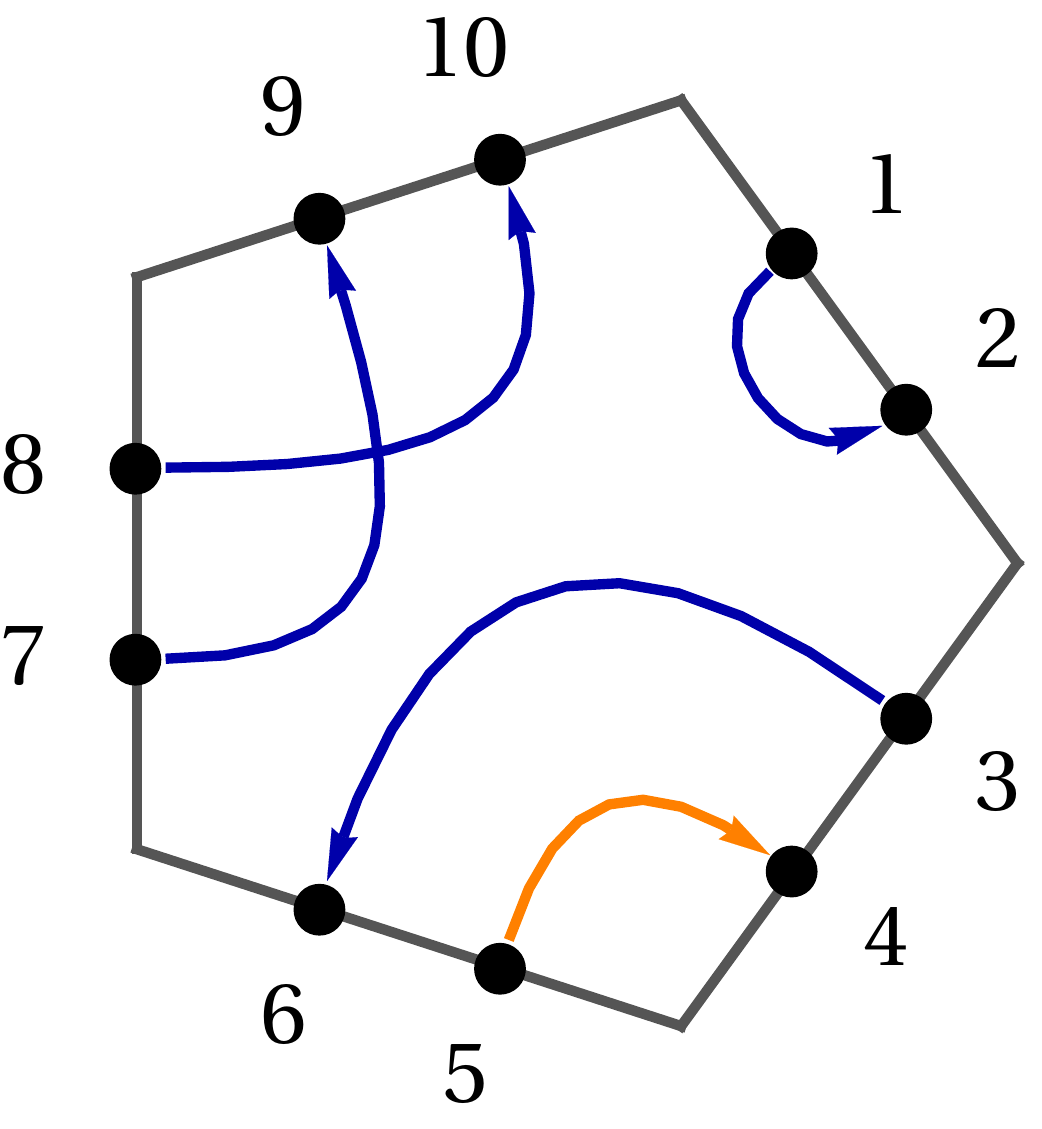}
\end{gathered}\;
\left(
\begin{gathered}
\includegraphics[height=0.1\textheight]{pentagon_net_3a.pdf}
\end{gathered}
\right)^\dagger
\quad=\quad
\begin{gathered}
\includegraphics[height=0.1\textheight]{pentagon_net_3a.pdf}
\end{gathered}\quad
\begin{gathered}
\includegraphics[height=0.1\textheight]{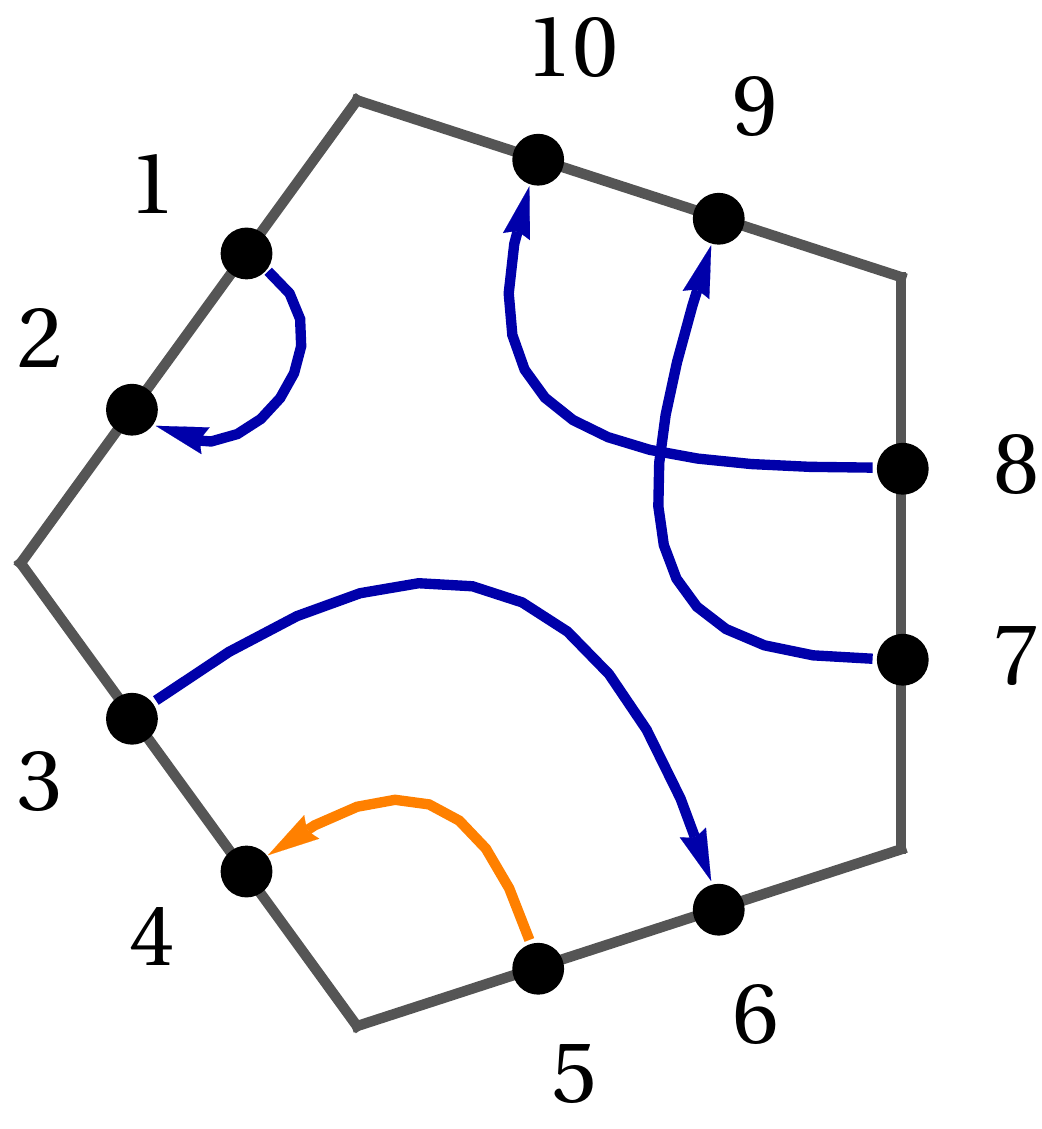}
\end{gathered}
\end{equation}
Here we are effectively using a Choi-Jamiolkowski isomorphism, representing a density matrix as a state in a doubled Hilbert space. 
In order to produce a reduced density matrix $\rho_A$ of some subsystem $A$, we sum over a complete set of states projected onto the edges that are part of $A^{\text{C}}$ (the complement of $A$), which we saw in \eqref{EQ_EPR2} to be equivalent to a contraction. For instance, the green-shaded subsystem $A$ in the following example leads to a reduced density matrix of the form
\begin{align*}
\label{EQ_RHO_A}
\rho_{\textcolor{darkgreen}{A}} \;=
\begin{gathered}
\includegraphics[height=0.1\textheight]{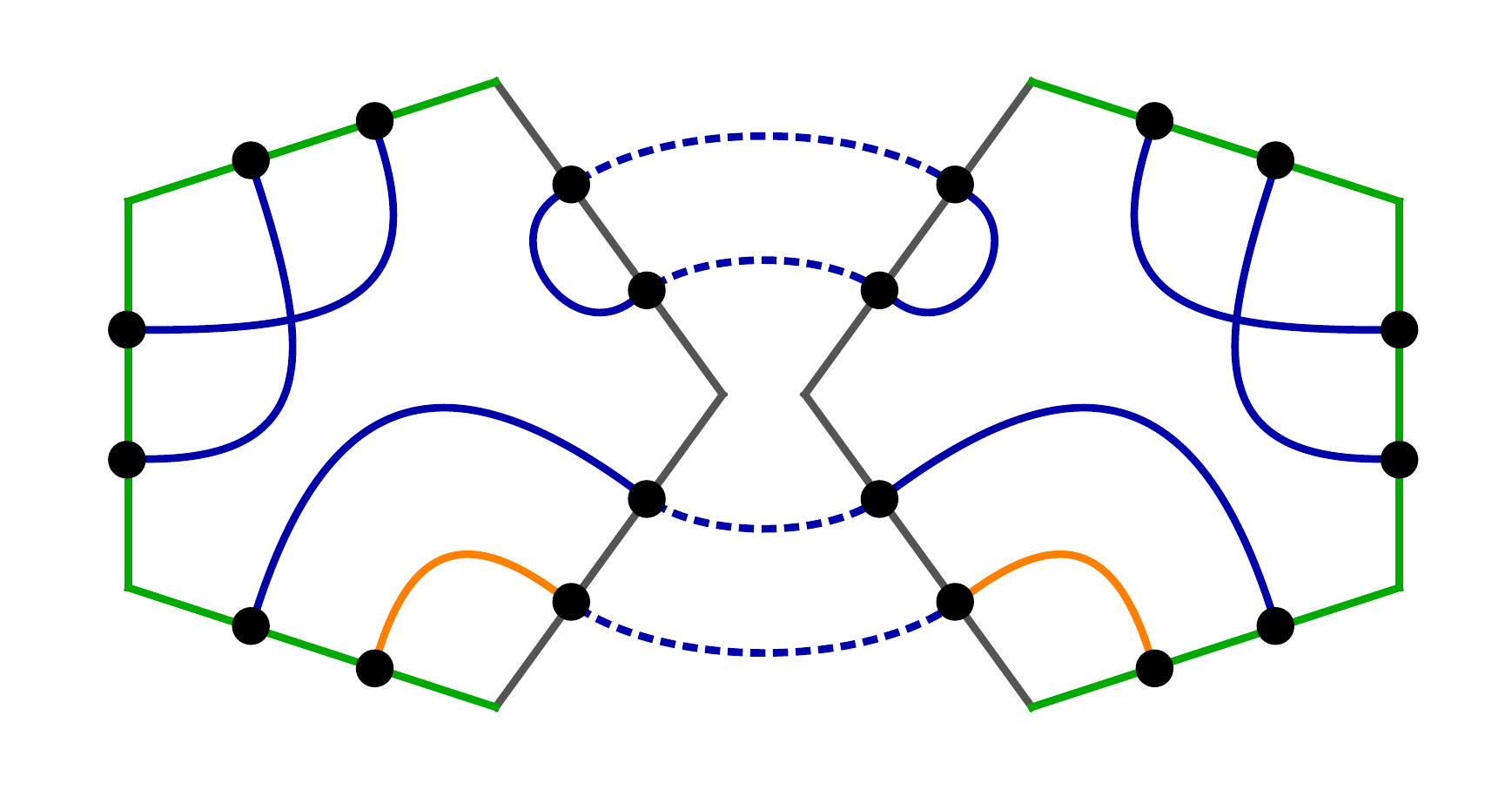}
\end{gathered}
\;=\;&
\begin{gathered}
\includegraphics[height=0.1\textheight]{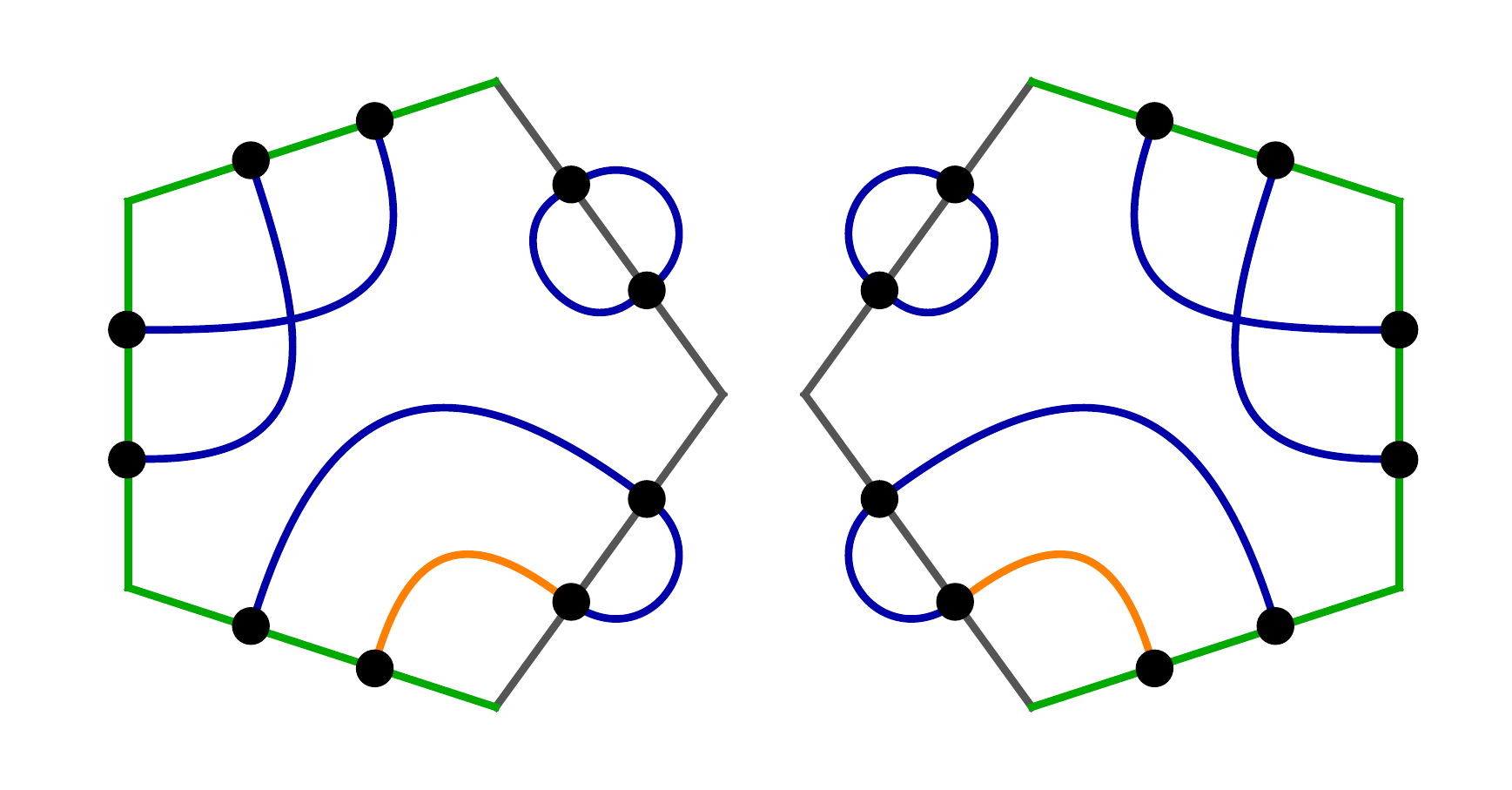}
\end{gathered}
+
\begin{gathered}
\includegraphics[height=0.1\textheight]{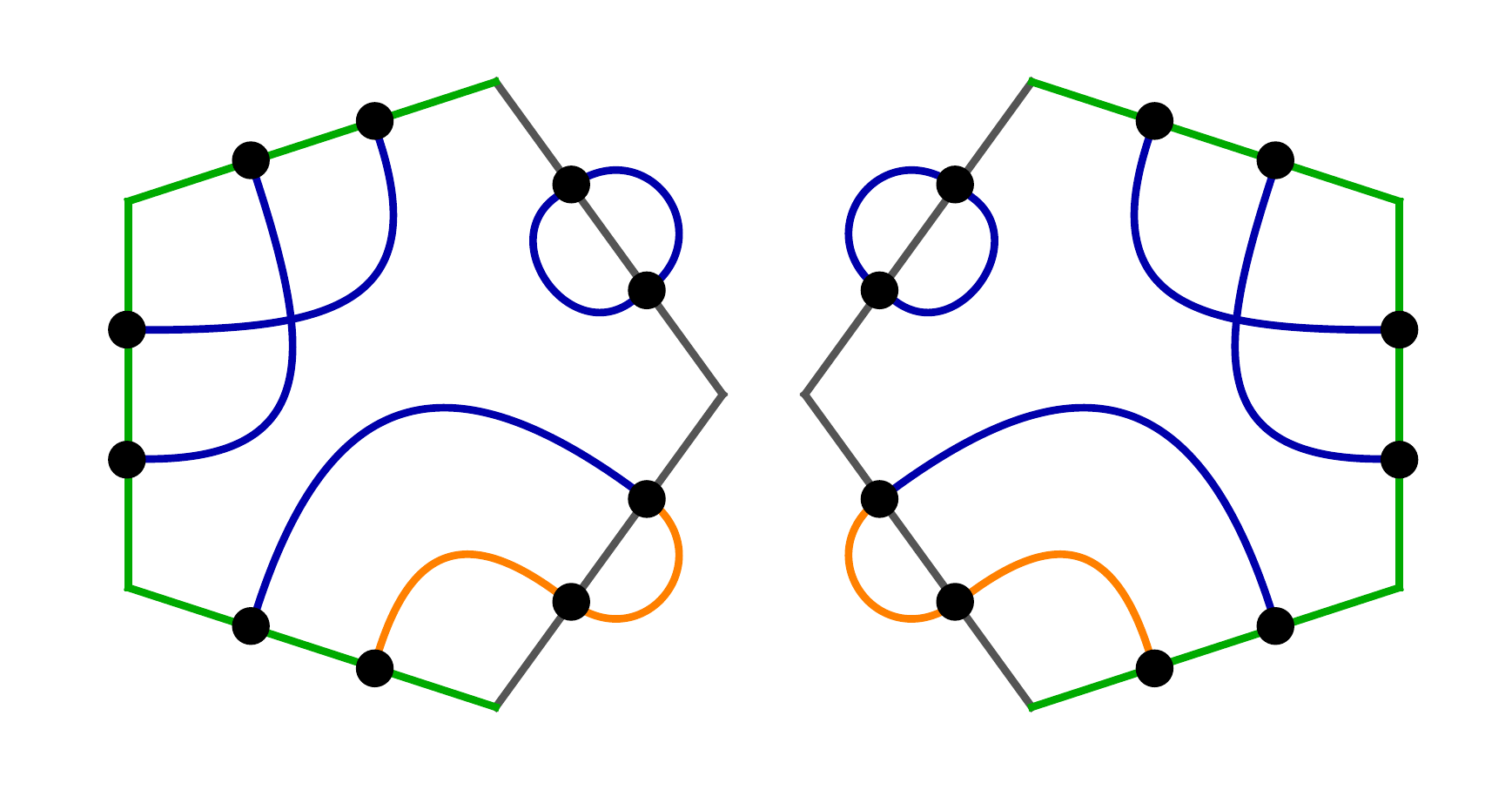}
\end{gathered} \\
\;+&
\begin{gathered}
\includegraphics[height=0.1\textheight]{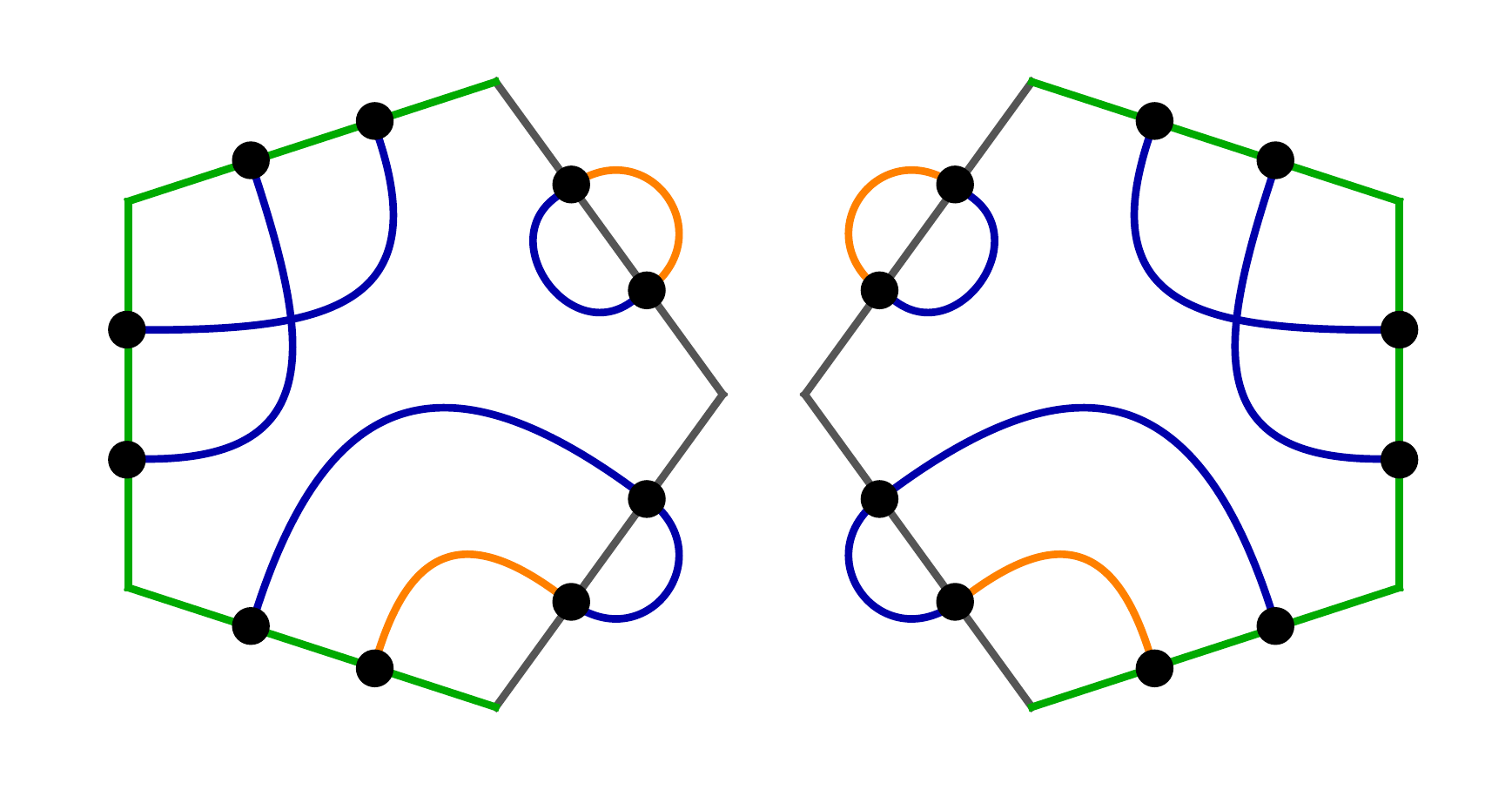}
\end{gathered}
+
\begin{gathered}
\includegraphics[height=0.1\textheight]{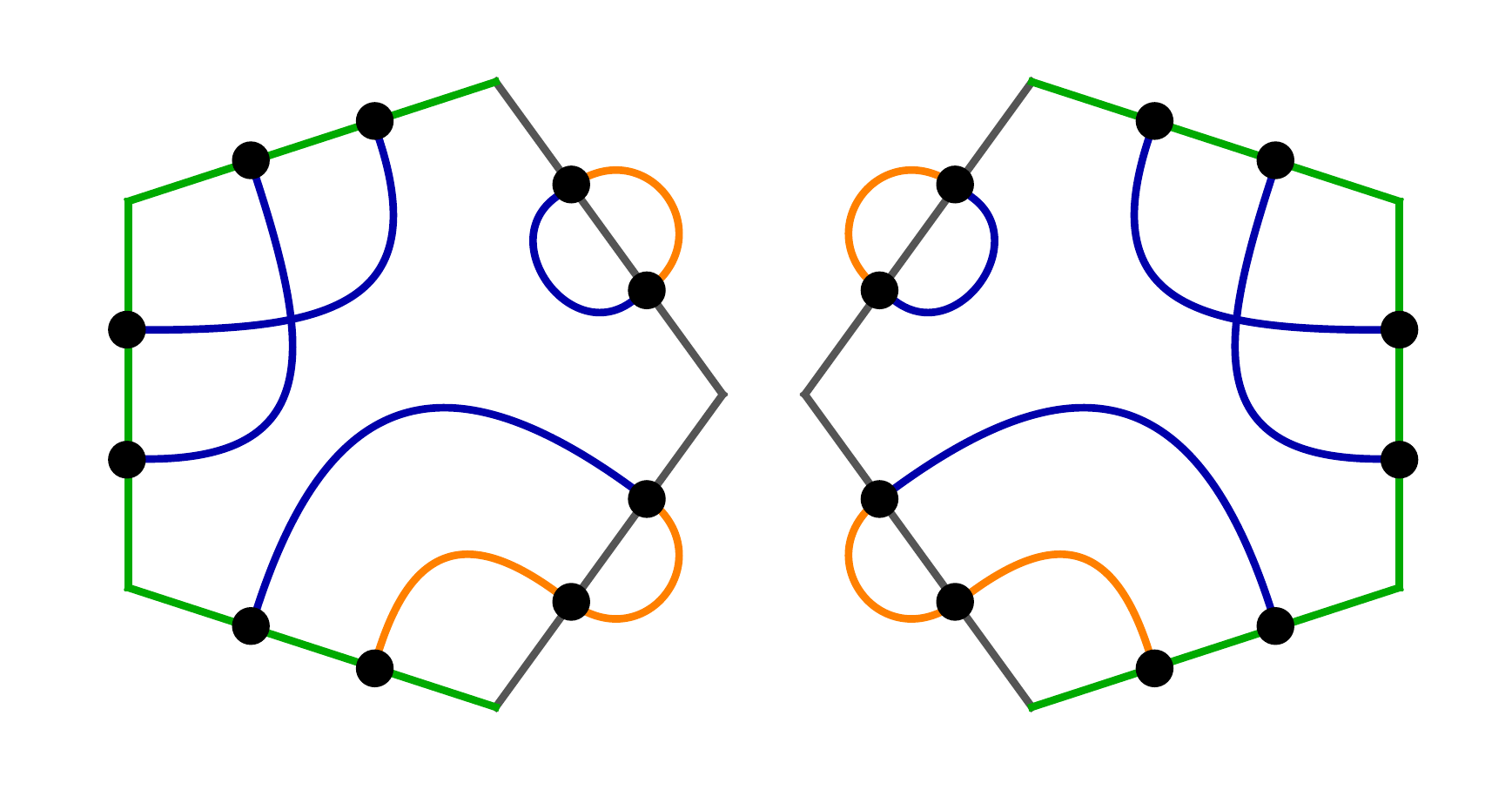}
\end{gathered} 
\end{align*}
\begin{align}
\;=\;
\frac{1}{2}
\begin{gathered}
\includegraphics[height=0.1\textheight]{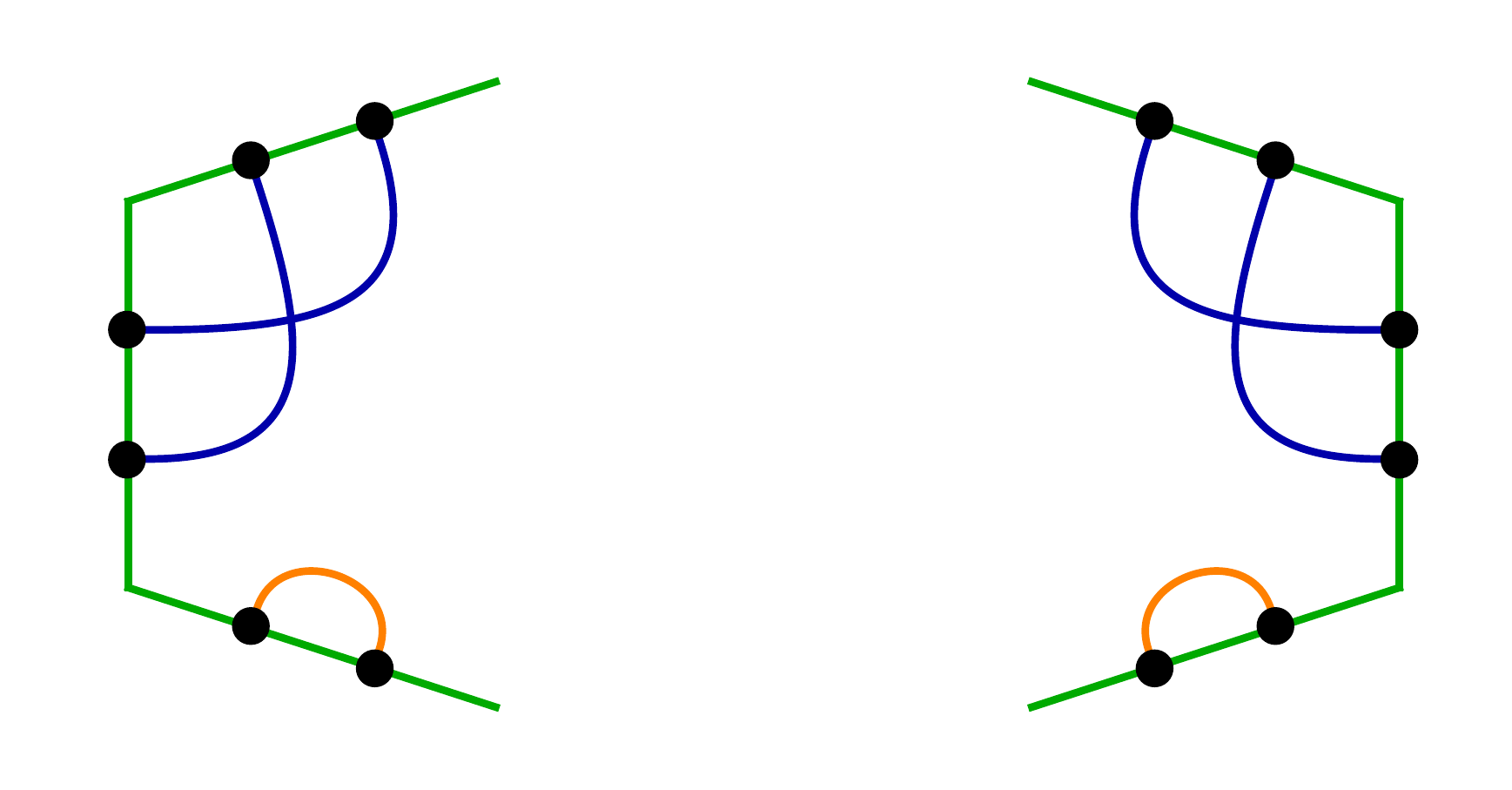}
\end{gathered}
+\;
\frac{1}{2}
\begin{gathered}
\includegraphics[height=0.1\textheight]{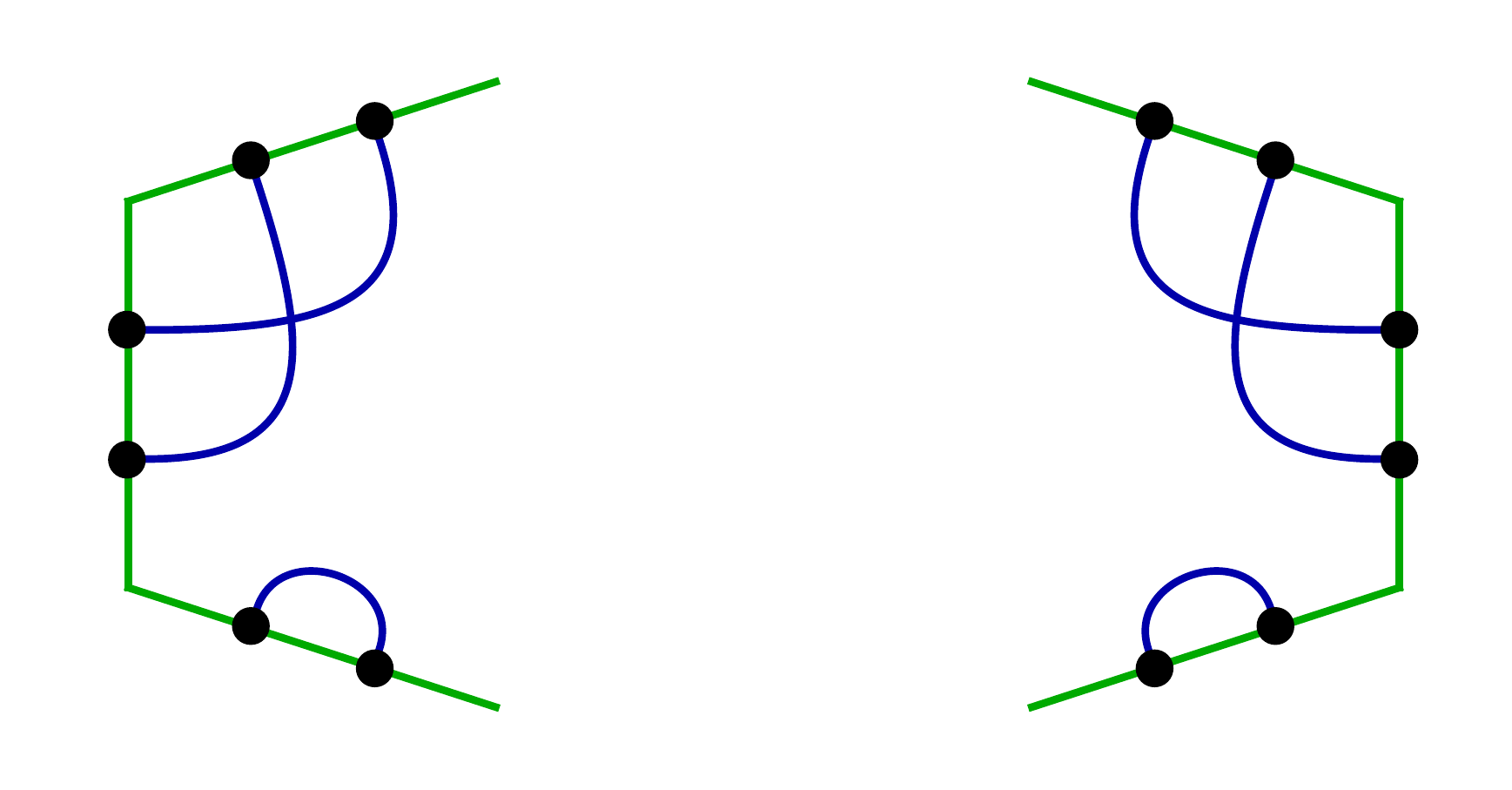}
\end{gathered}
\;=\;
\frac{1}{\sqrt{2}}
\begin{gathered}
\includegraphics[height=0.1\textheight]{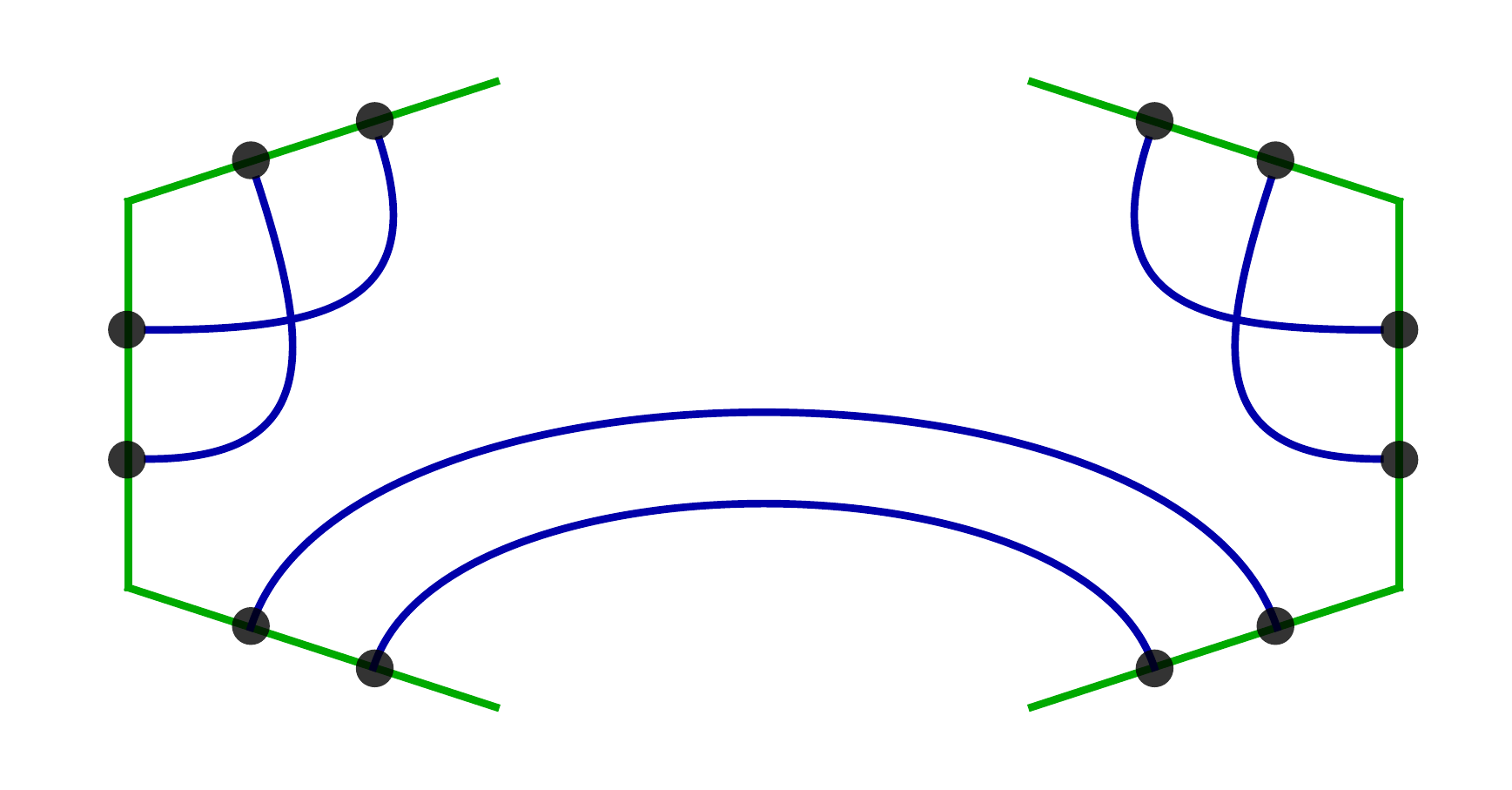}
\end{gathered}
\end{align}
We have omitted Majorana labels for clarity. 
In the first step, we used \eqref{EQ_EPR2} to relate partial trace and contraction, then applied \eqref{EQ_EPR2EX} in the second, yielding proper normalization factors. The third step merely uses \eqref{EQ_EPR2} in reverse. In summary, we see that normalization (requiring $\tr \rho_A = 1$ at each step) leads to a simple rule: Each contraction that glues two pairs of dimer together produces a factor of $1/\sqrt{2}$.

The entanglement entropy now follows from the eigenvalue spectrum of $\rho_A$. We can compute the eigenstates by projecting a full basis of Majorana dimer states onto the contracted edges of $\ket\psi$. For simplicity, we choose the basis of local Fock states, i.e.,\ with dimers only between the Majorana modes on each edge. As two edges are contracted out, there are four such basis states, of which only two are non-vanishing. These eigenvectors $\ket{\psi_{A,1}}$ and $\ket{\psi_{A,2}}$ are given by
\begin{align}
\label{EQ_RHOA_ES_EXAMPLE}
\ket{\psi_{\textcolor{darkgreen}{A},1}} \;&= 
\sqrt{2}
\begin{gathered}
\includegraphics[height=0.1\textheight]{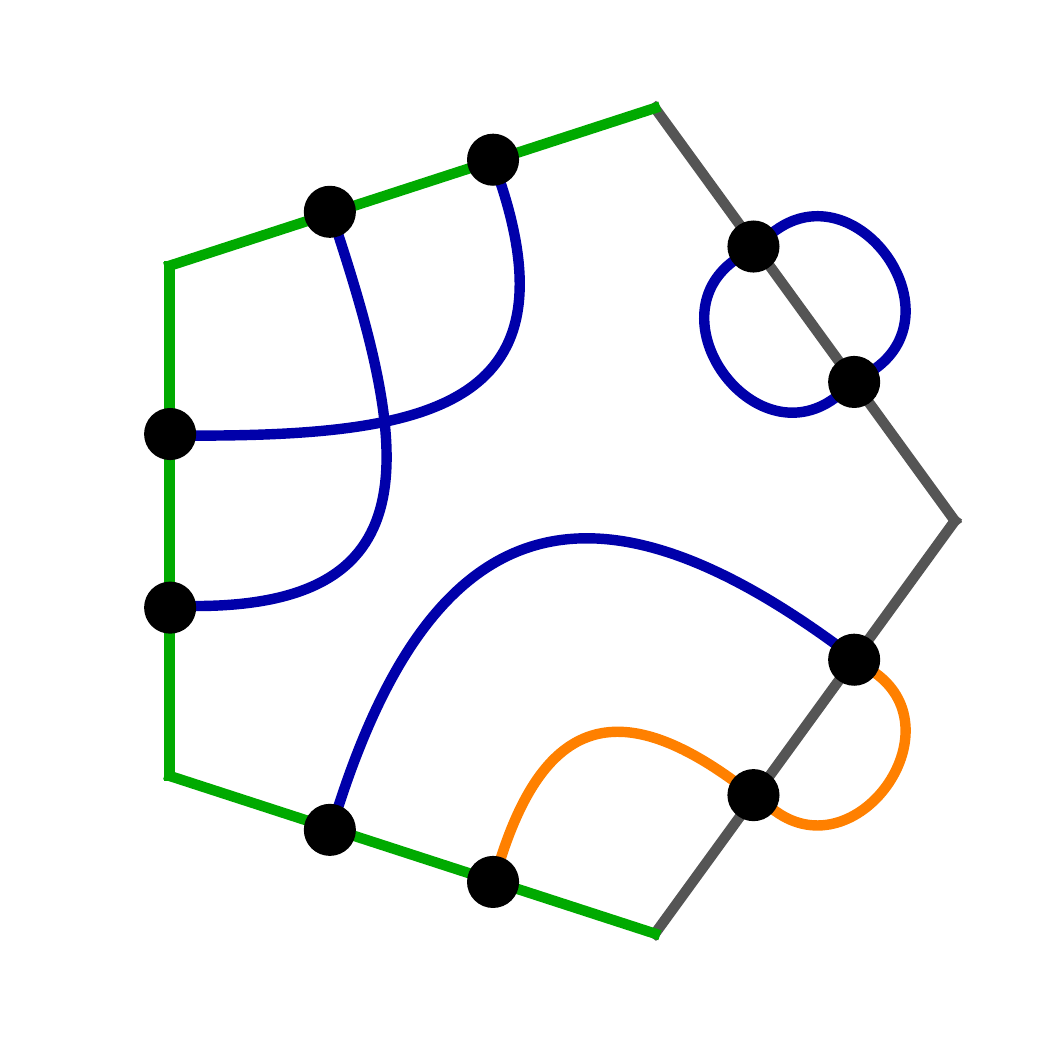}
\end{gathered} 
\;=\;
\begin{gathered}
\includegraphics[height=0.1\textheight]{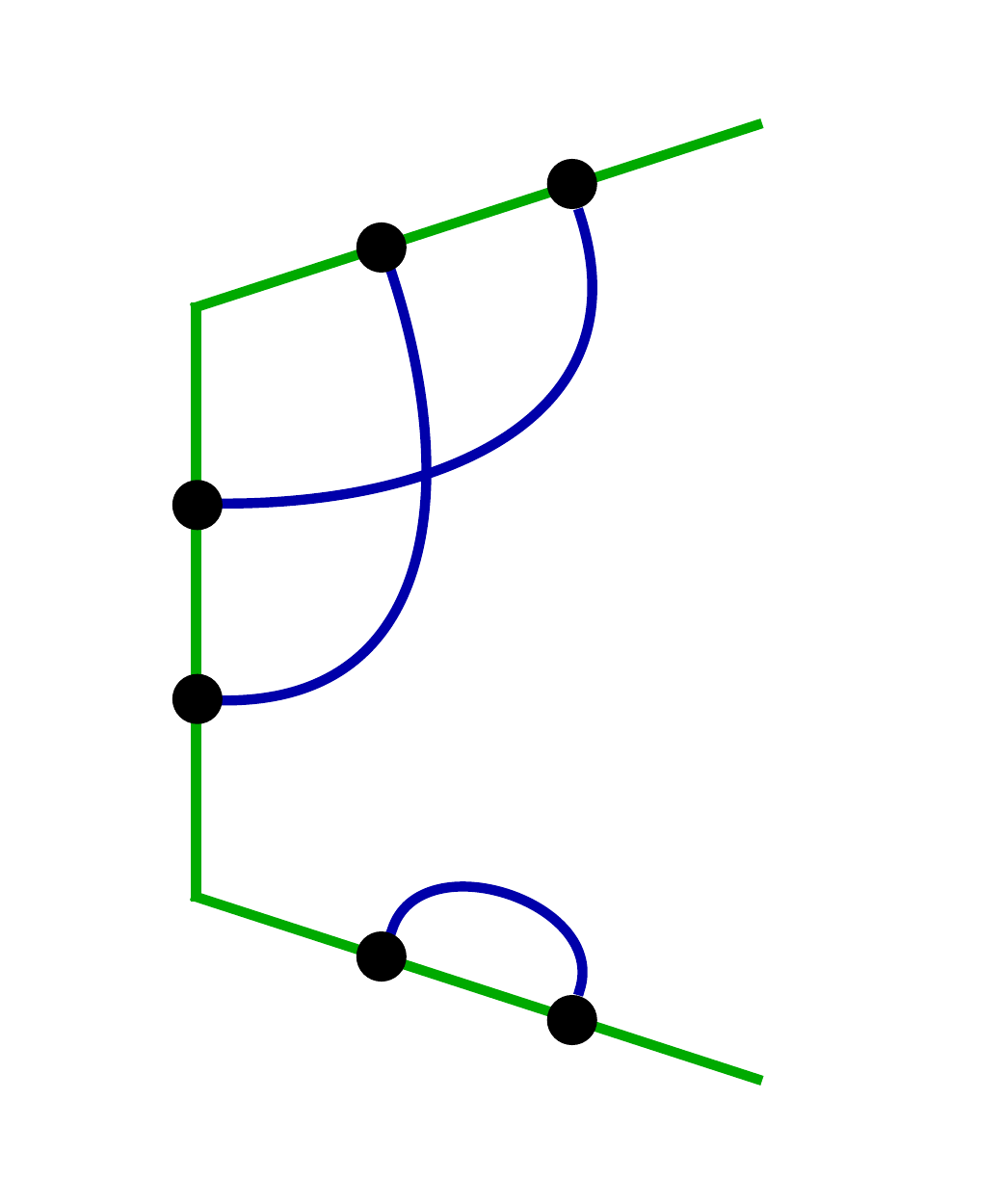}
\end{gathered} \;
&
\ket{\psi_{\textcolor{darkgreen}{A},2}} \;&=
\sqrt{2}
\begin{gathered}
\includegraphics[height=0.1\textheight]{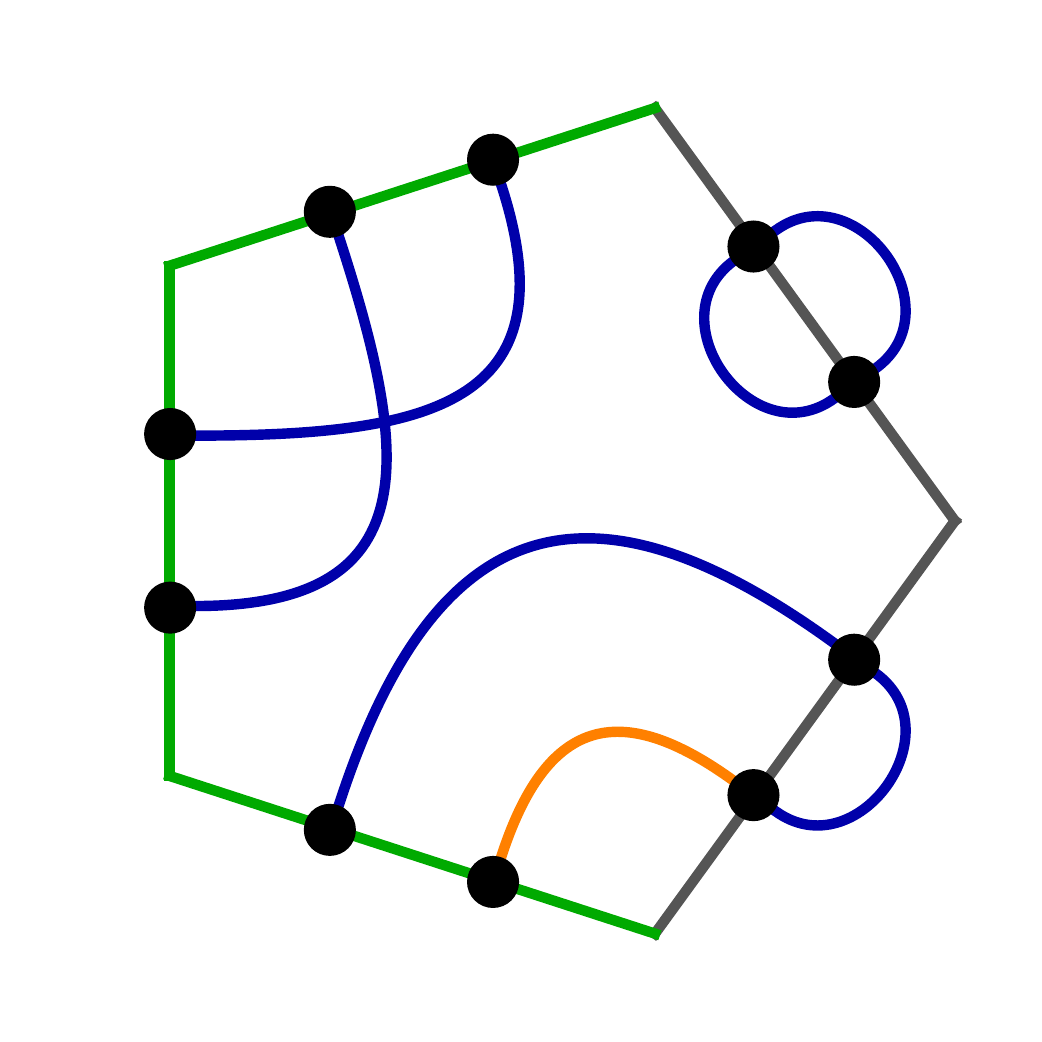}
\end{gathered} 
\;=\;
\begin{gathered}
\includegraphics[height=0.1\textheight]{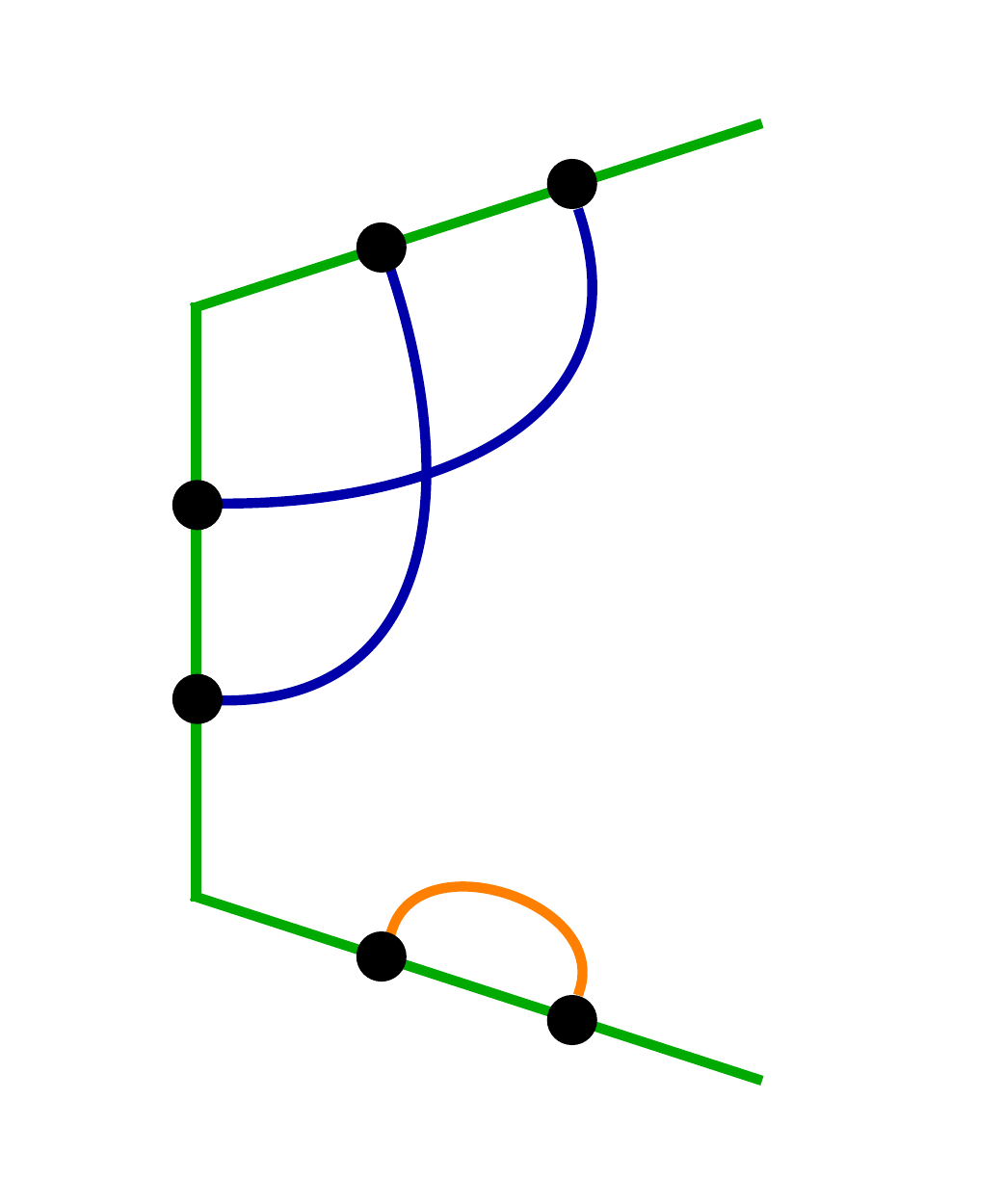}
\end{gathered}
\end{align}
Here both diagrams represent normalized states. To see that our construction indeed yields eigenstates of $\rho_A$, consider the eigenvalue equation for the second eigenvector $\ket{\psi_{A,2}}$:
\begin{equation}
\rho_{\textcolor{darkgreen}{A}} \ket{\psi_{\textcolor{darkgreen}{A},2}} \;=\;
\frac{1}{\sqrt{2}}
\begin{gathered}
\includegraphics[height=0.116\textheight]{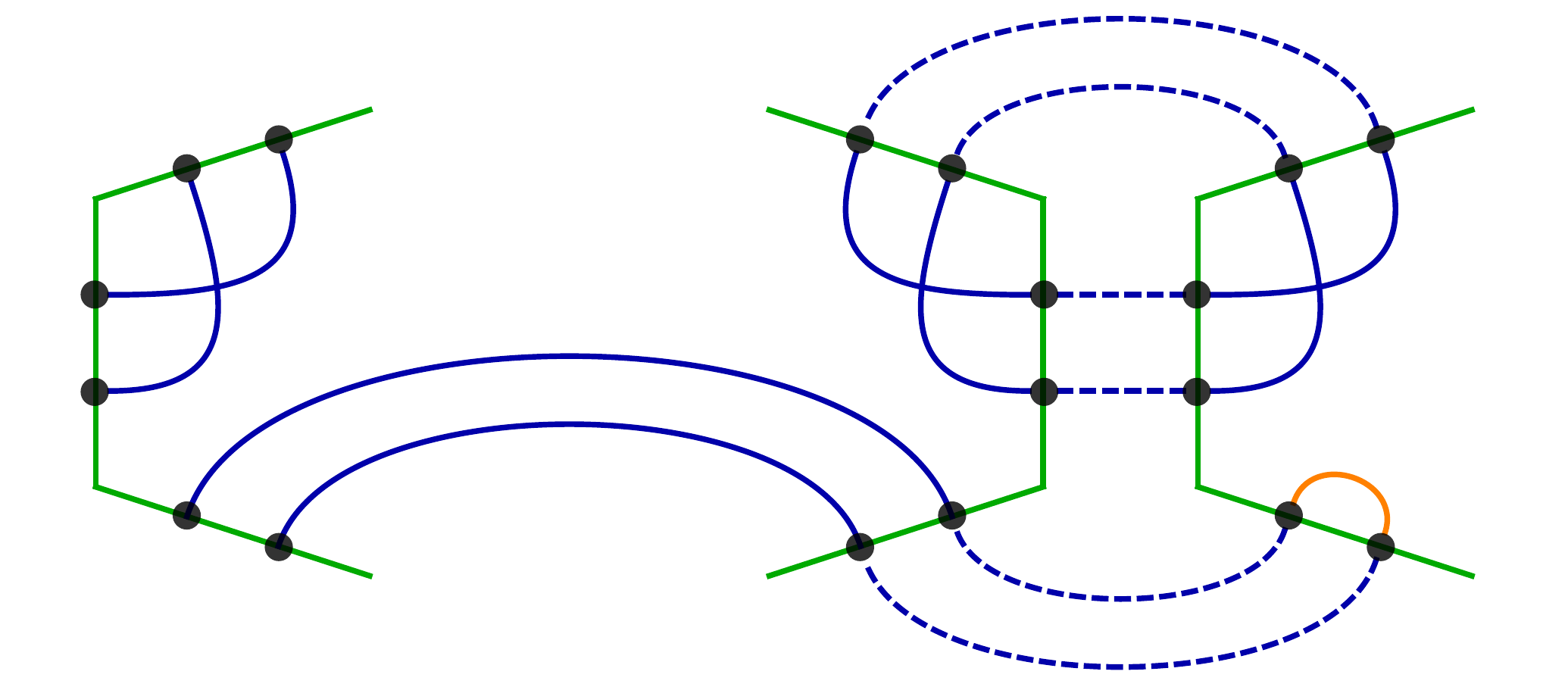}
\end{gathered}
\;=\;\frac{1}{2}
\begin{gathered}
\includegraphics[height=0.1\textheight]{pentagon_net_rhoA_2b}
\end{gathered}
\end{equation}
With a similar diagram for $\rho_A \ket{\psi_{A,1}}$ it is found that both eigenvalues are $1/2$. Thus, the entanglement entropy is given by
\begin{equation}
S_A = -\tr_A\, \rho_A \log\rho_A = - \frac{1}{2} \log\frac{1}{2} - \frac{1}{2} \log\frac{1}{2} = \log{2} \text{ .}
\end{equation}
We find that evaluating $S_A$ reduces to counting the dimers connecting the ``ket'' edges with the ''bra'' edges, which determines how mixed $\rho_A$ is. In general, any reduced density matrix contains $2m$ ``mixing dimers'' that span an eigenspace of $2^m$ orthogonal dimer states (whose diagrammatic representation is not unique for $m>1$). The entanglement entropy follows as
\begin{equation}
\label{EQ_EE_DIMER_PROOF}
S_A = - 2^m \left( \frac{1}{2^m}\log\frac{1}{2^m} \right) = m \log{2} \text{ .}
\end{equation}
Equivalently, as $2m$ dimers connect $A$ with the complementary (contracted) region $A^{\text{C}}$, each dimer contributes $\frac{1}{2}\log{2}$ to the entanglement entropy $S_A$, as in \eqref{EQ_DIMER_EE}.

Using a similar strategy, we can compute the R\'enyi entropy $S_A^{(n)} = \log (\tr \rho_A^n)/(1{-}n)$. This requires evaluating the $n$th power of the reduced density matrix $\rho_A$. As an example, consider the square of \eqref{EQ_RHO_A}:
\begin{align}
\label{EQ_RENYI_DIMERS}
\rho_{\textcolor{darkgreen}{A}}^2 \;= \;
\left(
\frac{1}{\sqrt{2}}
\begin{gathered}
\includegraphics[height=0.1\textheight]{pentagon_net_rhoA_1b}
\end{gathered}
\right)^2
&=\;
\frac{1}{2}
\begin{gathered}
\includegraphics[height=0.116\textheight]{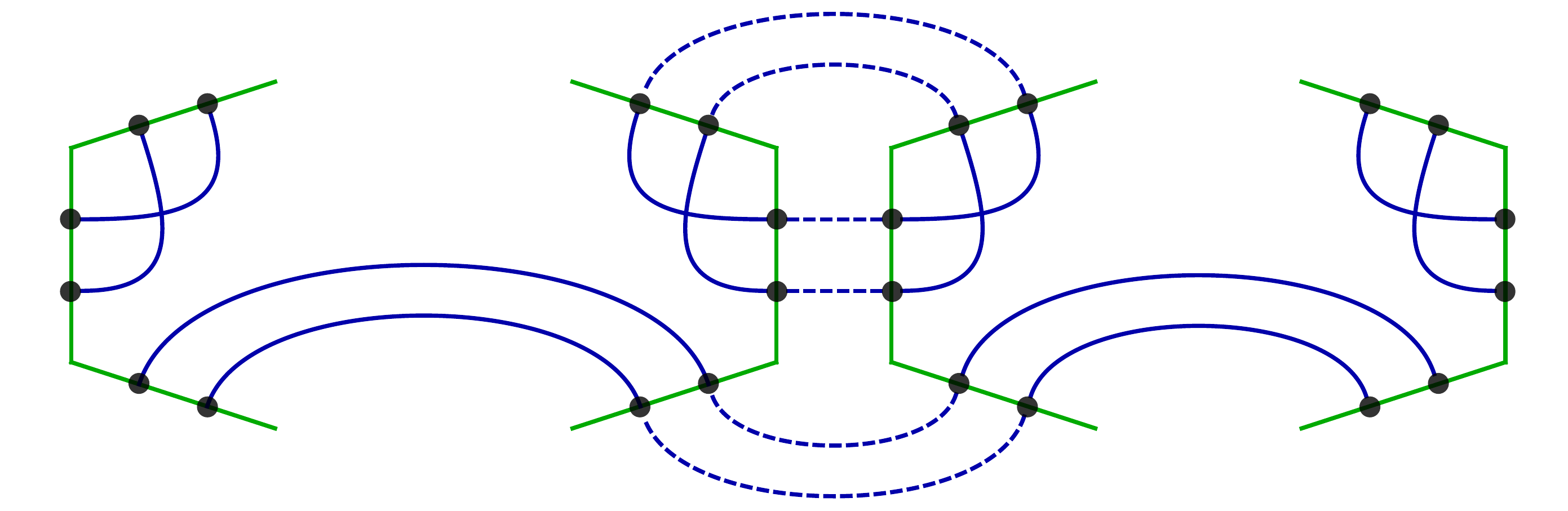}
\end{gathered} \nonumber\\
&=\;
\frac{1}{2\sqrt{2}}
\begin{gathered}
\includegraphics[height=0.1\textheight]{pentagon_net_rhoA_1b}
\end{gathered}
\;=\;\frac{1}{2}\, \rho_{\textcolor{darkgreen}{A}}\ .
\end{align}
Thus, it follows that $\rho_A^n = \rho_A / 2^{n-1}$, and hence $S_A^{(n)} = \log(\tr \rho_A / 2^{n-1})/(1{-}n)=\log{2}$. This property of a ``flat R\'enyi spectrum'', i.e.,\ $S_A^{(n)}=S_A$, holds for any Majorana dimer state. For a generic $\rho_A$, \eqref{EQ_RENYI_DIMERS} involves $n-1$ contractions of $2m$ mixing dimers, leading to the following analog of \eqref{EQ_EE_DIMER_PROOF} for R\'enyi entropies:
\begin{equation}
S_A^{(n)} = \frac{1}{1-n} \log \tr \rho_A^n = \frac{1}{1-n} \log \frac{\tr \rho_A}{2^{m (n-1)}} = m \log 2 = S_A \text{ .}
\end{equation}
Note that this property of a flat entanglement spectrum is a proven feature of stabilizer codes states \cite{PhysRevLett.103.261601}, thus making Majorana dimers ideal for the study of stabilizers. Contrary to the stabilizer picture, however, we can also diagrammatically evaluate the entanglement entropy for classes of superpositions, as we will now see. First, we take a look at superpositions of input states of the HyPeC. Consider a single tile of the $[[5,1,3]]$ code. For arbitrary bulk input, the boundary state is given by
\begin{equation}
\label{EQ_HAPPY_SUPERPOS}
\ket\psi = \alpha\ket{\bar{0}}_5 + \beta\ket{\bar{1}}_5
\; = \quad \alpha\quad
\begin{gathered}
\includegraphics[height=0.1\textheight]{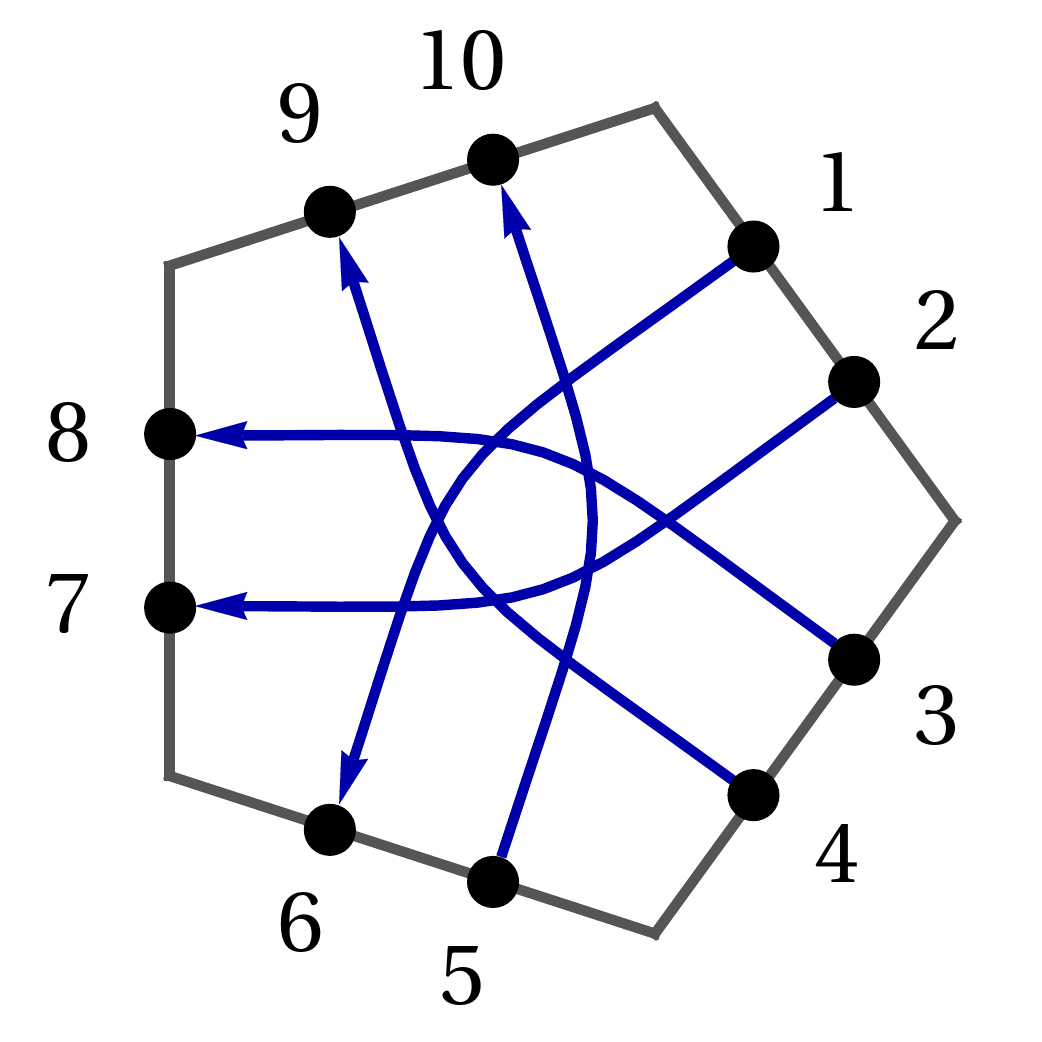}
\end{gathered}
\quad\scalebox{1.3}{$+$}\quad \beta\quad
\begin{gathered}
\includegraphics[height=0.1\textheight]{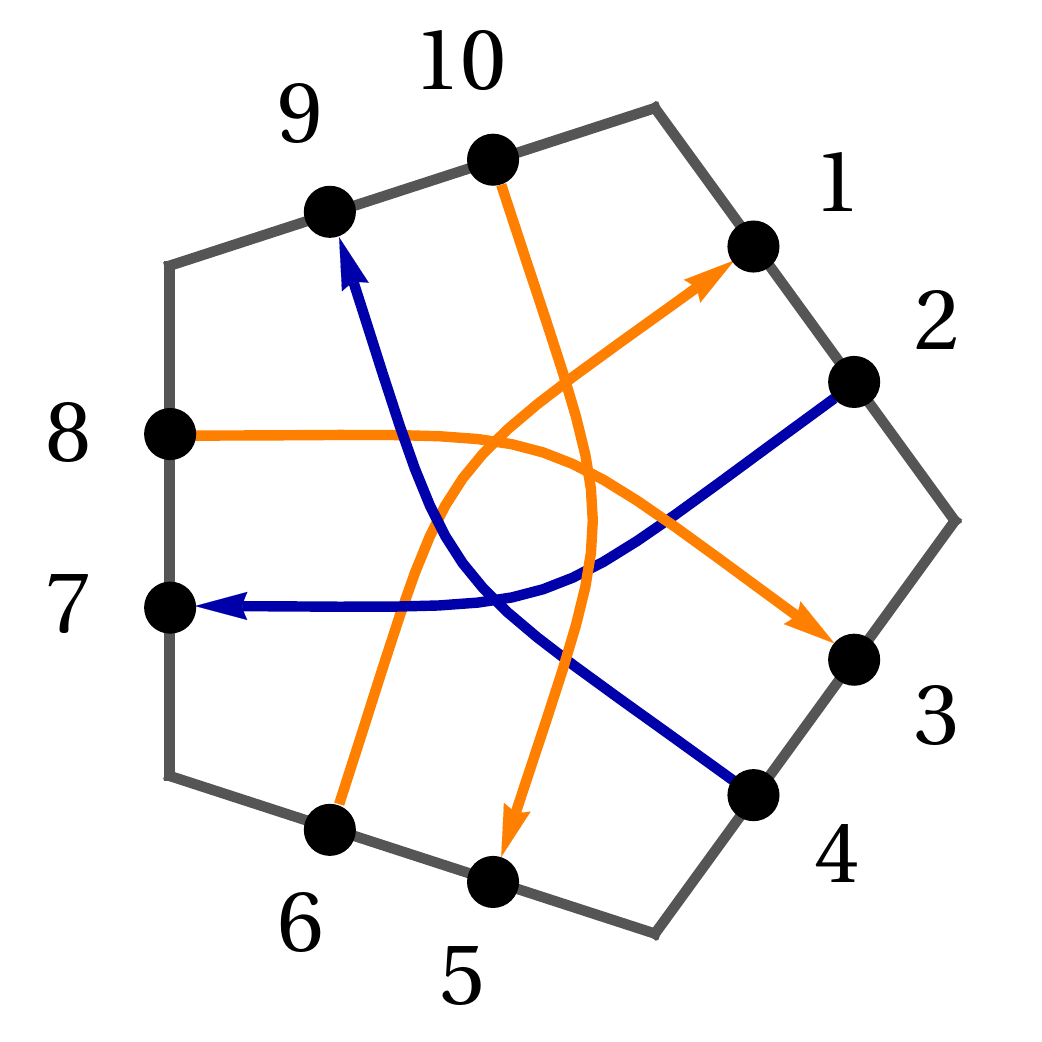}
\end{gathered} 
\quad \equiv \;
\begin{gathered}
\includegraphics[height=0.1\textheight]{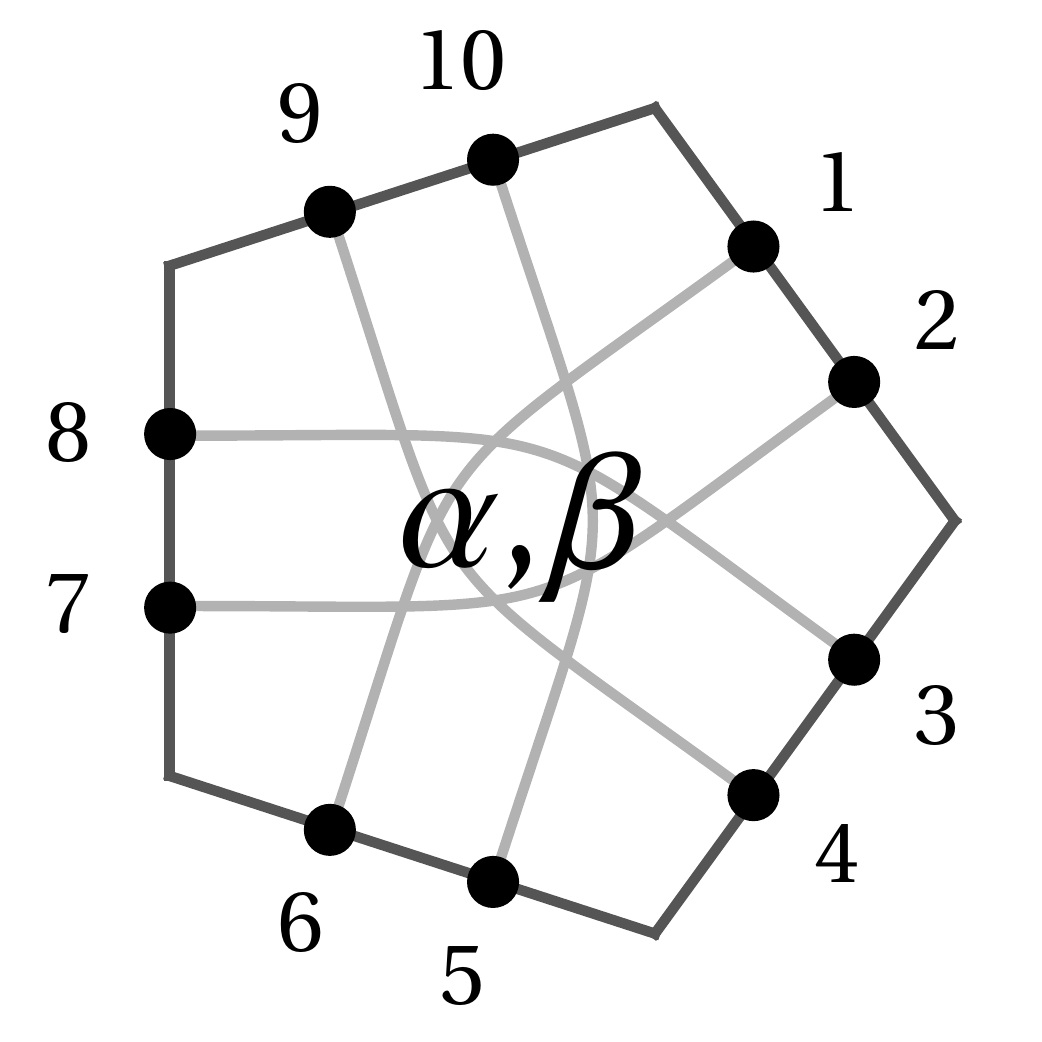}
\end{gathered} 
\end{equation}
On the right-hand side, we have used a new notation for superpositions of $[[5,1,3]]$ computational basis states with complex factors $\alpha$ and $\beta$. Normalization requires $|\alpha|^2+|\beta|^2=1$.
We now show that the reduced density matrix $\rho_A$ of this superposition becomes an identity on the subsystem $A$ when it consists of only two ($A=(1,2)$) or one edge ($A=(1)$):

\begin{align}
\label{EQ_HAPPY_SP_3C}
\rho_{\textcolor{darkgreen}{(1,2)}} \;=\; 
\begin{gathered}
\includegraphics[height=0.117\textheight]{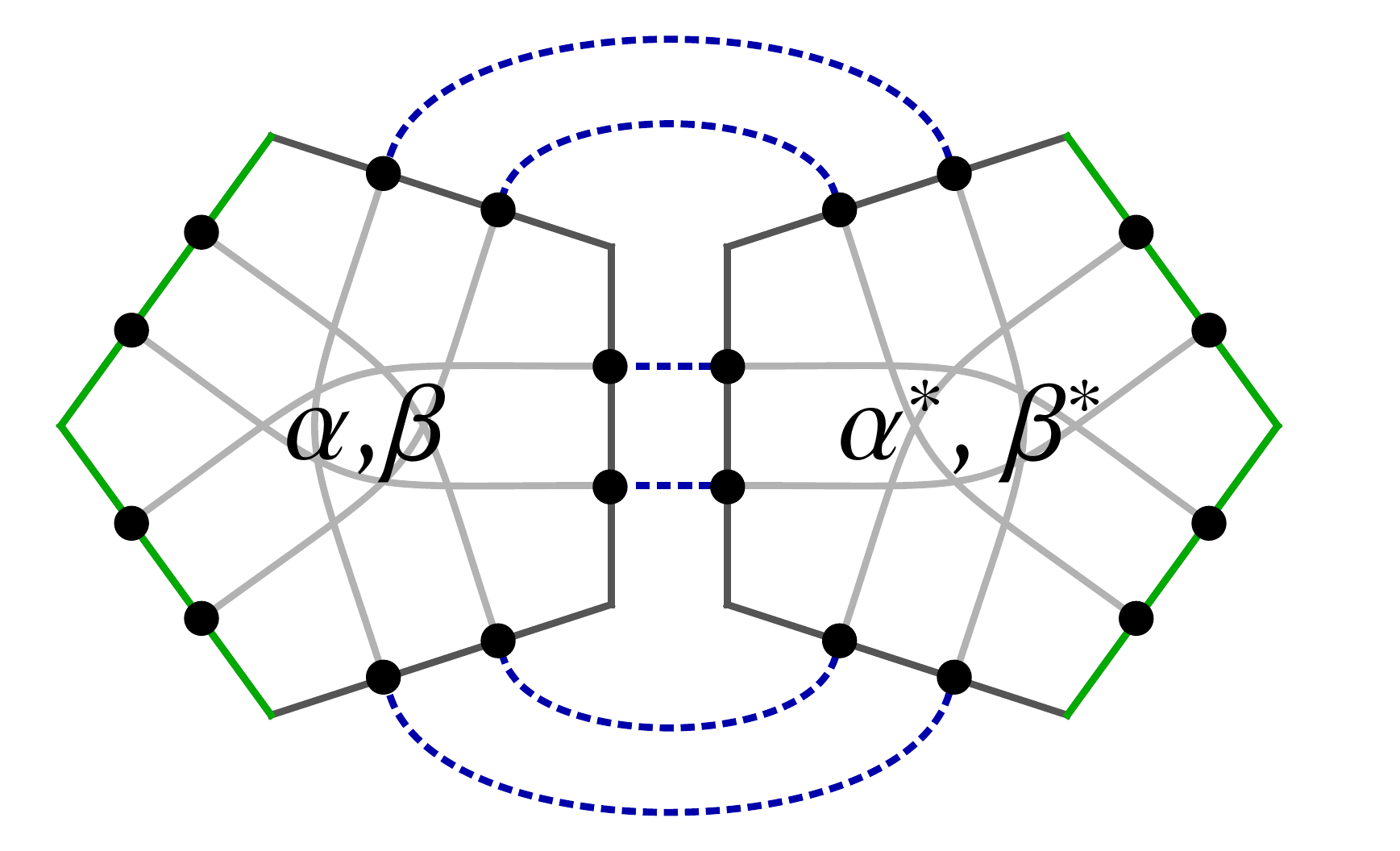}
\end{gathered}
\;&=\;
\alpha \alpha^\star
\begin{gathered}
\includegraphics[height=0.117\textheight]{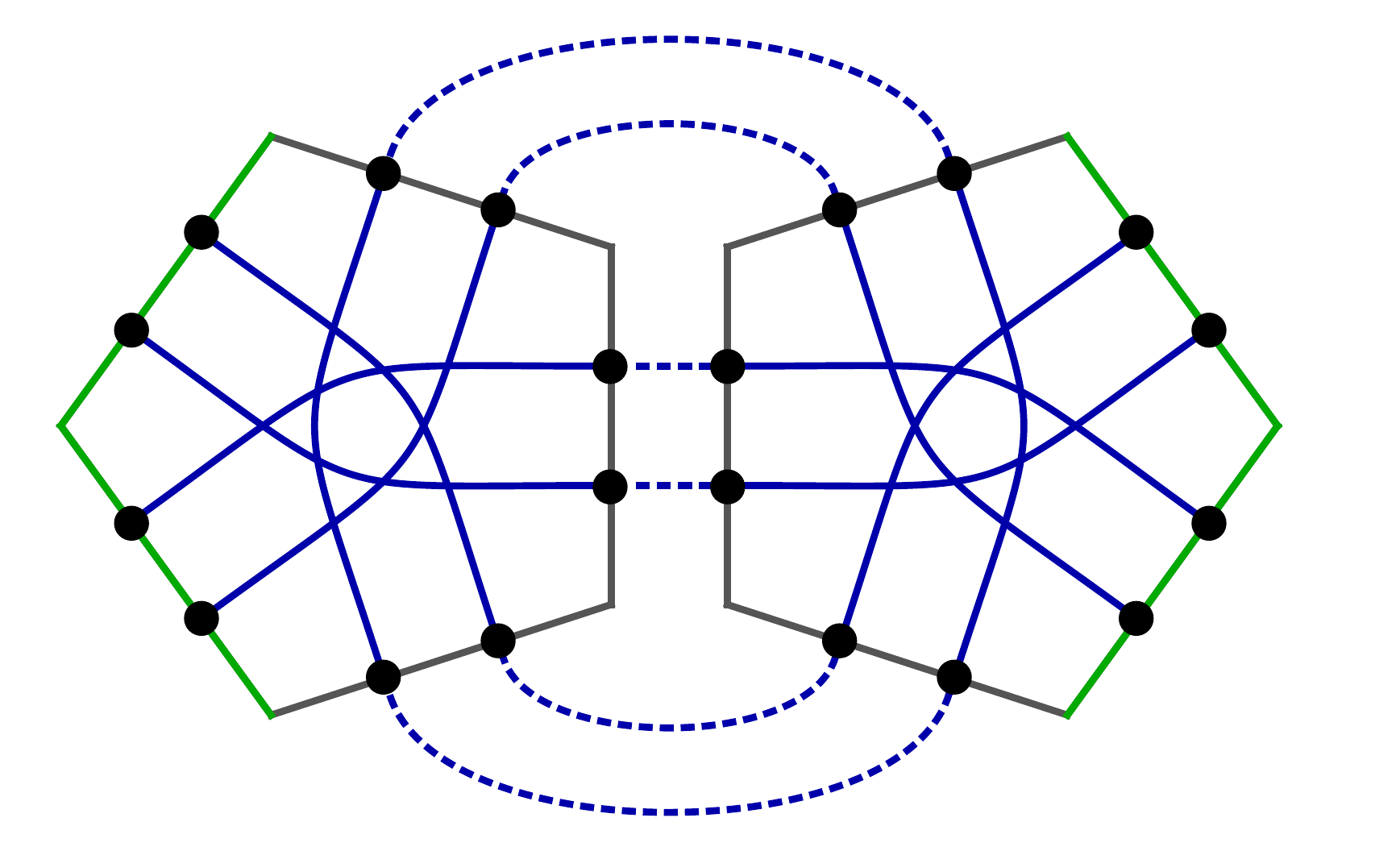}
\end{gathered}
\;+\;
\alpha \beta^\star
\begin{gathered}
\includegraphics[height=0.117\textheight]{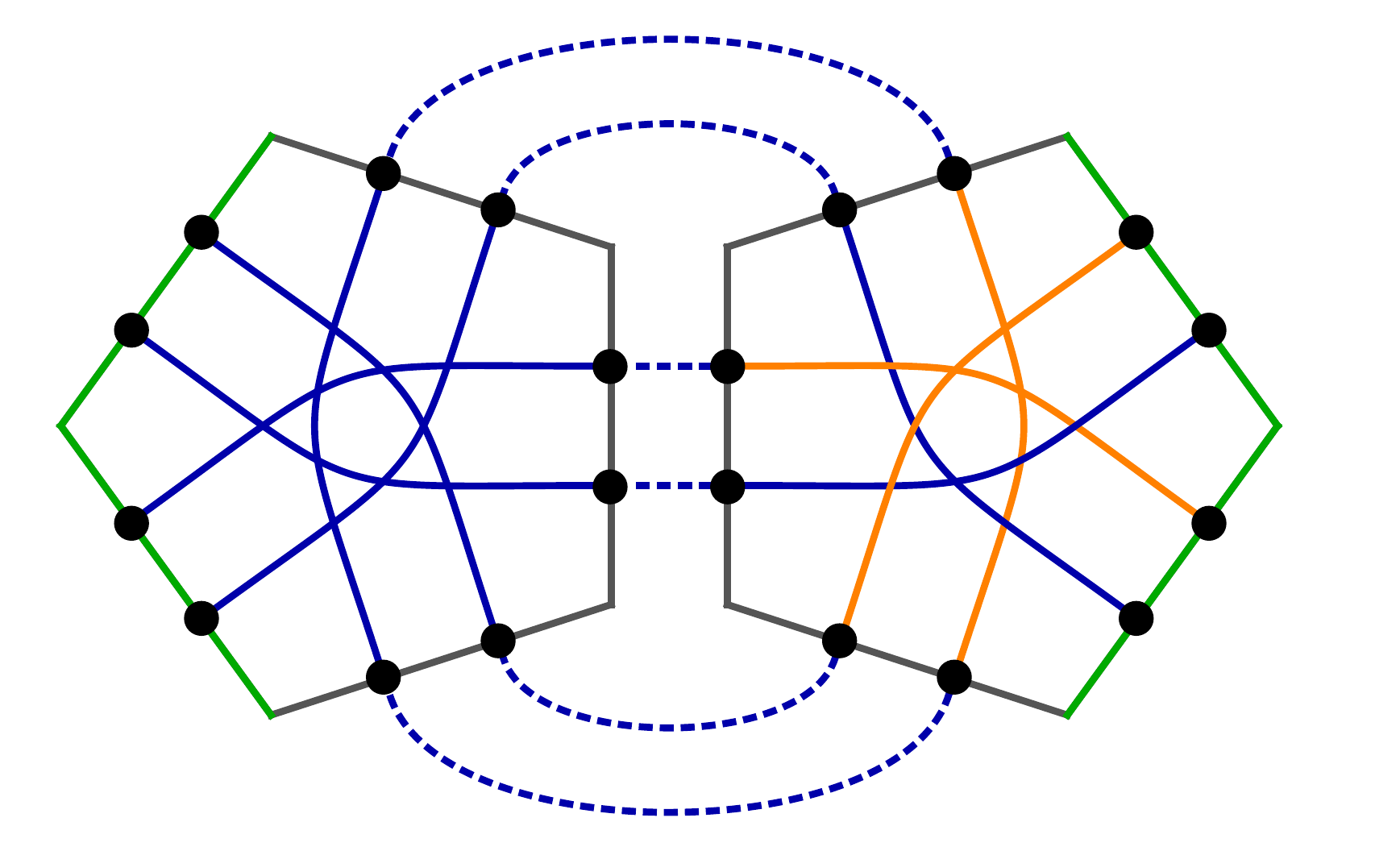}
\end{gathered} \nonumber \\
&+\;
\alpha^\star\beta
\begin{gathered}
\includegraphics[height=0.117\textheight]{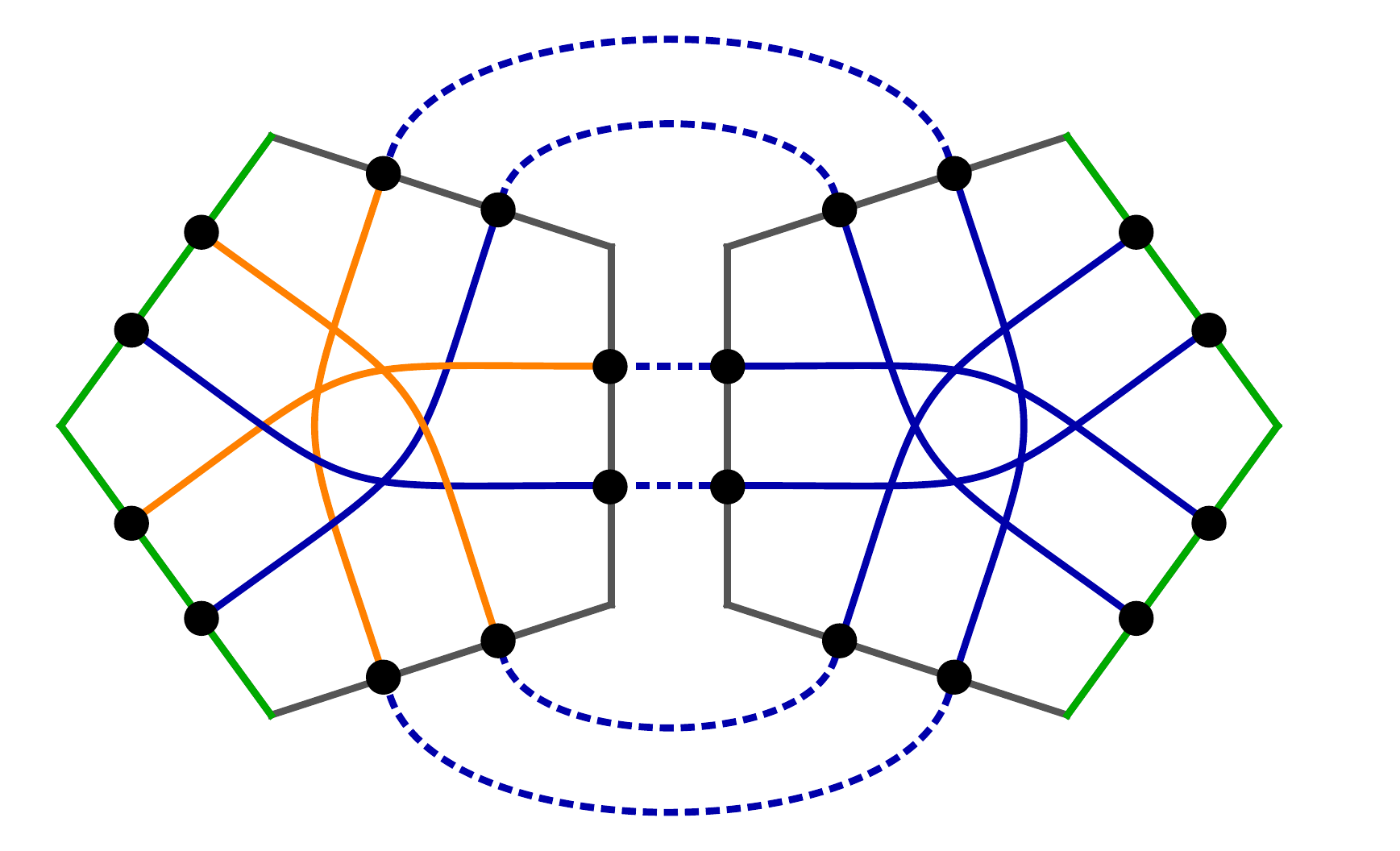}
\end{gathered}
\;+\;
\beta\beta^\star
\begin{gathered}
\includegraphics[height=0.117\textheight]{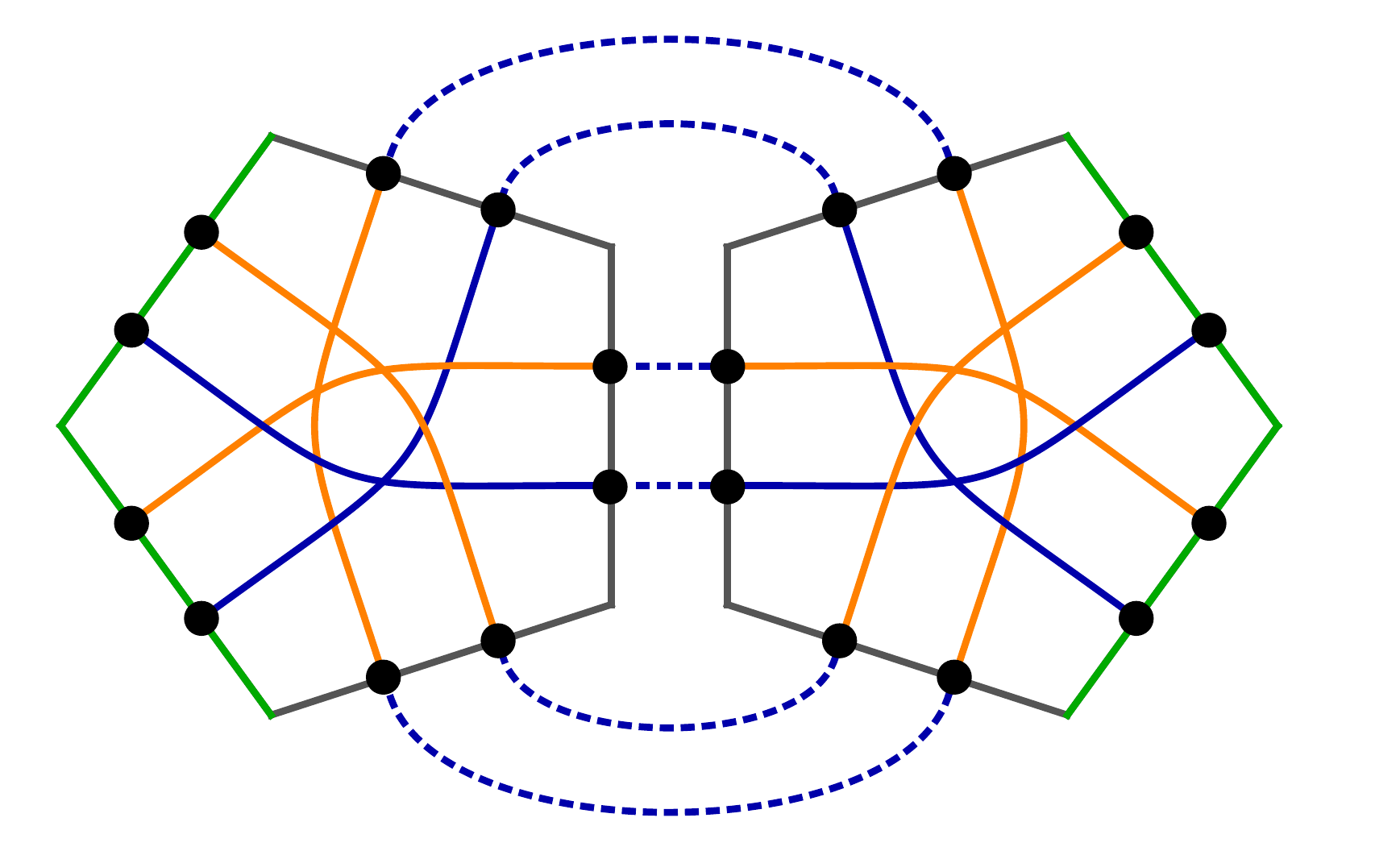}
\end{gathered} \nonumber\\
&=\; \frac{\alpha \alpha^\star + \beta\beta^\star}{2}
\begin{gathered}
\includegraphics[height=0.117\textheight]{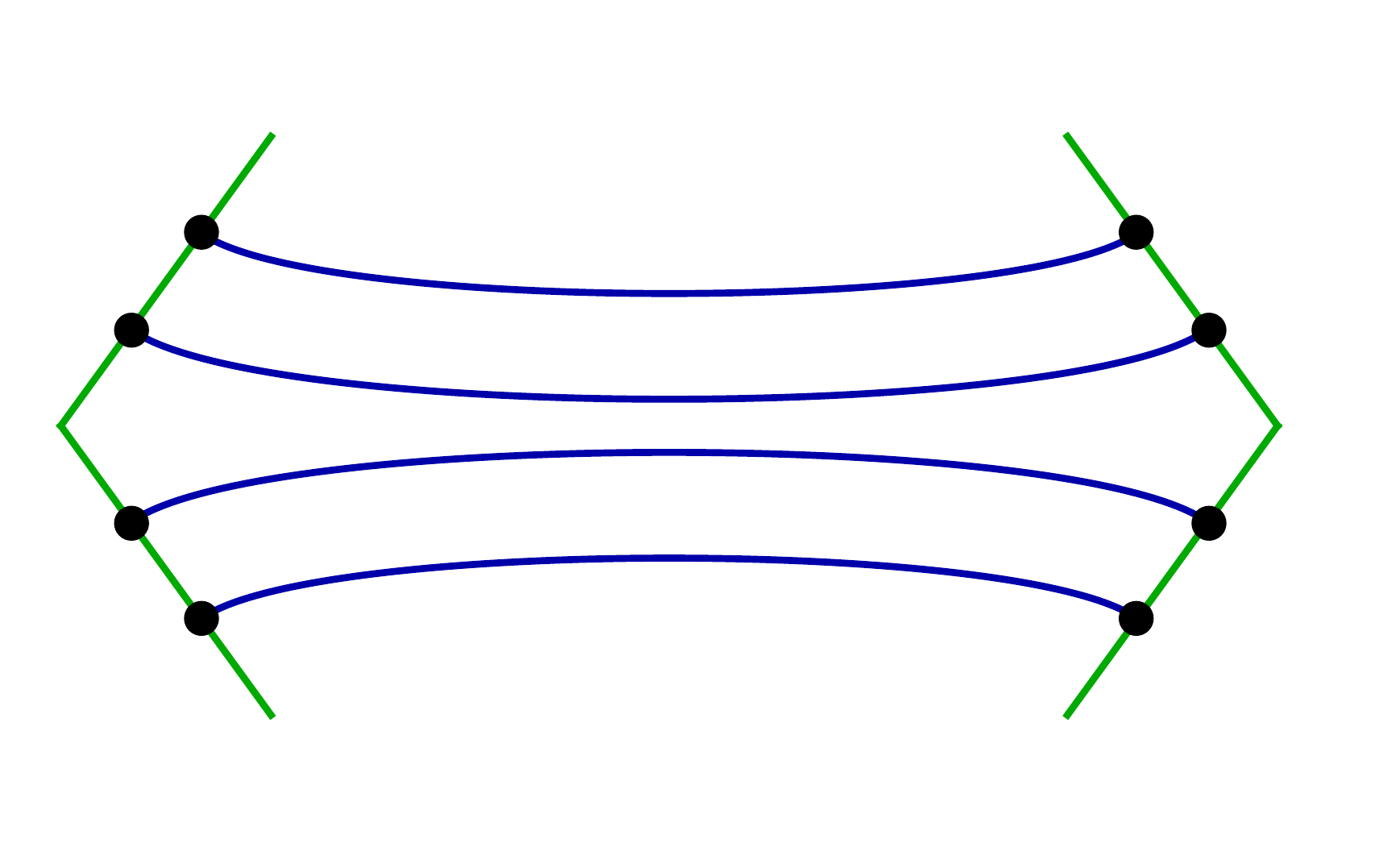}
\end{gathered}
\;=\; \frac{\id_4}{4}\ .
\end{align}
\begin{align}
\label{EQ_HAPPY_SP_4C}
\rho_{\textcolor{darkgreen}{(1)}} \;=\; 
\begin{gathered}
\includegraphics[height=0.125\textheight]{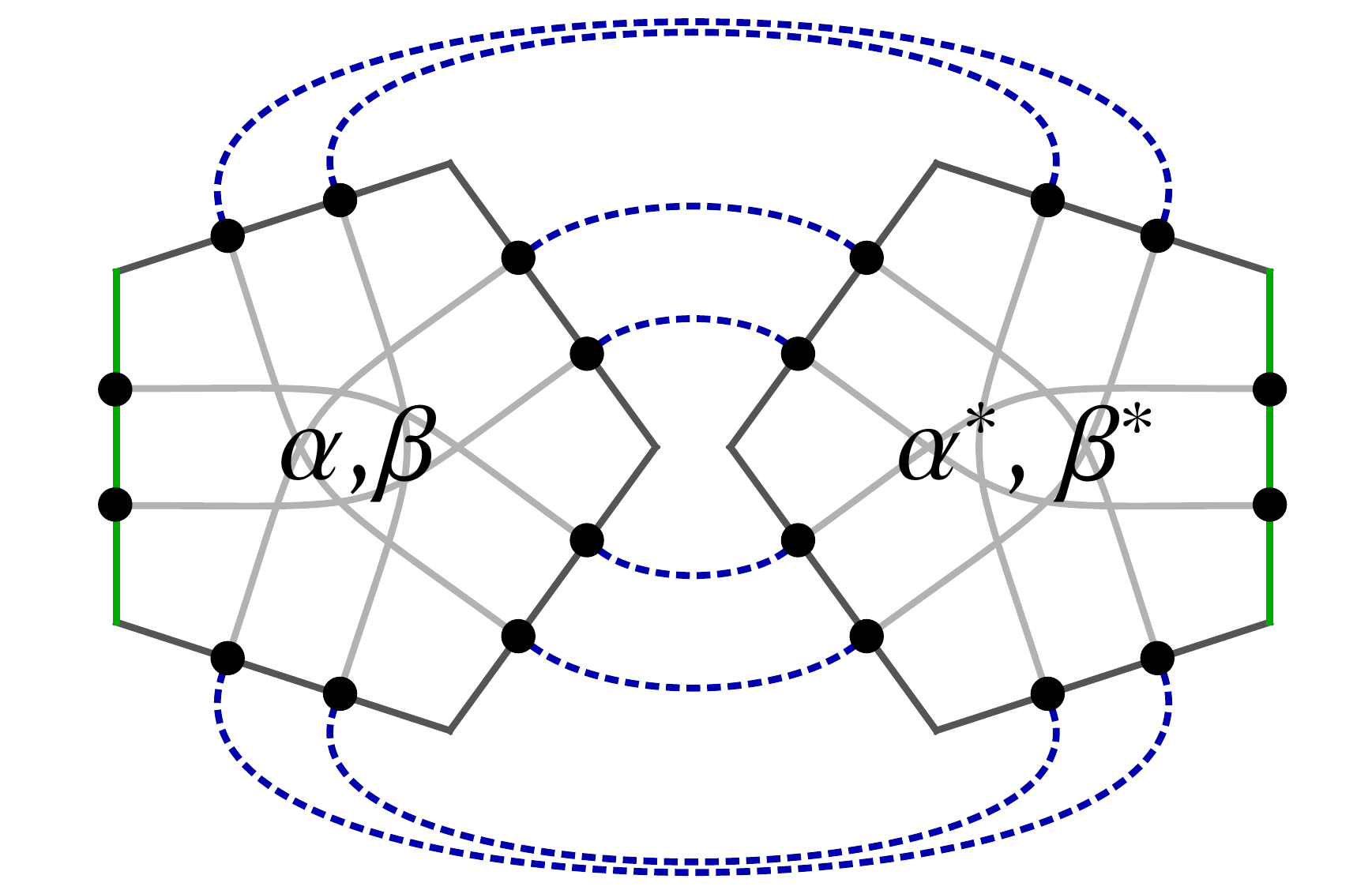}
\end{gathered}
\;&=\;
\alpha \alpha^\star
\begin{gathered}
\includegraphics[height=0.125\textheight]{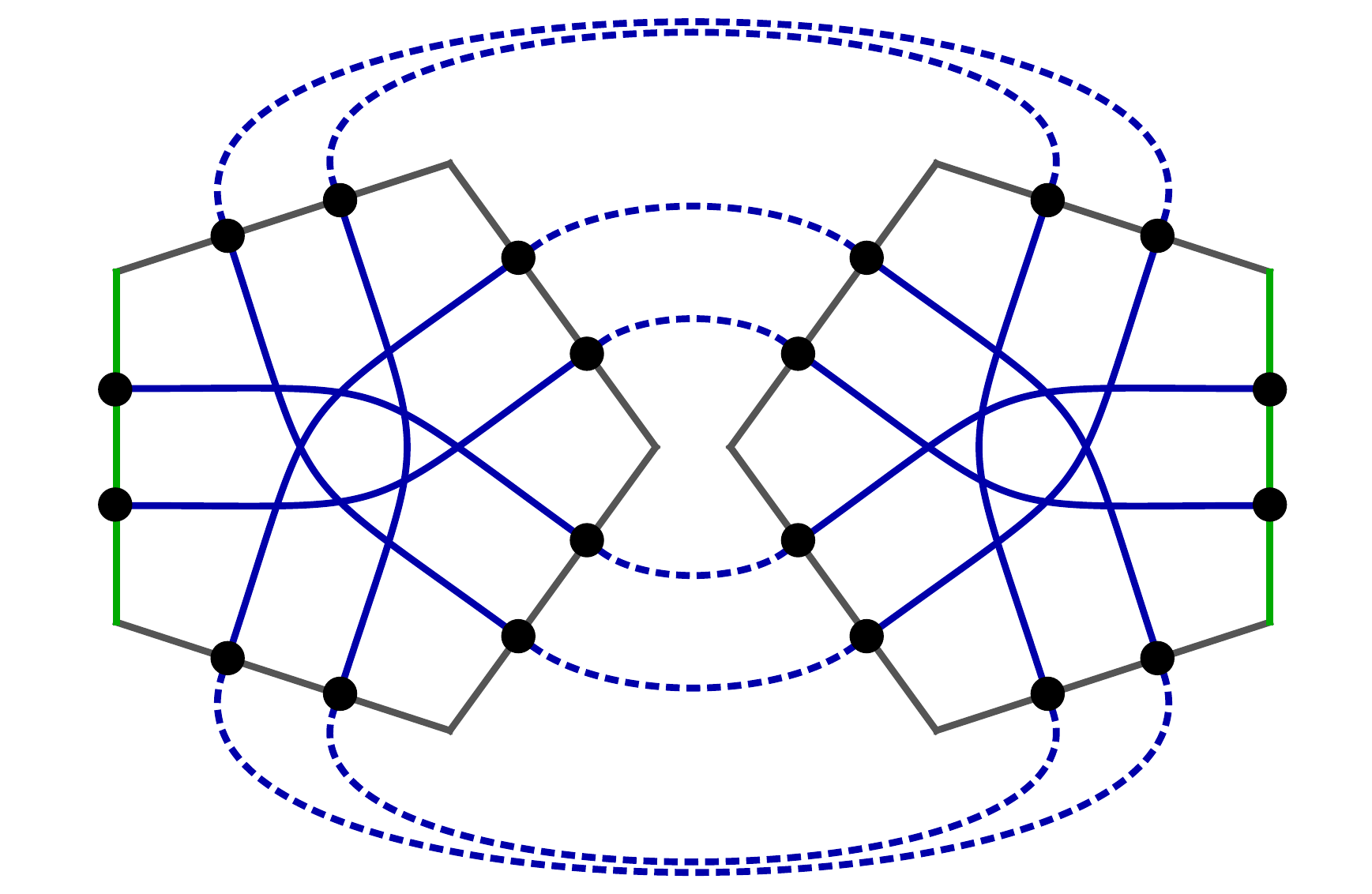}
\end{gathered}
\;+\;
\alpha \beta^\star
\begin{gathered}
\includegraphics[height=0.125\textheight]{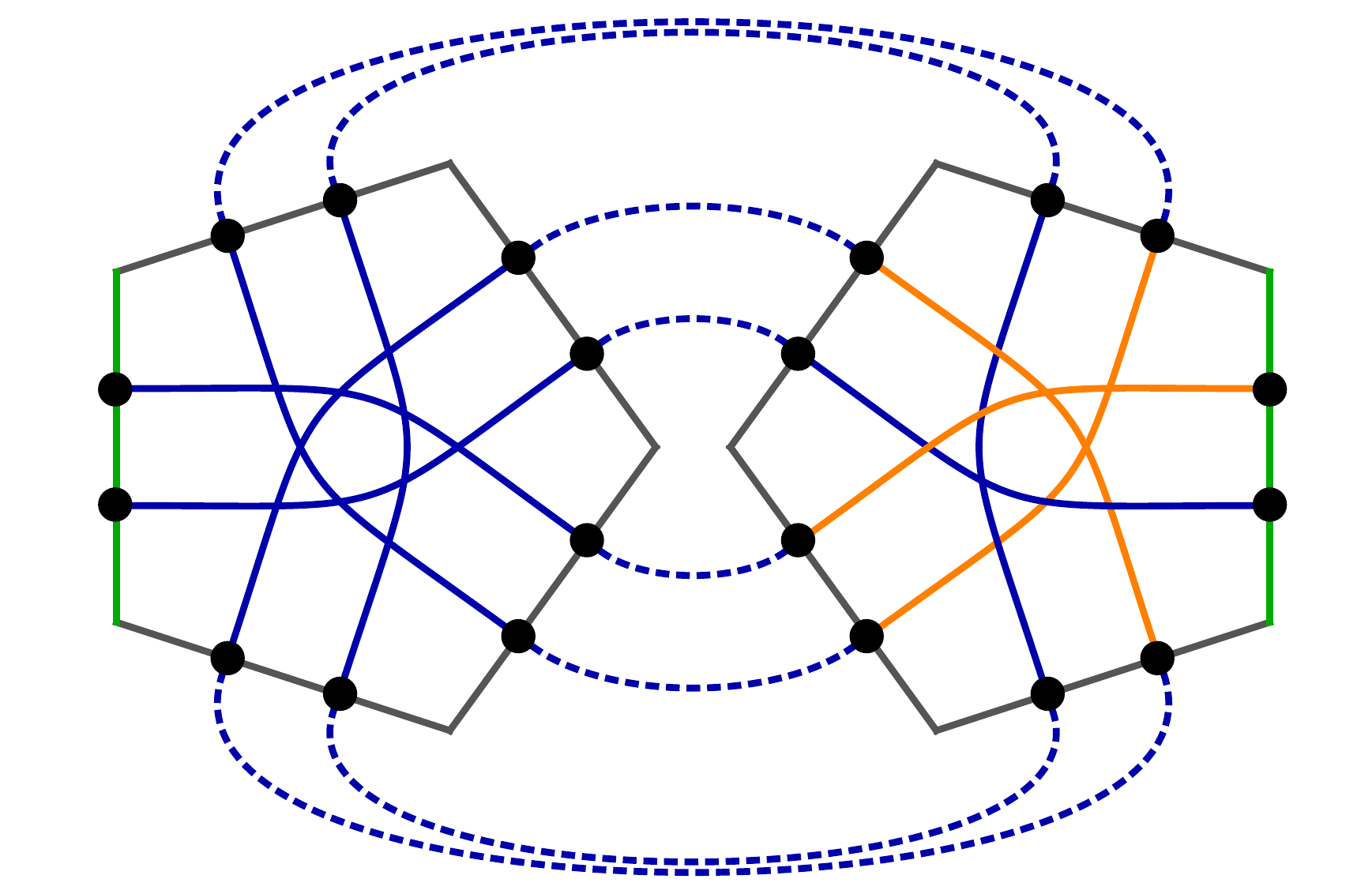}
\end{gathered} \nonumber \\
&+\;
\alpha^\star\beta
\begin{gathered}
\includegraphics[height=0.125\textheight]{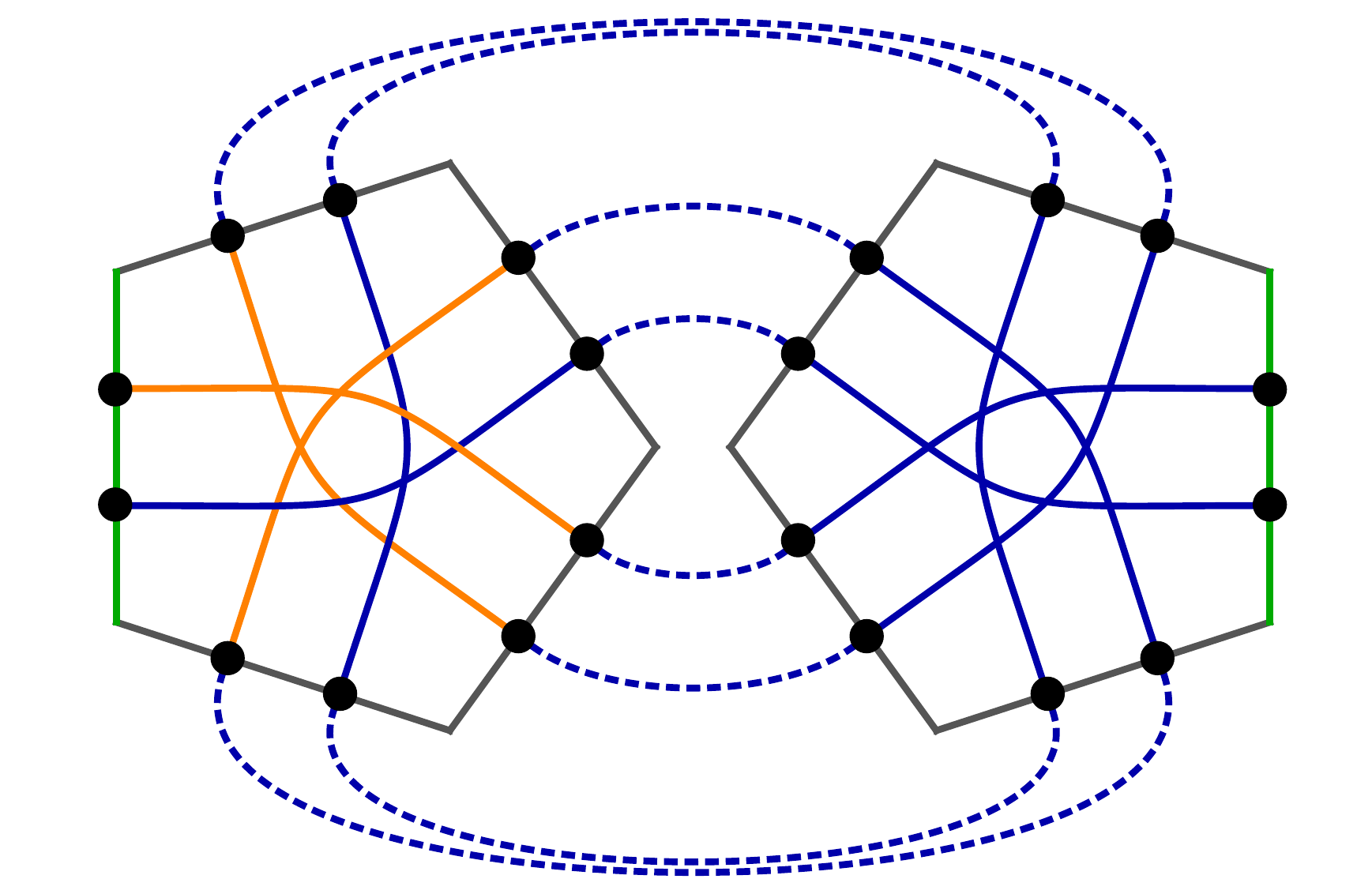}
\end{gathered}
\;+\;
\beta\beta^\star
\begin{gathered}
\includegraphics[height=0.125\textheight]{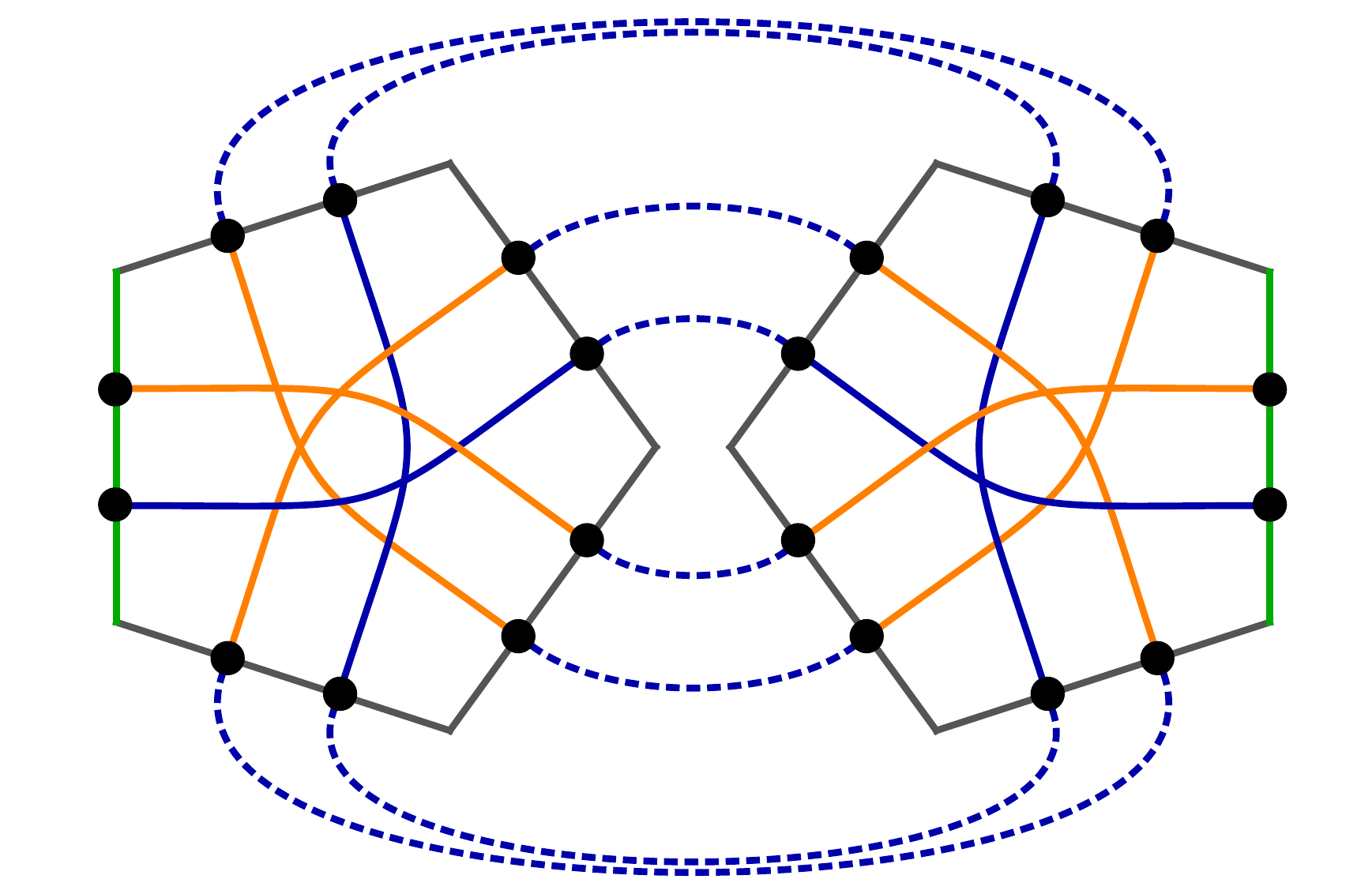}
\end{gathered} \nonumber\\
&=\; \frac{\alpha \alpha^\star + \beta\beta^\star}{\sqrt{2}}
\begin{gathered}
\includegraphics[height=0.125\textheight]{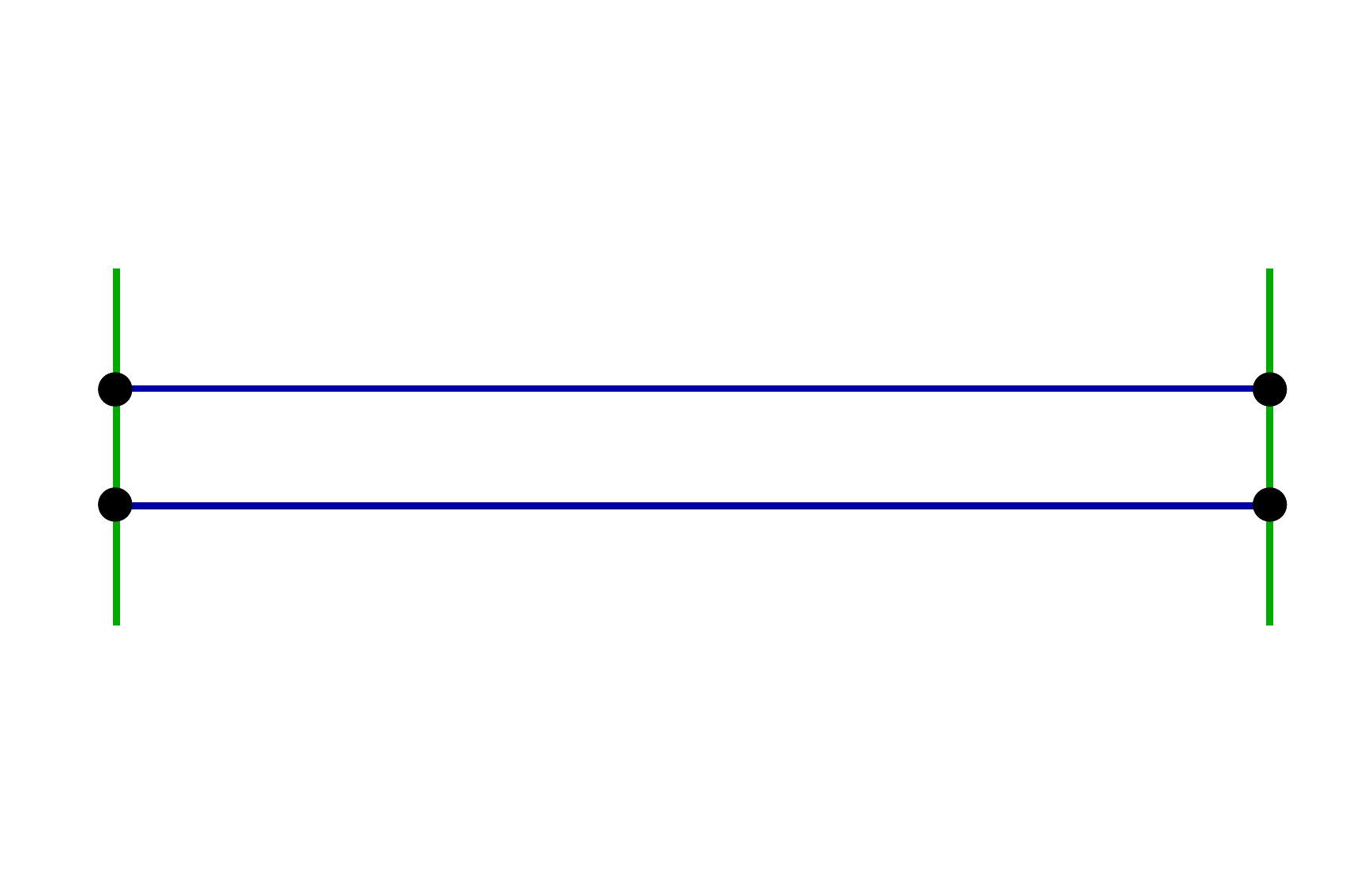}
\end{gathered}
\;=\; \frac{\id_2}{{2}}\ .
\end{align}
Note that the actual values for $\alpha$ and $\beta$ do not change $\rho_A$, and thus the entanglement spectrum is the same for \textsl{any} superposition of the $[[5,1,3]]$ logical code states. As we can easily see, $S_{(1)}=\log 2$ and $S_{(1,2)}=2\log 2$, identical to the result for the logical code states. The corresponding eigenstates are simply a complete basis of Majorana dimers on one or two edges, respectively.

Let us now consider the $|A|=3$ case. We easily find the entanglement entropy $S_A = S_{A^{\text{C}}}=2 \log 2$. The eigenstates of $\rho_A$ can be found, as in \eqref{EQ_RHOA_ES_EXAMPLE}, by starting with the state vector $\ket\psi$ and contracting a complete basis on the edges in $A^{\text{C}}$, yielding four eigenstates. We compute the first one explicitly:
\begin{align}
\label{EQ_RHO_A_513_EV1}
|\psi_{\textcolor{darkgreen}{A}}^{0,0}\rangle \;&=
\begin{gathered}
\includegraphics[height=0.1\textheight]{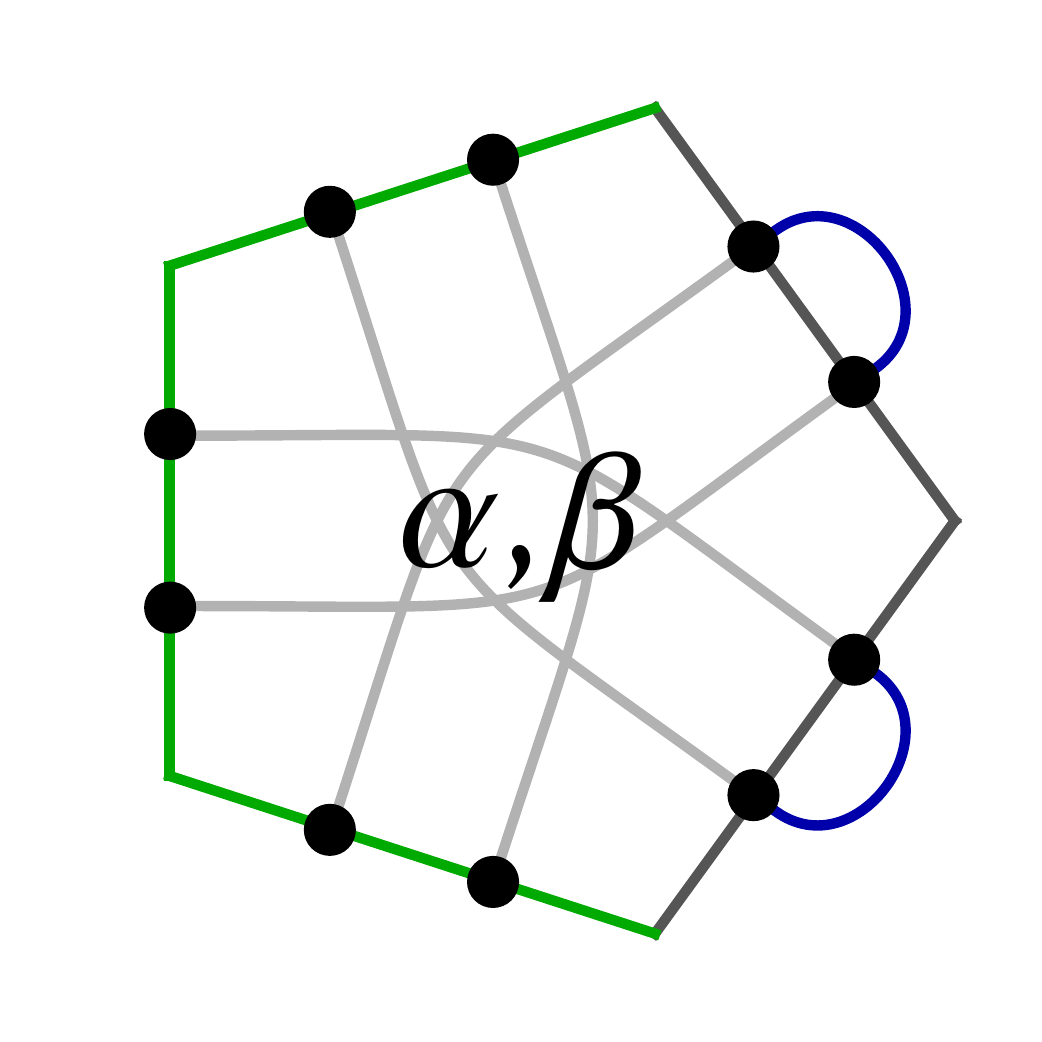}
\end{gathered}
\;=\;
\alpha
\begin{gathered}
\includegraphics[height=0.1\textheight]{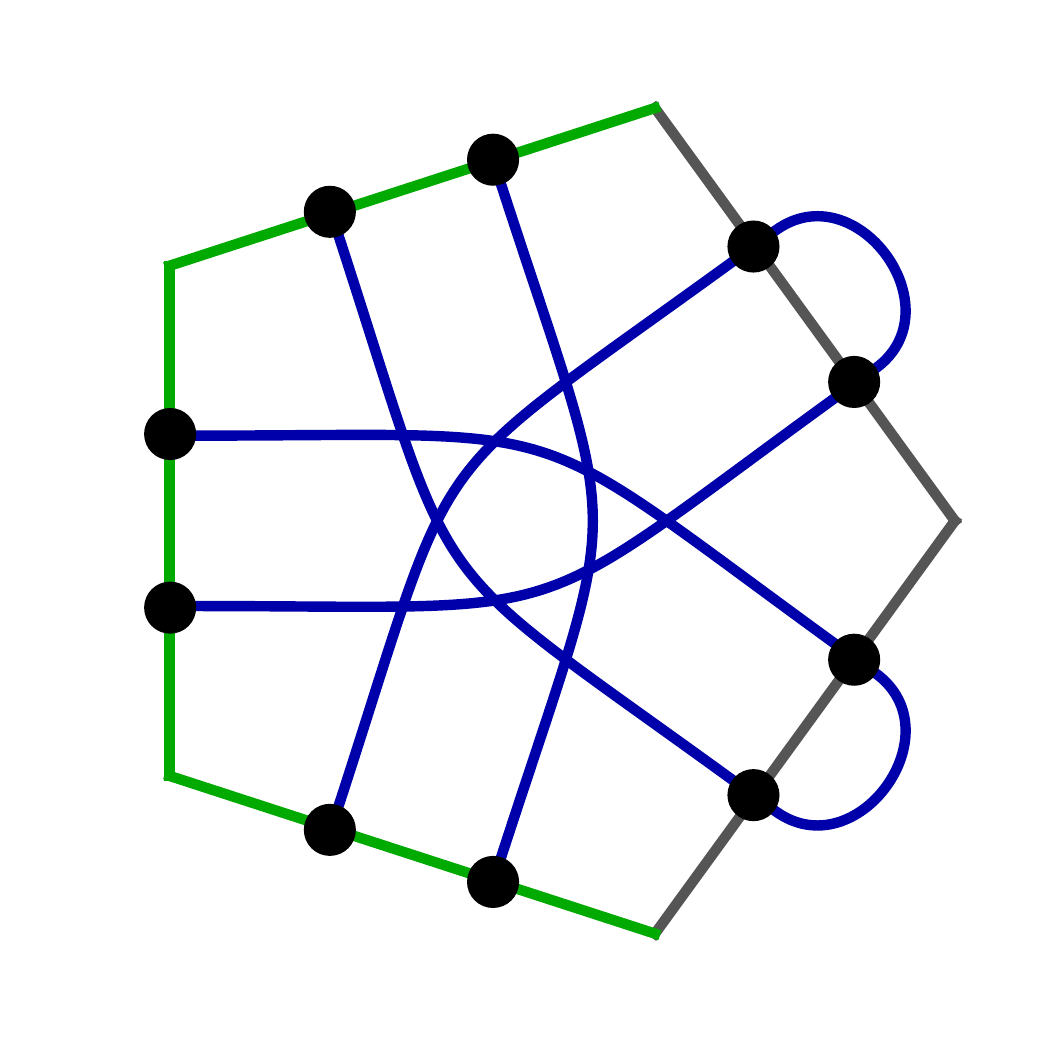}
\end{gathered}
\;+\;
\beta
\begin{gathered}
\includegraphics[height=0.1\textheight]{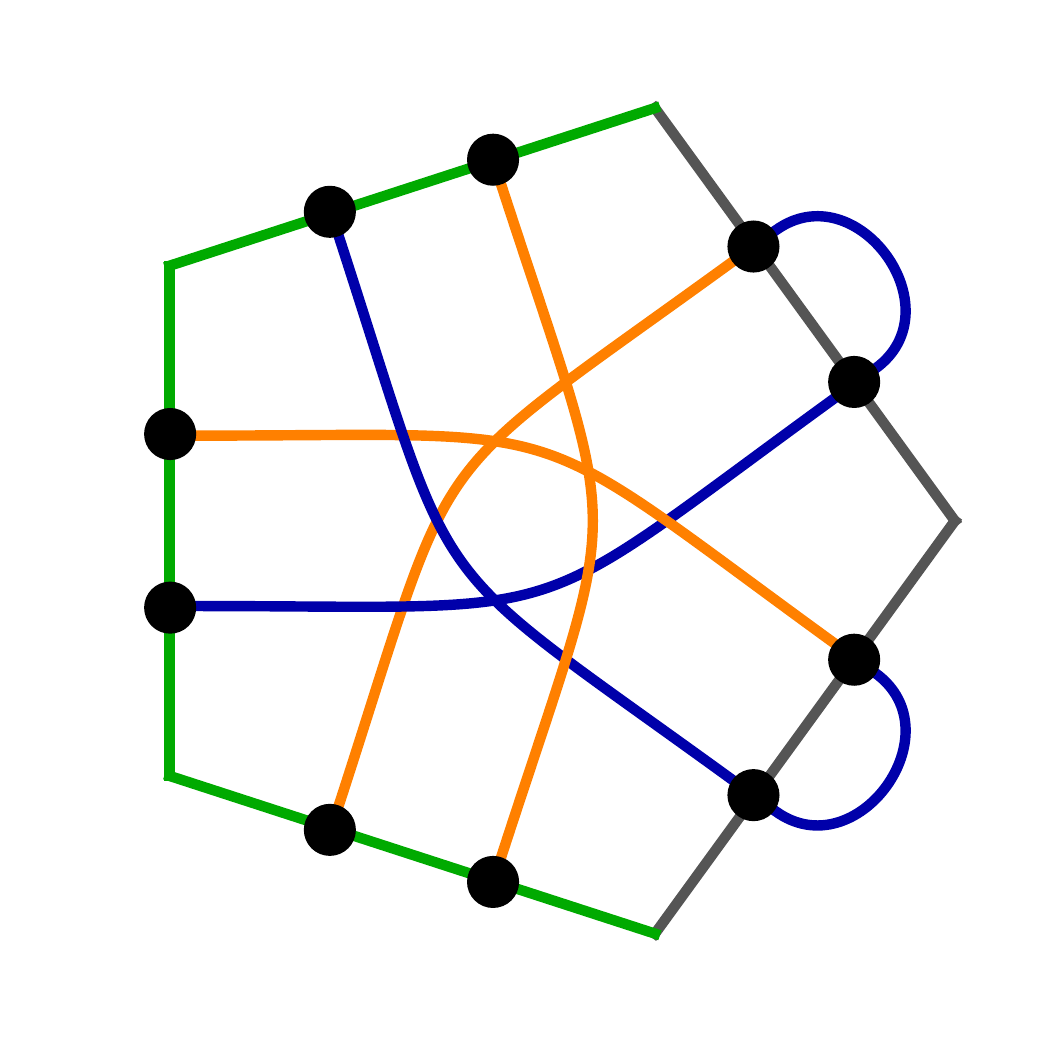}
\end{gathered}
\;=\;
\frac{\alpha}{2}
\begin{gathered}
\includegraphics[height=0.1\textheight]{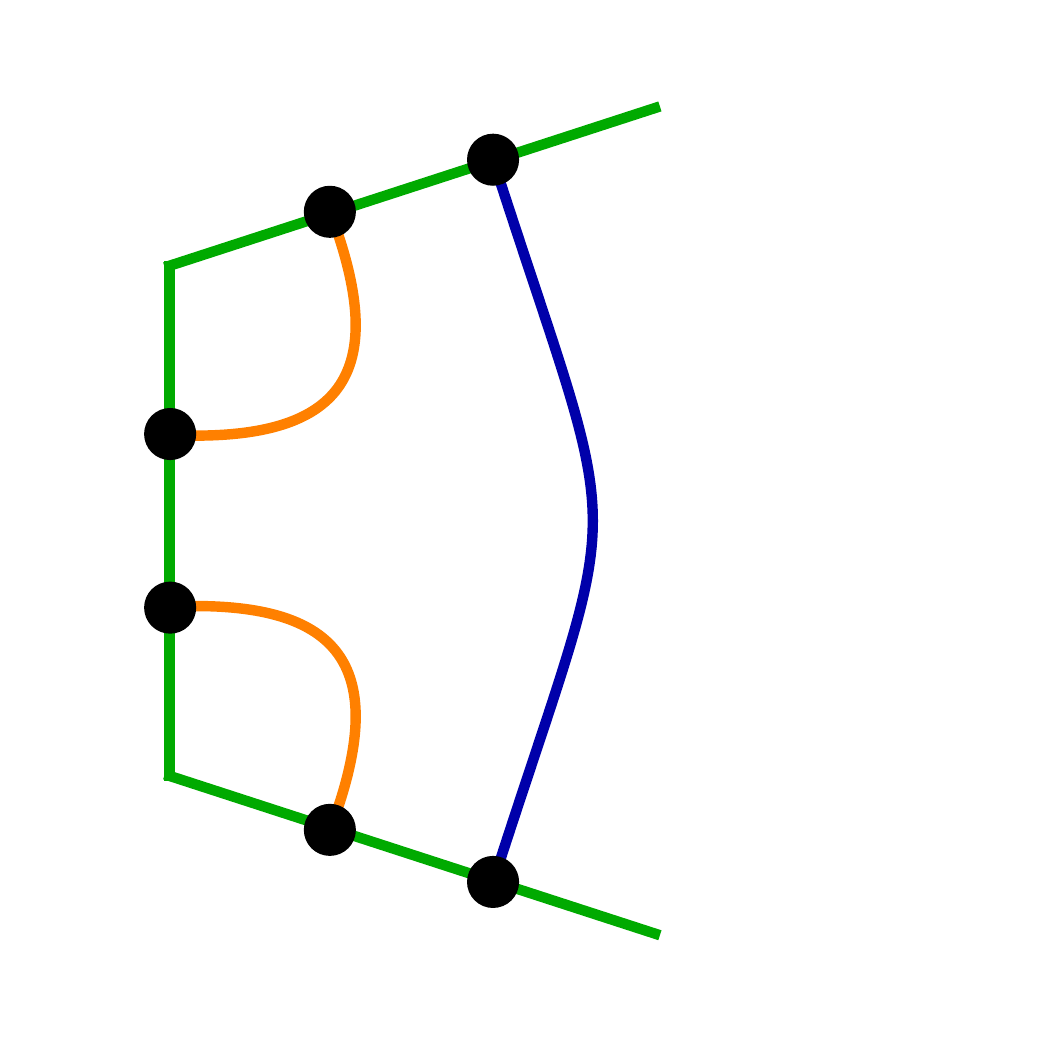}
\end{gathered}
+\;
\frac{\beta}{2}
\begin{gathered}
\includegraphics[height=0.1\textheight]{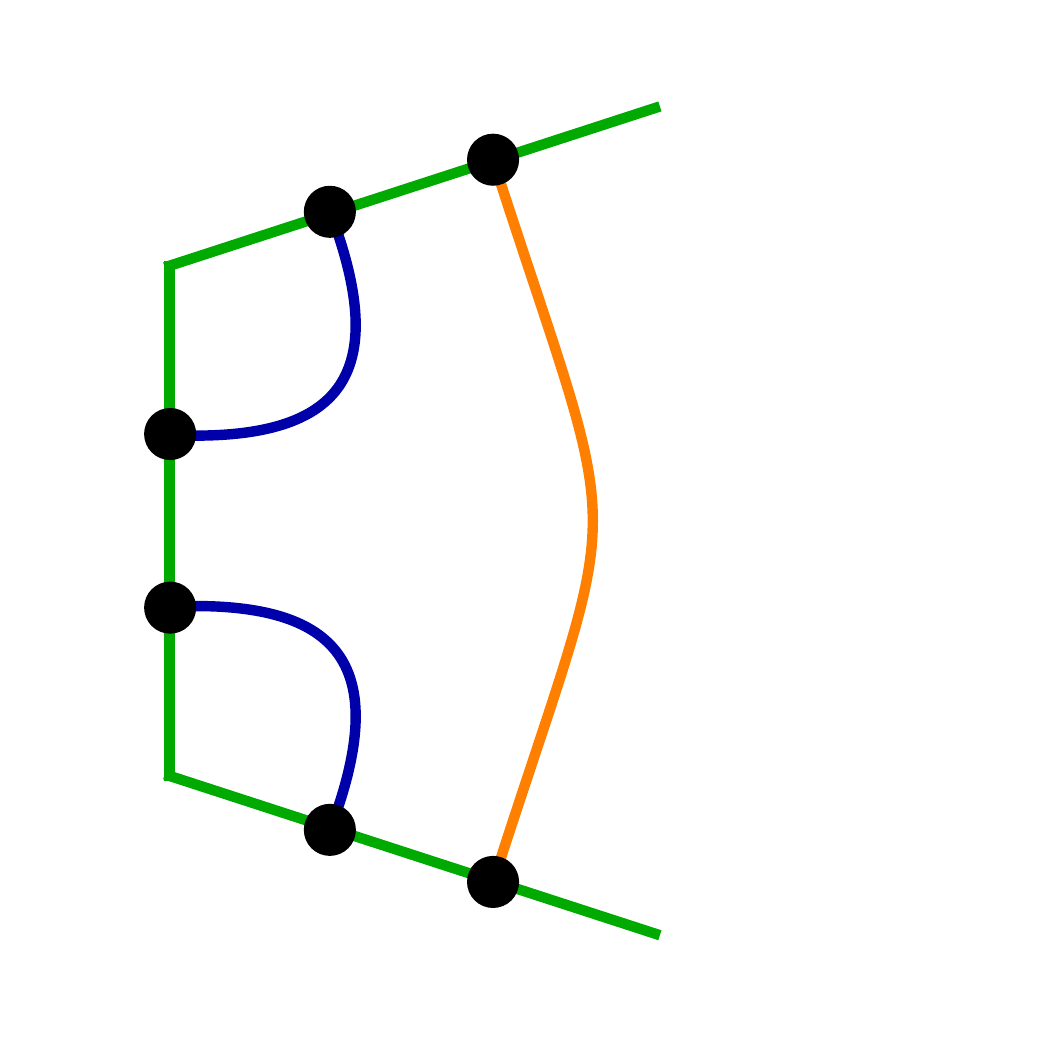}
\end{gathered}
\hspace{-1cm}
\end{align}
The remaining three eigenstates are given by:
\begin{align}
\label{EQ_RHO_A_513_EV2}
|\psi_{\textcolor{darkgreen}{A}}^{0,1}\rangle \;&=
\begin{gathered}
\includegraphics[height=0.1\textheight]{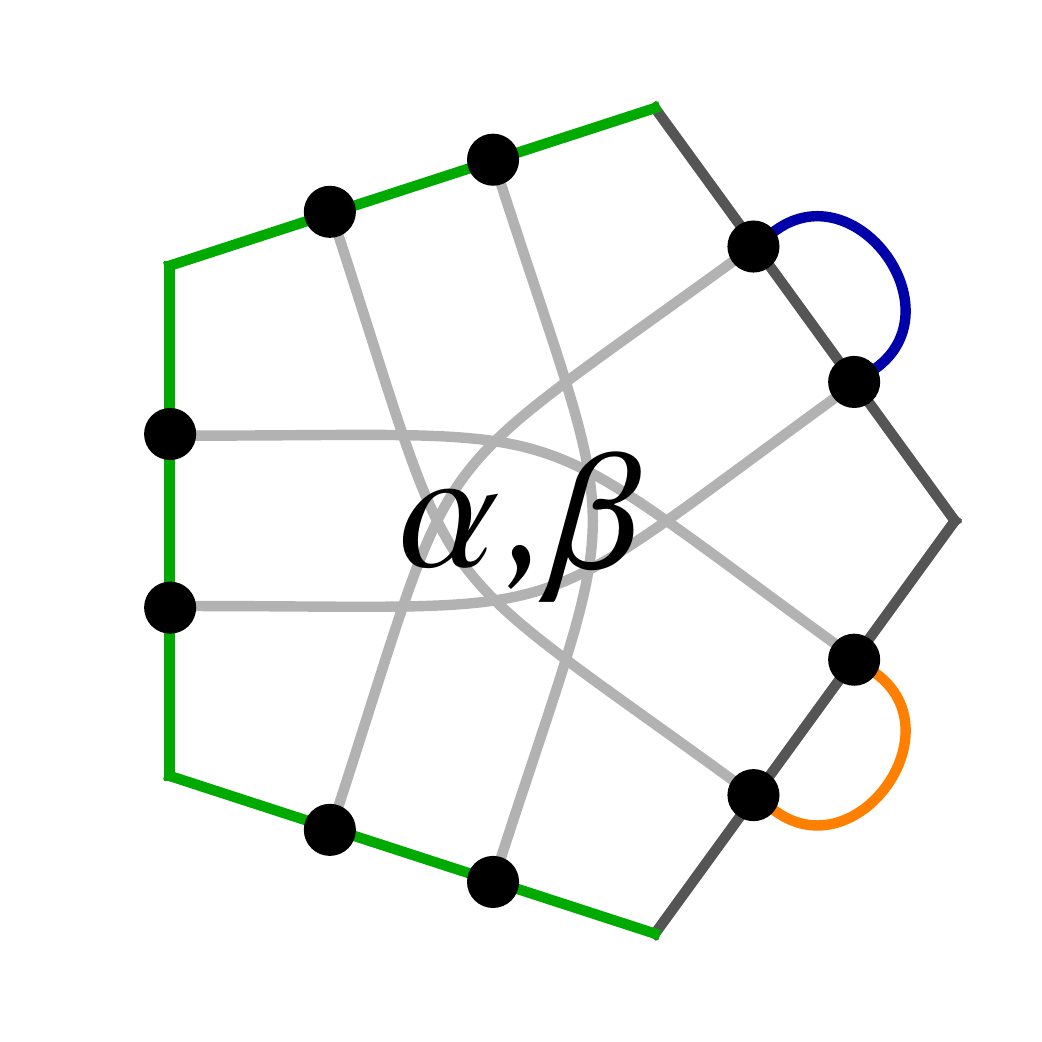}
\end{gathered}&
|\psi_{\textcolor{darkgreen}{A}}^{1,0}\rangle \;&=
\begin{gathered}
\includegraphics[height=0.1\textheight]{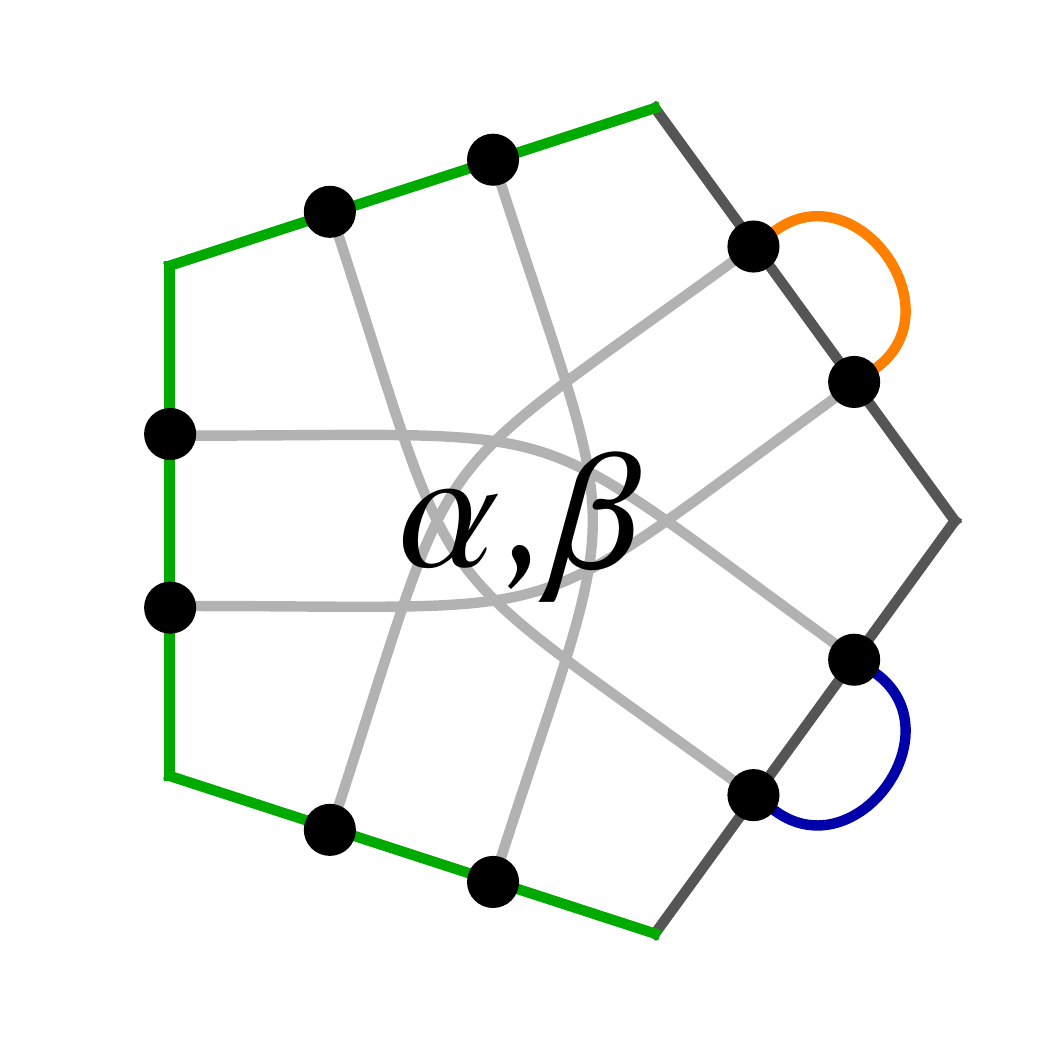}
\end{gathered}&
|\psi_{\textcolor{darkgreen}{A}}^{1,1}\rangle \;&=
\begin{gathered}
\includegraphics[height=0.1\textheight]{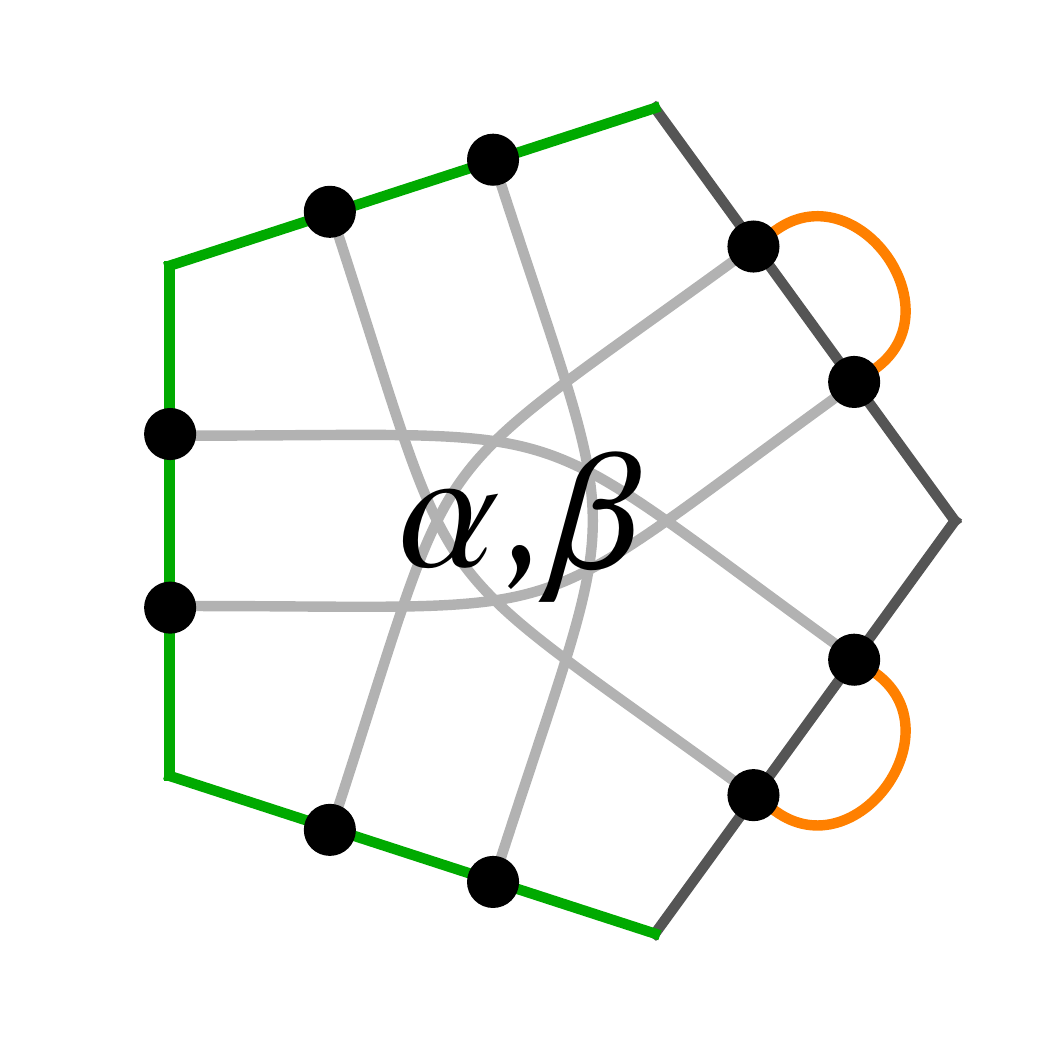}
\end{gathered}
\end{align}
To see that these are eigenstates, we do not need to actually evaluate these contractions. Instead, using \eqref{EQ_HAPPY_SP_3C}, we compute the eigenvalue equation for the first eigenvector as follows:
\begin{align}
\label{EQ_RHO_A_513}
\rho_{\textcolor{darkgreen}{A}}\, |\psi_{\textcolor{darkgreen}{A}}^{0,0}\rangle \;=
\begin{gathered}
\includegraphics[height=0.117\textheight]{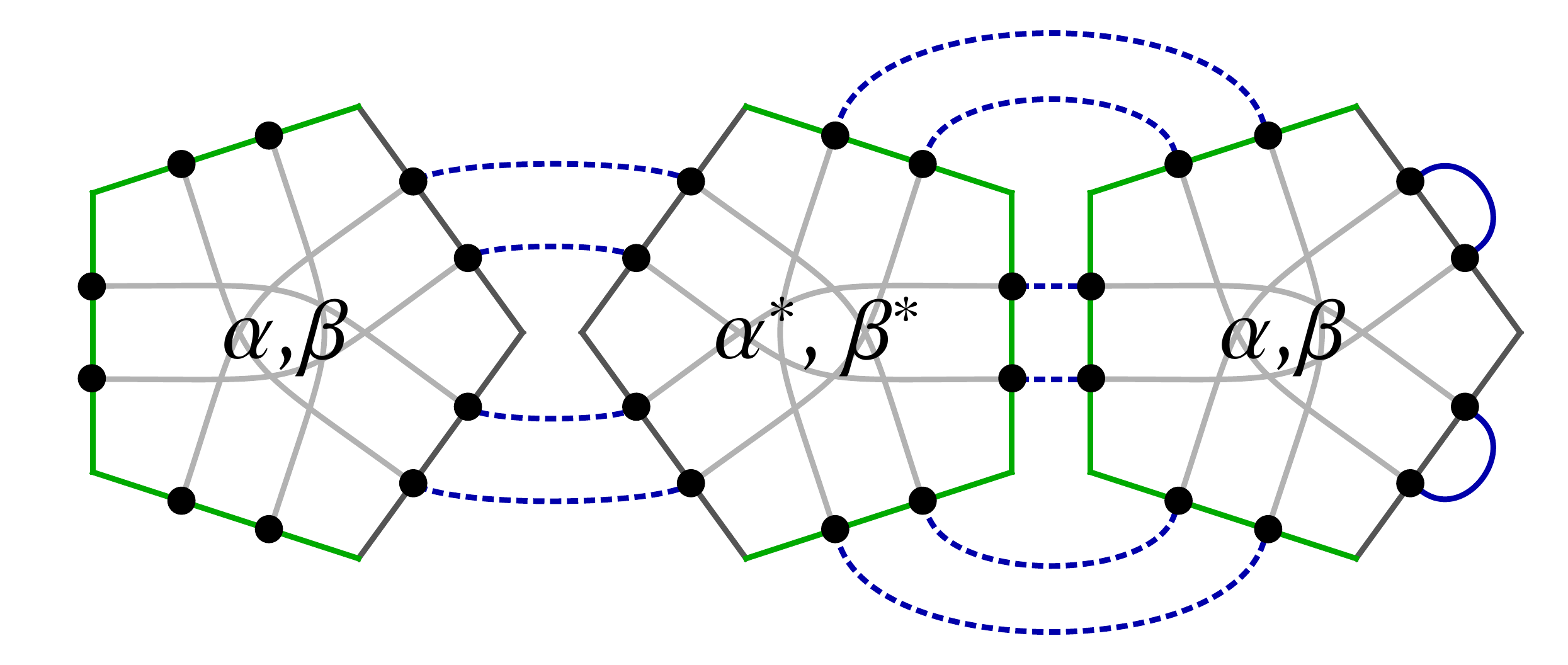}
\end{gathered}
\;&= \;\frac{1}{2}
\begin{gathered}
\includegraphics[height=0.117\textheight]{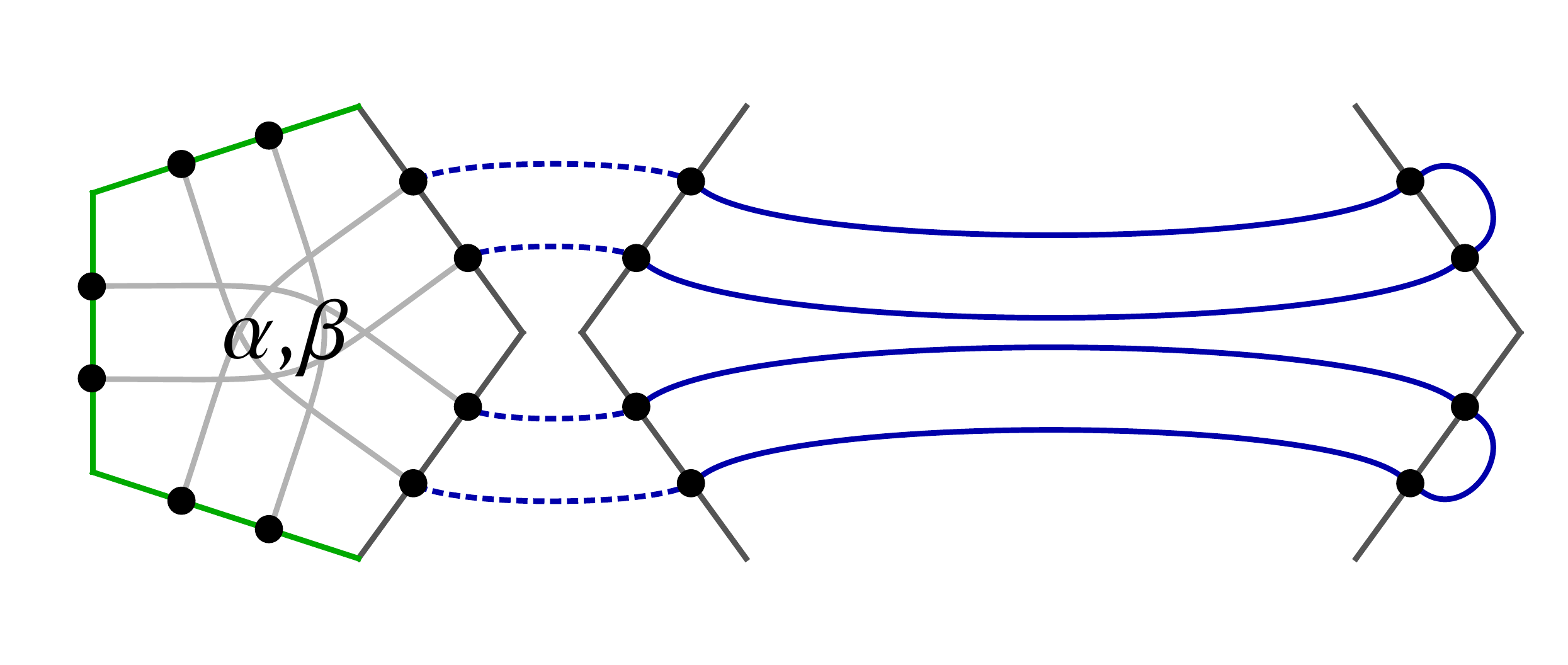}
\end{gathered} \nonumber\\
&= \;\frac{1}{4}
\begin{gathered}
\includegraphics[height=0.1\textheight]{pentagon_net_rhoA_happy3_00.pdf}
\end{gathered}
\;=\; \frac{1}{4}\; |\psi_{\textcolor{darkgreen}{A}}^{0,0}\rangle\ .
\end{align}
The equations for the other eigenstates follow equivalently, leading to an entanglement entropy $S_A=2 \log 2$ (i.e.,\ $m=2$).
For more than one tile, we can generalize \eqref{EQ_RHO_A_513} for \textsl{local} superpositions, i.e.,\ superpositions that factorize along the tiles. As an example, consider a $\ket{\psi^\prime}$ resulting from contracting two states of the form \eqref{EQ_HAPPY_SUPERPOS}:
\begin{align}
\label{EQ_HAPPY_TWO_SUPERPOS}
\ket{\psi^\prime} =\;& C_{1\leftrightarrow 4} \left( \alpha_1 \ket{\bar{0}}_{5} + \beta_1\ket{\bar{1}}_{5} \right) \left( \alpha_2 \ket{\bar{0}}_{5} + \beta_2 \ket{\bar{1}}_{5} \right)
\;=\;
\begin{gathered}
\includegraphics[height=0.117\textheight]{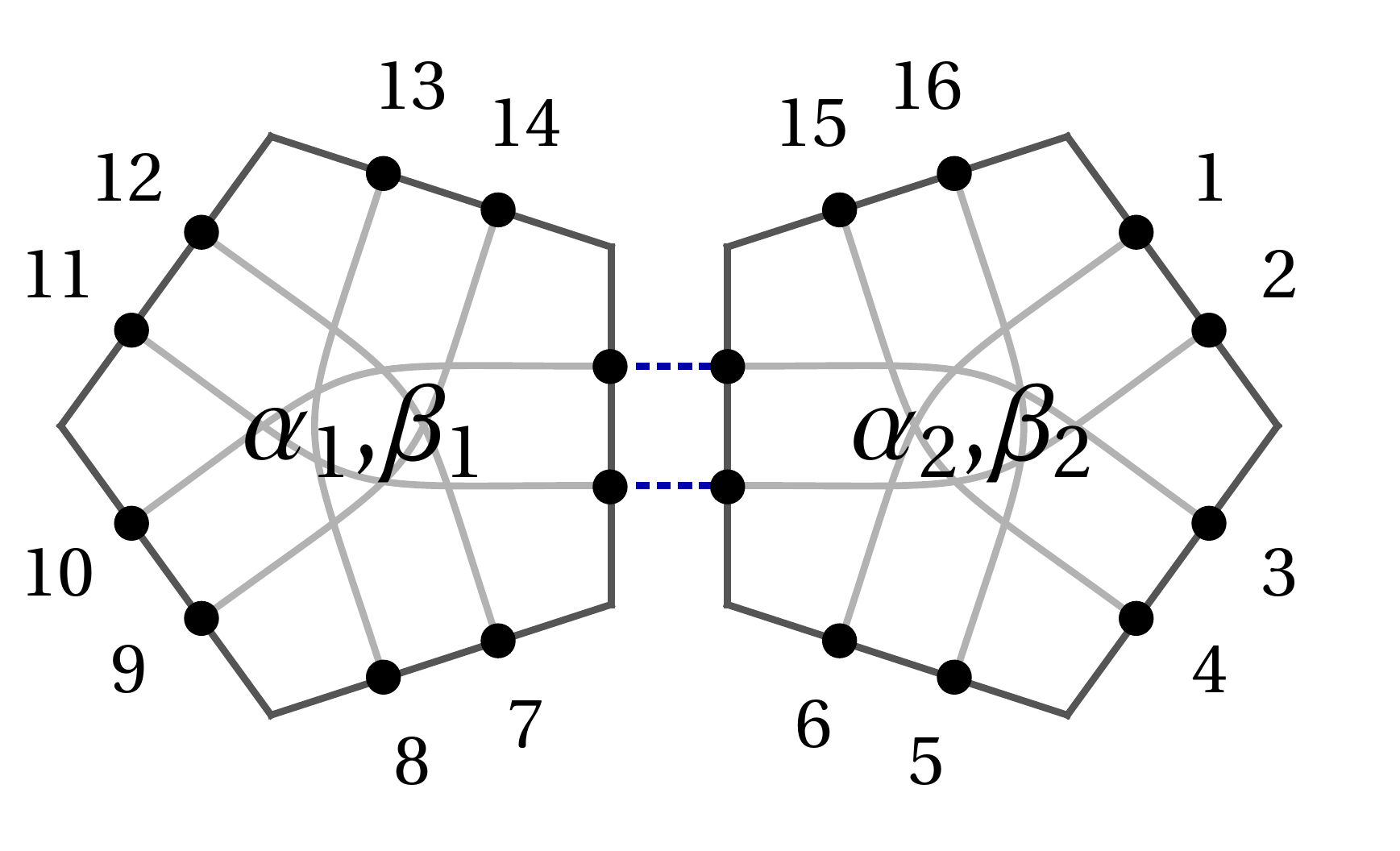}
\end{gathered}
\end{align}
Here we have defined the contraction operator $C_{j\leftrightarrow k}$ contracting the $j$th edge of the first dimer state on the $k$th edge of the second. 
We can now show that \eqref{EQ_RHO_A_513} generalizes if we extend region $A\to A^\prime$ onto a neighbouring pentagon tile. The reduced density matrix becomes
\begin{align}
\rho^\prime_{\textcolor{darkgreen}{A^\prime}} = 2\,
\begin{gathered}
\includegraphics[height=0.117\textheight]{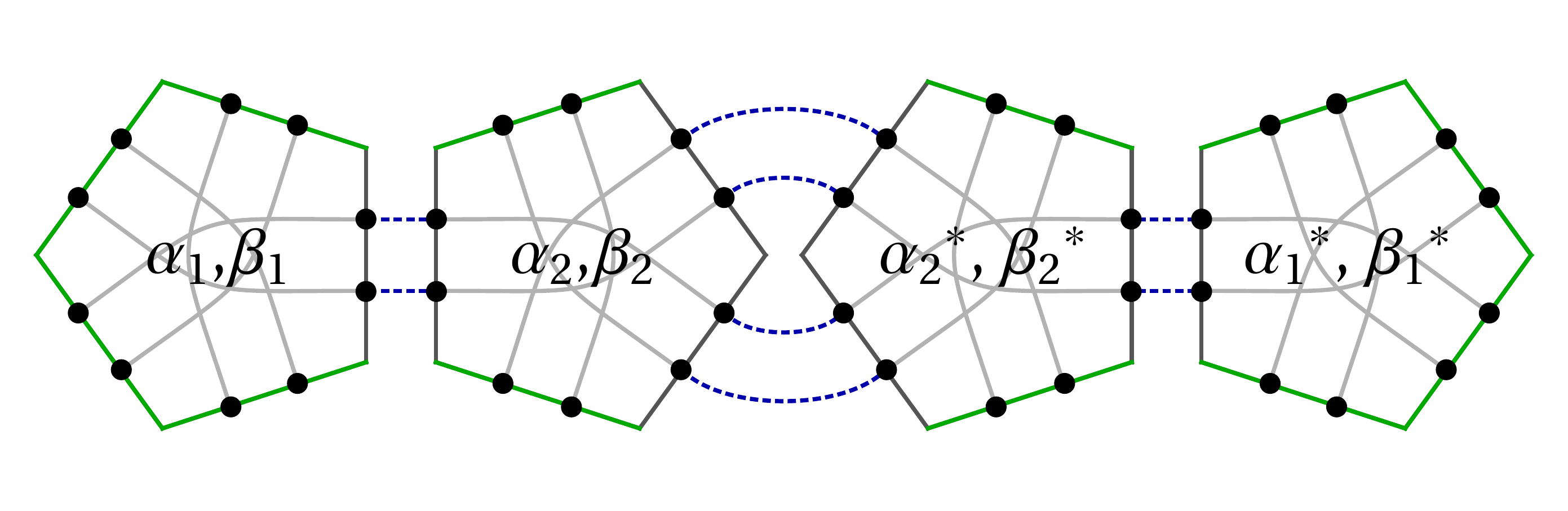}
\end{gathered}
\end{align}
Note that a normalization factor of $2$ appears as a result of the unresolved contraction within both $\ket{\psi^\prime}$ and $\bra{\psi^\prime}$.
To see that the eigenvalue spectrum of $\rho^\prime_{A^\prime}$ is the same as that of $\rho_A$, we simply extend the eigenvectors \eqref{EQ_RHO_A_513_EV1} and \eqref{EQ_RHO_A_513_EV2} onto the region $A^\prime$ by contracting them with the first pentagon, which is equivalent to contracting a complete basis on the extended vector $\ket{\psi^\prime}$. For the first eigenvector, we thus find
\begin{align}
|{\psi^\prime}_{\textcolor{darkgreen}{A^\prime}}^{0,0}\rangle \;&=\;
\begin{gathered}
\includegraphics[height=0.117\textheight]{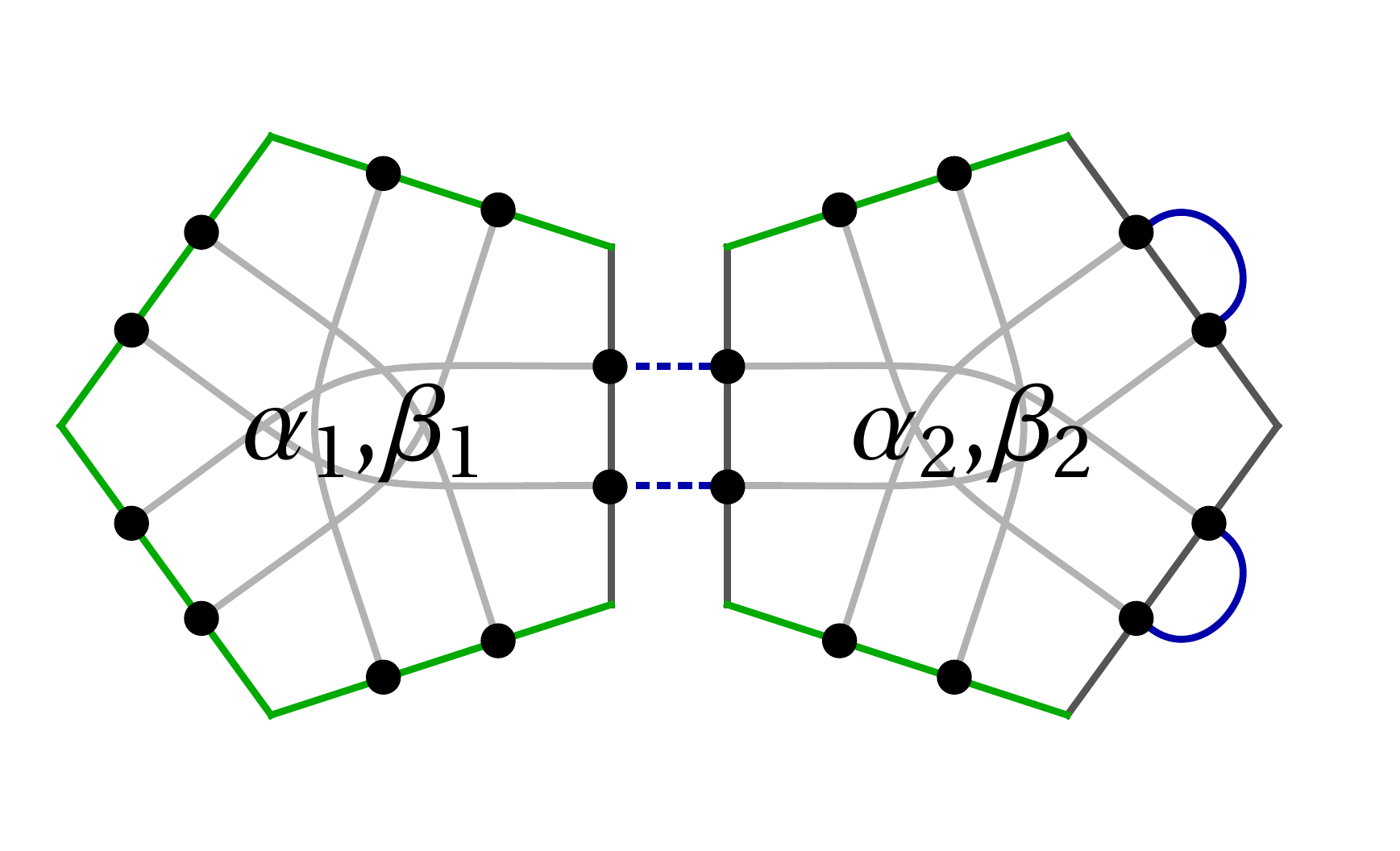}
\end{gathered}
\end{align}
The complete eigenvalue equation can be resolved by applying \eqref{EQ_HAPPY_SP_3C} and \eqref{EQ_HAPPY_SP_4C} successively:
\begin{align}
\label{EQ_HAPPY_RHO_A_RED1}
\rho^\prime_{\textcolor{darkgreen}{A^\prime}}\,  
|{\psi^\prime}_{\textcolor{darkgreen}{A^\prime}}^{0,0}\rangle
&=\; 2\,
\begin{gathered}
\includegraphics[height=0.146\textheight]{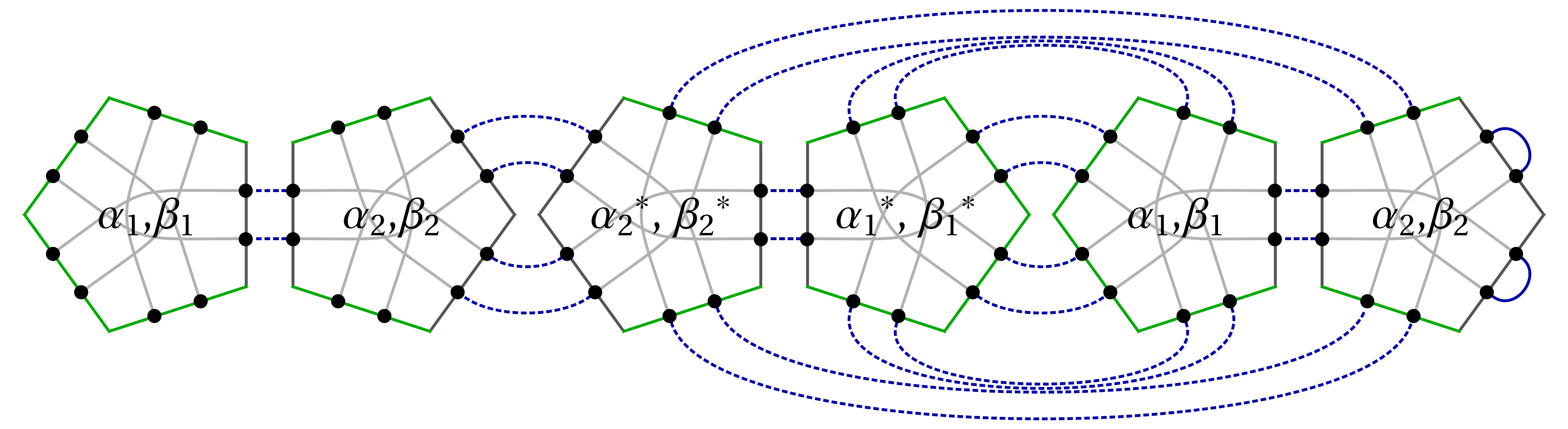}
\end{gathered} \nonumber\\
&=\; \sqrt{2}
\begin{gathered}
\includegraphics[height=0.146\textheight]{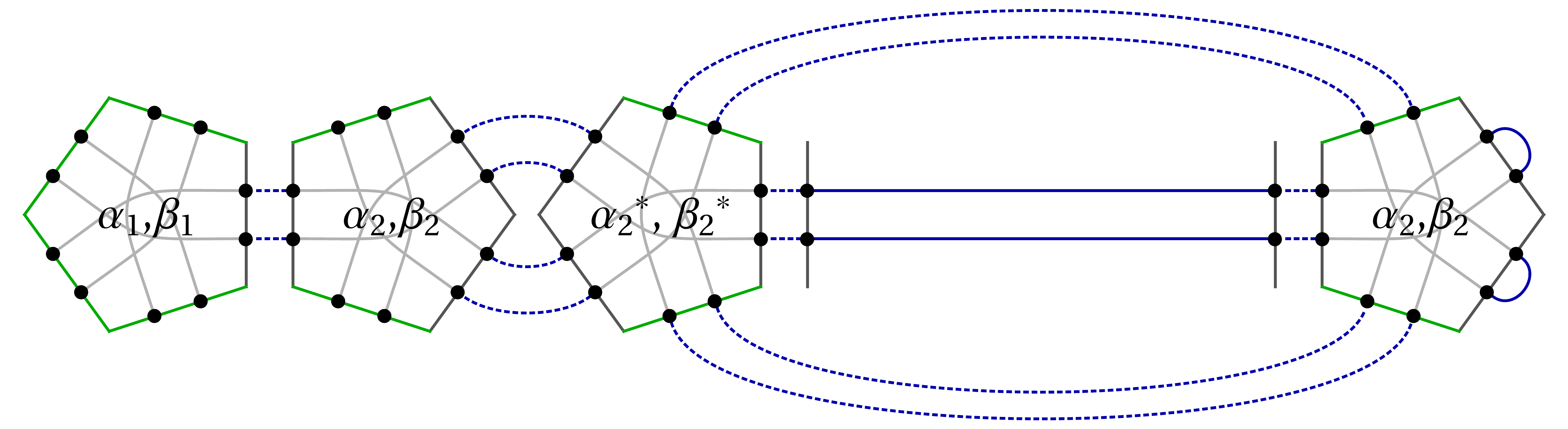}
\end{gathered} \nonumber\\
&=\;
\begin{gathered}
\includegraphics[height=0.117\textheight]{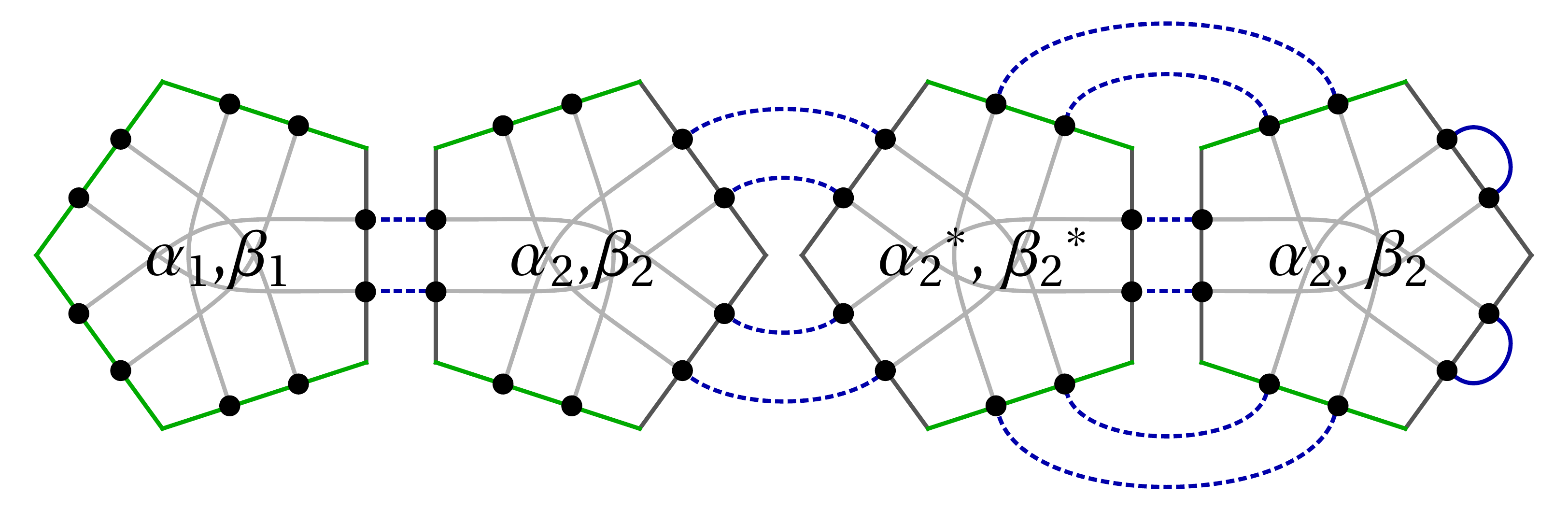}
\end{gathered} \nonumber\\
&=\; \frac{1}{2}
\begin{gathered}
\includegraphics[height=0.117\textheight]{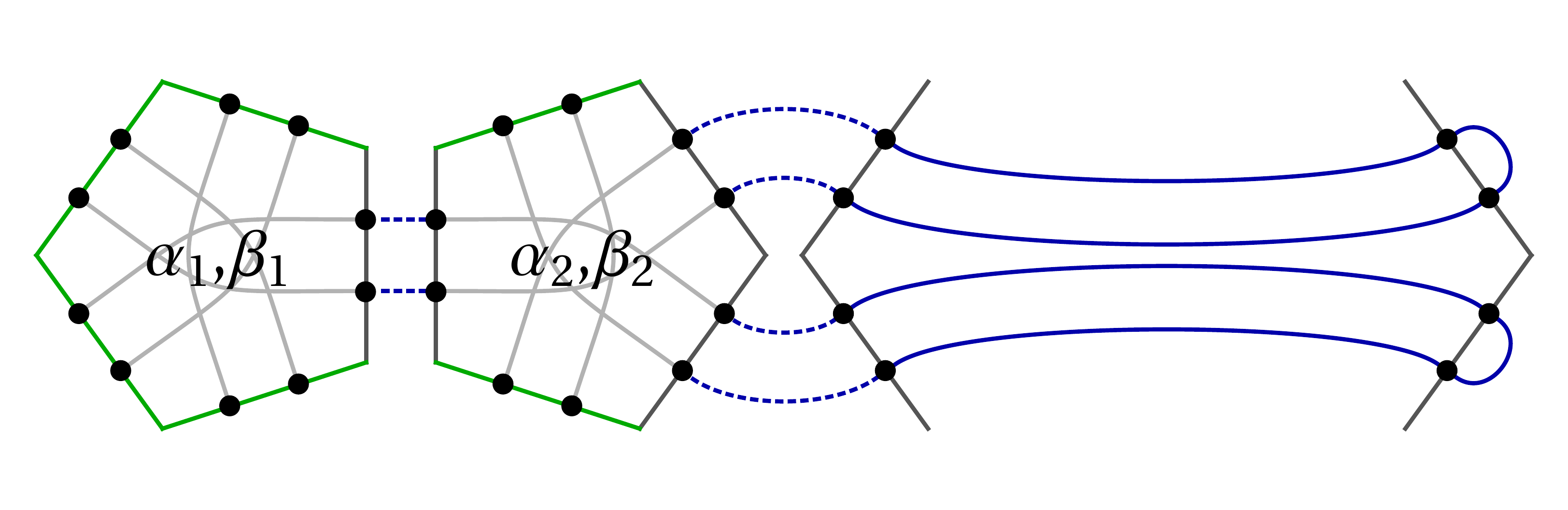}
\end{gathered} \nonumber\\
&=\; \frac{1}{4}
\begin{gathered}
\includegraphics[height=0.117\textheight]{pentagon_net_happy_sp2.pdf}
\end{gathered}
\;=\; \frac{1}{4}\; |{\psi^\prime}_{\textcolor{darkgreen}{A^\prime}}^{0,0}\rangle
\end{align}
Again, this procedure holds for all eigenstates, leading to the same eigenvalue spectrum as for $\rho_A$. Thus we see that ``gluing'' $[[5,1,3]]$ tiles onto a region $A$ on an original tile only projects the eigenvalues onto a larger space of Majorana dimer states, leaving their eigenvalues invariant. This procedure can also be extended to cases where a subsystem $B$ and its complement $B^{\text{C}}$ both cover different tiles, as in the following example:
\begin{align}
\rho_{\textcolor{darkgreen}{B}}\; =\; 2\,
\begin{gathered}
\includegraphics[height=0.191\textheight]{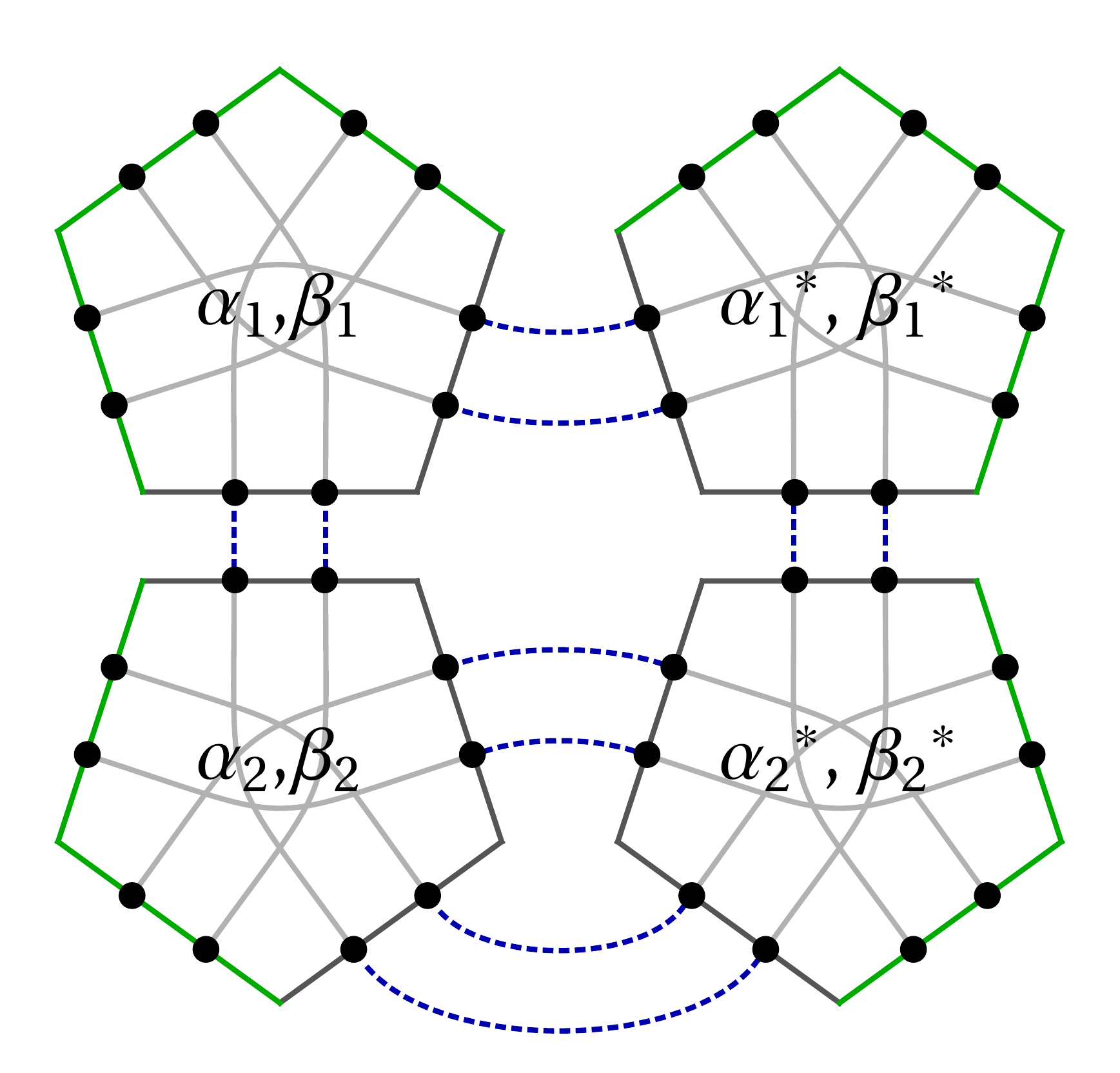}
\end{gathered}
\end{align}
Here, we have rotated the configuration \eqref{EQ_HAPPY_TWO_SUPERPOS} for easier visualization; as before, the adjoint part of $\rho_{B}^\prime$ is on the right. Even in this configuration, we can construct a set of eigenvectors by projecting a complete dimer basis onto $B^{\text{C}}$:
\begin{align}
|\psi_{\textcolor{darkgreen}{B}}^{0,0,0}\rangle \;&=\;
\begin{gathered}
\includegraphics[height=0.192\textheight]{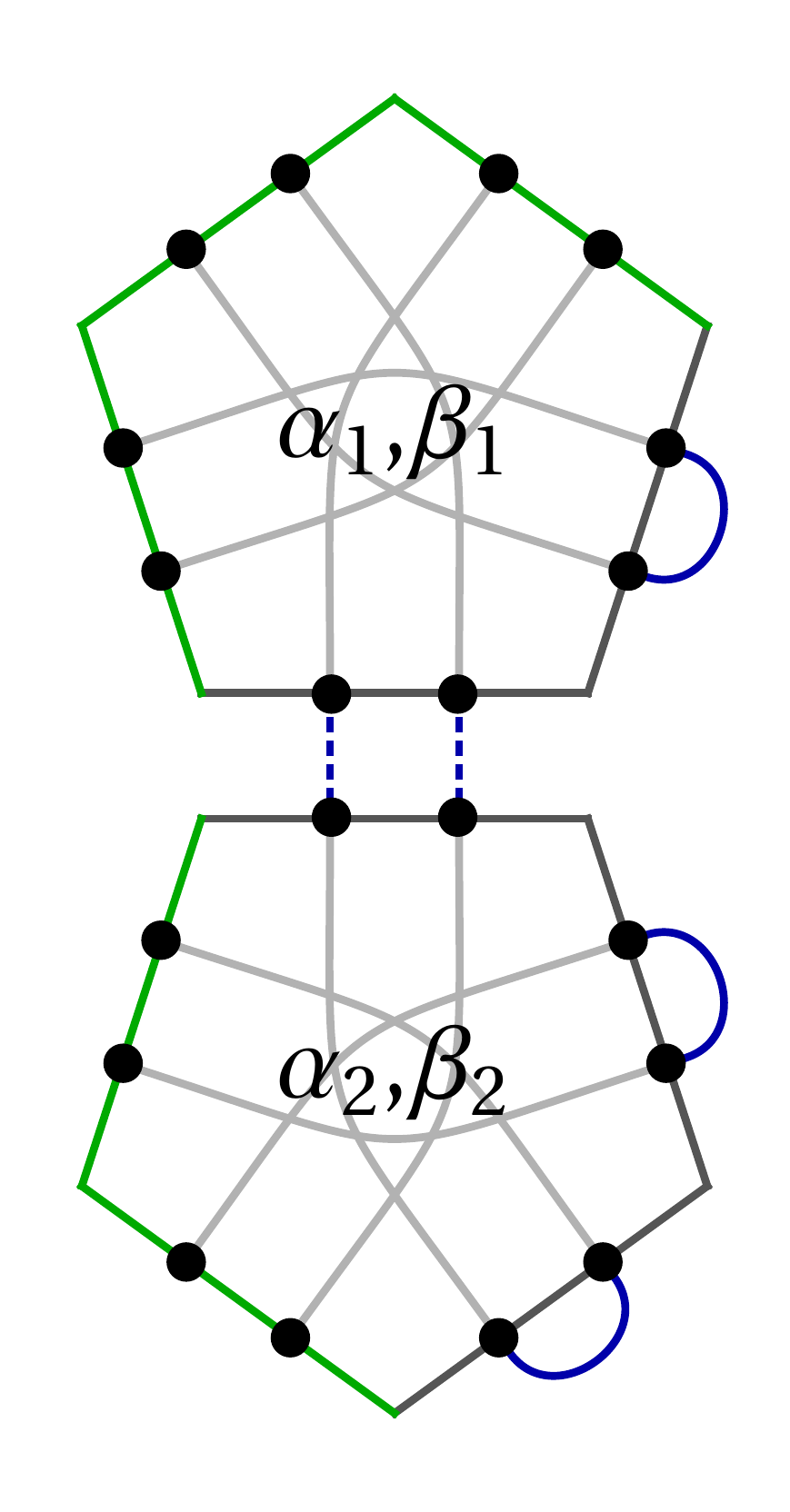}
\end{gathered}&
|\psi_{\textcolor{darkgreen}{B}}^{0,0,1}\rangle \;&=\;
\begin{gathered}
\includegraphics[height=0.192\textheight]{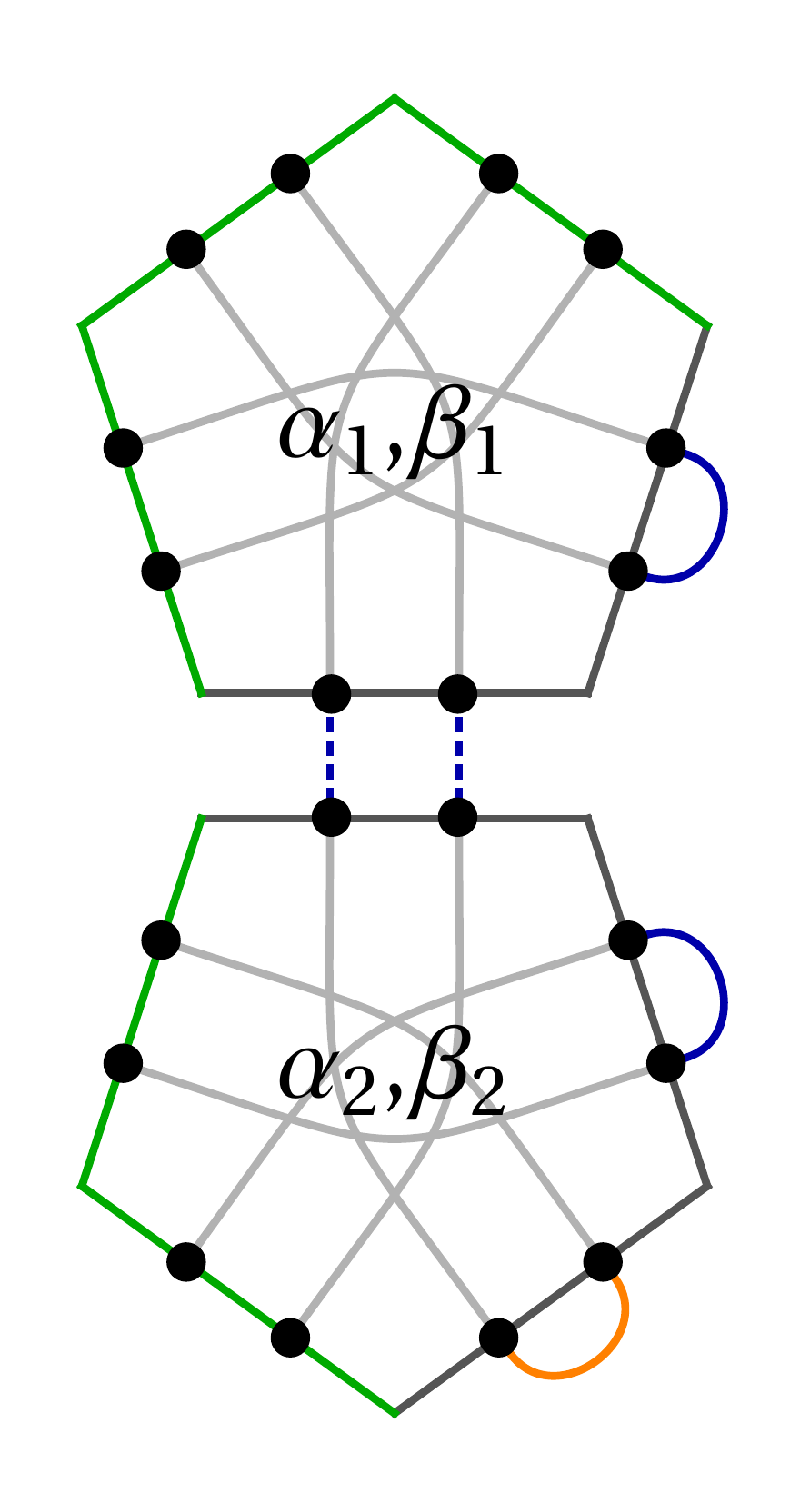}
\end{gathered}&
|\psi_{\textcolor{darkgreen}{B}}^{0,1,0}\rangle \;&=\;
\begin{gathered}
\includegraphics[height=0.192\textheight]{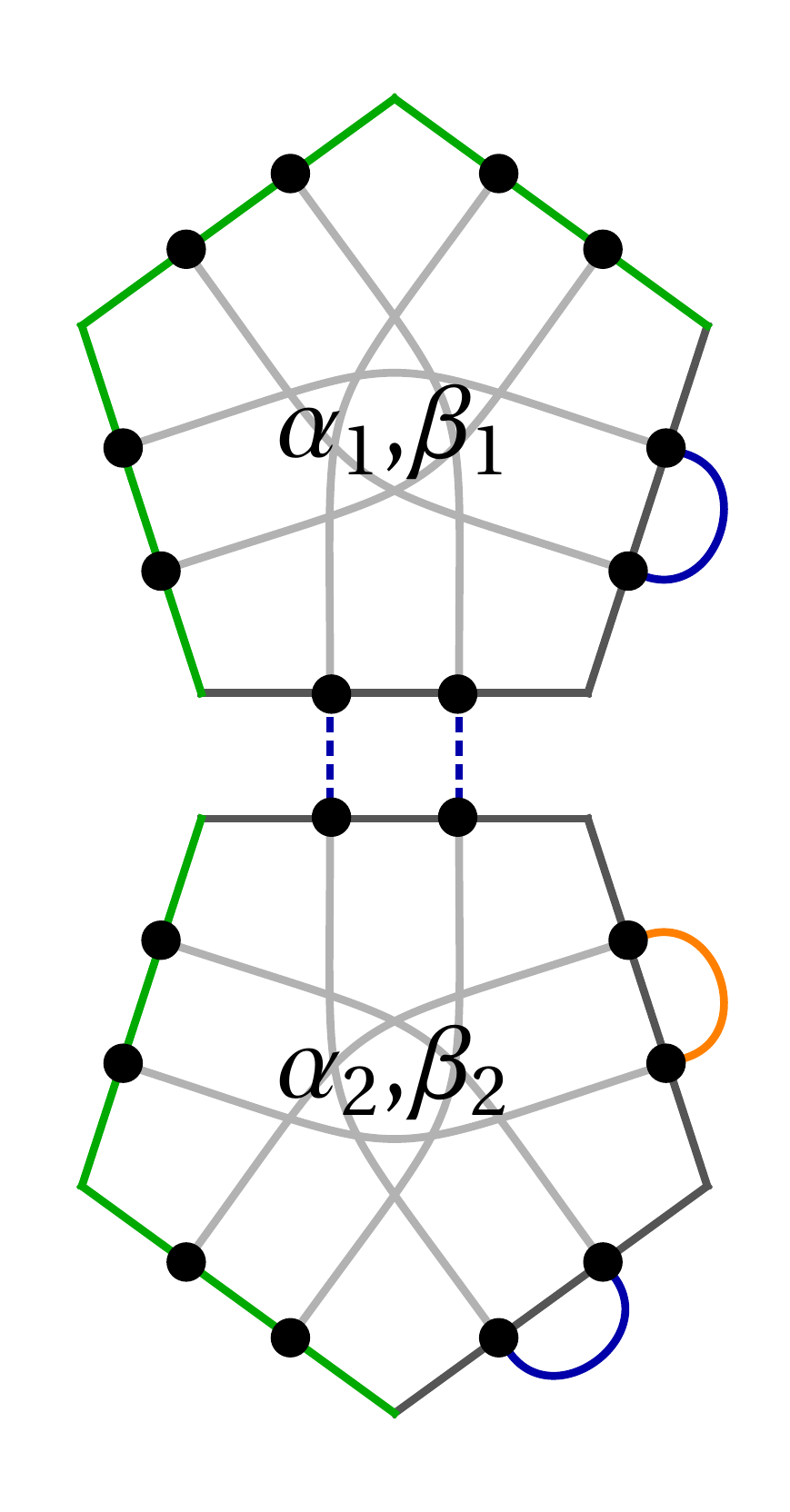}
\end{gathered}&
&\dots
\end{align}
Explicitly, the eigenvalue equation for $|\psi_{\textcolor{darkgreen}{B}}^{0,0,0}\rangle$ is given by
\begin{align}
\label{EQ_HAPPY_RHO_A_RED2}
\rho_{\textcolor{darkgreen}{B}}\, |\psi_{\textcolor{darkgreen}{B}}^{0,0,0}\rangle \; &= \; 2\;
\begin{gathered}
\includegraphics[height=0.191\textheight]{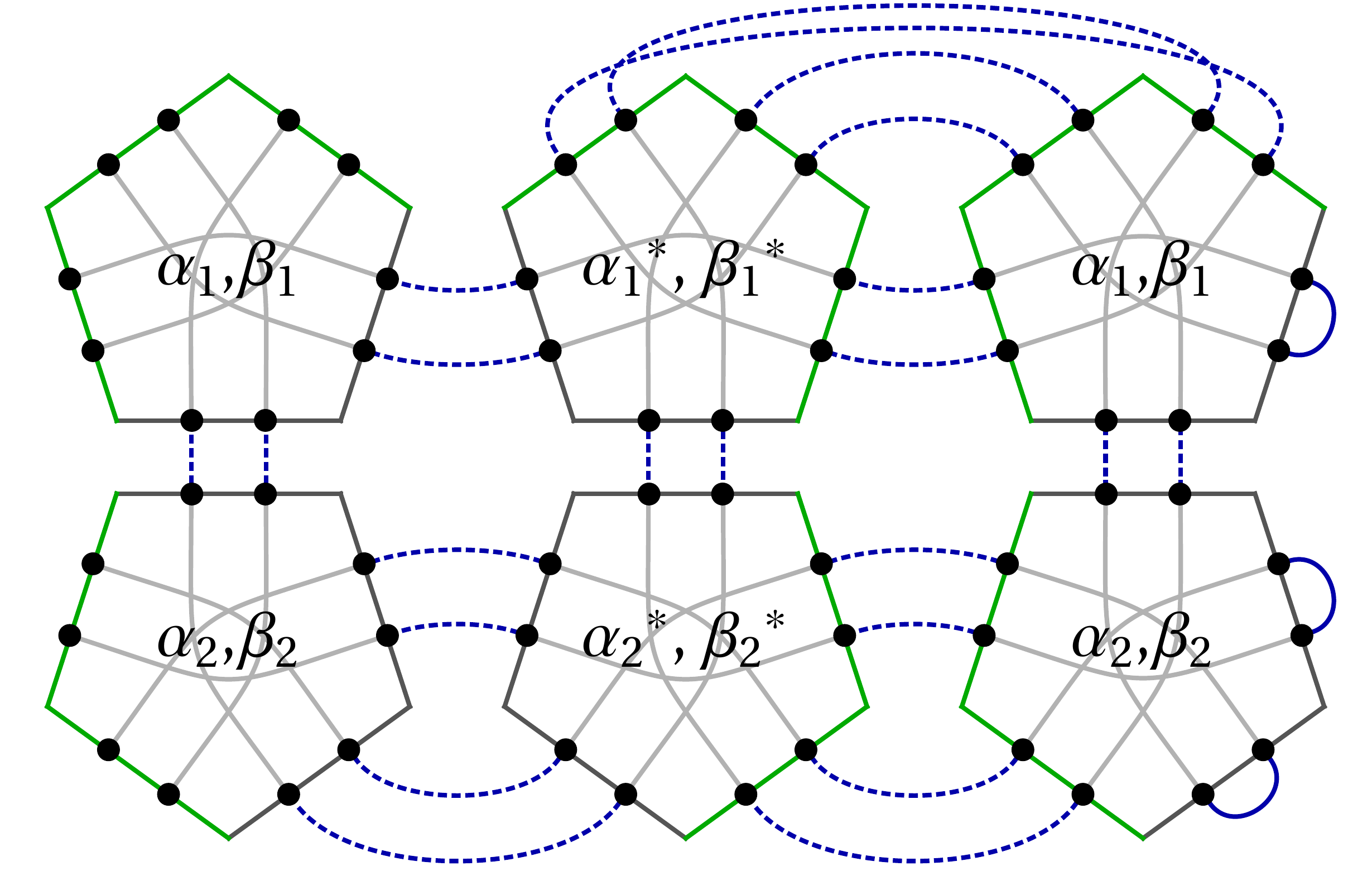}
\end{gathered} \nonumber\\
&=\; \frac{1}{2}
\begin{gathered}
\includegraphics[height=0.191\textheight]{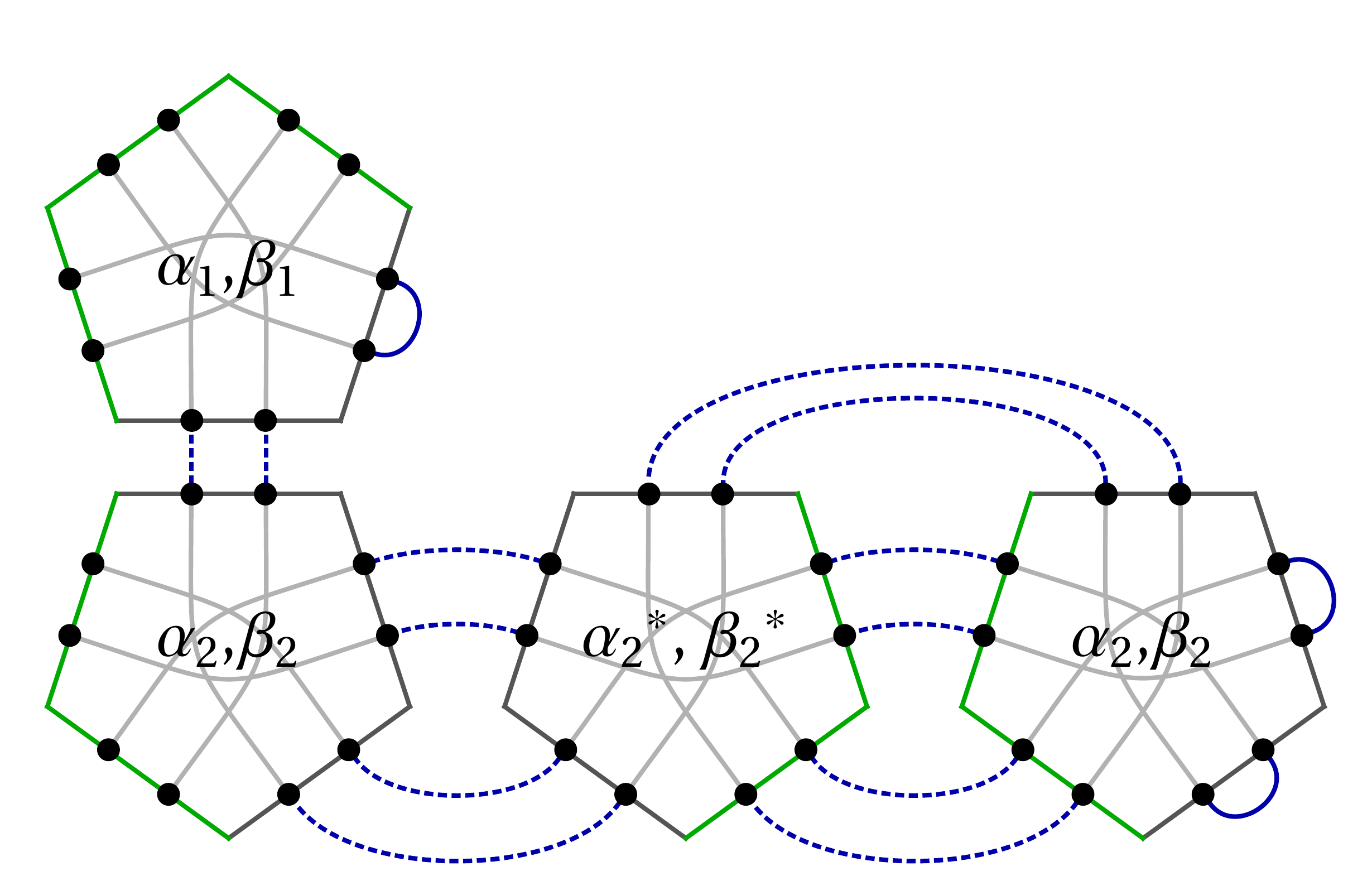}
\end{gathered}
\;=\; \frac{1}{8}
\begin{gathered}
\includegraphics[height=0.191\textheight]{pentagon_net_rhoA_7_000.pdf}
\end{gathered}
\end{align}
Again, after repeating this procedure for all eight eigenstates, we find that the entanglement entropy corresponds to the result for a fixed logical input, $S_B=3\log 2$. An important condition for computing reduced density matrix eigenstates in this way is that when projecting a complete basis of eigenvectors onto $\ket\psi$, the resulting states must be orthogonal. This is always the case when no dimers connect sites within region $A^{\text{C}}$. If they do, we can still simplify the reduced density matrix to an effective density matrix of a reduced state, as in the following example for a region $C$:
\begin{align}
\label{EQ_HAPPY_RHO_A_RED3}
\rho_{\textcolor{darkgreen}{C}} \; = \; 2\;
&\begin{gathered}
\includegraphics[height=0.2\textheight]{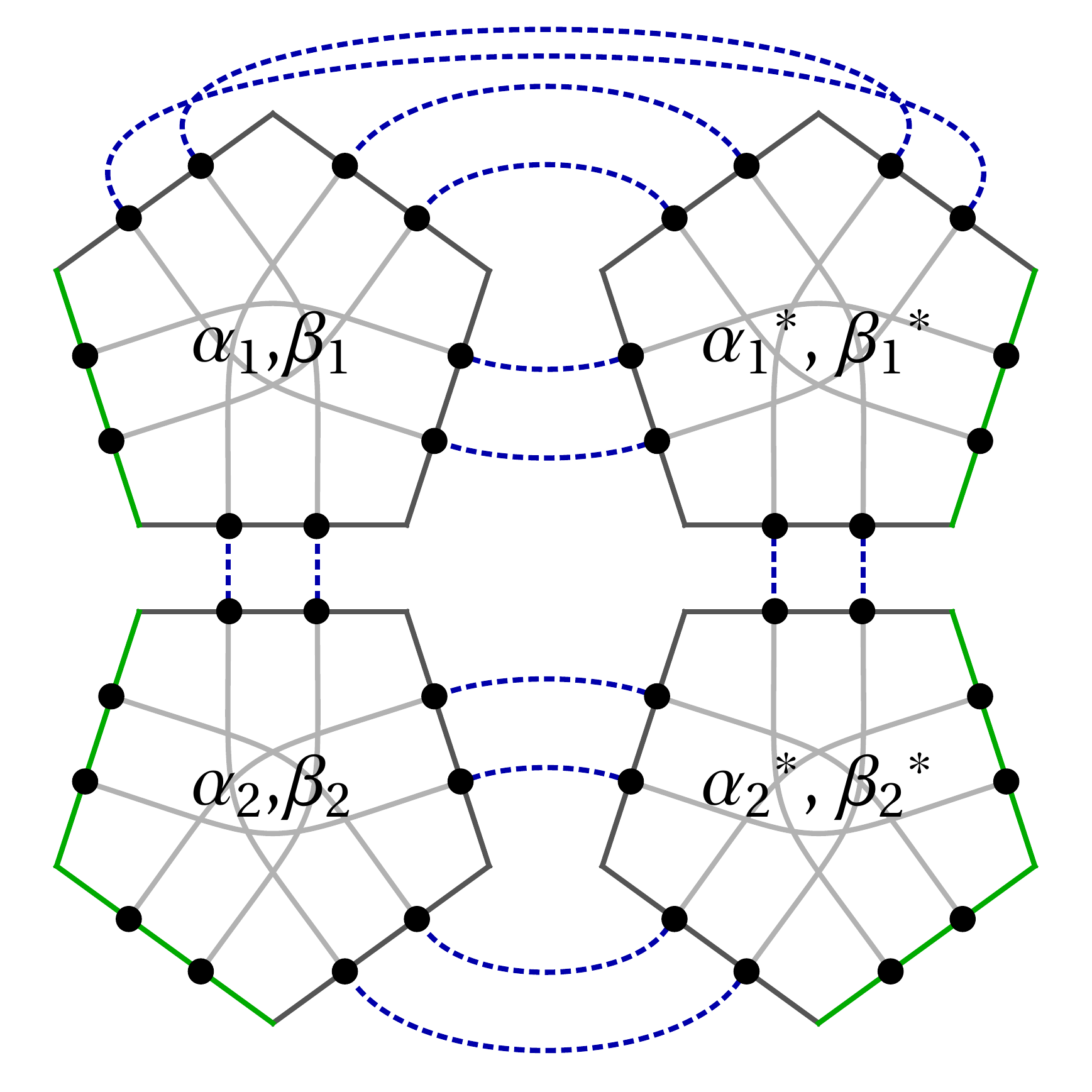}
\end{gathered}
\;=\; \frac{1}{\sqrt{2}}
\begin{gathered}
\includegraphics[height=0.2\textheight]{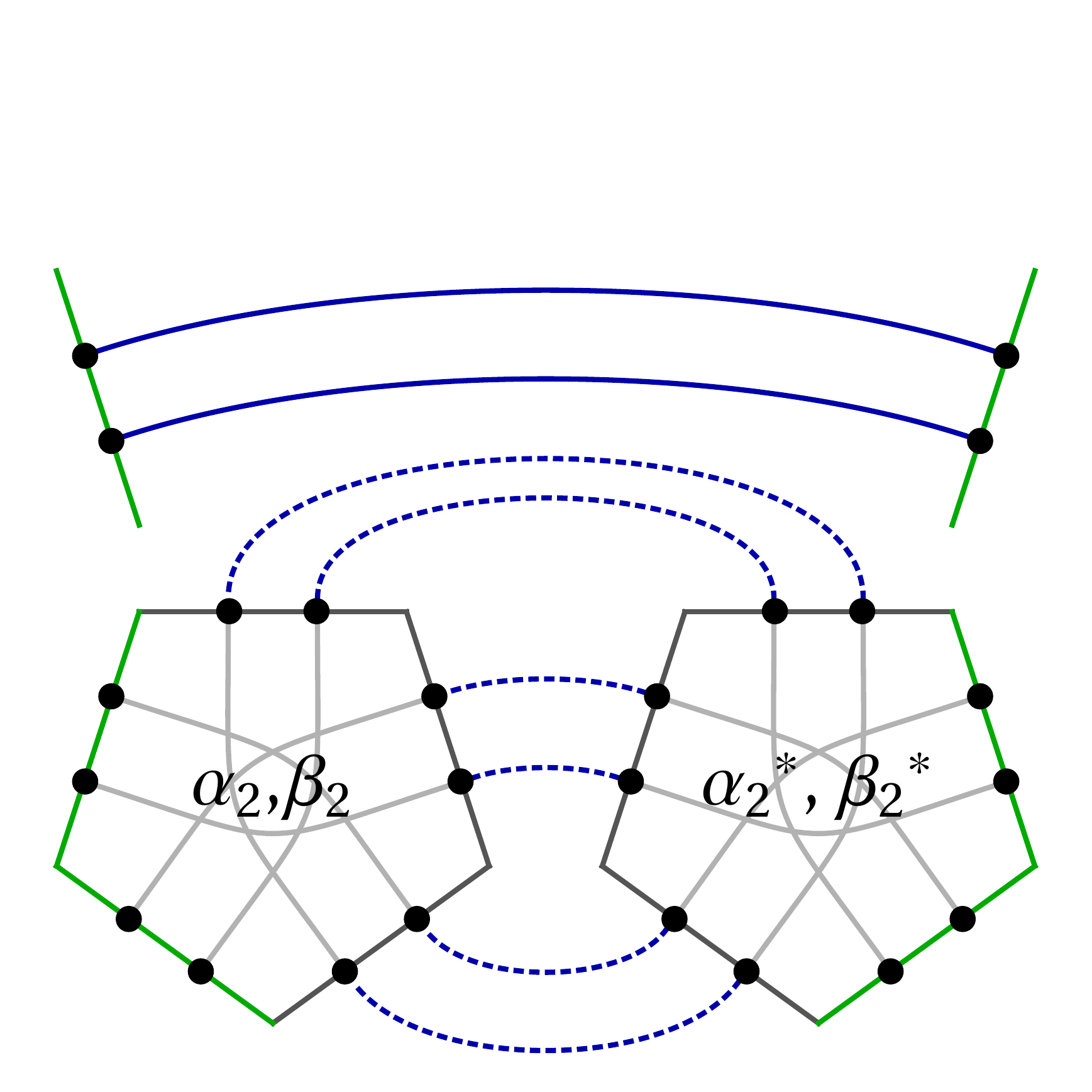}
\end{gathered} \\
=\; \frac{1}{2\sqrt{2}}
&\begin{gathered}
\includegraphics[height=0.2\textheight]{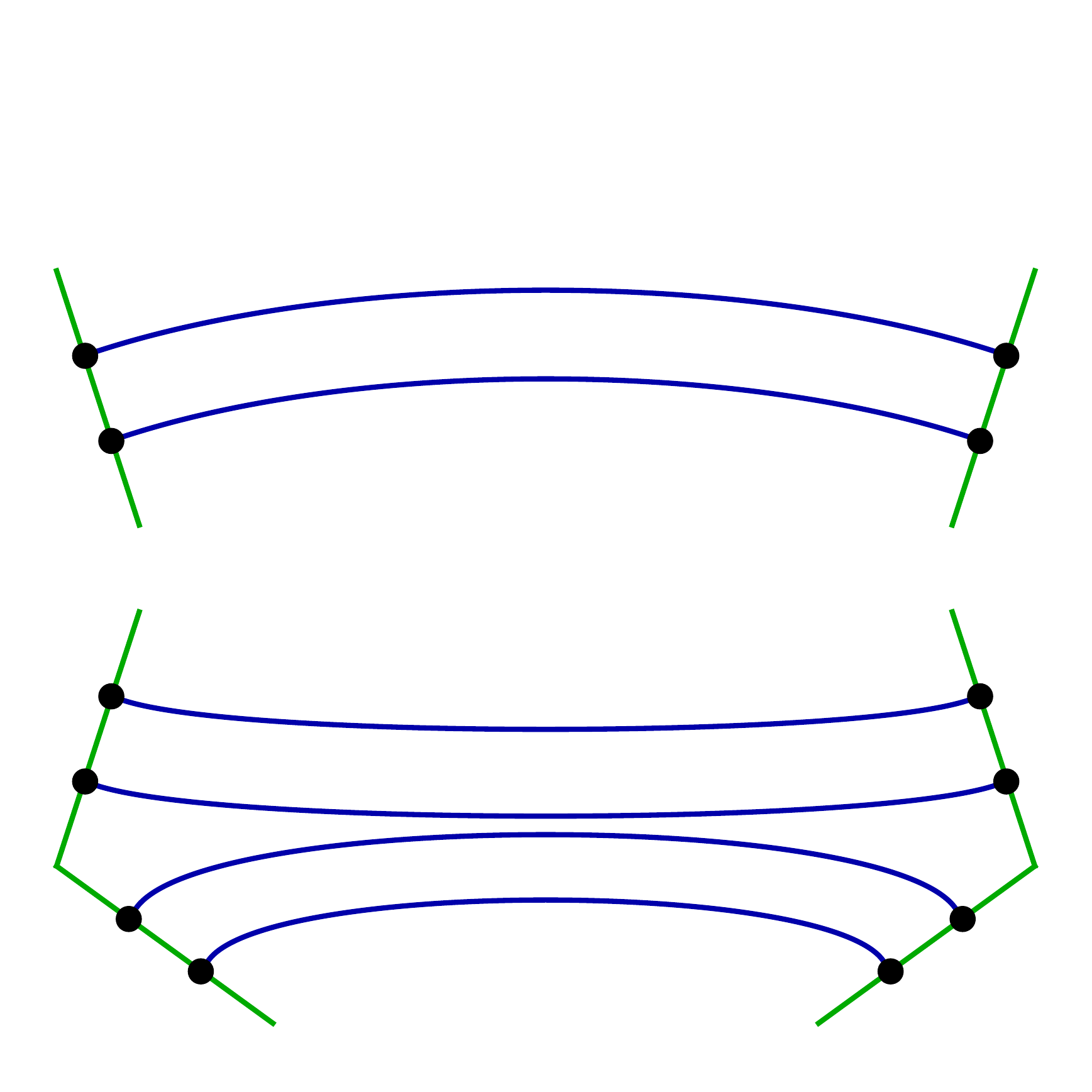}
\end{gathered} 
\;=\; \frac{\id_8}{8}
\end{align}
Instead of $2^{|C^{\text{C}}|}=32$ eigenstates, as in the previous example, we now 
only find $2^{|\gamma_C|}=8$, where $\gamma_C$ is the complement region of $C$ after simplifying $\rho_C$ (with $\partial \gamma_C = \partial C$; here, $\gamma_C = C$). This is because 
a basis set contracted onto $C^{\text{C}}$ of the original state does not lead to fully orthogonal states, 
for example:
\begin{align}
\begin{gathered}
\includegraphics[height=0.2\textheight]{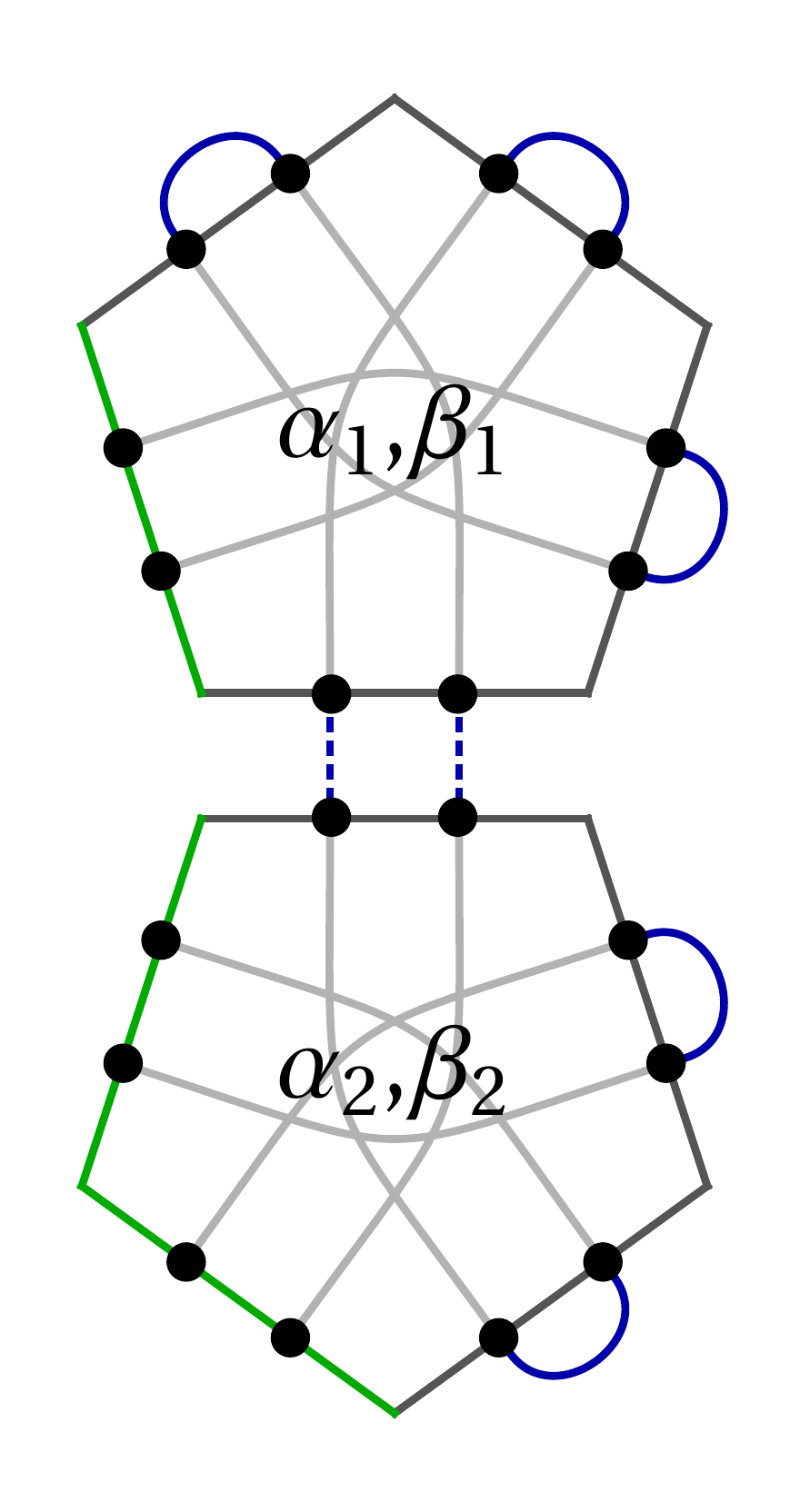}
\end{gathered}
\;&=\;
\begin{gathered}
\includegraphics[height=0.2\textheight]{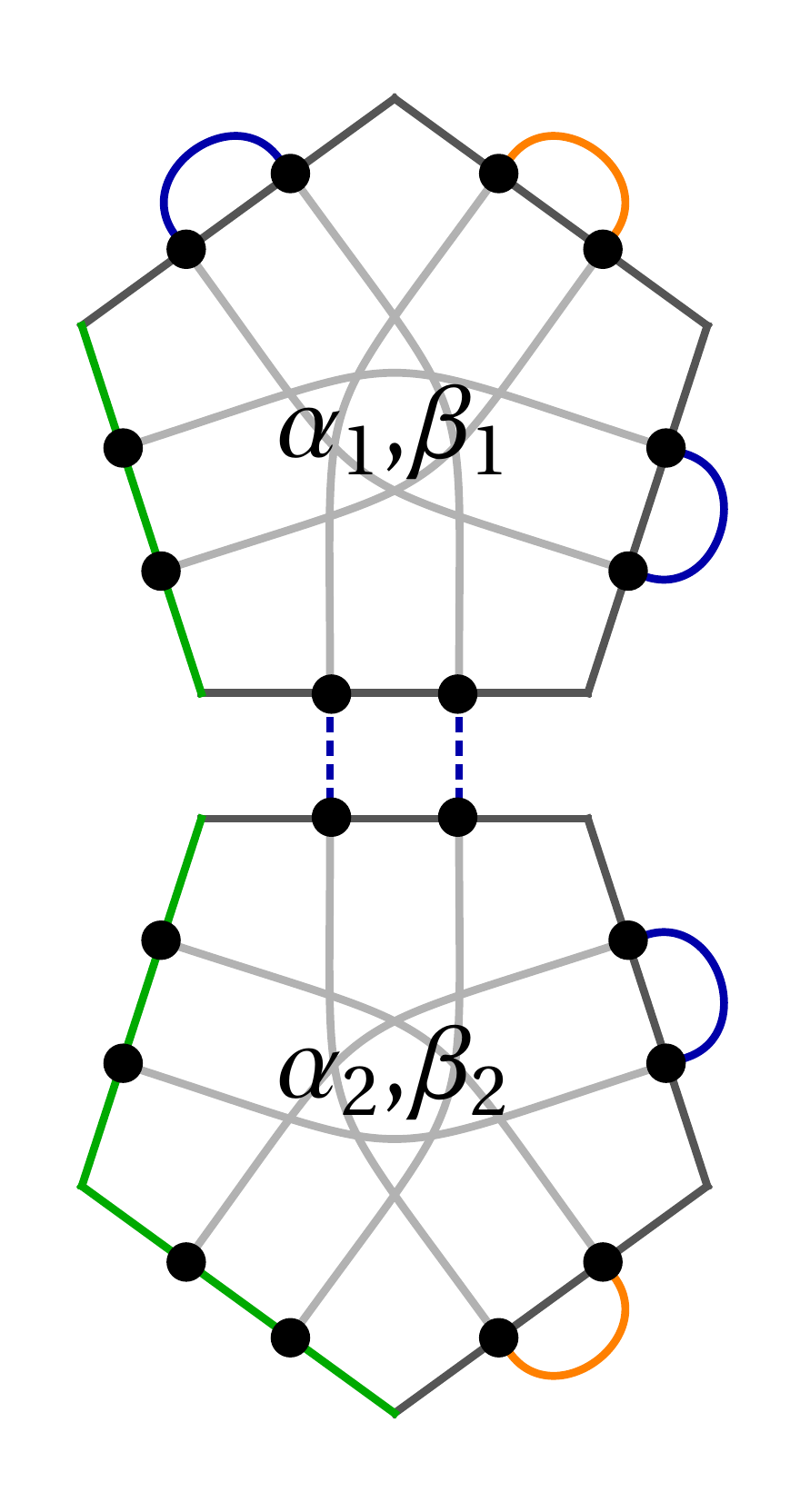}
\end{gathered} &
\begin{gathered}
\includegraphics[height=0.2\textheight]{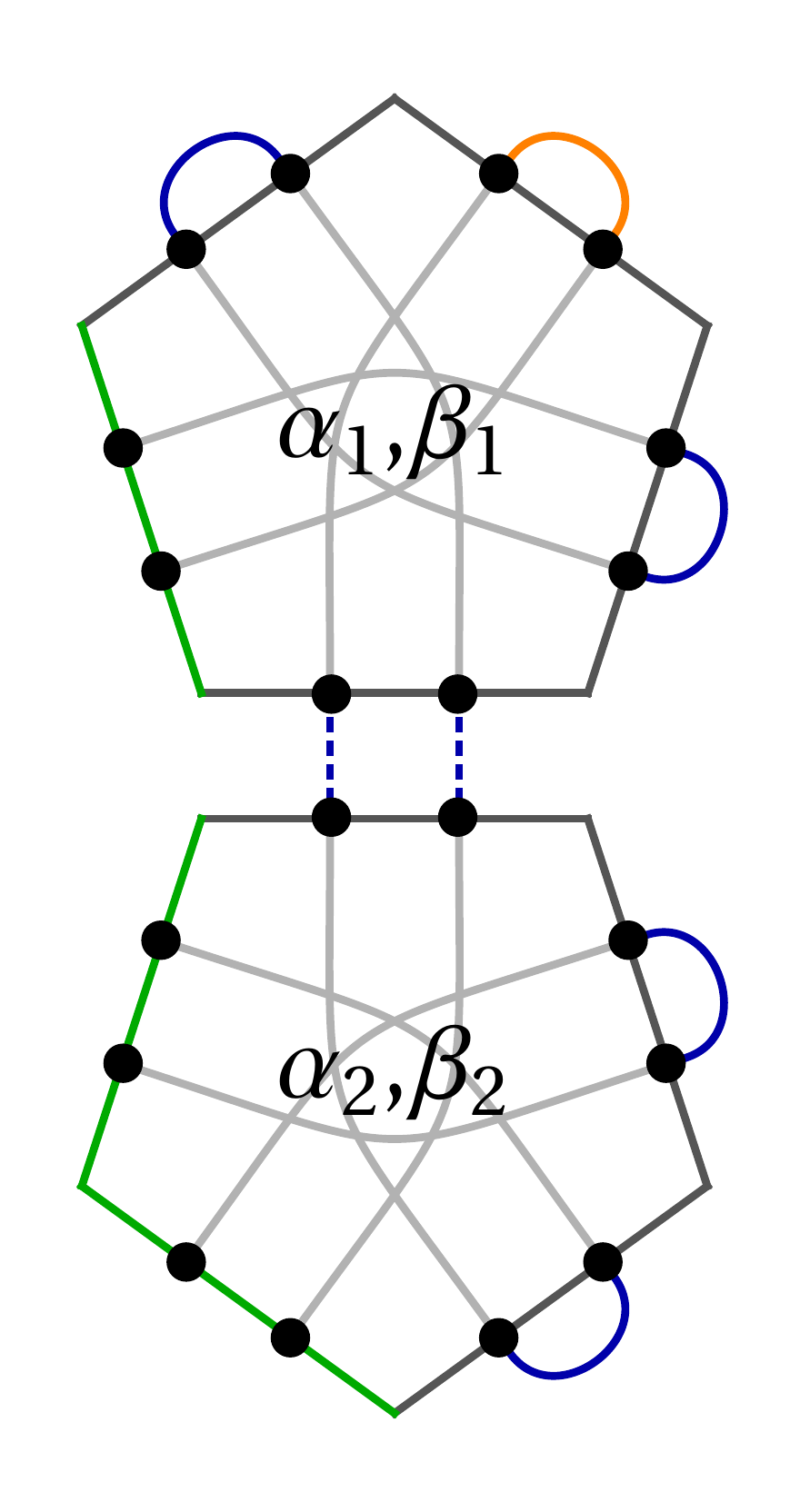}
\end{gathered}
\;&=\;
\begin{gathered}
\includegraphics[height=0.2\textheight]{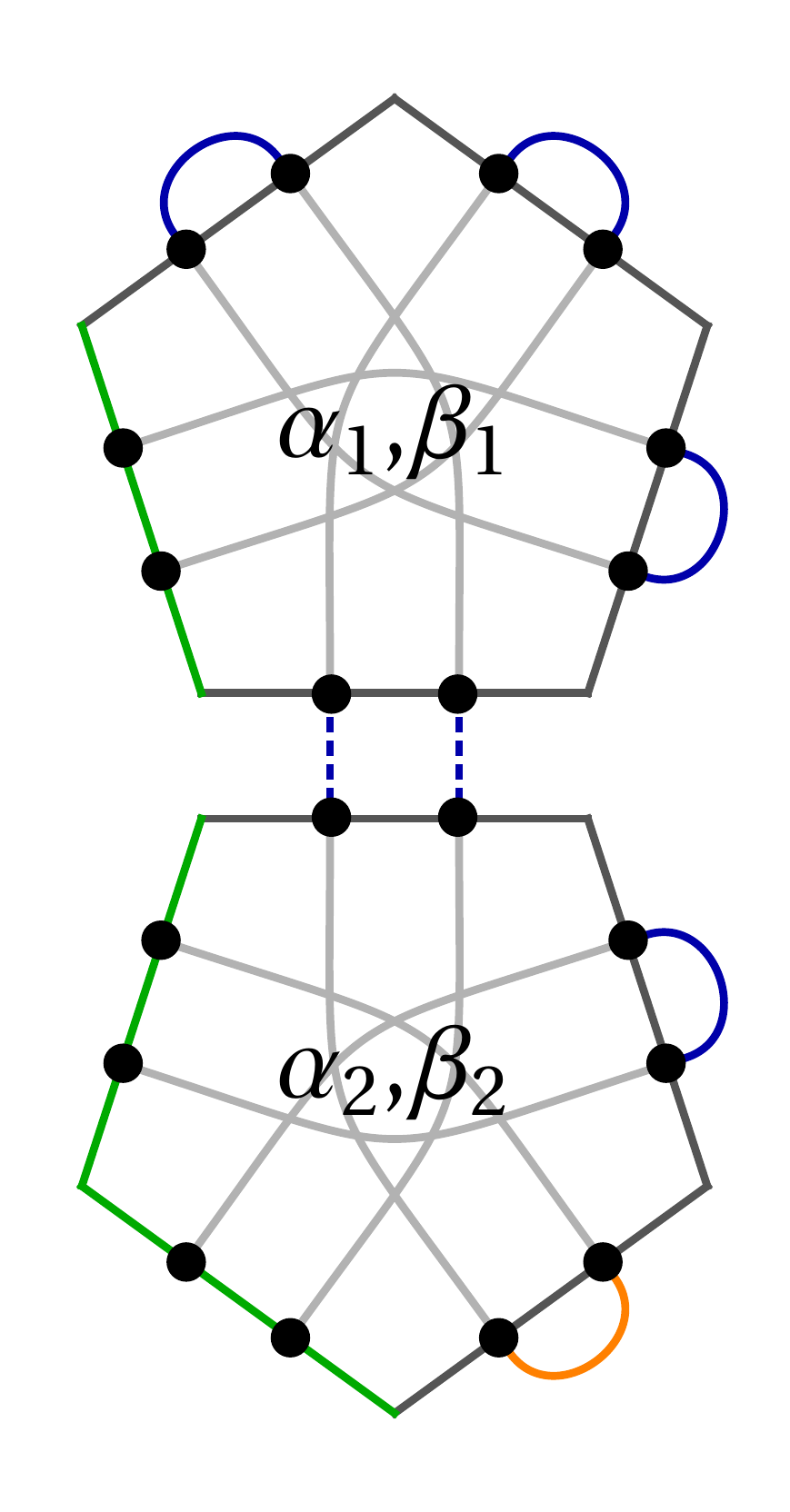}
\end{gathered} 
\end{align}
Thus we conclude that if by applying \eqref{EQ_HAPPY_SP_3C} and \eqref{EQ_HAPPY_SP_4C} a reduced density matrix $\rho_A$ can be simplified so that no dimers connect sites within $\gamma_A$, then there are $2^{|\gamma_A|}$ eigenstates with equal eigenvalues and an entanglement entropy $S_A = |\gamma_A| \log 2$.
When such a simplification is not possible, the entanglement entropy can depend on the bulk input. If we extend the region $C \to D$ onto half of the two-pentagon system, we cannot apply \eqref{EQ_HAPPY_SP_3C} and \eqref{EQ_HAPPY_SP_4C}:
\begin{align}
\rho_{\textcolor{darkgreen}{D}} \; = \; 2\;
&\begin{gathered}
\includegraphics[height=0.2\textheight]{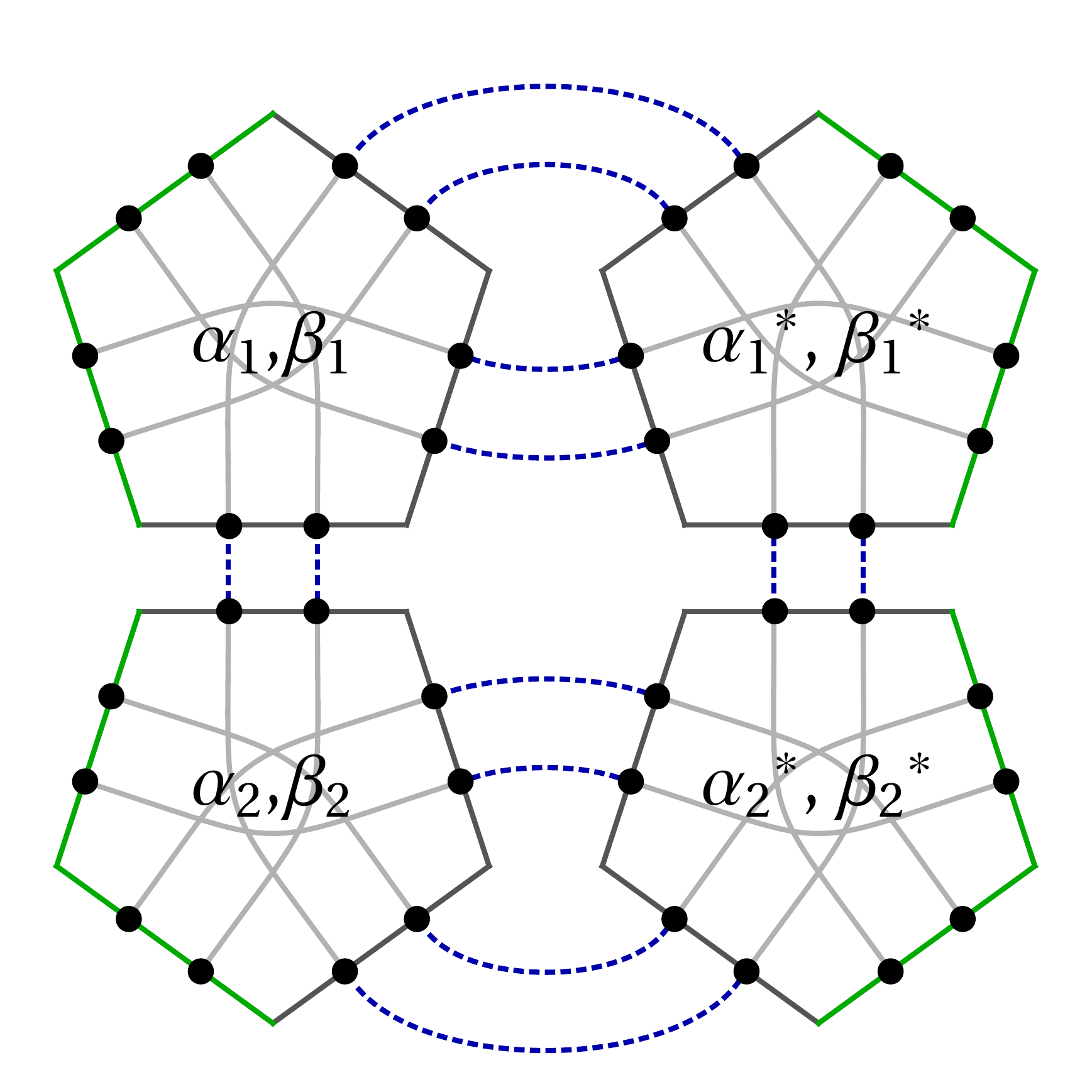}
\end{gathered}
\end{align}
Because of a dimer connecting Majorana modes \emph{within} $D$, an eigenbasis projected onto its edges becomes mixed. Indeed, the reduced density matrix $\rho_{D}$ separates into a sum of parity-even and parity-odd terms, as cross-terms between both vanish:
\begin{align}
\begin{gathered}
\includegraphics[height=0.2\textheight]{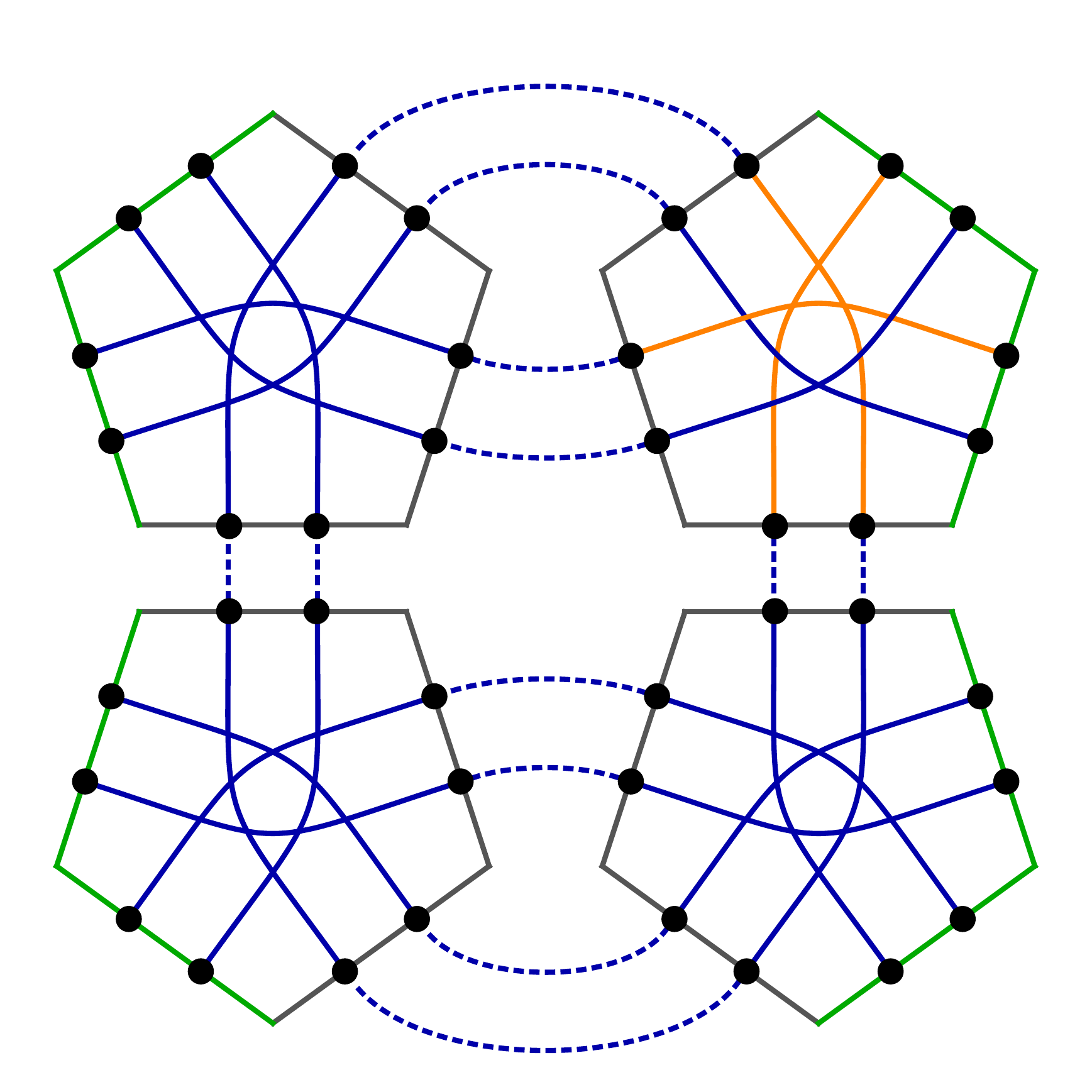}
\end{gathered}
\;=\;
&\begin{gathered}
\includegraphics[height=0.2\textheight]{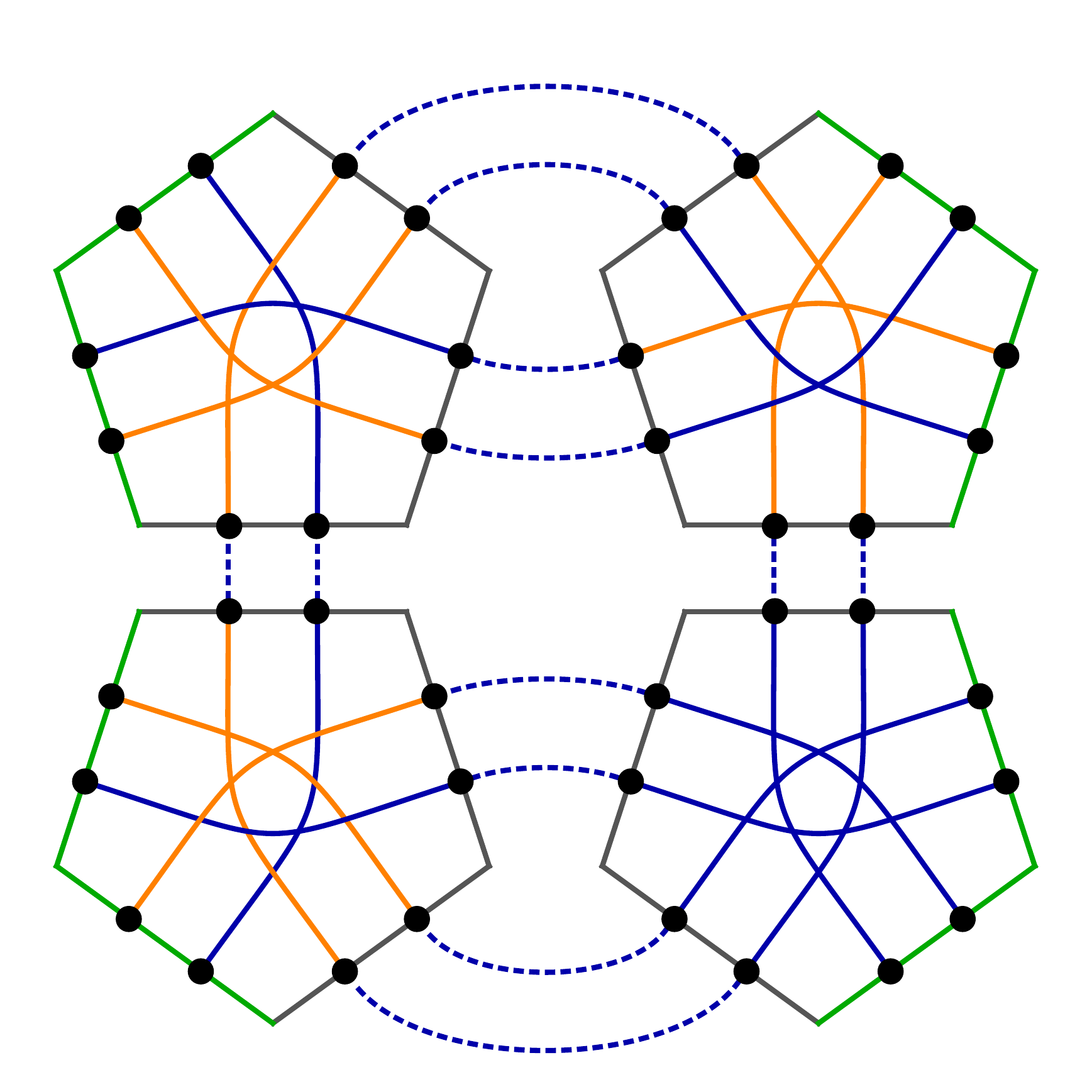}
\end{gathered}
\;=\; 0\ .
\end{align}
We can thus write
\begin{align}
\rho_D =\, &\tr_{D^\text{C}} \left( \alpha_1 \alpha_2 \ket{\bar{0},\bar{0}}+\beta_1 \beta_2 \ket{\bar{1},\bar{1}})(\alpha_1^\star \alpha_2^\star \bra{\bar{0},\bar{0}}+\beta_1^\star \beta_2^\star \bra{\bar{1},\bar{1}} \right) \nonumber\\
+ &\tr_{D^\text{C}} \left( \alpha_1 \beta_2 \ket{\bar{0},\bar{1}}+\beta_1 \alpha_2 \ket{\bar{1},\bar{0}})(\alpha_1^\star \beta_2^\star \bra{\bar{0},\bar{1}}+\beta_1^\star \alpha_2^\star \bra{\bar{1},\bar{0}} \right) \nonumber\\
\equiv\, & \tr_{D^\text{C}} |\psi^+\rangle \langle\psi^+| +  \tr_{D^\text{C}} |\psi^-\rangle \langle\psi^-| \text{ .}
\end{align}
We have defined as $|\psi^\pm\rangle$ the parity-even and parity-odd part of the total state vector $\ket\psi$. For each of the two states, we can still apply our previous approach of finding the eigenbasis by projecting a complete dimer basis on the state itself, yielding $S_D(\psi^+) = S_D(\psi^-) = 3 \log 2$, as 3 dimers connect $D$ and $D^C$. 
Following the rules for the entanglement of superpositions for biorthogonal states \cite{PhysRevLett.97.100502}, we can now compute the entanglement entropy of the full state as
\begin{align}
S_D &= \langle \psi^+ |\psi^+\rangle S_D(\psi^+) + \langle \psi^- |\psi^-\rangle S_D(\psi^-) + h_2(\langle \psi^+ |\psi^+\rangle) \nonumber\\
&= 3 \log 2 + h_2(|\alpha_1\alpha_2|^2 + |\beta_1\beta_2|^2) \leq 4 \log 2 \text{ ,}
\end{align}
where we have used the binary entropy function $x\mapsto h_2(x) := - x \log x - (1-x) \log(1-x)$. We are thus in a position to compute the entanglement entropy even for complicated superpositions of dimer states.

Assuming a boundary region $A$ that can be simplified using \eqref{EQ_HAPPY_SP_3C} and \eqref{EQ_HAPPY_SP_4C}, however, we can easily compute the entanglement entropy of the full HyPeC \emph{independent of the bulk input}. For this we follow the steps laid out in \eqref{EQ_HAPPY_RHO_A_RED1}, \eqref{EQ_HAPPY_RHO_A_RED2} and \eqref{EQ_HAPPY_RHO_A_RED3} for the construction of reduced density matrices and their eigenstates. Using our previous notation for superpositions, an example is given by 
\begin{align}
\ket\psi \; = \;2^{N_C/2}\;
&\begin{gathered}
\includegraphics[height=0.267\textheight]{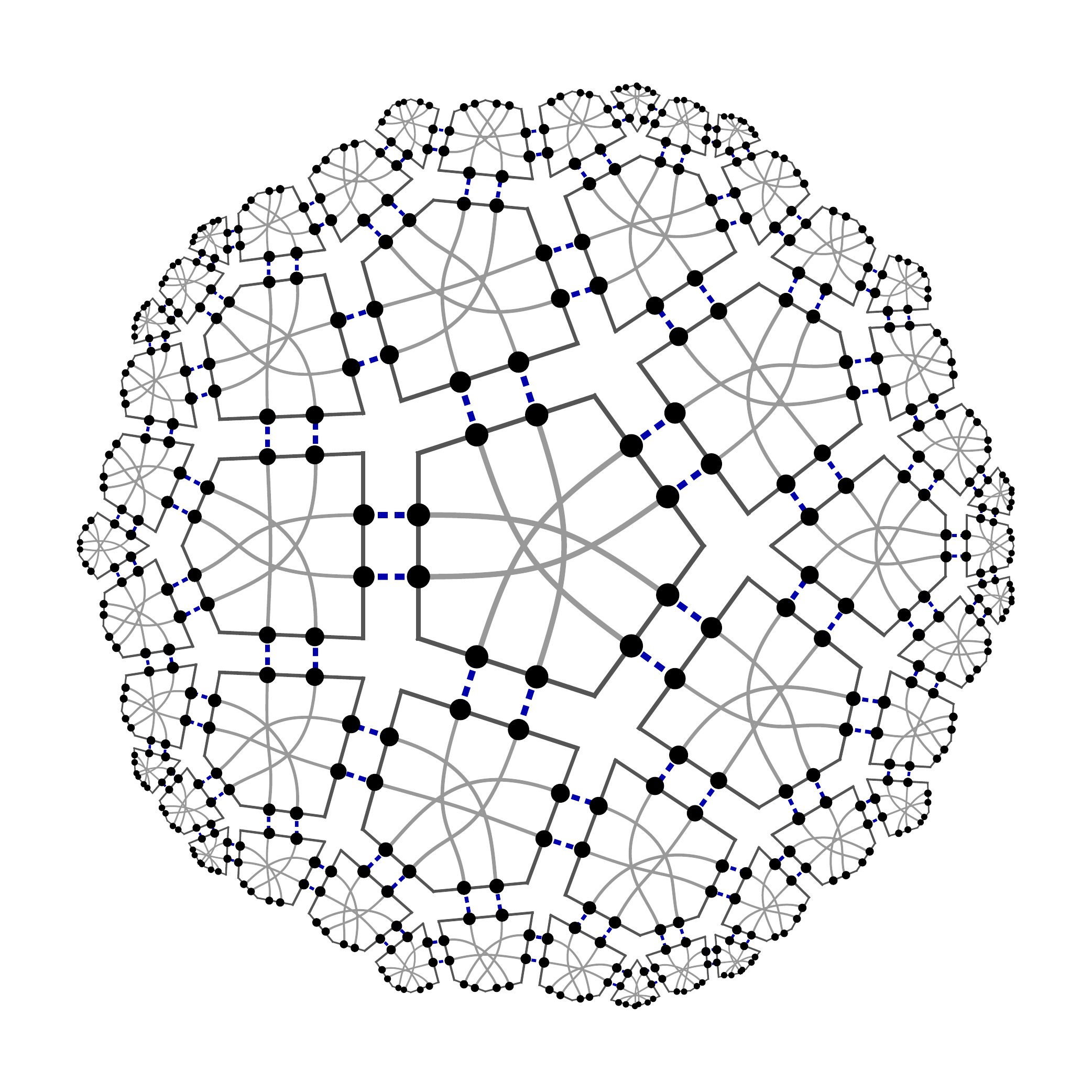}
\end{gathered}
\end{align}
The normalization depends on the number $N_C$ of internal contractions.
We omitted the superposition labels $\alpha,\beta$ for clarity, but still assume only local superpositions within each tile. Given a boundary region $A$, we first simplify the reduced density matrix $\rho_A$ as in \eqref{EQ_HAPPY_RHO_A_RED3}, being left with a wedge $\mathcal{W}$ bounded by minimal cut (or ``bulk geodesic'') $\gamma_A$:
\begin{align}
\rho_{\textcolor{darkgreen}{A}} \; = \; 2^{N_C} \;
&\begin{gathered}
\includegraphics[height=0.376\textheight]{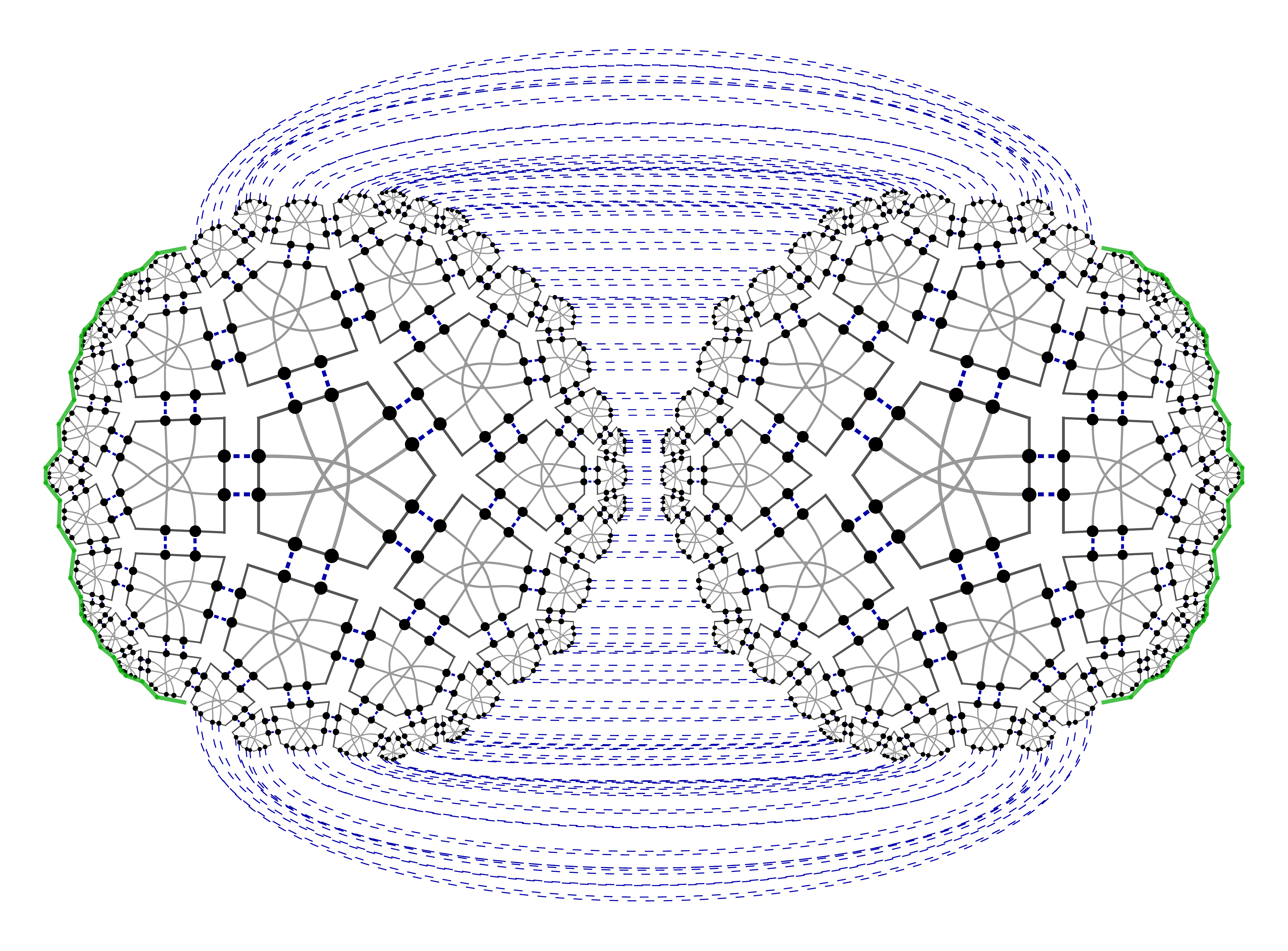}
\end{gathered} \nonumber\\
=\; 2^{N_{C,\mathcal{W}}}\;
&\begin{gathered}
\includegraphics[height=0.218\textheight]{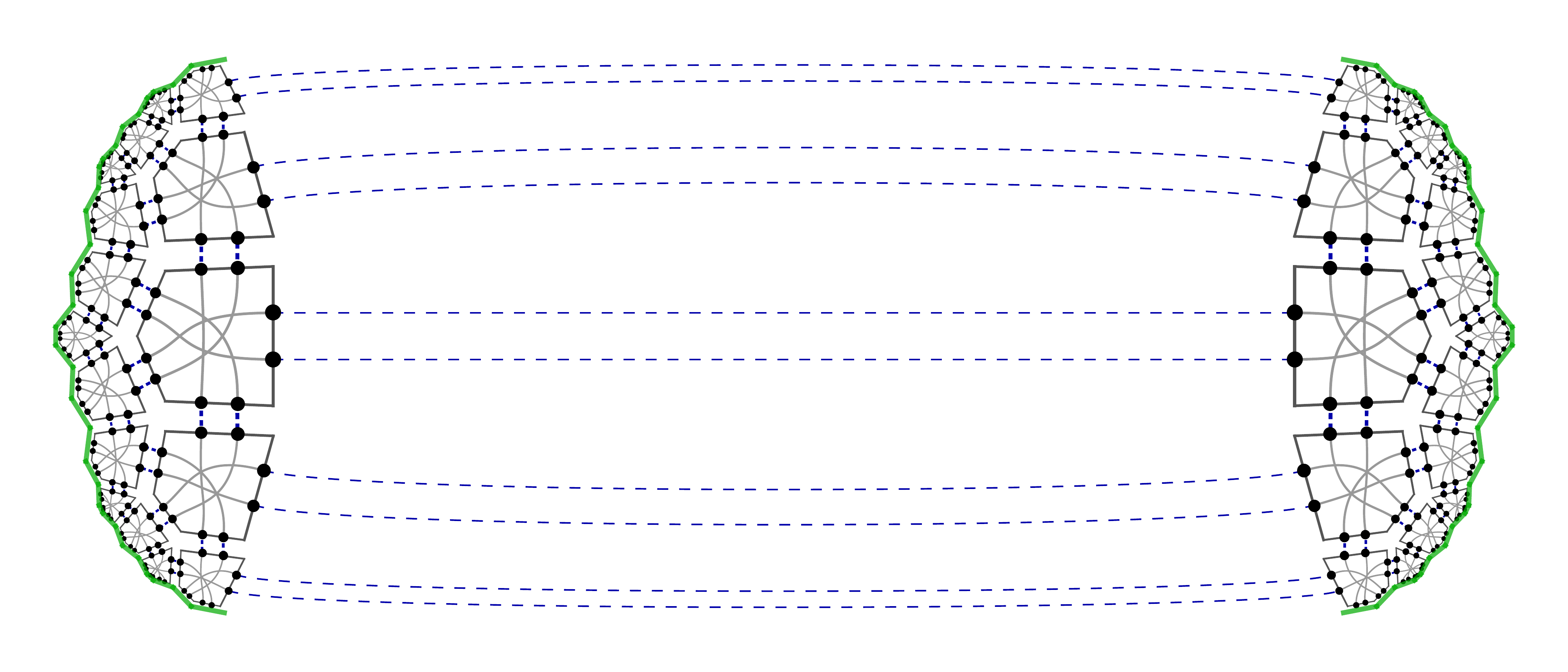}
\end{gathered}
\end{align}
Here $N_{C,\mathcal{W}}$ is the number of (still unresolved) contractions in the wedge $\mathcal{W}$. The eigenstate basis of $\rho_A$ can again be constructed by projecting a complete basis of dimer states onto the edges of $\gamma_A$, leading to a $2^{|\gamma_A|}$-dimensional space of states, where $|\gamma_A|$ is the number of edges along the cut. For illustration, we consider the eigenstate with all-even projections:
\begin{align}
|\psi_{\textcolor{darkgreen}{A}}^{0,0,0,0,0}\rangle \;&=\;
\begin{gathered}
\includegraphics[height=0.218\textheight]{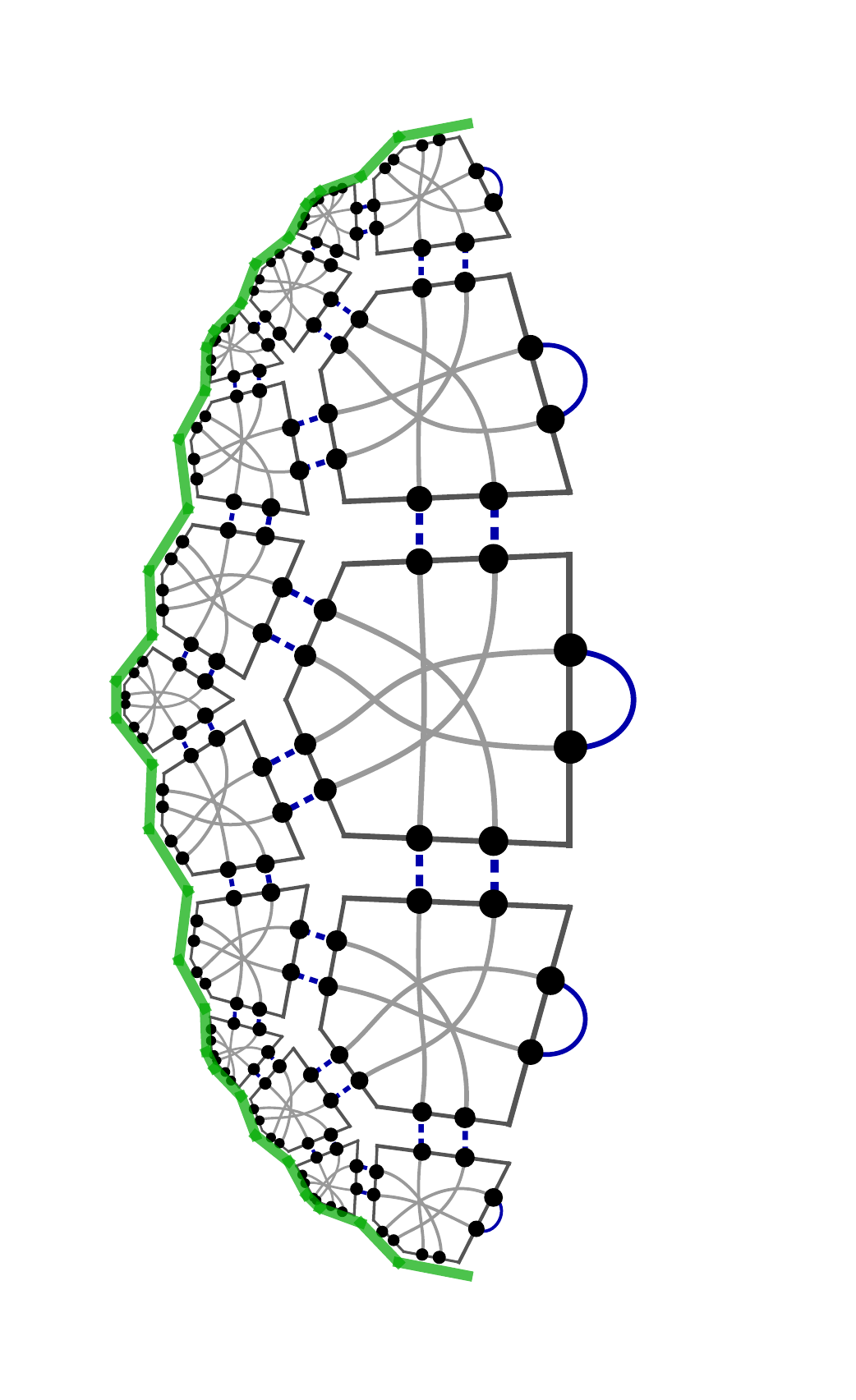}
\end{gathered}
\end{align}
The eigenvalue equation can be evaluated with the same techniques that we have used for reducing \eqref{EQ_HAPPY_RHO_A_RED1} and \eqref{EQ_HAPPY_RHO_A_RED2},
\begin{align}
\rho_{\textcolor{darkgreen}{A}}\, |\psi_{\textcolor{darkgreen}{A}}^{0,0,0,0,0}\rangle \; &= \;2^{N_{C,\mathcal{W}}}\; 
\begin{gathered}
\includegraphics[height=0.218\textheight]{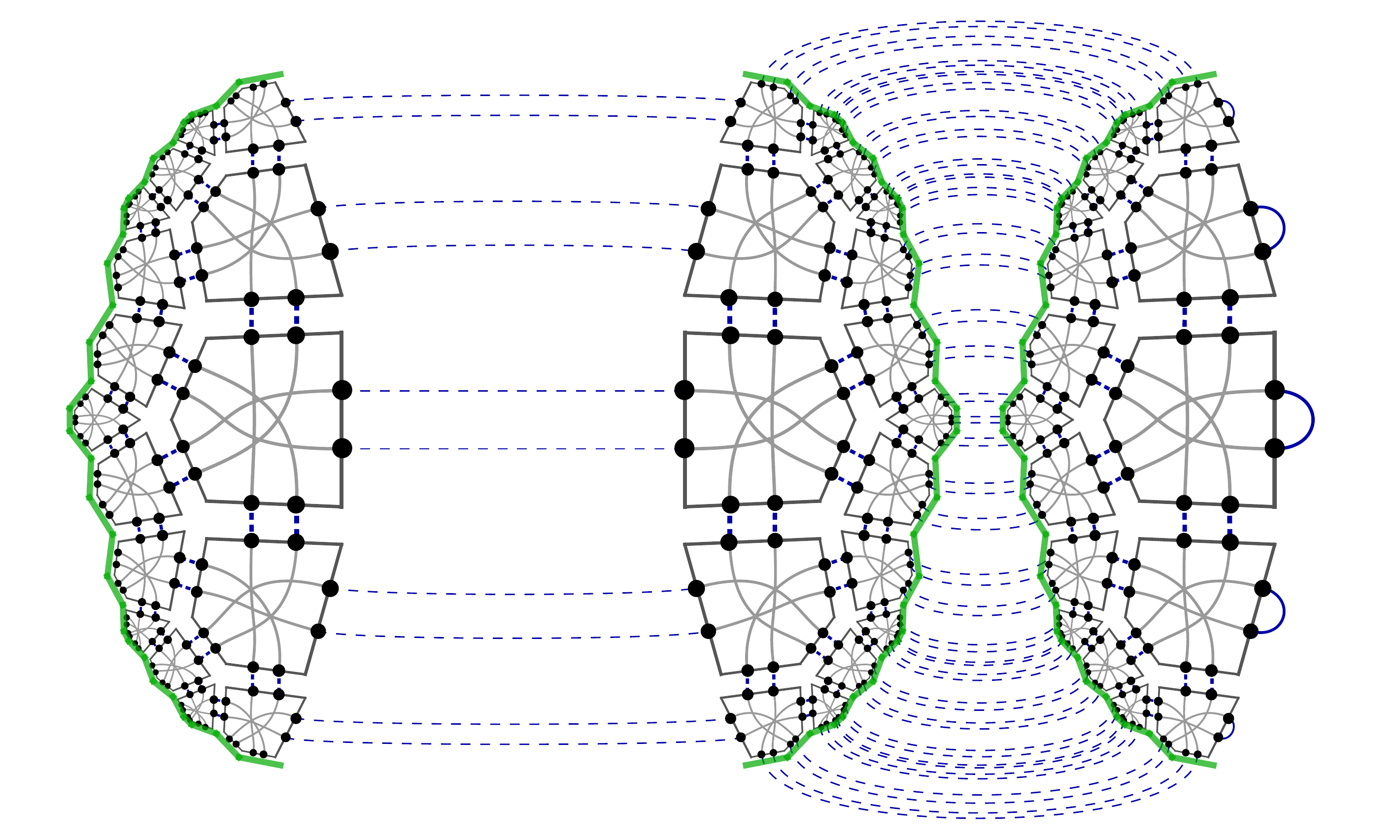}
\end{gathered} \nonumber\\
&=\; \frac{1}{2^{|\gamma_A|}}
\begin{gathered}
\includegraphics[height=0.218\textheight]{pentagon_net_rhoA_happy_full_00000.pdf}
\end{gathered}
\end{align}
Thus we find the same entanglement entropy as for the case of fixed logical input states, $S_A= |\gamma_A| \log 2$ (with $|\gamma_A|=5$). Our procedure is equivalent to the \textsl{greedy algorithm} \cite{Pastawski2015}, which in dimer language is manifested through the reduction steps \eqref{EQ_HAPPY_SP_4C} and \eqref{EQ_HAPPY_SP_3C}. As for the greedy algorithm, our approach only applies when both $A$ and its complement $A^{\text{C}}$ are reduced to the same $\gamma_A$ after simplification. In that case, we can draw the following conclusion about the reduced density matrix:
\begin{align}
\label{EQ_RHO_A_REDUC}
\rho_{\textcolor{darkgreen}{A}}^2 \; = \;2^{2N_{C,\mathcal{W}}}\;
&\begin{gathered}
\includegraphics[height=0.22\textheight]{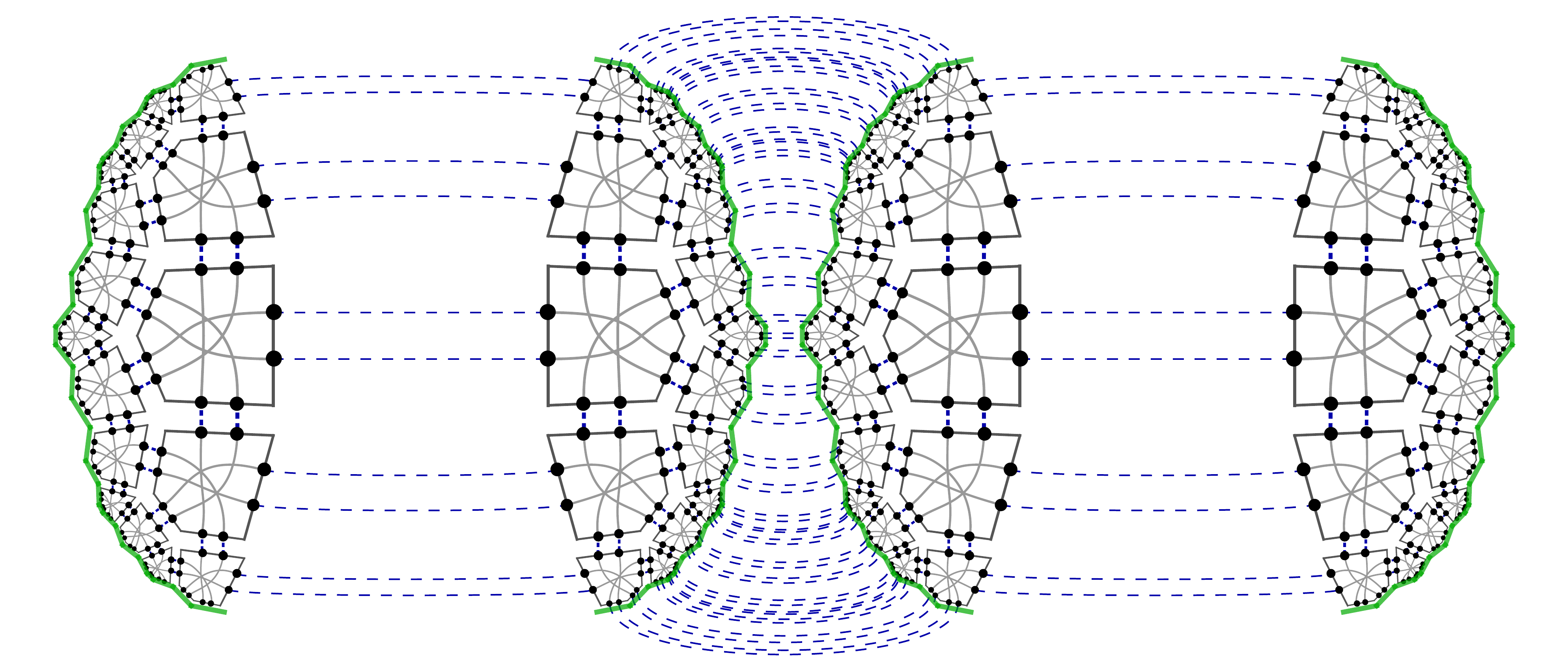}
\end{gathered}
\nonumber\\
=\;\frac{2^{N_{C,\mathcal{W}}}}{2^{|\gamma_{\textcolor{darkgreen}{A}}|}}
&\begin{gathered}
\includegraphics[height=0.22\textheight]{pentagon_net_rhoA_happy_full3.pdf}
\end{gathered}
= \frac{\rho_{\textcolor{darkgreen}{A}}}{2^{|\gamma_{\textcolor{darkgreen}{A}}|}}\text{ ,}
\end{align}
As a result, we find a flat spectrum of R\'enyi entropies $S_A^{(n)}=S_A=|\gamma|\log 2$. When reductions from $A$ and $A^{\text{C}}$ are not equivalent, i.e.,\ when the greedy algorithm does not converge to a geodesic, dimers will be ``lost'' during each power of $\rho_A$, and $S_A^{(n)}$ will decrease with $n$.

\newpage
\section{Contraction order}
\label{APP_CONTR_ORDER}

Contracting Majorana dimer states on a given tiling can give rise to ambiguities regarding contraction order: Before contraction, each tile has its own ordering of indices, some of which are contracted out and some remain on the edges of the final geometry. We consider here the HyPeC with its underlying spin tensor network description, of which the Majorana dimers form an effective representation. Let us start with the simplest case of an $\bar{0}$ (read: ``logical zero'') input everywhere in the bulk and a successive contraction of neighbouring tiles, starting from the centre:
\begin{equation}
\begin{aligned}
&\begin{gathered}
\includegraphics[height=0.17\textheight]{happy_contr1.pdf}
\end{gathered}
\quad\quad \scalebox{1.5}{$\rightarrow$}\quad\quad
\begin{gathered}
\includegraphics[height=0.17\textheight]{happy_contr2.pdf}
\end{gathered}\quad \\
\quad \scalebox{1.5}{$\rightarrow$}\quad\quad
&\begin{gathered}
\includegraphics[height=0.17\textheight]{happy_contr3.pdf}
\end{gathered}
\quad\quad \scalebox{1.5}{$\rightarrow$}\quad\quad
\begin{gathered}
\includegraphics[height=0.17\textheight]{happy_contr4.pdf}
\end{gathered}
\end{aligned}\quad\quad
\end{equation}
As all tiles have been filled with even-parity input states, the dimer parities of the fully contracted state is entirely independent of the initial ordering: As shown in \eqref{EQ_PENTAGON_STATES}, any cyclic permutation (i.e.,\ pivot shifts) of the initial tiles or at intermediate contraction steps would have left the dimer parities invariant.
For a general bulk input, however, the initial index labeling does matter: The $\bar{1}$ input has odd parity and its dimer parities thus change under cyclic permutations, as shown in \eqref{EQ_PENTAGON_STATES2}. Thus, whenever a $\bar{1}$ tile is contracted in, the total parity of the contraction changes, and while the total parity is odd, any cyclic permutation leads to a string of $Z$ edge operator, as discussed in Sec.\ \ref{APP_CYCL_PERM}. This leads to the following contraction rule for arbitrary fixed bulk input:
\begin{lemma}[Dimer parities of the fixed-input HyPeC]
\label{LEM_HAPPY_FIXED_INPUT}
Contracting fixed $[[5,1,3]]$ code states on a pentagon tiling is equivalent to multiplying dimer parities of contracted dimer pairs (regardless of the initial orientation of tiles) and adding $Z$ strings between the pivots of pairs of tiles with $\bar{1}$ input. If the number of $\bar{1}$ inputs is odd, then an additional $Z$ string connects the pivot of the unpaired $\bar{1}$ tile with the pivot of the full contraction. 
\end{lemma}
\begin{proof}
Without loss of generality, consider a particular contraction order and initial tile orientation. Whenever the total parity of the contraction at any step is even, contracting a $\bar{1}$ tile will cause pivot shifts in all following contraction steps to produce $Z$ strings, until another $\bar{1}$ tile is reached and total parity becomes even again. The starting and end points of these $Z$ strings are the pivots of the first and second $\bar{1}$ tile. If the number of $\bar{1}$ tiles is even, then the final contraction will contains $Z$ strings between each successive pair of $\bar{1}$ tile pivots. If it is odd, then the $Z$ string from the last $\bar{1}$ tile will continue until the boundary of the full contraction. Consider the previous contraction for a pair of $\bar{1}$ tiles, with pivots of the odd tiles (whose orientation is now relevant) shown by a small circle:
\begin{equation}
\begin{gathered}
\includegraphics[height=0.17\textheight]{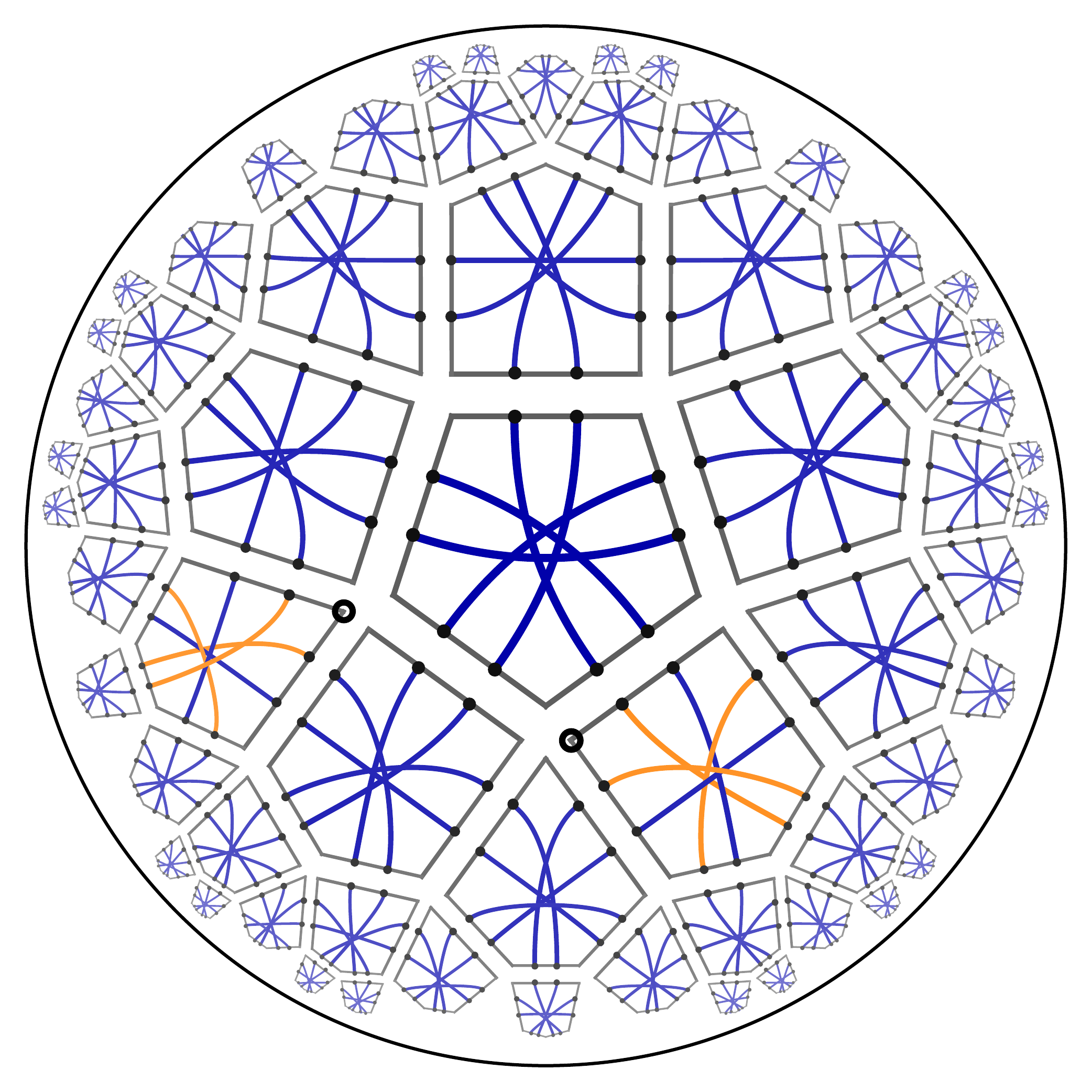}
\end{gathered}
\quad\quad \scalebox{1.5}{$\rightarrow$}\quad\quad
\begin{gathered}
\includegraphics[height=0.17\textheight]{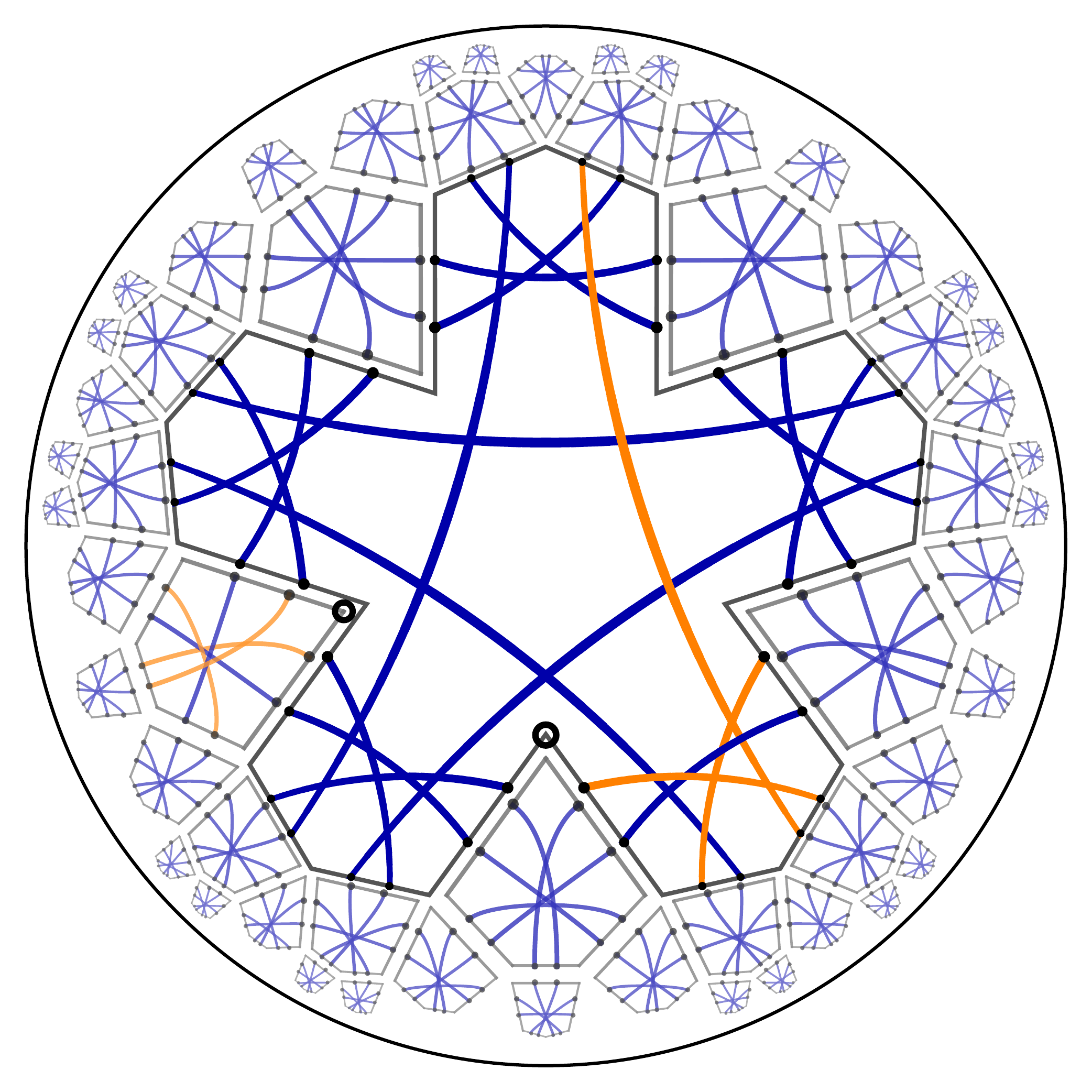}
\end{gathered}
\end{equation}
During the first iteration of contractions, the last contraction involves the first $\bar{1}$ tile and makes the contraction parity-odd. We thus need to mark it with a pivot, which is simply the pivot of the original $\bar{1}$ tile. To contract the other $\bar{1}$ tile, a pivot shift is required, which produces a $Z$ string (red line):
\begin{equation}
\begin{gathered}
\includegraphics[height=0.17\textheight]{happy_contrB_2.pdf}
\end{gathered}
\quad\quad \scalebox{1.5}{$=$}\quad\quad
\begin{gathered}
\includegraphics[height=0.17\textheight]{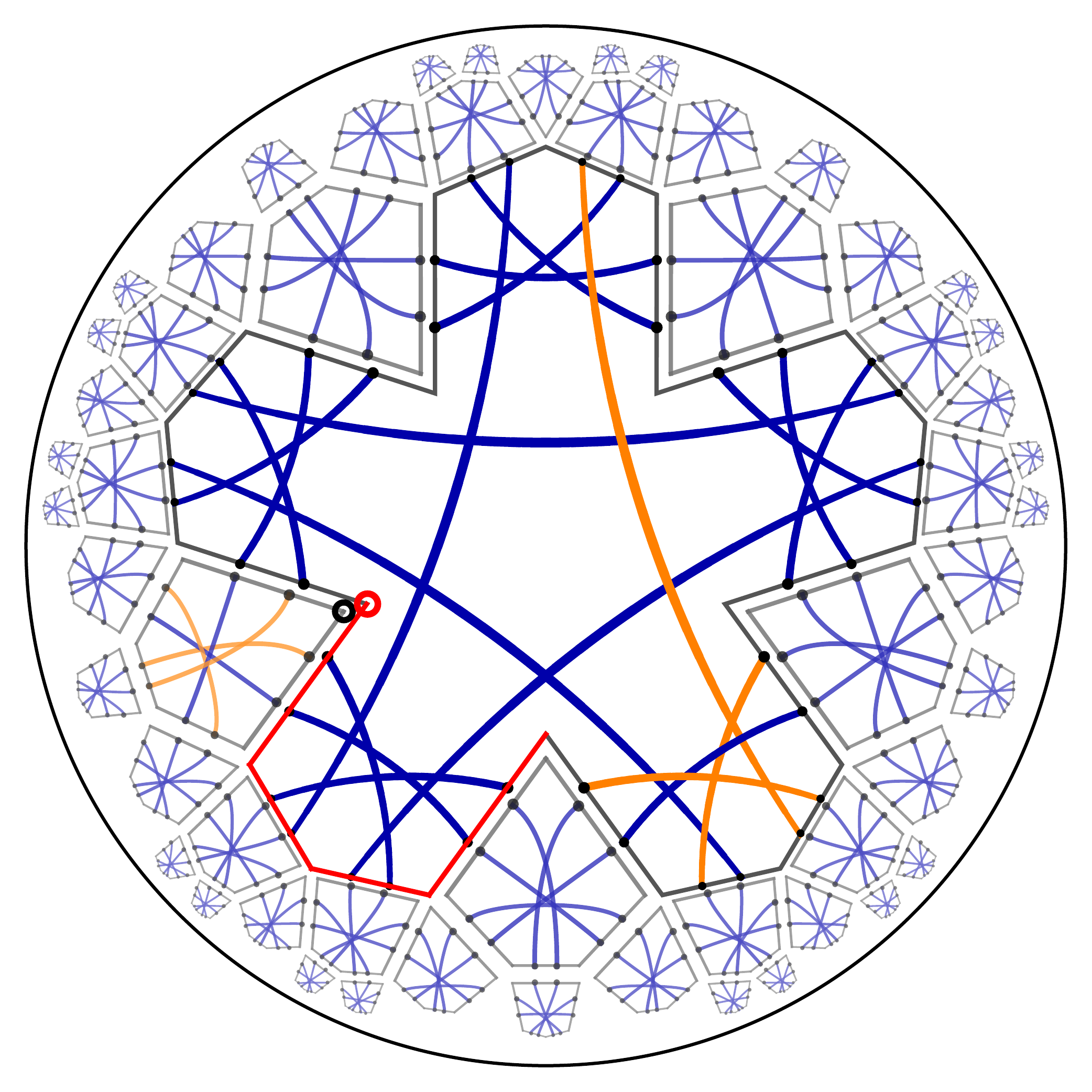}
\end{gathered}
\end{equation}
The contraction is now again parity-even (the pivots ``cancel each other out''), and the rest of the contraction can be performed without worrying about orientations:
\begin{equation}
\begin{gathered}
\includegraphics[height=0.17\textheight]{happy_contrB_2b.pdf}
\end{gathered}
\quad\quad \scalebox{1.5}{$\rightarrow$}\quad\quad
\begin{gathered}
\includegraphics[height=0.17\textheight]{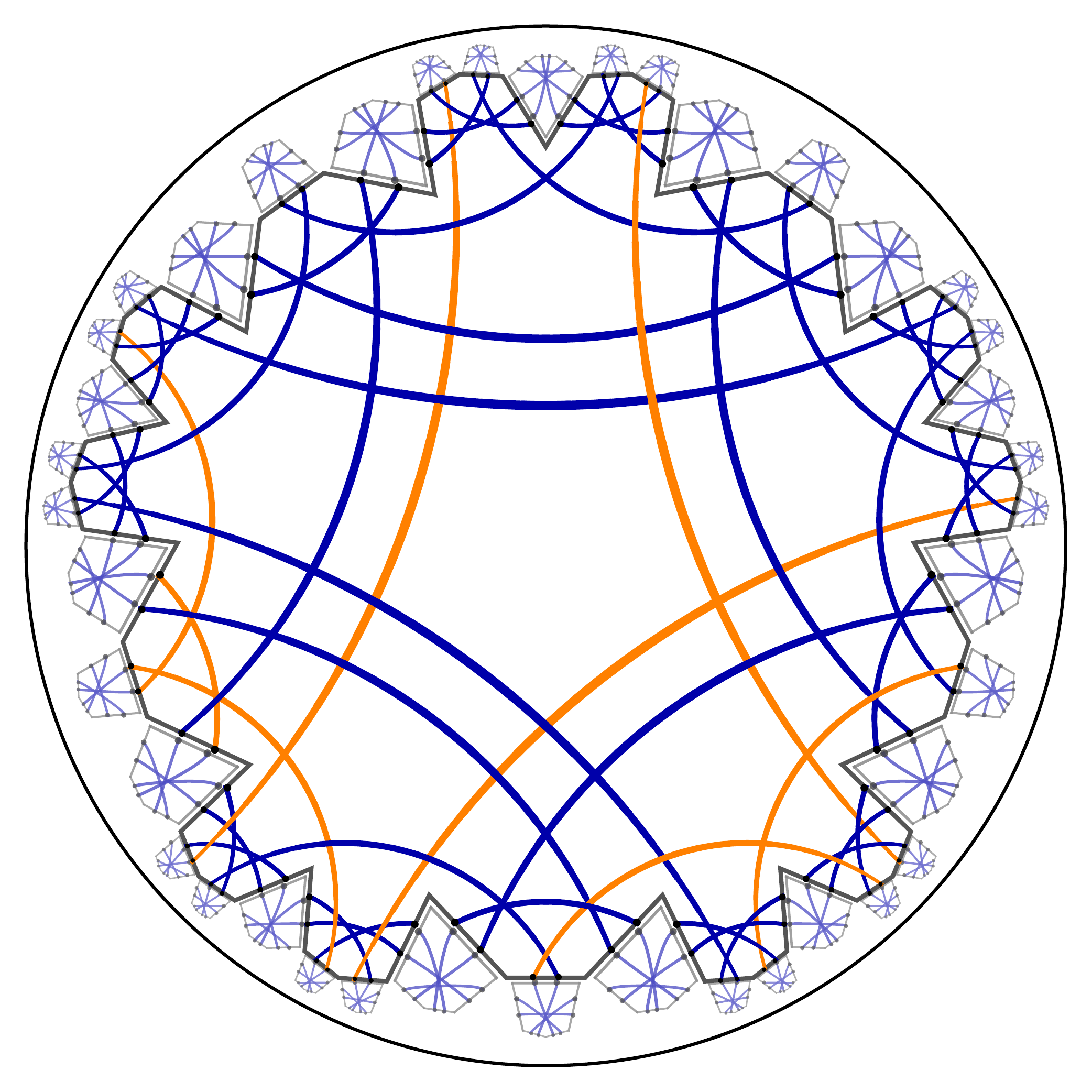}
\end{gathered}
\quad\quad \scalebox{1.5}{$\rightarrow$}\quad\quad
\begin{gathered}
\includegraphics[height=0.17\textheight]{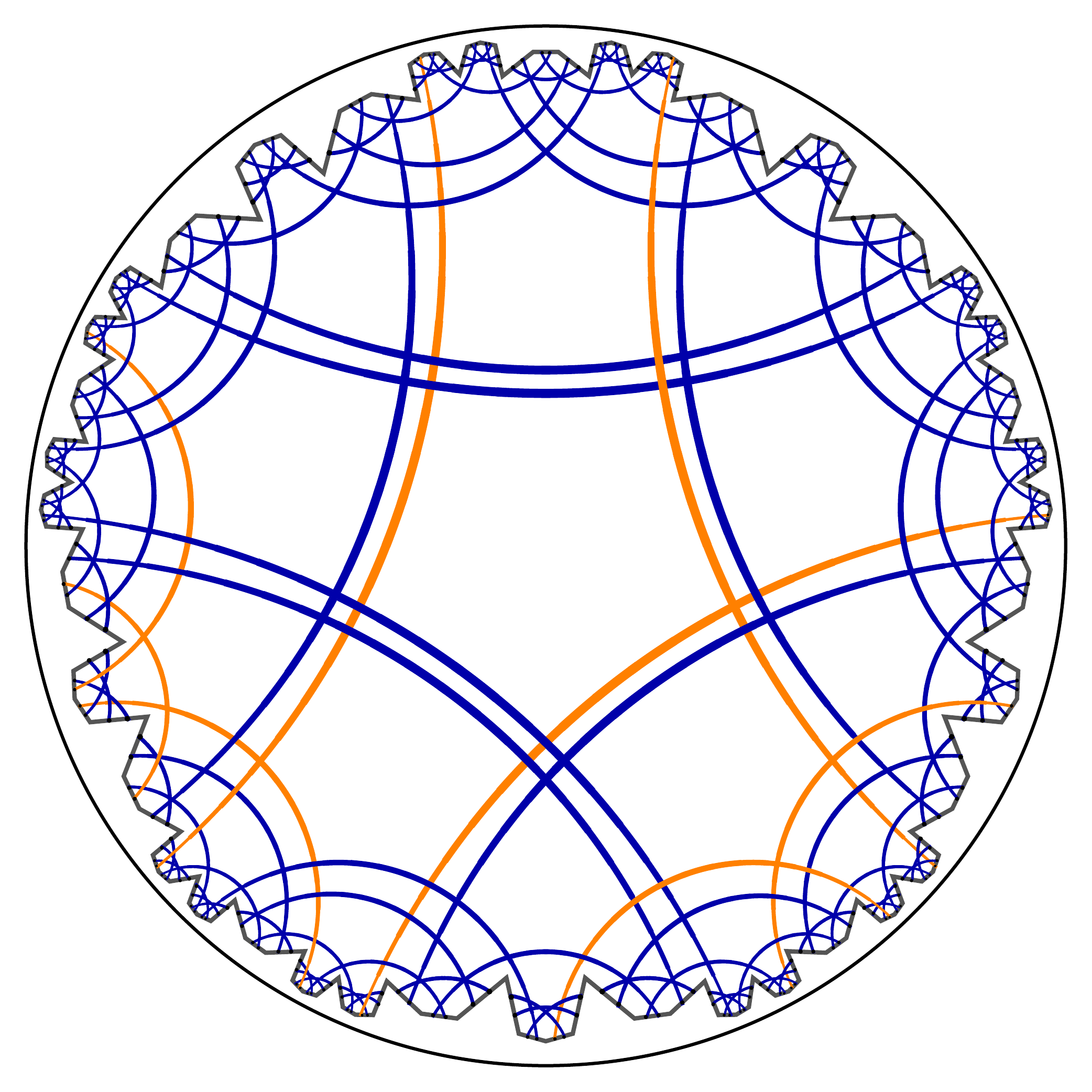}
\end{gathered}
\end{equation}
This result is independent of the ordering of the previous contraction, as we can freely deform the $Z$ strings through the $\bar{0}$ tiles:
\begin{equation}
\begin{gathered}
\includegraphics[height=0.17\textheight]{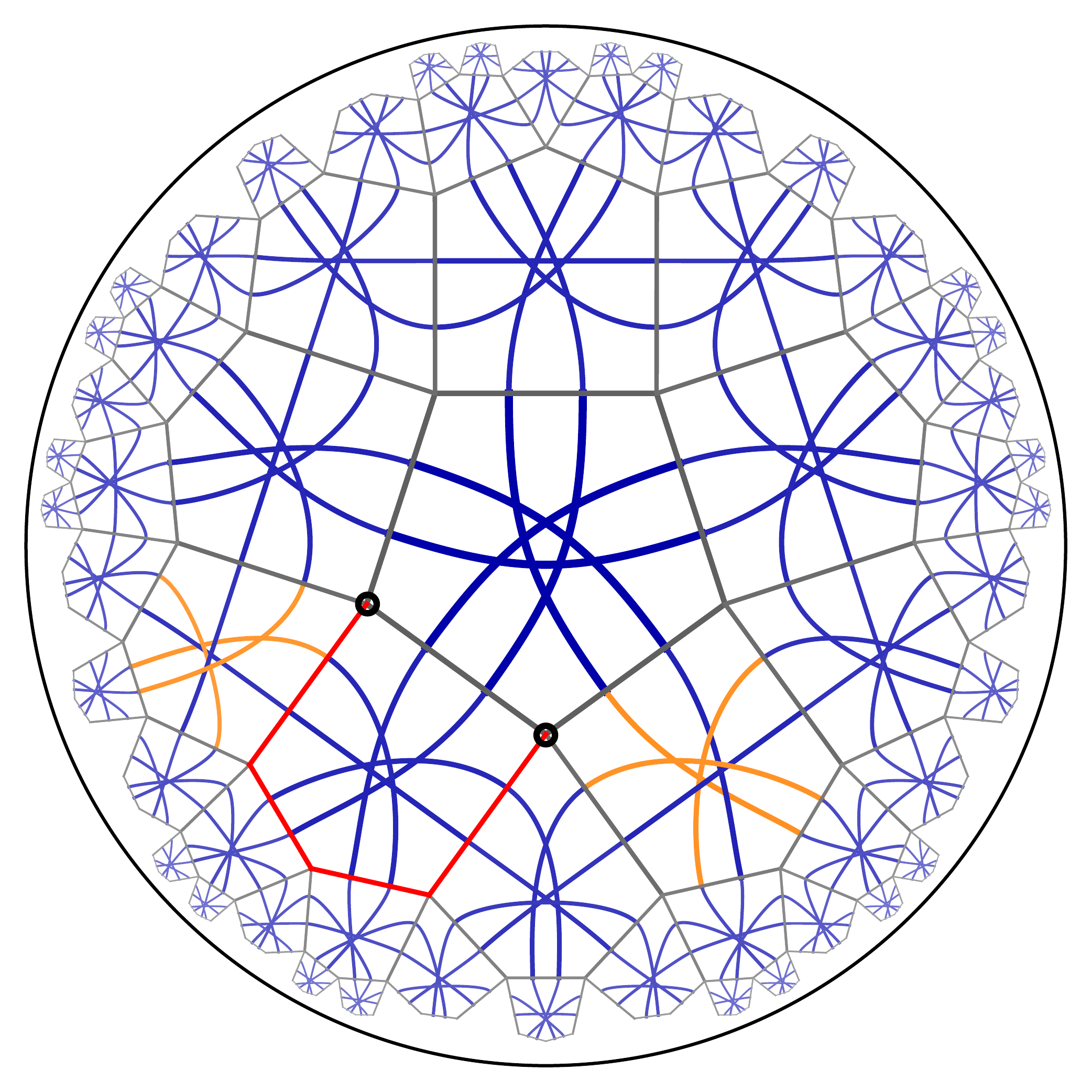}
\end{gathered}
\quad\quad \scalebox{1.5}{$=$}\quad\quad
\begin{gathered}
\includegraphics[height=0.17\textheight]{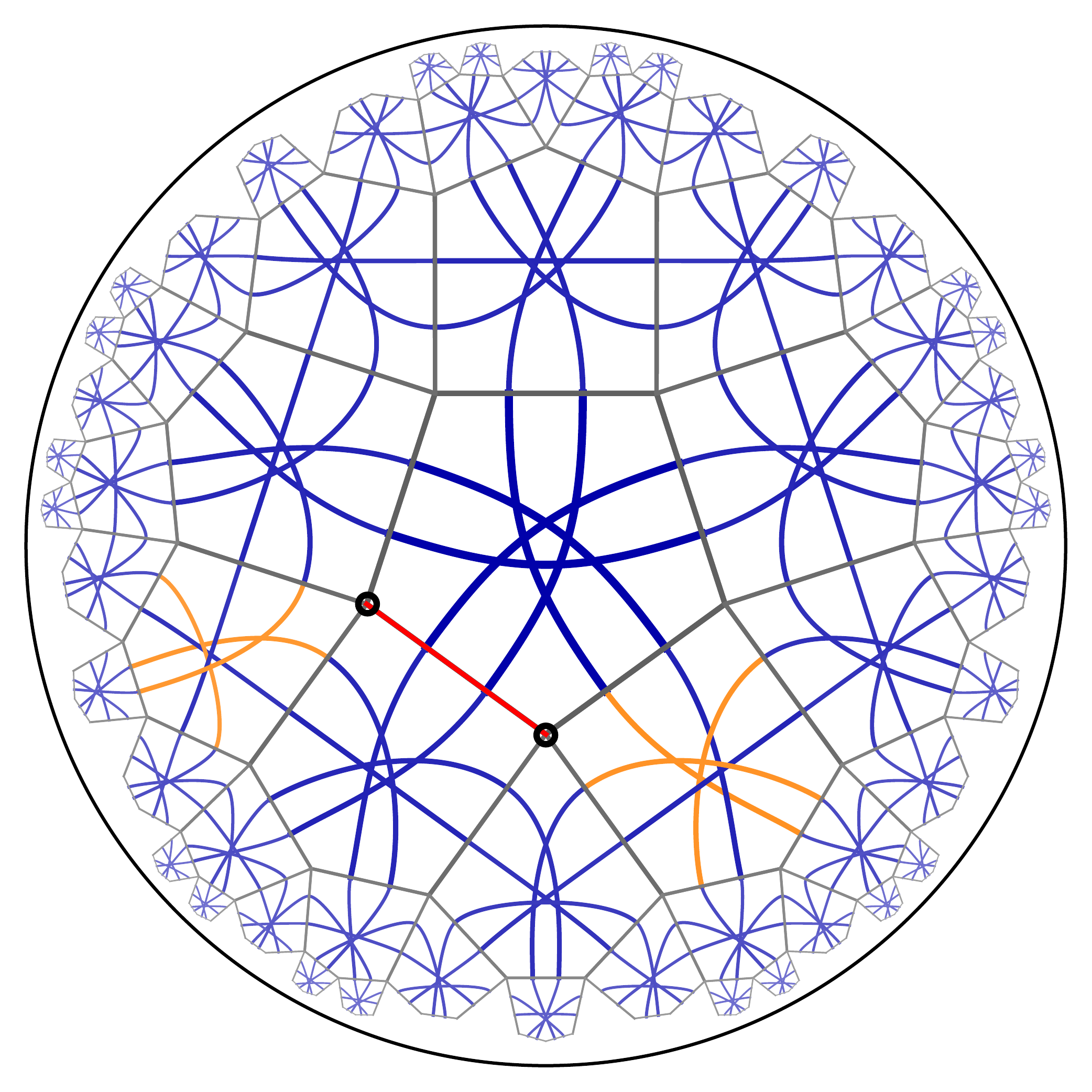}
\end{gathered}
\quad\quad \scalebox{1.5}{$=$}\quad\quad
\begin{gathered}
\includegraphics[height=0.17\textheight]{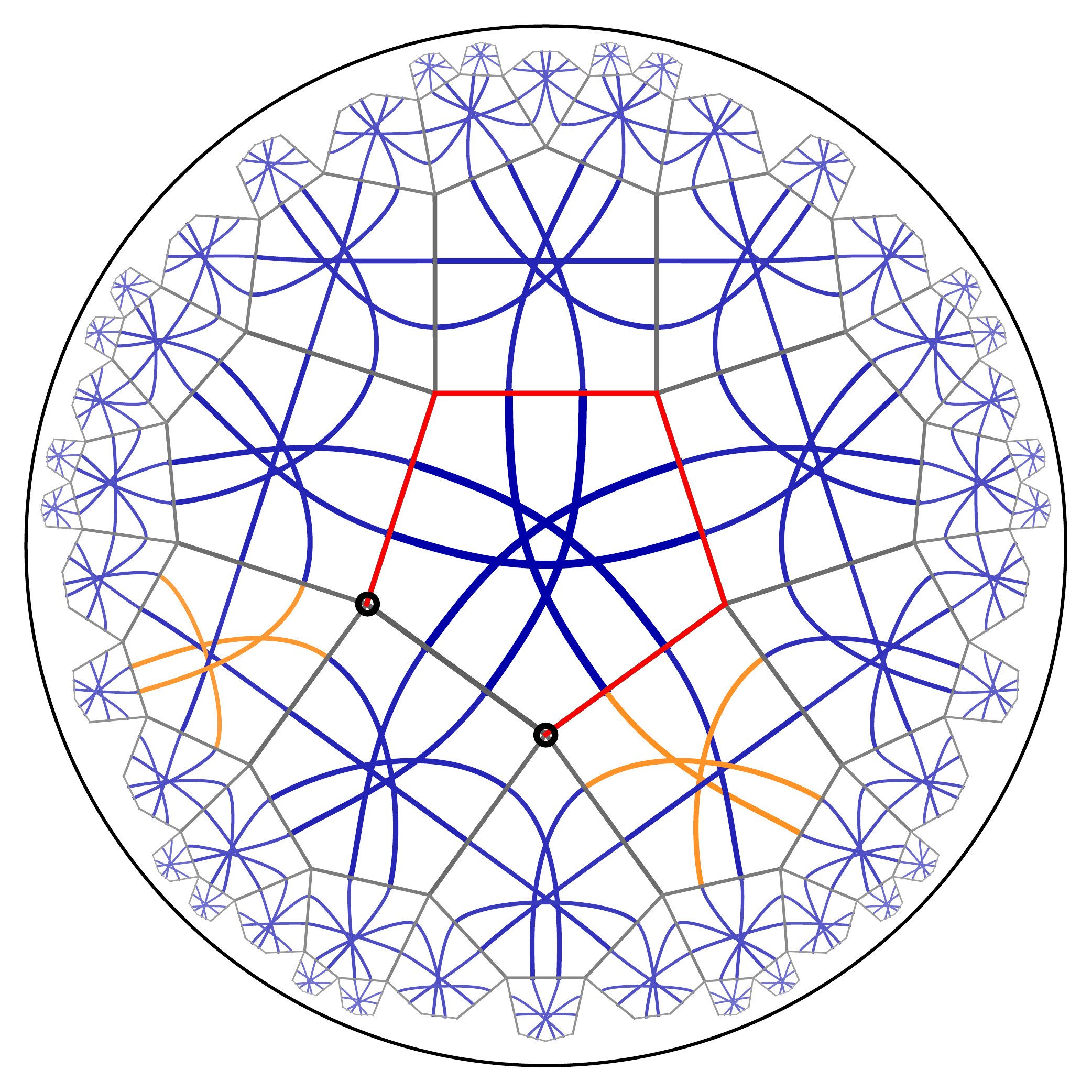}
\end{gathered}
\end{equation}
To indicate the action of the $Z$ strings on the full contraction, we have omitted the spaces between tiles in the previous diagram.
Furthermore, the result is independent of the initial orientations of the $\bar{1}$ tiles, as rotating these is equivalent to extending or shortening the $Z$ strings, as we have found in \eqref{EQ_PENTAGON_STATES}:
\begin{equation}
\begin{gathered}
\includegraphics[height=0.17\textheight]{happy_contrB_order1.pdf}
\end{gathered}
\quad\quad \scalebox{1.5}{$=$}\quad\quad
\begin{gathered}
\includegraphics[height=0.17\textheight]{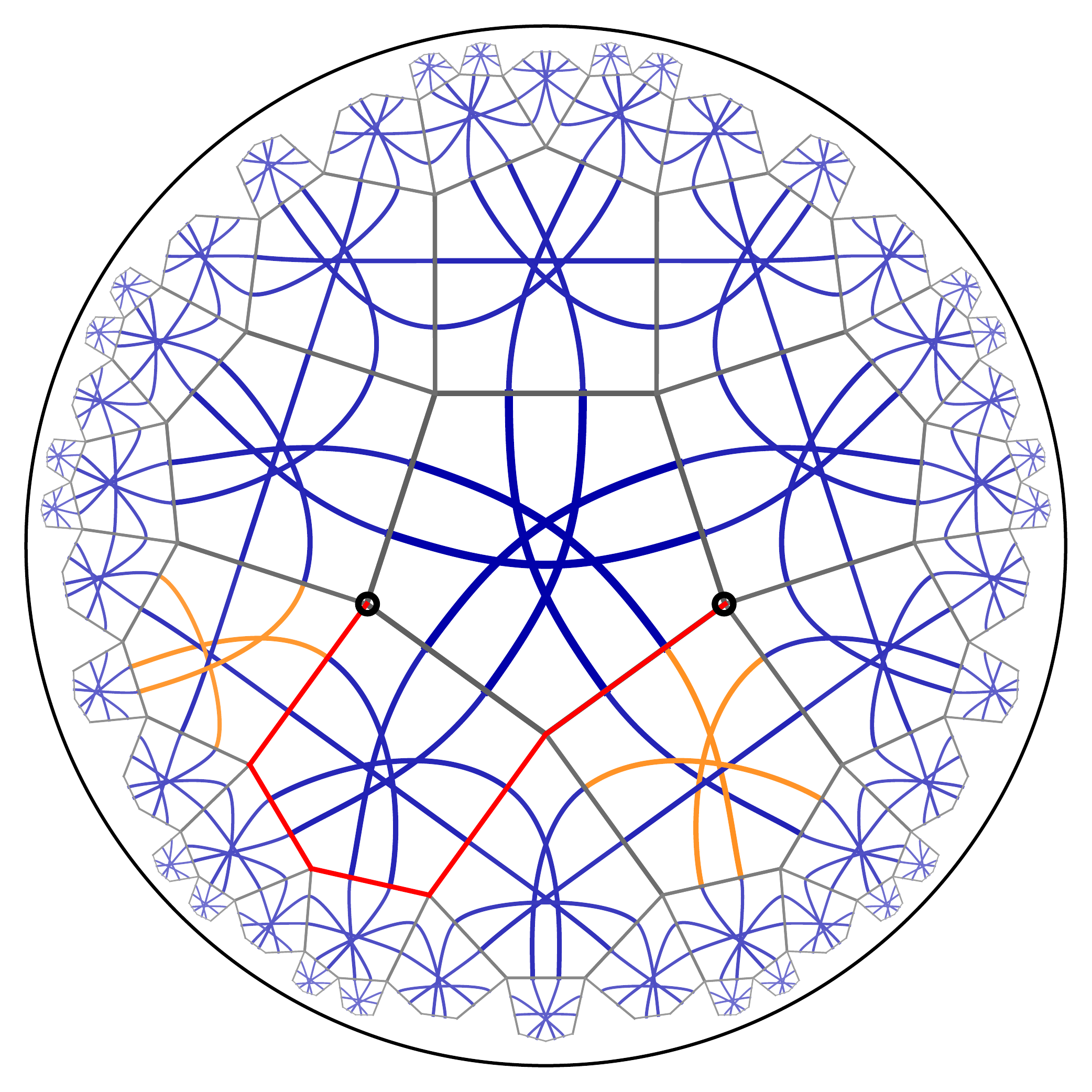}
\end{gathered}
\quad\quad \scalebox{1.5}{$=$}\quad\quad
\begin{gathered}
\includegraphics[height=0.17\textheight]{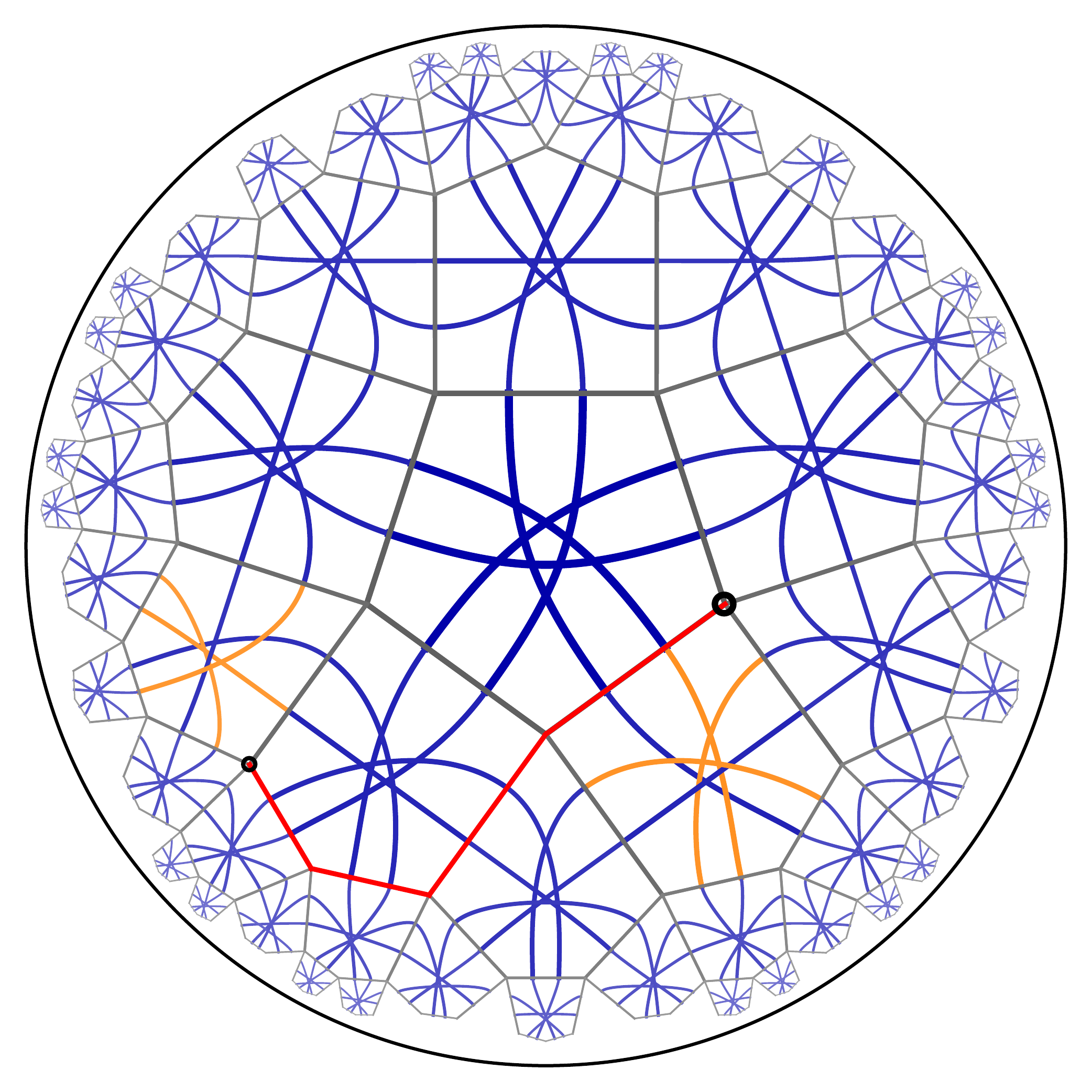}
\end{gathered}
\end{equation}
Contracting more than two $\bar{1}$ tiles will create $Z$ strings between pairs of them, in the order in which we contracted. This ordering, however, also does not affect the final contraction, as we can change this pairing using the same rules:
\begin{equation}
\begin{gathered}
\includegraphics[height=0.17\textheight]{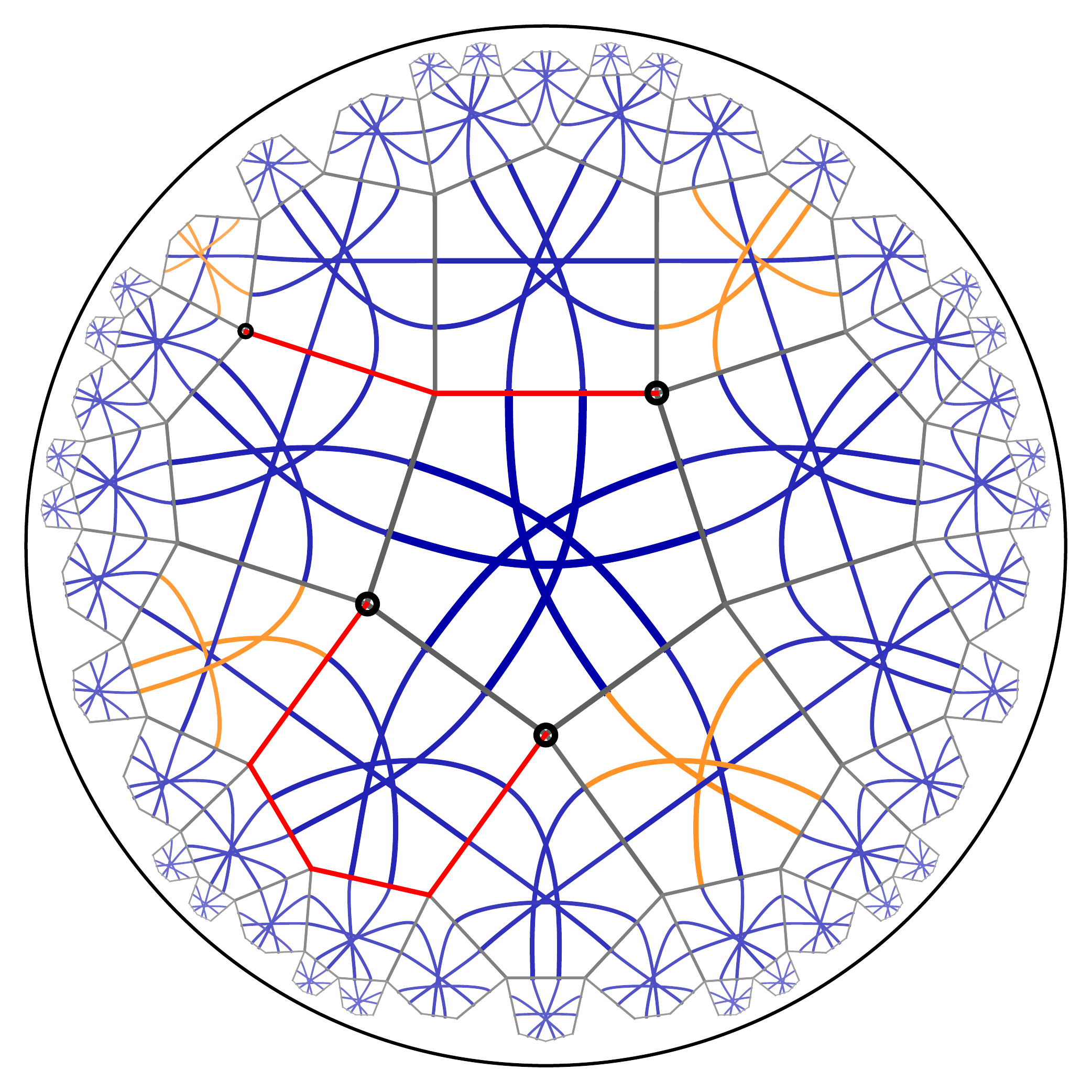}
\end{gathered}
\quad\quad \scalebox{1.5}{$=$}\quad\quad
\begin{gathered}
\includegraphics[height=0.17\textheight]{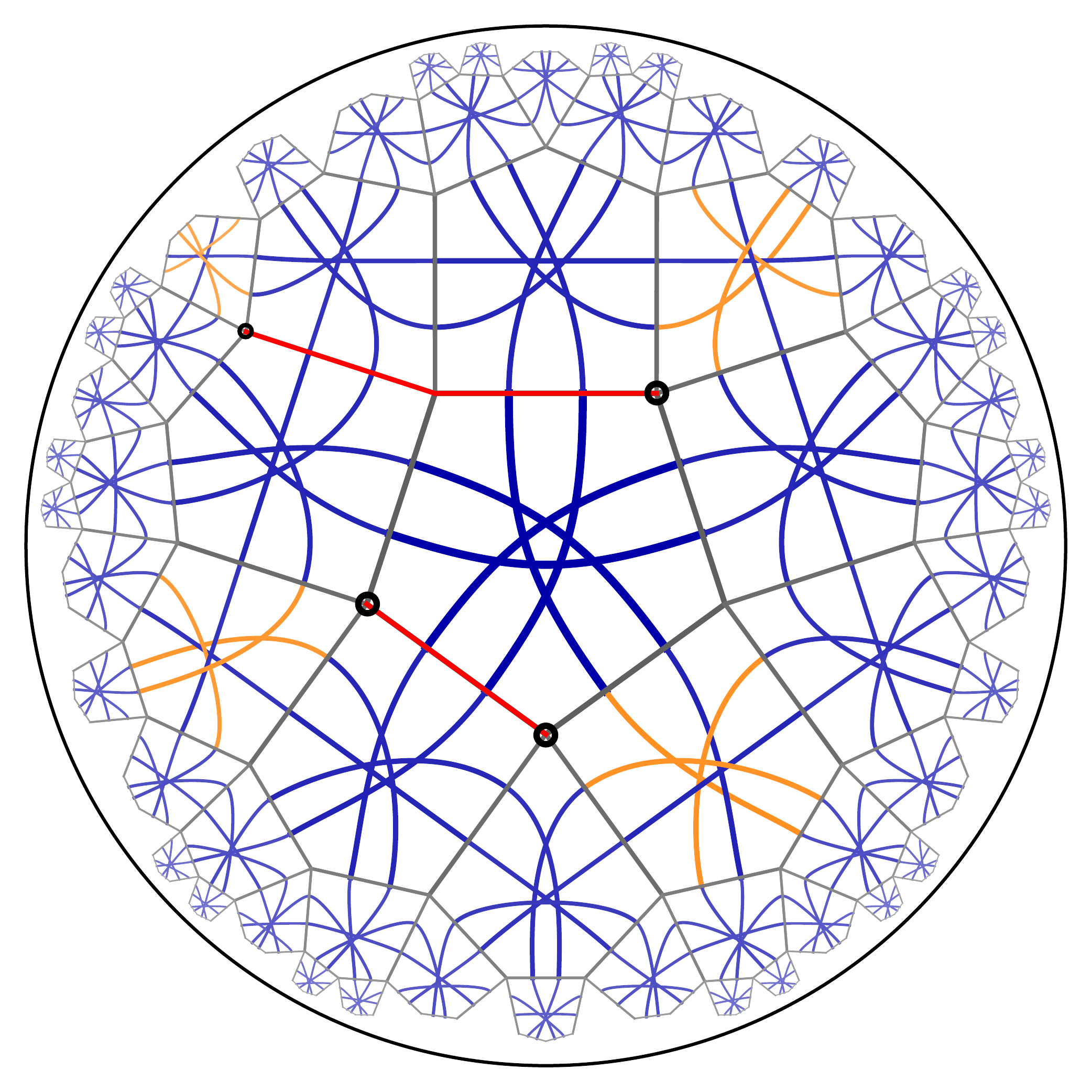}
\end{gathered}
\quad\quad \scalebox{1.5}{$=$}\quad\quad
\begin{gathered}
\includegraphics[height=0.17\textheight]{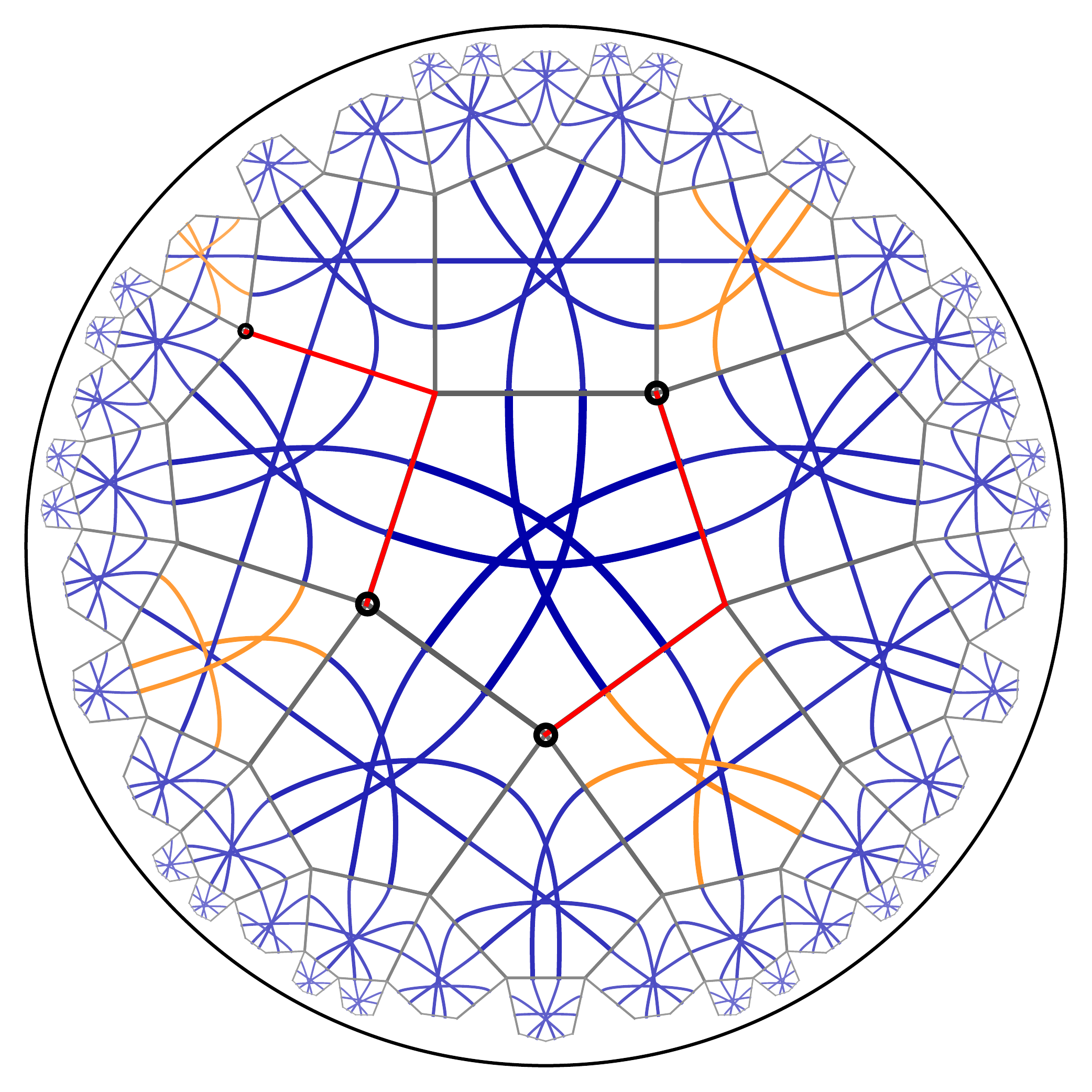}
\end{gathered}
\end{equation}
The same rules apply if we have an odd number of $\bar{1}$ tiles in the bulk. As the entire contraction now has odd parity, it also requires a pivot, which pairs up with the last $\bar{1}$ tile in the ordering. Again, this choice of a ``last'' tile does not change the outcome:
\begin{equation}
\begin{gathered}
\includegraphics[height=0.17\textheight]{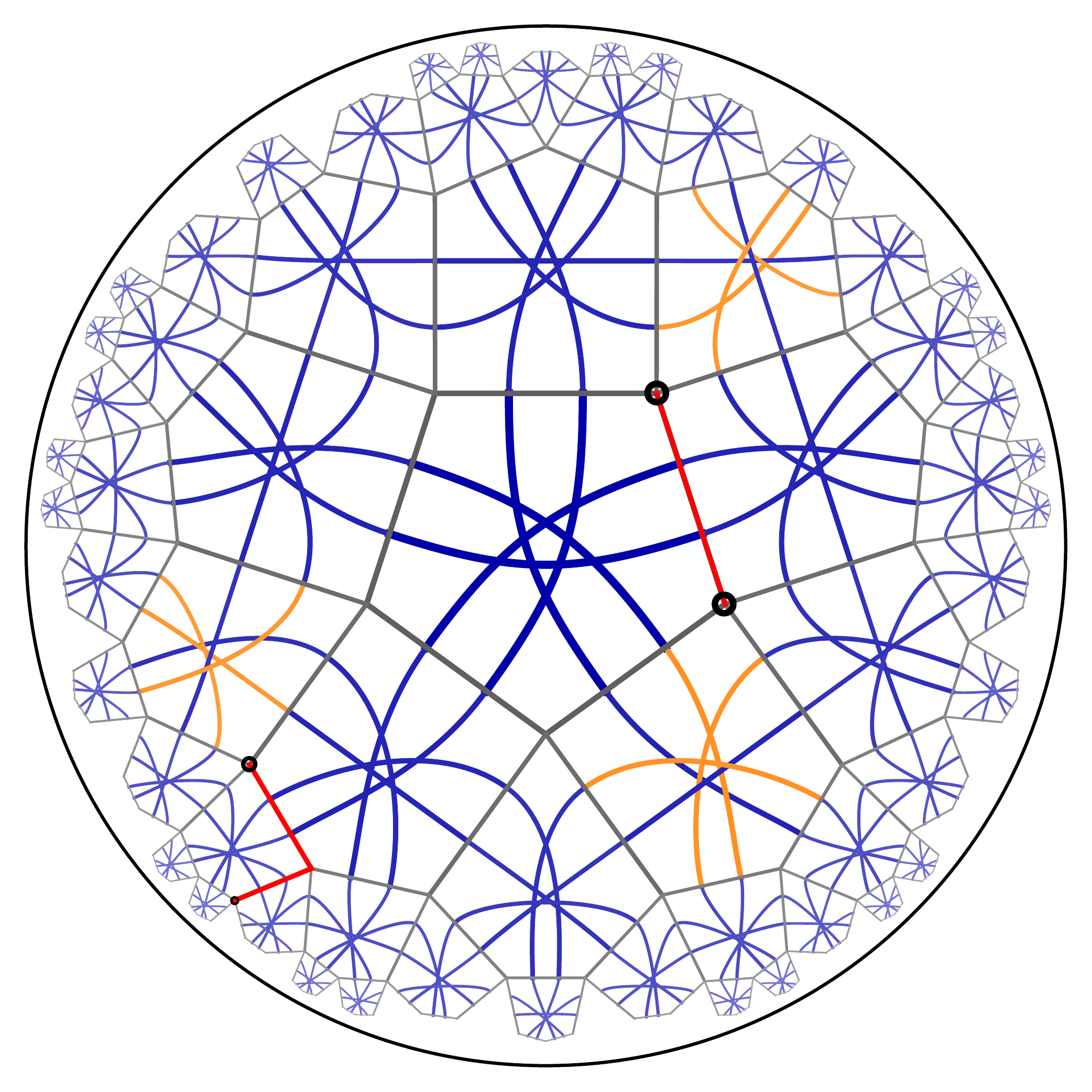}
\end{gathered}
\quad\quad \scalebox{1.5}{$=$}\quad\quad
\begin{gathered}
\includegraphics[height=0.17\textheight]{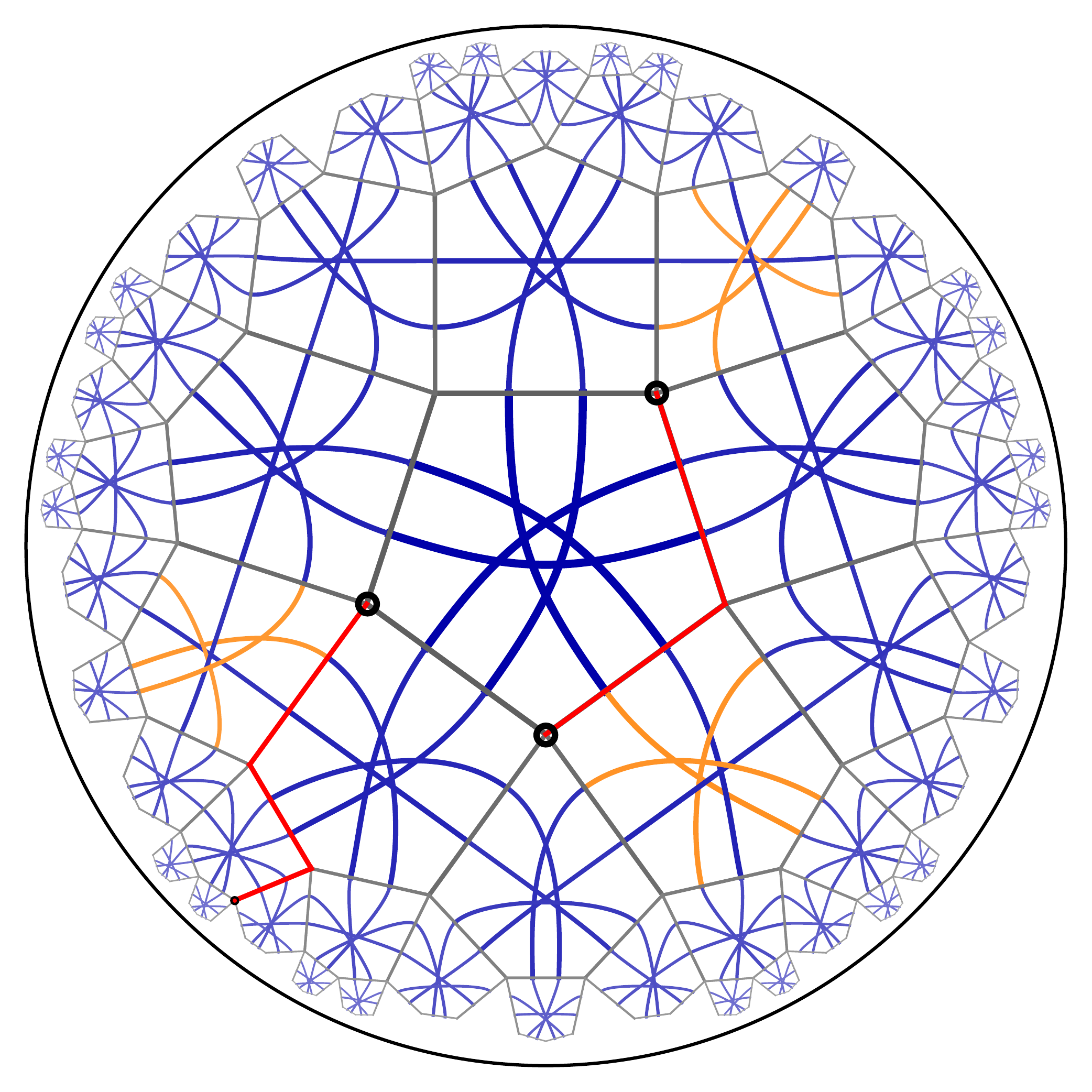}
\end{gathered}
\quad\quad \scalebox{1.5}{$=$}\quad\quad
\begin{gathered}
\includegraphics[height=0.17\textheight]{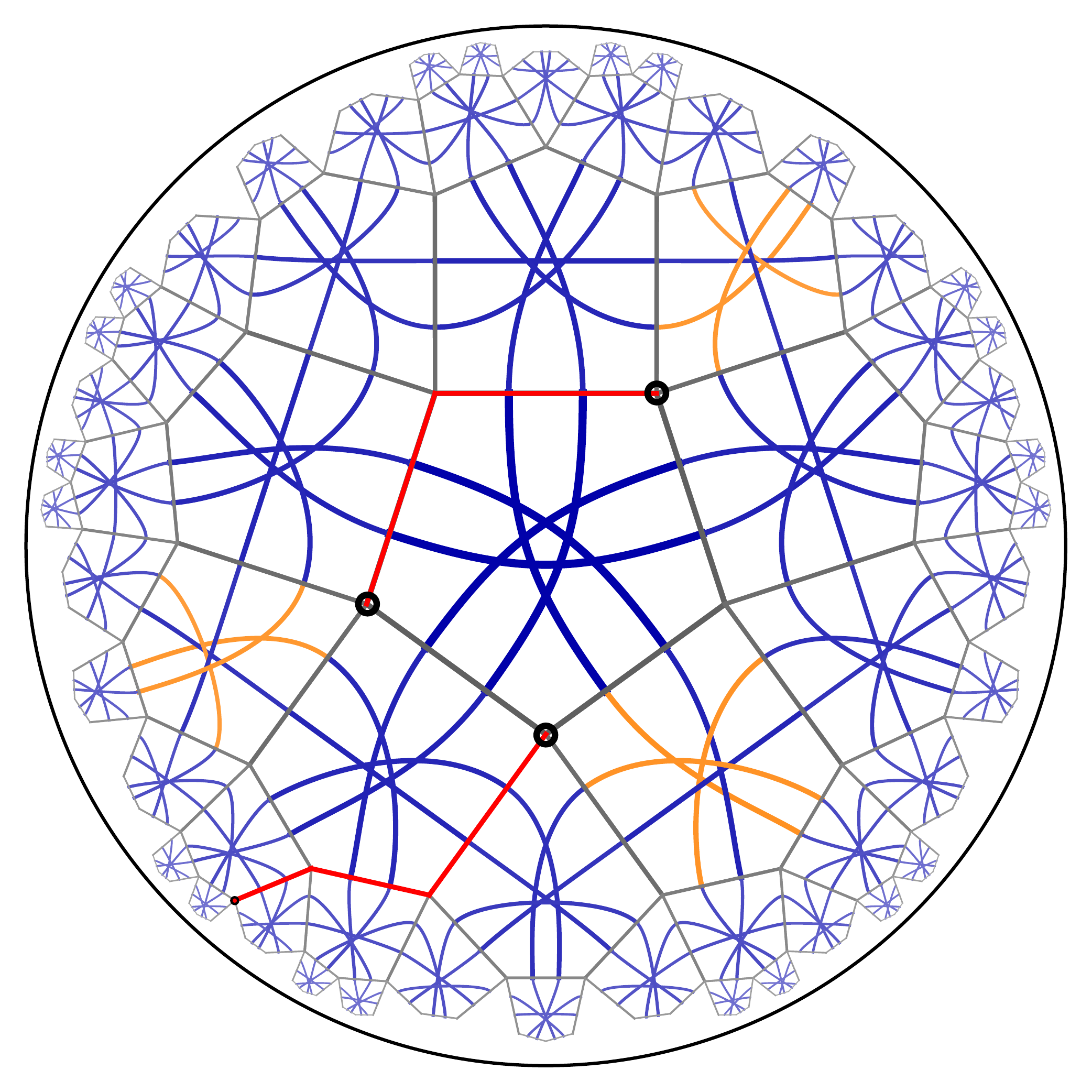}.
\end{gathered}
\end{equation}
Note that in the last step, we pushed $Z$ strings through two pentagon tiles. Consistent with the lemma above, moving the boundary pivot of the full contraction extends the $Z$ string attached to it along the boundary, which is the expected behaviour for a cyclic permutation of an odd-parity dimer state.
\end{proof}

Lemma \ref{LEM_HAPPY_FIXED_INPUT} now allows us to make some statements on the distance between Majorana boundary states for different bulk inputs. Let us define the \textsl{Majorana weight} $\mathfrak{w}$ as the number of Majorana operators (i.e.,\ dimer parity flips) required to transform one state into another. Given a boundary state vector $\ket{\bar{0},\bar{1},\bar{1},\bar{0},\dots}$ for an arbitrary bulk input, what is the lowest $\mathfrak{w}_{\text{min}}$ with respect to a state with any other bulk input? We claim the following:

\begin{lemma}[Majorana distance of HyPeC boundary states]
The boundary states of the HyPeC for fixed logical input in the bulk have a code distance $\mathfrak{w}>2$ between any two inputs. 
\end{lemma}

\begin{proof}
We will now show that starting from any fixed-input HyPeC boundary state, no number of logical input flips in the bulk can lead to a state which is closer that $\mathfrak{w}=3$ to the original one. This bound is clearly saturated for such an input flip $\bar{0}\leftrightarrow\bar{1}$ of a tile on the boundary, which flips three dimer parities. If we instead push the input flip further into the bulk, we will produce a $Z$ string from the boundary (or annihilate one, if the original contraction is parity-odd). The further in the bulk the flip occurs, the longer the $Z$ string grows, increasing $\mathfrak{w}$.
Due to the hyperbolic geometry, there is also no way that the dimer flips by neighbouring $\bar{1}$ insertions can cancel each out. For neighbouring pairs of $\bar{1}$ insertions, we always find $\mathfrak{w} > 3$:
\begin{equation}
\begin{gathered}
\includegraphics[height=0.17\textheight]{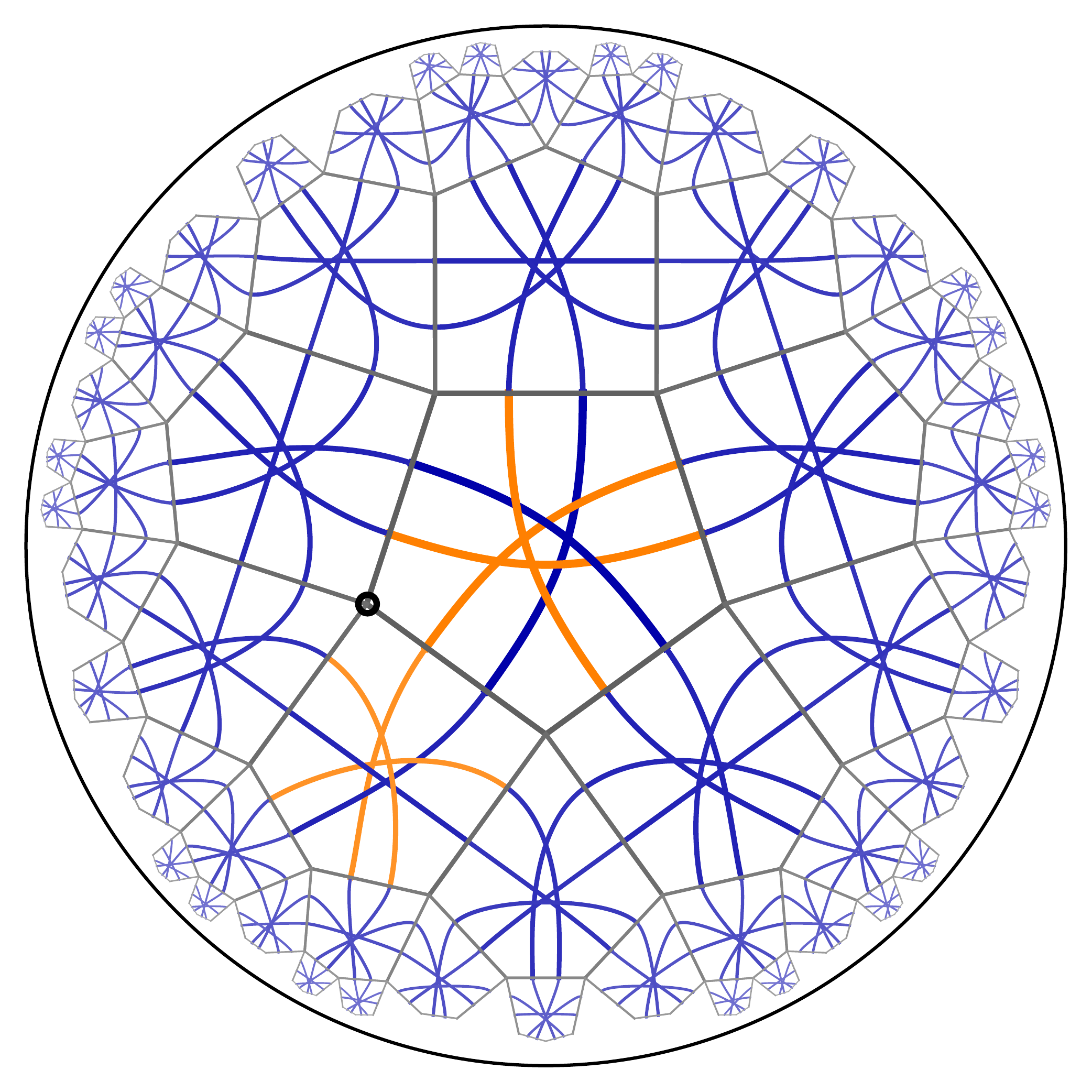}
\end{gathered}
\quad\quad \scalebox{1.5}{$=$}\quad\quad
\begin{gathered}
\includegraphics[height=0.17\textheight]{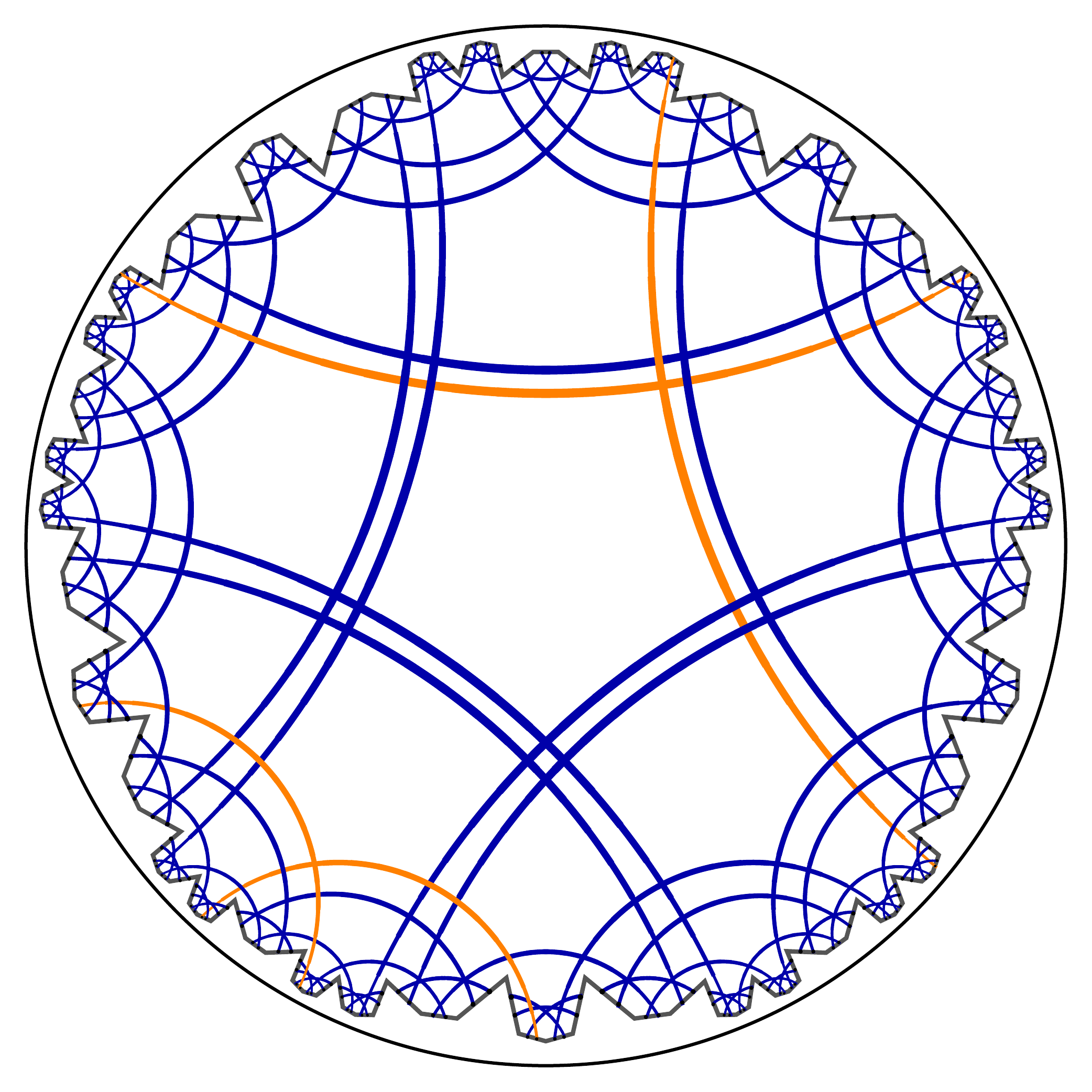}
\end{gathered}
\quad\quad\scalebox{1.1}{$(\mathfrak{w}=4)$}
\end{equation}
\begin{equation}
\begin{gathered}
\includegraphics[height=0.17\textheight]{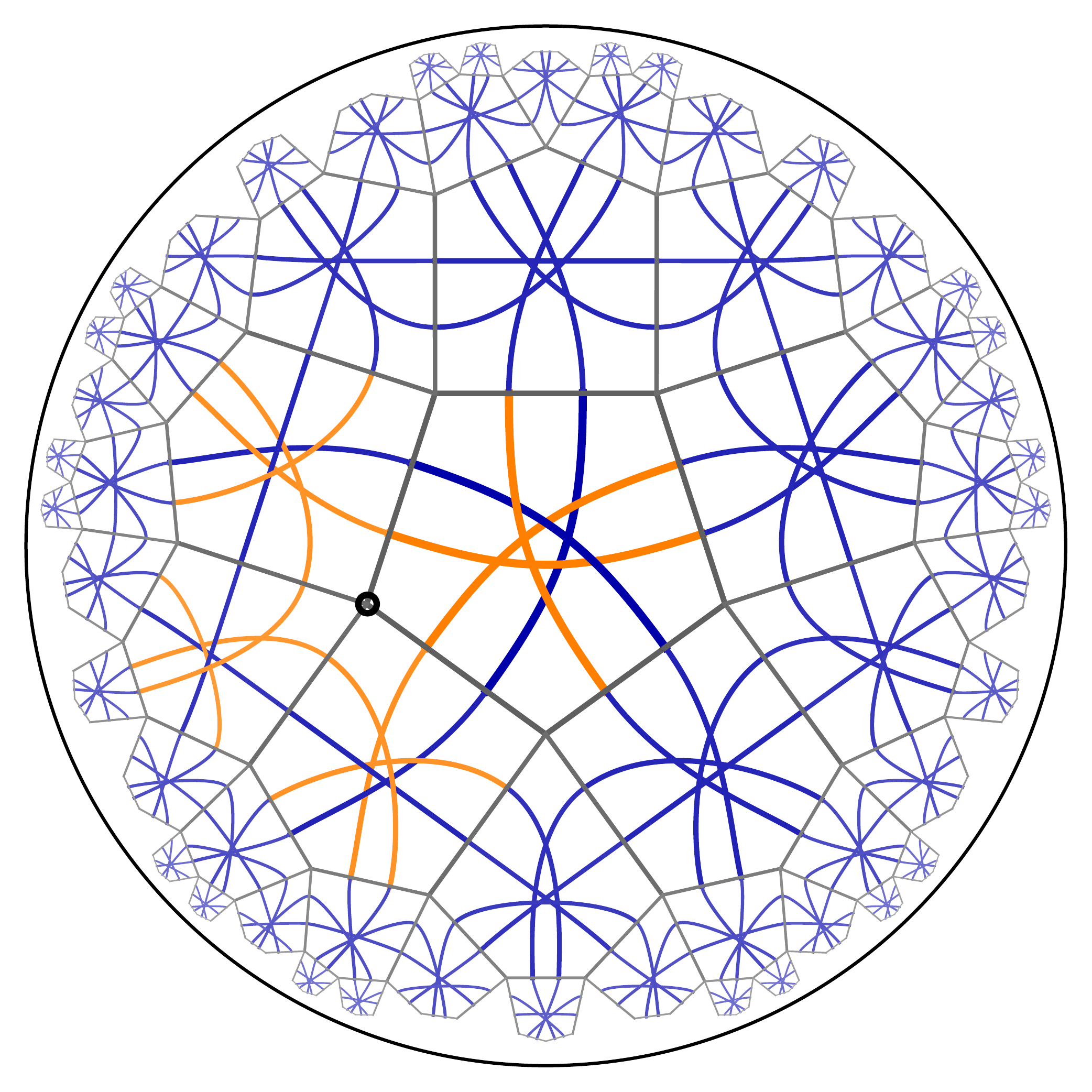}
\end{gathered}
\quad\quad \scalebox{1.5}{$=$}\quad\quad
\begin{gathered}
\includegraphics[height=0.17\textheight]{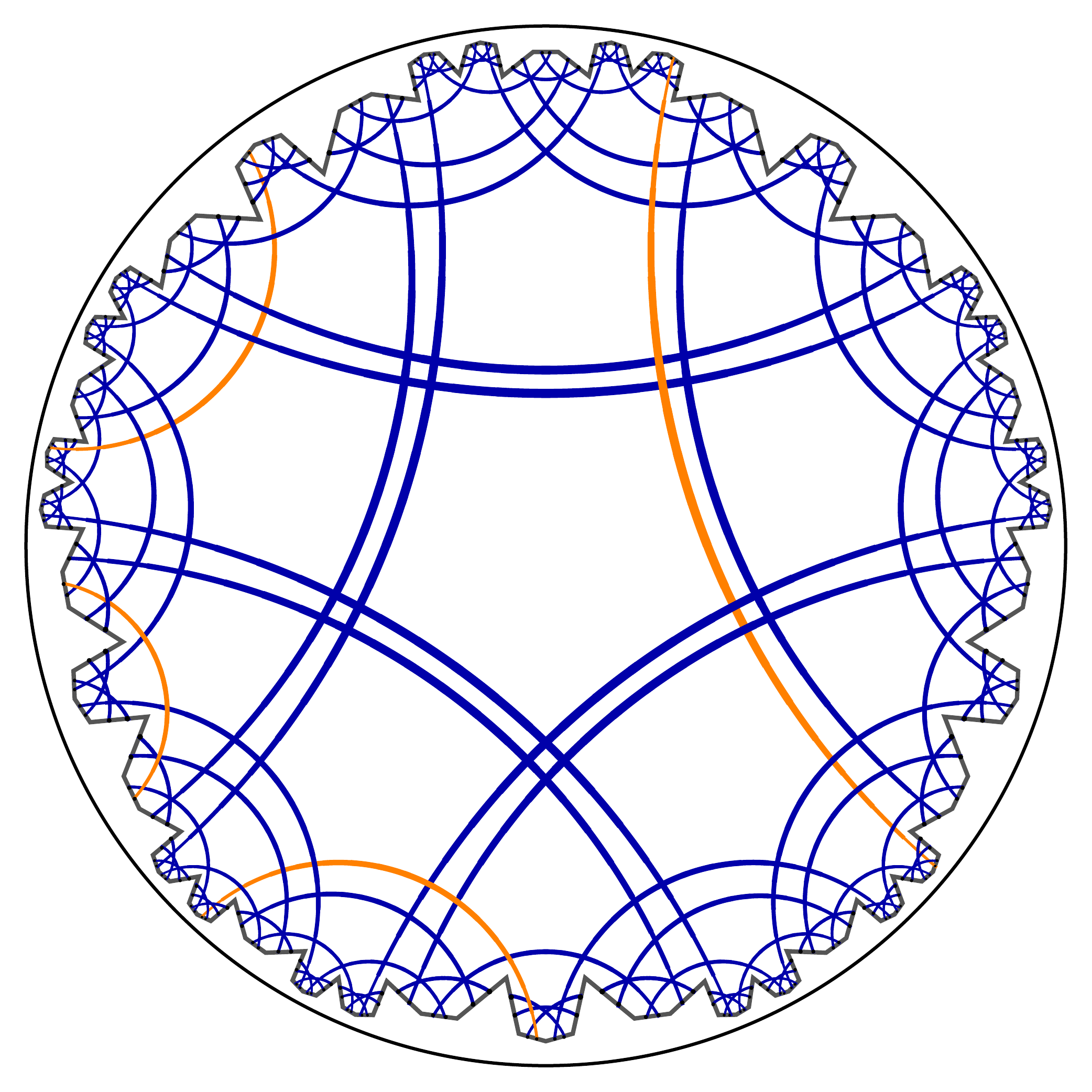}
\end{gathered}
\quad\quad\scalebox{1.1}{$(\mathfrak{w}=4)$}
\end{equation}
We have defined $\mathfrak{w}$ relative to the all-$\bar{0}$ input, but the result clearly holds for insertions on any fixed code input. When non-neighbouring pairs are added, the resulting $Z$ strings cause additional dimer flips:
\begin{equation}
\begin{gathered}
\includegraphics[height=0.17\textheight]{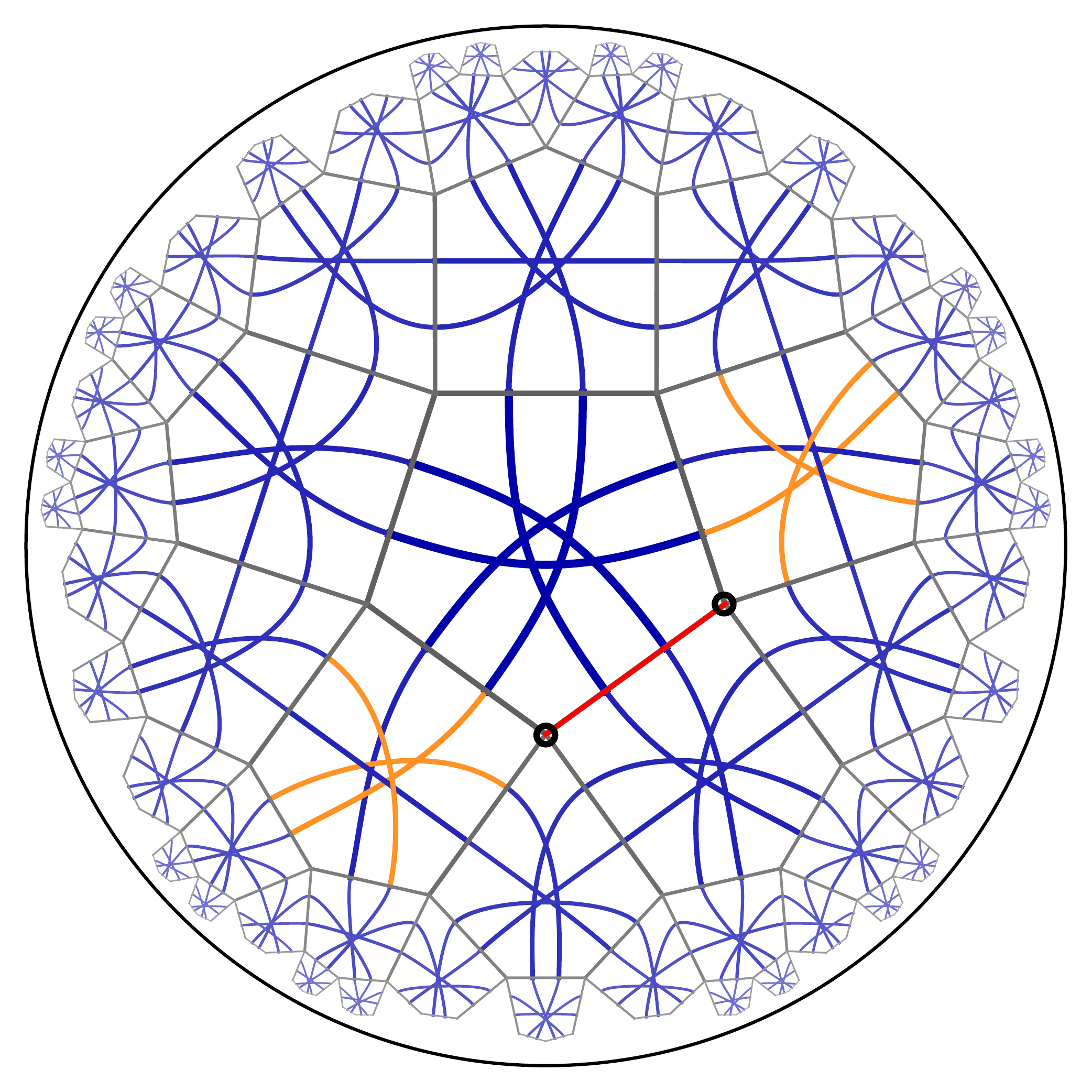}
\end{gathered}
\quad\quad \scalebox{1.5}{$=$}\quad\quad
\begin{gathered}
\includegraphics[height=0.17\textheight]{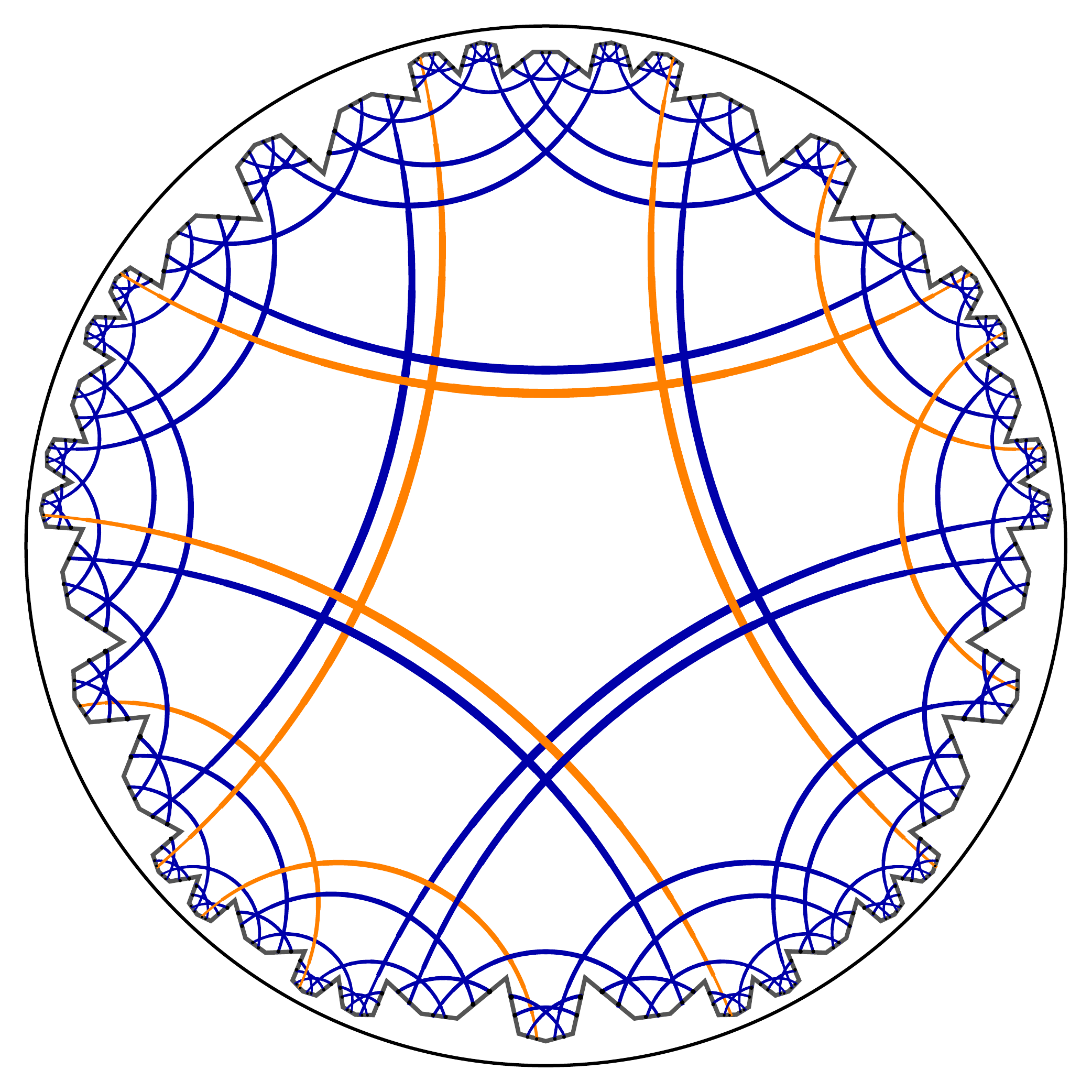}
\end{gathered}
\quad\quad\scalebox{1.1}{$(\mathfrak{w}=8)$}
\end{equation}
Similarly, adding even more pairs to make two $Z$ strings ``cancel'' out does not bring down $\mathfrak{w}$:
\begin{equation}
\begin{gathered}
\includegraphics[height=0.17\textheight]{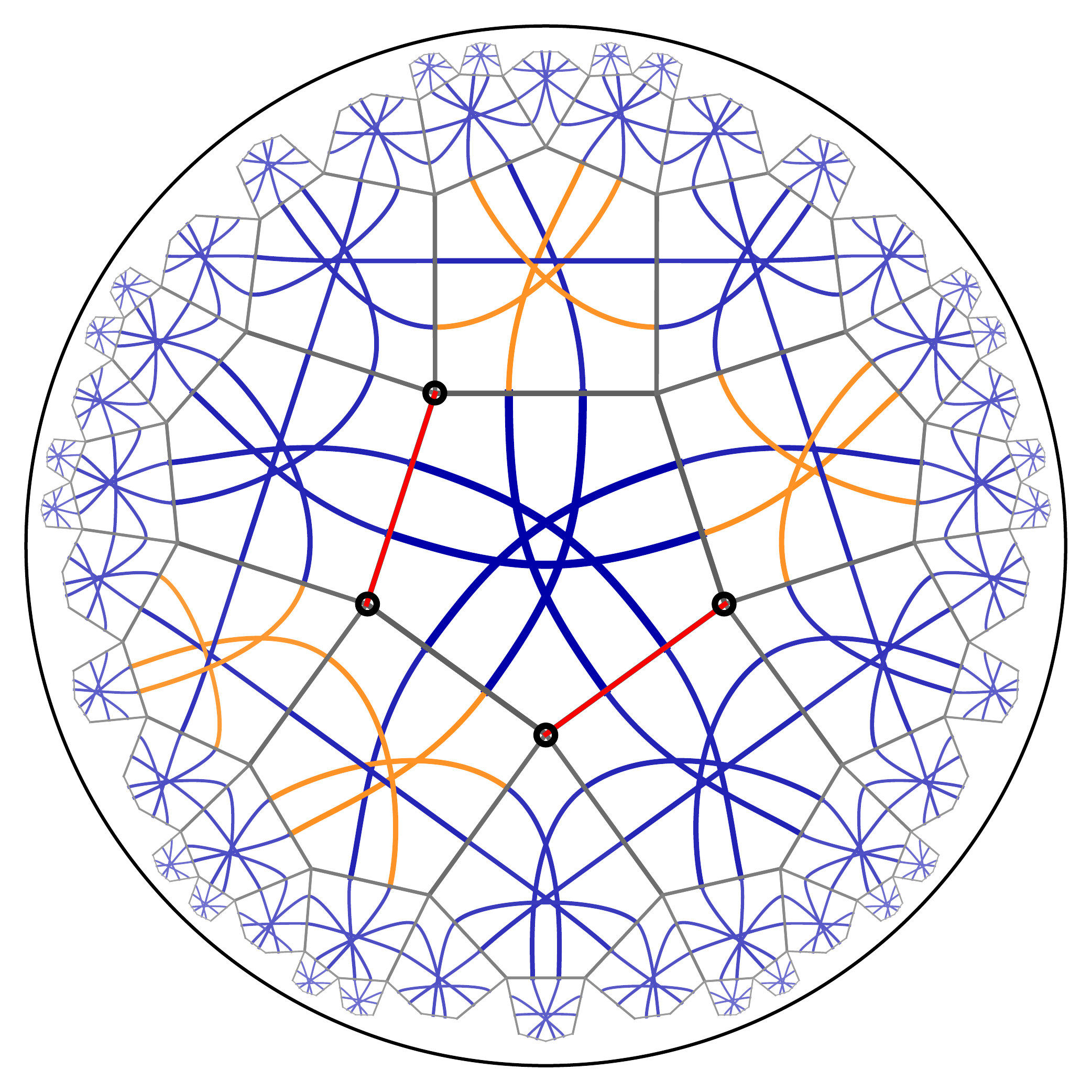}
\end{gathered}
\quad\quad \scalebox{1.5}{$=$}\quad\quad
\begin{gathered}
\includegraphics[height=0.17\textheight]{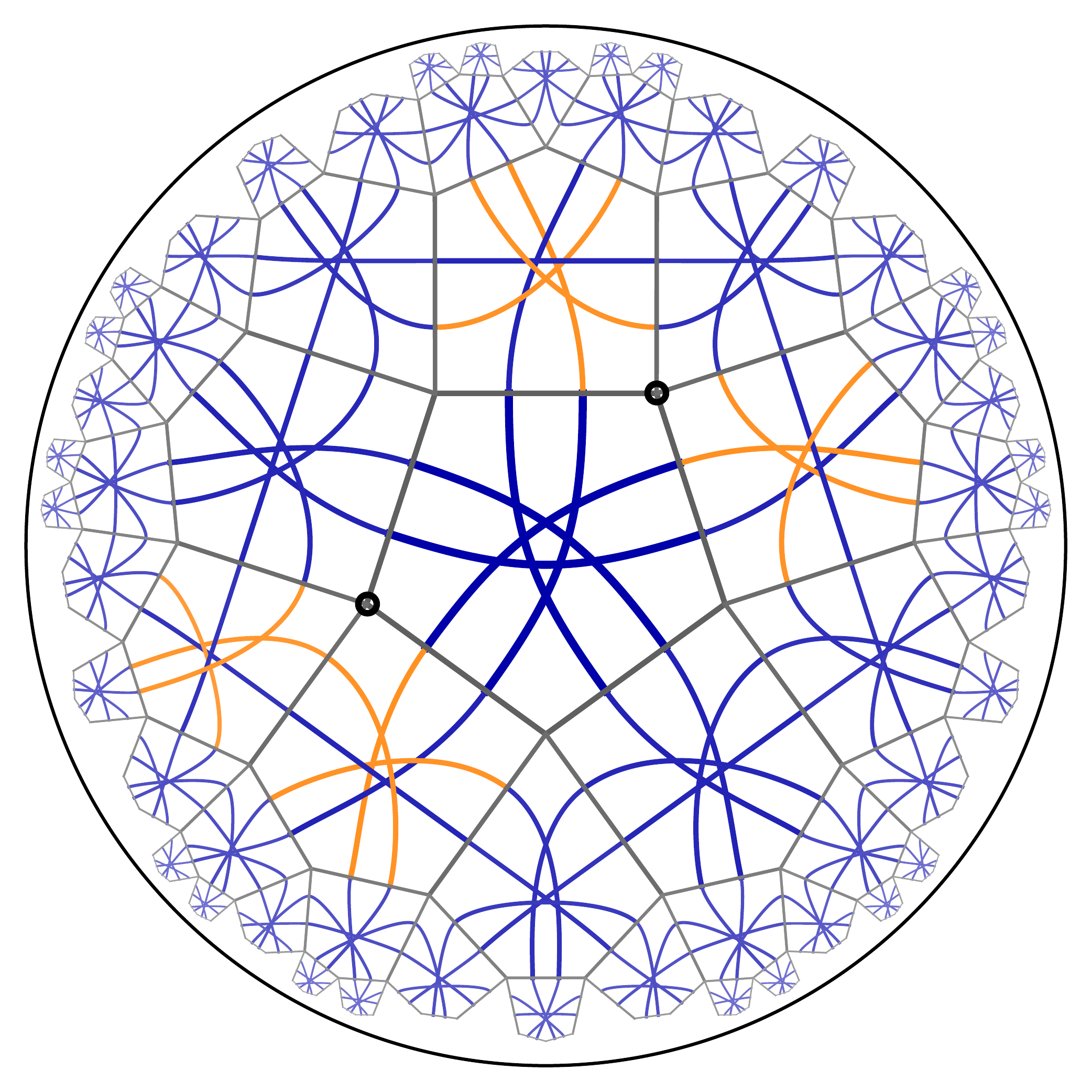}
\end{gathered}
\quad\quad \scalebox{1.5}{$=$}\quad\quad
\begin{gathered}
\includegraphics[height=0.17\textheight]{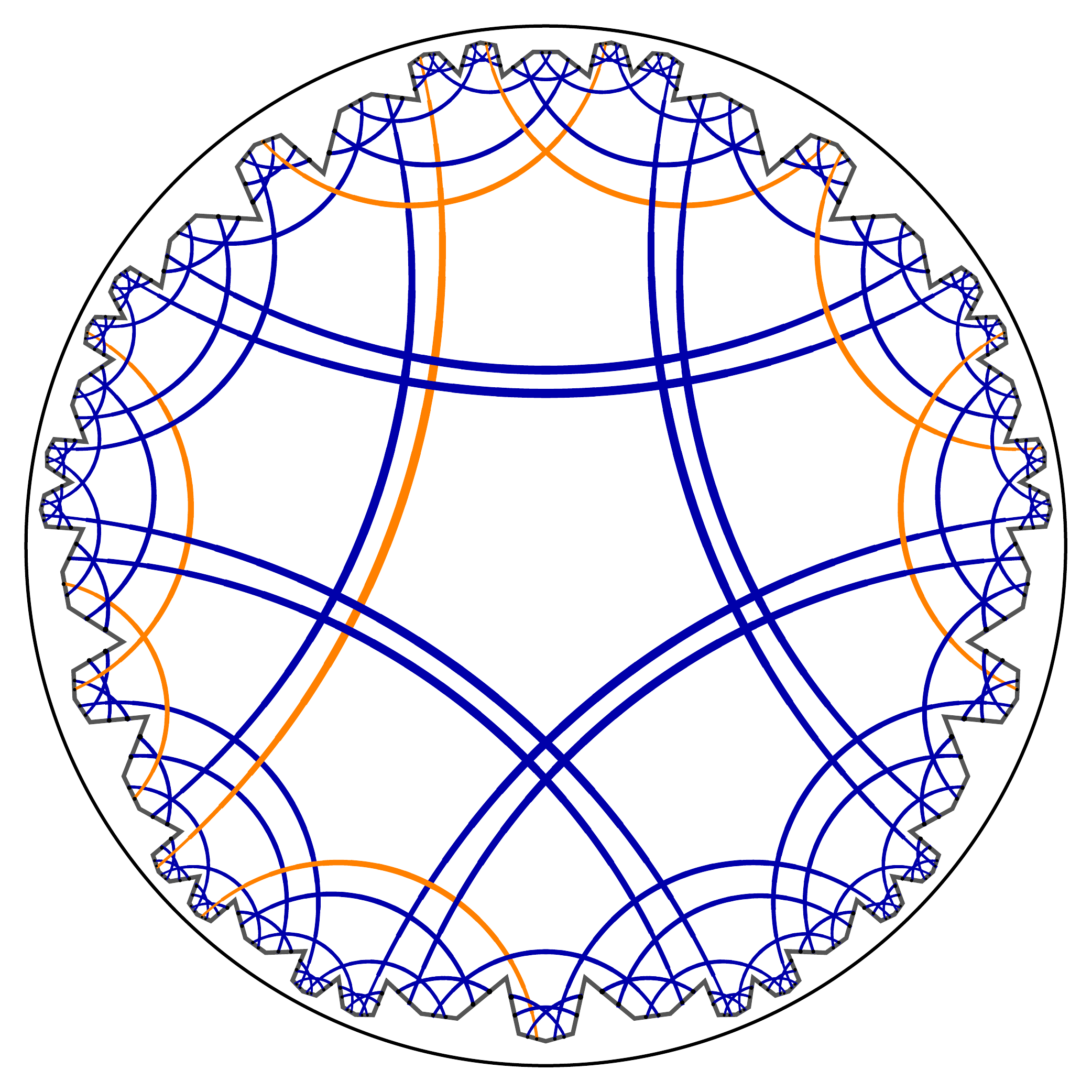}
\end{gathered}
\quad\quad\scalebox{1.1}{$(\mathfrak{w}=8)$}
\end{equation}
As a result, it is impossible to produce Majorana dimer states on the boundary of the fixed-input HyPeC that can be mapped to each other with less than $\mathfrak{w}=3$ Majorana operators. The underlying reason for this can be found in the geometrical construction: The number of possible boundary configurations $2^M$ on $L$ boundary edges increases much faster than the $2^N$ configurations on the $N$ bulk tiles, as the geometry is hyperbolic. 
\end{proof}

The property $\mathfrak{w}>2$ resembles the code properties of the HyPeC: Due to the tiles corresponding to $[[5,1,3]]$ code states, it requires three Pauli-type operations (``errors'') to map one code state to another. Thus, it requires at least three Pauli errors on the boundary to map any HyPeC state to another one. Here, we found that it also requires at least three ``Majorana errors'' to perform such a mapping. This is not a trivial result, as the number of Pauli operations corresponding to just two Majorana operations already grows in the distance between the two sites on which the Majorana operators act. For example, 
\begin{equation}
\m_2\m_{2k-1} = \i\, X_1 Z_2 Z_3 \dots Z_{k-1} X_k \text{ .}  
\end{equation}
In general, applying two Majorana operators $\m_j$ and $\m_k$ at some distance on the boundary produces a $Z$ string between the edges on which $\m_j$ and $\m_k$ act.
Fortunately, the $[[5,1,3]]$ code states upon which the HyPeC are built allow for the expression of long $Z$ strings as an action of just two Pauli operators as
\begin{align}
Z_1 Z_2 Z_3 \ket{\bar{0}}\; = \;
\begin{gathered}
\includegraphics[height=0.09\textheight]{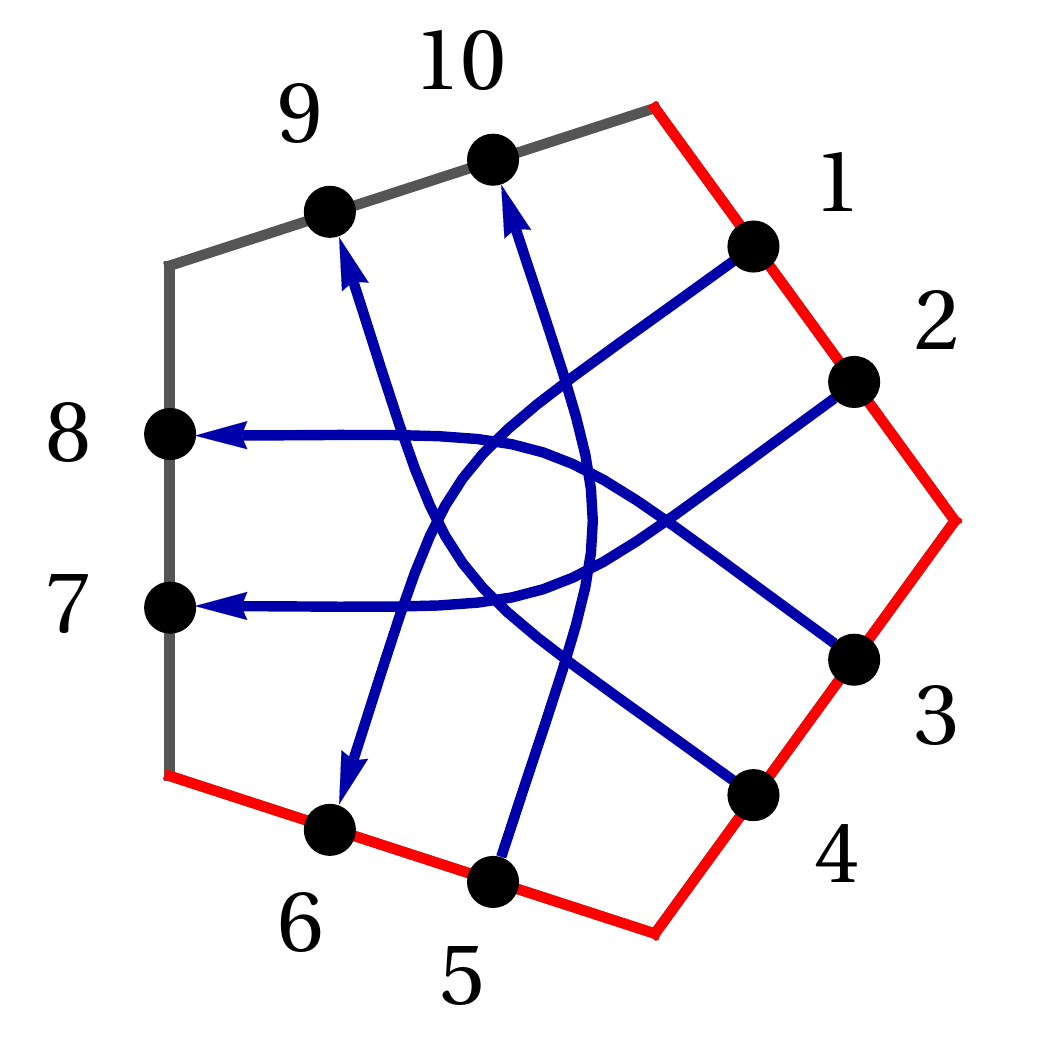}
\end{gathered}
\; = \;
\begin{gathered}
\includegraphics[height=0.09\textheight]{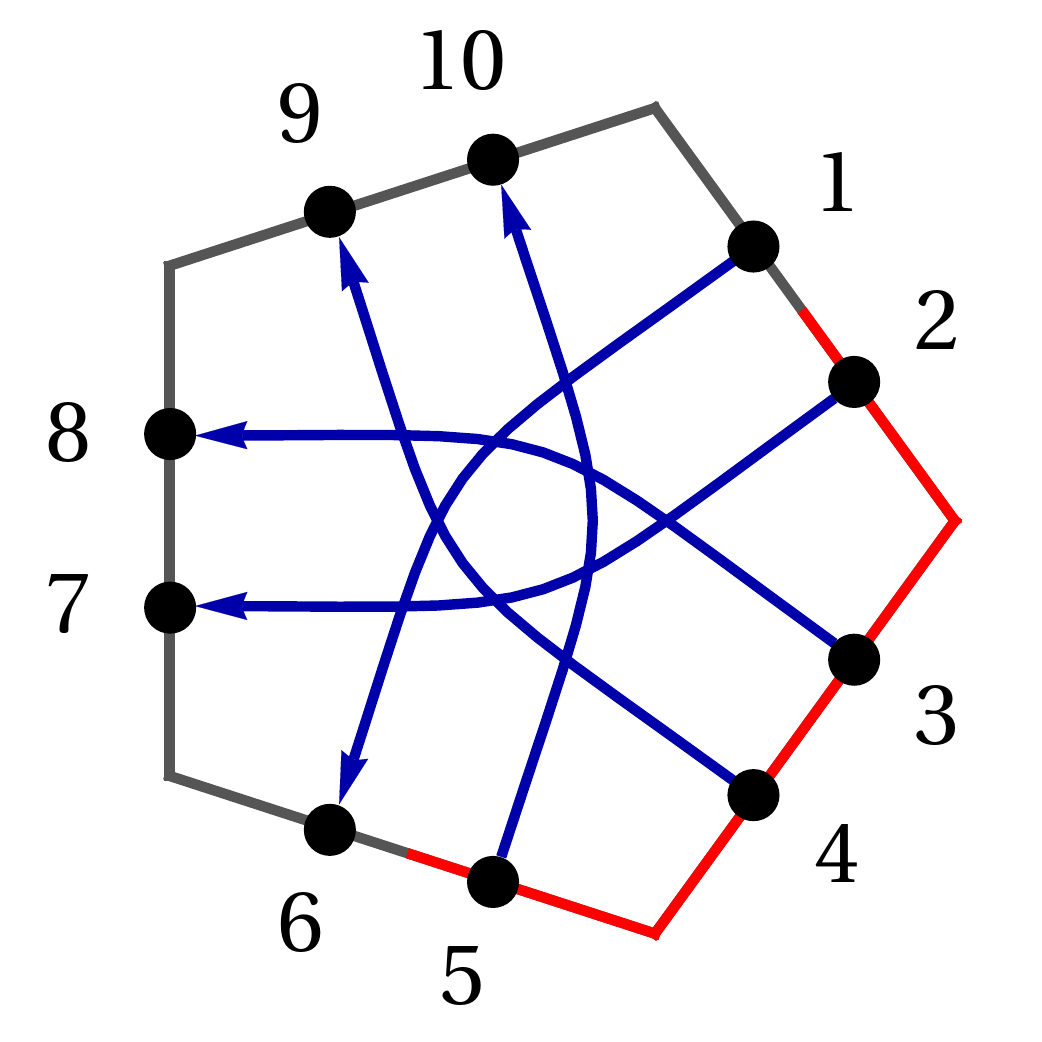}
\end{gathered} 
\; = \; X_1 X_3 \ket{\bar{0}},
\\
Z_1 Z_2 Z_3 Z_4 \ket{\bar{0}}\; = \;
\begin{gathered}
\includegraphics[height=0.09\textheight]{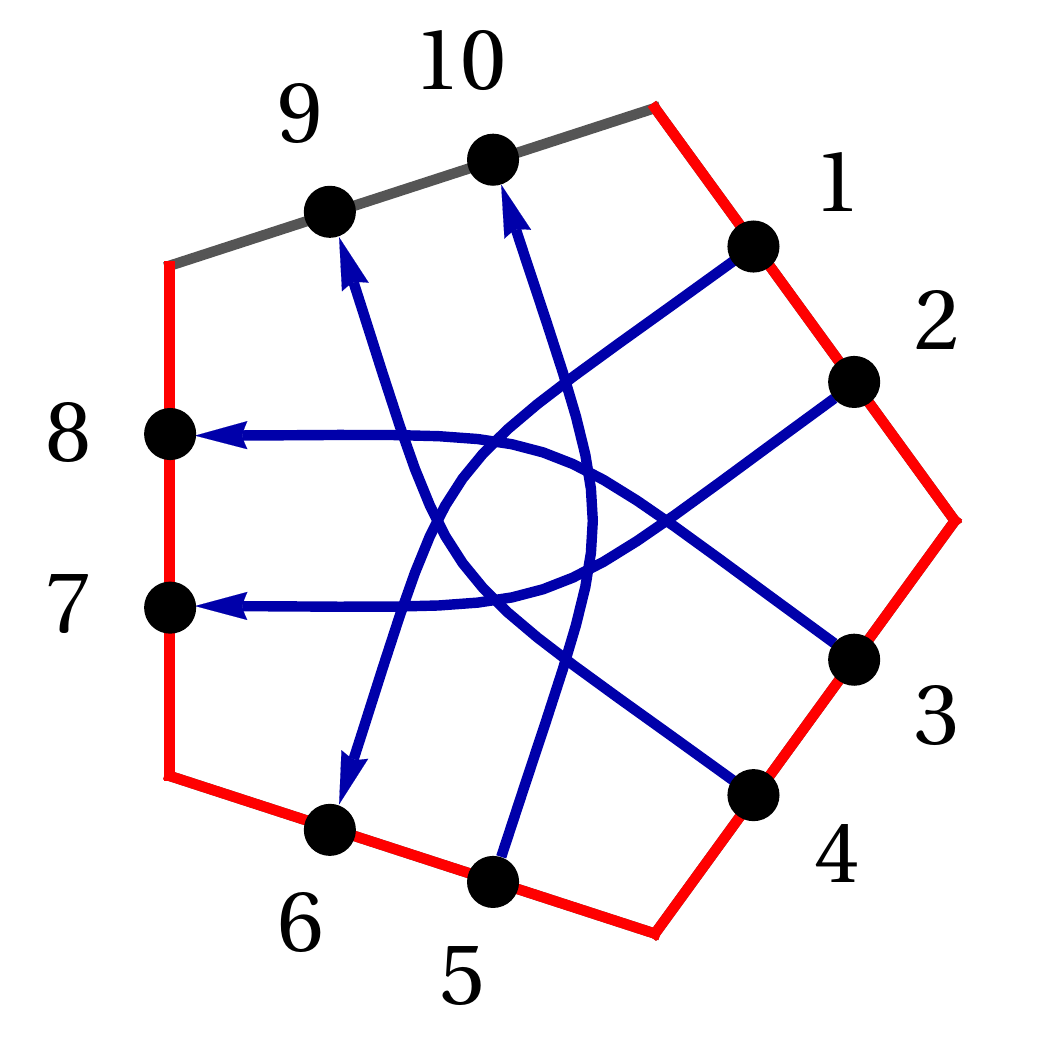}
\end{gathered}
\; = \;
\begin{gathered}
\includegraphics[height=0.09\textheight]{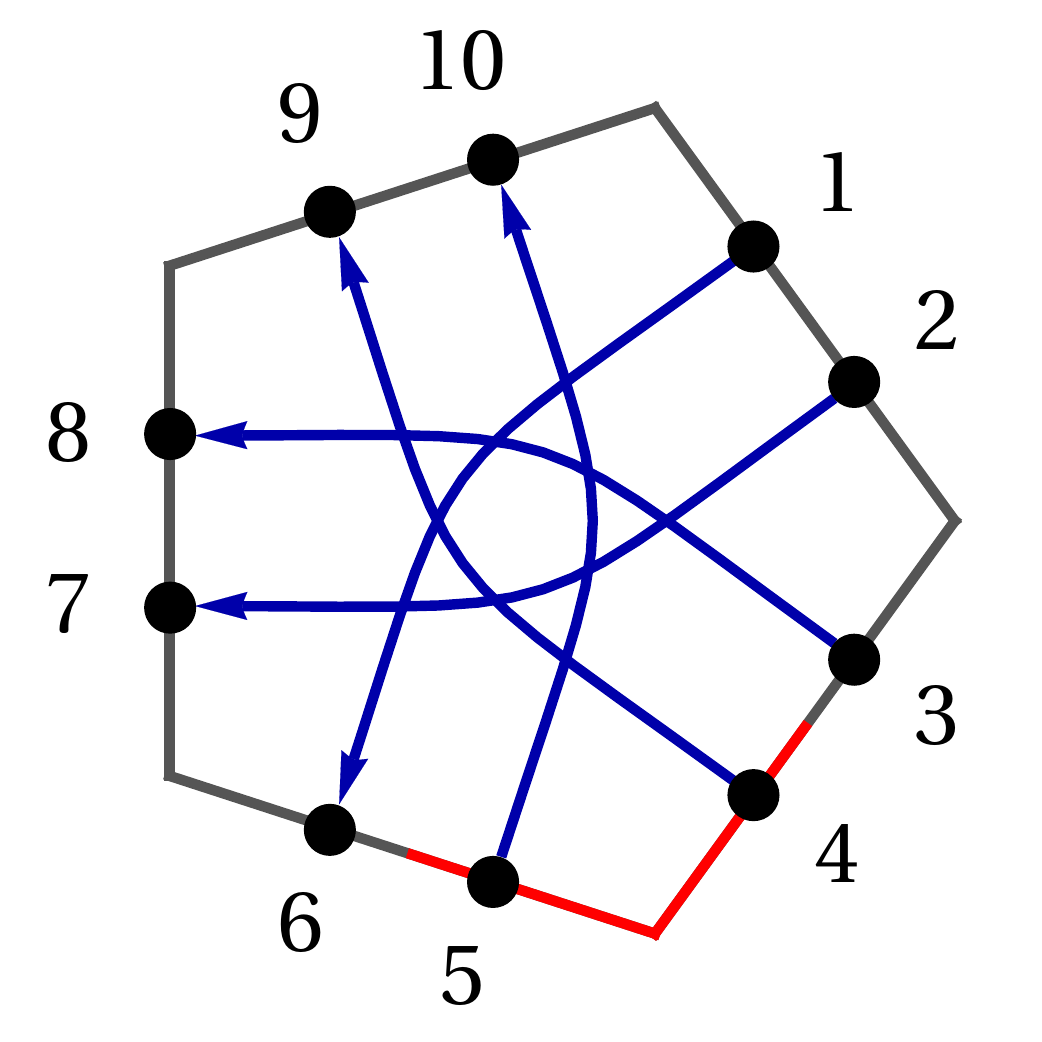}
\end{gathered} 
\; = \; X_2 X_3 \ket{\bar{0}}.
\end{align}
Here, we have used the $\bar{0}$ input for illustration. The relative sign between the left- and rand-hand side of these equations changes when using the $\bar{1}$ input instead, corresponding to a ``phase flip'' in the language of quantum error correction.
We conclude that a pair of Majorana operators on the boundary of the HyPeC is related, up to a complex phase, to no more than two Pauli operators acting on tiles on the boundary. As each of these tiles corrects one Pauli errors, no overlap between states for different bulk inputs can be produced with such an operation, supporting our earlier geometrical explanation.

\end{widetext}

\end{document}